\documentclass[12pt]{article}%
\usepackage{amssymb}
\usepackage{amsmath}
\usepackage{amsfonts}
\usepackage{sw20dft1}%
\usepackage{graphicx}
\providecommand{\U}[1]{\protect\rule{.1in}{.1in}}
\newtheorem{theorem}{Theorem}

\newtheorem{algorithm}[theorem]{Algorithm}

\newtheorem{corollary}[theorem]{Corollary}

\newtheorem{definition}[theorem]{Definition}

\newtheorem{lemma}[theorem]{Lemma}

\newtheorem{proposition}[theorem]{Proposition}
\newtheorem{remark}[theorem]{Remark}

\newenvironment{proof}[1][Proof]{\noindent\textbf{#1.} }{\ \rule{0.5em}{0.5em}}
\begin{document}

\title{Higher order Influence Functions and Minimax Estimation of Nonlinear Functionals}
\runningtitle{Higher Order Influence Functions}
\author{James Robins$^{a}$, Lingling Li$^{b}$, Eric Tchetgen Tchetgen$^{a}$, Aad van
der Vaart$^{c}$\\$^{a}$Harvard School of Public Health, $^{b}$Harvard Pilgrim Health Care
Institute, $^{c}$University of Leiden }
\keyphrases{U-statistics, Minimax, Influence Functions, Nonparametric, Robust Inference}
\primaryclass{123}
\secondaryclass{456}
\maketitle

\begin{abstract}
Robins et al, 2008, published a theory of higher order influence functions for
inference in semi- and non-parametric models. \ This paper is a comprehensive
manuscript from which Robins et al, was drawn. \ The current paper includes
many results and proofs that were not included in Robins et al due to space
limitation. Particular results contained in the present paper that were not
reported in Robins et al include the following. Given a set of functionals and
their corresponding higher order influence functions, we show how to derive
the higher order influence function of their product. We apply this result to
obtain higher order influence functions and associated estimators for the mean
of a response $Y$ subject to monotone missingness under missing at random.
These results also apply to estimating the causal effect of a time dependent
treatment on an outcome $Y$ in the presence of time-varying confounding.
\ Finally, we include an appendix that contains proofs for all theorems that
were stated without proof in Robins et al, 2008. The initial part of the paper
is closely related to Robins et al,, the latter parts differ.

\ Specifically, we present a theory of point and interval estimation for
nonlinear functionals in parametric, semi-, and non-parametric models based on
higher order influence functions (Robins \cite{robins3}, Sec. 9; Li et al.
\cite{Li}, Tchetgen et al, \cite{tchetgen}, Robins et al, \cite{Vaart2}).
Higher order influence functions are higher order U-statistics. Our theory
extends the first order semiparametric theory of Bickel et al. \cite{bickel1}
and van der Vaart \cite{vaart} by incorporating the theory of higher order
scores considered by Pfanzagl \cite{pfanzagl}, Small and McLeish
\cite{mcleish} and Lindsay and Waterman \cite{lindsay}. The theory reproduces
many previous results, produces new non-$\sqrt{n}$ results, and opens up the
ability to perform optimal non-$\sqrt{n}$ inference in complex high
dimensional models. We present novel rate-optimal point and interval
estimators for various functionals of central importance to biostatistics in
settings in which estimation at the expected $\sqrt{n}$ rate is not possible,
owing to the curse of dimensionality. We also show that our higher order
influence functions have a multi-robustness property that extends the double
robustness property of first order influence functions described by Robins and
Rotnitzky \cite{rotnitzky} and van der Laan and Robins \cite{laan2}.

\end{abstract}

%\newline

\section{Introduction}

Robins et al, 2008, published a theory of higher order influence functions for
inference in semi- and non-parametric models. \ This paper is a comprehensive
manuscript from which Robins et al, was drawn. \ The current paper includes
many results and proofs that were not included in Robins et al due to space
limitation. Particular results contained in the present paper that were not
reported in Robins et al include the following. Given a set of functionals and
their corresponding higher order influence functions, we show how to derive
the higher order influence function of their product. We apply this result to
obtain higher order influence functions and associated estimators for the mean
of a response $Y$ subject to monotone missingness under missing at random.
These results also apply to estimating the causal effect of a time dependent
treatment on an outcome $Y$ in the presence of time-varying confounding.
\ Finally, we include an appendix that contains proofs for all theorems that
were stated without proof in Robins et al, 2008. The initial part of the paper
is closely related to Robins et al,, the latter parts differ.

We have developed a theory of point and interval estimation for nonlinear
functionals $\psi\left(  F\right)  $ in parametric, semi-, and non-parametric
models based on higher order likelihood scores and influence functions that
applies equally to both $\sqrt{n}$ and non-$\sqrt{n}$ problems (Robins 2004,
Sec. 9; Li et al, 2006, Tchetgen et al, 2006, Robins et al, 2007). The theory
reproduces results previously obtained by the modern theory of non-parametric
inference, produces many new non-$\sqrt{n}$ results, and most importantly
opens up the ability to perform non-$\sqrt{n}$ inference in complex high
dimensional models, such as models for the estimation of the causal effect of
time varying treatments in the presence of time varying confounding and
informative censoring. See Tchetgen et al. (2007) for examples of the latter.

Higher order influence functions are higher order U-statistics. Our theory
extends the first order semiparametric theory of Bickel et al. (1993) and van
der Vaart (1991) by incorporating the theory of higher order scores and
Bhattacharrya bases considered by Pfanzagl (1990), Small and McLeish (1994)
and Lindsay and Waterman (1996).

The purpose of this paper is to demonstrate the scope and flexibility of our
methodology by deriving rate-optimal point and interval estimators for various
functionals that are of central importance to biostatistics. We now describe
some of these functionals. We suppose we observe i.i.d copies of a random
vector $O=\left(  Y,A,X\right)  $ with unknown distribution $F$ on each of $n$
study subjects. In this paper, we largely study non-parametric models that
place no restrictions on $F,$ other than bounds on both the $L_{p}$ norms and
on the smoothness of certain density and conditional expectation functions.
The variable $X$ represents a random vector of baseline covariates such as
age, height, weight, hematocrit, and laboratory measures of lung, renal,
liver, brain, and heart function. $X$ is assumed to have compact support and a
density $f_{X}\left(  x\right)  $ with respect to the Lebesgue measure in
$R^{d},\ $where, in typical applications, $d$ is in the range 5 to 100. $A$ is
a binary treatment and $Y$ is a response, higher values of which are
desirable. Then, in the absence of confounding by additional unmeasured
factors, the functional $\psi\left(  F\right)  =E\left\{  E\left[
Y|A=1,X\right]  \right\}  -E\left\{  E\left[  Y|A=0,X\right]  \right\}  $ is
the mean effect of treatment in the total study population. Our results for
$E\left\{  E\left[  Y|A=1,X\right]  \right\}  -E\left\{  E\left[
Y|A=0,X\right]  \right\}  $ follow from results for the functional
$\psi\left(  F\right)  =E\left\{  E\left[  Y|A=1,X\right]  \right\}  $ based
on data $\left(  AY,A,X\right)  $ rather than $\left(  Y,A,X\right)  .$ If $Y$
is missing for some study subjects, and $A$ is now the indicator that takes
the value $1$ when $Y$ is observed and zero otherwise, then the functional
$E\left\{  E\left[  Y|A=1,X\right]  \right\}  $ is the marginal mean of $Y$
under the missing at random assumption that the probability $P\left[
A=0|X,Y\right]  =P\left[  A=0|X\right]  \ $ that $Y$ is missing does not
depend on the unobserved $Y.$

Returning to data $O=\left(  Y,A,X\right)  ,$ the functional
\begin{align*}
\psi\left(  F\right)   &  =E\left\{  Cov\left(  Y,A|X\right)  \right\}
/E\left[  var\left\{  A|X\right\}  \right] \\
&  =E\left[  w\left(  X\right)  \left\{  E\left[  Y|A=1,X\right]  -E\left[
Y|A=0,X\right]  \right\}  \right]
\end{align*}
with $w\left(  X\right)  =var\left\{  A|X\right\}  /E\left[  var\left\{
A|X\right\}  \right]  \ $is the variance weighted average treatment effect.
Our results for $E\left\{  Cov\left(  Y,A|X\right)  \right\}  /E\left[
var\left\{  A|X\right\}  \right]  $ are derived from results for the
functionals $\psi\left(  F\right)  =E\left\{  Cov\left(  Y,A|X\right)
\right\}  $ and $\psi\left(  F\right)  =E\left[  \left\{  E\left(  Y|X\right)
\right\}  ^{2}\right]  .$

We note that Robins and van der Vaart's (2006) construction of an adaptive
confidence set for a regression function $E\left(  Y|X=x\right)  $ depended on
being able to construct a confidence interval for $\psi\left(  F\right)
=E\left[  \left\{  E\left(  Y|X\right)  \right\}  ^{2}\right]  .$ They
constructed an interval for $E\left[  \left\{  E\left(  Y|X\right)  \right\}
^{2}\right]  $ when the marginal distribution of $X$ was known. In this paper,
we construct a confidence interval for $E\left[  \left\{  E\left(  Y|X\right)
\right\}  ^{2}\right]  $ when the marginal of $X$ is unknown and, in Section
\ref{adaptive_section}, use it to obtain an adaptive confidence set for
$E\left(  Y|X=x\right)  $.

The functional $E\left\{  Cov\left(  Y,A|X\right)  \right\}  $ is the
functional $E\left\{  var\left(  Y|X\right)  \right\}  $ in the special case
in which $Y=A$ wp1. Minimax estimation of $var\left(  Y|X\right)  $ has
recently been discussed by Wang et al. (2006) and Cai et al. (2006) in the
setting of non-random $X$.

The function $\gamma\left(  x\right)  =E\left[  Y|A=1,X=x\right]  -E\left[
Y|A=0,X=x\right]  $ is the effect of treatment on the subgroup with $X=x.$ It
is important to estimate the function $\gamma\left(  x\right)  $, in addition
to the average treatment effect in the total population, because treatment
should be given, since beneficial, to those subjects with $\gamma\left(
x\right)  >0$ but withheld, since harmful, from subjects with $\gamma\left(
x\right)  <0.$ We show that one can obtain adaptive confidence sets for
$\gamma\left(  x\right)  $ if one can set confidence intervals for the
functional $\psi\left(  F\right)  =E\left[  \gamma\left(  X\right)
^{2}\right]  $. We construct intervals for $E\left[  \gamma\left(  X\right)
^{2}\right]  $ under the additional assumption that the data $O=\left(
Y,A,X\right)  $ came from a randomized trial. In a randomized trial, in
contrast to an observational study, the randomization probabilities, $P\left(
A=1|X\right)  =E\left(  A|X\right)  $ are known by design. We plan to report
confidence intervals for $E\left[  \gamma\left(  X\right)  ^{2}\right]  $ with
$E\left(  A|X\right)  $ unknown elsewhere.

All of the above functionals $\psi\left(  F\right)  $ have a positive
semiparametric information bound (SIB) and thus a (first order) efficient
influence function with a finite variance. In fact all the functionals
$\psi\left(  F\right)  $ have efficient influence function
\begin{equation}
IF\left(  b\left(  F\right)  ,p\left(  F\right)  ,\psi\left(  F\right)
\right)  \equiv if\left(  O,b\left(  X,F\right)  ,p\left(  X,F\right)
,\psi\left(  F\right)  \right) \label{iif}%
\end{equation}
where $b\left(  x,F\right)  ,p\left(  x,F\right)  $ are monotone functions of
certain conditional expectations, and, for any $b^{\ast}\left(  x\right)
,p^{\ast}\left(  x\right)  ,$%
\[
E_{F}\left[  IF\left(  b^{\ast},p^{\ast},\psi\left(  F\right)  \right)
\right]  =E_{F}\left[  h_{1}\left(  O\right)  \left\{  b^{\ast}\left(
X\right)  -b\left(  X;F\right)  \right\}  \left\{  p^{\ast}\left(  X\right)
-p\left(  X;F\right)  \right\}  \right]
\]
where $h_{1}\left(  O\right)  $ is a known function. \ We refer to functionals
in our class as doubly-robust to indicate that $IF\left(  b\left(  F\right)
,p\left(  F\right)  ,\psi\left(  F\right)  \right)  $ continues to have mean
zero when either (but not both) $p\left(  F\right)  $ is misspecified as
$p^{\ast}$ or $b\left(  F\right)  $ is misspecified as $b^{\ast}.$ \ The
functions $b\left(  x,F\right)  ,p\left(  x,F\right)  ,if\left(  O,b\left(
X,F\right)  ,p\left(  X,F\right)  ,\psi\left(  F\right)  \right)  ,$ and
$h_{1}\left(  O\right)  $ differ depending on the functional $\psi\left(
F\right)  $ of interest.

As the functionals $\psi\left(  F\right)  $ are all closely related, we shall
use $E\left\{  Cov\left(  Y,A|X\right)  \right\}  $ as a prototype in this
introduction. For $\psi\left(  F\right)  \equiv E\left\{  Cov\left(
Y,A|X\right)  \right\}  ,$ $b\left(  X;F\right)  =E_{F}\left(  Y|X\right)  ,$
$p\left(  X;F\right)  =E_{F}\left(  A|X\right)  ,$%
\[
IF\left(  b\left(  F\right)  ,p\left(  F\right)  ,\psi\left(  F\right)
\right)  =\left\{  Y-b\left(  X;F\right)  \right\}  \left\{  A-p\left(
X;F\right)  \right\}  -\psi\left(  F\right)  ,
\]
and $h_{1}\left(  O\right)  \equiv1$.

Whenever a functional $\psi\left(  F\right)  $ has a non-zero SIB, given
sufficiently stringent bounds on $L_{p}$ norms and on smoothness, it is
possible to use the estimated first order influence function to construct
regular estimators and honest asymptotic confidence intervals whose width
shrinks at the usual parametric rate of $n^{-1/2}$ $.$ We recall that, by
definition, regular estimators are $n^{1/2}$-consistent. When $X$ is high
dimensional, the apriori smoothness restrictions on $p\left(  X;F\right)  $
and $b\left(  X;F\right)  $ necessary for point or interval estimators of
$E\left\{  Cov\left(  Y,A|X\right)  \right\}  $ to achieve the parametric rate
of $n^{-1/2}$ are so severe as to be substantively implausible. As a
consequence, we replace the usual approach based on first order influence
functions by one based on higher order influence functions.

To provide quantitative results, we require a measure of the maximal possible
complexity (e.g. smoothness) of $p\left(  \cdot;F\right)  $ and $b\left(
\cdot;F\right)  $ believed substantively plausible.$\ $

$\ $We use H\"{o}lder balls for concreteness, although our methods extend to
other measures of complexity. A function $h\left(  \cdot\right)  $ lies in the
H\"{o}lder ball $H(\beta,C),$ with H\"{o}lder exponent $\beta>0$ and radius
$C>0,$ if and only if $h\left(  \cdot\right)  $ is bounded in supremum norm by
$C$ and all partial derivatives of $h(x)$ up to order $\left\lfloor
\beta\right\rfloor $ exist, and all partial derivatives of order $\left\lfloor
\beta\right\rfloor $ are Lipschitz with exponent $\left(  \beta-\left\lfloor
\beta\right\rfloor \right)  \ $and constant $C$. We make the assumption that
$b\left(  \cdot,F\right)  ,\ p\left(  \cdot,F\right)  \ $lie in given
H\"{o}lder balls $H(\beta_{b},C_{b}),$ $H(\beta_{p},C_{p}).$ Furthermore, it
turns out we must also make assumptions about the complexity of the function
$g\left(  X;F\right)  \equiv E_{F}\left[  h_{1}\left(  O\right)  |X\right]
f_{X}\left(  X\right)  ,$ which we take to lie in a given $H(\beta_{g}%
,C_{g}).$ For $\psi\left(  F\right)  =E\left\{  Cov\left(  Y,A|X\right)
\right\}  ,g\left(  X;F\right)  =f_{X}\left(  X\right)  $

Using higher order influence functions, we construct regular estimators and
honest (i.e uniform over our model) asymptotic confidence intervals for
functionals $\psi\left(  F\right)  $ in our class whose width shrinks at the
usual parametric rate of $n^{-1/2}$ whenever $\beta/d\equiv\frac{\beta
_{b}+\beta_{p}}{2}/d>1/4$ and $\beta_{g}>0.$ This result cannot be improved
on, since even when $g\left(  x\right)  $ is known apriori, $\beta/d>1/4$ is
necessary for a regular estimator to exist.

When $\beta/d\leq1/4$ and $g\left(  x\right)  $ is known apriori$,$ we have
shown using arguments similar to those of Birge and Massart (1995) that the
minimax rate of convergence for an estimator and minimax rate of shrinkage of
a confidence interval is $n^{-\frac{4\beta/d}{4\beta/d+1}}\geq n^{-\frac{1}%
{2}}.$ When $g\left(  x\right)  $ is unknown, we construct point and interval
estimators with this same rate of $n^{-\frac{4\beta/d}{4\beta/d+1}}$ whenever
\begin{equation}
\beta_{g}/d>\beta/d\frac{2\left(  \Delta+1\right)  \left(  1-4\beta/d\right)
}{\left(  \Delta+2\right)  \ \left(  1+4\beta/d\right)  -4\left(
\beta/d\right)  \left(  1-4\beta/d\right)  \left(  \ \Delta+1\right)
},\label{ief}%
\end{equation}
where $\Delta=\left\vert \frac{\beta_{p}}{\beta_{b}}-1\right\vert .$ For
example if $\Delta=0,\beta/d=1/8,$ we require $\beta_{g}/d\ $exceed $1/22$ to
achieve the rate $n^{-\frac{4\beta/d}{4\beta/d+1}}.$ When the previous
inequality does not hold and $\Delta=0,$ we have constructed, in a yet
unpublished paper, estimators that converge at rate
\begin{align}
&  \log\left(  n\right)  n^{-\frac{1}{2}+\frac{\beta_{g}/d}{1+2\beta_{g}%
/d}\frac{\left(  m^{\ast}+1\right)  ^{2}}{2\beta/d}},\text{ with}\label{bc}\\
m^{\ast}  &  \equiv\lceil\left(  \left[  \frac{\beta}{d}\left(  4\frac{\beta
}{d}+\left(  1-4\frac{\beta}{d}\right)  \frac{1+2\beta_{g}/d}{\beta_{g}%
/d}\right)  \right]  ^{1/2}-\left(  1+2\beta/d\right)  \right)  \rceil
.\nonumber
\end{align}
We conjecture that this rate is minimax, possibly up to log factors. In this
paper, however, the estimators we construct are inefficient when the previous
inequality fails to hold, converging at rates less than the conjectured
minimax rate of Eq $\left(  \ref{bc}\right)  $.

Let us return to the case where $Y=A$ {}{}{}{}{}{}{}{}{}${}{}{}{}${}{}{}{}%
{}${}${}${}${}{}${}$wp1. Then $\psi\left(  F\right)  =E\left\{  var\left(
Y|X\right)  \right\}  $ and $p\left(  \cdot\right)  =$ $b\left(  \cdot\right)
$ so $\Delta=0.$ Now, for fixed $\beta,$ Eq $\left(  \ref{bc}\right)  $
converges to $\log\left(  n\right)  n^{-2\beta/d}$ as $\beta_{g}\rightarrow0,$
which agrees (up to a log factor) with the minimax rate of $n^{-2\beta/d}$
given by Wang et al. (2006) and Cai et al. (2006) under the semiparametric
homoscedastic model $var\left(  Y|X\right)  =\sigma^{2}$ with equal-spaced
non-random $X.$ This result might suggest that $X\ $being random rather than
equal-spaced can result in faster rates of convergence only when the density
of $X$ has some smoothness, as quantified here by $\beta_{g}>0.$ But this
suggestion is not correct. Recall that we obtained the rate $\log\left(
n\right)  n^{-2\beta/d}$ for $\psi\left(  F\right)  =E\left\{  var\left(
Y|X\right)  \right\}  $ as $\beta_{g}\rightarrow0$ under a non-parametric
model. In section {\ref{minimax_section}}, we construct a simple estimator of
$\sigma^{2}$ under the homoscedastic model with $X$ random with unknown
density that, for $\beta/d<1/4$, $\beta<1,$ and without smoothness
restrictions on $f_{X}\left(  x\right)  $, converges at the rate
$n^{-\frac{4\beta/d}{4\beta/d+1}},$ which is faster than the equal-spaced
non-random minimax rate of $n^{-2\beta/d}.$

The paper is organized as follows. In Section 2, we define the higher order
(estimation) influence functions of a functional $\psi\left(  F\right)  $ for
$F$ contained in a model $\mathcal{M}$ and prove two fundamental theorems -
the extended information equality theorem and the efficient estimation
influence function theorem. Further, in the context of a parametric model
whose dimension increases with sample size, we outline why estimators based on
higher order influence can outperform those based on first order influence
functions in high-dimensional models. In Section 3, we introduce the class of
functionals we study in the remainder of the paper and describe their
importance in biostatistics. The theory of section 2, however, is not directly
applicable to these functionals because they have first order but not higher
order influence functions. We show that higher order influence functions fail
to exist precisely because the Dirac delta function is not an element of the
Hilbert space $L_{2}$ of square integrable functions. We describe two
approaches to overcoming this difficulty. The first approach is based on
approximating the Dirac delta function by a projection operator onto a
subspace of $L_{2}$ of dimension $k\left(  n\right)  ,$ where $k\left(
n\right)  $ can be as large as the square of the sample size $n.$ The second
approach is based on approximating the functional $\psi\left(  F\right)  $ by
a truncated functional $\widetilde{\psi}_{k\left(  n\right)  }\left(
F\right)  $. The truncated functional has influence functions of all orders,
is equal to $\psi\left(  F\right)  $ if either a $k\left(  n\right)  $
dimensional working parametric model (with $k\left(  n\right)  <n^{2})$ for
the function $b\left(  \cdot\right)  $ or the function $p\left(  \cdot\right)
$ in Eq. $\left(  {\ref{iif}}\right)  $ is correct, and remains close to
$\psi\left(  F\right)  $ even if both working models are misspecified. We then
use higher order influence function based estimators of $\widetilde{\psi
}_{k\left(  n\right)  }\left(  F\right)  $ as estimators of $\psi\left(
F\right)  $. These estimators $\widehat{\psi}_{m,k\left(  n\right)  }$ are
asymptotically normal with variance and bias for $\psi\left(  F\right)  $
depending both on the choice of the dimension $k\left(  n\right)  $ of the
working models and on the order $m$ of the influence function of
$\widetilde{\psi}_{k\left(  n\right)  }\left(  F\right)  .$ We show that these
same estimators $\widehat{\psi}_{m,k\left(  n\right)  }$ can also be obtained
under the approximate Dirac delta function approach. We derive the optimal
estimator $\widehat{\psi}_{m_{opt},k_{opt}\left(  n\right)  }\left(  \beta
_{b},\beta_{p},\beta_{g}\right)  $ in the class as a function of the
H\"{o}lder balls in which the functions $b,p,$ and $g$ are assumed to lie.
Finally we conclude section 3 by showing that the estimators $\widehat{\psi
}_{m,k\left(  n\right)  }$ have a multi-robustness property that extends the
double-robustness property of the first order influence function estimator
$\widehat{\psi}_{1}.$

In Section 4, we consider whether the estimators $\widehat{\psi}%
_{m_{opt},k_{opt}\left(  n\right)  }\left(  \beta_{b},\beta_{p},\beta
_{g}\right)  $ are rate-minimax. We show that whenever $\beta/d\equiv
\frac{\beta_{b}+\beta_{p}}{2}/d>1/4\ $and $\beta_{g}>0$,\newline%
$\widehat{\psi}_{m_{opt},k_{opt}\left(  n\right)  }\left(  \beta_{b},\beta
_{p},\beta_{g}\right)  $ is not only rate minimax but is semiparametric
efficient. Further, by letting the order $m=m\left(  n\right)  $ of the
U-statistic depend on sample size, we construct a single estimator
$\widehat{\psi}_{m\left(  n\right)  ,k\left(  n\right)  }$ that is
semiparametric efficient for all $\beta/d>1/4$ even when $g\left(
\cdot\right)  $ cannot be estimated at an algebraic rate. We show, however,
that when $\beta/d<1/4,\ \widehat{\psi}_{m_{opt},k_{opt}\left(  n\right)
}\left(  \beta_{b},\beta_{p},\beta_{g}\right)  $ does not in general converge
at the minimax rate. In Section 4.1, however, we construct a new estimator
$\widehat{\psi}_{\mathcal{K}_{J}}^{eff}\left(  \beta_{g},\beta_{b},\beta
_{p}\right)  $ that converges at the minimax rate of $n^{-\frac{4\beta
/d}{4\beta/d+1}}$ whenever Eq. $\left(  {\ref{ief}}\right)  $ holds. In
Section 5, we use the results obtained in earlier sections to construct
adaptive confidence intervals for a regression function $E\left[
Y|X=x\right]  $ when the marginal of $X$ is unknown and for the treatment
effect function and optimal treatment regime in a randomized clinical trial.
In Section 6.1, we discuss how to obtain higher order U-statistic point
estimators and confidence intervals for functionals $\tau\left(  F\right)  $
that are implicitly defined as the solution to an equation $\psi\left(
\tau,F\right)  =0.$ In Section 6.2, we define higher order testing influence
functions and efficient scores and describe their relationship to the higher
order estimation influence functions and efficient influence functions of
Section 2. Finally, in Section 6.3, we discuss the relationship between the
higher order U-statistic point estimators of an implicitly defined functional
$\tau\left(  F\right)  $ and higher order testing influence functions.

Before proceeding, several additional comments are in order. In this paper, we
investigate the asymptotic properties of our higher order U-statistic point
and interval estimators. The reader is referred to Li et al (2006) for an
investigation of the finite sample properties of our procedures through
simulation.\ \ In addition, precise regularity conditions are sometimes
omitted from both the statements and the proofs of various theorems. This
reflects the fact that the goal of this paper is to provide a broad overview
of our theory as it currently stands.

Different subject matter experts will clearly disagree as to the maximum
possible complexity of $p\left(  x;F\right)  $,$b\left(  x;F\right)  $ and
$g\left(  x;F\right)  .$ Thus it is important to have methods that adapt to
the actual smoothness of these functions. Elsewhere, we plan to provide point
estimators that optimally adapt to unknown smoothness. In contrast to point
estimators, however, for honest confidence intervals, the degree of possible
adaption to unknown smoothness is small. Therefore we propose that an analyst
should report a mapping from apriori smoothness assumptions encoded in the
exponents and radii of H\"{o}lder balls (or in other measures of complexity)
to the associated optimal $1-\alpha$ honest confidence intervals proposed in
this paper. Such a mapping is finally only useful if substantive experts can
approximately quantify their informal opinions concerning the shape and
wiggliness of $p,b,$and $g$ using the measure of complexity on offer by the
analyst. It is an open question which, if any, complexity measure is suitable
for this purpose.

Finally, most of our mathematical results concern rates of convergence. We
offer only a few results on the constants in front of those rates. This is not
because the constant is less important than the rate in predicting how a
proposed procedure will perform in the moderate sized samples occurring in
practice. Rather, at present, we do not possess the mathematical tools
necessary to obtain useful results concerning constants. A more extended
discussion of the issue is found in Section 3 of Li et al. (2006).

In the following, we use $X_{n}\asymp Y_{n}$ to mean $X_{n}=O_{p}\left(
Y_{n}\right)  $ and $Y_{n}=O_{p}\left(  X_{n}\right)  ;$ {}$X_{n}\sim Y_{n}$
to mean $\frac{X_{n}}{Y_{n}}\overset{P}{\rightarrow}1$ $;$ and $X_{n}\gg
Y_{n}$ $\left(  X_{n}\ll Y_{n}\right)  $ to respectively mean $\frac{Y_{n}%
}{X_{n}}\overset{P}{\rightarrow}0$ $\left(  \frac{X_{n}}{Y_{n}}%
\overset{P}{\rightarrow}0\right)  $ as $n\rightarrow\infty.$

\section{Theory of Higher Order Influence Functions\label{hoif}}

Given $n$ i.i.d observations $\mathbf{O\equiv O}_{n}\mathbf{\equiv}\left\{
O_{i},i=1,...n\right\}  $ from a model
\[
\mathcal{M}\left(  \Theta\right)  =\left\{  F(\cdot;\theta),\theta\in
\Theta\right\}  ,
\]
we consider inference on a nonlinear functional $\psi\left(  \theta\right)  .$
In general, $\psi\left(  \theta\right)  $ can be infinite dimensional but for
now we only consider the one dimensional case. In the following all quantities
can depend on the sample size $n,$ including the support of $O,$ the parameter
space $\Theta,$ and the functional $\psi\left(  \theta\right)  $. We generally
suppress the dependence on $n$ in the notation. We will be particularly
interested in models in which the parameter $\theta$ is infinite dimensional
and $\theta,\Theta,$ and $\psi\left(  \cdot\right)  $ do not depend on $n.$ We
also briefly discuss models in which subvectors of $\theta$ are
finite-dimensional parameters whose dimension $k\left(  n\right)  =n^{1+\rho}$
increases as power $1+\rho$ (often $\rho>0)$ of $n$ and thus $\theta
_{n},\Theta_{n},$ and $\psi_{n}\left(  \cdot\right)  $ depend on $n.$

Our first task is to define higher order influence functions. \ Before
proceeding we recall some facts about $U-$statistics. Consider a function
$b_{m}\left(  o_{1},o_{2},...,o_{m}\right)  \equiv b\left(  o_{1}%
,o_{2},...,o_{m}\right)  $ where we often suppress $b^{\prime}s$ subscript
$m.$ For integers $i_{1},i_{2}...,i_{m}\ $lying in $\left\{  1,...,n\right\}
,$ we define
\[
B_{m,i_{1},...,i_{m}}\mathbf{\equiv}b_{m}\left(  O_{i_{1}},O_{i_{2}%
},...,O_{i_{m}}\right)  \equiv b\left(  O_{i_{1}},O_{i_{2}},...,O_{i_{m}%
}\right)  .
\]
and%
\[
\mathbb{V}_{n}\left[  b_{m\ }\right]  \mathbf{\equiv}\frac{(n-m)!}{n!}%
\sum\limits_{i_{1}\neq i_{2}...\neq i_{m}}B_{m,i_{1},...,i_{m}}.
\]

In an abuse of notation, we will consider the following expressions to be
equivalent
\[
\mathbb{V}_{n}\left[  B_{m\ }\right]  \mathbf{\equiv}\mathbb{V}_{n}\left[
B_{m,,i_{1},...,i_{m}}\right]  \mathbf{\equiv}\mathbb{V}_{n}\left[
b_{m\ }\right]  .
\]
Thus $\mathbb{V}_{n}\left[  b_{m\ }\right]  $ is a $mth$ order U-statistic
with kernel $b_{m}\left(  o_{1},o_{2},...,o_{m}\right)  $. We do not assume
that $b_{m}\left(  o_{1},o_{2},...,o_{m}\right)  $ is symmetric. We will write
$\mathbb{V}_{n}\left[  B_{m}\right]  $ as $\mathbb{B}_{n,m}.$ So, suppressing
the dependence on $n,$ $\ \mathbb{B}_{m}\mathbf{\equiv}\mathbb{V}\left[
B_{m}\right]  $.

Any $\mathbb{B}_{m}\mathbb{\ }$has a unique (up to permutation) decomposition
$\mathbb{B}_{m}=\sum_{s=1}^{m}\mathbb{D}_{s}^{\left(  b\right)  }\left(
\theta\right)  $ under any $F(.;\theta)$ as a sum of degenerate U-statistics
$\mathbb{D}_{s}^{\left(  b\right)  }\left(  \theta\right)  ,$ where degeneracy
of $\mathbb{D}_{s}^{\left(  b\right)  }\left(  \theta\right)  $ means that
$D_{s}^{\left(  b\right)  }\left(  \theta\right)  =d_{s}^{\left(  b\right)
}\left(  O_{i_{1}},O_{i_{2}},...,O_{i_{s}};\theta\right)  $ satisfies
\[
E_{\theta}\left[  d_{s}^{\left(  b\right)  }\left(  o_{i_{1}},...o_{i_{l-1}%
},O_{i_{l}},o_{i_{l+1}}...,o_{i_{s}};\theta\right)  \right]  =0,l=1,....,s
\]
where upper and lower case letters, respectively, denote random variables and
their possible realizations.

Let $\mathcal{U}_{m}\left(  \theta\right)  $ be the Hilbert space of all
$U-$statistics of order $m$ with mean zero and finite variance with inner
product defined by covariances with respect to the $n$-fold product measure
$F^{n}\left(  \cdot;\theta\right)  .$ Note that any $U-$statistic
$\mathbb{B}_{s}$ of order $s$, $s<m,$ is also an mth order $U-$ statistic with
$\mathbb{D}_{l}^{\left(  b\right)  }\left(  \theta\right)  $ identically zero
for $m\geq l>s$ .

Since any two degenerate $U-$ statistics of different orders are uncorrelated,
the $\mathcal{U}_{m}\left(  \theta\right)  $-Hilbert space projection of
$\mathbb{B}_{m}$ on $\mathcal{U}_{l}\left(  \theta\right)  $ is $\sum
_{s=1}^{l}\mathbb{D}_{s}^{\left(  b\right)  }\left(  \theta\right)  $ for
$l<m.$ Thus a $U-$statistic $\mathbb{B}_{m}$ is degenerate $\Leftrightarrow$
$\mathbb{B}_{m}=$ $\mathbb{D}_{m}^{\left(  b\right)  }\left(  \theta\right)  $
$\Leftrightarrow$ $\Pi_{\theta}\left[  \mathbb{B}_{m}|\mathcal{U}_{m-1}\left(
\theta\right)  \right]  =0\Leftrightarrow\mathbb{B}_{m}\in\mathcal{U}%
_{m-1}\left(  \theta\right)  ^{\perp_{m,\theta}},\ $where $\Pi_{\theta}\left[
\cdot|\cdot\right]  \mathbf{\equiv}\Pi_{\theta,m}\left[  \cdot|\cdot\right]  $
is the projection operator of the Hilbert space $\mathcal{U}_{m}\left(
\theta\right)  $ (with the dependence on $m$ suppressed when no ambiguity can
arise) and, for any linear subspace $\mathcal{R}$ of $\mathcal{U}_{m}\left(
\theta\right)  ,$ $\mathcal{R}^{\perp_{m,\theta}}$ is its orthocomplement in
the Hilbert space $\mathcal{U}_{m}\left(  \theta\right)  .$ Given any
$\mathbb{B}_{m}=\mathbb{V}\left[  B_{m}\right]  ,$ $\mathbb{D}_{m}^{\left(
b\right)  }\left(  \theta\right)  $ is explicitly given by $\mathbb{V}\left[
d_{m,\theta}\left\{  B_{m}\right\}  \right]  $ where $d_{m,\theta}\ $\ maps
$\ B_{m\ }\equiv b\left(  O_{i_{1}},O_{i_{2}},..O_{i_{m}}\right)  $ to
\begin{gather}
d_{m,\theta}\left\{  B_{m}\right\}  =b\left(  O_{i_{1}},O_{i_{2}},..O_{i_{m}%
}\right) \label{deg}\\
+\sum_{t=0}^{m-1}\left(  -1\right)  ^{m-t}\sum_{i_{r_{1}}\neq i_{r_{2}}..\neq
i_{r_{t}}}E_{\theta}\left(  b\left(  O_{i_{1}},O_{i_{2}},..O_{i_{m}}\right)
|O_{i_{r_{1}}},O_{i_{r_{2}}},..O_{i_{r_{t}}}\right) \nonumber
\end{gather}

Given a function $g\left(  \zeta\right)  ,\zeta\mathbf{\equiv}\left\{
\zeta_{1},...,\zeta_{r}\right\}  ^{T},$ define for $m=0,1,2,...,$
\[
g_{\backslash\overline{l}_{m}}\left(  \zeta\right)  \mathbf{\equiv
}g_{\backslash l_{1},...l_{m}}\left(  \zeta\right)  \mathbf{\equiv}%
\frac{\partial^{m}g\left(  \zeta\right)  }{\partial\zeta_{l_{1}}%
...\partial\zeta_{l_{m}}}%
\]
with $l_{s}\in\left\{  1,...,r\right\}  $ where the $\backslash$ symbol
denotes differentiation by the variables occurring to its right and the
overbar $\overline{l}_{m}$ denotes the vector $\left(  l_{1},...l_{m}\right)
$. \ Given a sufficiently smooth $r-$dimensional parametric submodel
$\widetilde{\theta}\left(  \zeta\right)  $ mapping $\zeta\in R^{r}$
injectively into $\Theta$, define for $\theta$ in the range of
$\widetilde{\theta}\left(  \cdot\right)  ,$ $\psi_{\backslash\overline{l}_{m}%
}\left(  \theta\right)  \mathbf{\equiv}\left(  \psi\circ\widetilde{\theta
}\right)  _{\backslash l_{1},...l_{m}}\left(  \zeta\right)  |_{\zeta
=\widetilde{\theta}^{-1}\left\{  \theta\right\}  }$ and $f_{\backslash
\overline{l}_{m}}\left(  \mathbf{O}_{n};\theta\right)  \mathbf{\equiv}\left(
f\circ\widetilde{\theta}\right)  _{\backslash l_{1},...l_{m}}\left(
\zeta\right)  |_{\zeta=\widetilde{\theta}^{-1}\left\{  \theta\right\}  },$
where $f\left(  \mathbf{O}_{n}\mathbf{;}\theta\right)  \equiv\prod_{i}%
f(O_{i};\theta)$ is the density of $\mathbf{O}_{n}\ $wrt a dominating measure.
That is $\psi_{\backslash\overline{l}_{m}}\left(  \theta\right)  $ and
$f_{\backslash\overline{l}_{m}}\left(  \mathbf{O}_{n}\mathbf{,}\theta\right)
$ are higher order derivatives of $\psi\left(  \cdot\right)  \ $and $f\left(
\mathbf{O}_{n}\mathbf{;}\cdot\right)  $ under a parametric submodel
$\widetilde{\theta}\left(  \zeta\right)  ,$ where the model $\widetilde{\theta
}$ has been suppressed in the notation. \ An $sth$ order score associated with
the submodel $\widetilde{\theta}\left(  \zeta\right)  $ is defined to be
\[
\widetilde{\mathbb{S}}_{s,\overline{l}_{s}}\left(  \theta\right)  \equiv
f_{\backslash\overline{l}_{s}}\left(  \mathbf{O}_{n}\mathbf{;}\theta\right)
/f\left(  \mathbf{O}_{n}\mathbf{;}\theta\right)
\]
where $\widetilde{\mathbb{S}}_{s,\overline{l}_{s}}\left(  \theta\right)  $ is
a U-statistic of order $s.$ To understand why $\widetilde{\mathbb{S}%
}_{s,\overline{l}_{s}}\left(  \theta\right)  $ is a $U-$statistic we provide
formulae for an arbitrary score $\widetilde{\mathbb{S}}_{s,\overline{l}_{s}%
}\left(  \theta\right)  $ in terms of the subject specific scores
\[
S_{l_{1}...l_{m},j}\left(  \theta\right)  \mathbf{\equiv}f_{/l_{1}...l_{m}%
,j}\left(  O_{j};\theta\right)  /f_{j}\left(  O_{j};\theta\right)
\]
$j=1,...,n$ for $s=1,2,3.$ \ Suppressing the $\theta-$dependence, results in
Waterman and Lindsay (1996) imply%
\[
\widetilde{\mathbb{S}}_{1,\overline{l}_{1}}=\sum_{j}S_{l_{1},j}%
\]

\[
\widetilde{\mathbb{S}}_{2,\overline{l}_{2}}=\sum_{j}S_{l_{1}l_{2},j}%
+\sum_{s\neq j}S_{l_{1},j}S_{l_{2},s}%
\]

\[
\widetilde{\mathbb{S}}_{3,\overline{l}_{3}}=\sum_{j}S_{l_{1}l_{2}l_{3},j}%
+\sum_{s\neq j}S_{l_{1}l_{2},j}S_{l_{3},s}+S_{l_{3}l_{2},j}S_{l_{1}%
,s}+S_{l_{1}l_{3},j}S_{l_{2},s}+\sum_{s\neq j\neq t}S_{l_{1},j}S_{l_{2}%
,s}S_{l_{3},t}.
\]

Note these equations express each $\widetilde{\mathbb{S}}_{m,\overline{l}_{m}%
}$ as a sum of degenerate U-statistics. We now define a $mth$ order estimation
influence function $\mathbb{IF}_{m,\psi\left(  \cdot\right)  }\left(
\theta\right)  \mathbf{\equiv}\mathbb{IF}_{m,\psi}\left(  \theta\right)
\mathbf{\equiv}\mathbb{IF}_{m}\left(  \theta\right)  $ for $\psi\left(
\theta\right)  $ where we suppress the dependence on $\psi$ when no ambiguity
will arise.

\begin{definition}
A U-statistic $\mathbb{IF}_{m}\left(  \theta\right)  $ of order $m$ and finite
variance is said to be an $mth$ order estimation influence function for
$\psi\left(  \theta\right)  $ if (i) $E_{\theta}\left[  \mathbb{IF}_{m}\left(
\theta\right)  \right]  =0,$ $\theta\in\Theta$ and (ii) for $s=1,2,...,m$ and
every suitably smooth and regular (see Appendix) $r$ dimensional parametric
submodel $\widetilde{\theta}\left(  \zeta\right)  ,r=1,2,...m,$%
\[
\psi_{\backslash\overline{l}_{s}}\left(  \theta\right)  =E_{\theta}\left[
\mathbb{IF}_{m}\left(  \theta\right)  \widetilde{\mathbb{S}}_{s,\overline
{l}_{s}}\left(  \theta\right)  \right]  .
\]
\ 
\end{definition}

Estimation influence functions need not always exist, but when they do they
are useful for deriving point estimators of $\psi$ with small bias and for
deriving confidence interval estimators centered on an estimate of $\psi.$ We
will generally refer to estimation influence functions simply as influence
functions. We remark that $\mathbb{IF}_{m}\left(  \theta\right)  $ is an
influence function under the above definition if and only if it is one under
the modified version in which the dimension of the parametric submodel
$\widetilde{\theta}\left(  \zeta\right)  $ is unrestricted. A key result is
the following theorem which is related to results of Small and McLeish (1994).

\begin{theorem}
\label{eiet}\textbf{Extended Information Equality Theorem:} Given a $mth$
order influence function $\mathbb{IF}_{m}\left(  \theta\right)  ,$ for any
smooth and regular submodels $\widetilde{\theta}\left(  \zeta\right)  $ and
$s\leq m,$%
\[
\partial^{s}E_{\theta}\left[  \mathbb{IF}_{m}\left(  \widetilde{\theta}\left(
\zeta\right)  \right)  \right]  /\partial\zeta_{l_{1}}...\partial
\zeta_{l_{_{s}}}|_{\zeta=\widetilde{\theta}^{-1}\left\{  \theta\right\}
}=-\psi_{\backslash\overline{l}_{s}}\left(  \theta\right)
\]
Thus, if the functionals $E_{\theta}\left[  \mathbb{IF}_{m}\left(
\theta^{\ast}\right)  \right]  $ and $-\left[  \psi\left(  \theta^{\ast
}\right)  -\psi\left(  \theta\right)  \right]  $ have bounded Fr\'{e}chet
derivatives w.r.t. $\theta^{\ast}$ to order $m+1$ for a norm $\left\vert
\left\vert \cdot\right\vert \right\vert ,$
\[
E_{\theta}\left[  \mathbb{IF}_{m}\left(  \theta+\delta\theta\right)  \right]
=-\left[  \psi\left(  \theta+\delta\theta\ \right)  -\psi\left(
\theta\right)  \right]  +\text{ }O\left(  \left\vert \left\vert \delta
\theta\ \right\vert \right\vert ^{m+1}\right)
\]
since the functions $E_{\theta}\left[  \mathbb{IF}_{m}\left(  \theta^{\ast
}\right)  \right]  $ and $-\left[  \psi\left(  \theta^{\ast}\right)
-\psi\left(  \theta\right)  \right]  $ of $\theta^{\ast}$ have the same Taylor
expansion around $\theta$ up to order $m.$
\end{theorem}

The proof is in the Appendix. Define the $m$th order tangent space $\Gamma
_{m}\left(  \theta\right)  $ at $\theta$ for the model $\mathcal{M}\left(
\Theta\right)  $ to be the subspace of $\mathcal{U}_{m}\left(  \theta\right)
$ formed by taking the closed linear span of all scores of order $m$ or less
as we vary over all regular parametric submodels $\widetilde{\theta}\left(
\varsigma\right)  $ (whose range includes $\theta)$ of our model
$\mathcal{M}\left(  \Theta\right)  .$ We say a model is (locally)
nonparametric for $mth$ order inference if $\Gamma_{m}\left(  \theta\right)
=\mathcal{U}_{m}\left(  \theta\right)  .$

Given any $mth$ order estimation influence function $\mathbb{IF}_{m}\left(
\theta\right)  ,$ define the mth order efficient estimation influence function
to be
\[
\mathbb{IF}_{m}^{eff}\left(  \theta\right)  =\Pi_{\theta}\left[
\mathbb{IF}_{m}\left(  \theta\right)  |\Gamma_{m}\left(  \theta\right)
\right]
\]
$\ $where $\Pi_{\theta}\left[  \cdot|\cdot\right]  \mathbf{\equiv}\Pi
_{\theta,m}\left[  \cdot|\cdot\right]  $ is the $\mathcal{U}_{m}\left(
\theta\right)  -$projection operator. In the appendix, we prove the following:

\begin{theorem}
\label{eift}\textbf{Efficient Estimation Influence Function Theorem : }

\begin{enumerate}
\item $\mathbb{IF}_{m}^{eff}\left(  \theta\right)  $ is unique in the sense
that for any two mth order influence functions
\[
\Pi_{\theta}\left[  \mathbb{IF}_{m}^{\left(  1\right)  }\left(  \theta\right)
|\Gamma_{m}\left(  \theta\right)  \right]  =\Pi_{\theta}\left[  \mathbb{IF}%
_{m}^{\left(  2\right)  }\left(  \theta\right)  |\Gamma_{m}\left(
\theta\right)  \right]  \text{ a.s.}%
\]
\textbf{\ }

\item $\mathbb{IF}_{m}^{eff}\left(  \theta\right)  $ is a mth order estimation
influence function and has variance less than or equal to any other $mth$
order estimation influence function.

\item $\mathbb{IF}_{m}\left(  \theta\right)  $ is a $mth$ order estimation
influence function if and only if $\ $%
\[
\mathbb{IF}_{m}\left(  \theta\right)  \in\left\{  \mathbb{IF}_{m}^{eff}\left(
\theta\right)  +\mathbb{U}_{m}\left(  \theta\right)  ;\mathbb{U}_{m}\left(
\theta\right)  \in\Gamma_{m}^{^{\perp_{m,\theta}}}\left(  \theta\right)
\right\}
\]
where $\Gamma_{m}^{^{\perp_{m,\theta}}}\left(  \theta\right)  $ is the
ortho-complement of $\Gamma_{m}\left(  \theta\right)  \ $in $\mathcal{U}%
_{m}\left(  \theta\right)  .$

\item If $\mathbb{IF}_{m}\left(  \theta\right)  $ exists then $\mathbb{IF}%
_{s}^{eff}\left(  \theta\right)  \mathbb{\ }$exists for $s<m$ and $\Pi
_{\theta}\left[  \mathbb{IF}_{m}\left(  \theta\right)  |\Gamma_{s}\left(
\theta\right)  \right]  =\mathbb{IF}_{s}^{eff}\left(  \theta\right)  .$

\item If the model $\mathcal{M}\left(  \Theta\right)  $ is (locally)
nonparametric, then

\begin{enumerate}
\item there is at most one $m$th order estimation influence function
$\mathbb{IF}_{m}\left(  \theta\right)  $ for $\psi\left(  \theta\right)  ,$

\item
\[
\mathbb{IF}_{m}\left(  \theta\right)  =\mathbb{IF}_{m-1}\left(  \theta\right)
+\mathbb{IF}_{mm}\left(  \theta\right)
\]
where%
\[
\mathbb{IF}_{m-1}\left(  \theta\right)  =\Pi_{m,\theta}\left[  \mathbb{IF}%
_{m}\left(  \theta\right)  |\mathcal{U}_{m-1}\left(  \theta\right)  \right]
\]
and $\mathbb{IF}_{mm}\left(  \theta\right)  $ is a degenerate $mth$ order
U-statistic and thus
\[
E_{\theta}\left[  \mathbb{IF}_{m-1}\left(  \theta\right)  \mathbb{IF}%
_{mm}\left(  \theta\right)  \right]  =0.
\]

\item (i): Suppose, for a given $m\geq2,$ $\mathbb{IF}_{m-1}\left(
\theta\right)  $ exists and a kernel \newline$if_{m-1,m-1}\left(  o_{i_{1}%
},...,o_{i_{m-1}};\theta\right)  $ of $\mathbb{IF}_{m-1,m-1}\left(
\theta\right)  $ has a first order influence function with kernel
$if_{1,if_{m-1,m-1}\left(  o_{i_{1}},...,o_{i_{m-1}};\cdot\right)  \ }\left(
O_{i_{m}};\theta\right)  $ for all $o_{i_{1}},...,o_{i_{m-1}}$ in a set
$\mathcal{O}_{m-1}$ which has probability $1$ under $F^{\left(  m-1\right)
}\left(  \cdot,\theta\right)  .$ Then $\mathbb{IF}_{m}\left(  \theta\right)  $
exists and
\begin{equation}
m\mathbb{IF}_{m,m}\left(  \theta\right)  =\mathbb{V}\left(  d_{m,\theta
}\left[  if_{1,if_{m-1,m-1}\left(  O_{i_{1}},...,O_{i_{m-1}};\cdot\right)
\ {}}\left(  O_{i_{m}};\theta\right)  \right]  \right) \label{mm1}%
\end{equation}
where the operator $d_{m,\theta}$ is given in $Eq.\left(  \ref{deg}\right)  .$

(ii) Conversely, if $\mathbb{IF}_{m}$ exists then the symmetric kernel
\newline$if_{m-1,m-1}^{sym}\left(  o_{i_{1}},...,o_{i_{m-1}};\theta\right)  $
of $\mathbb{IF}_{m-1,m-1}\left(  \theta\right)  $ has a first order influence
function for all $o_{i_{1}},...,o_{i_{m-1}}$ in a set $\mathcal{O}_{m-1}$
which has probability $1$ under $F^{\left(  m-1\right)  }\left(  \cdot
,\theta\right)  .$ Further
\[
m^{-1}d_{m,\theta}\left[  if_{1,if_{m-1,m-1}^{sym}\left(  O_{i_{1}%
},...,O_{i_{m-1}};\cdot\right)  {}}\left(  O_{i_{m}};\theta\right)  \right]
=if_{m,m}^{sym}\left(  O_{i_{1}},...,O_{i_{m}};\theta\right)  .
\]

\end{enumerate}
\end{enumerate}
\end{theorem}

\begin{remark}
\smallskip Pfanzagl (1990) previously proved part 5.c(i) for $m=2.$ Our
theorem offers a generalization of his result. Note, in part (i) of 5(c), we
can always take the kernel to be the symmetric kernel.

\begin{remark}
\label{proj}Provided one knows how to calculate first order influence
functions, one can obtain $\mathbb{IF}_{2}\left(  \theta\right)
,...,\mathbb{IF}_{m}\left(  \theta\right)  $ recursively using part $\left(
5.c\right)  .$ An example of such a calculation is given in Section
\ref{dhoif_subsection} below. \ Thus part $\left(  5.c\right)  $ has the
interesting implication that even though higher order influence functions are
defined in terms of their inner products with higher order scores
$\widetilde{\mathbb{S}}_{m,\overline{l}_{m}},$ nevertheless, in (locally)
nonparametric models, one can derive all the higher order influence functions
of a functional $\psi\left(  \theta\right)  $ without even knowing how to
compute the scores $\widetilde{\mathbb{S}}_{m,\overline{l}_{m}}$ for any
$m>1.$ In fact, one need not even be aware of the structure of the scores
$\widetilde{\mathbb{S}}_{m,\overline{l}_{m}}$ in terms of the subject-specific
higher order scores $S_{l_{1}...l_{s},j}\left(  \theta\right)  .$ In contrast,
in parametric or semiparametric models whose tangent space $\Gamma_{m}\left(
\theta\right)  $ \ does not equal the set $\mathcal{U}_{m}\left(
\theta\right)  $ of all $m$th order $U-statistics$, one can often (but not
always) still obtain an inefficient influence $\mathbb{IF}_{m}\left(
\theta\right)  $ \ by applying part $\left(  5.c\right)  $ of the Theorem.
However, calculation of the efficient influence function $\mathbb{IF}%
_{m}^{eff}\left(  \theta\right)  =\Pi_{\theta}\left[  \mathbb{IF}_{m}\left(
\theta\right)  |\Gamma_{m}\left(  \theta\right)  \right]  $ by projection
generally requires explicit knowledge of the scores $\widetilde{\mathbb{S}%
}_{m,\overline{l}_{m}}\ $to derive $\Gamma_{m}\left(  \theta\right)  .$ For
this reason,\ it can be considerably more difficult to analyze certain
parametric models (with dimension increasing with sample size) than to analyze
(locally) nonparametric models. \ We will consider derivation of and
projections onto $\Gamma_{m}\left(  \theta\right)  $ in a forthcoming paper.
In the current paper, however, we do calculate $IF_{2}^{eff}\left(
\theta\right)  $ in one model that is not (locally) nonparametric so as to
provide some sense of the issues that arise. Specifically in Section
\ref{minimax_section}, we calculate $IF_{2}^{eff}\left(  \theta\right)  $ for
the functional $E\left[  \left\{  E\left[  Y|X\right]  \right\}  ^{2}\right]
$ in a model that assumes the marginal distribution of $X$ is known.
\end{remark}
\end{remark}

\begin{remark}
\textbf{\ \label{JJ1}:}\ \textbf{Implications of Theorem }$\left(
\ref{eift}\right)  $\textbf{\ for the Variance of Unbiased Estimators:}%
\ Suppose we have $n$ iid draws $\mathbf{O=}\left(  O_{1},...,O_{n}\right)  $
from $F(o;\theta),\theta\in\Theta,\ $and a U-statistic $\widehat{\psi}_{m}$ of
order $m\leq n$ with $var_{\theta}\left[  \widehat{\psi}_{m}\right]  <\infty$
for $\theta\in\Theta\ $satisfying $E_{\theta}\left[  \widehat{\psi}%
_{m}\right]  =\psi\left(  \theta\right)  $ for all $\theta\in\Theta.$ That is,
$\widehat{\psi}_{m}$ is unbiased for $\psi\left(  \theta\right)  $. We will
use Theorem $\left(  \ref{eift}\right)  $ to generalize a number of well-known
results on minimum variance unbiased estimation to arbitrary models. \ \ 

By $E_{\theta}\left[  \widehat{\psi}_{m}\right]  =\psi\left(  \theta\right)
,$ we immediately conclude that, viewing $\widehat{\psi}_{m}\ $as a kth order
U-statistic, $\widehat{\psi}_{m}-\psi\left(  \theta\right)  $ is a $k^{th}$
order estimation influence function for $\psi\left(  \theta\right)  $ for
$n\geq k\geq m.$ By Theorem $\left(  \ref{eift}\right)  ,$ $var_{\theta
}\left[  \widehat{\psi}_{m}\right]  \geq var_{\theta}\left[  \mathbb{IF}%
_{m}^{eff}\left(  \theta\right)  \right]  .$ We refer to $var_{\theta}\left[
\mathbb{IF}_{m}^{eff}\left(  \theta\right)  \right]  $ as the $mth$ order
Bhattacharyya variance bound at $\theta$ for the parameter $\psi\left(
\theta\right)  $ in model $\mathcal{M}\left(  \Theta\right)  ,$ as this
definition, in a precise analogy to Bickel et al. (1993)'s generalization of
the Cramer-Rao variance bound, generalizes Bhattacharyya's (1947) variance
bound to arbitrary semi- and non- parametric models. Indeed our 1st order
Bhattacharyya bound is precisely Bickel et al.'s (1993) generalization of the
Cramer-Rao variance bound.

We shall refer to an mth order U-statistic estimator $\widehat{\psi}_{m}$ as
mth order `unbiased locally efficient' at $\theta^{\ast}$ for $\psi\left(
\theta\right)  $ in model $\mathcal{M}\left(  \Theta\right)  $ if it is
unbiased for $\psi\left(  \theta\right)  $ under the model with variance at
$\theta^{\ast}$ equal to the $mth$ order Bhattacharyya bound at $\theta^{\ast
}.$ If $\widehat{\psi}_{m}$ is `unbiased locally efficient' at $\theta^{\ast}
$ for all $\theta^{\ast}\in\Theta,$ we say it is `unbiased globally
efficient'. By Theorem $\left(  \ref{eift}\right)  ,$ $var_{\theta}\left[
\mathbb{IF}_{k}^{eff}\left(  \theta\right)  \right]  \geq var_{\theta}\left[
\mathbb{IF}_{m}^{eff}\left(  \theta\right)  \right]  $ for $n\geq k>m.$ As a
consequence if an mth order `unbiased locally efficient' estimator
$\widehat{\psi}_{m,eff}$ exists at $\theta^{\ast}$ then, for $n\geq k\geq m,$
$\mathbb{IF}_{k}^{eff}\left(  \theta^{\ast}\right)  =\mathbb{IF}_{m}%
^{eff}\left(  \theta^{\ast}\right)  $ so the $mth$ and $kth$ order
Bhattacharyya bounds are equal at $\theta^{\ast}$ and $\widehat{\psi}_{m,eff}
$ is also kth order `unbiased locally efficient' at $\theta^{\ast}$.

From the fact that for, an unbiased estimator $\widehat{\psi}_{m},$
$\widehat{\psi}_{m}-\psi\left(  \theta\right)  $ is an $mth$ order influence
function, we conclude that the variance of $\widehat{\psi}_{m}$ attains the
the bound $var_{\theta^{\ast}}\left[  \mathbb{IF}_{m}^{eff}\left(
\theta^{\ast}\right)  \right]  $ at $\theta^{\ast}$ if and only if
$\widehat{\psi}_{m}-\psi\left(  \theta^{\ast}\right)  =\mathbb{IF}_{m}%
^{eff}\left(  \theta^{\ast}\right)  ,$ It follows that $\widehat{\psi}_{m}$ is
`unbiased globally efficient' if and only if $\widehat{\psi}_{m}-\psi\left(
\theta\right)  =\mathbb{IF}_{m}^{eff}\left(  \theta\right)  $ for all
$\theta\in\Theta.$ We thus have proved the following theorem in the
$\Rightarrow$ direction. The $\Leftarrow$ direction is immediate.
\end{remark}

\begin{theorem}
\label{global_eff}\textbf{:} In a model $\mathcal{M}\left(  \Theta\right)  ,$
there exists an mth order unbiased globally efficient U-statistic estimator of
$\psi\left(  \theta\right)  ,$ if and only if, for all $\theta\in\Theta,$
$\mathbb{IF}_{m}^{eff}\left(  \theta\right)  +\psi\left(  \theta\right)  $ is
a function $\widehat{\psi}_{m,eff}$ of the data $\mathbf{O,}$ not depending on
$\theta.$ In that case, $\widehat{\psi}_{m,eff}$ is the unique unbiased
globally efficient estimator.
\end{theorem}

In a locally nonparametric model all unbiased mth order estimators are
unbiased globally efficient, as there is a unique $mth$ order influence
function. For example, the usual unbiased estimator $\widehat{\sigma}^{2}%
=\sum_{i=1}^{n}\left\{  X_{i}-\sum_{j=1}^{n}X_{j}/n\right\}  ^{2}/\left(
n-1\right)  $ of the variance of a random variable $X$ is a second order
U-statistic and thus is a kth order unbiased globally efficient U-statistic
for $k\geq2\ $ in the locally nonparametric model consisting of all
distributions under which $\widehat{\sigma}^{2}$ has a finite variance.

In Section \ref{minimax_section} we use the results from this remark to
compare the relative efficiencies of competing rate-optimal unbiased
estimators in a model which is not locally nonparametric.

We now describe the main heuristic idea behind using higher order influence
functions. Technical details are suppressed. \ Consider the estimator
\begin{equation}
\widehat{\psi}_{m}=\psi\left(  \widehat{\theta}\right)  +\mathbb{IF}_{m}%
^{eff}\left(  \widehat{\theta}\right) \label{est}%
\end{equation}
based on a sample size $n,$ where $\widehat{\theta}$ is an initial rate
optimal estimator of $\theta$ from a separate independent training sample.
That is we assume that our actual sample size is $N$ and we randomly split the
$N$ observations into two samples: an analysis sample of size $n$ and a
training sample of size $N-n$ where $\left(  N-n\right)  /N=c^{\ast}$ $,$
$1>c^{\ast}\ >0.\ $We obtain our initial estimate $\widehat{\theta}$ from the
training sample data. Sample splitting has no effect on optimal rates of
convergence, although in the form described here does affect 'constants'.
Throughout the paper, we derive the properties of our estimators conditional
on the data in the training sample. In a later section, we describe how one
can sometimes obtain an optimal constant by choosing $\left(  N-n\right)
/N=N^{-\epsilon},\epsilon>0$ rather than $c^{\ast}.$

\begin{remark}
\label{semi} Note that sample splitting is avoided in most statistical
applications by using modern \textquotedblleft empirical process
theory\textquotedblright\ to prove that `plug-in' estimators such as
$\widehat{\psi}_{m}=\left\{  \psi\left(  \theta\right)  +\mathbb{IF}_{m}%
^{eff}\left(  \theta\right)  \right\}  _{\theta=\widehat{\theta}}$ that
estimate $\theta$ from the same sample used to calculate $\mathbb{IF}%
_{m}^{eff}\left(  \cdot\right)  $ have nice statistical properties. However
empirical process theory is not applicable in our setting because we are
interested in function classes whose size (entropy) is so large that they fail
to be Donsker. For this reason we initially believed that explicit sample
splitting would be difficult to avoid in our methodology. However, in Robins
et al. (2007), we describe a new method, more analogous to the jackknife than
to sample splitting, that effectively allows one to use all the data for
estimator construction.
\end{remark}

Expanding and evaluating conditionally on the training sample (or equivalently
on $\widehat{\theta}),$ we find by Theorem $\ref{eiet}$ that the conditional
bias
\[
E_{\theta}\left[  \widehat{\psi}_{m}-\psi\left(  \theta\right)
|\widehat{\theta}\right]  =\psi\left(  \widehat{\theta}\right)  -\psi\left(
\theta\right)  +E_{\theta}\left[  \mathbb{IF}_{m}^{eff}\left(  \widehat{\theta
}\right)  |\widehat{\theta}\right]
\]
is $O_{p}\left(  ||\widehat{\theta}-\theta||^{m+1}\right)  \ $which decreases
with $m\ $provided $||\widehat{\theta}-\theta||<1.$

In theorem \ref{3.19} below, we show that if
\[
sup_{o\in\mathcal{O}}\left\vert f\left(  o;\widehat{\theta}\right)  -f\left(
o;\theta\right)  \right\vert \rightarrow0
\]
as $||\widehat{\theta}-\theta||\rightarrow0\ ,$ where $f\left(  o;\theta
\right)  $ is the density of $O$ under $\theta$ and $\mathcal{O}$ has
probability one under all $\theta\in\Theta,$ then
\[
var_{\theta}\left[  \widehat{\psi}_{m}|\widehat{\theta}\right]  \equiv
var_{\theta}\left[  \mathbb{IF}_{m}^{eff}\left(  \widehat{\theta}\right)
|\widehat{\theta}\right]  =var_{\widehat{\theta}}\left[  \mathbb{IF}_{m}%
^{eff}\left(  \widehat{\theta}\right)  \right]  \left(  1+O_{p}%
(||\widehat{\theta}-\theta||)\right)
\]
\qquad

Now, by Theorem \ref{eift}, $var_{\widehat{\theta}}\left[  \mathbb{IF}%
_{m}^{eff}\left(  \widehat{\theta}\right)  \right]  $ increases with $m$.
Further, $var_{\widehat{\theta}}\left[  \mathbb{IF}_{1}^{eff}\left(
\widehat{\theta}\right)  \right]  \asymp1/n$, since, conditional on
$\widehat{\theta},$ $\mathbb{IF}_{1}^{eff}\left(  \widehat{\theta}\right)  $
is the sample average of $iid$ random variables.

To proceed further we shall need to be more explicit about the model
$\mathcal{M}\left(  \Theta\right)  .$ For now, we consider finite-dimensional
parametric models whose dimension $k\left(  n\right)  $ increases with sample
size. That is $\theta\equiv\theta_{n}$ depends on $n$ and the dimension of
$\Theta\equiv\Theta_{n}$ is $k\left(  n\right)  $. Suppose $k\left(  n\right)
\asymp n^{\gamma},\gamma\geq0.$ Let $\widehat{\theta}_{n}$ be the maximum
likelihood estimator of $\theta.$ If $k\left(  n\right)  $ increases slower
than the sample size (i.e., $\gamma<1),$ then, a) under regularity conditions,
$||\widehat{\theta}_{n}-\theta_{n}||=O_{p}\left(  \left\{  k\left(  n\right)
/n\right\}  ^{1/2}\right)  =O_{p}\left(  n^{-\frac{1}{2}\left(  1-\gamma
\right)  }\right)  $ with $||\cdot||$ the usual Euclidean norm in $R^{k\left(
n\right)  \text{ }};$ and b) $var_{\widehat{\theta}}\left[  \mathbb{IF}%
_{m}^{eff}\left(  \widehat{\theta}\right)  \right]  ,\ $although increasing
with $m,$ remains order $1/n;$ as a consequence, if $m\ $is chosen greater
than the solution $m^{\ast}$ to $n^{-\frac{m^{\ast}+1}{2}\left(
1-\gamma\right)  }=n^{-1/2},$ the bias of $\widehat{\psi}_{m}$ will be
$o_{p}\left(  n^{-1/2}\right)  ,$ the rate of convergence will be the usual
parametric rate of $n^{-1/2},$ and thus, for $n$ sufficiently large, the
squared bias of $\widehat{\psi}_{m}$ will be less than the variance. As a
consequence, as discussed in section \ref{DR_CI_Section}, we can construct
honest (i.e uniform over $\theta_{n}\in\Theta_{n})$ asymptotic confidence
intervals centered at $\widehat{\psi}_{m^{\ast}}$ with width of order
$n^{-1/2}.$ Here is a concrete example.

\textbf{Example:} Suppose $O=\left(  Y,X\right)  $ with $Y$ Bernouilli and
with $X\ $having a density with respect to the uniform measure $\mu\left(
\cdot\right)  $ on the unit cube $\left[  0,1\right]  ^{d}$ in\ $R^{d}.$
Suppose $\psi=E\left[  \left(  E\left[  Y|X\right]  \right)  ^{2}\right]  .$
Let $\left\{  z_{l}\left(  \cdot\right)  \right\}  \equiv\left\{  z_{l}\left(
x\right)  ;1,2,...\right\}  $ be a countable, linearly independent, sequence
of either spline, polynomial, or compact wavelet basis functions dense in
$L_{2}\left(  \mu\right)  .$ Set $\overline{z}_{k}\left(  x\right)  =\left(
z_{1}\left(  x\right)  ,...z_{k}\left(  x\right)  \right)  ^{T}.$ {}We assume
\[
E\left(  Y|X=x\right)  \in\left\{
\begin{array}
[c]{c}%
b\left(  x;\overline{\eta}_{k^{\ast}\left(  n\right)  }\right)  \equiv\left[
1+\exp\left(  -\overline{\eta}_{k^{\ast}\left(  n\right)  }^{T}\overline
{z}_{k^{\ast}\left(  n\right)  }\left(  x\right)  \right)  \right]  ^{-1};\\
\overline{\eta}_{k^{\ast}\left(  n\right)  }\in\mathcal{N}_{k^{\ast}\left(
n\right)  }%
\end{array}
\right\}  ,
\]

\[
f_{X}\left(  x\right)  \in\left\{
\begin{array}
[c]{c}%
f\left(  x;\overline{\omega}_{k^{\ast\ast}\left(  n\right)  }\right)  \equiv
c\left(  \overline{\omega}_{k^{\ast\ast}\left(  n\right)  }\right)
\exp\left[  \overline{\omega}_{k^{\ast\ast}\left(  n\right)  }^{T}\overline
{z}_{k^{\ast\ast}\left(  n\right)  }\left(  x\right)  \right]  ;\\
\overline{\omega}_{k^{\ast\ast}\left(  n\right)  }\in\mathcal{W}_{k^{\ast\ast
}\left(  n\right)  }%
\end{array}
\right\}
\]
where $c\left(  \overline{\omega}_{k^{\ast\ast}\left(  n\right)  }\right)  $
is a normalizing constant and $\mathcal{N}_{k^{\ast}\left(  n\right)  }$ and
$\mathcal{W}_{k^{\ast\ast}\left(  n\right)  }$ are open bounded subsets of
$R^{k^{\ast}\left(  n\right)  }$ and $R^{k^{\ast\ast}\left(  n\right)  }.$
Hence, $\Theta_{n}=\mathcal{N}_{k\left(  n\right)  }\times\mathcal{W}%
_{k\left(  n\right)  }$ has dimension $k\left(  n\right)  =k^{\ast}\left(
n\right)  +k^{\ast\ast}\left(  n\right)  \ $and $\psi\left(  \theta\right)
=\psi_{n}\left(  \theta_{n}\right)  =\int b^{2}\left(  x;\overline{\eta
}_{k^{\ast}\left(  n\right)  }\right)  f\left(  x;\overline{\omega}%
_{k^{\ast\ast}\left(  n\right)  }\right)  d\mu\left(  x\right)  .$

He (2000) and Portnoy (1988) prove that, under regularity conditions,
$||\widehat{\theta}_{n}-\theta_{n}||=O_{p}\left(  \left\{  k\left(  n\right)
/n\right\}  ^{1/2}\right)  $ when $k\left(  n\right)  =n^{\gamma}\ll n.$ Below
we shall see that $var_{\widehat{\theta}}\left[  \mathbb{IF}_{m}^{eff}\left(
\widehat{\theta}\right)  |\widehat{\theta}\right]  \asymp1/n$ for $n^{\gamma
}\ll n.$

Consider next models whose dimension $k\left(  n\right)  \asymp n^{\gamma}$
increases faster than $n$ (i.e., $\gamma>1).$ In such models, the MLE
$\widehat{\theta}_{n}$ is generally inconsistent and indeed there may exist no
consistent estimator of $\theta_{n}$. In that case, $||\widehat{\theta}%
_{n}-\theta_{n}||\ $fails to be $o_{p}\left(  1\right)  $ and the conditional
bias $E_{\theta}\left[  \widehat{\psi}_{m}-\psi\left(  \theta\right)
|\widehat{\theta}\right]  $ may not decrease with $m.$ In order to guarantee
consistent estimators of $\theta_{n}\ $exist$,$ it is necessary to place
further apriori restrictions on the complexity of $\Theta_{n}.$ Typical
examples of complexity-reducing assumptions would be an $\epsilon-$sparseness
assumption that only $k\left(  n\right)  ^{\epsilon},0<\epsilon<1,$ of the
$k\left(  n\right)  $ parameters are non-zero or a smoothness assumption that
specifies that the rate of decrease of the $j^{th}$component of $\theta_{n}$
is equal to $1/j$ raised to a given (positive) power. Even after imposing such
complexity-reducing assumptions, $\psi\left(  \theta\right)  \equiv\psi
_{n}\left(  \theta_{n}\right)  $ may not be estimable at rate $n^{-1/2}$.

For instance consider the previous example but now with $\gamma^{\ast}$ and
$\gamma^{\ast\ast}$ exceeding $1,$ so $k^{\ast\ast}\left(  n\right)
=n^{\gamma^{\ast\ast}}>>n$, $k^{\ast}\left(  n\right)  =n^{\gamma^{\ast}}>>n$
and $k\left(  n\right)  =k^{\ast\ast}\left(  n\right)  +k^{\ast}\left(
n\right)  \asymp n^{\gamma}>>n\ $with $\gamma=\max\left(  \gamma^{\ast\ast
},\gamma^{\ast\ast}\right)  .$ Consider the norms $\left\Vert \overline{\eta
}_{k^{\ast}\left(  n\right)  }\right\Vert $ =$\left\{  \int b^{2}\left(
x;\overline{\eta}_{k^{\ast}\left(  n\right)  }\right)  d\mu\left(  x\right)
\right\}  ^{1/2},\ $\newline$\left\Vert \overline{\omega}_{k^{\ast\ast}\left(
n\right)  }\right\Vert _{p}=\left\{  \int f\left(  x;\overline{\omega
}_{k^{\ast\ast}\left(  n\right)  }\right)  ^{p}d\mu\left(  x\right)  \right\}
^{1/p}$ and \newline$\left\Vert \theta\right\Vert _{p}=\left\Vert
\overline{\eta}_{k^{\ast}\left(  n\right)  }\right\Vert +\left\Vert
\overline{\omega}_{k^{\ast\ast}\left(  n\right)  }\right\Vert _{p}$. Suppose,
under a particular smoothness assumption, optimal rate estimators
$\widehat{\overline{\eta}}_{k^{\ast}\left(  n\right)  }$ and
$\widehat{\overline{\omega}}_{k^{\ast\ast}\left(  n\right)  }$ of
$\overline{\eta}_{k^{\ast}\left(  n\right)  }$ and $\overline{\omega}%
_{k^{\ast\ast}\left(  n\right)  }$ satisfy $\left\Vert \widehat{\overline
{\eta}}_{k^{\ast}\left(  n\right)  }-\overline{\eta}_{k^{\ast}\left(
n\right)  }\right\Vert =O_{p}\left(  n^{-\gamma_{\eta}}\right)  $ and
$\left\Vert \widehat{\overline{\omega}}_{k^{\ast\ast}\left(  n\right)
}-\overline{\omega}_{k^{\ast\ast}\left(  n\right)  }\right\Vert _{p}%
=O_{p}\left(  n^{-\gamma_{\omega}}\right)  \ $for some $\gamma_{\eta}%
>0,\gamma_{\omega}>0$ and all $p\geq2.$ Hence, $||\widehat{\theta}%
-\theta||_{p}=O_{p}\left(  \max\left\{  n^{-\gamma_{\eta}},n^{-\gamma_{\omega
}}\right\}  \right)  .$ For $\gamma>1,$ based on arguments given later, we
expect that $var_{\widehat{\theta}}\left[  \widehat{\psi}_{m}-\psi\left(
\theta\right)  |\widehat{\theta}\right]  \asymp\frac{n^{\left(  \gamma
-1\right)  \left(  m-1\right)  }}{n}$ and $E_{\theta}\left[  \widehat{\psi
}_{m}-\psi\left(  \theta\right)  |\widehat{\theta}\right]  =O_{p}\left(
\left\Vert \widehat{\overline{\eta}}_{k^{\ast}\left(  n\right)  }%
-\overline{\eta}_{k^{\ast}\left(  n\right)  }\right\Vert ^{2}\left\Vert
\widehat{\overline{\omega}}_{k^{\ast\ast}\left(  n\right)  }-\overline{\omega
}_{k^{\ast\ast}\left(  n\right)  }\right\Vert _{m-1}^{m-1}\right)
=O_{p}\left(  n^{-2\gamma_{\eta}-\left(  m-1\right)  \gamma_{\omega}}\right)
$

$=O_{p}\left(  ||\widehat{\theta}-\theta||_{m-1}^{m+1}\right)  .$

To find the estimator $\widehat{\psi}_{m_{_{best}}}$ in the class
$\widehat{\psi}_{m}$with optimal rate of convergence$,$ let $m^{\ast}%
=1+\frac{1-4\gamma_{\eta}}{\left(  \gamma-1\right)  +2\gamma_{\omega}}$ be the
value of $m$ that equates the order $n^{-4\gamma_{\eta}-2\left(  m-1\right)
\gamma_{\omega}}$ of the squared bias and the order $\frac{n^{\left(
\gamma-1\right)  \left(  m-1\right)  }}{n}$ of the variance. Then $m_{_{best}%
}=\left\lfloor m^{\ast}\right\rfloor $ if the order $n^{-4\gamma_{\eta
}-2\left(  m-1\right)  \gamma_{\omega}}+n^{\left(  \gamma-1\right)  \left(
m-1\right)  -1}$of the mean squared error$\ $at $\left\lfloor m^{\ast
}\right\rfloor $ is less than or equal to that at $\left\lceil m^{\ast
}\right\rceil .$ Otherwise, $m_{_{best}}=\left\lceil m^{\ast}\right\rceil $.
The rate of convergence of $\widehat{\psi}_{m_{_{best}}}$ will often be slower
than $n^{-1/2}.$ Note $m_{_{best}}=1$ whenever $\gamma>2,$ regardless of
$\gamma_{\eta}$ and $\gamma_{\omega}.$

By using the estimator $\widehat{\psi}_{\left\lceil m^{\ast}\right\rceil }$
rather than $\widehat{\psi}_{m_{_{best}}},$we can guarantee that the variance
asymptotically dominates bias and construct honest (i.e uniform over
$\theta_{n}\in\Theta_{n})$ asymptotic confidence intervals centered at
$\widehat{\psi}_{\left\lceil m^{\ast}\right\rceil }.$ Of course, the sample
size $n$ at which, for all $\theta_{n}\in\Theta_{n},$ the finite sample
coverage of the intervals discussed above is close to the asymptotic (i.e.
nominal) coverage is generally unknown and could be very large. For this
reason, a better, but unfortunately as yet technically out of reach, approach
to confidence interval construction is discussed in section
\ref{DR_CI_Section}.

In contrast to the case of parametric models of increasing dimension, in the
infinite dimensional models which we consider in the following section, the
functionals $\psi\left(  \theta\right)  $ of interest have first order
influence functions $\mathbb{IF}_{1}\left(  \theta\right)  $ but do not have
higher order influence functions. As a consequence, an initial 'truncation'
step is needed before we can apply the approach outlined in the preceding paragraph.

Finally, even in the case of parametric models with $k\left(  n\right)  >>n$
and complexity reducing assumptions imposed, , when the minimax rate for
estimation of $\psi\left(  \theta\right)  $ is slower than $n^{-1/2},$ the
optimal estimator $\widehat{\psi}_{m_{_{best}}}$ in the class $\widehat{\psi
}_{m}$ will generally not be rate minimax. See Section 3.2.6 and Sections
4.1.1 for additional discussion.

\section{Inference for a Class of Doubly Robust
Functionals\label{drfunctionals}:}

\subsection{The class of functionals:}

In this Section we consider models in which the parameter $\theta$ is infinite
dimensional and $\theta,\Theta,$ and $\psi\left(  \cdot\right)  $ do not
depend on $n.$ We make the following 4 assumptions $Ai)-Aiv)$:\ 

Ai) The data $O\ $includes a vector $X,\ $where, for all $\ \theta\in\Theta,$
the distribution of $X\ $ is supported on the unit cube $\left[  0,1\right]
^{d}$ ( or more generally a compact set) in\ $R^{d}\ $and has a density
$f\left(  x\right)  $ with respect to the Lebesgue measure. Further
$\Theta=\Theta_{1}\times\Theta_{2}$ where $\theta_{1}\in\Theta_{1}$ governs
the marginal law of $X$ and $\theta_{2}\in\Theta_{2}$ governs the conditional
distribution of $O|X.$

Aii ) The parameter $\theta_{2}$ contains components $b=b\left(  \cdot\right)
$ and $p=p\left(  \cdot\right)  $, $b:\left[  0,1\right]  ^{d}%
\mathcal{\rightarrow R}$ and $p:\left[  0,1\right]  ^{d}\mathcal{\rightarrow
R},$ such that the functional $\psi\left(  \theta\right)  $ of interest has a
first order influence function $\mathbb{IF}_{1,\psi}\left(  \theta\right)
=\mathbb{V}\left[  IF_{_{1,\psi\ }}\left(  \theta\right)  \right]  ,$ where
\begin{gather}
IF_{1,\psi\ }\left(  \theta\right)  =H\left(  b,p\right)  -\psi\left(
\theta\right)  ,\text{ }\label{Hdef}\\
\text{ with }H\left(  b,p\right)  \equiv h\left(  O,b\left(  X\right)
,p\left(  X\right)  \right) \nonumber\\
\equiv b\ \left(  X\right)  p\left(  X\right)  h_{1}\left(  O\right)
+b\left(  X\right)  h_{2}\left(  O\right)  +p\left(  X\right)  h_{3}\left(
O\right)  +h_{4}\left(  O\right) \\
\equiv BPH_{1}+BH_{2}+PH_{3}+H_{4},\nonumber
\end{gather}
and the known functions $h_{1}\left(  \cdot\right)  ,h_{2}\left(
\cdot\right)  ,h_{3}\left(  \cdot\right)  ,h_{4}\left(  \cdot\right)  \ $do
not depend on $\theta.$

Aiii )

a) $\Theta_{2b}\times\Theta_{2p}\subseteq\Theta_{2}$ where $\Theta_{2b} $ and
$\Theta_{2p}$ are the parameter spaces for the functions $b$ and $p$.
Furthermore the sets $\Theta_{2b}$ and $\Theta_{2p}$ are dense in
$L_{2}\left(  F_{X}\left(  x\right)  \right)  $ at each $\theta_{1}^{\ast}%
\in\Theta_{1}.$

or

b) $b^{\ast}\left(  \cdot\right)  =p^{\ast}\left(  \cdot\right)  $,
$h_{3}\left(  O\right)  =h_{2}\left(  O\right)  \ wp1,\ $and $\Theta
_{2b}\subseteq\Theta_{2}$ is dense in $L_{2}\left(  F_{X}\left(  x\right)
\right)  $ at each $\theta_{1}^{\ast}\in\Theta_{1}.$

\textbf{Remark: }Aiiib) can be viewed as a special case of Aiiia) as discussed
in Example 1a below, so we need only prove results under assumption Aiiia).

Assumptions Ai)-Aiii) have a number of important implications that we
summarize in a Theorem and two Lemmas.

\begin{theorem}
\label{DR}Double-Robustness: Assume Ai)-Aiii) hold, and recall $p\ $and $b$
are elements of $\theta.$ Then
\[
E_{\theta}\left[  H\left(  b\ ,p^{\ast}\right)  \right]  =E_{\theta}\left[
H\left(  b^{\ast},p\right)  \right]  =E_{\theta}\left[  H\left(  b,p\right)
\right]  =\psi\left(  \theta\right)
\]
for all $\left(  p^{\ast},b^{\ast}\right)  \in\Theta_{2p}\times\Theta
_{2b},\theta\in\Theta.$
\end{theorem}

\begin{proof}
: $E_{\theta}\left[  H\left(  b^{\ast},p\right)  \right]  -E_{\theta}\left[
H\left(  b,p\right)  \right]  =E_{\theta}\left[  \left\{  H_{1}p\left(
X\right)  +H_{2}\right\}  \left\{  b\left(  X\right)  -b^{\ast}\left(
X\right)  \right\}  \right]  $ and $E_{\theta}\left[  H\left(  b\ ,p^{\ast
}\right)  \right]  -E_{\theta}\left[  H\left(  b,p\right)  \right]
=E_{\theta}\left[  \left\{  H_{1}b\left(  X\right)  +H_{3}\right\}  \left\{
p\left(  X\right)  -p^{\ast}\left(  X\right)  \right\}  \right]  .$ \ The
theorem then follows from part 1) of the following lemma.
\end{proof}

Theorem \ref{DR} states that $H\left(  \cdot,\cdot\right)  $\ has mean
$\psi\left(  \theta\right)  $\ under $F\left(  \cdot;\theta\right)  $\ even
when $p $\ is misspecified as $p^{\ast}$\ or $b$\ is misspecified as $b^{\ast
}.$\ We refer to the functional $\psi\left(  \theta\right)  $\ as doubly
robust because of this property. The next lemma shows that $H\left(  b^{\ast
},p^{\ast}\right)  $\ is not unbiased if both $b$\ and $p$\ are simultaneously
misspecified. That is, $E_{\theta}\left[  H\left(  b^{\ast},p^{\ast}\right)
\right]  \neq\psi\left(  \theta\right)  .$

\begin{lemma}
\label{condE} Assume Ai)-Aiii) hold. Then

\begin{enumerate}
\item $E_{\theta}\left[  \left\{  H_{1}B+H_{3}\right\}  |X\right]  =E_{\theta
}\left[  \left\{  H_{1}P+H_{2}\right\}  |X\right]  =0$

\item $E_{\theta}\left[  H\left(  b^{\ast},p^{\ast}\right)  \right]
-E_{\theta}\left[  H\left(  b,p\right)  \right]  =E_{\theta}\left[  \left(
B-B^{\ast}\right)  \left(  P-P^{\ast}\right)  H_{1}\right]  $

and $\psi\left(  \theta\right)  \equiv E_{\theta}\left[  H\left(  b,p\right)
\right]  =E_{\theta}\left[  -BPH_{1}+H_{4}\right]  $
\end{enumerate}
\end{lemma}

\begin{proof}
Part 1): By assumptions $Ai)$ and $Aiiia)$ we have paths $\widetilde{\theta
}_{l}\left(  t\right)  ,l=1,2,...,$ in our model with $\widetilde{\theta}%
_{l}\left(  0\right)  =\theta$ and$\ p_{l}\left(  t\right)  =p_{l}\left(
x;t\right)  =p\left(  x\right)  +tc_{l}\left(  x\right)  ,b_{l}\left(
x;t\right)  =b\left(  x\right)  ,F_{l}\left(  x;t\right)  =F\left(  x\right)
$ for $l=1,2,..., $ where the sequence $c_{l}\left(  \cdot\right)  $ is dense
in $L_{2}\left[  F_{X}\left(  x\right)  \right]  .$ Let $S_{l}\left(
\theta\right)  \ $be the score for path $\widetilde{\theta}_{l}\left(
t\right)  $ at $t=0.$ Then by $\psi\left(  \widetilde{\theta}_{l}\left(
t\right)  \right)  =E_{\widetilde{\theta}_{l}\left(  t\right)  }\left[
H\left(  b,p_{l}\left(  t\right)  \right)  \right]  $
\begin{align*}
d\psi\left(  \widetilde{\theta}_{l}\left(  t\right)  \right)  /dt_{|t=0}  &
=E_{\theta}\left[  \left\{  H_{1}B+H_{3}\right\}  c_{l}\left(  X\right)
\right] \\
&  +E_{\theta}\left[  H\left(  b,p\right)  S_{l}\left(  \theta\right)
\right]
\end{align*}
By $\mathbb{IF}_{1,\psi}\left(  \theta\right)  $=$H\left(  b,p\right)
-\psi\left(  \theta\right)  ,$%
\[
d\psi\left(  \widetilde{\theta}_{l}\left(  t\right)  \right)  /dt_{|t=0}%
=E_{\theta}\left[  H\left(  b,p\right)  S_{l}\right]
\]
Thus $E\left[  \left\{  H_{1}B+H_{3}\right\}  c_{l}\left(  X\right)  \right]
=0.$ But $\left\{  c_{l}\left(  \cdot\right)  \right\}  $ is dense in
$L_{2}\left[  F_{0}\left(  X\right)  \right]  \ $\ so
\[
E\left[  H_{1}B+H_{3}|X\right]  =0
\]
An analogous argument proves $E_{\theta}\left[  \left\{  H_{1}P+H_{2}\right\}
|X\right]  =0.$ Part 2): $E_{\theta}\left[  H\left(  b^{\ast},p^{\ast}\right)
\right]  -E_{\theta}\left[  H\left(  b,p\right)  \right]  =$%
\begin{align*}
&  E_{\theta}\left[  \left(  B^{\ast}P^{\ast}-BP\right)  H_{1}+\left(
B^{\ast}-B\right)  H_{2}+\left(  P^{\ast}-P\right)  H_{3}\right] \\
&  =E_{\theta}\left[  \left(  B^{\ast}P^{\ast}-BP\right)  H_{1}-\left(
B^{\ast}-B\right)  PH_{1}-\left(  P^{\ast}-P\right)  BH_{1}\right] \\
&  =E_{\theta}\left[  \left(  B-B^{\ast}\right)  \left(  P-P^{\ast}\right)
H_{1}\right]
\end{align*}
where the second equality is by part 1). Choosing $P^{\ast}=B^{\ast}=0$ wp1
completes the proof of the theorem since then $E_{\theta}\left[  H\left(
b^{\ast},p^{\ast}\right)  \right]  =E_{\theta}\left[  H_{4}\right]  $.
\end{proof}

Below we will need the following partial converse to Lemma \ref{condE}.

\begin{lemma}
\label{IF1} Let $\Theta_{2b},\Theta_{2p},\Theta_{1}$ and $\Theta$ and
$H\left(  b,p\right)  $ be as defined in Ai)- Aiiia). Suppose that
\[
E_{\theta}\left[  \left\{  H_{1}B+H_{3}\right\}  |X\right]  =E_{\theta}\left[
\left\{  H_{1}P+H_{2}\right\}  |X\right]  =0
\]
and $\psi\left(  \theta\right)  =E_{\theta}\left[  H\left(  b,p\right)
\right]  .$ Then $\mathbb{V}\left[  H\left(  b,p\right)  -\psi\left(
\theta\right)  \right]  \ $is the first order influence function of
$\psi\left(  \theta\right)  .$
\end{lemma}

\begin{proof}
: The influence function of the functional $E_{\theta}\left[  H\left(
b^{\ast},p^{\ast}\right)  \right]  $ for known functions $b^{\ast},p^{\ast} $
is $\mathbb{V}\left[  H\left(  b^{\ast},p^{\ast}\right)  -E_{\theta}\left[
H\left(  b^{\ast},p^{\ast}\right)  \right]  \right]  .$ Thus by the linearity
of first order influence functions, the Lemma is true if and only if for each
$\theta_{0}\in\Theta,$ the functional $\tau\left(  b,p\right)  =E_{\theta_{0}%
}\left[  H\left(  b,p\right)  \right]  $ with $\theta_{0}$ fixed has influence
function equal to $0$ wp1 at $\left(  b,p\right)  =\left(  b_{0},p_{0}\right)
\subset\theta_{0}.$ That the influence function is equal to $0$ follows from
the fact that, under the assumptions of the Lemma, for sets $\left\{
c_{l}\left(  \cdot\right)  \right\}  $ and $\left\{  d_{l}\left(
\cdot\right)  \ \right\}  $ dense in $L_{2}\left[  F_{0}\left(  X\right)
\right]  $,
\begin{align*}
&  dE_{\theta_{0}}\left[  H\left(  b_{0}\left(  X\right)  +tc_{l}\left(
X\right)  ,p_{0}\left(  X\right)  +td_{l}\left(  X\right)  \right)  \right]
/dt_{|t=0}\\
&  =E_{\theta}\left[  \left\{  H_{1}b_{0}\left(  X\right)  +H_{3}\right\}
d_{l}\left(  X\right)  \right]  +E_{\theta}\left[  \left\{  H_{1}p_{0}\left(
X\right)  +H_{2}\right\}  c_{l}\left(  X\right)  \right]  =0
\end{align*}

\end{proof}

Results of Ritov and Bickel (1990) and Robins and Ritov (1997) imply it is not
possible to construct honest asymptotic confidence intervals for $\psi\left(
\theta\right)  $ whose width shrinks to $0$ as $n\rightarrow\infty$ if
$b\left(  \cdot\right)  \ $\ and $p\left(  \cdot\right)  $ are too
rough$.$\ Therefore we also place apriori bounds on their roughness. Our
bounds will be based on the following definition.

\begin{definition}
A function $h(\cdot)$ with domain $\left[  0,1\right]  ^{d}$ is said to belong
to a H\"{o}lder ball $H(\beta,C),$ with H\"{o}lder exponent $\beta>0 $ and
radius $C>0,$ if and only if $h\left(  \cdot\right)  $ is uniformly bounded by
$C$, all partial derivatives of $h(\cdot)$ up to order $\left\lfloor
\beta\right\rfloor $ exist and are bounded, and all partial derivatives
$\nabla^{\left\lfloor \beta\right\rfloor }$ of order $\left\lfloor
\beta\right\rfloor $ satisfy
\[
\sup_{x,x+\delta x\in\left[  0,1\right]  ^{d}}\left\vert \nabla^{\left\lfloor
\beta\right\rfloor }h(x+\delta x)-\nabla^{\left\lfloor \beta\right\rfloor
}h(x)\right\vert \leq C||\delta x||^{\beta-\left\lfloor \beta\right\rfloor }.
\]

\end{definition}

We note that the $L_{p},2<p<\infty$ and $L_{\infty}$ rates of convergence for
estimation of a marginal density or conditional expectation $h\left(
\cdot\right)  \in$ $H(\beta,C)$ are $O\left(  n^{-\frac{\beta}{2\beta+d}%
}\right)  $ and $O\left(  \left(  \frac{n}{\log n}\right)  ^{-\frac{\beta
}{2\beta+d}}\right)  $ respectively$.$ We refer to an estimator attaining
these rates as rate optimal.

Aiv) We assume $b\left(  \cdot\right)  ,\ p\left(  \cdot\right)  ,$ and
g$\left(  \cdot\right)  $ lie in given H\"{o}lder balls $H(\beta_{b},C_{b}),$
$H(\beta_{p},C_{p}),$ $H(\beta_{g},C_{g})$ where
\begin{equation}
g\left(  x\right)  \equiv E\left\{  H_{1}|X=x\right\}  f\left(  x\right)
\label{gee}%
\end{equation}

Furthermore we assume $g\left(  X\right)  >\sigma_{g}>0$ wp1. {}Finally we
assume, as can always be arranged by a suitable choice of estimator, that the
initial training sample estimators $\widehat{b}\left(  .\right)
,\widehat{p}\left(  .\right)  ,$ and $\widehat{g}\left(  \cdot\right)  $ are
rate optimal, have more than max$\left\{  \beta_{b},\beta_{g},\beta
_{p}\right\}  $ derivatives, and have $L_{\infty}$ norm bounded by a constant
$c_{\infty}.$ Further inf$_{x\in\left[  0,1\right]  ^{d}}$ $\widehat{g}\left(
x\right)  >\sigma_{g.}.$ The reason for the restrictions on $g\left(
\cdot\right)  $ will become clear below.

The restrictions $Ai)-Aiv)$ are the only restrictions common to all
functionals and models in the class. Additional model and/or functional
specific restrictions will be given below.

To motivate our interest in such a class of functionals and models we provide
a number of examples. In each case, one can use Lemma $\left(  \ref{IF1}%
\right)  $ to verify that the influence function of $\psi\left(
\theta\right)  $ is as given. All but examples 3 and 4 are examples of
(locally) nonparametric models.

\textbf{Example 1:} Suppose $O$=$\left(  A,Y,X\right)  \ $with $A$ and $Y$
univariate random variables.

\textbf{Example 1a:} \textbf{Expected Product of Conditional Expectations:
}Let\textbf{\ }$\psi\left(  \theta\right)  =E_{\theta}\left[  p\left(
X\ \right)  b\left(  X\ \right)  \right]  $ where $b\left(  X\ \right)
=E_{\theta}\left[  Y|X\right]  ,p\left(  X\ \right)  =E_{\theta}\left[
A|X\right]  .$ In this model
\begin{align*}
IF_{1,\psi}\left(  \theta\right)   &  =p\left(  X\ \right)  b\left(
X\ \right)  -\psi\left(  \theta\right) \\
&  +p\left(  X\ \right)  \left\{  Y-b\left(  X\right)  \right\}  +b\left(
X\ \right)  \left\{  A-p\left(  X\right)  \right\}
\end{align*}
so $H_{1}=-1,H_{2}=A,H_{3}=Y,H_{4}=0.$

We also consider the special case of this model where $A=Y$ with probability
one $\left(  w.p.1\right)  $. Then, as in assumption $Aiiib),$ $b\left(
X\ \right)  =p\left(  X\ \right)  $ $w.p.1,$ $H_{2}=H_{3}$ wp1. Then
$\psi\left(  \theta\right)  =E_{\theta}\left[  b^{2}\left(  X\ \right)
\right]  .$ In Section \ref{adaptive_section}, we show how our confidence
interval for $E_{\theta}\left[  b^{2}\left(  X\ \right)  \right]  $ can be
used to obtain an adaptive confidence interval for the regression function
$b\left(  \cdot\ \right)  $.

\textbf{Example 1b :} \textbf{Expected Conditional Covariance}
\[
\psi\left(  \theta\right)  =E_{\theta}\left[  AY\right]  \ -E_{\theta}\left[
p\left(  X\ \right)  b\left(  X\ \right)  \right]  =E_{\theta}\left[
Cov_{\theta}\left\{  Y,A|X\right\}  \right]
\]
has influence function
\[
AY-\left\{  p\left(  X\ \right)  b\left(  X\ \right)  +p\left(  X\ \right)
\left\{  Y-b\left(  X\right)  \right\}  +b\left(  X\ \right)  \left\{
A-p\left(  X\right)  \right\}  \right\}  -\psi\left(  \theta\right)
\]
so $H_{1}=\ 1,H_{2}=-A,H_{3}=-Y,H_{4}=AY$ .

The next example 1c shows that a confidence interval and point estimators for
$E_{\theta}\left[  Cov_{\theta}\left\{  Y,A|X\right\}  \right]  $ can be used
to obtain confidence intervals and point estimator for the variance weighted
average treatment effect in an observational study.

\textbf{Example 1c: }\textbf{Variance-weighted average treatment
effect:}\ Suppose, in an observational study, $O=\left\{  Y^{\ast
},A,X\right\}  $, $A$ is a binary treatment taking values in $\left\{
0,1\right\}  $, $Y^{\ast}$ is a univariate response and $X$ is a vector of
pretreatment covariates. Consider the parameter $\tau\left(  \theta\right)  $
given by:%
\begin{equation}
\tau\left(  \theta\right)  =\frac{E_{\theta}\left[  \ cov_{\theta}(Y^{\ast
},A|X)\right]  }{E_{\theta}\left[  \ var_{\theta}(A|X)\right]  }%
=\frac{E_{\theta}\left[  \ cov_{\theta}(Y^{\ast},A|X)\right]  }{E_{\theta
}\left[  \ \pi\left(  X\right)  \left\{  1-\pi\left(  X\right)  \right\}
\right]  },\text{ }%
\end{equation}
where $\pi\left(  X\right)  =pr\left(  A=1|X\right)  $ is often referred to as
the propensity score. We are interested in $\tau\left(  \theta\ \right)  $ for
several reasons. First, in the absence of confounding by unmeasured factors,
$\tau\left(  \theta\right)  $ is the variance-weighted average treatment
effect since $\tau\left(  \theta\right)  $ can be rewritten as $E_{\theta
}\left[  w_{\theta}(X)\gamma\left(  X;\theta\right)  \ \right]  $ where
$w_{\theta}(X)=\frac{var_{\theta}(A|X)}{E_{\theta}\left[  var_{\theta
}(A|X)\right]  }$ and
\[
\gamma\left(  x;\theta\right)  =E_{\theta}(Y^{\ast}|A=1,X=x)-E_{\theta
}(Y^{\ast}|A=0,X=x)
\]
is the average conditional treatment effect at level $x$ of the covariates.
Second, under the semiparametric model
\begin{equation}
\gamma\left(  X;\theta\right)  =\upsilon\left(  \theta\right)  \text{
}w.p.1\label{sr}%
\end{equation}
that assumes the treatment effect does not depend on $X,$ $\tau\left(
\theta\ \right)  =\upsilon\left(  \theta\right)  .$ However since the model
$\left(  \ref{sr}\right)  $ may not hold and therefore the parameter
$\upsilon\left(  \theta\right)  $ may be undefined, we choose to make
inference on $\tau\left(  \theta\ \right)  $ without imposing $\left(
\ref{sr}\right)  $. \ 

Now if for any $\tau\in R,$ we define $\psi\left(  \tau,\theta\right)  $ to
be
\[
\psi\left(  \tau,\theta\right)  =E_{\theta}\left[  \left\{  Y^{\ast}\left(
\tau\right)  -E_{\theta}\left(  Y^{\ast}\left(  \tau\right)  |X\right)
\right\}  \left\{  A-E_{\theta}\left(  A|X\right)  \right\}  \right]
\]
\ with $Y^{\ast}\left(  \tau\right)  =Y^{\ast}-\tau A,$ it is easy to verify
that $\tau\left(  \theta\right)  $ may also be characterized as the solution
$\tau=\tau\left(  \theta\right)  $ to the equation $\psi\left(  \tau
,\theta\right)  =0$. Thus inference on $\tau\left(  \theta\ \right)  $ is
easily obtained from inference on $\psi\left(  \tau,\theta\right)  .$ In
particular a $(1-\alpha)\ $confidence set for $\tau\left(  \theta\right)  $ is
the set of $\tau\ $such that a $(1-\alpha)\ $CI interval for $\psi\left(
\tau,\theta\right)  $ contains $0.\ $Therefore, with no loss of generality, we
consider the construction of a $(1-\alpha)\ $CI for $\psi\left(
\widetilde{\tau},\theta\right)  $ for a fixed value $\tau=\widetilde{\tau}$,
and write $Y=Y^{\ast}\left(  \widetilde{\tau}\right)  $ and $\psi\left(
\theta\right)  =\psi\left(  \widetilde{\tau},\theta\right)  .$ Thus
$\psi\left(  \theta\right)  =E_{\theta}\left[  Cov_{\theta}\left\{
Y,A|X\right\}  \right]  $ and we are in the setting of Example 1b.

In section \ref{testing_section}$,$ we show the rates at which the width of
the confidence sets for $\psi\left(  \widetilde{\tau},\theta\right)  $ and for
$\tau\left(  \theta\right)  $ shrink with $n$ are equal$.$

\textbf{Example 2a:} \textbf{Missing at Random: }Suppose $O=(AY,A\ ,X\ )$
where $Y$ is an outcome that is not always observed, $A\ $is the binary
missingness indicator, $X$ is a $d-$dimensional vector of always observed
continuous covariates, and let $b\left(  X\right)  =E(Y|A=1,X),\pi\left(
X\right)  =P(A=1|X)$ be the propensity score, and $p\left(  X\right)
=1/\pi\left(  X\right)  $. We suppose $\pi\left(  X\right)  >\sigma>0\ $and
define
\begin{equation}
\psi\left(  \theta\right)  =E_{\theta}\left[  \frac{AY}{\pi\left(  X\right)
}\right]  =E_{\theta}\left[  b\left(  X\right)  \right]
\end{equation}
Interest in $\psi\left(  \theta\right)  $ lies in the fact that $\psi\left(
\theta\right)  $ is the marginal mean of $Y$ under the missing (equivalently,
coarsening) at random (MAR) assumption that $P(A=1|X,Y)=\pi\left(  X\right)
.$ In this model $IF_{1,\psi}\left(  \theta\right)  =Ap\left(  X\right)
(Y-b(X))+b(X)-\psi\left(  \theta\right)  $ so $H_{1}=-A,H_{2}=1,H_{3}%
=AY,H_{4}=0$.

Note that if one has assumed apriori that $f_{X}\left(  \cdot\right)  $ and
$p\left(  X\right)  $ lay in H\"{o}lder balls with respective exponents
$\beta_{f_{X}}$ and $\beta_{p},$ then $\beta_{g}$ would be $\min\left(
\beta_{f_{X}}\ ,\beta_{p}\right)  ,$ since $g\left(  X\right)  =-f_{X}\left(
X\right)  /p\left(  X\right)  $.

\textbf{Example 2b:} \textbf{Missing Not-at Random:} Consider again the
setting of Example 2a but we no longer assume MAR. Rather we assume
\[
P(A=1|X,Y)=\left\{  1+\exp\left\{  -\left[  \gamma\left(  X\right)  +\alpha
Y\right]  \right\}  \right\}  ^{-1}%
\]
may depend on $Y$, where now $\gamma\left(  X\right)  $ is an unknown function
and $\alpha$ is a known constant (to be later varied in a sensitivity
analysis). In this case the marginal mean of $Y$ is given by $\psi\left(
\theta\right)  =E_{\theta}\left(  AY\left[  1+\exp\left\{  -\left[
\gamma\left(  X\right)  +\alpha Y\right]  \right\}  \right]  \right)  .$
Robins and Rotnitzky $\left(  2001\right)  $ proved this model places no
restrictions on $F\left(  o\right)  $ and derived
\[
IF_{1,\psi}\left(  \theta\right)  =A\left\{  1+\exp\left\{  -\alpha Y\right\}
p\left(  X\right)  \right\}  \left\{  Y-b\left(  X\ \right)  \right\}
+b\left(  X\right)  -\psi\left(  \theta\right)
\]
where, now,
\[
b\left(  X\ \right)  =E\left[  Y\exp\left\{  -\alpha Y\right\}  |A=1,X\right]
/E\left[  \exp\left\{  -\alpha Y\right\}  |A=1,X\right]
\]
$\ $ and $p\left(  X\right)  =\exp\left\{  -\gamma\left(  X\right)  \right\}
$. Thus
\[
H_{1}=-\exp\left\{  -\alpha Y\right\}  A,\ H_{2}=\left\{  1-A\right\}
,H_{3}=AY\exp\left\{  -\alpha Y\right\}  ,
\]
and $H_{4}=AY.$ When $\alpha=0$ this provides an alternate parametrization of
Example 2a.

\textbf{Example 3: Marginal Structural Models and The Average Treatment
Effect: }Consider the set-up of Example 1c including the non-identifiable
assumption of no unmeasured confounders, except now $A$ is discrete with
possibly many levels and $f\left(  a|X\right)  >\delta>0$ wp1. A marginal
structural model assumes $E_{f_{X}}\left\{  E_{\theta}(Y^{\ast}%
|A=a,X)\right\}  =d\left(  a,\upsilon\left(  \theta\right)  \right)  ,$ where
$d\left(  a,\upsilon\right)  $ is a known function and $\upsilon\left(
\theta\right)  $ is an unknown vector parameter of dimension $d^{\ast}.$ When
$A$ is dichotomous with $a\in\left\{  0,1\right\}  $ and $d\left(
a,\upsilon\right)  =\upsilon_{1}+\upsilon_{2}a$, then $\upsilon_{2}\left(
\theta\right)  $ is the average treatment effect parameter. Let $f^{\ast
}\left(  a\right)  $ be any density with the same support as $A$ and let
$s^{\ast}\left(  a\right)  $ be a $d^{\ast}$-vector function, both chosen by
the analyst. Then $\upsilon\left(  \theta\right)  \ $is identified as the
(assumed) unique value of $\upsilon$ satisfying
\[
\psi_{\upsilon}\left(  \theta\right)  \equiv E_{\theta}\left[  s\left(
O,A,\upsilon\right)  \frac{f^{\ast}\left(  A\right)  }{f\left(  A|X\right)
}\right]  =0,
\]
where $s\left(  O,a,\upsilon\right)  =\left\{  Y^{\ast}-d\left(
a,\upsilon\right)  \right\}  s^{\ast}\left(  a\right)  $. Thus a
$(1-\alpha)\ $confidence set for $\upsilon\left(  \theta\right)  $ is the set
of vectors $\upsilon\ $such that a $(1-\alpha)\ $CI for $\psi_{\upsilon
}\left(  \theta\right)  $ contains $0.\ $Therefore, with no loss of
generality, we consider the construction of a $(1-\alpha)\ $CI for the
$d-$vector functional $\psi\left(  \theta\right)  \equiv\psi
_{\widetilde{\upsilon}}\left(  \theta\right)  $ for a fixed value
$\widetilde{\upsilon}$ and define $h\left(  O,A\ \right)  \equiv s\left(
O,a,\widetilde{\upsilon}\right)  $ and $b\left(  a,X\ \right)  \equiv
E_{\theta}\left[  h\left(  O,a\ \right)  |A=a,X\right]  .$ Then $\psi
_{\widetilde{\upsilon}}\left(  \theta\right)  $ has influence function
\[
\ IF_{1}\left(  \theta\right)  =\frac{f^{\ast}\left(  A\right)  }{f\left(
A|X\right)  }\left\{  h\left(  O,A\ \right)  -b\left(  A,X\ \right)  \right\}
+\int b\left(  a,X\ \right)  dF^{\ast}\left(  a\right)  -\psi\left(
\theta\right)  .
\]
Next define $p\left(  a,X\right)  =1/f\left(  a|X\right)  ,\psi\left(
\theta,a\right)  =E_{f_{X}}\left[  b\left(  a,X\ \right)  \right]  $. Then
$IF_{1}\left(  \theta\right)  $ is the\ integral
\begin{align*}
\ IF_{1}\left(  \theta\right)   &  =\int dF^{\ast}\left(  a\right)
IF_{1}\left(  a,\theta\right)  ,\\
IF_{1}\left(  a,\theta\right)   &  =H_{1}\left(  a\right)  p\left(
a,X\right)  b\left(  a,X\ \right) \\
&  +H_{2}\left(  a\right)  b\left(  a,X\ \right)  +H_{3}\left(  a\right)
p\left(  a,X\ \right)  -\psi\left(  \theta,a\right)  ,\\
H_{1}\left(  a\right)   &  =-I\left(  A=a\right)  ,H_{2}\left(  a\right)
=1,H_{3}\left(  a\right)  =I\left(  A=a\right)  h\left(  O,a\ \right)  .\
\end{align*}

It follows that $IF_{1}\left(  \theta\right)  $ is a integral over
$a\in\emph{A}$ of influence functions $IF_{1}\left(  a,\theta\right)  $ for
parameters $\psi\left(  \theta,a\right)  $ in our class with $H_{4}\left(
a\right)  =0$. Thus we can estimate $\psi\left(  \theta\right)  $ by $\int
dF^{\ast}\left(  a\right)  \widehat{\psi}\left(  a\right)  ,$ where
$\widehat{\psi}\left(  a\right)  $ is an estimator of $\psi\left(
\theta,a\right)  ,$ If the support of $A$ is of greater cardinality than
$d^{\ast},$ the model is not locally nonparametric. Different choices for
$s^{\ast}\left(  a\right)  $ and $f^{\ast}\left(  a\right)  $ for which
$\left\{  \partial/\partial\upsilon^{T}\right\}  E_{\theta}\left[  s\left(
O,A,\upsilon\right)  \frac{f^{\ast}\left(  A\right)  }{f\left(  A|X\right)
}\right]  $ is invertible may result in difference influence functions. All
yield the same rate of convergence, although the constants differ. See Remark
\ref{proj} above. Extension of our methods to continuous $A$ will be treated elsewhere.

\textbf{Example 4: Confidence Intervals for The Optimal Treatment Strategy in
a Randomized Clinical Trial:} Consider a randomized clinical trial with data
$O=\left\{  Y,Y^{\ast},A,X\right\}  $, $A$ is a binary treatment taking values
in $\left\{  0,1\right\}  $, $Y^{\ast}$ and $Y$ are univariate responses, $X$
is a vector of pretreatment covariates. In a randomized trial, the
randomization probabilities $\pi_{0}\left(  X\right)  =P\left(  A=1|X\right)
$ are known by design. Let $b\left(  x\right)  =E_{\theta}(Y^{\ast
}|A=1,X=x)-E_{\theta}(Y^{\ast}|A=0,X=x\ )$ and $p\left(  x\right)  =E_{\theta
}(Y\ |A=1,X=x)-E_{\theta}(Y\ |A=0,X=x\ )$ be the average treatment effects at
level $X=x$ on $Y^{\ast}$ and $Y.$ We assume $Y\ $\ and $Y^{\ast}$ have been
coded so that positive treatment effects are desirable. Let $\psi\left(
\theta\right)  =E\left[  b\left(  X\right)  p\left(  X\right)  \right]  .$
Because the model is not locally nonparametric there exists more than a single
first order influence function. Indeed, for any given function $c\left(
\cdot\right)  ,$%

\begin{align*}
IF_{1,\psi}\left(  \theta,c\right)  =  &  b\left(  X\right)  p\left(
X\right)  -\psi\left(  \theta\right)  +\left[  b\left(  X\right)  \left\{
Y-Ap\left(  X\right)  \right\}  +p\left(  X\right)  \left\{  Y^{\ast
}-Ab\left(  X\right)  \right\}  \right] \\
&  \times\left\{  A-\pi_{0}\left(  X\right)  \right\}  \sigma_{0}^{-2}\left(
X\right)  +c\left(  X\right)  \left\{  A-\pi_{0}\left(  X\right)  \right\}
\end{align*}
with $\sigma_{0}^{2}\left(  X\right)  =\pi_{0}\left(  X\right)  \left\{
1-\pi_{0}\left(  X\right)  \right\}  $ is an influence function in our class
[provided it is square integrable] with $H_{1}=1-2A\left\{  A-\pi_{0}\left(
X\right)  \right\}  \sigma_{0}^{-2}\left(  X\right)  ,$ \newline%
$H_{2}=\left\{  A-\pi_{0}\left(  X\right)  \right\}  \sigma_{0}^{-2}\left(
X\right)  Y,$ $H_{3}=\left\{  A-\pi_{0}\left(  X\right)  \right\}  \sigma
_{0}^{-2}\left(  X\right)  Y^{\ast},$ \newline$H_{4}=c\left(  X\right)
\left\{  A-\pi_{0}\left(  X\right)  \right\}  $. As $c\left(  \cdot\right)  $
is varied, one obtains all first order influence functions. We do not discuss
the efficient choice of $c\left(  \cdot\right)  $ in this paper.

Our interest lies in the special case where $Y=Y^{\ast}$ wp1 (so there is but
one response of interest) and thus, as in assumption $Aiiib),$ $b=p,H_{2}%
=H_{3}\ $and we construct confidence interval for $\psi\left(  \theta\right)
=E\left[  b^{2}\left(  X\right)  \right]  .$ In Section \ref{adaptive_section}
we describe how we can use a confidence interval for $\psi\left(
\theta\right)  =E\left[  b^{2}\left(  X\right)  \right]  $ to obtain
confidence intervals for the treatment effect function $b\left(  x\right)  $
and, most importantly, for the optimal treatment strategy $d_{opt}\left(
x\right)  =I\left[  b\left(  x\right)  >0\right]  $ under which a subject with
covariate value $x$ is treated if and only if the treatment effect $b\left(
x\right)  $ is positive ( i.e., $d_{opt}\left(  x\right)  =1)$.

\subsection{ Higher Order Influence Functions for Our Model
:\label{DRHOIsection}}

\subsubsection{Dirac Kernels, Truncation Bias, and A Truncated Parameter}

In all of our examples the functions $p\left(  \cdot\right)  $ and $b\left(
\cdot\right)  $ are functions of conditional expectations given the continuous
random variable $X.$ \ It is well known that the associated point-evaluation
functional $p\left(  x\right)  $ and $b\left(  x\right)  $ do not have first
order influence functions. It then follows from part 5c of Theorem \ref{eift}
and the dependence of $\mathbb{IF}_{1,\psi}\left(  \theta\right)
=\mathbb{V}\left[  if_{1,\psi}\left(  O_{i_{1}};\theta\right)  \right]  $ on
$b\left(  \cdot\right)  $ and $p\left(  \cdot\right)  $ evaluated at the point
$X$ that, in none of our examples, does $\psi\left(  \theta\right)  $ have a
second (or higher) order influence function.

As a precise understanding of the reason for the nonexistence of higher order
influence functions for $\psi\left(  \theta\right)  $ is fundamental to our
approach, we now use part 5c of Theorem \ref{eift} to prove that
$\mathbb{IF}_{2,\psi}\left(  \theta\right)  $ does not exist by showing that
the functional $if_{1,\psi}\left(  o;\theta\right)  $ does not have a first
order influence function $\mathbb{V}\left[  if_{1,if_{1,\psi}\left(
o;\cdot\right)  \ }\left(  O;\theta\right)  \right]  .$ In this proof , we do
not assume that $b\left(  \cdot\right)  $ and $p\left(  \cdot\right)  $ are
functions of conditional expectations. Rather we only assume that our
functional satisfies Assumptions A(i-iv). Let $F_{X}$ and $f_{X}=f_{X}\left(
\cdot\right)  $ denote the marginal CDF and density of $X.$

Consider paths (parametric submodels) $\widetilde{\theta}_{l}\left(  t\right)
$ such that $\widetilde{\theta}_{l}\left(  0\right)  =\theta\ $satisfying
\begin{align*}
p_{l}\left(  t\ \right)   &  \equiv p_{l}\left(  x,t\ \right)  \equiv p\left(
x\right)  +tc_{l}\left(  x\right)  ,\\
b_{l}\left(  t\ \right)   &  \equiv b_{l}\left(  x,t\ \right)  \equiv b\left(
x\right)  +ta_{l}\left(  x\right)  ,
\end{align*}
where the sequences $c_{l}\left(  \cdot\right)  $ and $a_{l}\left(
\cdot\right)  ,l=1,2,...,$ are each dense in $L_{2}\left[  F_{X}\left(
x\right)  \right]  .$ Let
\[
s_{l}\left(  O;\theta\right)  =s_{l}\left(  O|X;\theta\right)  +s_{l}\left(
X;\theta\right)  ,
\]
$s_{l}\left(  O|X;\theta\right)  ,$ and $s_{l}\left(  X;\theta\right)  $
denote the overall, conditional, and marginal scores%
\[
\partial lnf\left(  O;\mathbb{\widetilde{\mathbb{\theta}}}_{l}\left(
0\right)  \right)  /\partial t,\text{ }\partial lnf\left(
O|X;\mathbb{\widetilde{\mathbb{\theta}}}_{l}\left(  0\right)  \right)
/\partial t,\text{ }\partial lnf_{X}\left(
X;\mathbb{\widetilde{\mathbb{\theta}}}_{l}\left(  0\right)  \right)  /\partial
t
\]

By linearity, $if_{1,\psi}\left(  o;\theta\right)  $ has an influence function
only if the functionals $b\left(  x\right)  $ and $p\left(  x\right)  $ have
one as well. Now by differentiating the identity
\[E_{\mathbb{\widetilde{\mathbb{\theta}}}_{l}\left(  t\right)  }\left[  \left\{
H_{1}b_{l}\left(  X,t\ \right)  +H_{3}\right\}  |X=x\right]  =0
\]
wrt to $t$ and evaluating at $t=0,$ we have
\[
-E_{\theta}\left[  \left\{  \left\{  H_{1}b\left(  X\ \right)  +H_{3}\right\}
\right\}  s_{l}\left(  O|X\right)  |X=x\right]  =E_{\mathbb{\theta}}\left[
H_{1}|X=x\right]  a_{l}\left(  x\right)
\]
However, by definition, $b\left(  x\right)  $ has an influence function
$\mathbb{V}\left[  if_{1,b\left(  x\right)  }\left(  O;\theta\right)  \right]
$ at $\theta\ $only if for $l=1,2,...,$ both $\partial b_{l}\left(
x,t\ \right)  /\partial t_{|t=0}=a_{l}\left(  x\right)  $ equals $E_{\theta
}\left[  if_{1,b\left(  x\right)  \ }\left(  O;\theta\right)  s_{l}\left(
O;\theta\right)  \right]  $ and $E_{\theta}\left[  if_{1,b\left(  x\right)
\ }\left(  O;\theta\right)  \right]  =0.$ Thus if $if_{1,b\left(  x\right)
}\left(  O;\theta\right)  $ exists, it must satisfy
\begin{align*}
&  -E_{\theta}\left[  \left\{  H_{1}b\left(  X\ \right)  +H_{3}\right\}
s_{l}\left(  O|X\right)  |X=x\right] \\
&  =E_{\theta}\left[  H_{1}|X=x\right]  E_{\theta}\left[  if_{1,b\left(
x\right)  \ }\left(  O;\theta\right)  s_{l}\left(  O;\theta\right)  \right]
\end{align*}
Without loss of generality, suppose $H_{1}\geq0$ wp 1$.$ Now if we could find
a 'kernel' $K_{f_{X},\infty}\left(  x,X\right)  $ such that%
\begin{align}
r\left(  x\right)   &  =E_{f_{X}}\left[  K_{f_{X},\infty}\left(  x,X\ \right)
r\left(  X\ \right)  \right] \nonumber\\
&  \equiv\int K_{f_{X},\infty}\left(  x,x^{\ast}\right)  r\left(  x^{\ast
}\right)  f_{X}\left(  x^{\ast}\right)  dx^{\ast}\text{ for all }r\left(
\cdot\right)  \in L_{2}\left(  F_{X}\right) \label{dirac2}%
\end{align}
then
\[
if_{1,b\left(  x\right)  \ }\left(  O\ ;\theta\right)  \equiv-\left[
\begin{array}
[c]{c}%
\left\{  E_{\theta}\left[  H_{1}|X=x\right]  \right\}  ^{-1/2}K_{f_{X},\infty
}\left(  x,X\ \right) \\
\times\left\{  E_{\theta}\left[  H_{1}|X\right]  \right\}  ^{-1/2}\left\{
H_{1}b\left(  X\ \right)  +H_{3}\right\}
\end{array}
\right]
\]
would be an influence function since
\begin{align*}
&  E_{\theta}\left[  H_{1}|X=x\right]  E_{\theta}\left[
\begin{array}
[c]{c}%
-\left\{  E_{\theta}\left[  H_{1}|X=x\right]  \right\}  ^{-1/2}K_{f_{X}%
,\infty}\left(  x,X\right)  \times\\
\left\{  E_{\theta}\left[  H_{1}|X\right]  \right\}  ^{-1/2}\left\{
H_{1}b\left(  X\ \right)  +H_{3}\right\}  s_{l}\left(  O;\theta\right)
\end{array}
\right] \\
&  =E\left[  H_{1}|X=x\right]  ^{1/2}\ E_{\theta}\left[
\begin{array}
[c]{c}%
-K_{f_{X},\infty}\left(  x,X\right)  \left\{  E_{\theta}\left[  H_{1}%
|X\right]  \right\}  ^{-1/2}\\
\times\left\{  H_{1}b\left(  X\ \right)  +H_{3}\right\}  \left\{  s_{l}\left(
O|X\right)  +s_{l}\left(  X\right)  \right\}
\end{array}
\right] \\
&  =E\left[  H_{1}|X=x\right]  ^{1/2}\ E_{f_{X}}\left\{  E_{\theta}\left[
\begin{array}
[c]{c}%
-K_{f_{X},\infty}\left(  x,X\right)  \left\{  E_{\theta}\left[  H_{1}%
|X\right]  \right\}  ^{-1/2}\times\\
\left\{  H_{1}b\left(  X\ \right)  +H_{3}\right\}  s_{l}\left(  O|X\right)  |X
\end{array}
\right]  \right\} \\
&  =-E_{\theta}\left[  \left(  H_{1}b\left(  X\right)  +H_{3}\right)
s_{l}\left(  O|X\right)  |X=x\right]
\end{align*}
By an analogous argument
\[
if_{1,p\left(  x\right)  \ }\left(  O;\theta\right)  =-\left[
\begin{array}
[c]{c}%
\left\{  E_{\theta}\left[  H_{1}|X=x\right]  \right\}  ^{-1/2}K_{f_{X},\infty
}\left(  x,X\right) \\
\times\left\{  E_{\theta}\left[  H_{1}|X\right]  \right\}  ^{-1/2}\left\{
H_{1}p\left(  X\ \right)  +H_{2}\right\}
\end{array}
\right]
\]
would be an influence function.

Indeed since the sequences $\left\{  c_{l}\left(  \cdot\right)  \right\}
\ $\ and $\left\{  a_{l}\left(  \cdot\right)  \right\}  $ are dense the
existence of such a kernel is also a necessary condition for $if_{1,b\left(
x\right)  \ }\left(  O;\theta\right)  $ \ and $if_{1,p\left(  x\right)
\ }\left(  O;\theta\right)  $ to exist and thus for $if_{1,\psi}\left(
o;\theta\right)  $ to exist. A kernel satisfying Eq.$\left(  \ref{dirac2}%
\right)  $ is referred to as the Dirac delta function wrt to the measure
$dF_{X}\left(  x\right)  $ and would clearly have to satisfy
\begin{equation}
K_{f_{X},\infty}\left(  x_{i_{1}},x_{i2}\right)  =0\text{ }if\text{ }x_{i_{2}%
}\neq x_{i_{1}}\label{diracpoint}%
\end{equation}
were it to exist. Of course a kernel satisfying Eq. (\ref{dirac2}) is known
not to exist in $L_{2}\left[  F_{X}\right]  \times L_{2}\left[  F_{X}\right]
.$ We conclude that $if_{1,\psi}\left(  o;\theta\right)  $ does not have an
influence function and therefore $\mathbb{IF}_{2,2,\psi}\left(  \theta\right)
$ does not exist.

\textbf{A Formal Approach:} To motivate how one might overcome this
difficulty, we note that kernels satisfying Eq. (\ref{dirac2}) exist as
generalized functions or kernels (also known as Schwartz functions or
distributions). We shall "formally' derive higher order influence functions
that appear to be elements of the space of generalized functions. However, we
use these calculations only as motivation for statistical procedures based on
ordinary kernels living in $L_{2}\left[  F_{X}\right]  \times L_{2}\left[
F_{X}\right]  .$ Thus it does not matter whether these formal calculations
could be made rigorous with appropriate redefinitions. Rather we can simply
regard the following as results obtained by applying a "formal calculus" to
part 5c of Theorem \ref{eift} that adds to the usual calculus additional
identities licensed by Eqs. (\ref{dirac2}) and (\ref{diracpoint})$.$

We will need the fact that$,$ for any function $v\left(  x;\theta\right)  ,$
Eq. $\left(  \ref{diracpoint}\right)  \ $implies that
\[
v\left(  x;\theta\right)  K_{f_{X},\infty}\left(  x,X\right)  =v\left(
X;\theta\right)  K_{f_{X},\infty}\left(  x,X\right)  .
\]

We now show that%
\[
\mathbb{IF}_{2,2,\psi}\left(  \theta\right)  \equiv\mathbb{V}\left[
IF_{2,2,\psi,i_{1},i_{2}}\left(  \theta\right)  \right]  =\Pi_{\theta
,2}\left[  \mathbb{V}\left[  if_{1,if_{1,\psi}\left(  O_{i_{1}};\cdot\right)
\ }\left(  O_{i_{2}};\theta\right)  /2\right]  |\mathcal{U}_{1}^{\perp
_{2,\theta}}\left(  \theta\right)  \right]
\]
would formally have U-statistic kernel
\begin{align}
IF_{2,2,\psi,i_{1},i_{2}}\left(  \theta\right)   &  =-\left[
\begin{array}
[c]{c}%
\varepsilon_{b,i_{1}}\left(  \theta\right)  E_{\theta}\left[  H_{1}|X_{i_{1}%
}\right]  ^{-\frac{1}{2}}K_{f_{X},\infty}\left(  X_{i_{1}},X_{i_{2}}\right) \\
E_{\theta}\left[  H_{1}|X_{i_{2}}\right]  ^{-\frac{1}{2}}\varepsilon_{p,i_{2}%
}\left(  \theta\right)
\end{array}
\right]  ,\label{infin2}\\
\text{with }\varepsilon_{b,i_{1}}\left(  \theta\right)   &  =\ \left\{
B_{i_{1}}H_{1,i_{1}}+H_{3,i_{1}}\right\}  ,\text{ }\varepsilon_{p,i_{2}%
}\left(  \theta\right)  =\ \left\{  H_{1,i_{2}}P_{i_{2}}+H_{2,i_{2}}\right\}
.\nonumber
\end{align}
To show Eq \ref{infin2} note, by
\[
\partial H\left(  b,p\right)  /\partial P=\partial\left\{  BPH_{1}%
+BH_{2}+PH_{3}+H_{4}\right\}  /\partial P=BH_{1}+H_{3}%
\]
$\ $and%
\[
\partial H\left(  b,p\right)  /\partial B=PH_{1}+H_{2},
\]
we have
\[
if_{1,if_{1,\psi}\left(  O_{i_{1}};\cdot\right)  \ }\left(  O_{i_{2}}%
;\theta\right)  =Q_{2,b,\overline{i}_{2}}\left(  \theta\right)
+Q_{2,p,\overline{i}_{2}}\left(  \theta\right)  -IF_{1,\psi,i_{2}}\left(
\theta\right)
\]
where%

\begin{align*}
Q_{2,p,\overline{i}_{2}}\left(  \theta\right)   &  \equiv\left\{  B_{i_{1}%
}H_{1,i_{1}}+H_{3,i_{1}}\right\}  if_{1,p\left(  X_{i_{1}}\right)  \ }\left(
O_{i_{2}};\theta\right) \\
&  =-\left\{  B_{i_{1}}H_{1,i_{1}}+H_{3,i_{1}}\right\}  E_{\theta}\left[
H_{1}|X_{i_{1}}\right]  ^{-\frac{1}{2}}\\
&  \times K_{f_{X},\infty}\left(  X_{i_{1}},X_{i_{2}}\right)  E_{\theta
}\left[  H_{1}|X_{i_{2}}\right]  ^{-\frac{1}{2}}\left\{  P_{i_{2}}H_{1,i_{2}%
}+H_{2,i_{2}}\right\} \\
&  =-\varepsilon_{b,i_{1}}\left(  \theta\right)  E_{\theta}\left[
H_{1}|X_{i_{1}}\right]  ^{-\frac{1}{2}}K_{f_{X},\infty}\left(  X_{i_{1}%
},X_{i_{2}}\right)  E_{\theta}\left[  H_{1}|X_{i_{2}}\right]  ^{-\frac{1}{2}%
}\varepsilon_{p,i_{2}}\left(  \theta\right) \\
Q_{2,b,\overline{i}_{2}}\left(  \theta\right)   &  \equiv\left\{  P_{i_{1}%
}H_{1,i_{1}}+H_{2,i_{1}}\right\}  if_{1,b\left(  X_{i_{1}}\right)  \ }\left(
O_{i_{2}};\theta\right) \\
&  =-\varepsilon_{b,i_{2}}\left(  \theta\right)  E_{\theta}\left[
H_{1}|X_{i_{2}}\right]  ^{-\frac{1}{2}}K_{f_{X},\infty}\left(  X_{i_{2}%
},X_{i_{1}}\right)  E_{\theta}\left[  H_{1}|X_{i_{1}}\right]  ^{-\frac{1}{2}%
}\varepsilon_{p,i_{1}}\left(  \theta\right)
\end{align*}

Thus, by part 5c of Theorem \ref{eift} $\ $%
\begin{align*}
\mathbb{IF}_{2,2,\psi}\left(  \theta\right)   &  =\Pi_{\theta,2}\left[
\frac{1}{2}\left\{  \mathbb{Q}_{2,p,\overline{i}_{2}}\left(  \theta\right)
+\mathbb{Q}_{2,b,\overline{i}_{2}}\left(  \theta\right)  +\mathbb{IF}%
_{1,\psi,i_{2}}\right\}  |\mathcal{U}_{1}^{\perp_{2,\theta}}\left(
\theta\right)  \right] \\
&  =\frac{1}{2}\left\{  \mathbb{Q}_{2,p,\overline{i}_{2}}\left(
\theta\right)  +\mathbb{Q}_{2,b,\overline{i}_{2}}\left(  \theta\right)
\right\}  =\mathbb{Q}_{2,p,\overline{i}_{2}}\left(  \theta\right)
\equiv\mathbb{V}\left[  RHS\text{ of Eq}.(\ref{infin2})\right]
\end{align*}
since $\mathbb{IF}_{1,\psi,i_{2}}$ is a function of only one subject's data
and $Q_{2,p,\overline{i}_{2}}\left(  \theta\right)  $ and $Q_{2,b,\overline
{i}_{2}}\left(  \theta\right)  $ are the same up to a permutation that
exchanges $i_{2}$ with $i_{1}.$

To obtain $IF_{3,3,\psi,\overline{i}_{m}}\left(  \theta\right)  ,$ one must
derive the influence function $if_{1,if_{2,2,\psi}\left(  O_{i_{1}},O_{i_{2}%
};\cdot\right)  \ }\left(  O_{i_{3}};\theta\right)  \ $of $if_{2,2,\psi
}\left(  O_{i_{1}},O_{i_{2}};\theta\right)  .$ The formula for $IF_{3,3,\psi
,\overline{i}_{m}}\left(  \theta\right)  $ is given in Eq. $\left(
\ref{imm}\right)  .$ A detailed derivation is given in the Appendix. Here we
simply note that the only essentially new point is that we now require the
influence function of $K_{f_{X},\infty}\left(  X_{i_{1}},X_{i_{2}}\right)  $,
which, as shown next, is given by
\begin{equation}
IF_{1,K_{f_{X},\infty}\left(  X_{i_{1}},X_{i_{2}}\right)  }=-\left\{
\begin{array}
[c]{c}%
K_{f_{X},\infty}\left(  X_{i_{1}},X_{i_{3}}\right)  K_{f_{X},\infty}\left(
X_{i_{3}},X_{i_{2}}\right) \\
-K_{f_{X},\infty}\left(  X_{i_{1}\ },X_{i_{2}}\right)
\end{array}
\right\} \label{kernelif}%
\end{equation}

To see that if Eq.$\left(  \ref{dirac2}\right)  $ held, Eq.$\left(
\ref{kernelif}\right)  $ would hold, note that for any path
$\widetilde{\mathbb{\theta}}\left(  t\right)  $ with
$\widetilde{\mathbb{\theta}}\left(  0\right)  = f_{X}(\cdot),$ $h\left(
x\right)  =E_{\widetilde{\mathbb{\theta}}\left(  t\right)  }\left[
K_{\widetilde{\mathbb{\theta}}\left(  t\right)  ,\infty}\left(  x,X_{i_{1}%
}\right)  h\left(  X_{i_{1}}\right)  \right]  .\ $Differentiating wrt to $t$
and evaluating at $t=0$ , we have
\[
0=E_{\theta}\left[  K_{f_{X},\infty}\left(  x,X\right)  h\left(  X\right)
S\left(  \theta\right)  \right]  +E_{_{\theta}}\left[  \left\{  \frac
{\partial}{\partial t}K_{\widetilde{\mathbb{\theta}}\left(  t\right)  ,\infty
}\left(  x,X_{i_{1}}\right)  _{|t=0}\right\}  h\left(  X_{i_{1}}\right)
\right]
\]
Hence it suffices to show that
\begin{align*}
&  -E_{\theta}\left[  K_{f_{X},\infty}\left(  x,X\right)  h\left(  X\right)
S\left(  \theta\right)  \right] \\
&  =E_{_{\theta}}\left[  \left\{  E_{\theta}\left\{  -K_{f_{X},\infty}\left(
x,X_{i_{2}}\right)  K_{f_{X},\infty}\left(  X_{i_{2}},X_{i_{1}}\right)
S_{i_{2}}\left(  \theta\right)  |X_{i_{1}}\right\}  \right\}  h\left(
X_{i_{1}}\right)  \right]
\end{align*}
But, by Eq.$\left(  \ref{dirac2}\right)  ,$
\begin{align*}
&  E_{_{\theta}}\left[  \left\{  E_{\theta}\left\{  -K_{f_{X},\infty}\left(
x,X_{i_{2}}\right)  K_{f_{X},\infty}\left(  X_{i_{2}},X_{i_{1}}\right)
S_{i_{2}}\left(  \theta\right)  |X_{i_{1}}\right\}  \right\}  h\left(
X_{i_{1}}\right)  \right] \\
&  =E_{_{\theta}}\left[  -K_{f_{X},\infty}\left(  x,X_{i_{1}}\right)
S_{i_{1}}\left(  \theta\right)  h\left(  X_{i_{1}}\right)  \right]  .
\end{align*}

\textbf{\ Feasible Estimators:} These "formal" calculations motivate a
"truncated Dirac" approach to estimate $\psi\left(  \theta\right)  .$ Let
$\left\{  z_{l}\left(  \cdot\right)  \right\}  \equiv\left\{  z_{l}\left(
X\right)  ;1,2,...\right\}  $ be a countable sequence of known basis{}
functions with dense span in $L_{2}\left(  F_{X}\right)  $ and define
$\overline{z}_{k}\left(  X\right)  ^{T}=\left(  z_{1}\left(  X\right)
,...,z_{k}\left(  X\right)  \right)  .$ Define
\[
K_{f_{X},k}\left(  X_{i_{1}},X_{i_{2}}\right)  \equiv\overline{z}_{k}\left(
X_{i_{1}}\right)  ^{T}\left\{  E_{f_{X}}\left[  \overline{z}_{k}\left(
X\ \right)  \overline{z}_{k}\left(  X\right)  ^{T}\right]  \right\}
^{-1}\overline{z}_{k}\left(  X_{i_{2}}\right)
\]
to be the projection kernel in $L_{2}\left(  F_{X}\right)  $ onto the
subspace
\[
lin\left\{  \overline{z}_{k}\left(  X\right)  \right\}  \equiv\left\{
\eta^{T}\overline{z}_{k}\left(  x\ \right)  ;\eta\in R^{k},\eta^{T}%
\overline{z}_{k}\left(  x\ \right)  \in L_{2}\left(  F_{X}\right)  \right\}
\]
spanned by the elements of $\overline{z}_{k}\left(  X\right)  .$ That is, for
any $h\left(  x\right)  ,$%
\begin{align*}
&  \Pi_{f_{X}}\left[  h\left(  X\right)  |lin\left\{  \overline{z}_{k}\left(
x\right)  \right\}  \right] \\
&  =E_{f_{X}}\left[  K_{f_{X},k}\left(  x,X\right)  h\left(  X\right)  \right]
\\
&  =\overline{z}_{k}\left(  x\right)  ^{T}\left\{  E_{f_{X}}\left[
\overline{z}_{k}\left(  X\ \right)  \overline{z}_{k}\left(  X\right)
^{T}\right]  \right\}  ^{-1}E_{f_{X}}\left[  \overline{z}_{k}\left(
X\ \right)  h\left(  X\right)  \right]
\end{align*}
Then we can view $K_{f_{X},k}\left(  x_{i_{1}},x_{i_{2}}\right)  $ as a
truncated at $k$ approximation to $K_{f_{X},\infty}\left(  x_{i_{1}},x_{i_{2}%
}\right)  $ that is in $L_{2}\left[  F_{X}\right]  \times L_{2}\left[
F_{X}\right]  $ and satisfies Eq.$\left(  \ref{dirac2}\right)  $ for all
$r\left(  x\right)  \in lin\left\{  \overline{z}_{k}\left(  X\right)
\right\}  .$ Then a natural idea would be to substitute
\[
IF_{2,2,\psi,i_{1},i_{2}}^{\left(  k\right)  }\left(  \widehat{\theta}\right)
\equiv\left(
\begin{array}
[c]{c}%
-\varepsilon_{b,i_{1}}\left(  \widehat{\theta}\right)  E_{\widehat{\theta}%
}\left[  H_{1}|X_{i_{1}}\right]  ^{-\frac{1}{2}}K_{\widehat{f}_{X},k}\left(
X_{i_{1}},X_{i_{2}}\right) \\
\times E_{\widehat{\theta}}\left[  H_{1}|X_{i_{2}}\right]  ^{-\frac{1}{2}%
}\varepsilon_{p,i_{2}}\left(  \widehat{\theta}\right)
\end{array}
\right)
\]
with, for example,
\[
\varepsilon_{b,i_{1}}\left(  \widehat{\theta}\right)  =\left\{  \widehat{B}%
_{i_{1}}H_{1,i_{1}}+H_{3,i_{1}}\right\}
\]
for the generalized function $IF_{2,2,\psi,i_{1},i_{2}}\left(  \widehat{\theta
}\right)  \ $based on Eqs. \ref{infin2} resulting in the feasible 2nd
U-statistic estimator%
\[
\widehat{\psi}_{2}^{\left(  k\right)  }=\psi\left(  \widehat{\theta}\right)
+\mathbb{IF}_{1,\psi}\left(  \widehat{\theta}\right)  +\mathbb{IF}%
_{2,2,\psi\left(  \theta\right)  }^{(k)}\left(  \widehat{\theta}\right)
\]
where%
\[
\mathbb{IF}_{2,2,\psi}^{(k)}\left(  \widehat{\theta}\right)  \equiv
\mathbb{V}\left[  IF_{2,2,\psi,i_{1},i_{2}}^{\left(  k\right)  }\left(
\widehat{\theta}\right)  \right]
\]

To avoid having to do a matrix inversion it is convenient to choose
$\overline{z}_{k}\left(  X\ \right)  =\overline{\varphi}_{k}\left(
X\ \right)  /\left\{  \widehat{f}_{X}\left(  X\right)  \right\}  ^{1/2}$ where
$\varphi_{1}\left(  X\right)  ,\varphi_{2}\left(  X\right)  ,...$ is a
complete orthonormal basis wrt to Lebesgue measure in $R^{d}.$ Then
$E_{\widehat{f}_{X}}\left[  \overline{z}_{k}\left(  X\ \right)  \overline
{z}_{k}\left(  X\right)  ^{T}\right]  =I_{k\times k}$ so
\begin{align*}
K_{\widehat{f}_{X},k}\left(  X_{i_{1}},X_{i_{2}}\right)   &  =\overline{z}%
_{k}\left(  X_{i_{1}}\ \right)  ^{T}\overline{z}_{k}\left(  X_{i_{2}%
}\ \right)  =\frac{K_{Leb,k}\left(  X_{i_{1}},X_{i_{2}}\right)  }{\left\{
\widehat{f}_{X}\left(  X_{i_{1}}\right)  \widehat{f}_{X}\left(  X_{i_{2}%
}\right)  \right\}  ^{1/2}},\\
\text{where }K_{Leb,k}\left(  X_{i_{1}},X_{i_{2}}\right)   &  \equiv
\overline{\varphi}_{k}\left(  X_{i_{1}}\right)  ^{T}\overline{\varphi}%
_{k}\left(  X_{i_{2}}\right)  .
\end{align*}
This choice corresponds to having taken
\[
K_{f_{X},\infty}\left(  X_{i_{1}},X_{i_{2}}\right)  =K_{Leb,\infty}\left(
X_{i_{1}},X_{i_{2}}\right)  /\left\{  f_{X}\left(  X_{i_{1}}\right)
f_{X}\left(  X_{i_{2}}\right)  \right\}  ^{1/2}%
\]
in our formal calculations where $K_{Leb,\infty}\left(  X_{i_{1}},X_{i_{2}%
}\right)  $ is the Dirac delta function wrt to Lebesgue measure. In that case
with $g\left(  X\right)  \equiv f_{X}\left(  X\right)  E_{\theta}\left[
H_{1}|X\right]  $ and $\widehat{g}\left(  X\right)  \equiv\widehat{f}%
_{X}\left(  X\right)  E_{\widehat{\theta}}\left[  H_{1}|X\right]  ,$
\begin{align}
IF_{2,2,\psi,i_{1},i_{2}}\left(  \theta\right)   &  =-\varepsilon_{b,i_{1}%
}\left(  \theta\right)  g\left(  X_{i_{1}}\right)  ^{-\frac{1}{2}%
}K_{Leb,\infty}\left(  X_{i_{1}},X_{i_{2}}\right)  g\left(  X_{i_{2}}\right)
^{-\frac{1}{2}}\varepsilon_{p,i_{2}}\left(  \theta\right) \label{221}\\
IF_{2,2,\psi,i_{1},i_{2}}^{\left(  k\right)  }\left(  \widehat{\theta}\right)
&  =-\varepsilon_{b,i_{1}}\left(  \widehat{\theta}\right)  \widehat{g}\left(
X_{i_{1}}\right)  ^{-\frac{1}{2}}K_{Leb,k}\left(  X_{i_{1}},X_{i_{2}}\right)
\widehat{g}\left(  X_{i_{2}}\right)  ^{-\frac{1}{2}}\varepsilon_{p,i_{2}%
}\left(  \widehat{\theta}\right) \label{22}%
\end{align}
In the appendix, we show one can proceed by induction to formally obtain that
for $m=3,4,...,$%
\begin{align}
&  IF_{m,m,\psi,\overline{i}_{m}}\left(  \theta\right) \\
&  =\varepsilon_{b,i_{1}}\left(  \theta\right)  g\left(  X_{i_{1}}\right)
^{-\frac{1}{2}}\left[
\begin{array}
[c]{c}%
\sum_{j=0}^{m-2}c(m,j)\times\\
\prod\limits_{s=1}^{j}\frac{H_{1,i_{s+1}}}{g\left(  X_{i_{s+1}}\right)
}K_{Leb,\infty}\left(  X_{i_{s}},X_{i_{s+1}}\right) \\
\times K_{Leb,\infty}\left(  X_{i_{j+1}},X_{i_{m}}\right)
\end{array}
\right]  g\left(  X_{i_{m}}\right)  ^{-\frac{1}{2}}\varepsilon_{p,i_{m}%
}\left(  \theta\right) \label{imm}%
\end{align}
where $c(m,j)=\binom{m-2}{j}\left(  -1\right)  ^{(j+1)},$ which we then use to
obtain $IF_{m,m,\psi,\overline{i}_{m}}^{\left(  k\right)  }\left(
\widehat{\theta}\right)  \ $and $\widehat{\psi}_{m}^{\left(  k\right)
}=\widehat{\psi}_{2}^{\left(  k\right)  }+\sum_{j=3}^{m}\mathbb{IF}_{j,j,\psi
}^{\left(  k\right)  }\left(  \widehat{\theta}\right)  .$

\textbf{Statistical Properties:} We shall prove below that the estimator
$\widehat{\psi}_{m}^{\left(  k\right)  }$ has variance
\[
var_{\theta}\left[  \widehat{\psi}_{m}^{\left(  k\right)  }\right]
\asymp\left(  \frac{1}{n}\max\left[  1,\left(  \frac{k}{n}\right)
^{m-1}\right]  \right)
\]
when $\left\{  \varphi_{l}\left(  X\right)  ;l=1,2,...\right\}  $ is a compact
wavelet basis$.$ (Robins et al. (2007) proves this result for more general
bases). We also prove that the bias
\[
E_{\theta}\left[  \widehat{\psi}_{m}^{\left(  k\right)  }\right]  -\psi\left(
\theta\right)  =TB_{k}\left(  \theta\right)  +EB_{m}\left(  \theta\right)  ,
\]
of $\widehat{\psi}_{m}^{\left(  k\right)  }$ is the sum of a truncation bias
term of order
\[
TB_{k}\left(  \theta\right)  =O_{p}\left(  k^{-\left(  \beta_{b}+\beta
_{p}\right)  /d}\right)
\]
(for a basis $\left\{  \varphi_{l}\left(  X\right)  ;l=1,2,...\right\}  $ that
provides optimal rate approximation for H\"{o}lder balls) and an estimation
bias term of order
\begin{align*}
EB_{m}\left(  \theta\right)   &  =O_{p}\left(  \left\{  P-\widehat{P}\right\}
\left\{  B-\widehat{B}\right\}  \left(  \frac{G-\widehat{G}}{\widehat{G}%
}\right)  ^{m-1}\right) \\
&  =O_{p}\left(  n^{-\frac{\left(  m-1\right)  \beta_{g}}{2\beta_{g}+d}%
-\frac{\beta_{b}}{2\beta_{b}+d}-\frac{\beta_{p}}{2\beta_{p}+d}}\right)  .
\end{align*}
The truncation bias is of this order only if $g$ has smoothness exceeding
$\max\left\{  \beta_{p},\beta_{b}\right\}  $. This restriction on $g$ is
removed later by using kernels based on Eq. $\left(  \ref{ON}\right)  $. Note
this estimation bias is $O_{P}\left(  \left\Vert \theta-\widehat{\theta
}\right\Vert ^{m+1}\right)  .$ It gets its name from the fact that, unlike the
truncation bias, it would be exactly zero if the initial estimator
$\widehat{\theta}$ happened to equal $\theta.$ Thus, {}the U-statistic
estimator $\widehat{\psi}_{m}^{\left(  k\right)  }$ for our functional
$\psi\left(  \theta\right)  $ (which does not admit a second order influence
function) differs from the U-statistic estimators $\widehat{\psi}_{m}$ of Eq.
(\ref{est}) for functionals that admit second order influence functions in
that, owing to truncation bias, the total bias of $\widehat{\psi}_{m}^{\left(
k\right)  }$ is not $O_{p}\left(  \left\Vert \theta-\widehat{\theta
}\right\Vert ^{m+1}\right)  .$ The choice of $k\ $determines the trade-off
between the variance and truncation bias. As $k\rightarrow\infty$ with $n$
fixed$,$ $var_{\theta}\left[  \widehat{\psi}_{m}^{\left(  k\right)  }\right]
\rightarrow\infty$ and $TB_{k}\left(  \theta\right)  \rightarrow0.$ \ Thus, we
can heuristically view the non-existent estimator $\widehat{\psi}%
_{m}=\widehat{\psi}_{m}^{\left(  k=\infty\right)  }$ as the choice of $k$ that
results in no truncation bias [and therefore a total bias of $O_{p}\left(
\left\Vert \theta-\widehat{\theta}\right\Vert ^{m+1}\right)  ]$ at the expense
of an infinite variance. \ Writing $k=k\left(  n\right)  =n^{\rho},$ the order
of the asymptotic MSE of $\widehat{\psi}_{m}^{\left(  k\right)  }$ is
minimized at the value of $\rho\ $for which order of the variance equals the
order of the sum of the truncation and estimation bias.

\begin{remark}
\label{1proj}The models of examples 1-4 exhibit a spectrum of different
likelihood functions and therefore a spectrum of different first order and
higher order scores. Nonetheless, because the first order influence functions
of the functionals $\psi\left(  \theta\right)  $ share a common structure, we
were able to use part 5c of Theorem \ref{eift} to formally derive
$IF_{m,m,\psi,\overline{i}_{m}}\left(  \theta\right)  $ and, thus, the
feasible $IF_{m,m,\psi,\overline{i}_{m}}^{\left(  k\right)  }\left(
\widehat{\theta}\right)  $ in examples 1-4 in a unified manner without needing
to consult the full likelihood function for any of the models. See Remark
(\ref{proj})\ above for a closely related discussion.
\end{remark}

\textbf{A Critical Non-uniqueness:} We have as yet neglected a critical
non-uniqueness in our definition of $\mathbb{IF}_{m,m,\psi\left(
\theta\right)  }^{(k)}\left(  \widehat{\theta}\right)  $ and thus
$\widehat{\psi}_{m}^{\left(  k\right)  }$ that poses a significant problem for
our "truncated Dirac" approach. \ For instance, when $m=3,$ the two
generalized $U-statistic $ kernels $IF_{3,3,\psi,i_{1},i_{2},i_{3}}\left(
\theta\right)  $ of Eq \ref{imm} and
\begin{align*}
&  IF_{3,3,\psi,i_{1},i_{2},i_{3}}^{\ast}\left(  \theta\right) \\
&  \equiv\frac{\varepsilon_{b,i_{1}}\left(  \theta\right)  }{g\left(
X_{i_{1}}\right)  ^{\frac{1}{2}}}\left[
\begin{array}
[c]{c}%
\frac{H_{1,i_{2}}}{g\left(  X_{i_{2}}\right)  }K_{Leb,\infty}\left(  X_{i_{1}%
},X_{i_{2}}\right) \\
-E_{\theta}\left[  \frac{K_{Leb,\infty}\left(  X_{i_{1}},X_{i_{2}}\right)
}{f\left(  X_{i_{2}}\right)  }|X_{i_{1}}\right]
\end{array}
\right] \\
&  \times K_{Leb,\infty}\left(  X_{i_{1}},X_{i_{3}}\right)  \frac
{\varepsilon_{p,i_{3}}\left(  \theta\right)  }{g\left(  X_{i_{3}}\right)
^{\frac{1}{2}}}%
\end{align*}
are precisely equal, by Eq. $\left(  \ref{diracpoint}\right)  $; nonetheless,
upon truncation, they result in different feasible kernels;
\begin{align*}
&  IF_{3,3,\psi,i_{1},i_{2}}^{\left(  k\right)  }\left(  \widehat{\theta
}\right) \\
&  =\frac{\widehat{\varepsilon}_{b,i_{1}}\left(  \theta\right)  }%
{\widehat{g}\left(  X_{i_{1}}\right)  ^{\frac{1}{2}}}\left[
\begin{array}
[c]{c}%
\frac{H_{1,i_{2}}}{\widehat{g}\left(  X_{i_{2}}\right)  }K_{Leb,k}\left(
X_{i_{1}},X_{i_{2}}\right)  K_{Leb,k}\left(  X_{i_{2}},X_{i_{3}}\right) \\
-K_{Leb,k}\left(  X_{i_{1}},X_{i_{3}}\right)
\end{array}
\right]  \times\frac{\widehat{\varepsilon}_{p,i_{3}}\left(  \theta\right)
}{\widehat{g}\left(  X_{i_{3}}\right)  ^{\frac{1}{2}}}%
\end{align*}
and
\begin{align*}
&  IF_{3,3,\psi,i_{1},i_{2},i_{3}}^{\left(  k\right)  ,\ast}\left(
\widehat{\theta}\right) \\
&  \equiv\frac{\widehat{\varepsilon}_{b,i_{1}}\left(  \theta\right)
}{\widehat{g}\left(  X_{i_{1}}\right)  ^{\frac{1}{2}}}\left[
\begin{array}
[c]{c}%
\frac{H_{1,i_{2}}}{\widehat{g}\left(  X_{i_{2}}\right)  }K_{Leb,k}\left(
X_{i_{1}},X_{i_{2}}\right) \\
-E_{\widehat{\theta}}\left[  \frac{K_{Leb,k}\left(  X_{i_{1}},X_{i_{2}%
}\right)  }{\widehat{f}\left(  X_{i_{2}}\right)  }|X_{i_{1}}\right]
\end{array}
\right] \\
&  \times K_{Leb,k}\left(  X_{i_{1}},X_{i_{3}}\right)  \frac
{\widehat{\varepsilon}_{p,i_{3}}\left(  \theta\right)  }{\widehat{g}\left(
X_{i_{3}}\right)  ^{\frac{1}{2}}}%
\end{align*}
with different orders of bias. For simplicity, we consider the case where
$H_{1}=1$ as in Examples \textbf{1a}-\textbf{1c. \ }Let $\delta B\equiv
B-\widehat{B},$ $\delta P\equiv P-\widehat{P},$ $\delta f=\delta g\equiv
\frac{f}{\widehat{f}}-1,$ and $\overline{Z}_{k}\equiv\frac{\overline{\varphi
}_{k}\left(  X\right)  }{\widehat{f}\left(  X\right)  ^{\frac{1}{2}}}%
=\frac{\overline{\varphi}_{k}\left(  X\right)  }{\widehat{g}\left(  X\right)
^{\frac{1}{2}}},$ then,%
\begin{align*}
&  E_{\theta}\left[  IF_{3,3,\psi,i_{1},i_{2},i_{3}}^{\left(  k\right)  ,\ast
}\left(  \widehat{\theta}\right)  \right] \\
&  =E_{\theta}\left[
\begin{array}
[c]{c}%
\frac{\delta B_{i_{1}}}{\widehat{f}\left(  X_{i_{1}}\right)  ^{\frac{1}{2}}%
}\times\\
E_{\mu}\left[  \left(  \frac{f\left(  X_{i_{2}}\right)  }{\widehat{f}\left(
X_{i_{2}}\right)  }-1\right)  \overline{\varphi}_{k}\left(  X_{i_{2}}\right)
^{T}\right]  \overline{\varphi}_{k}\left(  X_{i_{1}}\right) \\
\times E_{\theta}\left[  \frac{\delta P_{i_{3}}}{\widehat{f}\left(  X_{i_{3}%
}\right)  ^{\frac{1}{2}}}\overline{\varphi}_{k}\left(  X_{i_{3}}\right)
^{T}\right]  \overline{\varphi}_{k}\left(  X_{i_{1}}\right)
\end{array}
\right] \\
&  =E_{\widehat{\theta}}\left[
\begin{array}
[c]{c}%
\left(  \frac{f\left(  X_{i_{1}}\right)  }{\widehat{f}\left(  X_{i_{1}%
}\right)  }-1+1\right)  \widehat{f}\left(  X_{i_{1}}\right)  ^{\frac{1}{2}%
}\delta B_{i_{1}}\times\\
E_{\widehat{\theta}}\left[  \left(  \frac{f\left(  X_{i_{2}}\right)
}{\widehat{f}\left(  X_{i_{2}}\right)  }-1\right)  \widehat{f}\left(
X_{i_{2}}\right)  ^{-\frac{1}{2}}\frac{\overline{\varphi}_{k}\left(  X_{i_{2}%
}\right)  ^{T}}{\widehat{f}\left(  X_{i_{2}}\right)  ^{\frac{1}{2}}}\right]
\frac{\overline{\varphi}_{k}\left(  X_{i_{1}}\right)  }{\widehat{f}\left(
X_{i_{1}}\right)  ^{\frac{1}{2}}}\\
\times E_{\widehat{\theta}}\left[  \left(  \frac{f\left(  X_{i_{3}}\right)
}{\widehat{f}\left(  X_{i_{3}}\right)  }-1+1\right)  \delta P_{i_{3}}%
\frac{\overline{\varphi}_{k}\left(  X_{i_{3}}\right)  ^{T}}{\widehat{f}\left(
X_{i_{3}}\right)  ^{\frac{1}{2}}}\right]  \frac{\overline{\varphi}_{k}\left(
X_{i_{1}}\right)  }{\widehat{f}\left(  X_{i_{1}}\right)  ^{\frac{1}{2}}}%
\end{array}
\right] \\
&  =E_{\widehat{\theta}}\left[
\begin{array}
[c]{c}%
\widehat{f}\left(  X_{i_{1}}\right)  ^{\frac{1}{2}}\delta B_{i_{1}%
}E_{\widehat{\theta}}\left[  \delta f\left(  X_{i_{2}}\right)  \widehat{f}%
\left(  X_{i_{2}}\right)  ^{-\frac{1}{2}}\overline{Z}_{k,i_{2}}^{T}\right]
\overline{Z}_{k,i_{1}}\\
\times E_{\widehat{\theta}}\left[  \delta P_{i_{3}}\overline{Z}_{k,i_{3}}%
^{T}\right]  \overline{Z}_{k,i_{1}}%
\end{array}
\right] \\
&  +O_{p}\left(  \left\{  B-\widehat{B}\right\}  \left\{  P-\widehat{P}%
\right\}  \left\{  G-\widehat{G}\right\}  ^{2}\right)
\end{align*}
and%
\begin{align*}
&  E_{\theta}\left[  IF_{3,3,\psi,i_{1},i_{2},i_{3}}^{\left(  k\right)
}\left(  \widehat{\theta}\right)  \right] \\
&  =E_{\mu}\left[
\begin{array}
[c]{c}%
E_{\theta}\left[  \frac{\delta B_{i_{1}}}{\widehat{f}\left(  X_{i_{1}}\right)
^{\frac{1}{2}}}\overline{\varphi}_{k}\left(  X_{i_{1}}\right)  ^{T}\right]
\overline{\varphi}_{k}\left(  X_{i_{2}}\right)  \left(  \frac{f\left(
X_{i_{2}}\right)  }{\widehat{f}\left(  X_{i_{2}}\right)  }-1\right) \\
\times\overline{\varphi}_{k}\left(  X_{i_{2}}\right)  ^{T}E_{\theta}\left[
\frac{\delta P_{i_{3}}}{\widehat{f}\left(  X_{i_{3}}\right)  ^{\frac{1}{2}}%
}\overline{\varphi}_{k}\left(  X_{i_{3}}\right)  \right]
\end{array}
\right] \\
&  =E_{\widehat{\theta}}\left[
\begin{array}
[c]{c}%
E_{\widehat{\theta}}\left[  \left(  \delta f\left(  X_{i_{1}}\right)
+1\right)  \delta B_{i_{1}}\overline{Z}_{k,i_{1}}^{T}\right] \\
\times\overline{Z}_{k,i_{2}}\left(  \frac{f\left(  X_{i_{2}}\right)
}{\widehat{f}\left(  X_{i_{2}}\right)  }-1\right) \\
\times\overline{Z}_{k,i_{2}}^{T}E_{\widehat{\theta}}\left[  \left(  \delta
f\left(  X_{i_{3}}\right)  +1\right)  \delta P_{i_{3}}\overline{Z}_{k,i_{3}%
}\right]
\end{array}
\right] \\
&  =E_{\widehat{\theta}}\left[
\begin{array}
[c]{c}%
E_{\widehat{\theta}}\left[  \delta B_{i_{1}}\overline{Z}_{k,i_{1}}^{T}\right]
\overline{Z}_{k,i_{2}}\left(  \frac{f\left(  X_{i_{2}}\right)  }%
{\widehat{f}\left(  X_{i_{2}}\right)  }-1\right) \\
\times\overline{Z}_{k,i_{2}}^{T}E_{\widehat{\theta}}\left[  \delta P_{i_{3}%
}\overline{Z}_{k,i_{3}}\right]
\end{array}
\right] \\
&  +O_{p}\left(  \left\{  B-\widehat{B}\right\}  \left\{  P-\widehat{P}%
\right\}  \left\{  G-\widehat{G}\right\}  ^{2}\right)
\end{align*}
That is,
\begin{align*}
&  E_{\theta}\left[  IF_{3,3,\psi,i_{1},i_{2},i_{3}}^{\left(  k\right)  ,\ast
}\left(  \widehat{\theta}\right)  \right]  -E_{\theta}\left[  IF_{3,3,\psi
,i_{1},i_{2},i_{3}}^{\left(  k\right)  }\left(  \widehat{\theta}\right)
\right] \\
&  =E_{\widehat{\theta}}\left[
\begin{array}
[c]{c}%
\widehat{f}\left(  X_{i_{1}}\right)  ^{\frac{1}{2}}\delta B_{i_{1}%
}E_{\widehat{\theta}}\left[  \delta f\left(  X_{i_{2}}\right)  \widehat{f}%
\left(  X_{i_{2}}\right)  ^{-\frac{1}{2}}\overline{Z}_{k,i_{2}}^{T}\right]
\overline{Z}_{k,i_{1}}\\
E_{\widehat{\theta}}\left[  \delta P_{i_{3}}\overline{Z}_{k,i_{3}}^{T}\right]
\overline{Z}_{k,i_{1}}%
\end{array}
\right] \\
&  -E_{\widehat{\theta}}\left[
\begin{array}
[c]{c}%
E_{\widehat{\theta}}\left[  \delta B_{i_{1}}\overline{Z}_{k,i_{1}}^{T}\right]
\overline{Z}_{k,i_{2}}\delta f\left(  X_{i_{2}}\right) \\
\times\overline{Z}_{k,i_{2}}^{T}E_{\widehat{\theta}}\left[  \delta P_{i_{3}%
}\overline{Z}_{k,i_{3}}\right]
\end{array}
\right] \\
&  +O_{p}\left(  \left\{  B-\widehat{B}\right\}  \left\{  P-\widehat{P}%
\right\}  \left\{  G-\widehat{G}\right\}  ^{2}\right)
\end{align*}%
\[
=
\]%
\begin{align*}
&  E_{\widehat{\theta}}\left[  \Pi_{\widehat{\theta}}\left[  \delta
P|\overline{Z}_{k}\right]  \widehat{f}\left(  X\right)  ^{\frac{1}{2}}\delta
B\Pi_{\widehat{\theta}}\left[  \delta f\left(  X\right)  \widehat{f}\left(
X\right)  ^{-\frac{1}{2}}|\overline{Z}_{k}\right]  \right] \\
&  -E_{\widehat{\theta}}\left[  \Pi_{\widehat{\theta}}\left[  \delta
P|\overline{Z}_{k}\right]  \widehat{f}\left(  X\right)  ^{\frac{1}{2}}\delta
f\left(  X\right)  \widehat{f}\left(  X\right)  ^{-\frac{1}{2}}\Pi
_{\widehat{\theta}}\left[  \delta B|\overline{Z}_{k}\right]  \right] \\
&  +O_{p}\left(  \left\{  B-\widehat{B}\right\}  \left\{  P-\widehat{P}%
\right\}  \left\{  G-\widehat{G}\right\}  ^{2}\right)
\end{align*}%
\[
=
\]%
\begin{align*}
&  E_{\widehat{\theta}}\left[
\begin{array}
[c]{c}%
\Pi_{\widehat{\theta}}\left[  \delta P|\overline{Z}_{k}\right]  \widehat{f}%
\left(  X\right)  ^{\frac{1}{2}}\times\\
\left\{
\begin{array}
[c]{c}%
\Pi_{\widehat{\theta}}^{\bot}\left[  \delta B|\overline{Z}_{k}\right]
\Pi_{\widehat{\theta}}\left[  \delta f\left(  X\right)  \widehat{f}\left(
X\right)  ^{-\frac{1}{2}}|\overline{Z}_{k}\right] \\
-\Pi_{\widehat{\theta}}\left[  \delta B|\overline{Z}_{k}\right]
\Pi_{\widehat{\theta}}^{\bot}\left[  \delta f\left(  X\right)  \widehat{f}%
\left(  X\right)  ^{-\frac{1}{2}}|\overline{Z}_{k}\right]
\end{array}
\right\}
\end{array}
\right] \\
&  +O_{p}\left(  \left\{  B-\widehat{B}\right\}  \left\{  P-\widehat{P}%
\right\}  \left\{  G-\widehat{G}\right\}  ^{2}\right)
\end{align*}
where $\Pi_{\widehat{\theta}}\left[  h\left(  X\right)  |\overline{Z}%
_{k}\right]  $ and $\Pi_{\widehat{\theta}}^{\bot}\left[  h\left(  X\right)
|\overline{Z}_{k}\right]  $ respectively denote the projection under $F\left(
\cdot;\widehat{\theta}\right)  $ in $L_{2}\left(  \widehat{F}\right)  $ of
$h\left(  X\right)  $ on the $k$ dimensional linear subspace $lin\left\{
\overline{z}_{k}\left(  X\right)  \right\}  $ spanned by the components of the
vector $\overline{z}_{k}\left(  X\right)  $ and the projection on the
orthocomplement of this subspace.

Since the basis $\left\{  \varphi_{l}\left(  X\right)  ;l=1,2,...\right\}  $
provides optimal rate approximation for H\"{o}lder balls, it is easy to verify
that the difference is of order
\[
O_{p}\left(
\begin{array}
[c]{c}%
n^{-\frac{\beta_{p}/d}{1+2\beta_{p}/d}-\frac{\beta g/d}{1+2\beta_{g}/d}%
}k^{-\beta_{b}/d}+n^{-\frac{\beta_{p}/d}{1+2\beta_{p}/d}-\frac{\beta_{b}%
/d}{1+2\beta_{b}/d}}k^{-\beta_{g}/d}\\
+n^{-\frac{\beta_{p}/d}{1+2\beta_{p}/d}-\frac{\beta_{b}/d}{1+2\beta_{b}%
/d}-\frac{2\beta g/d}{1+2\beta_{g}/d}}%
\end{array}
\right)
\]
provided $g$ has smoothness exceeding $\max\left(  \beta_{p},\beta_{b}\right)
.$

For concreteness, we shall look at an example. Suppose $\beta_{b}/d=\beta
_{p}/d=0.3$ and $\beta_{g}/d=0.1,$ thus, by choosing $k=n^{\frac{5}{6}},$
$\widehat{\psi}_{3}^{\left(  k\right)  \ }$ converges to $\psi\left(
\theta\right)  $ at rate $n^{-\frac{1}{2}}.$ In contrast, the order,
\[
\min_{k}\left(  n^{-\frac{\beta_{p}/d}{1+2\beta_{p}/d}-\frac{\beta
g/d}{1+2\beta_{g}/d}}k^{-\beta_{b}/d}+n^{-\frac{\beta_{p}/d}{1+2\beta_{p}%
/d}-\frac{\beta_{b}/d}{1+2\beta_{b}/d}}k^{-\beta_{g}/d}+\sqrt{\frac{1}{n}%
\max\left(  1,\frac{k^{2}}{n^{2}}\right)  }\right)  ,
\]
of the optimal root mean squares error of $\widehat{\psi}_{3}^{\left(
k\right)  ,\ast}\mathbb{\ }$that uses $IF_{3,3,\psi,\overline{i}_{3}}^{\left(
k\right)  ,\ast}\left(  \widehat{\theta}\right)  $ is $n^{-0.477}\gg
n^{-0.5}.$ \ Thus $\widehat{\psi}_{3}^{\left(  k\right)  ,\ast}\mathbb{\ }%
$converges to $\psi\left(  \theta\right)  $ at a slower rate than
$\widehat{\psi}_{3}^{\left(  k\right)  }$ which uses $IF_{3,3,\psi
,\overline{i}_{3}}^{\left(  k\right)  }\left(  \widehat{\theta}\right)  .$
Nothing in our development up to this point provides any guidance as to which
of the many equivalent generalized U-statistic kernels should be selected for
truncation. To provide some guidance, we introduce an alternative approach to
the estimation of $\psi\left(  \theta\right)  $ based on truncated parameters
that admit higher order influence functions. The class of estimators we derive
using this alternative approach includes members algebraically identical to
the estimators $\widehat{\psi}_{m}^{\left(  k\right)  }$ but does not include
estimators equivalent to less efficient estimators such as $\widehat{\psi}%
_{3}^{\left(  k\right)  ,\ast}$.

\textbf{An Approach based on Truncated Parameters: }We introduce a class of
truncated parameters $\widetilde{\psi}_{k}\left(  \theta\right)  $ that (i)
depend on the sample size through a positive integer index $k=k\left(
n\right)  \ ($which we refer to as the truncation index and will be optimized
below), (ii) have influence functions $\mathbb{IF}_{m,\widetilde{\psi}_{k}%
}\left(  \theta\right)  $ of all orders $m$, (iii) equals $\psi\left(
\theta\right)  $ on a large subset $\Theta_{sub,k}$ of $\Theta$ and (iv) the
initial estimator $\widehat{\theta}\ $\ is an element of $\Theta_{sub,k}\ $so
that the plug-ins $\psi\left(  \widehat{\theta}\right)  \ $and
$\widetilde{\psi}_{k}\left(  \widehat{\theta}\right)  $\ are equal. To prepare
we introduce a simplified notation. For functions $h\left(  o,\cdot\right)  $
or $r\left(  \cdot\right)  $ of $\theta,$ we will often write $h\left(
o,\widehat{\theta}\right)  $ and $r\left(  \widehat{\theta}\right)  $ as
$\widehat{h}\left(  o\right)  $ and $\widehat{r},$ and $E_{\widehat{\theta
}_{\ }}\left[  \cdot\right]  $ as $\widehat{E}\left[  \cdot\right]  $.
Similarly , we often write $h\left(  o,\theta\right)  $ and $r\left(
\theta\right)  $ as $h\left(  o\right)  \ $and $r,$ and $E_{\theta}\left[
\cdot\right]  $ as $E\left[  \cdot\right]  .$ Further we shall introduce
slightly different definitions of truncation and estimation bias.

Define the estimator $\psi_{m,k}\left(  \widehat{\theta}\right)  \equiv
\psi\left(  \widehat{\theta}\right)  +\mathbb{IF}_{m,\widetilde{\psi}_{k}%
}\left(  \widehat{\theta}\right)  \ $or, equivalently, $\widehat{\psi}%
_{m,k}\equiv\widehat{\psi}+\widehat{\mathbb{IF}}_{m,\widetilde{\psi}_{k}}$.
\ Then the conditional bias $E\left[  \widehat{\psi}_{m,k}|\widehat{\theta
}\right]  -\psi$ of $\widehat{\psi}_{m,k}$ is $TB_{k}\ +EB_{m},$ where the
truncation bias $TB_{k}\ =\widetilde{\psi}_{k}\ -\psi\ $ is zero for
$\theta\in\Theta_{sub,k}$ and does not depend on $m$ and the estimation bias
$EB_{m,k}=E\left[  \widehat{\psi}_{m,k}|\widehat{\theta}\right]
-\widetilde{\psi}_{k}$ is $O_{P}\left(  ||\widehat{\theta}-\theta
||^{m+1}\right)  $ by Theorem \ref{eiet}. Since, as we show later, the order
of $EB_{m,k}$ does not depend on $k,$ we will abbreviate $EB_{m,k}$ as
$EB_{m},$ suppressing the dependence on $k.$ Under minimal conditions, the
conditional variance of $\widehat{\psi}_{m,k}$ is of the order of
$var\ \left[  \mathbb{IF}_{m,\widetilde{\psi}_{k}}\right]  $ whenever $k\equiv
k\left(  n\right)  \geq n.$ The rate of convergence of $\widehat{\psi}_{m,k}$
to $\psi$ can depend on the choice of $\widetilde{\psi}_{k}.$\ Nevertheless,
many different choices $\widetilde{\psi}_{k}$ result in estimators
$\widehat{\psi}_{m,k}$ that achieve what we conjecture to be the optimal rate
for estimators in our class of the form $\widehat{\psi}_{m,k}$. We choose,
among all such $\widetilde{\psi}_{k},$ the class that minimizes the
computational complexity of $\widehat{\psi}_{m,k}.$ Specifically for all
$\widetilde{\psi}_{k}$ in our chosen class and all $j,$ $\mathbb{IF}%
_{jj,\widetilde{\psi}_{k}}$ consists of a single term rather than a sum of
many terms. We conjecture this appealing property does not hold for any
$\widetilde{\psi}_{k}$ outside our class. We now describe this choice. The
parameter $\widetilde{\psi}_{k}\ $is defined in terms of $k\left(  n\right)
-$dimensional 'working' linear parametric submodels for $p\left(
\cdot\right)  $ and $b\left(  \cdot\right)  $ depending on unknown parameters
$\overline{\alpha}_{k}\ $and $\overline{\eta}_{K}$ through the basepoints
$\widehat{p}\left(  \cdot\right)  $ and $\widehat{b}\left(  \cdot\right)  ,$
where $\widehat{p}\left(  \cdot\right)  $ and $\widehat{b}\left(
\cdot\right)  $ are initial estimators from the training sample. Specifically
let $\dot{p}\left(  X\right)  $ and $\dot{b}\left(  X\right)  $ be arbitrary
bounded known functions chosen by the analyst satisfying Eqs $\left(
\text{\ref{dot1}}\right)  -\left(  \text{\ref{dot3}}\right)  $ below.%

\begin{gather}
\dot{p}\left(  X\right)  \dot{b}\left(  X\right)  E\left[  H_{1}|X\right]
\geq0\text{ }w.p.1\label{dot1}\\
\left\vert \left\vert \frac{\dot{p}\left(  X\right)  }{\dot{b}\left(
X\right)  }\right\vert \right\vert _{\infty}<C^{\ast},\left\vert \left\vert
\frac{\dot{b}\left(  X\right)  }{\dot{p}\left(  X\right)  }\right\vert
\right\vert _{\infty}<C^{\ast}\label{dot2}\\
\frac{\dot{p}\left(  X\right)  }{\dot{b}\left(  X\right)  }\text{ has at least
}\left\lceil max\left\{  \beta_{b},\ \beta_{p}\right\}  \right\rceil \text{
derivatives}\label{dot3}%
\end{gather}
Particular choices of $\dot{p}\left(  X\right)  $ and $\dot{b}\left(
X\right)  $ can make the form of $\mathbb{IF}_{m,\widetilde{\psi}_{k}}\left(
\widehat{\theta}\right)  $ more aesthetic. The choice has no bearing on the
rate of convergence of the estimator $\widehat{\psi}_{m,k}$ to $\psi\left(
\theta\right)  .$ Often there are fairly natural choices for $\dot{p}\left(
\cdot\right)  $ and $\dot{b}\left(  \cdot\right)  .$ See Remark \ref{Q} below
for examples$.$ Let $\overline{\alpha}_{k},\overline{\eta}_{k}$ be $k-$vectors
of unknown parameters and consider the 'working' linear models
\begin{align}
p^{\ast}\left(  X,\overline{\alpha}_{k}\right)   &  \equiv\widehat{p}\left(
X\right)  +\dot{p}\left(  X\right)  \overline{\alpha}_{k}^{T}\overline{z}%
_{k}\left(  X\right)  \equiv\widehat{P}+\dot{P}\overline{\alpha}_{k}%
^{T}\overline{Z}_{k}\label{worka}\\
b^{\ast}\left(  X,\overline{\eta}_{k}\right)   &  =\widehat{b}\left(
X\right)  +\dot{b}\left(  X\right)  \overline{\eta}_{k}^{T}\overline{z}%
_{k}\left(  X\right)  =\widehat{B}+\dot{B}\overline{\eta}_{k}^{T}\overline
{Z}_{k}\label{workb}%
\end{align}

We define the parameters $\widetilde{\overline{\eta}}_{k}\left(
\theta\right)  $ and $\widetilde{\overline{\alpha}}_{k}\left(  \theta\right)
$ respectively to be the solution to%
\begin{align}
0  &  =E_{\theta}\left[  \partial H\left(  b^{\ast}\left(  X,\overline{\eta
}_{k}\right)  ,p^{\ast}\left(  X,\overline{\alpha}_{k}\right)  \right)
/\partial\overline{\alpha}_{k}\right]  =E_{\theta}\left[  \left\{
H_{1}b^{\ast}\left(  X,\overline{\eta}_{k}\right)  +H_{3}\right\}  \dot
{P}\overline{Z}_{k}\right] \label{eta}\\
0  &  =E_{\theta}\left[  \partial H\left(  b^{\ast}\left(  X,\overline{\eta
}_{k}\right)  ,p^{\ast}\left(  X,\overline{\alpha}_{k}\right)  \right)
/\partial\overline{\eta}_{k}\right]  =E_{\theta}\left[  \left\{  H_{1}p^{\ast
}\left(  X,\overline{\alpha}_{k}\right)  +H_{2}\right\}  \dot{B}\overline
{Z}_{k}\right]  .\label{alpha}%
\end{align}

The solution to $\left(  \text{\ref{eta}}\right)  $ and $\left(
\text{\ref{alpha}}\right)  $ exist in closed form as
\begin{align}
\widetilde{\overline{\eta}}_{k}\left(  \theta\right)   &  =-E_{\theta}\left[
\dot{B}\dot{P}H_{1}\overline{Z}_{k}\overline{Z}_{k}^{T}\right]  ^{-1}%
E_{\theta}\left[  \overline{Z}_{k}\dot{P}\left\{  H_{1}\widehat{B}%
+H_{3}\right\}  \right] \label{eta1}\\
\widetilde{\overline{\alpha}}_{k}\left(  \theta\right)   &  =-E_{\theta
}\left[  \dot{P}\dot{B}H_{1}\overline{Z}_{k}\overline{Z}_{k}^{T}\right]
^{-1}E_{\theta}\left[  \overline{Z}_{k}\dot{B}\left\{  H_{1}\widehat{P}%
+H_{2}\right\}  \right]  .\label{alpha1}%
\end{align}

Next define $\widetilde{b}\left(  \theta\right)  =\widetilde{b}\left(
\cdot,\theta\right)  =b^{\ast}\left(  \cdot,\widetilde{\overline{\eta}}%
_{k}\left(  \theta\right)  \right)  $ \ and $\widetilde{p}\left(
\theta\right)  =\widetilde{p}\left(  \cdot,\theta\right)  =p^{\ast}\left(
\cdot,\widetilde{\overline{\alpha}}_{k}\left(  \theta\right)  \right)  \ $and
\[
\widetilde{\psi}_{k}\left(  \theta\right)  =E_{\theta}\left[  H\left(
\widetilde{b}\left(  \theta\right)  ,\widetilde{p}\left(  \theta\right)
\right)  \right]
\]

Note the models $p^{\ast}\left(  \cdot,\overline{\alpha}_{k}\right)  $ and
$b^{\ast}\left(  \cdot,\overline{\eta}_{k}\right)  $ are used only to define
the truncated parameter $\widetilde{\psi}_{k}\left(  \theta\right)  .$ They
are not assumed to be correctly specified. In particular, the training sample
estimates $\widehat{p},$ $\widehat{b}$ need not be based on the models
$p^{\ast}\left(  \cdot,\overline{\alpha}_{k}\right)  $,$b^{\ast}\left(
\cdot,\overline{\eta}_{k}\right)  .$We now compare our truncated parameter
$\widetilde{\psi}_{k}\left(  \theta\right)  $ with $\psi\left(  \theta\right)
$ and calculate the truncation bias. It is important to keep in mind that
$b,p$ are components of the unknown $\theta\ $while $\dot{p}$,$\dot
{b},\widehat{p},\widehat{b}\ \ $are regarded as known functions.

\begin{theorem}
If our model satisfies $Ai)-Aiii)$ $\ $and
\[
\theta\in\Theta_{sub,k}\ =\left\{  \theta;p\left(  \cdot\right)  =p^{\ast
}\left(  \cdot,\overline{\alpha}_{k}\right)  \text{ for some }\overline
{\alpha}_{k}\text{ or }b\left(  \cdot\right)  =b^{\ast}\left(  \cdot
,\overline{\eta}_{k}\right)  \text{ for some }\overline{\eta}_{k}\right\}
\cap\Theta
\]
then $\widetilde{\psi}_{k}\left(  \theta\right)  =\psi\left(  \theta\right)  $

Further $TB_{k}\left(  \theta\right)  \ =\widetilde{\psi}_{k}\ \left(
\theta\right)  -\psi\left(  \theta\right)  \ =E_{\theta\ }\left[  \left\{
\widetilde{B}\left(  \theta\right)  \ -B\right\}  \left\{  \widetilde{P}%
\left(  \theta\right)  \ -P\right\}  H_{1}\right]  $

\begin{proof}
Immediate from Theorem \ref{DR} and Lemma \ref{condE}.
\end{proof}
\end{theorem}

We know from the above Theorem that $TB_{k}\left(  \theta\right)  =0$ for
$\theta\in\Theta_{sub,k}.$ However to control the truncation bias in forming
confidence intervals for $\psi\left(  \theta\right)  $ we will need to know
how fast sup$_{\theta\in\Theta}\left\{  TB_{k\ }\left(  \theta\right)
\right\}  $ decreases as $k$ increases$.\ $The following theorem is a key step
towards determining an upper bound.

\begin{theorem}
\label{TBformula}Suppose $\dot{b}\left(  X\right)  $ and $\dot{p}\left(
X\right)  $ are chosen so that $\dot{B}\dot{P}E\left[  H_{1}|X\right]  \geq0$
wp1. Let
\[
Q\equiv q\left(  X\right)  =\left\{  \dot{B}\dot{P}E\left[  H_{1}|X\right]
\right\}  ^{1/2}%
\]
and $\Pi\left[  h\left(  Z\right)  |Q\overline{Z}_{k}\right]  $ and
$\Pi^{\perp}\left[  h\left(  X\right)  |Q\overline{Z}_{k}\right]  $ be,
respectively, the projection in $L_{2}\left(  F_{X}\left(  x\right)  \right)
$ of $h\left(  X\right)  $ on the $k$ dimensional linear subspace $lin\left\{
Q\overline{Z}_{k}\right\}  $ spanned by the components of the vector
$Q\overline{Z}_{k}=q\left(  X\right)  \overline{z}_{k}\left(  X\right)  $ and
the projection on the orthocomplement of this subspace. Then if $Ai)-Aiii)$
are satisfied $,$%
\[
TB_{k}=E\left[  \Pi^{\perp}\left[  \left(  \frac{P-\widehat{P}}{\dot{P}%
}\right)  Q|Q\overline{Z}_{k}\right]  \Pi^{\perp}\left[  \left(
\frac{B-\widehat{B}}{\dot{B}}\right)  Q|Q\overline{Z}_{k}\right]  \right]
\]

\begin{remark}
\label{Q} To simplify various formulae it is often convenient and
aesthetically pleasing to have $\widehat{Q}=1.$ We can choose $\dot{B}$ and
$\dot{P}$ to guarantee $\widehat{Q}=1$ wp1$.$ For the functional $\psi\left(
\theta\right)  =E_{\theta}\left[  b\left(  X\right)  p\left(  X\right)
\right]  $ of Example 1a, $H_{1}=-1$ wp1. Thus choosing $\dot{B}$ and $\dot
{P}$ equal to $1$ and $-1$, respectively, $wp1$ makes $\widehat{Q}=1$ wp1. In
the missing data example 2a, the function $H_{1}=-A$ so $\widehat{E}\left[
H_{1}|X\right]  =1/\widehat{P}$ and thus the choice $\dot{B}=-1,\dot
{P}=\widehat{P}\ $makes $\widehat{Q}=1$ wp1. Note since inference on
$\psi\left(  \theta\right)  $ is conditional on the training sample data, we
view the initial estimator $\widehat{p}\left(  \cdot\right)  $ of $p\left(
\cdot\right)  $ from the training sample as known and thus an analyst is free
to choose $\dot{P}$ to be $\widehat{P}.$
\end{remark}
\end{theorem}

\textbf{\ Examples continued.} In \textbf{Example} \textbf{1a, }recall
$\psi=E\left[  BP\right]  .$ Choose $\dot{B}=-\dot{P}=1$ $wp1$ so
$\widehat{Q}=Q=1,$ and take $\widehat{B}\in lin\left\{  \overline{Z}%
_{k}\right\}  .$ Then
\begin{align*}
\widetilde{B}  &  =\widehat{B}+\Pi\left[  \left(  B-\widehat{B}\right)
|\overline{Z}_{k}\right]  =\Pi\left[  B|\overline{Z}_{k}\right] \\
\widetilde{P}  &  =\Pi\left[  P|\overline{Z}_{k}\right]  ,\\
TB_{k}  &  =E\left\{  \left[  \Pi^{\perp}\left[  B|\overline{Z}_{k}\right]
\Pi^{\perp}\left[  P|\overline{Z}_{k}\right]  \right]  \right\}  ,\\
\widetilde{\psi}_{k}  &  =\psi-TB_{k}=E\left\{  \Pi\left[  B|\overline{Z}%
_{k}\right]  \Pi\left[  P|\overline{Z}_{k}\right]  \right\}
\end{align*}

Thus $\widetilde{\psi}_{k}$ appears to be the natural choice for a truncated parameter.

In the \textbf{Example 2a} with $\psi=E\left[  B\right]  ,\dot{B}=-1,\dot
{P}=\widehat{P}=1/\widehat{\pi},$ $\widehat{Q}=1,Q=\left\{  \widehat{P}%
/P\right\}  ^{1/2}=\left\{  \frac{\pi}{\widehat{\pi}}\right\}  ^{1/2}%
,\widehat{\pi}\equiv\widehat{\pi}\left(  X\right)  ,\pi\equiv\pi\left(
X\right)  ,$ we obtain%
\[
TB_{k}=E\left[
\begin{array}
[c]{c}%
\Pi^{\perp}\left[  \widehat{\pi}\left(  \frac{1}{\pi}-\frac{1}{\widehat{\pi}%
}\right)  \left\{  \frac{\pi}{\widehat{\pi}}\right\}  ^{1/2}|\left\{
\frac{\pi}{\widehat{\pi}}\right\}  ^{1/2}\overline{Z}_{k}\right] \\
\times\Pi^{\perp}\left[  \left\{  \frac{\pi}{\widehat{\pi}}\right\}
^{1/2}\left(  B-\widehat{B}\right)  |\left\{  \frac{\pi}{\widehat{\pi}%
}\right\}  ^{1/2}\overline{Z}_{k}\right]
\end{array}
\right]
\]

Thus the truncated parameter $\widetilde{\psi}_{k}=\psi-TB_{k}$ does not seem
to be a particular natural or obvious choice. The complexity of
$\widetilde{\psi}_{k}$ is not simply due to the fact that we chose $\dot
{P}=\widehat{P}$ rather than $\dot{P}=1$ as we now demonstrate.

In \textbf{Example 2a} with $\dot{B}=-1,\dot{P}=1\ ,\widehat{Q}=\widehat{\pi
}^{1/2},$ $Q=\pi^{1/2},$%
\[
TB_{k}=E\left[
\begin{array}
[c]{c}%
\Pi^{\perp}\left[  \left(  \frac{1}{\pi}-\frac{1}{\widehat{\pi}}\right)
\pi^{1/2}|\pi^{1/2}\overline{Z}_{k}\right]  \times\\
\Pi^{\perp}\left[  \left\{  \frac{\pi}{\widehat{\pi}}\right\}  ^{1/2}\left(
B-\widehat{B}\right)  |\pi^{1/2}\overline{Z}_{k}\right]
\end{array}
\right]
\]

Nonetheless we will see that, for either choice of $\left(  \dot{B},\dot
{P}\right)  ,$ the parameter $\widetilde{\psi}_{k}$ will result in estimators
with good properties. \ 

\begin{remark}
Henceforth, given $\left(  \beta_{p},\beta_{b},\beta_{g}\right)  ,$ $\left\{
\varphi_{l}\left(  X\right)  ,l=1,2,...\right\}  $ will always denote a
complete orthonormal basis wrt to Lebesgue measure in $R^{d}$ or in the unit
cube in $R^{d}$ that provides optimal rate approximation for H\"{o}lder balls
$H\left(  \beta^{\ast},C\right)  ,\beta^{\ast}\leq\lceil\max\left(  \beta
_{p},\beta_{b},\beta_{g}\right)  \rceil$, i.e.%
\begin{equation}
sup_{h\in H\left(  \beta^{\ast},C\right)  }inf_{\varsigma_{_{l}}}\int_{R^{d}%
}\left(  h\left(  x\right)  -\sum_{l=1}^{k}\varsigma_{_{l}}\varphi_{_{l}%
}\ \left(  x\right)  \right)  ^{2}dx=O\left(  k^{-2\beta^{\ast}/d}\right)
\label{appr}%
\end{equation}
The basis consisting of $d-fold$ tensor products of univariate orthonormal
polynomials satisfies $\left(  \text{\ref{appr}}\right)  $ for all
$\beta^{\ast}.$ The basis consisting of $d-fold$ tensor products of a
univariate Daubechies compact wavelet basis with mother wavelet $\varphi
_{w}\left(  u\right)  $ satisfying
\[
\int_{R^{1}}u^{m}\varphi_{w}\left(  u\right)  du=0,m=0,1,...,M
\]
also satisfies $\left(  \text{\ref{appr}}\right)  $ for $\beta^{\ast}<M+1.$
\end{remark}

\begin{theorem}
\label{TBrate}Suppose that $Ai)-Aiv)$ are satisfied, that $\dot{b}\left(
X\right)  $ and $\dot{p}\left(  X\right)  $ satisfy $\left(  \text{\ref{dot1}%
}\right)  -\left(  \text{\ref{dot3}}\right)  $ and in the remainder of the
paper, unless stated otherwise, we take
\begin{equation}
\overline{z}_{k}\left(  X\right)  \equiv E\left\{  \widehat{Q}^{2}%
\overline{\varphi}_{k}\left(  X\right)  \overline{\varphi}_{k}\left(
X\right)  ^{T}\right\}  ^{-1/2}\overline{\varphi}_{k}\left(  X\right)
\label{ON}%
\end{equation}
where recall that $\widehat{Q}^{2}=\left\{  \dot{B}\dot{P}\widehat{E}\left[
H_{1}|X\right]  \right\}  .$ Then
\[
sup_{\theta\in\Theta}\left\{  TB_{k}^{2}\left(  \theta\right)  \right\}
=O_{p}\left(  k^{-2\left(  \beta_{b}+\beta_{p}\right)  /d}\right)
\]

\end{theorem}

\subsubsection{Derivation of the Higher Order Influence Functions of the
Truncated Parameter\label{dhoif_subsection}}

We begin by proving that the first order influence functions of
$\widetilde{\psi}_{k}\ $and $\psi$ are identical except with $\widetilde{b}%
\left(  \theta\right)  ,\widetilde{p}\left(  \theta\right)  ,\widetilde{\psi
}_{k}\left(  \theta\right)  $ replacing $b,p,\psi\left(  \theta\right)  .$

\begin{theorem}
\label{FOIF}%
\[
\mathbb{IF}_{1,\widetilde{\psi}_{k}}\left(  \theta\right)  =\mathbb{V}\left[
IF_{1,\widetilde{\psi}_{k},i_{1}}\,\left(  \theta\right)  \right]
\]
with%
\[
IF_{1,\widetilde{\psi}_{k}}\left(  \theta\right)  =H\left(  \widetilde{b}%
\left(  \theta\right)  ,\widetilde{p}\left(  \theta\right)  \right)
-\widetilde{\psi}_{k}\left(  \theta\right)
\]

\end{theorem}

\begin{proof}
Since $\widetilde{\psi}_{k}\left(  \theta\right)  =E_{\theta}\left[  H\left(
\widetilde{b}\left(  \theta\right)  ,\widetilde{p}\left(  \theta\right)
\right)  \right]  $,
\begin{align*}
IF_{1,\widetilde{\psi}_{k}}\left(  \theta\right)   &  =H\left(  \widetilde{b}%
\left(  \theta\right)  ,\widetilde{p}\left(  \theta\right)  \right)
-\widetilde{\psi}_{k}\left(  \theta\right) \\
&  +E\left[  \partial H\left(  b^{\ast}\left(  X,\widetilde{\overline{\eta}%
}_{k}\left(  \theta\right)  \right)  ,p^{\ast}\left(  X,\widetilde{\overline
{\alpha}}_{k}\left(  \theta\right)  \right)  \right)  /\partial\overline{\eta
}_{k}^{T}\right]  IF_{1,\widetilde{\overline{\eta}}_{k}\left(  \cdot\right)
}\left(  \theta\right) \\
&  +E\left[  H\left(  b^{\ast}\left(  X,\widetilde{\overline{\eta}}_{k}\left(
\theta\right)  \right)  ,p^{\ast}\left(  X,\widetilde{\overline{\alpha}}%
_{k}\left(  \theta\right)  \right)  \right)  /\partial\overline{\alpha}%
_{k}^{T}\right]  IF_{1,\widetilde{\overline{\alpha}}_{k}\left(  \cdot\right)
}\left(  \theta\right)
\end{align*}
But, by definition of $\widetilde{\overline{\eta}}_{k}\left(  \theta\right)  $
and $\widetilde{\overline{\alpha}}_{k}\left(  \theta\right)  ,$ both
expectations are zero.
\end{proof}

Note that $\widetilde{\overline{\eta}}_{k}\left(  \theta\right)  $ and
$\widetilde{\overline{\alpha}}_{k}\left(  \theta\right)  $ are not maximizers
of the expected log-likelihood for $\overline{\alpha}_{k}$ and $\overline
{\eta}_{k}.$ This choice was deliberate. Had we defined $\widetilde{\overline
{\eta}}_{k}\left(  \theta\right)  $ and $\widetilde{\overline{\alpha}}%
_{k}\left(  \theta\right)  $ as the maximizers of the expected log-likelihood,
then $\mathbb{IF}_{1,\widetilde{\psi}_{k}}\left(  \theta\right)  $ would have
had additional terms since the expectations in the preceding proof would not
be zero. The existence of these extra terms would translate to many extra
terms in $\mathbb{IF}_{m,\widetilde{\psi}_{k}}\left(  \theta\right)  $ for
large $m$ leading to computational difficulties. Similarly had we chosen
models $p^{\ast}\left(  X,\overline{\alpha}_{k}\right)  \equiv\Phi\left(
\widehat{P}+\dot{P}\overline{\alpha}_{k}^{T}\overline{Z}_{k}\right)  $ and
$b^{\ast}\left(  X,\overline{\eta}_{k}\right)  =\Phi\left(  \widehat{B}%
+\dot{B}\overline{\eta}_{k}^{T}\overline{Z}_{k}\right)  $ with $\Phi\left(
\cdot\right)  $ a non-linear inverse-link function, $\mathbb{IF}%
_{m,\widetilde{\psi}_{k}}\left(  \theta\right)  $ would also have had many
extra terms without an improvement in the rate of convergence.

The following is proved in the Appendix.

\begin{theorem}
\label{DRHOIF}$\mathbb{IF}_{m,\widetilde{\psi}_{k}}\ =\mathbb{IF}%
_{1,\widetilde{\psi}_{k}}\ +\sum_{j=2}^{m}\mathbb{IF}_{jj,\widetilde{\psi}%
_{k}}\ $ where $\mathbb{IF}_{jj,\widetilde{\psi}_{k}}$=$\mathbb{V}\left[
IF_{jj,\widetilde{\psi}_{k},\overline{i}_{j}}\,\right]  $ is a jth order
degenerate $U-$statistic given by
\begin{align*}
IF_{22,\widetilde{\psi}_{k},\overline{i}_{2}}  &  =-\left\{
\begin{array}
[c]{c}%
\left[  \left(  H_{1}\widetilde{P}+H_{2}\right)  \dot{B}\overline{Z}_{k}%
^{T}\right]  _{i_{1}}\left\{  E\left[  \dot{P}\dot{B}H_{1}\overline{Z}%
_{k}\overline{Z}_{k}^{T}\right]  \right\}  ^{-1}\\
\times\left[  \overline{Z}_{k}\left(  H_{1}\widetilde{B}+H_{3}\right)  \dot
{P}\right]  _{i_{2}}%
\end{array}
\right\} \\
IF_{jj,\widetilde{\psi}_{k},\overline{i}_{j}}  &  =\left(  -1\right)
^{j-1}\left[  \left(  H_{1}\widetilde{P}+H_{2}\right)  \dot{B}\overline{Z}%
_{k}^{T}\right]  _{i_{1}}\\
\times &  \left[
\begin{array}
[c]{c}%
%TCIMACRO{\dprod \limits_{s=3}^{j}}%
%BeginExpansion
{\displaystyle\prod\limits_{s=3}^{j}}
%EndExpansion
\left\{  E\left[  \dot{P}\dot{B}H_{1}\overline{Z}_{k}\overline{Z}_{k}%
^{T}\right]  \right\}  ^{-1}\\
\left\{  \left(  \dot{P}\dot{B}H_{1}\overline{Z}_{k}\overline{Z}_{k}%
^{T}\right)  _{i_{s}}-E\left[  \dot{P}\dot{B}H_{1}\overline{Z}_{k}\overline
{Z}_{k}^{T}\right]  \right\}
\end{array}
\right] \\
&  \times\left\{  E\left[  \dot{P}\dot{B}H_{1}\overline{Z}_{k}\overline{Z}%
_{k}^{T}\right]  \right\}  ^{-1}\left[  \overline{Z}_{k}\left(  H_{1}%
\widetilde{B}+H_{3}\right)  \dot{P}\right]  _{i_{2}}%
\end{align*}

\end{theorem}

\subsubsection{The Estimator $\protect\widehat{\psi}_{m,k}\equiv
\protect\widehat{\psi}+\protect\widehat{\mathbb{IF}}%
_{m,\protect\widetilde{\psi}_{k}}$ and its Estimation Bias}

We can now calculate the estimator $\widehat{\psi}_{m,k}\equiv\widehat{\psi
}+\widehat{\mathbb{IF}}_{m,\widetilde{\psi}_{k}}\ $by substitution of
$\widehat{\theta}$ for $\theta$ in $\mathbb{IF}_{m,\widetilde{\psi}_{k}}%
\equiv\mathbb{IF}_{m,\widetilde{\psi}_{k}}\left(  \theta\right)  $ to obtain
the following.

\begin{theorem}
Suppose $\left(  \ref{ON}\right)  $ holds . Then $\widehat{\psi}%
_{m,k}=\widehat{\psi}+\widehat{\mathbb{IF}}_{1,\widetilde{\psi}_{k}}%
+\sum_{j=2}^{m}\widehat{\mathbb{IF}}_{jj,\widetilde{\psi}_{k}}$ where
\begin{align*}
\widehat{\psi}+\widehat{\mathbb{IF}}_{1,\widetilde{\psi}_{k}} &
=\widehat{B}\widehat{P}H_{1}+\widehat{B}H_{2}+\widehat{P}H_{3}+H_{4}\\
\widehat{IF}_{22,\widetilde{\psi}_{k},\overline{i}_{2}} &  =-\left[  \left(
H_{1}\widehat{P}+H_{2}\right)  \dot{B}\overline{Z}_{k}^{T}\right]  _{i_{1}%
}\left[  \overline{Z}_{k}\left(  H_{1}\widehat{B}+H_{3}\right)  \dot
{P}\right]  _{i_{2}}\\
\widehat{IF}_{jj,\widetilde{\psi}_{k},\overline{i}_{j}} &  =\left(  -1\right)
^{j-1}\left\{
\begin{array}
[c]{c}%
\left[  \left(  H_{1}\widehat{P}+H_{2}\right)  \dot{B}\overline{Z}_{k}%
^{T}\right]  _{i_{1}}\left[
%TCIMACRO{\dprod \limits_{s=3}^{j}}%
%BeginExpansion
{\displaystyle\prod\limits_{s=3}^{j}}
%EndExpansion
\left\{
\begin{array}
[c]{c}%
\left(  \dot{P}\dot{B}H_{1}\overline{Z}_{k}\overline{Z}_{k}^{T}\right)
_{i_{s}}\\
-I_{k\times k}%
\end{array}
\right\}  \right] \\
\times\left[  \overline{Z}_{k}\left(  H_{1}\widehat{B}+H_{3}\right)  \dot
{P}\right]  _{i_{2}}%
\end{array}
\right\}
\end{align*}

\end{theorem}

\begin{proof}
By Lemma \ref{condE} $E_{\widehat{\theta}}\left[  \left\{  H_{1}%
\widehat{B}+H_{3}\right\}  \dot{P}\overline{Z}_{k}\right]  =E_{\widehat{\theta
}}\left[  \left\{  H_{1}\widehat{P}+H_{2}\right\}  \dot{B}\overline{Z}%
_{k}\right]  =0.$ Thus by Eqs. $\left(  \ref{eta1}\right)  $and $\left(
\ref{alpha1}\right)  $ $\widetilde{\overline{\eta}}_{k}\left(  \widehat{\theta
}\right)  $= $\widetilde{\overline{\alpha}}_{k}\left(  \widehat{\theta
}\right)  =0$ so $\widetilde{B}\left(  \widehat{\theta}\right)  =\widehat{B}$
\ and $\widetilde{P}\left(  \widehat{\theta}\right)  =\widehat{P}$ . Further,
by Eq.$\left(  \ref{ON}\right)  ,$ $\widehat{E}\left[  \dot{P}\dot{B}%
H_{1}\overline{Z}_{k}\overline{Z}_{k}^{T}\right]  =\widehat{E}\left[  \dot
{P}\dot{B}\widehat{E}\left[  H_{1}|X\right]  \overline{Z}_{k}\overline{Z}%
_{k}^{T}\right]  =\widehat{E}\left[  \widehat{Q}^{2}\overline{Z}_{k}%
\overline{Z}_{k}^{T}\right]  =\int\overline{\varphi}_{k}\left(  x\right)
\overline{\varphi}_{k}\left(  x\right)  ^{T}=I_{k\times k}$ .
\end{proof}

It follows that by our judicious choice of $\overline{Z}_{k}$ in Eq.$\left(
\ref{ON}\right)  ,$ we have avoided the need to invert a $k\times k$ matrix to
compute $\widehat{\psi}_{m,k}.$

\begin{remark}
The reader can easily check that when we take
\[
\overline{Z}_{k}=\overline{\varphi}_{k}\left(  X\right)  /\left\{
\widehat{f_{X}}\left(  X\right)  \widehat{E}\left[  H_{1}|X\right]  \right\}
^{1/2}%
\]
$\ \dot{B}=\dot{P}=1$ and $H_{1}\geq0$ wp 1$,$ $\widehat{IF}%
_{j,j,\widetilde{\psi}_{k},\overline{i}_{2}}$ is precisely the same as
$IF_{j,j,\psi,i_{1},i_{2}}^{\left(  k\right)  }\left(  \widehat{\theta
}\right)  $ of equation {}(\ref{22}) in Section 3.2.1.
\end{remark}

To make our procedures less abstract, we provide explicit expressions for
$\widehat{IF}_{jj,\widetilde{\psi}_{k},\overline{i}_{j}}$ in examples 1a and 2a.

\textbf{Example 1a} \textbf{continued: } $\psi=E\left[  BP\right]  ,$ $\dot
{B}=-\dot{P}=1$ $wp1,$ $\widehat{Q}=1,H_{1}=-1,$ $\widehat{E}\left[
\overline{Z}_{k}\overline{Z}_{ki_{s}}^{T}\right]  =I_{k\times k}.$ Then
\[
\dot{P}\left\{  H_{1}\widehat{B}+H_{3}\right\}  =Y-\widehat{B},\dot{B}\left\{
H_{1}\widehat{P}+H_{2}\right\}  =A-\widehat{P}%
\]
and thus
\begin{align*}
\widehat{IF}_{22,\widetilde{\psi}_{k},\overline{i}_{2}} &  =-\left[  \left(
A-\widehat{P}\right)  \overline{Z}_{k}^{T}\right]  _{i_{1}}\left[
\overline{Z}_{k}\left(  Y-\widehat{B}\right)  \right]  _{i_{2}}\\
\widehat{IF}_{jj,\widetilde{\psi}_{k},\overline{i}_{j}} &  =\left(  -1\right)
^{j-1}\left[  \left(  A-\widehat{P}\right)  \overline{Z}_{k}^{T}\right]
_{i_{1}}\left[
%TCIMACRO{\dprod \limits_{s=3}^{j}}%
%BeginExpansion
{\displaystyle\prod\limits_{s=3}^{j}}
%EndExpansion
\left\{  \overline{Z}_{ki_{s}}\overline{Z}_{ki_{s}}^{T}-I_{k\times k}\right\}
\right]  \left[  \overline{Z}_{k}\left(  Y-\widehat{B}\right)  \right]
_{i_{2}}%
\end{align*}

\textbf{Example 2a continued:} $H_{1}=-A,$ $\dot{B}=-1,\dot{P}=\widehat{P}%
=1/\widehat{\pi},\widehat{Q}=1$, $\psi=E\left[  B\right]  ,$ $Q=\left\{
\widehat{P}/P\right\}  ^{1/2}=\left\{  \frac{\pi}{\widehat{\pi}}\right\}
^{1/2} $ and $\overline{Z}_{k}=\overline{\varphi}_{k}\left\{  \widehat{f}%
\left(  X\right)  \right\}  ^{-1/2},$ so $\widehat{E}\left[  \overline{Z}%
_{k}\overline{Z}_{ki_{s}}^{T}\right]  =I_{k\times k}.$

Then $\dot{P}\left\{  H_{1}\widehat{B}+H_{3}\right\}  =\frac{A}{\widehat{\pi}%
}\left(  Y-\widehat{B}\right)  ,$ $\dot{B}\left\{  H_{1}\widehat{P}%
+H_{2}\right\}  =\left(  \frac{A}{\widehat{\pi}}-1\right)  ,$ so%
\begin{align*}
\widehat{IF}_{22,\widetilde{\psi}_{k},\overline{i}_{2}}  &  =-\left[  \left(
\frac{A}{\widehat{\pi}}-1\right)  \overline{Z}_{k}^{T}\right]  _{i_{1}}\left[
\overline{Z}_{k}\frac{A}{\widehat{\pi}}\left(  Y-\widehat{B}\right)  \right]
_{i_{2}}\\
\widehat{IF}_{jj,\widetilde{\psi}_{k},\overline{i}_{j}}  &  =\left(
-1\right)  ^{j-1}\left[  \left(  \frac{A}{\widehat{\pi}}-1\right)
\overline{Z}_{k}^{T}\right]  _{i_{1}}\left[
%TCIMACRO{\dprod \limits_{s=3}^{j}}%
%BeginExpansion
{\displaystyle\prod\limits_{s=3}^{j}}
%EndExpansion
\left\{  \frac{A}{\widehat{\pi}}\overline{Z}_{k}\overline{Z}_{k}%
^{T}-I_{k\times k}\right\}  _{i_{s}}\right]  \left[  \overline{Z}_{k}\frac
{A}{\widehat{\pi}}\left(  Y-\widehat{B}\right)  \right]  _{i_{2}}%
\end{align*}

Consider Example 2a with $\dot{B}=-1,\dot{P}=1\ ,\widehat{Q}=\widehat{\pi
}^{1/2}$, $\widehat{E}\left[  \widehat{Q}^{2}\overline{Z}_{k}\overline
{Z}_{ki_{s}}^{T}\right]  =I_{k\times k},$ ${}\dot{P}\left\{  H_{1}%
\widehat{B}+H_{3}\right\}  =\left[  A\left(  Y-\widehat{B}\right)  \right]  ,$
$\dot{B}\left\{  H_{1}\widehat{P}+H_{2}\right\}  =\left(  \frac{A}%
{\widehat{\pi}}-1\right)  ,$ so
\begin{align*}
\widehat{IF}_{22,\widetilde{\psi}_{k},\overline{i}_{2}} &  =-\left[  \left(
\frac{A}{\widehat{\pi}}-1\right)  \overline{Z}_{k}^{T}\right]  _{i_{1}}\left[
\overline{Z}_{k}A\left(  Y-\widehat{B}\right)  \right]  _{i_{2}}\\
\widehat{IF}_{jj,\widetilde{\psi}_{k},\overline{i}_{j}} &  =\left(  -1\right)
^{j-1}\left[  \left(  \frac{A}{\widehat{\pi}}-1\right)  \overline{Z}_{k}%
^{T}\right]  _{i_{1}}\left[
%TCIMACRO{\dprod \limits_{s=3}^{j}}%
%BeginExpansion
{\displaystyle\prod\limits_{s=3}^{j}}
%EndExpansion
\left\{  A\overline{Z}_{k}\overline{Z}_{k}^{T}-I_{k\times k}\right\}  _{i_{s}%
}\right]  \left[  \overline{Z}_{k}A\left(  Y-\widehat{B}\right)  \right]
_{i_{2}}%
\end{align*}

Our next theorem, proved in the appendix, derives the estimation bias
$EB_{m}=E\left[  \widehat{\psi}_{m,k}\right]  -\widetilde{\psi}_{k}.$

\begin{theorem}
\label{EBrate}Suppose $\left(  \text{\ref{dot1}}\right)  -\left(
\text{\ref{dot3}}\right)  $ and $Ai)-Aiv)$ hold then%
\begin{gather}
EB_{m}=\left(  -1\right)  ^{m-1}\left\{
\begin{array}
[c]{c}%
E\left[  Q^{2}\left(  \frac{B-\widehat{B}}{\dot{B}}\right)  \overline{Z}%
_{k}^{T}\right]  \left\{  E\left[  Q^{2}\overline{Z}_{k}\overline{Z}_{k}%
^{T}\right]  -I_{k\times k}\right\}  ^{m-1}\\
\times\left\{  E\left[  Q^{2}\overline{Z}_{k}\overline{Z}_{k}^{T}\right]
\right\}  ^{-1}E\left[  \overline{Z}_{k}Q^{2}\left(  \frac{P-\widehat{P}}%
{\dot{P}}\right)  \right]
\end{array}
\right\} \label{EB3}\\
|EB_{m}|\nonumber\\
\leq\left\{
\begin{array}
[c]{c}%
\left\vert \left\vert \left\{  \frac{\dot{B}}{\dot{P}}G\right\}
^{1/2}\right\vert \right\vert _{\infty}\left\vert \left\vert \left\{
\frac{\dot{P}}{\dot{B}}G\right\}  ^{1/2}\right\vert \right\vert _{\infty
}\left\vert \left\vert \delta g\right\vert \right\vert _{\infty}^{m-1}\left(
1+o_{p}\left(  1\right)  \right)  \times\\
\left\{  \int\left(  p\left(  X\right)  -\widehat{p}\left(  X\right)  \right)
^{2}dX\right\}  ^{1/2}\left\{  \int\left(  b\left(  X\right)  -\widehat{b}%
\left(  X\right)  \right)  ^{2}dX\right\}  ^{1/2}%
\end{array}
\right\} \label{EB4}\\
=O_{P}\left(  \left(  \frac{\log n}{n}\right)  ^{\frac{\left(  m-1\right)
\beta_{g}}{d+2\beta_{g}}}n^{-\left(  \frac{\beta_{b}}{d+2\beta_{b}}%
+\frac{\beta_{p}}{d+2\beta_{p}}\right)  }\right) \label{EB6}%
\end{gather}

for $m\geq1,$\ where $\delta g=\frac{g\left(  X\right)  -\widehat{g}\left(
X\right)  }{\widehat{g}\left(  X\right)  }$
\end{theorem}

\begin{remark}
\label{pnorm}At the cost of a longer proof we could have used H\"{o}lder's
inequality repeatedly to control $\delta g$ in the $L_{p}$ norm $\left\vert
\left\vert \delta g\right\vert \right\vert _{m+1}$ with $p=m+1$ to show that
$|EB_{m}|=O_{P}\left(  \left\vert \left\vert \delta g\right\vert \right\vert
_{m+1}^{m-1}\left\vert \left\vert b\left(  \cdot\right)  -\widehat{b}\left(
\cdot\right)  \right\vert \right\vert _{m+1}\left\vert \left\vert p\left(
\cdot\right)  -\widehat{p}\left(  \cdot\right)  \right\vert \right\vert
_{m+1}\right)  .$ Thus, $|EB_{m}|$ is $O_{P}\left(  \left\vert \left\vert
\theta-\widehat{\theta}\ \right\vert \right\vert ^{m+1}\right)  ,$ consistent
with the form of the bias given in our fundamental theorem \ref{eiet}.
\end{remark}

\begin{remark}
\textbf{An alternate derivation of }$\widehat{\psi}_{m,k}.$ The above
derivation of $\widehat{\psi}_{m,k}$ required that one have facility in
calculating higher order influence functions $\mathbb{IF}_{m,\widetilde{\psi
}_{k}},\ $as done in the proof of Theorem \ref{DRHOIF} in the appendix.
However, there exists an alternate derivation of $\widehat{\psi}_{m,k}$ that
does not require one learn how to calculate influence functions. Specifically,
we know from Theorems \ref{eiet} and \ref{eift} that in a (locally)
nonparametric model $\widehat{\mathbb{IF}}_{jj,\widetilde{\psi}_{k}},j\geq2$
is the unique $j^{th}$ order U-statistic that is degenerate under
$\widehat{\theta}$ and satisfies
\begin{equation}
EB_{j-1}+E\left[  \widehat{\mathbb{IF}}_{jj,\widetilde{\psi}_{k}%
}|\widehat{\theta}\right]  \equiv EB_{j}=O_{p}\left(  \left\vert \left\vert
\widehat{\theta}-\theta\right\vert \right\vert ^{j+1}\right)
\end{equation}

with $EB_{1}=E\left[  \widehat{\psi}_{1}|\widehat{\theta}\right]
-\widetilde{\psi}_{k}.$ In fact, we first derived $\widehat{\psi}_{m,k}$ by
beginning with $\widehat{\psi}_{1}=\widehat{\psi}+\widehat{\mathbb{IF}%
}_{1,\widetilde{\psi}_{k}}$, calculating $EB_{1}=E\left[  \widehat{\psi}%
_{1}|\widehat{\theta}\right]  -\widetilde{\psi}_{k},$ and then, recursively
for $j=2,...,$ finding $\widehat{\mathbb{IF}}_{jj,\widetilde{\psi}_{k}}$
satisfying the above equation. For explicit details see the Appendix.In fact
if one did not even know how to derive $\mathbb{IF}_{1,\widetilde{\psi}_{k}},$
one could begin the recursion by obtaining $\widehat{\mathbb{IF}%
}_{1,\widetilde{\psi}_{k}}$ as the unique first order U-statistic with mean
zero under $\widehat{\theta}$ satisfying $\widehat{\psi}-\widetilde{\psi}%
_{k}+E\left[  \widehat{\mathbb{IF}}_{1,\widetilde{\psi}_{k}}|\widehat{\theta
}\right]  =O_{p}\left(  \left\vert \left\vert \widehat{\theta}-\theta
\right\vert \right\vert ^{2}\right)  $.
\end{remark}

\subsubsection{The Variance of $\protect\widehat{\psi}_{m,k}\equiv
\protect\widehat{\psi}+\protect\widehat{\mathbb{IF}}%
_{m,\protect\widetilde{\psi}_{k}}$ using compact
wavelets\label{variance_subsection}}

In this section, we derive the order of the variance of $\widehat{\psi}_{m,k}
$ when the orthonormal system $\left\{  \varphi_{j}\left(  X\right)  \right\}
$ used to construct our $U$-statistics are a compact wavelet basis. First
consider the case where $X$ is univariate; without loss of generality, assume
that $\ X\sim Uniform[0,1].$ {}Because we are primarily interested in
convergence rates, the fact that $X$ may not follow the uniform distribution
will not affect the rate results given below, but can influence the size of
the constants. We use $\phi_{j}\left(  X\right)  $ in place of $\varphi
_{j}\left(  X\right)  $ to indicate univariate basis functions.

Let $k^{\ast},$ $k$ be integer powers of two with $k>k^{\ast}$\ . Denote by
$\overline{\phi}\left(  X\right)  \equiv\overline{\phi}_{1}^{k}\left(
X\right)  $ the $k-$ dimensional basis vector whose first $k^{\ast}$
components $\overline{\phi}_{1}^{k^{\ast}}\left(  X\right)  $ are the
$k^{\ast}-$vector of level $\log_{2}k^{\ast}$ scaled and translated versions
of a compactly supported 'father' wavelet (Mallat, 1998) and whose last
$k-k^{\ast}$ components $\overline{\phi}_{k^{\ast}+1}^{k}\left(  X\right)  $
are the associated compact mother wavelets between levels $\log_{2}k^{\ast}$
and $\log_{2}k$. In particular, one may use periodic wavelets, folded wavelets
or Daubechies' boundary wavelets with enough vanishing moments to obtain the
optimal approximation rate of $O\left(  k^{-2\beta/d}\right)  $ for
$\beta=\max\left(  \beta_{g},\beta_{p},\beta_{b}\right)  $ $.$ \ The
multiresolution analysis (MRA) property of wavelets allows us to decompose the
vector space spanned by the $\log_{2}\left(  k\right)  $-level father wavelets
$\mathcal{V}_{\log_{2}\left(  k\right)  }$ into the direct sum of the subspace
spanned by $\log_{2}\left(  k^{\ast}\right)  $-level father wavelets
$\mathcal{V}_{\log_{2}\left(  k^{\ast}\right)  }=\left\{  a^{T}\overline{\phi
}_{1}^{k^{\ast}}\left(  X\right)  :a\in R^{k^{\ast}}\right\}  $ and the span
of mother wavelets for each level between $\log_{2}\left(  k^{\ast}\right)  $
and $\log_{2}\left(  k\right)  -1$ which we respectively write as
\begin{align*}
\mathcal{W}_{\log_{2}\left(  k^{\ast}\right)  }  &  =\left\{  a^{T}%
\overline{\phi}_{k^{\ast}+1}^{2k^{\ast}}\left(  X\right)  :a\in R^{k^{\ast}%
}\right\}  ,\\
\mathcal{W}_{\log_{2}\left(  k_{0}\right)  +1}  &  =\left\{  a^{T}%
\overline{\phi}_{2k^{\ast}+1}^{4k^{\ast}}\left(  X\right)  :a\in R^{2k^{\ast}%
}\right\}  ,\\
&  \vdots\\
\mathcal{W}_{\log_{2}\left(  k\right)  -1}  &  =\left\{  a^{T}\overline{\phi
}_{\frac{k}{2}+1}^{k}\left(  X\right)  :a\in R^{\frac{k}{2}}\right\}  .
\end{align*}
Then for any integer $s$ with $\log_{2}\left(  k^{\ast}\right)  +1\leq s,$ we
have%
\[
\mathcal{V}_{s}=\mathcal{V}_{\log_{2}\left(  k^{\ast}\right)  }\oplus\left(
\bigoplus\limits_{v=\log_{2}\left(  k^{\ast}\right)  }^{s-1}\mathcal{W}%
_{v}\right)
\]
As $s\rightarrow\infty,$ the resulting basis system is dense in $L_{2}\left(
X\right)  $ $\left(  \text{Mallat, 1998}\right)  .$ Since, in fact, $X$ is
$d-$dimensional we require a generalization that allows for multivariate
tensor wavelet basis functions. \ In fact, suppose $X^{T}=\left(
X^{1},...,X^{d}\right)  $ is now multivariate, and we again assume $X\sim
Uniform$ on $\left[  0,1\right]  ^{d}.$ \ Given $d$ univariate vector spaces
\[
\mathcal{V}_{1,\log_{2}\left(  k\right)  },\mathcal{V}_{2,\log_{2}\left(
k\right)  },...,\mathcal{V}_{d,\log_{2}\left(  k\right)  }%
\]
respectively spanned by vectors $\overline{\phi}_{1}^{k}\left(  X^{1}\right)
,\overline{\phi}_{1}^{k}\left(  X^{2}\right)  ,...,\overline{\phi}_{1}%
^{k}\left(  X^{d}\right)  ,$ so that for $1\leq r\leq d,$%
\[
\mathcal{V}_{r,\log_{2}\left(  k^{\ast}\right)  }\subset\mathcal{V}%
_{r,\log_{2}\left(  k^{\ast}\right)  +1}\subset...\subset\mathcal{V}%
_{r,\log_{2}\left(  k\right)  -1}\subset\mathcal{V}_{r,\log_{2}\left(
k\right)  }%
\]
and
\[
\mathcal{V}_{r,\log_{2}\left(  k\right)  }=\mathcal{V}_{r,\log_{2}\left(
k^{\ast}\right)  }\oplus\left(  \bigoplus\limits_{v=\log_{2}\left(  k^{\ast
}\right)  }^{\log_{2}\left(  k\right)  -1}\mathcal{W}_{r,v}\right)
\]
One may define $d$ dimensional tensor vector spaces
\[
\mathcal{Y}_{d,\log_{2}\left(  k^{\ast}\right)  },\mathcal{Y}_{d,\log
_{2}\left(  k^{\ast}\right)  +1},...,\mathcal{Y}_{d,\log_{2}\left(  k\right)
}%
\]
such that%
\[
\mathcal{Y}_{d,\log_{2}\left(  k^{\ast}\right)  }\subset\mathcal{Y}%
_{d,\log_{2}\left(  k^{\ast}\right)  +1}\subset...\subset\mathcal{Y}%
_{d,\log_{2}\left(  k\right)  }%
\]
where for $s\geq0,$
\[
\mathcal{Y}_{d,\log_{2}\left(  k_{0}\right)  +s}\mathcal{=}\bigotimes
\limits_{1\leq r\leq d}\mathcal{V}_{r,\log_{2}\left(  k_{0}\right)  +s}%
\]
As $s\rightarrow\infty,$ the resulting tensor basis system is dense in
$L_{2}\left(  X\right)  \left(  \text{Mallat, 1998}\right)  .$ \ 

Next, suppose that we have a set of multivariate basis functions
\[
\left\{  \overline{\varphi}_{1}^{k_{j}}\left(  X\right)
,j=0,1,...,2m\right\}
\]
such that for each $k_{j},$ $\overline{\varphi}_{1}^{k_{j}}\left(  X\right)  $
spans $\bigotimes\limits_{1\leq r\leq d}\mathcal{V}_{r,\log_{2}\left(
k_{j,r}\right)  }$ where $\prod\limits_{r=1}^{d}k_{j,r}=k_{j}$. Define
$||\cdot||_{2}$ as the $L_{2}$ norm with respect to the Lebesgue measure.
\ The following theorem is key to our derivation of the order of the variance
of $\widehat{\psi}_{m,k}$

\begin{theorem}
\label{var_multi}For $m\geq0,$ \
\begin{align*}
&  \left\Vert \overline{\varphi}_{1}^{k_{1}}\left(  X_{i_{1}}\right)  ^{T}%
%TCIMACRO{\dprod \limits_{j=1}^{m}}%
%BeginExpansion
{\displaystyle\prod\limits_{j=1}^{m}}
%EndExpansion
\left\{  \overline{\varphi}_{1}^{k_{j}}\left(  X_{i_{j+1}}\right)
\overline{\varphi}_{1}^{k_{j+1}}\left(  X_{i_{j+1}}\right)  ^{T}\right\}
\overline{\varphi}_{1}^{k_{m+1}}\left(  X_{i_{m+2}}\right)  \right\Vert
_{2}^{2}\\
&  =E\left(
%TCIMACRO{\dprod \limits_{j=1}^{m+1}}%
%BeginExpansion
{\displaystyle\prod\limits_{j=1}^{m+1}}
%EndExpansion
K_{\left(  1,k_{j}\right)  }\left(  X_{i_{j}},X_{i_{j+1}}\right)  \right)
^{2}\\
&  \asymp\Pi_{j=1}^{m+1}k_{j}%
\end{align*}

\end{theorem}

The following theorem is an immediate consequence of Theorem (\ref{var_multi})
obtained by taking $k_{j}=k^{\ast}=k\ $\ (which implies we use the father
wavelets at level log$_{2}(k)\ $but no mother wavelets.)

\begin{theorem}
\label{var_if}For all $\theta\in\Theta,$
\[
Var_{\theta}\left[  \mathbb{IF}_{1,\widetilde{\psi}_{k}}\left(  \theta\right)
\right]  \asymp\frac{1}{n}%
\]%
\[
Var_{\theta}\left[  \mathbb{IF}_{jj,\widetilde{\psi}_{k}}\left(
\theta\right)  \right]  \asymp\left(  \frac{1}{n}\max\left\{  1,\left(
\frac{k}{n}\right)  ^{j-1}\right\}  \right)
\]
$,$
\[
Var_{\theta}\left[  \mathbb{IF}_{m,\widetilde{\psi}_{k}}\left(  \theta\right)
\right]  \approx Var_{\widehat{\theta}}\left[  \widehat{\mathbb{IF}%
}_{m,\widetilde{\psi}_{k}}|\widehat{\theta}\right]  \asymp\frac{1}{n}%
\max\left\{  1,\left(  \frac{k}{n}\right)  ^{m-1}\right\}
\]

\end{theorem}

We now use Theorem $\left(  \ref{var_if}\right)  $ to derive the order of the
conditional variance of $\widehat{\psi}_{m,k}$ given $\widehat{\theta}$.

\begin{theorem}
\label{3.19}If sup$_{o\in\mathcal{O}}\left\vert f\left(  o;\widehat{\theta
}\right)  -f\left(  o;\theta\right)  \right\vert \rightarrow0\ $as
$||\widehat{\theta}-\theta||\rightarrow0,$ then for a fixed $m,$
\begin{align*}
Var_{\theta}\left[  \widehat{\mathbb{\psi}}_{m,\widetilde{\psi}_{k}%
}|\widehat{\theta}\right]   &  =Var_{\theta}\left[  \widehat{\mathbb{IF}%
}_{m,\widetilde{\psi}_{k}}|\widehat{\theta}\right] \\
&  =Var_{\widehat{\theta}}\left[  \widehat{\mathbb{IF}}_{m,\widetilde{\psi
}_{k}}|\widehat{\theta}\right]  \left(  1+o_{p}\left(  1\right)  \right) \\
&  \asymp\left(  \frac{1}{n}\max\left\{  1,\left(  \frac{k}{n}\right)
^{m-1}\right\}  \right)
\end{align*}

\end{theorem}

The proof is in the appendix.

For a given $m,$ the estimator $\widehat{\psi}_{m,k_{opt}\left(  m\right)  }$
that minimizes the maximum asymptotic MSE over the model $\mathcal{M}\left(
\Theta\right)  $ $\ $defined by $Ai)-Aiv)$ among the candidates $\widehat{\psi
}_{m,k}$ uses the value $k_{opt}\left(  m\right)  \equiv k_{opt}\left(
m,n\right)  $ of $k$ that equates the order $\frac{1}{n}\max\left\{  1,\left(
\frac{k}{n}\right)  ^{m-1}\right\}  $ of $Var\left[  \widehat{\mathbb{\psi}%
}_{m,\widetilde{\psi}_{k}}|\widehat{\theta}\right]  \ $to the order
\begin{align*}
&  \max\left[  \left\{  TB_{k}\right\}  ^{2},\left\{  EB_{m}\left(
\theta\right)  \right\}  ^{2}\right]  =\\
&  \max\left[
\begin{array}
[c]{c}%
\left(  \frac{\log n}{n}\right)  ^{\frac{2\left(  m-1\right)  \beta_{g}%
}{d+2\beta_{g}}}n^{-\left(  \frac{2\beta_{b}}{d+2\beta_{b}}+\frac{2\beta_{p}%
}{d+2\beta_{p}}\right)  },\\
k^{-\frac{2\left(  \beta_{b}+\beta_{p}\right)  }{d}}%
\end{array}
\right]
\end{align*}
of the maximal squared bias. \ The estimator $\widehat{\psi}_{m_{opt},k_{opt}%
}\equiv\widehat{\psi}_{m_{opt},k_{opt}\left(  m_{opt}\right)  }$ that
minimizes the maximum asymptotic MSE over the model $\mathcal{M}\left(
\Theta\right)  $ among all candidates $\widehat{\psi}_{m,k}\ $is the estimator
$\widehat{\psi}_{m,k_{opt}\left(  m,n\right)  }$ which minimizes $\frac{1}%
{n}\max\left(  1,\left(  \frac{k_{opt}\left(  m,n\right)  }{n}\right)
^{m-1}\right)  .$

\subsubsection{Distribution Theory and Confidence Interval
Construction\label{DR_CI_Section}}

We derive a consistent estimator of the variance and give the asymptotic
distribution of $\widehat{\psi}_{m,k}\ $\ for any model and functional
satisfying $Ai)-Aiv).$ Let $z_{\alpha}$ be the upper $\alpha-$quantile of a
standard normal, i.e. a $N\left(  0,1\right)  ,$ distribution.

\begin{theorem}
\textbf{\label{HOCI}}:

a) Letting $\widehat{\mathbb{W}}_{1,\widetilde{\psi}_{k}}^{2}=\ n^{-1}%
\mathbb{V}\left[  \left\{  \widehat{IF}_{1,\widetilde{\psi}_{k},i_{1}%
}\right\}  ^{2}\right]  $,
\[
\ \widehat{\mathbb{W}}_{jj,\widetilde{\psi}_{k}}^{2}=\dbinom{n}{j}%
^{-1}\mathbb{V}\left[  \left(  \widehat{IF}_{j,j,\widetilde{\psi}_{k}\left(
\cdot\right)  }^{\left(  s\right)  }\right)  ^{2}\right]  ,
\]
for $j\geq2,$ and
\[
\widehat{\mathbb{W}}_{m,\widetilde{\psi}_{k}}^{2}=\widehat{\mathbb{W}%
}_{1,\widetilde{\psi}_{k}}^{2}+\sum_{j=2}^{m}\widehat{\mathbb{W}%
}_{jj,\widetilde{\psi}_{k}}^{2},
\]
where $\widehat{IF}_{j,j,\widetilde{\psi}_{k}\left(  \cdot\right)  }^{\left(
s\right)  }$ is the symmetric kernel of ${}\widehat{\mathbb{IF}}%
_{jj,\widetilde{\psi}_{k}\left(  \cdot\right)  }.$

we have$,$
\begin{align*}
\widehat{E}\left[  \widehat{\mathbb{W}}_{1,\widetilde{\psi}_{k}}^{2}\right]
&  =\widehat{Var}\left[  \widehat{\mathbb{IF}}_{1,\widetilde{\psi}_{k}%
}|\widehat{\theta}\right] \\
\widehat{E}\left[  \widehat{\mathbb{W}}_{jj,\widetilde{\psi}_{k}}^{2}\right]
&  =\widehat{Var}\left[  \widehat{\mathbb{IF}}_{jj,\widetilde{\psi}_{k}%
}|\widehat{\theta}\right]  ,\\
\widehat{E}\left[  \widehat{\mathbb{W}}_{m,\widetilde{\psi}_{k}}^{2}\right]
&  =\widehat{Var}\left[  \widehat{\mathbb{IF}}_{m,\widetilde{\psi}_{k}%
}|\widehat{\theta}\right]
\end{align*}
where $\widehat{Var}\left[  \cdot\right]  =Var_{\widehat{\theta}}\left[
\cdot\right]  .$

b) Conditional on the training sample,
\[
\left\{  \frac{1}{n}\max\left\{  1,\left(  \frac{k_{opt}\left(  m,n\right)
}{n}\right)  ^{m-1}\right\}  \right\}  ^{-1/2}\left\{  \widehat{\psi
}_{m,k_{opt}\left(  m,n\right)  }-E\left[  \widehat{\psi}_{m,k_{opt}\left(
m,n\right)  }|\widehat{\theta}\right]  \right\}
\]
converges uniformly for $\theta\in\Theta$ to a normal distribution with finite
variance as $n\rightarrow\infty$. The asymptotic variance is uniformly
consistently estimated by
\[
\left\{  \frac{1}{n}\max\left\{  1,\left(  \frac{k}{n}\right)  ^{m-1}\right\}
\right\}  ^{\ ^{-1}}\widehat{\mathbb{W}}_{m,\widetilde{\psi}_{k_{opt}\left(
m,n\right)  }}^{2}%
\]
Thus
\[
\left\{  \widehat{\psi}_{m,k_{opt}\left(  m,n\right)  }-E\left[
\widehat{\psi}_{m,k_{opt}\left(  m,n\right)  }|\widehat{\theta}\right]
\right\}  /\widehat{\mathbb{W}}_{m,\widetilde{\psi}_{k_{opt}\left(
m,n\right)  }}%
\]
is converging in distribution to a standard normal distribution.

c) Define the interval $C_{m,k}=\widehat{\psi}_{m,k}\pm z_{\alpha
}\widehat{\mathbb{W}}_{m,\widetilde{\psi}_{k}}.$ Suppose $k_{opt}\left(
m,n\right)  =n^{\rho_{_{opt}\left(  m,n\right)  }}.$ Then for $k^{\ast
}=n^{\rho_{\ }^{\ast}},$ $\rho_{\ }^{\ast}>\rho_{_{opt}\left(  m,n\right)
},$
\[
sup_{\theta\in\Theta}\left[  \frac{E_{\theta}\left[  \widehat{\psi}%
_{2,k^{\ast}}|\widehat{\theta}\right]  }{\sqrt{Var_{\theta}\left[
\widehat{\psi}_{2,k^{\ast}}|\widehat{\theta}\right]  }}\right]  =o_{p}\left(
1\right)
\]
and $\left\{  \widehat{\psi}_{m,k^{\ast}}-\psi\left(  \theta\right)  \right\}
/\widehat{\mathbb{W}}_{m,\widetilde{\psi}_{k^{\ast}}}$ converges uniformly in
$\theta\in\Theta$ to a $N\left(  0,1\right)  $. \ Moreover, $C_{m,k^{\ast}}$
is a conservative uniform asymptotic $\left(  1-\alpha\right)  $ confidence
interval for $\psi\left(  \theta\right)  $.

d) Suppose we could derive a constant $C_{bias}$ and a constant $N^{\ast}%
\ $\ such that
\begin{align*}
&  \sup_{\theta}\left\vert E_{\theta}\left[  \left\{  \widehat{\psi
}_{m,k_{opt}\left(  m,n\right)  }-\psi\left(  \theta\right)  \right\}
\right]  \right\vert \\
&  =\sup_{\theta}\left\vert \left\{  TB_{k_{opt}\left(  m,n\right)  }\left(
\theta\right)  +EB_{m}\left(  \theta\right)  \right\}  \right\vert \\
&  \leq C_{bias}\left\{  \frac{1}{n}\max\left\{  1,\left(  \frac
{n^{\rho_{_{opt}\left(  m,n\right)  }}}{n}\right)  ^{m-1}\right\}  \right\}
^{1/2}%
\end{align*}
\ for $n>N^{\ast}$ . Then \
\begin{align*}
&  {}BC_{m,k_{opt}\left(  m,n\right)  }\\
&  =\widehat{\psi}_{m,k_{opt}\left(  m,n\right)  }\pm\left\{  z_{\alpha
}\widehat{\mathbb{W}}_{m,\widetilde{\psi}_{k^{\ast}}}+C_{bias}\left\{
\frac{1}{n}\max\left\{  1,\left(  \frac{n^{\rho_{_{opt}\left(  m,n\right)  }}%
}{n}\right)  ^{m-1}\right\}  \right\}  ^{1/2}\right\}
\end{align*}
is a conservative uniform asymptotic $\left(  1-\alpha\right)  $ confidence
interval for $\psi\left(  \theta\right)  .$
\end{theorem}

Part a) of the theorem is an easy calculation. The asymptotic normality of
$\widehat{\psi}_{m,k_{opt}\left(  m,n\right)  }$ is based on new results on
the asymptotic distribution of higher order $U-statistics$ with kernels
depending on $n$ to be published elsewhere (Robins et al, 2007).

Part c of the theorem implies we obtain a conservative uniform asymptotic
$\left(  1-\alpha\right)  $ confidence interval for any value of $\rho
_{\ }^{\ast}\ $\ exceeding $\rho_{_{opt}\left(  m,n\right)  }.$ However, for
the actual fixed sample size of \ our study, say $n=$5000$,$\ there is no
guarantee the interval of part c based on given difference $\rho_{\ }^{\ast
}-\rho_{_{opt}\left(  m,n\right)  },$ say $.3,$ will provide conservative
finite sample coverage.

Because of this difficulty, a better approach, described in part d, would be
to determine a constant $C_{bias}$ that can be used to bound the maximal bias
under the model at a sample sizes exceeding $N^{\ast},$ with $N^{\ast} $ no
greater than the actual fixed sample size $n$ of the study. Then the interval
$BC_{m,k_{opt}\left(  m,n\right)  }$ will be a honest conservative finite
sample $1-\alpha$\ confidence interval, provided that $\widehat{\psi
}_{m,k_{opt}\left(  m,n\right)  }$\ has nearly converged to its normal limit
at sample size $n.$\ Unfortunately as yet we do not know how to determine the
constants $C_{bias}$\ and $N^{\ast}$ of part d as a function of our model and
of our initial estimator $\widehat{\theta}.$ This\ is an important open problem.

\subsubsection{Models of Increasing Dimension and Multi-Robustness\label{mr}}

\textbf{A Model of Increasing Dimension:} The previous results can also be
used for the analysis of models whose dimension increases with sample size. In
fact, consider the $\mathcal{M}\left(  \Theta_{n^{\eta}}\right)  ,$ $\eta
\ $known, that differs from model $\mathcal{M}\left(  \Theta\right)  $ in
that, rather than assuming $b\left(  x\right)  $ and $p\left(  x\right)  $
live in particular H\"{o}lder balls, we instead assume the working models of
Eqs. \ref{worka} and \ref{workb} are precisely true for $k=n^{\eta},$ so
$\psi\left(  \theta\right)  \equiv\widetilde{\psi}_{n^{\eta}}\left(
\theta\right)  $ and the dimensions of $b\left(  x\right)  $ and $p\left(
x\right)  $ increase as $n^{\eta}$. Valid point and interval estimation for
$\widetilde{\psi}_{n^{\eta}}\left(  \theta\right)  \ $can still be based on
the estimators $\widehat{\psi}_{m,k}$ except now (i) there is truncation bias
only when $k<n^{\eta}$, (ii) the variance remains of the order of $\frac{1}%
{n}\max\left(  1,\left(  \frac{k}{n}\right)  ^{m-1}\right)  ,$ and (iii) the
estimation and trunction bias (when it exists)\ orders will be determined by
any additional complexity reducing restrictions placed on the fraction of
non-zero components or on the rate of decay of the components of the vectors
$\widetilde{\overline{\eta}}_{n^{\eta}}\left(  \theta\right)  \ $and
$\widetilde{\overline{\alpha}}_{n^{\eta}}\left(  \theta\right)  ,$ and, for
estimation bias, by $\beta_{g}$ as well. As a consequence, $m_{opt}$ and
$k_{opt}$ under model $\mathcal{M}\left(  \Theta_{n^{\eta}}\right)  $ will
differ from their values under model $\mathcal{M}\left(  \Theta\right)  .$
Note we need not take $k=n^{\eta}$ as we did in the heurisitic discussion
following Remark $\ref{semi}.$ Indeed $\widehat{\psi}_{m_{_{best}}}$ in that
discussion corresponds to the estimator in the class $\widehat{\psi
}_{m,k=n^{\eta}}$ with the fastest rate of convergence. In general,
$\widehat{\psi}_{m_{_{best}}}$ will have convergence rate slower than
$\widehat{\psi}_{m_{opt},k_{opt}}.$ Furthermore, the discussion in Section
4.1.1 implies that, when $n^{\eta}>>n$ and the minimax rate for estimation of
$\psi\left(  \theta\right)  $ is slower than $n^{-1/2},$ even $\widehat{\psi
}_{m_{opt},k_{opt}}$ will typically fail to converge at the minimax rate when
complexity reducing restrictions have been imposed on $\widetilde{\overline
{\eta}}_{n^{\eta}}\left(  \theta\right)  \ $and $\widetilde{\overline{\alpha}%
}_{n^{\eta}}\left(  \theta\right)  $.

\textbf{Multi-Robustness and a Practical Data Analysis Strategy:
}\label{multirob copy(1)} Conditional on $\widehat{\theta},$ for $m\geq2,$
$EB_{m} $ is zero and thus estimator $\widehat{\psi}_{m,k}$ is unbiased for
$\widetilde{\psi}_{k}$ if $\widehat{p}\left(  \cdot\right)  =p\left(
\cdot\right)  ,\widehat{b}\left(  \cdot\right)  =b\left(  \cdot\right)  ,$ or
$\widehat{g}\left(  \cdot\right)  =g\left(  \cdot\right)  .$ We refer to
$\widehat{\psi}_{m,k}$ as triply-robust for $\widetilde{\psi}_{k},$
generalizing Robins and Rotnitzky (2001) and van der Laan and Robins (2003)
who referred to $\widehat{\psi}_{1}$ as doubly-robust because of its being
unbiased for $\widetilde{\psi}_{k}$ if either $\widehat{p}\left(
\cdot\right)  =p\left(  \cdot\right)  $ or $\widehat{b}\left(  \cdot\right)
=b\left(  \cdot\right)  .$ \ In fact, for $m\geq3,$ we can construct a
modified estimator $\widehat{\psi}_{m,k}^{\operatorname{mod}}$ that is
$m+1-fold$ robust as follows. Let $\widehat{g}_{s}\left(  \cdot\right)  ,$
$s=3,...,m,$ denote $m-2$ additional initial estimators of $g\left(
\cdot\right)  \ $that differ from one another and from $\widehat{g}\left(
\cdot\right)  .$ Define $\widehat{\psi}_{m,k}^{\operatorname{mod}%
}=\widehat{\psi}+\widehat{\mathbb{IF}}_{1,\widetilde{\psi}_{k}}%
+\widehat{\mathbb{IF}}_{22,\widetilde{\psi}_{k},\overline{i}_{j}}+\sum
_{j=3}^{m}\widehat{\mathbb{IF}}_{jj,\widetilde{\psi}_{k}}^{\operatorname{mod}%
},$ where$\mathbb{\ }$
\begin{align*}
\widehat{IF}_{jj,\widetilde{\psi}_{k},\overline{i}_{j}}^{\operatorname{mod}}
&  =\left(  -1\right)  ^{j-1}\left[  \left(  H_{1}\widehat{P}+H_{2}\right)
\dot{B}\overline{Z}_{k}^{T}\right]  _{i_{1}}\left\{  \left(  \dot{P}\dot
{B}H_{1}\overline{Z}_{k}\overline{Z}_{k}^{T}\right)  _{i_{2}}-I_{k\times
k}\right\} \\
\times &  \left[
\begin{array}
[c]{c}%
%TCIMACRO{\dprod \limits_{s=3}^{j-1}}%
%BeginExpansion
{\displaystyle\prod\limits_{s=3}^{j-1}}
%EndExpansion
\left\{  \widehat{E}_{s}\left[  \dot{P}\dot{B}H_{1}\overline{Z}_{k}%
\overline{Z}_{k}^{T}\right]  \right\}  ^{-1}\\
\left\{  \left(  \dot{P}\dot{B}H_{1}\overline{Z}_{k}\overline{Z}_{k}%
^{T}\right)  _{i_{s}}-\widehat{E}_{s}\left[  \dot{P}\dot{B}H_{1}\overline
{Z}_{k}\overline{Z}_{k}^{T}\right]  \right\}
\end{array}
\right] \\
&  \times\left\{  \widehat{E}_{j}\left[  \dot{P}\dot{B}H_{1}\overline{Z}%
_{k}\overline{Z}_{k}^{T}\right]  \right\}  ^{-1}\times\left[  \overline{Z}%
_{k}\left(  H_{1}\widehat{B}+H_{3}\right)  \dot{P}\right]  _{i_{j}}%
\end{align*}
with $\widehat{E}_{s}$ defined like $\widehat{E},$ except with $\widehat{g}%
_{s}\left(  \cdot\right)  $ replacing $\widehat{g}\left(  \cdot\right)  .$ In
the appendix, we prove that $EB_{m}^{\operatorname{mod}}=E\left[
\widehat{\psi}_{m,k}^{\operatorname{mod}}\right]  -\widetilde{\psi}_{k}$ is
\begin{equation}
\left(  -1\right)  ^{m-1}\left\{
\begin{array}
[c]{c}%
E\left[  \dot{B}\dot{P}H_{1}\left(  \frac{P-\widehat{P}}{\dot{P}}\right)
\overline{Z}_{k}^{T}\right]  \left\{  E\left[  \dot{B}\dot{P}H_{1}\overline
{Z}_{k}\overline{Z}_{k}^{T}\right]  -I_{k\times k}\right\} \\
\times%
%TCIMACRO{\dprod \limits_{s=3}^{m}}%
%BeginExpansion
{\displaystyle\prod\limits_{s=3}^{m}}
%EndExpansion
\left\{  \widehat{E}_{s}\left[  \dot{B}\dot{P}H_{1}\overline{Z}_{k}%
\overline{Z}_{k}^{T}\right]  \right\}  ^{-1}\\
\times\left\{  E\left[  \dot{B}\dot{P}H_{1}\overline{Z}_{k}\overline{Z}%
_{k}^{T}\right]  -\widehat{E}_{s}\left[  \dot{B}\dot{P}H_{1}\overline{Z}%
_{k}\overline{Z}_{k}^{T}\right]  \right\} \\
\times\left\{  E\left[  \dot{B}\dot{P}H_{1}\overline{Z}_{k}\overline{Z}%
_{k}^{T}\right]  \right\}  ^{-1}E\left[  \dot{B}\dot{P}H_{1}\left(
\frac{B-\widehat{B}}{\dot{B}}\right)  \right]
\end{array}
\right\} \label{multirob_bias}%
\end{equation}
which is zero if $\widehat{p}\left(  \cdot\right)  =p\left(  \cdot\right)
,\widehat{b}\left(  \cdot\right)  =b\left(  \cdot\right)  ,$ $\widehat{g}%
\left(  \cdot\right)  =g\left(  \cdot\right)  ,$ or if any of the $m-2$
$\widehat{g}_{s}\left(  \cdot\right)  $ equals $g\left(  \cdot\right)  .$ (We
note that if $\widehat{p}\left(  \cdot\right)  =p\left(  \cdot\right)  $ or
$\widehat{b}\left(  \cdot\right)  =b\left(  \cdot\right)  ,$ $\psi
=\widetilde{\psi}_{k}$ and thus $\widehat{\psi}_{m,k}^{\operatorname{mod}}$
and $\widehat{\psi}_{m,k}$ are unbiased for $\psi$ .)

In settings where the dimension $d$ of $X$ is so large ( say $30-50)$ that the
above asymptotic results fail as a guide to the finite sample performance of
our procedures at the moderate sample sizes, say $n=500-5000,$ commonly found
in practice, one might consider, as a practical data analysis strategy, using
the $m+1-fold$ robust estimator $\widehat{\psi}_{m,k}^{\operatorname{mod}}%
\ $with $\widehat{p}\left(  \cdot\right)  ,\widehat{b}\left(  \cdot\right)
,\widehat{g}\left(  \cdot\right)  ,$ and the $\widehat{g}_{s}\left(
\cdot\right)  $ selected by cross-validation as in van der Laan and Dudoit
(2003). Specifically, the training sample is split into two random subsamples
- a candidate estimator subsample of size $n_{c}$ and a validation subsample
of size $n_{v},$ where both $n_{c}/n$ and $n_{v}/n$ are bounded away from $0$
as $n\rightarrow\infty.$ A large number ( e.g., $n^{3})$ candidate parametric
models of various dimensions and functional forms for $p,b,$ and $g$ are fit
to the candidate estimator subsample and the validation sample is used to find
the candidate estimators $\widehat{p}\left(  \cdot\right)  \ $and
$\widehat{b}\left(  \cdot\right)  $ for $p$ and $b $ and the $m-1$ candidate
estimators $\widehat{g}\left(  \cdot\right)  $ and $\widehat{g}_{s}\left(
\cdot\right)  ,$ $s=3,...,m,$ for $g$ with the smallest estimated risks (with
respect to an appropriate risk function such as squared error or
Kullback-Leibler.) An alternative approach would be to use the triply robust
estimator $\widehat{\psi}_{m,k}$ with $\widehat{g}\left(  \cdot\right)  $ the
candidate for $g$ with minimum estimated risk. We plan to explore through
simulation whether $\widehat{\psi}_{m,k}^{\operatorname{mod}}$ outperforms
$\widehat{\psi}_{m,k}$ in the setting of very high dimensional $X.$

\section{Rates of Convergence and Minimaxity\label{minimax_section}}

We consider a generic version in which we only assume a model and functional
satisfying $Ai)-Aiv).$ To examine efficiency issues, we first consider the
estimator $\widehat{\psi}_{1}$ based on the first order influence function and
sample splitting. Without loss of generality we assume $\beta_{p}\geq\beta
_{b}.$ (Otherwise simply interchange $\beta_{p}\ $and $\beta_{b}$ in what
follows.) \ It will be useful to consider the alternative parametrization
\begin{align*}
\beta &  =\frac{\beta_{p}+\beta_{b}}{2},\\
\Delta &  =\left(  \frac{\beta_{p}}{\beta_{b}}-1\right)  \geq0
\end{align*}
The (conditional) variance of $\widehat{\psi}_{1}$ is of the order of $1/n$
and the (conditional) bias of {}$\widehat{\psi}_{1}$ in estimating $\psi$ is
$O_{p}\left(  n^{-\left(  \frac{\beta_{b}}{d+2\beta_{b}}+\frac{\beta_{p}%
}{d+2\beta_{p}}\right)  }\right)  .$ If $\Delta=0$ and thus $\beta_{p}%
=\beta_{b}, $ the bias of ${}\widehat{\psi}_{1}$ is $n^{-\frac{2\beta
}{d+2\beta}\ }$ and $\widehat{\psi}_{1}$ is not $n^{1/2}-$ consistent for
$\psi$ when $\beta/d<1/2.$ At the other extreme, as $\Delta\rightarrow\infty,
$ i.e. $\beta_{b}\rightarrow0,$ the bias of $\widehat{\psi}_{1}$ is
$n^{-\frac{2\beta}{d+4\beta}\ {}}${}which fails to be $n^{1/2}-$%
consistent$\ $for any finite $\beta.$

\textbf{Minimaxity with }$g$\textbf{\ known: }To further examine efficiency
issues, it is instructive to first consider the estimation of $\psi$ with
$g\left(  \cdot\right)  $ known. If $g\left(  \cdot\right)  $ were known, we
could set $\widehat{g}\left(  X\right)  =g\left(  X\right)  $ when calculating
$\widehat{\psi}_{m,k}.$ Then $EB_{2}=0$ and $\widehat{\psi}_{2,k}$ would
therefore be an unbiased estimator of $\widetilde{\psi}_{k}.$ Letting a
superscript $g$ denote the model with $g$ known, it is easy to see that
$\widehat{\psi}_{m_{opt}^{g},k_{opt}^{g}\left(  m_{opt}^{g}\ \right)  }$ would
be $\widehat{\psi}_{2,k_{opt}^{g}\left(  2\ \right)  }$ where $k_{opt}%
^{g}\left(  2\ \right)  $ satisfies $\max\left(  1/n,k/n^{2}\right)  \asymp
Var\left(  \widehat{\psi}_{2,k}\right)  =TB_{k}^{2}=k^{-2\left(  \beta
_{b}+\beta_{p}\right)  /d}=k^{-4\beta/d}.\ $Solving this, we find that
when\ $\beta/d$ is greater than or equal to $1/4,$ we can take $k_{opt}%
^{g}=n^{\frac{1}{4\beta/d}}\leq n$ and $\left\vert \widehat{\psi}%
_{2,k_{opt}^{g}\left(  2\ \right)  }-\psi\right\vert =O_{p}\left(
n^{-\frac{1}{2}}\right)  $ regardless of $\Delta,$ which is, of course, the
minimax rate .

In contrast if $\beta/d$ $<1/4,$ $k_{opt}^{g}\left(  2\right)  =n^{\frac
{2}{1+4\beta/d}}$ and $\left(  \widehat{\psi}_{2,k_{opt}^{g}\left(
2\ \right)  }-\psi\right)  =n^{-\frac{4\beta\ }{4\beta+d}}.\ $In an
unpublished paper, we have proved that this is the minimax rate when $g\left(
\cdot\right)  $ is known.

This raises the question of whether the lower bounds of rate $n^{-\frac{1}{2}%
}$ for $\beta/d$ $\geq1/4$ and/or rate $n^{-\frac{4\beta\ }{4\beta+d}}$\ for
$\beta/d$ $<1/4\ $are still achievable when $g$ is unknown, without
restrictions on the smoothness of $g.$

Before addressing this question, we take the opportunity to compare the
relative efficiencies of competing rate-optimal unbiased estimators in the
case of $g$ known. This discussion will provide further insight into the
results given in Remark \ref{JJ1} for models which are not locally nonparametric.

\textbf{Relative Efficiency of Various Unbiased Estimator with }%
$g$\textbf{\ known}:

For simplicity, we restrict the following discussion to the truncated version
of the parameter $\psi=E\left[  \left\{  b\left(  X\right)  \right\}
^{2}\right]  ,$ with $b\left(  X\right)  =E[Y|X],$ $g\left(  \cdot\right)  $
known, and $Y$ Bernouilli. For this choice of $\psi,$ $g\left(  \cdot\right)
$ is the marginal density of $X.$ In this subsection ,we assume $\widehat{g}%
\left(  X\right)  \ $is chosen equal to the known $g\left(  X\right)  $ so
$E\left[  \overline{Z}_{k}\overline{Z}_{k}^{T}\right]  =I_{k\times k}.$ Also
we choose $\dot{B}=\dot{P}=1$ and take $\widehat{B}=\widehat{b}\left(
X\right)  \in lin\left\{  \overline{Z}_{k}\right\}  ,$ so $\widetilde{B}%
=\Pi\left[  B|\overline{Z}_{k}\right]  =E\left[  B\overline{Z}_{k}^{T}\right]
\overline{Z}_{k}$ and $\widetilde{\psi}_{k}=E\left[  \left\{  \Pi\left[
B|\overline{Z}_{k}\right]  \right\}  ^{2}\right]  $ do not depend on
$\widehat{B}. $ Further we only concern ourselves with efficiency relative to
the $n$ observations in the estimation sample. We thus ignore any efficiency
loss from using $N-n$ observations to construct $\widehat{b}$.

Let $\Theta^{g}=\left\{  b:x\mapsto b\left(  x\right)  \in\left[  0,1\right]
\right\}  \subset\Theta$ denote the subset of $\Theta$ corresponding to the
known $g$, which consists of all functions from the unit cube in $R^{d}$ to
the unit interval. The model $\mathcal{M}\left(  \Theta^{g}\right)  $ is not
locally nonparametric. For example, the $1st$ order tangent space $\Gamma
_{1}\left(  \theta\right)  $ does not include first order scores for $g.$ Its
2nd order tangent space $\Gamma_{2}\left(  b\right)  $ does not contain second
order scores for $g$ or mixed scores for $g$ and $b.$ Rather, $\Gamma
_{2}\left(  b\right)  $ is the closed linear span of the first and second
order scores for $b$. Thus
\[
\Gamma_{2}\left(  b\right)  =\left\{  \mathbb{S}\left(  a,c\right)
;var_{b}\left[  \mathbb{S}\left(  a,c\right)  \right]  \right\}  <\infty
;a\in\mathcal{A},c\in\mathcal{C}\ \}
\]
where%

\[
S_{ij}\left(  a,c\right)  =\left\{  \left(  Y-B\right)  a\left(  X\right)
\right\}  _{i}+\left[  \left(  Y-B\right)  _{i}c\left(  X_{i},X_{j}\right)
\left(  Y-B\right)  _{j}\right]  ,
\]
and $\mathcal{A}$ and $\mathcal{C}$ are the set of one and two dimensional
functions of $x.$ Since, for $\widehat{b}\in lin\left\{  \overline{z}%
_{k}\left(  x\right)  \right\}  $
\begin{align*}
\widehat{\psi}_{2,k}\left(  \widehat{b}\right)   &  \equiv\widehat{\psi}%
_{2,k}\\
&  =\mathbb{V}\left\{
\begin{array}
[c]{c}%
\left[  \widehat{B}^{2}+2\widehat{B}\left(  Y-\widehat{B}\right)  \right]
_{i}\\
+\left[  \left(  Y-\widehat{B}\right)  \overline{Z}_{k}^{T}\right]
_{i}\left[  \overline{Z}_{k}\left(  Y-\widehat{B}\right)  \right]  _{j}%
\end{array}
\right\}
\end{align*}
\newline is unbiased for $\widetilde{\psi}_{k}\left(  b\right)  =E\left[
\left\{  \Pi\left[  B|\overline{Z}_{k}\right]  \right\}  ^{2}\right]  $ in
model $\mathcal{M}\left(  \Theta^{g}\right)  $, we know, by Remark \ref{JJ1},
that $\mathbb{IF}_{2,\widetilde{\psi}_{k}}^{eff}\left(  b\right)  \ $for
$\widetilde{\psi}_{k}\left(  b\right)  $ is the projection $\Pi_{b}\left[
\widehat{\psi}_{2,k}-\widetilde{\psi}_{k}\left(  b\right)  |\Gamma_{2}\left(
\theta\right)  \right]  $ of the 2nd order influence function $\widehat{\psi
}_{2,k}-\widetilde{\psi}_{k}\left(  b\right)  $ onto $\Gamma_{2}\left(
b\right)  .$ Now if $\widehat{\psi}_{2,k}-\widetilde{\psi}_{k}\left(
b\right)  $ was an element of $\Gamma_{2}\left(  b\right)  ,$ $\widehat{\psi
}_{2,k}-\widetilde{\psi}_{k}\left(  b\right)  $ would equal $\mathbb{IF}%
_{2,\widetilde{\psi}_{k}}^{eff}\left(  b\right)  $ and thus be 2nd order
`unbiased locally efficient', at $b\in\Theta^{g},$ as defined earlier in
Remark \ref{JJ1}. However we show below that, when $\widehat{b}\left(
X\right)  =c$ for some $c$ wp1 does not hold, $\widehat{\psi}_{2,k}%
-\widetilde{\psi}_{k}\left(  b\right)  $ is not an element of $\Gamma
_{2}\left(  b\right)  $ for any $b$. Rather, a straightforward calculation
gives
\[
\mathbb{IF}_{2,\widetilde{\psi}_{k}}^{eff}\left(  b\right)  =\mathbb{V}%
\left\{
\begin{array}
[c]{c}%
\left[  2E\left[  B\overline{Z}_{k}^{T}\right]  \overline{Z}_{k}\left(
Y-B\right)  \right]  _{i}\\
+\left[  \left(  Y-B\right)  \overline{Z}_{k}^{T}\right]  _{i}\left[
\overline{Z}_{k}\left(  Y-B\right)  \right]  _{j}%
\end{array}
\right\}  .
\]

Now one can check that $\widetilde{\psi}_{k}\left(  \widehat{b}\right)
+\mathbb{IF}_{2,\widetilde{\psi}_{k}}^{eff}\left(  \widehat{b}\right)  $ is a
function of $\widehat{b},$ so by Theorem \ref{global_eff} of Remark \ref{JJ1},
we conclude no unbiased globally efficient estimator exists. However, we prove
below that $\widetilde{\psi}_{k}\left(  \widehat{b}\right)  +\mathbb{IF}%
_{2,\widetilde{\psi}_{k}}^{eff}\left(  \widehat{b}\right)  $ and
$\widehat{\psi}_{2,k}$ have identical means. It follows that $\widetilde{\psi
}_{k}\left(  \widehat{b}\right)  +\mathbb{IF}_{2,\widetilde{\psi}_{k}}%
^{eff}\left(  \widehat{b}\right)  $ is an unbiased estimator of
$\widetilde{\psi}_{k}\left(  b\right)  =E\left[  \left(  \Pi\left[
B|\overline{Z}_{k}\right]  \right)  ^{2}\right]  $ for any $\widehat{b}\in
lin\left\{  \overline{z}_{k}\left(  x\right)  \right\}  $. Thus, for a given
choice of $\widehat{b}\in lin\left\{  \overline{z}_{k}\left(  x\right)
\right\}  ,$ $\widetilde{\psi}_{k}\left(  \widehat{b}\right)  +\mathbb{IF}%
_{2,\widetilde{\psi}_{k}}^{eff}\left(  \widehat{b}\right)  $ is 2nd order
unbiased locally efficient at $b=\widehat{b}.$ However, one can show using a
proof analogous to that in theorem \ref{3.19}\ that for $k<<n^{2}$
\begin{align*}
&  var_{b}\left[  \widetilde{\psi}_{k}\left(  \widehat{b}\right)
+\mathbb{IF}_{2,\widetilde{\psi}_{k}}^{eff}\left(  \widehat{b}\right)
\right]  /var_{b}\left[  \mathbb{IF}_{2,\widetilde{\psi}_{k}}^{eff}\left(
b\right)  \right] \\
&  =1+o_{P}\left(  \left\vert \left\vert \widehat{b}-b\right\vert \right\vert
_{\infty}\right)  .
\end{align*}

Henceforth we assume that $b$ lies in a Holder ball $H(\beta_{b},C_{b}).$ That
is we consider the submodel $b\in\Theta^{g}\cap H(\beta_{b},C_{b})\ $and
assume $\widehat{b}\left(  x\right)  \in lin\left\{  \overline{z}_{k}\left(
x\right)  \right\}  $ converges to $b$ in sup norm at the optimal rate of
$\left(  \frac{n}{\log n}\right)  ^{-\beta_{b}/\left(  2\beta_{b}+d\right)  }%
$uniformly over $\Theta^{g}\cap H(\beta_{b},C_{b}).$ The submodel and the
model $\Theta^{g}$\ have identical tangent spaces. For all $b\in\Theta^{g}\cap
H(\beta_{b},C_{b}),$\ $\left(  max\left(  n^{-1},k/n^{2}\right)  \right)
^{-1/2}\left\{  \widetilde{\psi}_{k}\left(  \widehat{b}\right)  +\mathbb{IF}%
_{2,\widetilde{\psi}_{k}}^{eff}\left(  \widehat{b}\right)  -\widetilde{\psi
}_{k}\left(  b\right)  \right\}  $\ has an asymptotic distribution with mean
zero and variance equal to $\lim_{n\rightarrow\infty}\left(  max\left(
n^{-1},k/n^{2}\right)  \right)  ^{-1}var_{b}\left[  \mathbb{IF}%
_{2,\widetilde{\psi}_{k}}^{eff}\left(  b\right)  \right]  $\ for all
$b\in\Theta^{g}\cap H(\beta_{b},C_{b}).$\emph{\ }In a slight abuse of
language, we shall refer to $var_{b}\left[  \mathbb{IF}_{2,\widetilde{\psi
}_{k}}^{eff}\left(  b\right)  \right]  $ as the asymptotic variance of
$\left\{  \widetilde{\psi}_{k}\left(  \widehat{b}\right)  +\mathbb{IF}%
_{2,\widetilde{\psi}_{k}}^{eff}\left(  \widehat{b}\right)  -\widetilde{\psi
}_{k}\left(  b\right)  \right\}  .$Thus, as with standard first order theory,
even when no unbiased estimator has finite sample variance that attains the
Bhattacharyya bound for all $b\in\Theta^{g}\cap H(\beta_{b},C_{b})$, there can
exist an unbiased estimator sequence whose asymptotic variance does attain the
bound globally.

We next compare the means and variances of $\widetilde{\psi}_{k}\left(
\widehat{b}\right)  +\mathbb{IF}_{2,\widetilde{\psi}_{k}}^{eff}\left(
\widehat{b}\right)  $ and $\widehat{\psi}_{2,k}.$ Now the two estimators are
algebraically related by
\[
\widehat{\psi}_{2,k}=\left\{  \widetilde{\psi}_{k}\left(  \widehat{b}\right)
+\mathbb{IF}_{2,\widetilde{\psi}_{k}}^{eff}\left(  \widehat{b}\right)
\right\}  +\left\{  \mathbb{V}\left[  \widehat{B}^{2}\right]  -E\left[
\widehat{B}^{2}\right]  \right\}  .
\]
Since $\mathbb{V}\left[  \widehat{B}^{2}\right]  -E\left[  \widehat{B}%
^{2}\right]  $ is unbiased for zero, we conclude that $\widehat{\psi}_{2,k}$
and $\widetilde{\psi}_{k}\left(  \widehat{b}\right)  +\mathbb{IF}%
_{2,\widetilde{\psi}_{k}}^{eff}\left(  \widehat{b}\right)  $ have the same
mean but $var_{b}\left[  \widehat{\psi}_{2,k}\right]  /var_{b}\left[
\mathbb{IF}_{2,\widetilde{\psi}_{k}}^{eff}\left(  b\right)  \right]  >1$
except when $\widehat{b}\left(  X\right)  =b\left(  X\right)  =c\ wp$ $1\ $for
some $c.$ Thus, since $\widetilde{\psi}_{k}\left(  \widehat{b}\right)
+\mathbb{IF}_{2,\widetilde{\psi}_{k}}^{eff}\left(  \widehat{b}\right)  \ $has
asymptotic variance $var_{b}\left[  \mathbb{IF}_{2,\widetilde{\psi}_{k}}%
^{eff}\left(  b\right)  \right]  $ and, except when $\widehat{b}\left(
X\right)  =c+o_{p}\left(  1\right)  ,$ $var\left(  \mathbb{V}\left[
\widehat{B}^{2}\right]  \right)  \asymp\ n^{-1},$ we conclude the asymptotic
variance of $\widehat{\psi}_{2,k}$ attains the bound $var_{b}\left[
\mathbb{IF}_{2,\widetilde{\psi}_{k}}^{eff}\left(  b\right)  \right]  $ when
$k>>n,$ but exceeds the bound when $k\leq n$, except when $\widehat{b}\left(
X\right)  =c+o_{p}\left(  1\right)  $.

Finally, for completeness, Robins and van der Vaart (2006) considered an
alternative particularly simple rate-optimal unbiased estimator of
$\widetilde{\psi}_{k}\left(  b\right)  =E\left[  \left\{  \Pi\left[
B|\overline{Z}_{k}\right]  \right\}  ^{2}\right]  $ given by $\widehat{\psi
}_{RV}=\mathbb{V}\left\{  \left[  Y\overline{Z}_{k}^{T}\right]  _{i}\left[
\overline{Z}_{k}Y\right]  _{j}\right\}  $. The Hoeffding decomposition of
$\widehat{\psi}_{RV}-\widetilde{\psi}_{k}\left(  b\right)  $ is%

\begin{align*}
&  \mathbb{V}\left[  E\left[  B\overline{Z}_{k}^{T}\right]  \overline{Z}%
_{k}Y-\widetilde{\psi}_{k}\left(  b\right)  \right]  +\mathbb{V}\left\{
\left[  Y\overline{Z}_{k}^{T}-E\left[  B\overline{Z}_{k}^{T}\right]  \right]
_{i}\left[  \overline{Z}_{k}Y-E\left[  B\overline{Z}_{k}\right]  \right]
_{j}\right\} \\
&  =\mathbb{IF}_{2,\widetilde{\psi}_{k}}^{eff}\left(  b\right)  +Q+T
\end{align*}
with%

\begin{align*}
Q  &  =\mathbb{V}\left[  \left\{  \Pi\left[  B|\overline{Z}_{k}\right]
B-\psi\right\}  \right] \\
T  &  =\mathbb{V}\left\{
\begin{array}
[c]{c}%
2\left(  B_{i}\overline{Z}_{k,i}^{T}\overline{Z}_{k,j}-\Pi\left[
B|\overline{Z}_{k}\right]  _{j}\right)  \left(  Y-B\right)  _{j}\\
+B_{i}\overline{Z}_{k,i}^{T}\overline{Z}_{k,j}B_{j}-\Pi\left[  B|\overline
{Z}_{k}\right]  _{i}B_{i}-\Pi\left[  B|\overline{Z}_{k}\right]  _{j}B_{j}+\psi
\end{array}
\right\}
\end{align*}
Since, except when $B=c$ wp1, $var_{b}\left(  Q\right)  \ \asymp\ n^{-1}$ and
$var_{b}\left(  T\right)  \asymp k/n^{2},$ we conclude that the asymptotic
variance of $\widehat{\psi}_{RV}$ exceeds the bound $var_{b}\left[
\mathbb{IF}_{2,\widetilde{\psi}_{k}}^{eff}\left(  b\right)  \right]  $
regardless of whether $k>>n$ does or does not hold except when $b\left(
X\right)  =c$ wp1.

\textbf{Minimaxity with Unknown }$g$\textbf{\ and }$\beta/d$\textbf{\ }%
$\geq1/4:$ We now show that the bound $n^{-\frac{1}{2}}$ for $\beta/d$
$\geq1/4$ is achievable for each $\beta_{g}>0.$ Consider the estimator
$\widehat{\psi}_{m,k}$ with $n^{\frac{2}{1+4\beta/d}}\leq\ k\leq n$ and
\[
m\geq1+\left\{  \frac{1}{2}-\frac{\beta_{b}}{d+2\beta_{b}}-\frac{\beta_{p}%
}{d+2\beta_{p}}\right\}  \frac{2\beta_{g}+d}{\beta_{g}}%
\]
so that $EB_{m}=O_{p}\left(  n^{-\left(  \frac{\left(  m-1\right)  \beta_{g}%
}{2\beta_{g}+d}+\frac{\beta_{b}}{d+2\beta_{b}}+\frac{\beta_{p}}{d+2\beta_{p}%
}\right)  }\right)  $ is $O_{p}\left(  n^{-1/2}\right)  .$ Then $Var\left(
\widehat{\psi}_{m,k}\right)  \asymp1/n,$ $TB_{k}^{2}=O_{p}\left(  1/n\right)
$ and $EB_{m}^{2}=O_{p}\left(  1/n\right)  $ so $\widehat{\psi}_{m,k}$ will be
$n^{\frac{1\ }{2}}-$consistent for $\psi.$ If $\ \Delta=0$ and $\beta<1/2,$
the above expression implies that $m\geq\frac{\ d-2\beta}{2\left(
2\beta\ +d\right)  }/\frac{\beta_{g}\ }{\left(  2\beta_{g}+d\right)  }+1$ for
$n^{\frac{1\ }{2}}-$consistency. Similarly, if $\Delta\rightarrow\infty,$ i.e.
$\beta_{b}\rightarrow0,$ it is necessary that $m\geq\frac{d}{2\left(
4\beta\ +d\right)  }/\frac{\beta_{g}\ }{\left(  2\beta_{g}+d\right)  }\ +1$
for $n^{\frac{1}{2}}-$consistency. These results imply that estimators
$\widehat{\psi}_{m,k}$ in our class can always achieve $n^{\frac{1\ }{2}}%
-$consistency whenever $\beta_{g}>0$, but for fixed $\beta<d/2,$ the order $m$
of the required $U-$statistic increases without bound as the smoothness
$\beta_{g}$ of $g$ approaches zero.

\textbf{Efficiency:} We now show that when $\beta/d$ is strictly greater than
$1/4,$ we can construct an unconditional asymptotically linear estimator based
on all $N$ subjects with influence function $N^{-1}\sum_{i=1}^{N}IF_{1,\psi
,i}\left(  \theta\right)  $ by having the number of the $N$ subjects allotted
to the validation sample and analysis sample be $N^{1-\epsilon}$ and
$n=n\left(  \epsilon\right)  =N-N^{1-\epsilon},$ respectively, for
$1>\epsilon>0$. It then follows from van der Vaart$\ (1998) $ that the
estimator is regular and semiparametric efficient. Specifically, suppose
$\beta/d=1/4+\delta,$ $\delta>0.$ Consider the estimator $\widehat{\psi
}_{m^{\ast},k}$ with $m^{\ast}>1+\left\{  \frac{1}{2\left(  1-\epsilon\right)
}-\frac{\beta_{b}}{d+2\beta_{b}}-\frac{\beta_{p}}{d+2\beta_{p}}\right\}
\frac{2\beta_{g}+d}{\beta_{g}}$ so that $EB_{m^{\ast}}=O_{p}\left(
N^{-\left(  1-\epsilon\right)  \left(  \frac{\left(  m^{\ast}-1\right)
\beta_{g}}{2\beta_{g}+d}+\frac{\beta_{b}}{d+2\beta_{b}}+\frac{\beta_{p}%
}{d+2\beta_{p}}\right)  }\right)  $ is $o_{p}\left(  N^{-1/2}\right)  \ $and
$k=n\left(  \epsilon\right)  ^{\frac{1}{1+2\delta}}<n\left(  \epsilon\right)
$ so that $TB_{k}^{2}=o_{p}\left(  1/N\right)  $ and $var\left[
\widehat{IF}_{jj,\widetilde{\psi}_{k}}\right]  =o_{p}\left(  1/N\right)  $ for
$j\geq2.\ $Then, by our previous results,
\[
\widehat{\psi}_{m^{\ast},k}-\psi\left(  \theta\right)  =\ n\left(
\epsilon\right)  ^{-1}\sum_{i=1}^{n\left(  \epsilon\right)  }IF_{1,\psi
,i}\left(  \theta\right)  +o_{p}\left(  N^{-1/2}\right)  .
\]
It remains to show that $N^{-1}\sum_{i=1}^{N}IF_{1,\psi,i}\left(
\theta\right)  -n\left(  \epsilon\right)  ^{-1}\sum_{i=1}^{n\left(
\epsilon\right)  }IF_{1,\psi,i}\left(  \theta\right)  =o_{p}\left(
N^{-1/2}\right)  .$ But the LHS is
\begin{align*}
&  n\left(  \epsilon\right)  ^{-1}\sum_{i=1}^{n\left(  \epsilon\right)
}IF_{1,\psi,i}\left(  \theta\right)  \left\{  \frac{n\left(  \epsilon\right)
}{N}-1\right\}  +N^{-1}\sum_{i=n\left(  \epsilon\right)  +1}^{N\ }%
IF_{1,\psi,i}\left(  \theta\right) \\
&  =O_{p}\left(  n\left(  \epsilon\right)  ^{-1/2}N^{-\epsilon}\right)
+O_{p}\left(  N^{\ \left(  1-\epsilon\right)  /2}N^{-1}\right)  =O_{p}\left(
N^{-1/2}N^{-\epsilon}\right)  +O_{p}\left(  N^{-1/2}N^{-\epsilon/2}\right) \\
&  =o_{p}\left(  N^{-1/2}\right)  .
\end{align*}

\textbf{Adaptivity when }$\beta/d$\textbf{\ }$>1/4$\textbf{\ :} We next prove
that if we let $n\equiv n\left(  \epsilon\right)  =N-N^{1-\epsilon},m\equiv
m\left(  N\right)  =o\left(  N\right)  \ $with $\ln\left(  N\right)  =O\left(
m\left(  N\right)  \right)  $ and $k=n\left(  \epsilon\right)  /\ln\left(
n\right)  ,$ $\widehat{\psi}_{m,k\ }$ will be semiparametric efficient for
each $\beta>1/4$, provided $\left\{  \widehat{g}\left(  X\right)  -g\left(
X\right)  \right\}  =o_{p}\left(  m\left(  N^{\left(  1-\epsilon\right)
}\right)  ^{-2}\right)  .$ Clearly, the truncation bias is $o\left(
N^{-1/2}\right)  .$ The estimation bias $EB_{m\left(  N\right)  }$ is
$O_{p}\left(  m\left(  N^{\left(  1-\epsilon\right)  }\right)  ^{-2\left[
m\left(  N\right)  -1\right]  }N^{\ -\left(  1-\epsilon\right)  \left\{
\frac{\beta_{b}}{d+2\beta_{b}}+\frac{\beta_{p}}{d+2\beta_{p}}\right\}
}\right)  .$ Thus $EB_{m\left(  N\right)  }=o_{p}\left(  N^{-1/2}\right)  $ if
$m\left(  N^{\left(  1-\epsilon\right)  }\right)  ^{-2\left[  m\left(
N\right)  -1\right]  }=o\left(  N^{-\frac{1}{2}+\left(  1-\epsilon\right)
\left\{  \frac{\beta_{b}}{d+2\beta_{b}}+\frac{\beta_{p}}{d+2\beta_{p}%
}\right\}  }\right)  .$ So we require
\[
2\left[  m\left(  N\right)  -1\right]  ln\left\{  m\left(  N^{\left(
1-\epsilon\right)  }\right)  \right\}  /\left[  \frac{1}{2}-\left(
1-\epsilon\right)  \left\{  \frac{\beta_{b}}{d+2\beta_{b}}+\frac{\beta_{p}%
}{d+2\beta_{p}}\right\}  \right]  \ln\left(  N\right)  \rightarrow\infty,
\]
which is satisfied if $\ln\left(  N\right)  =O\left(  m\left(  N\right)
\right)  .$ In the appendix we prove that var$_{\theta}\left[  \widehat{\psi
}_{m,k\ }\right]  =var_{\widehat{\theta}}\left[  \widehat{\psi}_{m,k\ }%
\right]  \left\{  1+o_{p}\left(  1\right)  \right\}  $ provided $\left\{
\widehat{g}\left(  X\right)  -g\left(  X\right)  \right\}  =o_{p}\left(
m\left(  N^{\left(  1-\epsilon\right)  }\right)  ^{-2}\right)  .$ Now
$var_{\widehat{\theta}}\left[  \widehat{\psi}_{m,k\ }\right]  =\frac{1}%
{n}var_{\widehat{\theta}}\left\{  IF_{1,\psi,i}\left(  \widehat{\theta
}\right)  \right\}  \left[  O\left(  \sum_{l=0}^{m\left(  N\right)  }\left\{
\ln n\right\}  ^{-l}\right)  \right]  .{\Huge \ }$But $\sum_{l=0}^{m\left(
N\right)  }\left\{  \ln n\right\}  ^{-l}$ $=O\left(  \frac{1-\left(  \ln
n\right)  ^{-\left[  m\left(  N\right)  +1\right]  }}{1-\left\{  \ln
n\right\}  ^{-1}}\right)  =O\left(  1+\left\{  \ln n\right\}  ^{-1}\right)  ,$
so var$_{\theta}\left[  \widehat{\psi}_{m,k\ }\right]  $ is $n^{-1}%
var_{\widehat{\theta}}\left\{  IF_{1,\psi,i}\left(  \widehat{\theta}\right)
\right\}  \left\{  1+o_{p}\left(  1\right)  \right\}  =n^{-1}var\left\{
\mathbb{IF}_{1,\psi}\right\}  \left\{  1+o_{p}\left(  1\right)  \right\}  .$
The proof of efficiency now proceeds as above.

\textbf{Alternative Estimators when }$\beta/d$\textbf{\ }$>1/4$\textbf{\ :}
When $\beta/d$ $>1/4,$ there actually exist,at least for certain functionals
in our class, $n^{\frac{1\ }{2}}-$consistent estimators of $\psi$ that are
much simpler than our very high order U-statistic estimators.$\ $For example
consider the expected conditional covariance $\psi=E\left[  Cov\left\{
A,Y|X\right\}  \right]  $ of Example 1b with $d=1$.

\textbf{Example 1b (cont)}: Number the study subjects $i=0,...,N-1$ ordered by
their realized values $X_{i},$ where we have not split the sample$.$ Following
Wang et al. (2006), consider the difference -based estimator estimator
\[
\widehat{\psi}_{d}=N^{-1}\sum_{i=0}^{N/2-1}\left\{  Y_{2i}A_{2i}%
+Y_{2i+1}A_{2i+1}-Y_{2i+1}A_{2i}-Y_{2i}A_{2i+1}\right\}
\]
which has conditional mean given $\left\{  X_{1},...,X_{N}\right\}  $ of
\begin{align*}
&  N^{-1}\sum_{i=0}^{N/2-1}Cov\left\{  A,Y|X_{2i}\right\}  +Cov\left\{
A,Y|X_{2i+1}\right\} \\
&  +N^{-1}\sum_{i=0}^{N/2-1}\left(  \left\{  b\left(  X_{i+1}\right)
-b\left(  X_{i}\right)  \right\}  \left\{  p\left(  X_{i+1}\right)  -p\left(
X_{i}\right)  \right\}  \right)
\end{align*}

Hence
\begin{align*}
&  E\left[  \widehat{\psi}_{d}-\psi\right] \\
&  =N^{-1}E\left[  \sum_{i=0}^{N/2-1}\left\{  b\left(  X_{i+1}\right)
-b\left(  X_{i}\right)  \right\}  \left\{  p\left(  X_{i+1}\right)  -p\left(
X_{i}\right)  \right\}  \right] \\
&  =O_{p}\left(  N^{-1}%
%TCIMACRO{\tsum \limits_{i=0}^{\frac{N}{2}-1}}%
%BeginExpansion
{\textstyle\sum\limits_{i=0}^{\frac{N}{2}-1}}
%EndExpansion
E\left\{  X_{i+1}-X_{i}\right\}  ^{2\beta}\right)  =O\left(  N^{-2\beta
}\right)
\end{align*}
by the theory of spacings (Pyke, 1965). But $O\left(  N^{-2\beta}\right)  $ is
$o_{p}\left(  N^{-1/2}\right)  $ when $\beta$ $>1/4.$ The variance of
$\widehat{\psi}_{d}$ is $O\left(  N^{-1}\right)  \ $so $\widehat{\psi}_{d}$ is
$N^{1/2}-$consistent$.$ However, $\frac{var_{\theta}\left(  \widehat{\psi}%
_{d}\right)  }{var_{\theta}\left(  \mathbb{IF}_{1,\psi}\left(  \theta\right)
\right)  }\neq1+o_{p}\left(  1\right)  $ so $\widehat{\psi}_{d}$ is not
(semiparametric) efficient. As discussed by Arellano (2003), by using a $m$th
order rather than a second order difference operator and letting
$m\rightarrow\infty$ at an appropriate rate as $N\rightarrow\infty,$ the $m$th
order estimator $\widehat{\psi}_{d}$ can be made efficient.

\textbf{Minimaxity with Unknown }$g$\textbf{\ and }$\beta/d$\textbf{\ }$<1/4:
$ Consider next whether the lower bound of $n^{-\frac{4\beta\ }{4\beta+d}}%
$\ for $\beta/d$ $<1/4$ is achievable when $g$ is unknown$\ $but $\beta_{g}%
>0$. We will show the next section that the bound $\ n^{-\frac{4\beta}%
{4\beta+d}}$ is achievable provided
\begin{equation}
\frac{2\beta_{g}/d}{2\beta_{g}/d+1}>\frac{4\beta/d\frac{1-4\beta/d}%
{1+4\beta/d}\left(  \ \Delta+1\right)  }{\left(  \Delta+2\right)  },\text{
}\label{NE11}%
\end{equation}
i.e., $\beta_{g}>\frac{2\beta\left(  \ \Delta+1\right)  \left(  1-4\beta
/d\right)  }{\left(  \Delta+2\right)  \ \left(  1+4\beta/d\right)  -4\left(
\beta/d\right)  \left(  1-4\beta/d\right)  \left(  \ \Delta+1\right)  }.$ To
attain the bound $n^{-\frac{4\beta}{4\beta+d}}$ whenever eq.$\left(
\ref{NE11}\right)  $ holds, we introduce new more efficient estimators, owing
to the fact that an estimator $\widehat{\psi}_{m,k}$ in our class can attain
the bound $n^{-\frac{4\beta}{4\beta+d}}$ only in the special case where the
second order estimation bias $EB_{2}=O_{p}\left(  n^{-\left(  \frac{\beta_{g}%
}{2\beta_{g}+d}+\frac{\beta_{b}}{d+2\beta_{b}}+\frac{\beta_{p}}{d+2\beta_{p}%
}\right)  }\right)  $ is less than $n^{-\frac{4\beta}{4\beta+d}}.$

For a fixed $\beta=\left(  \beta_{p}+\beta_{b}\right)  /2,$ the right hand
side of eq.$\left(  \ref{NE11}\right)  $ is minimized over $\Delta\geq0$ at
$\Delta=0.$ At $\Delta=0,$ eq.$\left(  \ref{NE11}\right)  $ reduces to
\begin{align}
\frac{\beta_{g}/d}{2\beta_{g}/d+1}  &  >\frac{1-4\beta/d}{1+4\beta/d}%
\beta/d\text{ }\Rightarrow\label{boundE}\\
\beta_{g}  &  >\frac{\beta\left(  1-4\beta/d\right)  }{1+2\beta/d\ +8\left(
\beta/d\right)  ^{2}}\label{boundE2}%
\end{align}

The right hand side of eq.$\left(  \ref{NE11}\right)  $ increases with
$\Delta$ with asymptote equal to twice the RHS of eq.(\ref{boundE}) {}as
$\Delta\rightarrow\infty.$ Hence, in order to attain the optimal rate
$n^{-\frac{4\beta}{4\beta+d}}$ when $\beta_{p}=2\beta$ and $\beta_{b}=0,$ the
quantity $\frac{\beta_{g}}{2\beta_{g}+d}$ must be twice as large as when
$\beta_{p}=\beta_{b}=\beta.$

In the next section, we construct an estimator with a convergence rate of
$\log\left(  n\right)  n^{-\frac{4\beta}{4\beta+d}}$ at the cut-point
$\frac{\beta_{g}}{1+2\beta g}=\frac{\left(  1-4\beta\right)  \beta}{1+4\beta
}.$ In this paper we do not consider the construction of estimators that are
rate optimal below this cutpoint.

However, for the special case $\Delta=0,$ in an unpublished paper Li et. al.
(2007) have constructed estimators which converge at a rate given in
Eq.$\left(  \ref{bc}\right)  $, whenever inequality $\left(  \ref{NE11}%
\right)  $ fails to hold . We conjecture that this rate is minimax, possibly
only up to $\log$ factors, when inequality $\left(  \ref{NE11}\right)  $ fails
to hold and $\Delta=0.$ At the cut-point $\frac{\beta_{g}}{1+2\beta g}%
=\frac{\left(  1-4\beta\right)  \beta}{1+4\beta},$ we obtain $m^{\ast}=0$ and
thus$\ $Eq.$\left(  \ref{bc}\right)  $ becomes$\ \log\left(  n\right)
n^{-\frac{4\beta}{4\beta+d}},$ in agreement with the rate of the estimator of
Section \ref{case2} below. In the extreme case in which $\beta_{g}%
\rightarrow0$ with $\beta$ remaining fixed, $\log\left(  n\right)
n^{-\frac{1}{2}+\frac{\beta_{g}/d}{1+2\beta_{g}/d}\frac{\left(  m^{\ast
}+1\right)  ^{2}}{2\beta/d}}$ $\rightarrow\log\left(  n\right)  n^{-\frac
{1}{2}+\frac{\beta_{g}}{1+2\beta g}\frac{1}{\beta}\beta\left(  1-4\beta
/d\right)  \frac{1+2\beta_{g}}{2\beta_{g}}}=\log\left(  n\right)
n^{-2\beta/d},$ which agrees (up to a log factor) with the rate of
$n^{-2\beta/d}$ given by the simple estimator of Wang et al. (2006) analyzed
above under "Example 1b (cont)".

\textbf{Improved Rates of Convergence with }$X$\textbf{\ random in {}a
semiparametric model:} We now, as promised in the Introduction, construct an
estimator of $\sigma^{2}$ under the homoscedastic model $E\left[  Y|X\right]
=b\left(  X\right)  ,$ $var\left[  Y|X\right]  =\sigma^{2}$ with $X$ random
with unknown density that, whenever $\beta<\min\left\{  1,d/4\right\}  $ and,
regardless of the smoothness of $f_{X}\left(  x\right)  $, converges at the
rate $n^{-\frac{4\beta/d}{4\beta/d+1}},$ which is faster than equal-spaced
non-random minimax rate of $n^{-2\beta/d}.$ Specifically we divide the support
of $X,$ i.e., the unit cube in $R^{d},$ into $k=k\left(  n\right)  =n^{\gamma
},\gamma>1$ identical subcubes with edge length $k^{-1/d}.$ We continue to
assume the unknown density $f_{X}\left(  x\right)  $ is absolutely continuous
wrt to Lebesgue measure and both it and its inverse are bounded in sup-norm.
Then it is a standard probability calculation that the number of subcubes
containing at least two observations is $O_{p}\left(  n^{2}/k\right)  .$ We
estimate $\sigma^{2}$ in each such subcube by $\left(  Y_{i}-Y_{j}\right)
^{2}/2,$ where, for any subcube with $3$ or more observations,$\ i$ and $j$
are chosen randomly, without replacement. Our final estimator of $\sigma^{2}$
is the average of our subcube-specific estimates $\left(  Y_{i}-Y_{j}\right)
^{2}/2$ over the $O_{p}\left(  n^{2}/k\right)  $ subcubes with at least two
observations. The rate of convergence of the estimator is minimized at
$n^{-\frac{4\beta/d}{4\beta/d+1}}$ by taking $k=n^{\frac{2}{1+4\beta/d}.},$ as
we now show.

We note that $E\left[  \left(  Y_{i}-Y_{j}\right)  ^{2}/2|X_{i},X_{j}\right]
=\sigma^{2}+\left\{  b\left(  X_{i}\right)  -b\left(  X_{j}\right)  \right\}
^{2}/2,$ $\left\vert b\left(  X_{i}\right)  -b\left(  X_{j}\right)
\right\vert =O\left\Vert X_{i}-X_{j}\right\Vert ^{\beta}$ by $\beta<1,$ and
$\left\Vert X_{i}-X_{j}\right\Vert =d^{1/2}O\left(  k^{-1/d}\right)  \ $when
$X_{i}$ and$\ X_{j}$ are in the same subcube. It follows that the estimator
has variance $O_{p}\left(  k/n^{2}\right)  $ and bias of $O\left(
k^{-2\beta/d}\right)  .$ To minimize the convergence rate we equate the orders
of the variance and the squared bias by solving $k/n^{2}=k^{-4\beta/d}$ which
gives $k=n^{\frac{2}{1+4\beta/d}.}.$ Our random design estimator has better
bias control and hence converges faster than the optimal equal-spaced fixed
$X$ estimator, because the random design estimator exploits the $O_{p}\left(
n^{2}/n^{\frac{2}{1+4\beta/d}.}\right)  $ random fluctuations for which
$X^{\prime}s$ corresponding to two different observations are a distance of
$O\left(  \left\{  n^{\frac{2}{1+4\beta/d}.}\right\}  ^{-1/d}\right)  $ apart.
Our estimator will not converge at rate $n^{-\frac{4\beta/d}{4\beta/d+1}}$ to
$E\left[  var\left(  Y|X\right)  \right]  $ in our nonparametric model,
because it then no longer suffices to average estimates of $var\left(
Y|X\right)  $ only over subcubes containing 2 or more observations.

\subsection{More Efficient Estimators}

\subsubsection{Case 1: The estimation bias of the third order estimator is
less than the optimal rate}

In a (locally) nonparametric model $\mathcal{M}\left(  \Theta\right)  ,$ the
estimator $\widehat{\psi}_{m,k}=\widehat{\psi}+\widehat{\mathbb{IF}%
}_{m,\widetilde{\psi}_{k}}$ is essentially the unique $m-th$ order U-statistic
estimator of the truncated parameter $\widetilde{\psi}_{k}$ for which the
leading term in the bias is $O\left(  \left\vert \left\vert \widehat{\theta
}-\theta\right\vert \right\vert ^{m+1}\right)  .$ However, when the minimax
rate of convergence for $\psi$ is slower than $n^{-1/2}$, other $m^{th}$ order
U-statistics estimators will often converge to $\widetilde{\psi}_{k}$ (and
thus $\psi)$ at a faster rate uniformely over the model than does any
estimator $\widehat{\psi}_{m,k}$ (constructed from an estimated higher order
influence function $\widehat{\mathbb{IF}}_{m,\widetilde{\psi}_{k}}$ for
$\widetilde{\psi}_{k})$ by tolerating bias at orders less than $m+1$ in
exchange for a savings in variance.

\begin{remark}
\label{T7}A heuristic understanding as to why this is so can be gained from
the following considerations. The theory of higher order influence functions
as developed in theorems 2.2 and 2.3 is a theory of score functions
(derivatives). \ Thus it can directly incorporate the restriction that a
function, say $b\left(  x\right)  ,$ has an expansion $b\left(  x\right)
=\sum_{l=1}^{\infty}\eta_{l}z_{l}\left(  x\right)  $ for which $\eta_{l}=0$
for $l>k,$ \ as the restriction is equivalent to various scores being equal to
zero. \ However the theory cannot directly incorporate restrictions such as
$\sum_{l=k}^{\infty}\eta_{l}^{2}=k^{-2\beta_{b}}$ or $\eta_{l}\propto
l^{-\left(  \beta_{b}+\frac{1}{2}\right)  }$ that do not imply any
restrictions on score functions. Thus to find an optimal estimator, one must
perform additional \textquotedblleft side calculations\textquotedblright\ to
quantify the estimation and truncation bias of various candidate estimators
under these restrictions. As the assumption that $b\left(  x\right)  $ lies in
a Holder ball can be expressed in terms of such restrictions, this remark is
relevant to a search for an optimal rate estimator.
\end{remark}

We now construct such estimators. We first consider the case where $\beta
_{b},\beta_{b},$ and $\beta_{g}$ are such that the estimation bias $O\left(
n^{-\left(  \frac{\beta_{g}}{2\beta_{g}+d}+\frac{\beta_{b}}{d+2\beta_{b}%
}+\frac{\beta_{p}}{d+2\beta_{p}}\right)  }\right)  $ of the second order
estimator is greater than $O\left(  n^{-\frac{4\beta}{4\beta+d}}\right)  $ but
the estimation bias $O\left(  n^{-\left(  \frac{2\beta_{g}}{2\beta_{g}%
+d}+\frac{\beta_{b}}{d+2\beta_{b}}+\frac{\beta_{p}}{d+2\beta_{p}}\right)
}\right)  $ of the third order estimator is less than $O\left(  n^{-\frac
{4\beta}{4\beta+d}}\right)  .$ That is
\begin{equation}
n^{-\left(  \frac{2\beta_{g}}{2\beta_{g}+d}+\frac{\beta_{b}}{d+2\beta_{b}%
}+\frac{\beta_{p}}{d+2\beta_{p}}\right)  }<n^{-\frac{4\beta}{4\beta+d}%
}<n^{-\left(  \frac{\beta_{g}}{2\beta_{g}+d}+\frac{\beta_{b}}{d+2\beta_{b}%
}+\frac{\beta_{p}}{d+2\beta_{p}}\right)  }\label{fourth}%
\end{equation}
Then the most efficient estimator $\widehat{\psi}_{m,k\ }$ in our class has
rate of convergence slower than $n^{-\frac{4\beta}{4\beta+d}}$ because
$\widehat{\psi}_{2,k_{opt}\left(  2\right)  }$ converges at rate $n^{-\left(
\frac{\beta_{g}}{2\beta_{g}+d}+\frac{\beta_{b}}{d+2\beta_{b}}+\frac{\beta_{p}%
}{d+2\beta_{p}}\right)  }$ determined by the 2nd order estimation bias and,
for $m>3,$ $\widehat{\psi}_{m,k_{opt}\left(  m\right)  }$ converges at a rate
no faster than $n^{-\frac{6\beta}{\left(  d+2\beta\right)  }}=n^{-4\frac
{\beta}{d}3/\left(  \left(  3-1\right)  +4\frac{\beta}{d}\right)  }%
=\min_{\left\{  m;m>3\right\}  }n^{-4\frac{\beta}{d}m/\left(  \left(
m-1\right)  +4\frac{\beta}{d}\right)  }.$ [We obtained $n^{-4\frac{\beta}%
{d}m/\left(  \left(  m-1\right)  +4\frac{\beta}{d}\right)  }$ as $\left(
k^{-4\beta/d}\right)  ^{1/2},$ where $k$ solves the equation $k^{m}%
/n^{m+1}=k^{-4\beta/d}$ that equates the variance $k^{m}/n^{m+1}$ of
$\mathbb{IF}_{m}$ to the squared truncation bias $k^{-4\beta/d}.$]

To describe our more efficient estimator, define for nonnegative integers
$k\left(  0\right)  ,k\left(  1\right)  ,k^{\ast}\left(  0\right)  ,k^{\ast
}\left(  1\right)  $ with $k\left(  0\right)  <k\left(  1\right)  $ and
$k^{\ast}\left(  0\right)  <k^{\ast}\left(  1\right)  $ the $U-$statistic
\[
\widehat{\mathbb{U}}_{3}\left(  _{k\left(  0\right)  ,}^{k\left(  1\right)  ,}%
%TCIMACRO{\QATOP{k^{\ast}\left(  1\right)  }{k^{\ast}\left(  0\right)  }}%
%BeginExpansion
\genfrac{}{}{0pt}{}{k^{\ast}\left(  1\right)  }{k^{\ast}\left(  0\right)  }%
%EndExpansion
\right)  =\mathbb{V}\left(  \widehat{U}_{3}\left(  _{k\left(  0\right)
,}^{k\left(  1\right)  ,}%
%TCIMACRO{\QATOP{k^{\ast}\left(  1\right)  }{k^{\ast}\left(  0\right)  }}%
%BeginExpansion
\genfrac{}{}{0pt}{}{k^{\ast}\left(  1\right)  }{k^{\ast}\left(  0\right)  }%
%EndExpansion
\right)  \right)
\]
with%

\begin{align*}
&  \widehat{\mathbb{U}}_{3}\left(  _{k\left(  0\right)  ,k^{\ast}\left(
0\right)  }^{k\left(  1\right)  ,k^{\ast}\left(  1\right)  }\right) \\
&  =\widehat{\epsilon}_{i_{1}}\overline{Z}_{k\left(  0\right)  ,i_{1}%
}^{k\left(  1\right)  ,T}\left(  \left[  \dot{P}\dot{B}H_{1}\overline
{Z}_{k\left(  0\right)  \ }^{k\left(  1\right)  }\overline{Z}_{k^{\ast}\left(
0\right)  \ }^{k^{\ast}\left(  1\right)  ,T}\right]  _{i_{2}}-I_{\left\{
k\left(  1\right)  -k\left(  0\right)  \right\}  \times\left\{  k^{\ast
}\left(  1\right)  -k^{\ast}\left(  0\right)  \right\}  }\right)  \overline
{Z}_{k^{\ast}\left(  0\right)  ,i_{3}}^{k^{\ast}\left(  1\right)
}\widehat{\Delta}_{i_{3}}\\
&  =\sum_{s_{1}=k\left(  0\right)  +1}^{k\left(  1\right)  }\sum
_{s_{2}=k^{\ast}\left(  0\right)  +1}^{k^{\ast}\left(  1\right)  }\left\{
\begin{array}
[c]{c}%
\widehat{\epsilon}_{i_{1}}z_{s_{1}}\left(  X_{i_{1}}\right)  \times\\
\left\{  \left[  \dot{B}\dot{P}H_{1}\right]  _{i_{2}}z_{s_{1}}\left(
X_{i_{2}}\right)  z_{s_{2}}\left(  X_{i_{2}}\right)  -I\left[  s_{1}%
=s_{2}\right]  \right\}  z_{s_{2}}\left(  X_{i_{3}}\right)  \widehat{\Delta
}_{i_{3}}%
\end{array}
\right\}  ,
\end{align*}

where $\overline{Z}_{k\left(  0\right)  \ }^{k\left(  1\right)  }=\left(
Z_{k\left(  0\right)  +1\ },....,Z_{k\left(  1\right)  \ }\right)
^{T},\widehat{\epsilon}=\left(  H_{1}\widehat{P}+H_{2}\right)  \dot
{B},\widehat{\Delta}=\ \left(  H_{1}\widehat{B}+H_{3}\right)  \dot{P},$

$I_{r\times v}=\left(  I_{ij}\right)  _{r\times v}$ with $I_{ij}=I\left(
i=j\right)  .$

As an example $\widehat{\mathbb{IF}}_{33,\widetilde{\psi}_{k}}%
=\widehat{\mathbb{U}}_{3}\left(  _{0}^{k},_{0}^{k}\right)  .$ We can identify
$\left(  _{k\left(  0\right)  ,}^{k\left(  1\right)  ,}%
%TCIMACRO{\QATOP{k^{\ast}\left(  1\right)  }{k^{\ast}\left(  0\right)  }}%
%BeginExpansion
\genfrac{}{}{0pt}{}{k^{\ast}\left(  1\right)  }{k^{\ast}\left(  0\right)  }%
%EndExpansion
\right)  $ with the rectangle in $R^{2}$ defined by $\left\{  \left(
r_{1},r_{2}\right)  ;k\left(  0\right)  +1\leq r_{1}\leq k\left(  1\right)
,k^{\ast}\left(  0\right)  +1\leq r_{1}\leq k^{\ast}\left(  1\right)
\right\}  $ with $\left(  k\left(  0\right)  +1,k^{\ast}\left(  0\right)
+1\right)  $ and $\left(  k\left(  1\right)  +1,k^{\ast}\left(  1\right)
+1\right)  ,$ respectively, the vertices closest and furthest from the origin.
Thus $\widehat{\mathbb{IF}}_{33,\widetilde{\psi}_{k}}=\widehat{\mathbb{U}}%
_{3}\left(  _{0}^{k},_{0}^{k}\right)  $ is identified with the rectangle
$\left(  _{0}^{k},_{0}^{k}\right)  .$ Indeed we can write
\begin{align*}
&  \widehat{U}_{3}\left(  _{k\left(  0\right)  ,}^{k\left(  1\right)  ,}%
%TCIMACRO{\QATOP{k^{\ast}\left(  1\right)  }{k^{\ast}\left(  0\right)  }}%
%BeginExpansion
\genfrac{}{}{0pt}{}{k^{\ast}\left(  1\right)  }{k^{\ast}\left(  0\right)  }%
%EndExpansion
\right) \\
&  =\sum_{\left(  s_{1},s_{2}\right)  \in\left(  _{k\left(  0\right)
,}^{k\left(  1\right)  ,}%
%TCIMACRO{\QATOP{k^{\ast}\left(  1\right)  }{k^{\ast}\left(  0\right)  }}%
%BeginExpansion
\genfrac{}{}{0pt}{}{k^{\ast}\left(  1\right)  }{k^{\ast}\left(  0\right)  }%
%EndExpansion
\right)  }\left\{
\begin{array}
[c]{c}%
\widehat{\epsilon}_{i_{1}}z_{s_{1}}\left(  X_{i_{1}}\right)  \times\\
\left\{  \left[  \dot{B}\dot{P}H_{1}\right]  _{i_{2}}z_{s_{1}}\left(
X_{i_{2}}\right)  z_{s_{2}}\left(  X_{i_{2}}\right)  -I\left[  s_{1}%
=s_{2}\right]  \right\}  z_{s_{2}}\left(  X_{i_{3}}\right)  \widehat{\Delta
}_{i_{3}}%
\end{array}
\right\}
\end{align*}
where, here and below, $s_{1}$ and $s_{2}$ are restricted to be integers, so
$\left(  s_{1},s_{2}\right)  \in\left(  _{k\left(  0\right)  ,}^{k\left(
1\right)  ,}%
%TCIMACRO{\QATOP{k^{\ast}\left(  1\right)  }{k^{\ast}\left(  0\right)  }}%
%BeginExpansion
\genfrac{}{}{0pt}{}{k^{\ast}\left(  1\right)  }{k^{\ast}\left(  0\right)  }%
%EndExpansion
\right)  $ are the lattice points of the rectangle.

We next study the variance of $\widehat{\mathbb{U}}_{3}\left(  _{k\left(
0\right)  ,}^{k\left(  1\right)  ,}%
%TCIMACRO{\QATOP{k^{\ast}\left(  1\right)  }{k^{\ast}\left(  0\right)  }}%
%BeginExpansion
\genfrac{}{}{0pt}{}{k^{\ast}\left(  1\right)  }{k^{\ast}\left(  0\right)  }%
%EndExpansion
\right)  .$ It follows from Theorem \ref{var_multi} above that the number of
lattice points in $\left(  _{k\left(  0\right)  ,}^{k\left(  1\right)  ,}%
%TCIMACRO{\QATOP{k^{\ast}\left(  1\right)  }{k^{\ast}\left(  0\right)  }}%
%BeginExpansion
\genfrac{}{}{0pt}{}{k^{\ast}\left(  1\right)  }{k^{\ast}\left(  0\right)  }%
%EndExpansion
\right)  $ is proportional to the variance of $\widehat{\mathbb{U}}_{3}\left(
_{k\left(  0\right)  ,}^{k\left(  1\right)  ,}%
%TCIMACRO{\QATOP{k^{\ast}\left(  1\right)  }{k^{\ast}\left(  0\right)  }}%
%BeginExpansion
\genfrac{}{}{0pt}{}{k^{\ast}\left(  1\right)  }{k^{\ast}\left(  0\right)  }%
%EndExpansion
\right)  $ so if $k\left(  0\right)  <<k\left(  1\right)  $ and $k^{\ast
}\left(  0\right)  <<k^{\ast}\left(  1\right)  $ then $\ var\left[
\widehat{\mathbb{U}}_{3}\left(  _{k\left(  0\right)  ,}^{k\left(  1\right)  ,}%
%TCIMACRO{\QATOP{k^{\ast}\left(  1\right)  }{k^{\ast}\left(  0\right)  }}%
%BeginExpansion
\genfrac{}{}{0pt}{}{k^{\ast}\left(  1\right)  }{k^{\ast}\left(  0\right)  }%
%EndExpansion
\right)  \right]  $ and $var\left[  \widehat{\mathbb{U}}_{3}\left(
_{0,}^{k\left(  1\right)  ,}%
%TCIMACRO{\QATOP{k^{\ast}\left(  1\right)  }{0}}%
%BeginExpansion
\genfrac{}{}{0pt}{}{k^{\ast}\left(  1\right)  }{0}%
%EndExpansion
\right)  \right]  $ are both of order $k\left(  1\right)  k^{\ast}\left(
1\right)  /n^{3}.$ Hence the order of the variance of $\widehat{\mathbb{U}%
}_{3}\left(  _{k\left(  0\right)  ,}^{k\left(  1\right)  ,}%
%TCIMACRO{\QATOP{k^{\ast}\left(  1\right)  }{k^{\ast}\left(  0\right)  }}%
%BeginExpansion
\genfrac{}{}{0pt}{}{k^{\ast}\left(  1\right)  }{k^{\ast}\left(  0\right)  }%
%EndExpansion
\right)  $ is determined by the vertex of the rectangle $\left(  _{k\left(
0\right)  ,}^{k\left(  1\right)  ,}%
%TCIMACRO{\QATOP{k^{\ast}\left(  1\right)  }{k^{\ast}\left(  0\right)  }}%
%BeginExpansion
\genfrac{}{}{0pt}{}{k^{\ast}\left(  1\right)  }{k^{\ast}\left(  0\right)  }%
%EndExpansion
\right)  $ furthest from the origin.

In contrast by a theorem in the appendix, the mean $E\left[
\widehat{\mathbb{U}}_{3}\left(  _{k\left(  0\right)  ,}^{k\left(  1\right)  ,}%
%TCIMACRO{\QATOP{k^{\ast}\left(  1\right)  }{k^{\ast}\left(  0\right)  }}%
%BeginExpansion
\genfrac{}{}{0pt}{}{k^{\ast}\left(  1\right)  }{k^{\ast}\left(  0\right)  }%
%EndExpansion
\right)  \right]  $ is
\[
\widehat{E}\left(  \widehat{\Pi}\left[  \delta b|\overline{Z}_{k\left(
0\right)  }^{k\left(  1\right)  }\right]  \delta g\widehat{Q}^{2}\widehat{\Pi
}\left[  \delta p|\overline{Z}_{k^{\ast}\left(  0\right)  }^{k^{\ast}\left(
1\right)  }\right]  \right)  \left(  1+o_{p}\left(  1\right)  \right)
\]
with $\delta b$ $=\dot{P}\widehat{E}\left(  H_{1}|X\right)  \left(
\widehat{B}-B\right)  ,$ $\delta p=\dot{B}\widehat{E}\left(  H_{1}|X\right)
\left(  \widehat{P}-P\right)  ,$ $\delta g=\frac{g\left(  X\right)
-\widehat{g}\left(  X\right)  }{\widehat{g}\left(  X\right)  }$ and
$\widehat{Q}^{2}=\dot{B}\dot{P}\widehat{E}\left(  H_{1}|X\right)  .$ It
follows that if $k\left(  0\right)  <<k\left(  1\right)  $ and $k^{\ast
}\left(  0\right)  <<k^{\ast}\left(  1\right)  $ then $E\left[
\widehat{\mathbb{U}}_{3}\left(  _{k\left(  0\right)  ,}^{k\left(  1\right)  ,}%
%TCIMACRO{\QATOP{k^{\ast}\left(  1\right)  }{k^{\ast}\left(  0\right)  }}%
%BeginExpansion
\genfrac{}{}{0pt}{}{k^{\ast}\left(  1\right)  }{k^{\ast}\left(  0\right)  }%
%EndExpansion
\right)  \right]  $ and $E\left[  \widehat{\mathbb{U}}_{3}\left(  _{k\left(
0\right)  ,}^{\infty,}%
%TCIMACRO{\QATOP{\infty}{k^{\ast}\left(  0\right)  }}%
%BeginExpansion
\genfrac{}{}{0pt}{}{\infty}{k^{\ast}\left(  0\right)  }%
%EndExpansion
\right)  \right]  $ are both of order
\[
O_{p}\left[  k\left(  0\right)  ^{-\beta_{b}}k^{\ast}\left(  0\right)
^{-\beta_{p}}\left(  n/\log n\right)  ^{\frac{-\beta_{g}}{2\beta_{g}+1}%
}\right]  .
\]
To see this for $E\left[  \widehat{\mathbb{U}}_{3}\left(  _{k\left(  0\right)
,}^{k\left(  1\right)  ,}%
%TCIMACRO{\QATOP{k^{\ast}\left(  1\right)  }{k^{\ast}\left(  0\right)  }}%
%BeginExpansion
\genfrac{}{}{0pt}{}{k^{\ast}\left(  1\right)  }{k^{\ast}\left(  0\right)  }%
%EndExpansion
\right)  \right]  $, we $`$sup out' $\left\vert \delta g\widehat{Q}%
^{2}\right\vert $ from $\widehat{E}\left(  \left\vert \widehat{\Pi}\left[
\delta b|\overline{Z}_{k\left(  0\right)  }^{k\left(  1\right)  }\right]
\delta g\widehat{Q}^{2}\widehat{\Pi}\left[  \delta p|\overline{Z}_{k^{\ast
}\left(  0\right)  }^{k^{\ast}\left(  1\right)  }\right]  \right\vert \right)
$ which is
\[
O_{p}\left[  \left(  n/\log n\right)  ^{\frac{-\beta_{g}}{2\beta_{g}+1}%
}\right]  \widehat{E}\left(  \left\vert \widehat{\Pi}\left[  \delta
b|\overline{Z}_{k\left(  0\right)  }^{k\left(  1\right)  }\right]
\widehat{\Pi}\left[  \delta p|\overline{Z}_{k^{\ast}\left(  0\right)
}^{k^{\ast}\left(  1\right)  }\right]  \right\vert \right)  .
\]
We then apply Cauchy Schwartz to $\widehat{E}\left(  \left\vert \widehat{\Pi
}\left[  \delta b|\overline{Z}_{k\left(  0\right)  }^{k\left(  1\right)
}\right]  \widehat{\Pi}\left[  \delta p|\overline{Z}_{k^{\ast}\left(
0\right)  }^{k^{\ast}\left(  1\right)  }\right]  \right\vert \right)  ,$
noting that $\widehat{E}\left(  \left\{  \widehat{\Pi}\left[  \delta
b|\overline{Z}_{k\left(  0\right)  }^{k\left(  1\right)  }\right]  \right\}
^{2}\right)  ^{1/2}=O\left(  k\left(  0\right)  ^{-\beta_{b}}\right)  $. Again
a more careful argument using H\"{o}lder's inequality would show the log
factor is unnecessary. Hence the order of the mean of $\widehat{\mathbb{U}%
}_{3}\left(  _{k\left(  0\right)  ,}^{k\left(  1\right)  ,}%
%TCIMACRO{\QATOP{k^{\ast}\left(  1\right)  }{k^{\ast}\left(  0\right)  }}%
%BeginExpansion
\genfrac{}{}{0pt}{}{k^{\ast}\left(  1\right)  }{k^{\ast}\left(  0\right)  }%
%EndExpansion
\right)  $ is determined by the vertex of the rectangle $\left(  _{k\left(
0\right)  ,}^{k\left(  1\right)  ,}%
%TCIMACRO{\QATOP{k^{\ast}\left(  1\right)  }{k^{\ast}\left(  0\right)  }}%
%BeginExpansion
\genfrac{}{}{0pt}{}{k^{\ast}\left(  1\right)  }{k^{\ast}\left(  0\right)  }%
%EndExpansion
\right)  $ closest to the origin.

\textbf{Motivation:} With this background we are ready to motivate our new
estimator. Recall from Section \ref{DR_CI_Section}, that with $g$ known, the
choice $k_{opt}^{g}\left(  2\right)  =n^{\frac{2}{1+4\beta/d}}$ gives $\left(
\widehat{\psi}_{2,k_{opt}^{g}\left(  2\ \right)  }-\psi\right)  =O_{p}\left(
n^{-\frac{4\beta\ }{4\beta+d}}\right)  $ because the truncation bias
$\left\vert \widetilde{\psi}_{k_{opt}^{g}\left(  2\right)  }-\psi\right\vert $
and variance are of order $n^{-\frac{4\beta\ }{4\beta+d}}$ and the estimation
bias is zero. Any choice of $k$ larger than $k_{opt}^{g}\left(  2\right)  $
will result in a slower rate of convergence.

However, when $g$ is unknown and thus estimated, $\widehat{\psi}%
_{2,k_{opt}^{g}\left(  2\ \right)  }-\psi$ does not attain the optimal rate of
convergence because the estimation bias $n^{-\left(  \frac{\beta_{g}}%
{2\beta_{g}+d}+\frac{\beta_{b}}{d+2\beta_{b}}+\frac{\beta_{p}}{d+2\beta_{p}%
}\right)  }$ exceeds $n^{-\frac{4\beta\ }{4\beta+d}}.$ The estimator
$\widehat{\psi}_{3,k_{opt}^{g}\left(  2\ \right)  }=\widehat{\psi}%
_{2,k_{opt}^{g}\left(  2\ \right)  }+\widehat{\mathbb{U}}_{3}\left(
_{0,}^{k_{opt}^{g}\left(  2\ \right)  ,}%
%TCIMACRO{\QATOP{k_{opt}^{g}\left(  2\ \right)  }{0}}%
%BeginExpansion
\genfrac{}{}{0pt}{}{k_{opt}^{g}\left(  2\ \right)  }{0}%
%EndExpansion
\right)  $ also fails to attain the rate $n^{-\frac{4\beta\ }{4\beta+d}}$
because it has variance of the order of
\[
\frac{k_{opt}^{g}\left(  2\ \right)  }{n}\frac{k_{opt}^{g}\left(  2\ \right)
}{n^{2}}=O\left(  \frac{n^{\frac{2}{1+4\beta/d}}}{n}n^{-\frac{8\beta}%
{4\beta+d}}\right)  ,
\]
which exceeds $O\left(  n^{-\frac{8\beta}{4\beta+d}}\right)  .$ On the other
hand, $\widehat{\psi}_{3,k_{opt}^{g}\left(  2\ \right)  }$ has bias of
$O_{p}\left(  n^{-\frac{4\beta}{4\beta+d}}\right)  $ because the truncation
bias is $O_{p}\left(  n^{-\frac{4\beta}{4\beta+d}}\right)  $ and the
estimation bias $O_{p}\left(  n^{-\left(  \frac{2\beta_{g}}{2\beta_{g}%
+d}+\frac{\beta_{b}}{d+2\beta_{b}}+\frac{\beta_{p}}{d+2\beta_{p}}\right)
}\right)  $ is also $O_{p}\left(  n^{-\frac{4\beta}{4\beta+d}}\right)  $ under
our assumption $\left(  \ref{fourth}\right)  $. Our strategy will be to try to
replace the term $\widehat{\mathbb{U}}_{3}\left(  _{0,}^{k_{opt}^{g}\left(
2\ \right)  ,}%
%TCIMACRO{\QATOP{k_{opt}^{g}\left(  2\ \right)  }{0}}%
%BeginExpansion
\genfrac{}{}{0pt}{}{k_{opt}^{g}\left(  2\ \right)  }{0}%
%EndExpansion
\right)  $ in the estimator $\widehat{\psi}_{3,k_{opt}^{g}\left(  2\ \right)
}=\widehat{\psi}_{2,k_{opt}^{g}\left(  2\ \right)  }+\widehat{\mathbb{U}}%
_{3}\left(  _{0,}^{k_{opt}^{g}\left(  2\ \right)  ,}%
%TCIMACRO{\QATOP{k_{opt}^{g}\left(  2\ \right)  }{0}}%
%BeginExpansion
\genfrac{}{}{0pt}{}{k_{opt}^{g}\left(  2\ \right)  }{0}%
%EndExpansion
\right)  $ by%

\[
\widehat{U}_{3}\left(  \Omega\right)  =\sum_{\left(  s_{1},s_{2}\right)
\in\Omega}\left\{
\begin{array}
[c]{c}%
\widehat{\epsilon}_{i_{1}}z_{s_{1}}\left(  X_{i_{1}}\right)  z_{s_{2}}\left(
X_{i_{3}}\right)  \widehat{\Delta}_{i_{3}}\times\\
\left\{  \left[  \dot{B}\dot{P}H_{1}\right]  _{i_{2}}z_{s_{1}}\left(
X_{i_{2}}\right)  z_{s_{2}}\left(  X_{i_{2}}\right)  -I\left[  s_{1}%
=s_{2}\right]  \right\}
\end{array}
\right\}
\]

where $\Omega$ is a subset of the rectangle $\left(  _{0,}^{k_{opt}^{g}\left(
2\ \right)  ,}%
%TCIMACRO{\QATOP{k_{opt}^{g}\left(  2\ \right)  }{0}}%
%BeginExpansion
\genfrac{}{}{0pt}{}{k_{opt}^{g}\left(  2\ \right)  }{0}%
%EndExpansion
\right)  $ such that $var\left(  \widehat{U}_{3}\left(  \Omega\right)
\right)  \asymp n^{-\frac{8\beta\ }{4\beta+d}}$ but the additional bias
\begin{align*}
&  E\left[  \widehat{\mathbb{U}}_{3}\left(  _{0,}^{k_{opt}^{g}\left(
2\ \right)  ,}%
%TCIMACRO{\QATOP{k_{opt}^{g}\left(  2\ \right)  }{0}}%
%BeginExpansion
\genfrac{}{}{0pt}{}{k_{opt}^{g}\left(  2\ \right)  }{0}%
%EndExpansion
\right)  -\widehat{U}_{3}\left(  \Omega\right)  \right] \\
&  =E\left[  \widehat{\mathbb{U}}_{3}\left(  \left(  _{0,}^{k_{opt}^{g}\left(
2\ \right)  ,}%
%TCIMACRO{\QATOP{k_{opt}^{g}\left(  2\ \right)  }{0}}%
%BeginExpansion
\genfrac{}{}{0pt}{}{k_{opt}^{g}\left(  2\ \right)  }{0}%
%EndExpansion
\right)  \backslash\Omega\right)  \right] \\
&  \equiv E\left[  \sum_{\left(  s_{1},s_{2}\right)  \in\left(  _{0,}%
^{k_{opt}^{g}\left(  2\ \right)  ,}%
%TCIMACRO{\QATOP{k_{opt}^{g}\left(  2\ \right)  }{0}}%
%BeginExpansion
\genfrac{}{}{0pt}{}{k_{opt}^{g}\left(  2\ \right)  }{0}%
%EndExpansion
\right)  \backslash\Omega}\left\{
\begin{array}
[c]{c}%
\widehat{\epsilon}_{i_{1}}z_{s_{1}}\left(  X_{i_{1}}\right)  \left\{
\begin{array}
[c]{c}%
\left[  \dot{B}\dot{P}H_{1}\right]  _{i_{2}}z_{s_{1}}\left(  X_{i_{2}}\right)
z_{s_{2}}\left(  X_{i_{2}}\right) \\
-I\left[  s_{1}=s_{2}\right]
\end{array}
\right\} \\
\times z_{s_{2}}\left(  X_{i_{3}}\right)  \widehat{\Delta}_{i_{3}}%
\end{array}
\right\}  \right]
\end{align*}
is $O_{p}\left(  n^{-\frac{4\beta\ }{4\beta+d}}\right)  .$ This approach will
succeed if we can chose $\Omega$ and thus $\left(  _{0,}^{k_{opt}^{g}\left(
2\ \right)  ,}%
%TCIMACRO{\QATOP{k_{opt}^{g}\left(  2\ \right)  }{0}}%
%BeginExpansion
\genfrac{}{}{0pt}{}{k_{opt}^{g}\left(  2\ \right)  }{0}%
%EndExpansion
\right)  \backslash\Omega$ to be sums of rectangles (whose number does not
increase with $n)$ such that (i) each rectangle in $\left(  _{0,}^{k_{opt}%
^{g}\left(  2\ \right)  ,}%
%TCIMACRO{\QATOP{k_{opt}^{g}\left(  2\ \right)  }{0}}%
%BeginExpansion
\genfrac{}{}{0pt}{}{k_{opt}^{g}\left(  2\ \right)  }{0}%
%EndExpansion
\right)  \backslash\Omega$ has its closest vertex to the origin, say $\left(
k\left(  0\right)  ,k^{\ast}\left(  0\right)  \right)  ,$ satisfying
$\ O_{p}\left[  k\left(  0\right)  ^{-\beta_{b}}k^{\ast}\left(  0\right)
^{-\beta_{p}}n^{\frac{-\beta_{g}}{2\beta_{g}+1}}\right]  \leq n^{-\frac
{4\beta\ }{4\beta+d}}$ and (ii) simultaneously each rectangle in $\Omega$ has
its furthest vertex from the origin, say $\left(  k\left(  1\right)  ,k^{\ast
}\left(  1\right)  \right)  ,$ satisfying $O\left(  k\left(  1\right)
k^{\ast}\left(  1\right)  /n^{3}\right)  =O\left(  n^{-\frac{8\beta\ }%
{4\beta+d}}\right)  .$

We index the vertices of our set of rectangles as follows. Consider a natural
number $J$ and a set of non-negative integers $\mathcal{K}_{J,tot}=\left\{
k_{-2},k_{-1},k_{0},.....,k_{2J},k_{2J+1},k_{2J+2}\right\}  $ satisfying
$0=k_{-2}<k_{0}<k_{2}<,...,<k_{2J-2}<k_{2J}<k_{2J+2}=k_{2J+1}<k_{2J-1}%
<,...,<k_{1}<k_{-1} $

Note the elements with even subscripts increase from $0$ to $2J+2$ while
elements with odd subscripts decrease from $-1$ to $2J-1.$ Further the
smallest element with odd subscript equals the largest element with even
subscript. We will use two such sets $\mathcal{K}_{b,J,tot}$ and
$\mathcal{K}_{p,J,tot}$ with corresponding elements $k_{bl}$ and $k_{pl}%
\ $with $\ k_{b,-1}=k_{p,-1}$.$\ $

$\ $Set for $s\in\left\{  -1,0,..,J\right\}  $
\begin{align}
k_{b,2s+1}  &  =n^{\frac{3d+4\beta}{\left(  d+4\beta\right)  }}/k_{p,2s+2}%
,\label{variance1}\\
k_{p,2s+1}  &  =n^{\frac{3d+4\beta}{\left(  d+4\beta\right)  }}/k_{b,2s+2}%
,\text{ so}\label{variance3}\\
\frac{k_{p,2s+1}k_{b,2s+2}}{n^{3}}  &  =\frac{k_{b,2s+1}k_{p,2s+2}}{n^{3}%
}=n^{-\frac{8\beta}{4\beta+d}}\nonumber
\end{align}
We leave $J$,$\mathcal{K}_{p,J}=\left\{  k_{p,2s},s=0,...,J+1\right\}  ,$ and
$\mathcal{K}_{b,J}=\left\{  k_{b,2s},s=0,...,J+1\right\}  $ unspecified for
now but derive optimal values below.

Let $\Omega=\Omega\left(  \mathcal{K}_{pJ},\mathcal{K}_{bJ}\ \right)  $ be the
union of rectangles%

\[
\Omega\left(  \mathcal{K}_{pJ},\mathcal{K}_{bJ}\right)  =\left\{  \cup
_{s=0}^{J}\left(  _{k_{p,2s-2\ \ ,}k_{b,2s-2\ \ }}^{k_{p,2s-1},k_{b,2s}%
}\right)  \cup\left(  _{k_{p,2s-2\ \ }k_{b,2s\ \ }}^{k_{p,2s},k_{b,2s-1}%
}\right)  \right\}  \cup\left(  _{k_{p,2J\ \ }k_{b,2J\ }}^{k_{p,2J+1}%
k_{b,2J+1}}\right)
\]

The points $\left(  k_{p,2s+1},k_{b,2s+2}\right)  ,\left(  k_{p,2s+2}%
,k_{b,2s+1}\right)  $ for $s=-1,0...,J+1$ lie on a hyperbola $Hy\ $ in $R^{2}$
defined by $Hy\ =\left\{  \left(  r_{1},r_{2}\right)  ;r_{1}r_{2}%
=n^{\frac{3d+4\beta}{\left(  d+4\beta\right)  }}\right\}  $ shown in figure 1
for $J=2.$ The set $\Omega\left(  \mathcal{K}_{pJ},\mathcal{K}_{bJ}\right)
\subset\left(  _{0,}^{k_{opt}^{g}\left(  2\ \right)  ,}%
%TCIMACRO{\QATOP{k_{opt}^{g}\left(  2\ \right)  }{0}}%
%BeginExpansion
\genfrac{}{}{0pt}{}{k_{opt}^{g}\left(  2\ \right)  }{0}%
%EndExpansion
\right)  $ lies below $Hy.$%

%TCIMACRO{\FRAME{ftbpF}{4.9165in}{4.3154in}{0pt}{}{}{Figure}%
%maintain-aspect-ratio TRUE;  display "USEDEF";  valid_file "T";
%width 4.9165in;  height 4.3154in;  depth 0pt;  original-width 2.4621in;
%original-height 2.1577in;  cropleft "0";  croptop "1";  cropright "1";
%cropbottom "0";  tempfilename 'NYNBDB00.wmf';tempfile-properties "XPR";}} }%
%BeginExpansion

%\begin{figure}[ptb]%
%\centering
%\includegraphics[
%natheight=2.157700in,
%natwidth=2.462100in,
%height=4.3154in,
%width=4.9165in
%]%
%{NYNBDB00.wmf}%
%\end{figure}
%EndExpansion

\begin{figure}
\includegraphics[scale=0.5]{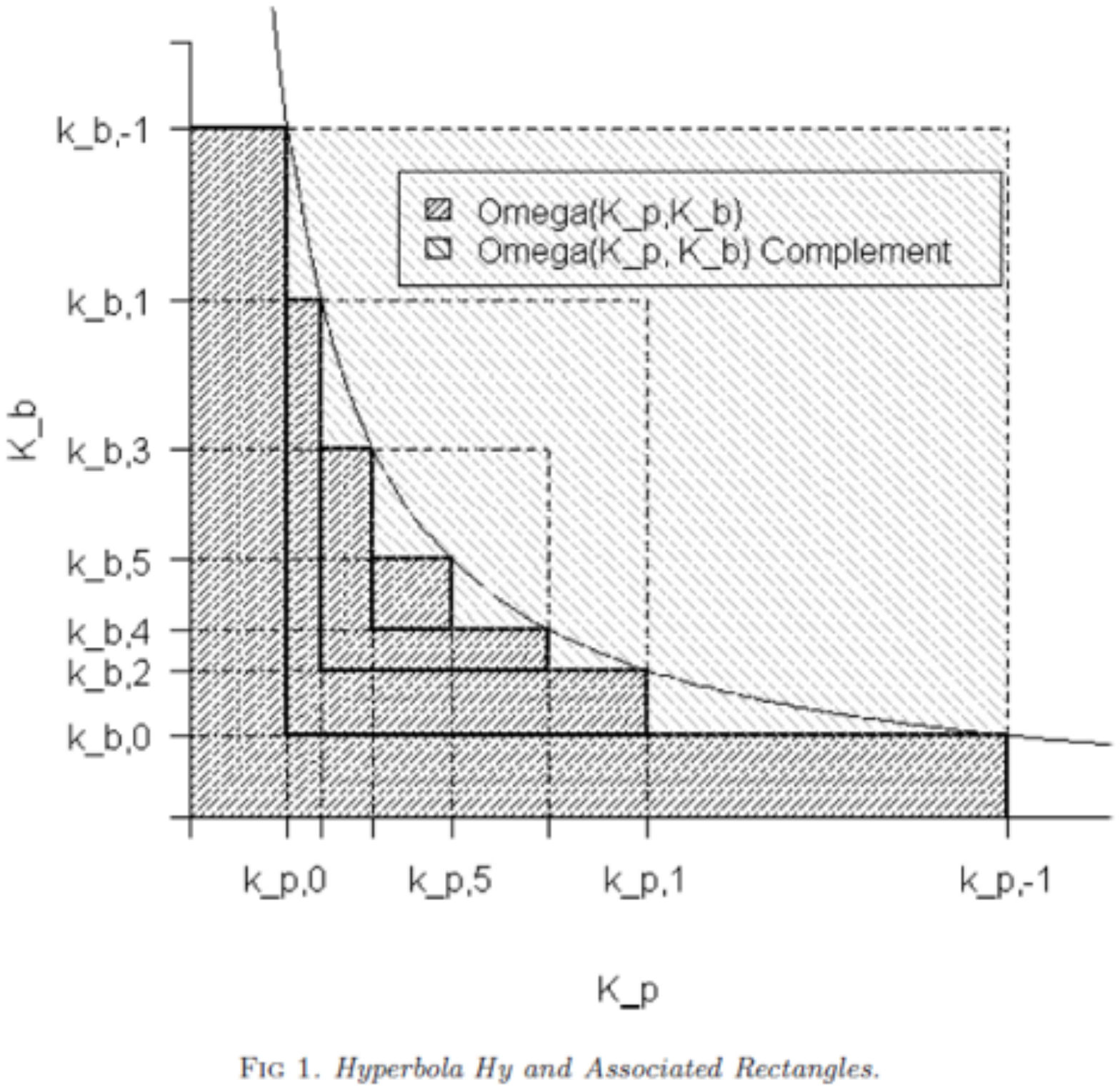}
\end{figure}

Define%
\[
\widehat{\overline{\psi}}_{3,\left(  \mathcal{K}_{pJ},\mathcal{K}_{bJ}\right)
}=\widehat{\psi}_{2,k_{-1}\ }+\widehat{\mathbb{U}}_{3}\left(  \Omega\left(
\mathcal{K}_{pJ},\mathcal{K}_{bJ}\right)  \right)  .
\]

We then have

\begin{theorem}
\label{uy}(i): The estimator $\widehat{\overline{\psi}}_{3,\left(
\mathcal{K}_{pJ},\mathcal{K}_{bJ}\right)  }$ has variance of the order of
\[
\frac{k_{-1}}{n^{2}}+\left(  2J+1\right)  n^{-\frac{8\beta}{4\beta+d}}%
\]
and bias $E\left(  \widehat{\overline{\psi}}_{3,\left(  \mathcal{K}%
_{pJ},\mathcal{K}_{bJ}\right)  }\right)  -\psi$ of order
\begin{gather*}
O_{p}\left\{  n^{-\frac{\beta_{g}}{2\beta_{g}+d}}\left(  \sum_{s=0}^{J}\left(
k_{p,2s+1}^{-\beta_{p}/d}k_{b,2s}^{-\beta_{b}/d}+k_{b,2s+1}^{-\beta_{b}%
/d}k_{p,2s}^{-\beta_{p}/d}\right)  \right)  \right\} \\
+O_{p}\left(  n^{-\left(  \frac{2\beta_{g}}{2\beta_{g}+d}+\frac{\beta_{b}%
}{d+2\beta_{b}}+\frac{\beta_{p}}{d+2\beta_{p}}\right)  }\right)  +O_{p}\left(
k_{-1}^{-\left(  \beta_{p}+\beta_{b}\right)  /d}\right)
\end{gather*}

\end{theorem}

Proof: Each of the $2J+1$ rectangles whose union is $\Omega\left(
\mathcal{K}_{pJ},\mathcal{K}_{bJ}\right)  $ has $\left(  k_{p,2s+1}%
,k_{b,2s+2}\right)  $ or $\left(  k_{p,2s+2},k_{b,2s+1}\right)  $ for some
$s\in\left\{  -1,0,..,J\right\}  $ as the vertex furthest from the origin and
thus contributes $\frac{k_{p,2s+1}k_{b,2s+2}}{n^{3}}=n^{-\frac{8\beta}%
{4\beta+d}}$ to the variance of $\widehat{\overline{\psi}}_{3,\left(
\mathcal{K}_{pJ},\mathcal{K}_{bJ}\right)  }.$ The variance of $\widehat{\psi
}_{2,k_{-1}}\asymp\frac{k_{-1}}{n^{2}}.$ Now
\begin{align*}
&  E\left(  \widehat{\overline{\psi}}_{3,\left(  \mathcal{K}_{pJ}%
,\mathcal{K}_{bJ},\right)  }\right)  -\psi\\
&  =\left\{  E\left(  \widehat{\psi}_{3,k_{-1}}\right)  -\psi\right\}
\,+\left\{  E\left[  \widehat{\mathbb{U}}_{3}\left\{  \Omega\left(  \left(
\mathcal{K}_{pJ},\mathcal{K}_{bJ}\right)  \right)  \right\}  \right]
-E\left[  \widehat{\mathbb{U}}_{3}\left\{  \left(  _{0,}^{k_{-1},}%
%TCIMACRO{\QATOP{k_{-1}}{0}}%
%BeginExpansion
\genfrac{}{}{0pt}{}{k_{-1}}{0}%
%EndExpansion
\right)  \right\}  \right]  \right\} \\
&  =O_{p}\left(  k_{-1}^{-\left(  \beta_{p}+\beta_{b}\right)  /d}\right)
+O_{p}\left(  n^{-\left(  \frac{2\beta_{g}}{2\beta_{g}+d}+\frac{\beta_{b}%
}{d+2\beta_{b}}+\frac{\beta_{p}}{d+2\beta_{p}}\right)  }\right) \\
&  +E\left[  \widehat{\mathbb{U}}_{3}\left\{  \left(  _{0,}^{k_{-1},}%
%TCIMACRO{\QATOP{k_{-1}}{0}}%
%BeginExpansion
\genfrac{}{}{0pt}{}{k_{-1}}{0}%
%EndExpansion
\right)  \backslash\Omega\left(  \left(  \mathcal{K}_{pJ},\mathcal{K}%
_{bJ}\right)  \right)  \right\}  \right]
\end{align*}
As is evident from Figure 1, $\Omega^{c}(\mathcal{K}_{pJ}, \mathcal{K}_{bJ})
\equiv\left(  _{0,}^{k_{-1},}%
%TCIMACRO{\QATOP{k_{-1}}{0}}%
%BeginExpansion
\genfrac{}{}{0pt}{}{k_{-1}}{0}%
%EndExpansion
\right)  \backslash\Omega\left(  \left(  \mathcal{K}_{pJ},\mathcal{K}%
_{bJ}\right)  \right)  $ is the union of rectangles $\cup_{s=0}^{J}\left\{
\left(  _{k_{p,2s},}^{k_{p,2s-1},}%
%TCIMACRO{\QATOP{\ k_{b,2s-1}}{k_{b,2s+1}}}%
%BeginExpansion
\genfrac{}{}{0pt}{}{\ k_{b,2s-1}}{k_{b,2s+1}}%
%EndExpansion
\right)  \cup\left(  _{k_{p,2s+1},}^{k_{p,2s-1},}%
%TCIMACRO{\QATOP{\ k_{b,2s+1}}{k_{b,2s}}}%
%BeginExpansion
\genfrac{}{}{0pt}{}{\ k_{b,2s+1}}{k_{b,2s}}%
%EndExpansion
\right)  \right\}  \ $which have
\[
\left\{  \left(  k_{p,2s},k_{b,2s+1}\right)  ,\left(  k_{p,2s+1}%
,k_{b,2s}\right)  ;s\in\left\{  -1,0,..,J\right\}  \right\}
\]
as the set of vertices closest to the origin, leading to the expression for
the bias given in the theorem.

\begin{theorem}
\label{uz}Given $\left(  \beta_{b},\beta_{p},\beta_{g}\right)  \ $with
$\beta_{p}\geq\beta_{b}$ so $\Delta\geq0,$ Eq.($\ref{NE11}$) holds if and only
if there exists $J,\mathcal{K}_{pJ},\mathcal{K}_{bJ}$ such that
$\widehat{\overline{\psi}}_{3,\left(  \mathcal{K}_{pJ},\mathcal{K}%
_{bJ},\ \right)  }-\psi=O_{p}\left(  n^{-\frac{4\beta}{4\beta+d}}\right)  .$

If Eq.($\ref{NE11}$) holds, $E\left[  \widehat{\mathbb{U}}_{3}\left\{  \left(
_{0,}^{k_{-1},}%
%TCIMACRO{\QATOP{k_{-1}}{0}}%
%BeginExpansion
\genfrac{}{}{0pt}{}{k_{-1}}{0}%
%EndExpansion
\right)  \backslash\Omega\left(  \left(  \mathcal{K}_{pJ},\mathcal{K}%
_{bJ}\right)  \right)  \right\}  \right]  =O_{p}\left(  n^{-\frac{4\beta
}{4\beta+d}}\right)  $ and thus $\widehat{\overline{\psi}}_{3,\left(
\mathcal{K}_{pJ},\mathcal{K}_{bJ}\right)  }-\psi=O_{p}\left(  n^{-\frac
{4\beta}{4\beta+d}}\right)  ,$ when we choose $J$ to be the smallest integer
such that

$^{{}}\left(  1+\Delta\right)  \left(  J+1\right)  +c^{\ast}\left(  \beta
_{g},\beta,\Delta\right)  \sum_{l=1}^{J+1}\left(  1+\Delta\right)
^{l-1}>\frac{3+4\beta/d}{\ 2\left(  1+4\beta/d\right)  }$ with
\[
c^{\ast}\left(  \beta_{g},\beta,\Delta\right)  =\left(  \frac{2\beta_{g}%
/d}{2\beta_{g}/d+1}\ \right)  \frac{\left(  \Delta+2\right)  }{4\beta
/d\ }-\frac{2\left(  \Delta+2\right)  }{4\beta/d+1}+\frac{3\ +4\beta
/d}{\left(  1+4\beta/d\right)  },
\]

$k_{b,0}=k_{p,0}=n$, $k_{b,2s}=k_{p,2s}=n^{\left(  1+\Delta\right)  s}%
n^{q\sum_{l=1}^{s}\left(  1+\Delta\right)  ^{l-1}}\ $for $s=1,...,J+1,$ with
$q=\left\{  \frac{3+4\beta/d}{\ 2\left(  1+4\beta/d\right)  }-\left(
1+\Delta\right)  \left(  J+1\right)  \right\}  /\sum_{l=1}^{J+1}\left(
1+\Delta\right)  ^{l-1}.$

Note $J$ does not depend on the sample size $n.$
\end{theorem}

Proof: From Theorem \ref{uy}, for the variance of $\widehat{\overline{\psi}%
}_{3,\left(  \mathcal{K}_{pJ},\mathcal{K}_{bJ}\right)  }$ to be $O_{p}\left(
n^{-\frac{8\beta}{4\beta+d}}\right)  ,$ $J$ cannot increase with $n.$ Further
for the second order truncation bias $O_{p}\left(  k_{-1}^{-\left(  \beta
_{p}+\beta_{b}\right)  /d}\right)  $ and the square root of the variance
$\frac{k_{-1}}{n^{2}}$ of $\ \widehat{\psi}_{2,k_{-1}}$ both to be
$O_{p}\left(  n^{-\frac{4\beta}{4\beta+d}}\right)  ,\ $we must have
$k_{-1}=k_{opt}^{g}\left(  2\right)  =n^{\frac{2}{1+4\beta/d}}.$ It then
follows from eqs. ($\ref{variance1}$) and (\ref{variance3}) that
$k_{p,0}=k_{b,0}=n.$

In order for $E\left[  \widehat{\mathbb{U}}_{3}\left\{  \left(  _{0,}%
^{k_{-1},}%
%TCIMACRO{\QATOP{k_{-1}}{0}}%
%BeginExpansion
\genfrac{}{}{0pt}{}{k_{-1}}{0}%
%EndExpansion
\right)  \backslash\Omega\left(  \left(  \mathcal{K}_{pJ},\mathcal{K}%
_{bJ}\right)  \right)  \right\}  \right]  =O_{p}\left(  n^{-\frac{4\beta
}{4\beta+d}}\right)  ,$ we require for $s=0,..,J$%

\begin{align}
n^{-\frac{2\beta_{g}}{2\beta_{g}+d}}\left\{  k_{b,2s}^{-2\beta_{b}%
/d}k_{p,2s+1}^{-2\beta_{p}/d}\right\}   &  \leq\ n^{-\frac{8\beta/d}%
{4\beta/d+1}}\label{bias}\\
n^{-\frac{2\beta_{g}}{2\beta_{g}+d}}\left\{  k_{p,2s}^{-2\beta_{p}%
/d}k_{b,2s+1}^{-2\beta_{b}/d}\right\}   &  \leq\ n^{-\frac{8\beta/d}%
{4\beta/d+1}}\label{bias2}%
\end{align}
Substituting for $k_{b,2s+1}$ in eq. (\ref{bias2}) using eq.(\ref{variance1})
and recalling that $\beta_{p}\geq\beta_{b}$ so $\Delta\geq0$, we obtain%

\begin{gather}
n^{-\frac{2\beta_{g}}{2\beta_{g}+d}}k_{p,2s}^{-2\beta_{p}/d}\left\{
\frac{n^{\frac{3d+4\beta}{\left(  d+4\beta\right)  }}}{k_{p,2s+2}}\right\}
^{-2\beta_{b}/d}\leq\ n^{-\frac{8\beta}{4\beta+d}}\label{gy}\\
\Leftrightarrow k_{p,2s+2}^{2\beta_{b}/d}\leq\ n^{\frac{2\beta_{g}/d}%
{2\beta_{g}/d+1}\ }n^{-\frac{8\beta/d}{4\beta/d+1}}k_{p,2s}^{2\beta_{p}%
/d}\left(  n^{\frac{3+4\beta/d}{\left(  1+4\beta/d\right)  }}\right)
^{2\beta_{b}/d}\nonumber\\
\Leftrightarrow k_{p,2s+2}\leq\ n^{\left(  \frac{2\beta_{g}/d}{2\beta_{g}%
/d+1}\ \right)  \frac{1}{2\beta_{b}/d}}n^{-\frac{8\beta/d}{4\beta/d+1}\frac
{1}{2\beta_{b}/d}}k_{p,2s}^{\frac{\beta_{p}}{\beta_{b}}}\left(  \ n^{\frac
{3+4\beta/d}{\left(  1+4\beta/d\right)  }}\right) \nonumber\\
\Leftrightarrow1\leq\frac{k_{p,2s+2}}{k_{p,2s}}\leq\ n^{\left(  \frac
{2\beta_{g}/d}{2\beta_{g}/d+1}\ \right)  \frac{1}{2\beta_{b}/d}}%
n^{-\frac{8\beta/d}{4\beta/d+1}\frac{1}{2\beta_{b}/d}}k_{p,2s}^{\Delta}\left(
\ n^{\frac{3\ +4\beta/d}{\left(  1+4\beta/d\right)  }}\right) \nonumber\\
\Leftrightarrow1\leq\frac{k_{p,2s+2}}{k_{p,2s}}\leq n^{c^{\ast}\left(
\beta_{g},\beta,\Delta\right)  }k_{p,2s}^{\Delta}\label{gx}\\
\Leftrightarrow1\leq n^{c^{\ast}\left(  \beta_{g},\beta,\Delta\right)
}n^{\Delta}\ \nonumber\\
\Leftrightarrow0\leq c^{\ast}\left(  \beta_{g},\beta,\Delta\right)
+\Delta\nonumber
\end{gather}

since $n=k_{0}\leq k_{p,2s}\leq k_{p,2s+2}.$

Solving the last expression for $\frac{2\beta_{g}/d}{2\beta_{g}/d+1}$,we obtain%

\begin{equation}
\frac{2\beta_{g}/d}{2\beta_{g}/d+1}\geq\frac{\frac{1-4\beta/d}{1+4\beta
/d}+\Delta\left\{  \frac{2}{4\beta/d+1}-1\ \right\}  }{\frac{\left(
\Delta+2\right)  }{4\beta/d\ }}=\left\{  \frac{4\beta/d}{\left(
\Delta+2\right)  }\right\}  \left(  \Delta+1\right)  \frac{1-4\beta
/d}{1+4\beta/d}\ \label{NE2}%
\end{equation}
which is equation eq.$\left(  \ref{NE11}\right)  ,$ except with a nonstrict
inequality. We have just deduced that the constraint $\left(  \ref{NE2}%
\right)  $ was due to restriction (\ref{bias2}). We have not yet considered
whether the restriction $\left(  \ref{bias}\right)  $ implies additional
constraints. We now show that it does not. Specifically if we set
$k_{p,2l}=k_{b,2l}$ for all $l\in\left\{  1,2,....,\ J+1\right\}  ,$ then
eq.$\left(  \ref{bias}\right)  $ is true whenever eq.$\left(  \ref{bias2}%
\right)  $ holds because of our assumption that $\Delta\geq0.$ Thus we can set
$\mathcal{K}_{pJ}=\mathcal{K}_{bJ}.$

Thus we have shown that if $\widehat{\overline{\psi}}_{3,\left(
\mathcal{K}_{pJ},\mathcal{K}_{bJ},\ \right)  }-\psi=O_{p}\left(
n^{-\frac{4\beta}{4\beta+d}}\right)  ,$ then $k_{-1}=n^{\frac{2}{1+4\beta/d}%
},$ $\left(  \ref{NE2}\right)  $ holds, and $J$ must not increase with $n.$

We next show that when the inequality is strict in $\left(  \ref{NE2}\right)
\ $and eq.$\left(  \ref{fourth}\right)  $ holds, we can find $\mathcal{K}%
_{J}=\mathcal{K}_{pJ}=\mathcal{K}_{bJ}$ for which $\widehat{\overline{\psi}%
}_{3,\mathcal{K}_{J}\mathcal{\ }}-\psi=O_{p}\left(  n^{-\frac{4\beta}%
{4\beta+d}}\right)  .$ We then complete the proof of the theorem by showing
that when $\left(  \ref{NE2}\right)  $ holds with an equality, there is no
choice of $\mathcal{K}_{J}$ for which $\widehat{\overline{\psi}}%
_{3,\mathcal{K}_{J}\mathcal{\ }}\ $converges at a rate better than
$O_{p}\left(  \left(  \log n\right)  n^{-\frac{4\beta}{4\beta+d}}\right)  $

Suppose the inequality is strict in $\left(  \ref{NE2}\right)  .$ Since
$k_{0}=n,$ eq.$\left(  \ref{gx}\right)  $ applied recursively suggests we
define $k_{2s}=n^{\left(  1+\Delta\right)  s}n^{c^{\ast}\left(  \beta
_{g},\beta,\Delta\right)  \sum_{l=1}^{s}\left(  1+\Delta\right)  ^{l-1}}$ for
$s=1,..,J+1$ and take $k_{2s+1}=\ \frac{n^{\frac{3d+4\beta}{\left(
d+4\beta\right)  }}}{k_{2s+2}}.$ {}However, this will not generally give
$k_{2J+1}=k_{2J+2}=n^{\left\{  \frac{3d+4\beta}{\left(  d+4\beta\right)
}\right\}  \frac{1}{2}}$ as required when $\mathcal{K}_{pJ}=\mathcal{K}_{bJ}.$
Instead we use the modified algorithm given in the statement of the theorem
which insures that $k_{2J+1}=k_{2J+2}=n^{\frac{3+4\beta/d}{\ 2\left(
1+4\beta/d\right)  }}$, as required. Since $J$ is not a function of $n,$ in
order to show $\widehat{\overline{\psi}}_{3,\mathcal{K}_{J}}$ converges at
rate $n^{-\frac{4\beta}{4\beta+d}},$ we only need to check the bias.

Now $\frac{k_{2s+2}}{k_{2s}}=n^{\left(  1+\Delta\right)  }n^{q\left(
1+\Delta\right)  ^{s-1}}=k_{0}^{\left(  1+\Delta\right)  }n^{q\left(
1+\Delta\right)  ^{s-1}}\leq$ $k_{0}^{\left(  1+\Delta\right)  }n^{c^{\ast
}\left(  \beta_{g},\beta,\Delta\right)  \left(  1+\Delta\right)  ^{s-1}}$
since $q\leq c^{\ast}\left(  \beta_{g},\beta,\Delta\right)  \ $so the bias of
$\widehat{\overline{\psi}}_{3,\mathcal{K}_{J}\mathcal{\ }}$ is $O_{P}\left(
n^{-\frac{4\beta}{4\beta+d}}\right)  ,$ as required.

Suppose now the equality holds in eq.$\left(  \ref{NE2}\right)  $ so $c^{\ast
}\left(  \beta_{g},\beta,\Delta\right)  +\Delta=0$ and continue to assume
eq.$\left(  \ref{fourth}\right)  $ holds. We now construct an estimator
$\widehat{\overline{\psi}}_{3,\mathcal{K}_{J}\mathcal{\ }}$ that converges at
rate $O_{P}\left(  n^{-\frac{4\beta}{4\beta+d}}\ln\left(  n\right)  \right)
\ $and show that no estimator in our class $\widehat{\overline{\psi}%
}_{3,\mathcal{K}_{J}\mathcal{\ }}$ converges at a faster rate. We conjecture
this rate is minimax when the equality in eq.$\left(  \ref{NE2}\right)  $
holds. Again $k_{2s+1}=\frac{n^{\frac{3d+4\beta}{\left(  d+4\beta\right)  }}%
}{k_{2s+2}}$ and by the previous arguments, $k_{0}=n,k_{-1}=n^{\frac
{2}{\ \left(  1+4\beta/d\right)  }},$ $k_{2J+1}=k_{2J+2}=\left\{
n^{\frac{3d+4\beta}{\left(  d+4\beta\right)  }}\right\}  ^{1/2}.$ We can
suppose that $k_{2s}=n\left\{  v\left(  n\right)  \right\}  ^{s}.\ $It remains
to determine $v\left(  n\right)  $ and $J=J\left(  n\right)  .$ We
know$\ J\left(  n\right)  $ must satisfy
\begin{align*}
k_{2J\left(  n\right)  +2}  &  =\left\{  n^{\frac{3d+4\beta}{\left(
d+4\beta\right)  }}\right\}  ^{1/2}=n\left\{  v\left(  n\right)  \right\}
^{J\left(  n\right)  +1}\text{ so}\\
v\left(  n\right)   &  =n^{\left(  \frac{3d+4\beta}{2\left(  d+4\beta\right)
}-1\right)  \frac{1}{J\left(  n\right)  +1}}.
\end{align*}
The variance of $\widehat{\overline{\psi}}_{3,\mathcal{K}_{J}\mathcal{\ }}$ is
of order $n^{-\frac{8\beta}{4\beta+d}}J\left(  n\right)  .$ Thus the order of
the bias will still equal that of the variance provided we multiply the RHS of
eq.$\left(  \ref{gy}\right)  $ by $J\left(  n\right)  $. Then eq.$\left(
\ref{gx}\right)  $ becomes $1\leq\frac{k_{p,2s+2}}{k_{p,2s}}\leq n^{c^{\ast
}\left(  \beta_{g},\beta,\Delta\right)  }k_{p,2s}^{\Delta}J\left(  n\right)
^{\frac{1}{2\beta/d\ }}.$ Since, $\frac{k_{p,2s+2}}{k_{p,2s}}=v\left(
n\right)  \ $and $\ n=k_{0}\leq k_{p,2s},$ we substitute $n^{\Delta}%
=k_{0}^{\Delta}$ for $k_{p,2s}^{\Delta}$ in the modified eq.$\left(
\ref{gx}\right)  $ which gives $v\left(  n\right)  =J\left(  n\right)
^{\frac{1}{2\beta/d\ }}.$ Hence $n^{\left(  \frac{3d+4\beta}{2\left(
d+4\beta\right)  }-1\right)  \frac{1}{J\left(  n\right)  +1}}=J\left(
n\right)  ^{\frac{1}{2\beta/d\ }}$ which implies that$.$
\begin{equation}
\frac{\ln\left(  n\right)  }{J\left(  n\right)  }=O\left(  \ln\left[  J\left(
n\right)  \right]  \right) \label{log}%
\end{equation}

To minimize the variance, we want the slowest growing function of $n$ that
satisfies eq.$\left(  \ref{log}\right)  ,$ which is $J\left(  n\right)
=\ln\left(  n\right)  ,$ as claimed.

\subsubsection{ Case 2: The estimation bias of the third order estimator
exceeds the optimal rate\label{case2}}

In this section we no longer assume that the estimation bias $n^{-\left(
\frac{2\beta_{g}}{2\beta_{g}+d}+\frac{\beta_{b}}{d+2\beta_{b}}+\frac{\beta
_{p}}{d+2\beta_{p}}\right)  }\ $of a third order estimator is less than
$n^{-\frac{4\beta}{4\beta+d}}.$\ Then even when eq.$\left(  \ref{NE2}\right)
$ holds with a strict inequality, $\widehat{\overline{\psi}}_{3,\mathcal{K}%
_{J}}$ does not achieve a $n^{-\frac{4\beta}{4\beta+d}}$ rate of convergence
because the fourth order bias $n^{-\left(  \frac{2\beta_{g}}{2\beta_{g}%
+d}+\frac{\beta_{b}}{d+2\beta_{b}}+\frac{\beta_{p}}{d+2\beta_{p}}\right)  }$%
{}exceeds $n^{-\frac{4\beta}{4\beta+d}}.$ \ However, we will now construct an
estimator $\widehat{\psi}_{\mathcal{K}_{J}}^{eff}\equiv\widehat{\psi
}_{\mathcal{K}_{J}}^{eff}\left(  \beta_{g},\beta_{b},\beta_{p}\right)  $ that
under our assumptions $Ai)-Aiv)$ does converge at rate $n^{-\frac{4\beta
}{4\beta+d}}$ whenever $\left(  \beta_{g},\beta_{b},\beta_{p}\right)  $ given
in assumption $Aiv)$ satisfy eq.$\left(  \ref{NE2}\right)  $ with a strict
inequality. Because the estimator is very complicated, we have chosen to only
define the estimator and give its properties in the text. The motivating ideas
for and the formal proofs of these properties are provided in the appendix.

To define the estimator, we need some additional notation. Define%
\begin{align*}
&  \widehat{\mathbb{U}}_{m}\left(  \left(  l\right)  _{k\left(  l,0\right)
}^{k\left(  l,1\right)  },1\leq l\leq m-1\right) \\
&  =\mathbb{V}_{m}\left(  \widehat{\epsilon}_{i_{1}}\overline{Z}_{k\left(
1,0\right)  ,i_{1}}^{k\left(  1,1\right)  T}%
%TCIMACRO{\dprod \limits_{u=2}^{m-1}}%
%BeginExpansion
{\displaystyle\prod\limits_{u=2}^{m-1}}
%EndExpansion
\left(  \dot{B}\dot{P}H_{1}\overline{Z}_{k\left(  u-1,0\right)  }^{k\left(
u-1,1\right)  }\overline{Z}_{k\left(  u,0\right)  }^{k\left(  u,1\right)
T}-I_{k_{u-1}\times k_{u}}\right)  \overline{Z}_{k\left(  m-1,0\right)
}^{k\left(  m-1,1\right)  }\widehat{\Delta}_{i_{m}}\right)
\end{align*}

where , $k_{u}=k\left(  u,1\right)  -k\left(  u,0\right)  ,$ $I_{k_{u-1}\times
k_{u}}=\left(  I_{ij}\right)  _{k_{u-1}\times k_{u}}$ with $I_{ij}=I\left(
i=j\right)  .$

Then define $\widehat{\mathbb{U}}_{m}\left(  _{k\left(  0\right)  }^{k\left(
1\right)  }\right)  $ as $\widehat{\mathbb{U}}_{m}\left(  \left(  l\right)
_{k\left(  0\right)  }^{k\left(  1\right)  },1\leq l\leq m-1\right)  .$
$\widehat{\mathbb{U}}_{m}^{\left(  u\right)  }\left(  _{k^{\ast}\left(
0\right)  }^{k^{\ast}\left(  1\right)  },_{k\left(  0\right)  }^{k\left(
1\right)  }\right)  $\ is defined as $\widehat{\mathbb{U}}_{m}\left(  \left(
l\right)  _{k\left(  l,0\right)  }^{k\left(  l,1\right)  },1\leq l\leq
m-1\right)  $ with $k\left(  l,1\right)  =k\left(  1\right)  ,k\left(
l,0\right)  =k\left(  0\right)  $ for $l\neq u,$ and $k\left(  u,1\right)
=k^{\ast}\left(  1\right)  ,$ $k\left(  u,0\right)  =k^{\ast}\left(  0\right)
.$ Next\ $\widehat{\mathbb{U}}_{m}^{\left(  u,u+1\right)  }\left(  _{k^{\ast
}\left(  0\right)  }^{k^{\ast}\left(  1\right)  },_{k^{\ast\ast}\left(
0\right)  }^{k^{\ast\ast}\left(  1\right)  },_{k\left(  0\right)  }^{k\left(
1\right)  }\right)  $ is defined as $\widehat{\mathbb{U}}_{m}\left(  \left(
l\right)  _{k\left(  l,0\right)  }^{k\left(  l,1\right)  },1\leq l\leq
m-1\right)  $ with $k\left(  l,1\right)  =k\left(  1\right)  ,k\left(
l,0\right)  =k\left(  0\right)  $ for $l\neq u$ and $l\neq u+1,$ $k\left(
u,1\right)  =k^{\ast}\left(  1\right)  , $ $k\left(  u,0\right)  =k^{\ast
}\left(  0\right)  ,$ $k\left(  u+1,1\right)  =k^{\ast\ast}\left(  1\right)
,$ $k\left(  u+1,0\right)  =k^{\ast\ast}\left(  0\right)  .$ We will use this
notation for $m=3,$ even though $\widehat{\mathbb{U}}_{3}^{\left(  1,2\right)
}\left(  _{k^{\ast}\left(  0\right)  }^{k^{\ast}\left(  1\right)  }%
,_{k^{\ast\ast}\left(  0\right)  }^{k^{\ast\ast}\left(  1\right)  },_{k\left(
0\right)  }^{k\left(  1\right)  }\right)  $ does not depend on $k\left(
0\right)  ,k\left(  1\right)  $ and is equal to $\widehat{\mathbb{U}}%
_{3}\left(  _{k^{\ast}\left(  0\right)  }^{k^{\ast}\left(  1\right)
},_{k^{\ast\ast}\left(  0\right)  }^{k^{\ast\ast}\left(  1\right)  }\right)  $
of the previous subsection.

Finally define%
\begin{align*}
\mathbb{H}_{v}^{\ast}  &  =\widehat{\mathbb{U}}_{v}\left(  _{0}^{k_{0}%
}\right)  +%
%TCIMACRO{\dsum \limits_{u=1}^{v-1}}%
%BeginExpansion
{\displaystyle\sum\limits_{u=1}^{v-1}}
%EndExpansion
\widehat{\mathbb{U}}_{v}^{\left(  u\right)  }\left(  _{k_{0}}^{k_{-1}}%
,_{0}^{k_{0}}\right) \\
\mathbb{G}\left(  s,v\right)   &  =%
%TCIMACRO{\dsum \limits_{u=1}^{v-2}}%
%BeginExpansion
{\displaystyle\sum\limits_{u=1}^{v-2}}
%EndExpansion
\left\{  \widehat{\mathbb{U}}_{v}^{\left(  u,u+1\right)  }\left(  _{k_{2s-2}%
}^{k_{2s-1}},_{k_{2s-2}}^{k_{2s}},_{0}^{k_{0}}\right)  +\widehat{\mathbb{U}%
}_{v}^{\left(  u,u+1\right)  }\left(  _{k_{2s-2}}^{k_{2s}},_{k_{2s}}%
^{k_{2s-1}},_{0}^{k_{0}}\right)  \right\} \\
\mathbb{Q}_{v}  &  =%
%TCIMACRO{\dsum \limits_{u=1}^{v-2}}%
%BeginExpansion
{\displaystyle\sum\limits_{u=1}^{v-2}}
%EndExpansion
\widehat{\mathbb{U}}_{v}^{\left(  u,u+1\right)  }\left(  _{k_{2J}}^{k_{2J+1}%
},_{k_{2J}}^{k_{2J+1}},_{0}^{k_{0}}\right)
\end{align*}

\begin{theorem}
\label{beyond4th}Given $\left(  \beta_{g},\beta_{b},\beta_{p}\right)  $
satisfying Eq.\ref{NE2} with a strict inequality$,$ define
\begin{equation}
m\left(  \beta_{g},\beta_{b},\beta_{p}\right)  =int\left\{  \left(
\frac{4\beta}{d+4\beta}-\frac{\beta_{b}}{d+2\beta_{b}}-\frac{\beta_{p}%
}{d+2\beta_{p}}\right)  \left(  2+\frac{d}{\beta_{g}}\right)  +1\right\}
+1\label{m}%
\end{equation}
$\ $to be the smallest integer such that $\left(  \frac{\log n}{n}\right)
^{\frac{\left(  m-1\right)  \beta_{g}}{d+2\beta_{g}}}n^{-\frac{\beta_{b}%
}{d+2\beta_{b}}-\frac{\beta_{p}}{d+2\beta_{p}}}<n^{-\frac{4\beta}{d+4\beta}},$
where $\beta=\frac{\beta_{b}+\beta_{p}}{2}.$ {}Let $\mathcal{K}_{J}$ ,$J,$
$\widehat{\overline{\psi}}_{3,\mathcal{K}_{J}}$ be as in Theorem
\ref{uz}{\Large \ }and define
\begin{align*}
&  \widehat{\psi}_{\mathcal{K}_{J}}^{eff}\left(  \beta_{g},\beta_{b},\beta
_{p}\right) \\
&  =\widehat{\overline{\psi}}_{3,\mathcal{K}_{J}}+%
%TCIMACRO{\dsum \limits_{v=4}^{m\left(  \beta_{g},\beta_{b},\beta_{p}\right)
%}}%
%BeginExpansion
{\displaystyle\sum\limits_{v=4}^{m\left(  \beta_{g},\beta_{b},\beta
_{p}\right)  }}
%EndExpansion
\left(  -1\right)  ^{v-1}\mathbb{H}_{v}^{\ast}+%
%TCIMACRO{\dsum \limits_{s=1}^{J}}%
%BeginExpansion
{\displaystyle\sum\limits_{s=1}^{J}}
%EndExpansion%
%TCIMACRO{\dsum \limits_{v=4}^{m\left(  \beta_{g},\beta_{b},\beta_{p}\right)
%}}%
%BeginExpansion
{\displaystyle\sum\limits_{v=4}^{m\left(  \beta_{g},\beta_{b},\beta
_{p}\right)  }}
%EndExpansion
\left(  -1\right)  ^{v-1}\mathbb{G}\left(  s,v\right) \\
&  +%
%TCIMACRO{\dsum \limits_{v=4}^{m\left(  \beta_{g},\beta_{b},\beta_{p}\right)
%}}%
%BeginExpansion
{\displaystyle\sum\limits_{v=4}^{m\left(  \beta_{g},\beta_{b},\beta
_{p}\right)  }}
%EndExpansion
\left(  -1\right)  ^{v-1}\mathbb{Q}_{v}\\
&  =\mathbb{V}_{n,1}\left(  H_{1}\widehat{B}\widehat{P}+H_{2}\widehat{B}%
+H_{3}\widehat{P}+H_{4}\right)  -\mathbb{H}_{2}^{\ast}\\
&  +%
%TCIMACRO{\dsum \limits_{v=3}^{m\left(  \beta_{g},\beta_{b},\beta_{p}\right)
%}}%
%BeginExpansion
{\displaystyle\sum\limits_{v=3}^{m\left(  \beta_{g},\beta_{b},\beta
_{p}\right)  }}
%EndExpansion
\left(  -1\right)  ^{v-1}\mathbb{H}_{v}^{\ast}+%
%TCIMACRO{\dsum \limits_{s=1}^{J}}%
%BeginExpansion
{\displaystyle\sum\limits_{s=1}^{J}}
%EndExpansion%
%TCIMACRO{\dsum \limits_{v=3}^{m\left(  \beta_{g},\beta_{b},\beta_{p}\right)
%}}%
%BeginExpansion
{\displaystyle\sum\limits_{v=3}^{m\left(  \beta_{g},\beta_{b},\beta
_{p}\right)  }}
%EndExpansion
\left(  -1\right)  ^{v-1}\mathbb{G}\left(  s,v\right)  +%
%TCIMACRO{\dsum \limits_{v=3}^{m\left(  \beta_{g},\beta_{b},\beta_{p}\right)
%}}%
%BeginExpansion
{\displaystyle\sum\limits_{v=3}^{m\left(  \beta_{g},\beta_{b},\beta
_{p}\right)  }}
%EndExpansion
\left(  -1\right)  ^{v-1}\mathbb{Q}_{v}%
\end{align*}

Then%
\begin{align*}
&  E\left(  \widehat{\psi}_{\mathcal{K}_{J}}^{eff}\left(  \beta_{g},\beta
_{b},\beta_{p}\right)  \right)  -\psi\left(  \theta\right) \\
&  =O_{p}\left(  \max\left[
\begin{array}
[c]{c}%
k_{-1}^{2\beta/d},\left(  \frac{\log n}{n}\right)  ^{-\frac{\beta_{g}%
}{d+2\beta_{g}}}k_{2s}^{-\beta_{b}/d}k_{2s+1}^{-\beta_{p}/d},\left(
\frac{\log n}{n}\right)  ^{-\frac{\beta_{g}}{d+2\beta_{g}}}k_{2s+1}%
^{-\beta_{b}/d}k_{2s}^{-\beta_{p}/d},\\
\left(  \frac{\log n}{n}\right)  ^{-\frac{2\beta_{g}}{d+2\beta_{g}}}%
k_{0}^{-2\beta/d},\left(  \frac{\log n}{n}\right)  ^{-\frac{\left(
m-1\right)  \beta_{g}}{d+2\beta_{g}}}n^{-\frac{\beta_{b}}{d+2\beta_{b}}%
-\frac{\beta_{p}}{d+2\beta_{p}}},\\
\left(  \frac{\log n}{n}\right)  ^{-\frac{2\beta_{g}}{d+2\beta_{g}}%
}\underset{1\leq s\leq J}{\max}\left(  k_{2s}^{-\beta_{b}/d}k_{0}^{-\beta
_{p}/d},k_{0}^{-\beta_{b}/d}k_{2s}^{-\beta_{p}/d}\right)
\end{array}
\right]  \right) \\
&  =O_{p}\left(  \max\left[
\begin{array}
[c]{c}%
k_{-1}^{2\beta/d},\left(  \frac{\log n}{n}\right)  ^{-\frac{\beta_{g}%
}{d+2\beta_{g}}}k_{2s}^{-\beta_{b}/d}k_{2s+1}^{-\beta_{p}/d},\left(
\frac{\log n}{n}\right)  ^{-\frac{\beta_{g}}{d+2\beta_{g}}}k_{2s+1}%
^{-\beta_{b}/d}k_{2s}^{-\beta_{p}/d},\\
\left(  \frac{\log n}{n}\right)  ^{-\frac{2\beta_{g}}{d+2\beta_{g}}}%
k_{0}^{-2\beta/d},\left(  \frac{\log n}{n}\right)  ^{-\frac{\left(
m-1\right)  \beta_{g}}{d+2\beta_{g}}}n^{-\frac{\beta_{b}}{d+2\beta_{b}}%
-\frac{\beta_{p}}{d+2\beta_{p}}}%
\end{array}
\right]  \right) \\
&  =O_{p}\left(  n^{-\frac{4\beta}{d+4\beta}}\right)
\end{align*}

\end{theorem}

\begin{theorem}
and
\begin{align*}
&  Var\left(  \widehat{\psi}_{\mathcal{K}_{J}}^{eff}\left(  \beta_{g}%
,\beta_{b},\beta_{p}\right)  \right) \\
&  \asymp\frac{k_{-1}}{n^{2}}+%
%TCIMACRO{\tsum \limits_{s=0}^{J}}%
%BeginExpansion
{\textstyle\sum\limits_{s=0}^{J}}
%EndExpansion
\frac{k_{2s}k_{2s-1}}{n^{3}}+\frac{k_{2J+1}^{2}}{n^{3}}\asymp n^{-\frac
{8\beta}{d+4\beta}}%
\end{align*}

\end{theorem}

\textbf{Inference:} Elsewhere, we prove that $\widehat{\psi}_{\mathcal{K}_{J}%
}^{eff}\left(  \beta_{g},\beta_{b},\beta_{p}\right)  $ is asymptotically
normal. Here, to avoid the problem of unknown 'constants' for confidence
interval construction that we discussed in Section \ref{DR_CI_Section}, we
will construct nearly optimal rather than optimal confidence intervals. We
suppose that Eq. $\left(  \ref{NE2}\right)  $ holds with strict equality for
the $\left(  \beta_{g},\beta_{b},\beta_{p}\right)  \ $associated with the
parameter space $\Theta$. Then there exists $\epsilon>0$ such that for all
$0<\sigma<\epsilon,$ $\left(  \beta_{g},\beta_{b}-\sigma,\beta_{p}%
-\sigma\right)  $ satisfies Eq $\left(  \ref{NE2}\right)  $ with strict
equality$,$
\[
sup_{\theta\in\Theta}\left[  \frac{E_{\theta}\left[  \widehat{\psi
}_{\mathcal{K}_{J}}^{eff}\left(  \beta_{g},\beta_{b}-\sigma,\beta_{p}%
-\sigma\right)  |\widehat{\theta}\right]  }{Var_{\theta}\left[  \widehat{\psi
}_{\mathcal{K}_{J}}^{eff}\left(  \beta_{g},\beta_{b}-\sigma,\beta_{p}%
-\sigma\right)  |\widehat{\theta}\right]  }\right]  =o_{p}\left(  1\right)
\]
and
\[
sup_{\theta\in\Theta}\left\{  Var_{\theta}\left[  \widehat{\psi}%
_{\mathcal{K}_{J}}^{eff}\left(  \beta_{g},\beta_{b}-\sigma,\beta_{p}%
-\sigma\right)  |\widehat{\theta}\right]  \right\}  \asymp n^{-\frac{8\left(
\beta-\sigma\right)  }{d+4\left(  \beta-\sigma\right)  }}.
\]
Let $\widehat{\mathbb{W}}\left[  \widehat{\psi}_{\mathcal{K}_{J}}^{eff}\left(
\beta_{g},\beta_{b}\ ,\beta_{p}\ \right)  \right]  $ be a uniformly consistent
estimator of (the properly standardized) $Var_{\theta}\left[  \widehat{\psi
}_{\mathcal{K}_{J}}^{eff}\left(  \beta_{g},\beta_{b}\ ,\beta_{p}\ \right)
|\widehat{\theta}\right]  $ constructed in the same manner as{\Large \ }in
Section \ref{DR_CI_Section}. Then, for all $\sigma<\epsilon,$
\[
\left\{  \widehat{\psi}_{\mathcal{K}_{J}}^{eff}\left(  \beta_{g},\beta
_{b}-\sigma,\beta_{p}-\sigma\right)  -\psi\left(  \theta\right)  \right\}
\left(  \widehat{\mathbb{W}}\left[  \widehat{\psi}_{\mathcal{K}_{J}}%
^{eff}\left(  \beta_{g},\beta_{b}-\sigma,\beta_{p}-\sigma\right)  \right]
\right)  ^{-1}%
\]
converges uniformly in $\theta\in\Theta$ to a $N\left(  0,1\right)  $.
Moreover,
\[
\ \widehat{\psi}_{\mathcal{K}_{J}}^{eff}\left(  \beta_{g},\beta_{b}%
-\sigma,\beta_{p}-\sigma\right)  \pm z_{\alpha}\widehat{\mathbb{W}}\left[
\widehat{\psi}_{\mathcal{K}_{J}}^{eff}\left(  \beta_{g},\beta_{b}-\sigma
,\beta_{p}-\sigma\right)  \right]
\]
is a conservative uniform asymptotic $\left(  1-\alpha\right)  $ confidence
interval for $\psi\left(  \theta\right)  $ with diameter of the order of
$n^{-\frac{4\left(  \beta-\sigma\right)  }{d+4\left(  \beta-\sigma\right)  }%
}.$

\begin{remark}
If Eq.\ref{NE2} holds with an equality and $\mathcal{K}_{J}$ ,$J,$
$\widehat{\overline{\psi}}_{3,\mathcal{K}_{J}}$ are as in the final paragraph
of the preceding subsection then the proof of Theorem \ref{beyond4th} in the
appendix implies $\widehat{\psi}_{\mathcal{K}_{J}}^{eff}\left(  \beta
_{g},\beta_{b},\beta_{p}\right)  -\psi\left(  \theta\right)  =O_{p}\left(
\left(  \log n\right)  n^{-\frac{4\beta}{d+4\beta}}\right)  $
\end{remark}

\section{Adaptive Confidence Intervals for Regression and Treatment Effect
Functions \ with unknown marginal of $X$:\label{adaptive_section}\ }

In this section we describe how to construct adaptive confidence intervals (i)
for a regression function $b\left(  X\right)  =E\left[  Y|X\right]  $ when the
marginal of $X$ is unknown and (ii) for the treatment effect function and
optimal treatment regime in a randomized clinical trial.

\subsection{Regression Functions:}

\textbf{Example 1a continued: }Consider the case $b=p$, $O=\left(  Y,X\right)
\ $with $b\left(  X\right)  =E\left(  Y|X\right)  .$ As usual, we assume for
all $\theta\in\Theta,$ $b\left(  \cdot\right)  $ and the density $g\left(
\cdot\right)  $ of $X$ are contained in known H\"{o}lder balls $H\left(
\beta_{b},C_{b}\right)  $ and $H\left(  \beta_{g},C_{g}\right)  .$ Redefine
$\psi\left(  \theta\right)  \equiv E_{\theta}\left[  \left(  b\left(
X\right)  -\widehat{b}\left(  X\right)  \right)  ^{2}\right]  $ where
$\widehat{b}\left(  \cdot\right)  $ is an adaptive estimate of $b\left(
\cdot\right)  $ from the training sample and expectations and probabilities
remain conditional on the training sample. Adaptivity of $\widehat{b}\left(
\cdot\right)  \ $implies that if $b\left(  \cdot\right)  \in\theta$ is also
contained in a smaller H\"{o}lder ball $H\left(  \beta^{\ast},C\right)  ,$
$\beta^{\ast}>\beta_{b},C<C_{b}$, then $\widehat{b}\left(  \cdot\right)  $
will converge to $b\left(  \cdot\right)  $ under $F\left(  \cdot
,\theta\right)  $ at rate $O_{p}\left(  n^{-\frac{\beta^{\ast}/d}%
{1+2\beta^{\ast}/d}}\right)  .$ Robins and van der Vaart (2006) showed that,
when the marginal density $g\left(  x\right)  $ of $X$ is known, the key to
constructing optimal (rate) adaptive confidence balls for $b\left(  X\right)
$ was to find a rate optimal estimator of $E_{\theta}\left[  \left(  b\left(
X\right)  -\widehat{b}\left(  X\right)  \right)  ^{2}\right]  .$ We shall show
that their approach fails when the marginal of $X$ is unknown, but that a
modification described below succeeds. Specifically, if $b\left(
\cdot\right)  \in\theta$ lies in a smaller H\"{o}lder ball $H\left(
\beta^{\ast},C\right)  ,$ $\beta^{\ast}>\beta_{b},C<C_{b},$ our modification
results in honest asymptotic confidence balls under $F\left(  \cdot
,\theta\right)  ,$ $\theta\in\Theta, $ whose diameter is (essentially) of the
same order $O_{p}\left(  \max\left\{  n^{-\frac{\beta^{\ast}/d}{1+2\beta
^{\ast}/d}},n^{-\frac{2\beta_{b}}{d+4\beta_{b}}}\right\}  \right)  $ as the
diameter of Robins and van der Vaart's optimal adaptive region or ball,
provided either $\left(  i\right)  $ $\beta_{b}/d>1/4$ and $\beta_{g}/d>0$ or
$\left(  ii\right)  $ $\beta_{b}/d\leq1/4$ and eq.$\left(  \ref{NE11}\right)
$ holds with $\beta=\beta_{b}.$ This order is the maximum of the minimax rate
$n^{-\frac{\beta^{\ast}/d}{1+2\beta^{\ast}/d}}$ of convergence of
$\widehat{b}\left(  X\right)  $ to $b\left(  X\right)  $ were $b\left(
X\right)  $ known to lie in $H\left(  \beta^{\ast},C\right)  $ and the square
root of the minimax rate of convergence of an estimator of $E_{\theta}\left[
\left(  b\left(  X\right)  -\widehat{b}\left(  X\right)  \right)  ^{2}\right]
\ $in the larger model $\mathcal{M}\left(  \Theta\right)  $ with $b\left(
\cdot\right)  $ and $g\left(  \cdot\right)  $ only known to lie in $H\left(
\beta_{b},C_{b}\right)  \ $and $H\left(  \beta_{g},C_{g}\right)  .$

The case where $\beta_{b}/d\leq1/4$ and eq.$\left(  \ref{NE11}\right)  $ does
not hold will be considered elsewhere.

Now, since $E_{\theta}\left[  \widehat{b}\left(  X\right)  b\left(  X\right)
\right]  =E_{\theta}\left[  \widehat{b}\left(  X\right)  Y\right]  ,$
\[
\psi\left(  \theta\right)  \equiv E_{\theta}\left[  \left(  b\left(  X\right)
-\widehat{b}\left(  X\right)  \right)  ^{2}\right]  =E_{\theta}\left[
\left\{  b\left(  X\right)  \right\}  ^{2}\right]  -2E_{\theta}\left[
\widehat{b}\left(  X\right)  b\left(  X\right)  \right]  +E_{\theta}\left[
\left\{  \widehat{b}\left(  X\right)  \right\}  ^{2}\right]
\]
has first order influence function $\mathbb{IF}_{1,\psi}\left(  \theta\right)
=\mathbb{V}\left[  H\left(  b,b\right)  -\psi\left(  \theta\right)  \right]
\ $where
\[
H\left(  b,b\right)  =b^{2}\left(  X\right)  +2b\left(  X\right)  \left[
Y-b\left(  X\right)  \right]  -2\widehat{b}\left(  X\right)  Y+\widehat{b}%
^{2}\left(  X\right)  ,
\]
$\ $so $H_{1}=-1,H_{2}=H_{3}=Y,H_{4}=-2\widehat{b}\left(  X\right)
Y+\widehat{b}^{2}\left(  X\right)  .$ Thus $H\left(  b,b\right)  $ for
$E_{\theta}\left[  b\left(  X\right)  ^{2}\right]  $ differs from $H\left(
b,b\right)  $ for $E_{\theta}\left[  \left(  b\left(  X\right)  -\widehat{b}%
\left(  X\right)  \right)  ^{2}\right]  $ only in $H_{4}.$ Since the
truncation bias $\widetilde{\psi}_{k}\left(  \theta\right)  -\psi\left(
\theta\right)  $, higher order influence functions of $\widetilde{\psi}%
_{k}\left(  \theta\right)  $ and estimation bias do not depend on $H_{4}$, it
follows that $TB_{k}\left(  \theta\right)  ,\mathbb{IF}_{jj,\widetilde{\psi
}_{k}}\left(  \theta\right)  $,$\widehat{\mathbb{W}}_{jj,\widetilde{\psi}_{k}%
}^{2},$ and $EB_{m}\left(  \theta\right)  $ are identical for $\psi\left(
\theta\right)  \equiv E_{\theta}\left[  \left(  b\left(  X\right)
-\widehat{b}\left(  X\right)  \right)  ^{2}\right]  $ and $\psi\left(
\theta\right)  \equiv E_{\theta}\left[  b\left(  X\right)  ^{2}\right]  .$ In
contrast, $IF_{1,\psi}\left(  \widehat{\theta}\right)  $ is identically zero
for $\psi\left(  \theta\right)  \equiv E_{\theta}\left[  \left(  b\left(
X\right)  -\widehat{b}\left(  X\right)  \right)  ^{2}\right]  $ but not for
$\psi\left(  \theta\right)  \equiv E_{\theta}\left[  b\left(  X\right)
^{2}\right]  .$ Thus, by Theorem $\left(  \ref{var_if}\right)  ,$\ for
$\psi\left(  \theta\right)  \equiv E_{\theta}\left[  \left(  b\left(
X\right)  -\widehat{b}\left(  X\right)  \right)  ^{2}\right]  ,$%
\ var$_{\theta}\left[  \widehat{\mathbb{\psi}}_{m,\widetilde{\psi}_{k}%
}\right]  \asymp\frac{1}{n}\left(  \frac{k}{n}\right)  ^{m-1}$\ if
$k>n\ $\ and $m>1,$\ and var$_{\theta}\left[  \widehat{\mathbb{\psi}%
}_{m,\widetilde{\psi}_{k}}\right]  =0$\ if $k\leq n$ and $m=1.$ In the case
when $k\leq n\ $\ and $m>1,$ by the Hoeffding decomposition,%
\[
var_{\theta}\left[  \widehat{\mathbb{\psi}}_{m,\widetilde{\psi}_{k}}\right]
=var_{\theta}\left(  \sum_{s=1}^{m}\left(  \mathbb{D}_{s}^{\left(
\widehat{\mathbb{\psi}}_{m,\widetilde{\psi}_{k}}\right)  }\left(
\theta\right)  \right)  \right)
\]
where $\mathbb{D}_{s}^{\left(  \widehat{\mathbb{\psi}}_{m,\widetilde{\psi}%
_{k}}\right)  }$ is a $s$th order degenerate U-statistic. Further by Theorem
$\left(  \ref{var_if}\right)  ,$ we have%
\[
var_{\theta}\left[  \widehat{\mathbb{\psi}}_{m,\widetilde{\psi}_{k}}\right]
\asymp\max\left(  var_{\theta}\left(  \mathbb{D}_{1}^{\left(
\widehat{\mathbb{\psi}}_{m,\widetilde{\psi}_{k}}\right)  }\right)
,var_{\theta}\left(  \mathbb{D}_{2}^{\left(  \widehat{\mathbb{\psi}%
}_{m,\widetilde{\psi}_{k}}\right)  }\right)  \right)
\]
as $var_{\theta}\left(  \mathbb{D}_{s}^{\left(  \widehat{\mathbb{\psi}%
}_{m,\widetilde{\psi}_{k}}\right)  }\right)  \asymp\frac{1}{n}\left(  \frac
{k}{n}\right)  ^{s-1}=o\left(  \frac{k}{n^{2}}\right)  $ for any $s>2.$
Moreover,
\[
var_{\theta}\left(  \mathbb{D}_{1}^{\left(  \widehat{\mathbb{\psi}%
}_{m,\widetilde{\psi}_{k}}\right)  }\right)  \asymp\frac{\left\vert \left\vert
b\left(  X\right)  -\widehat{b}\left(  X\right)  \right\vert \right\vert
_{2}^{2}}{n}%
\]
since the kernel of $\mathbb{D}_{1}^{\left(  \widehat{\mathbb{\psi}%
}_{m,\widetilde{\psi}_{k}}\right)  }$ is of order $O_{p}\left(  \left\vert
\left\vert b\left(  X\right)  -\widehat{b}\left(  X\right)  \right\vert
\right\vert _{2}\right)  .$ In summary
\begin{align*}
var_{\theta}\left[  \widehat{\mathbb{\psi}}_{m,\widetilde{\psi}_{k}}\right]
&  \asymp\max\left(  \frac{\left\vert \left\vert b\left(  X\right)
-\widehat{b}\left(  X\right)  \right\vert \right\vert _{2}^{2}}{n},\frac
{k}{n^{2}}\right) \\
&  =\max\left(  n^{-\frac{2\beta_{b}/d}{1+2\beta_{b}/d}-1},\frac{k}{n^{2}%
}\right)
\end{align*}
\ if $k\leq n\ $\ and $m>1$ (In contrast, for $\psi\left(  \theta\right)
\equiv E_{\theta}\left[  b\left(  X\right)  ^{2}\right]  ,$\ var$_{\theta
}\left[  \widehat{\mathbb{\psi}}_{m,\widetilde{\psi}_{k}}\right]  \asymp
\max\left(  \frac{1}{n},\frac{k}{n^{2}}\right)  =\frac{1}{n}$\ if $k\leq n$).
Thus if $\beta_{b}/d>1/4,$ $\left(  i\right)  $ $\widehat{\psi}_{m_{opt}%
,k_{opt}\left(  m_{opt}\right)  }$ has $k_{opt}\left(  m_{opt}\right)  $ of
$O\left(  n^{\frac{2}{1+4\beta_{b}/d}}\right)  ,$ where $n^{\frac{2}%
{1+4\beta/d}}<n$ comes from equating the order $k^{-4\beta_{b}/d}$ of
$TB_{k}^{2}\left(  \theta\right)  $ to the order $k/n^{2}=n^{-\frac{8\beta
_{b}/d}{1+4\beta_{b}/d}}<<n^{-1}$ of the variance ($n^{-\frac{2\beta_{b}%
/d}{1+2\beta_{b}/d}-1}\ll n^{-\frac{8\beta_{b}/d}{1+4\beta_{b}/d}}$ for
$\forall$ $\beta_{b}>0$) and $\left(  ii\right)  \ m_{opt}$ is the smallest
integer $m$ such that the order $n^{-\left(  \frac{\left(  m-1\right)
\beta_{g}}{2\beta_{g}+d}+\frac{2\beta_{b}}{d+2\beta_{b}}\right)  }$ of
$EB_{m}=O_{p}\left(  n^{-\left(  \frac{\left(  m-1\right)  \beta_{g}}%
{2\beta_{g}+d}+\frac{2\beta_{b}}{d+2\beta_{b}}\right)  }\right)  \ $is less
than the order $n^{-\frac{4\beta_{b}/d}{1+4\beta_{b}/d}}$ of the standard
error$.$ It follows that , for $\beta_{b}/d>1/4,$ in contrast to $\psi\left(
\theta\right)  \equiv E_{\theta}\left[  b\left(  X\right)  ^{2}\right]  ,$ we
can estimate $\psi\left(  \theta\right)  \equiv E_{\theta}\left[  \left(
b\left(  X\right)  -\widehat{b}\left(  X\right)  \right)  ^{2}\right]  $ at
(the minimax) rate $n^{-\frac{4\beta_{b}/d}{1+4\beta_{b}/d}} $ which is faster
(i.e., less ) than the usual parametric rate of $n^{-1/2}.$

When $\beta_{b}/d<1/4\ $, the minimax rates for $\psi\left(  \theta\right)
\equiv E_{\theta}\left[  \left(  b\left(  X\right)  -\widehat{b}\left(
X\right)  \right)  ^{2}\right]  $ and $\psi\left(  \theta\right)  \equiv
E_{\theta}\left[  b\left(  X\right)  ^{2}\right]  $ are identical and, when
eq.$\left(  \ref{NE11}\right)  $ holds, it follows from Theorem
\ref{beyond4th}{\LARGE \ }that $\widehat{\psi}_{\mathcal{K}_{J}}^{eff}\left(
\beta_{g},\beta_{b},\beta_{b}\right)  $ achieves the minimax rate of
$n^{-\frac{4\beta_{b}/d}{1+4\beta_{b}/d}}\geq n^{-1/2}$.

Henceforth assume either $\left(  i\right)  \beta_{b}/d>1/4$ or $\left(
ii\right)  \beta_{b}/d<1/4$ and eq.$\left(  \ref{NE11}\right)  $ holds. Pick
an $\epsilon$ so that eq.$\left(  \ref{NE11}\right)  $ holds for $\left(
\beta_{g},\beta_{b}-\epsilon,\beta_{p}-\epsilon\right)  .$ Let $0<\sigma
<\epsilon$ and define $\widehat{\psi}_{\ }^{\ast}\equiv\widehat{\psi}\left(
\sigma\right)  =\widehat{\psi}_{m_{opt},\left\{  k_{opt}\left(  m_{opt}%
\right)  \right\}  ^{1+\sigma}}$ and $\widehat{\mathbb{W}}^{\ast}%
\equiv\widehat{\mathbb{W}}^{\ast}\left(  \sigma\right)  =\widehat{\mathbb{W}%
}_{m_{opt},\widetilde{\psi}_{\left\{  k_{opt}\left(  m_{opt}\right)  \right\}
^{1+\sigma}}} $if $\beta_{b}/d>1/4$ and $\widehat{\psi}_{\ }^{\ast
}=\widehat{\psi}_{\mathcal{K}_{J}}^{eff}\left(  \beta_{g},\beta_{b}%
-\sigma,\beta_{p}-\sigma\right)  $ and $\widehat{\mathbb{W}}^{\ast
}=\widehat{\mathbb{W}}\left[  \widehat{\psi}_{\mathcal{K}_{J}}^{eff}\left(
\beta_{g},\beta_{b}-\sigma,\beta_{p}-\sigma\right)  \right]  \ $if $\beta
_{b}/d<1/4.$ Note $\widehat{\mathbb{W}}^{\ast}$ is $O_{p}\left(
n^{-\frac{4\left(  \beta_{b}-\sigma\right)  }{d+4\left(  \beta_{b}%
-\sigma\right)  }}\right)  $ uniformly over $\Theta,$ where $\Theta$ is the
parameter space with smoothness parameters $\left(  \beta_{g},\beta
_{b}\right)  .$ Then, by eq.$\left(  \ref{NE11}\right)  $ and results in
Section \ref{case2}, as $n\rightarrow\infty,$ $\inf_{\theta\in\Theta}%
$pr$_{\theta}\left[  \left\{  \widehat{\psi}_{\ }^{\ast}-\psi\left(
\theta\right)  \right\}  \geq-z_{\alpha}\widehat{\mathbb{W}}^{\ast}\right]
\geq1-\alpha.$\ Thus, if $\psi\left(  \theta\right)  $ were a function of
$\theta$ only through $b\left(  \cdot\right)  \ $so $\psi\left(
\theta\right)  =\psi\left(  b\right)  $, the set
\begin{equation}
\left\{  b^{\ast}\left(  \cdot\right)  ;\psi\left(  \theta\right)
\leq\widehat{\psi}_{\ }^{\ast}+z_{\alpha}\widehat{\mathbb{W}}^{\ast}\right\}
\label{CIa}%
\end{equation}
would be an uniform asymptotic $\left(  1-\alpha\right)  $ confidence region
for $b\left(  \cdot\right)  .$ However, for $\psi\left(  \theta\right)
=E_{\theta}\left[  \left(  b\left(  X\right)  -\widehat{b}\left(  X\right)
\right)  ^{2}\right]  ,$ this approach fails because $\psi\left(
\theta\right)  $ also depends on $\theta$ through the unknown density
$g\left(  x\right)  $ of $X.$ This approach succeeded in Robins and van der
Vaart (2006) because $g\left(  x\right)  $ was assumed known.

We consider two solutions. The first gives (near) optimal adaptive honest
intervals. The second would give honest, but non-optimal, intervals. The first
solution is to replace $\psi\left(  \theta\right)  $ with its empirical mean
$\psi_{emp}\left(  b\right)  \equiv\mathbb{V}\left[  \left\{  b\left(
X\right)  -\widehat{b}\left(  X\right)  \right\}  ^{2}\right]  \ $in
eq.(\ref{CIa})$.$
\[
\psi_{emp}\left(  b\right)  -\psi\left(  \theta\right)  =O_{p}\left(  \left[
\left\{  b\left(  X\right)  -\widehat{b}\left(  X\right)  \right\}
^{2}\right]  n^{-1/2}\right)  =O_{p}\left(  n^{-\left(  \frac{2\beta_{b}%
}{d+2\beta_{b}}+\frac{1}{2}\right)  }\right)
\]
uniformly in $\theta\in\Theta.$ It is straightforward to check that for all
$\beta_{b}>0,$ $n^{-\left(  \frac{2\beta_{b}}{d+2\beta_{b}}+\frac{1}%
{2}\right)  }<<n^{-\frac{4\beta_{b}/d}{1+4\beta_{b}/d}}.$ Thus, for
$\sigma<\epsilon,$ $\left\{  \widehat{\psi}_{\ }^{\ast}-\psi_{emp}\left(
b\right)  \right\}  /\left\{  \widehat{\psi}_{\ }^{\ast}-\psi\left(
\theta\right)  \right\}  =1+o_{p}\left(  1\right)  $ uniformly over $\theta
\in\Theta,$ so $\inf_{\theta\in\Theta}$pr$_{\theta}\left[  \left\{
\widehat{\psi}_{\ }^{\ast}-\psi_{emp}\left(  b\right)  \right\}
\geq-z_{\alpha}\widehat{\mathbb{W}}^{\ast}\right]  \geq1-\alpha$ and
\begin{equation}
\left\{  b^{\ast}\left(  \cdot\right)  ;\mathbb{V}\left[  \left\{  b^{\ast
}\left(  X\right)  -\widehat{b}\left(  X\right)  \right\}  ^{2}\right]
\leq\widehat{\psi}_{\ }^{\ast}+z_{\alpha}\widehat{\mathbb{W}}^{\ast}\right\}
\label{CItrue}%
\end{equation}
is a uniform asymptotic $\left(  1-\alpha\right)  $ confidence region for
$b\left(  \cdot\right)  .$ Moreover, if $b\left(  \cdot\right)  \in\theta$
lies in a smaller H\"{o}lder ball $H\left(  \beta^{\ast},C\right)  ,$
$\beta^{\ast}>\beta_{b},C<C_{b},$ then, under $F\left(  \cdot,\theta\right)
,$ the diameter
\begin{align*}
\left\{  \widehat{\psi}_{\ }^{\ast}+z_{\alpha}\widehat{\mathbb{W}}^{\ast
}\right\}  ^{1/2}  &  =\left\{  \psi\left(  \theta\right)  +O_{p}\left(
n^{-\frac{4\left(  \beta_{b}-\sigma\right)  }{d+4\left(  \beta_{b}%
-\sigma\right)  }}\right)  \right\}  ^{1/2}\\
&  =O_{p}\left(  \max\left\{  n^{-\frac{2\beta^{\ast}/d}{1+2\beta^{\ast}/d}%
},n^{-\frac{4\left(  \beta_{b}-\sigma\right)  }{d+4\left(  \beta_{b}%
-\sigma\right)  }}\right\}  \right)  ^{1/2}\\
&  =O_{p}\left(  \max\left\{  n^{-\frac{\ \beta^{\ast}/d}{1+2\beta^{\ast}/d}%
},n^{-\frac{2\left(  \beta_{b}-\sigma\right)  }{d+4\left(  \beta_{b}%
-\sigma\right)  }}\right\}  \right)
\end{align*}
since $\psi\left(  \theta\right)  =O_{p}\left(  n^{-\frac{2\beta^{\ast}%
/d}{1+2\beta^{\ast}/d}}\right)  $ and $\widehat{\psi}_{\ }^{\ast}-\psi\left(
\theta\right)  $ and $\widehat{\mathbb{W}}^{\ast}$ are $O_{p}\left(
n^{-\frac{4\left(  \beta_{b}-\sigma\right)  }{d+4\left(  \beta_{b}%
-\sigma\right)  }}\right)  .$

The second, non-optimal, solution would be to replace the functional
$\psi\left(  \theta\right)  \equiv E_{\theta}\left[  \left(  b\left(
X\right)  -\widehat{b}\left(  X\right)  \right)  ^{2}\right]  $ with
$\psi\left(  b\right)  =\int\left\{  b\left(  x\right)  -\widehat{b}\left(
x\right)  \right\}  ^{2}dx.$ The functional $\psi\left(  b\right)  $ is the
first functional we have considered that is not in our doubly robust class of
functionals. Arguing as above, if we can construct an asymptotically normal
higher order $U-$ statistic estimator $\widehat{\psi}_{\ }^{\ast}$ that
converges to $\psi\left(  b\right)  $ at rate $n^{-\omega}$ on $\mathcal{M}%
\left(  \Theta\right)  $ and a consistent estimator $\widehat{\mathbb{W}%
}^{\ast}\ $of its standard error, then $\left\{  b^{\ast}\left(  \cdot\right)
;\int\left\{  b\left(  x\right)  -\widehat{b}\left(  x\right)  \right\}
^{2}dx\leq\widehat{\psi}_{\ }^{\ast}+z_{\alpha}\widehat{\mathbb{W}}^{\ast
}\right\}  $ would be an honest adaptive confidence interval of diameter
$O_{p}\left(  \max\left\{  n^{-\frac{\beta^{\ast}/d}{1+2\beta^{\ast}/d}%
},n^{-\omega/2}\right\}  \right)  .$ We conjecture, based on arguments given
elsewhere, that the minimax rate for estimation of $\psi\left(  b\right)
=\int\left\{  b\left(  x\right)  -\widehat{b}\left(  x\right)  \right\}
^{2}dx$ exceeds $O_{p}\left[  n^{-\frac{4\beta_{b}}{d+4\beta_{b}}}\right]  $
whenever $\frac{\beta_{g}/d}{2\beta_{g}/d+1}<\frac{\beta/d}{\left(
1+4\beta/d\right)  \left(  1+2\beta/d\right)  }$. Since $\frac{\beta
/d}{\left(  1+4\beta/d\right)  \left(  1+2\beta/d\right)  }>\frac{1-4\beta
/d}{1+4\beta/d}\beta/d\ $for all $\beta>0,$ it follows that, when the marginal
of $X$ is unknown and $\frac{\beta/d}{\left(  1+4\beta/d\right)  \left(
1+2\beta/d\right)  }>\frac{\beta_{g}/d}{2\beta_{g}/d+1}>\frac{1-4\beta
/d}{1+4\beta/d}\beta/d$, intervals based on $\mathbb{V}\left[  \left\{
b^{\ast}\left(  X\right)  -\widehat{b}\left(  X\right)  \right\}  ^{2}\right]
$ will, but intervals based on $\int\left\{  b\left(  x\right)  -\widehat{b}%
\left(  x\right)  \right\}  ^{2}dx $ will not, have diameter of the same order
as the optimal interval with the marginal of $X$ known.

\subsection{Treatment Effect Functions in a Randomized Trial}

\textbf{Example 4 continued: } Consider the case $b=p$, $Y=Y^{\ast}$ wp1 so we
have data $O=\left\{  Y,A,X\right\}  $, where $A$ is a binary treatment, $Y$
is the response, and $X$ is a vector of prerandomization covariates. The
randomization probabilities $\pi_{0}\left(  X\right)  =P\left(  A=1|X\right)
$ are known by design and $b\left(  x\right)  =E_{\theta}(Y|A=1,X=x)-E_{\theta
}(Y|A=0,X=x\ )$ is the average treatment effects function. For $\theta
\in\Theta,$ $b\left(  \cdot\right)  $ and the density $g\left(  \cdot\right)
$ of $X$ are contained in known H\"{o}lder balls $H\left(  \beta_{b}%
,C_{b}\right)  $ and $H\left(  \beta_{g},C_{g}\right)  .$ Suppose we have an
adaptive estimator $\widehat{b}\left(  \cdot\right)  $ of $b\left(
\cdot\right)  \ $based on the training sample constructed as described below.
Now, since $E_{\theta}\left[  \widehat{b}\left(  X\right)  b\left(  X\right)
\right]  =E_{\theta}\left[  \widehat{b}\left(  X\right)  Y|A=1\right]
-E_{\theta}\left[  \widehat{b}\left(  X\right)  Y|A=0\right]  $ has influence
function $\frac{A}{\pi_{0}\left(  X\right)  }Y\widehat{b}\left(  X\right)
-\frac{1-A}{1-\pi_{0}\left(  X\right)  }Y\widehat{b}\left(  X\right)
-E_{\theta}\left[  \widehat{b}\left(  X\right)  b\left(  X\right)  \right]
=\left(  A-\pi_{0}\left(  X\right)  \right)  \sigma_{0}^{-2}\left(  X\right)
Y\widehat{b}\left(  X\right)  -E_{\theta}\left[  \widehat{b}\left(  X\right)
b\left(  X\right)  \right]  ,$ where $\sigma_{0}^{2}\left(  X\right)  =\pi
_{0}\left(  X\right)  \left\{  1-\pi_{0}\left(  X\right)  \right\}  ,$
$\psi\left(  \theta\right)  \equiv E_{\theta}\left[  \left(  b\left(
X\right)  -\widehat{b}\left(  X\right)  \right)  ^{2}\right]  $ has first
order influence functions, indexed by arbitrary functions $c\left(  x\right)
,$ $\mathbb{IF}_{1,\psi}\left(  \theta,c\right)  \equiv\mathbb{IF}_{1,\psi
}\left(  \theta\right)  =\mathbb{V}\left[  H\left(  b,b\right)  -\psi\left(
\theta\right)  \right]  \ $with
\begin{align*}
H_{1}  &  =1-2A\left\{  A-\pi_{0}\left(  X\right)  \right\}  \sigma_{0}%
^{-2}\left(  X\right)  ,\\
H_{2}  &  =H_{3}=\left\{  A-\pi_{0}\left(  X\right)  \right\}  \sigma_{0}%
^{-2}\left(  X\right)  Y,\\
H_{4}  &  =\left\{  A-\pi_{0}\left(  X\right)  \right\}  c\left(  X\right)
-2\left(  A-\pi_{0}\left(  X\right)  \right)  \sigma_{0}^{-2}\left(  X\right)
Y\widehat{b}\left(  X\right)  +\widehat{b}^{2}\left(  X\right)
\end{align*}
Thus $H\left(  b,b\right)  $ for $E_{\theta}\left[  \left(  b\left(  X\right)
-\widehat{b}\left(  X\right)  \right)  ^{2}\right]  $ differs from $H\left(
b,b\right)  $ for $\psi\left(  \theta\right)  \equiv E_{\theta}\left[
b\left(  X\right)  ^{2}\right]  $ only in $H_{4}.$ It follows that all the
properties of the confidence ball \ref{CItrue} for $b\left(  \cdot\right)
=E_{\theta}(Y|\ X=\cdot)$ in the setting of the last subsection remain true
for $b\left(  \cdot\right)  =E_{\theta}(Y|A=1,X=\cdot)-E_{\theta}\left(
Y|A=0,X=\cdot\right)  $ in the setting of this subsection.

Now define $d_{b^{\ast}}\left(  x\right)  =I\left[  b^{\ast}\left(  x\right)
>0\right]  .$ Then it then follows that an honest $1-\alpha$ uniform
asymptotic confidence set for the optimal treatment regime $d_{opt}\left(
\cdot\right)  =I\left[  b\left(  \cdot\right)  >0\right]  $ is given by
$\left\{  d_{b^{\ast}}\left(  \cdot\right)  ;\mathbb{V}\left[  \left\{
b^{\ast}\left(  X\right)  -\widehat{b}\left(  X\right)  \right\}  ^{2}\right]
\leq\widehat{\psi}_{\ }^{\ast}+z_{\alpha}\widehat{\mathbb{W}}^{\ast}\right\}
.$

\textbf{Adaptive Estimator of The Treatment Effect Function: }One among many
approaches to constructing a rate-adaptive estimator of $b\left(
\cdot\right)  $ is as follows. Split the training sample into two random
subsamples - a candidate estimator subsample of size $n_{c}$ and a validation
subsample of size $n_{v},$ where both $n_{c}/n$ and $n_{v}/n$ are bounded away
from $0$ as $n\rightarrow\infty.$ Noting that $0=\mathbb{E}_{\theta}\left[
\left\{  Y-Ab\left(  X\right)  \right\}  q\left(  X\right)  \left\{  A-\pi
_{0}\left(  X\right)  \right\}  \right]  $ for all $q\left(  \cdot\right)  ,$
we construct candidate estimators of $b\left(  \cdot\right)  $ as follows. For
$s=1,2,...,$ $n-1,$ let $\widehat{\overline{\varkappa}}_{s}$ be the solution,
if any, to the $s$ equations
\[
0=\mathbb{P}_{c}\left[  \left\{  Y-A\overline{\varkappa}_{s}^{T}%
\overline{\varphi}_{s}\left(  X\right)  \right\}  \overline{\varphi}%
_{s}\left(  X\right)  \left\{  A-\pi_{0}\left(  X\right)  \right\}  \right]
\]
where $\varphi_{1}\left(  X\right)  ,\varphi_{2}\left(  X\right)  ,...$ is a
complete basis$\ $wrt to Lebesgue measure in $R^{d}$ that provides optimal
rate approximation for H\"{o}lder balls and $\mathbb{P}_{c}$ is the empirical
measure for the candidate estimator subsample. Our candidates for $b\left(
X\right)  $ are the $\widehat{b}^{\left(  s\right)  }\left(  X\right)
=\overline{\varphi}_{s}\left(  X\right)  ^{T}\widehat{\overline{\varkappa}%
}_{s}. $ Robins (2004) proved that $b\left(  \cdot\right)  $ is the unique
function $b^{\ast}\left(  \cdot\right)  $\textbf{\ }minimizing $Risk\left(
b^{\ast}\right)  \equiv E_{\theta}\left[  \sigma_{0}^{-2}\left(  X\right)
\left\{  Y-\left[  A-\pi_{0}\left(  X\right)  \right]  b^{\ast}\left(
X\right)  \right\}  ^{2}\right]  .$ In fact, the candidate $\widehat{b}%
^{\left(  s\right)  }\left(  X\right)  $ in our set for which $Risk\left(
\widehat{b}^{\left(  s\right)  }\right)  \ $is smallest is also the candidate
that minimizes $E\left[  \left(  b\left(  X\right)  -\widehat{b}^{\left(
s\right)  }\left(  X\right)  \right)  ^{2}\ \right]  $ since $Risk\left(
\widehat{b}^{\left(  s\right)  }\right)  -Risk\left(  b\right)  =E\left[
\left(  b\left(  X\right)  -\widehat{b}^{\left(  s\right)  }\left(  X\right)
\right)  ^{2}\right]  . $ Specifically,
\begin{align*}
&  E\left[
\begin{array}
[c]{c}%
\sigma_{0}^{-2}\left(  X\right)  \left\{  Y-\left[  A-\pi_{0}\left(  X\right)
\right]  \widehat{b}^{\left(  s\right)  }\left(  X\right)  \right\}  ^{2}\\
-\sigma_{0}^{-2}\left(  X\right)  \left\{  Y-\left[  A-\pi_{0}\left(
X\right)  \right]  b\left(  X\right)  \right\}  ^{2}%
\end{array}
\right] \\
&  =E\left[
\begin{array}
[c]{c}%
\sigma_{0}^{-2}\left(  X\right)  \left(  A-\pi_{0}\left(  X\right)  \right)
\left(  b\left(  X\right)  -\widehat{b}^{\left(  s\right)  }\left(  X\right)
\right)  \times\\
\left(  2\left(  Ab\left(  X\right)  -E\left(  Y|A=0,X\right)  \right)
-\left(  A-\pi_{0}\left(  X\right)  \right)  \left(  b\left(  X\right)
+\widehat{b}^{\left(  s\right)  }\left(  X\right)  \right)  \right)
\end{array}
\right] \\
&  =E\left(  \sigma_{0}^{-2}\left(  X\right)  \left(  A-\pi_{0}\left(
X\right)  \right)  A\left(  b\left(  X\right)  -\widehat{b}^{\left(  s\right)
}\left(  X\right)  \right)  ^{2}\right) \\
&  =E\left[  \left(  b\left(  X\right)  -\widehat{b}^{\left(  s\right)
}\left(  X\right)  \right)  ^{2}\right]
\end{align*}

We use these results to select among our candidates by cross-validation. Let
$\widehat{b}\left(  \cdot\right)  $ be the $\widehat{b}^{\left(  s\right)
}\left(  \cdot\right)  $ minimizing $\mathbb{P}_{v}\left[  \sigma_{0}%
^{-2}\left(  X\right)  \left\{  Y-\left[  A-\pi_{0}\left(  X\right)  \right]
\widehat{b}^{\left(  s\right)  }\left(  X\right)  \right\}  ^{2}\right]  $
over $s=1,2,...,n-1 $, where $\mathbb{P}_{v}$ is the validation subsample
empirical measure. If $b\left(  \cdot\right)  $ were known to lie in a
H\"{o}lder ball $H\left(  \beta,C\right)  ,$ it is easy to check that the
candidate $\widehat{b}^{\left(  s\right)  }\left(  \cdot\right)  $ with
$s=\lfloor n^{\frac{1}{2\beta+1}}\rfloor$ obtains the optimal rate of
$n^{\frac{-\beta}{2\beta+1}}$ for estimating $E\left[  \left(  b\left(
X\right)  -\widehat{b}^{\left(  s\right)  }\left(  X\right)  \right)
^{2}\right]  .$ Since the number of candidates at sample size $n$ is less than
$n,$ it then follows at once from van der Laan and Dudoit's (2003) results on
model selection by cross validation that $\widehat{b}\left(  \cdot\right)  $
is adaptive over H\"{o}lder balls.

\section{Testing, Confidence Sets, and Implicitly Defined
Functionals:\label{testing_section}}

In Example $1c$ of section 3.1$,$ we considered the following problem. We were
given a functional $\psi\left(  \tau,\theta\right)  $ indexed by a real number
$\tau$ and the parameter $\theta\in\Theta.$ The implicitly defined-functional
$\tau\left(  \theta\right)  $ was the assumed unique solution to
$0=\psi\left(  \tau,\theta\right)  .$ We noted that a $(1-\alpha)\ $confidence
set for $\tau\left(  \theta\right)  $ is the set of $\tau\ $such that a
$(1-\alpha)\ $CI interval for $\psi\left(  \tau,\theta\right)  $ contains $0.$
In the following subsection we derive the width of the confidence set for
$\tau\left(  \theta\right)  .$ We then generalize the problem in the second
subsection by introducing the notions of the testing tangent space, a testing
influence function, and the higher order efficient testing score. In the final
subsection, we show how the two earlier subsections are related.

\subsection{Confidence Intervals for Implicitly Defined Functionals:\ }

To derive the order of the length of the confidence interval for the parameter
$\tau\left(  \theta\right)  $\ in Example 1c, we can use the next theorem as
follows. \ Assume eq (\ref{NE11})\ holds and $\beta\leq1/4$. Then we can take
the estimator $\widetilde{\psi}\left(  \tau\right)  $\ and rate $n^{-\gamma}$
in the theorem to be the estimator $\widehat{\psi}_{\mathcal{K}_{J}}^{eff}%
$\ and rate $n^{-\frac{4\beta}{4\beta+1}+\sigma}$\ for a very small positive
$\sigma$\ and conclude that the length of the confidence interval for
$\tau\left(  \theta\right)  $\ in Example 1c to be $O_{p}\left(
n^{-\frac{4\beta}{4\beta+1}+\sigma}\right)  .$

\begin{theorem}
\label{J1}\textbf{: }Suppose for an estimator $\widehat{\psi}\left(
\tau\right)  $ and functional $\psi\left(  \tau,\theta\right)  ,$ there is a
scale estimator $\widehat{\mathbb{W}}\left(  \tau\right)  $ such that
$n^{\gamma}\widehat{\mathbb{W}}\left(  \tau\right)  \rightarrow w\left(
\tau,\theta\right)  $ in $\theta-$probability $,w\left(  \tau,\theta\right)
>c^{\ast}>0$ and $\left(  \widehat{\psi}\left(  \tau\right)  -\psi\left(
\tau,\theta\right)  \right)  /\widehat{\mathbb{W}}\left(  \tau\right)  $
converges in law to $N\left(  0,1\right)  \ $uniformly for $\theta\in\Theta$,
$\tau\in\left\{  \tau\left(  \theta\right)  ;\theta\in\Theta\right\}  $. Then,
(i) with $z_{\alpha}$ the $\alpha-$quantile and $\Phi\left(  \cdot\right)  $
the CDF of a $N\left(  0,1\right)  ,$ the confidence set $\mathcal{C}%
_{n}=\left\{  \tau;-z_{1-\alpha/2}<\frac{\widehat{\psi}\left(  \tau\right)
}{\ \widehat{\mathbb{W}}\left(  \tau\right)  }<z_{1-\alpha/2}\right\}  $ is a
uniform asymptotic $1-\alpha$ confidence set for the (assumed) unique solution
$\tau\left(  \theta\right)  $ to $\psi\left(  \tau,\theta\right)  =0;$
$\left(  ii\right)  $ the probability under $\theta$ that a sequence
$\tau=\tau_{n}\ $satisfying $\psi\left(  \tau_{n},\theta\right)
=a_{n}n^{-\rho},a_{n}\rightarrow a\neq0$ is contained in $\mathcal{C}_{n}$
converges to $1$ when $\rho>\gamma,$ is $o\left(  1\right)  $ when
$\rho<\gamma,$ and converges to $\Phi\left(  z_{1-\alpha/2}-\frac{a}{w\left(
\tau\left(  \theta\right)  ,\theta\right)  \ }\right)  -\Phi\left(
-z_{1-\alpha/2}-\frac{a}{w\left(  \tau\left(  \theta\right)  ,\theta\right)
\ }\right)  \ $when $\rho=\gamma.$ $\left(  iii\right)  $ If $\psi\left(
\tau,\theta\right)  $ is uniformly twice continuously differentiable in $\tau$
and\emph{\ }$0<\sigma<\left\vert \psi_{\tau}\left(  \tau\left(  \theta\right)
,\theta\right)  \right\vert <c$\emph{\ }and\emph{\ }$\left\vert \psi_{\tau
^{2}}\left(  \tau\left(  \theta\right)  ,\theta\right)  \right\vert
<c$\emph{\ }for constants\emph{\ }$\left(  \sigma,c\right)  ,$ then $\left(
ii\right)  $ holds for a sequence $\tau=\tau_{n}$ satisfying $\tau_{n}%
-\tau\left(  \theta\right)  =\left\{  \psi_{\tau}\left(  \tau\left(
\theta\right)  ,\theta\right)  \right\}  ^{-1}a_{n}n^{-\rho},a_{n}\rightarrow
a\neq0,\rho>0.$
\end{theorem}

\begin{proof}
(i): That $\mathcal{C}_{n}$ is a uniform asymptotic $1-\alpha$ confidence set
is immediate. (ii): Now
\begin{align*}
&  Pr_{\theta}\left\{  z_{1-\alpha/2}>\frac{\widehat{\psi}\left(  \tau
_{n}\right)  }{\widehat{\mathbb{W}}\left(  \tau_{n}\right)  \ }>-z_{1-\alpha
/2}\right\} \\
&  =Pr_{\theta}\left\{  z_{1-\alpha/2}-\frac{\psi\left(  \tau_{n}%
,\theta\right)  }{\widehat{\mathbb{W}}\left(  \tau_{n}\right)  \ }%
>\frac{\widehat{\psi}\left(  \tau_{n}\right)  -\psi\left(  \tau_{n}%
,\theta\right)  }{\widehat{\mathbb{W}}\left(  \tau_{n}\right)  \ }%
>-z_{1-\alpha/2}-\frac{\psi\left(  \tau_{n},\theta\right)  }%
{\widehat{\mathbb{W}}\left(  \tau_{n}\right)  \ }\right\} \\
&  \underset{n\rightarrow\infty}{\rightarrow}\Phi\left(  z_{1-\alpha/2}%
-\lim_{n\rightarrow\infty}\frac{n^{\gamma\ }\psi\left(  \tau_{n}%
,\theta\right)  }{n^{\gamma\ }\widehat{\mathbb{W}}\left(  \tau_{n}\right)
\ }\right)  -\Phi\left(  -z_{1-\alpha/2}-\lim_{n\rightarrow\infty}%
\frac{n^{\gamma\ }\psi\left(  \tau_{n},\theta\right)  }{n^{\gamma
\ }\widehat{\mathbb{W}}\left(  \tau_{n}\right)  \ }\right) \\
&  =\Phi\left(  z_{1-\alpha/2}-\frac{a\lim_{n\rightarrow\infty}n^{\gamma-\rho
}}{w\left(  \tau\left(  \theta\right)  ,\theta\right)  \ }\right)
-\Phi\left(  -z_{1-\alpha/2}-\frac{a\lim_{n\rightarrow\infty}n^{\gamma-\rho}%
}{w\left(  \tau\left(  \theta\right)  ,\theta\right)  \ }\right)  .
\end{align*}
(iii): Since $\psi\left(  \tau_{n},\theta\right)  =\psi_{\tau}\left(
\tau\left(  \theta\right)  ,\theta\right)  \left(  \tau_{n}-\tau\left(
\theta\right)  \right)  +\frac{1}{2}\psi_{\tau^{2}}\left(  \tau^{\ast}\left(
\theta\right)  ,\theta\right)  \left(  \tau_{n}-\tau\left(  \theta\right)
\right)  ^{2}$ for some $\tau^{\ast}\left(  \theta\right)  $ between
$\tau\left(  \theta\right)  $ and $\tau,$ we have that $\psi\left(  \tau
_{n},\theta\right)  =a_{n}n^{-\rho}+o_{p}\left(  a_{n}n^{-\rho}\right)
=a_{n}\left(  1+o_{p}\left(  1\right)  \right)  n^{-\rho}$ satisfies the
assumption in (ii).
\end{proof}

\textbf{Remark:}\ Under some further regularity conditions, the solution
$\widetilde{\tau}$ to $0=\widetilde{\psi}\left(  \tau\right)  $ is
asymptotically normal with mean $\tau\left(  \theta\right)  $ and variance
$\psi_{\tau}^{-2}\left(  \tau\ ,\theta\right)  \left[  \left\{  w\left(
\tau\left(  \theta\right)  ,\theta\right)  \right\}  ^{2}\right]  $ uniformly
over $\theta\in\Theta$,$\tau\in\left\{  \tau\left(  \theta\right)  ;\theta
\in\Theta\right\}  .$

\subsection{Testing influence functions and a higher order efficient score}

In the following, we repeatedly use definitions from Sec. 2, which might
usefully be reviewed at this point.

\begin{definition}
\textbf{\ m}$^{\mathbf{th}}$\textbf{\ order testing nuisance tangent space,
testing tangent space, testing influence functions, efficient score, efficient
information, and efficient testing variance: }Given a model $\mathcal{M}%
\left(  \Theta\right)  $ with parameter space $\Theta$ and a functional
$\tau\left(  \theta\right)  ,$ define $\mathcal{M}\left(  \Theta\left(
\tau^{\dagger}\right)  \right)  $ to be the submodel with parameter space
$\Theta\left(  \tau^{\dagger}\right)  \equiv\Theta\cap\left\{  \theta
;\tau\left(  \theta\right)  =\tau^{\dagger}\right\}  )$. Thus $\mathcal{M}%
\left(  \Theta\left(  \tau^{\dagger}\right)  \right)  $ is the submodel with
$\tau\left(  \theta\right)  $ equal to $\tau^{\dagger}.$ Define, for
$\theta\in\Theta\left(  \tau^{\dagger}\right)  ,$ the $m$th order (i) testing
nuisance tangent space $\Gamma_{m}^{nuis,test}\left(  \theta,\tau^{\dagger
}\right)  $ to be the $m^{th}$ order tangent space for the submodel
$\mathcal{M}\left(  \Theta\left(  \tau^{\dagger}\right)  \right)  ,$ (ii)
testing tangent space $\Gamma_{m}^{test}\left(  \theta,\tau^{\dagger}\right)
$ to be the closed linear span of $\mathbb{IF}_{1,\tau\left(  \cdot\right)
}\left(  \theta\right)  \cup\Gamma_{m}^{nuis,test}\left(  \theta,\tau
^{\dagger}\right)  ,$ (iiia) set $\Gamma_{m}^{nuis,test,\perp}\left(
\theta,\tau^{\dagger}\right)  \equiv\left\{  \mathbb{IF}_{m,\tau\left(
\cdot\right)  }^{test}\right\}  $ of testing influence functions to be the
orthocomplement of $\Gamma_{m}^{nuis,test}\left(  \theta,\tau^{\dagger
}\right)  $ in $\mathcal{U}_{m}\left(  \theta\right)  $ $,$ (iiib) set
$\Gamma_{m}^{std,nuis,test,\perp}\left(  \theta,\tau^{\dagger}\right)
\equiv\left\{  \mathbb{IF}_{m,\tau\left(  \cdot\right)  }^{std,test}\right\}
$ of standardized testing influence functions to be
\[
\left\{  \mathbb{IF}_{m,\tau\left(  \cdot\right)  }^{std,test}\in\Gamma
_{m}^{nuis,test,\perp}\left(  \theta,\tau^{\dagger}\right)  ;\text{ {}%
}E_{\theta}\left[  \mathbb{IF}_{m,\tau\left(  \cdot\right)  }^{std,test}%
\mathbb{IF}_{1,\tau\left(  \cdot\right)  }^{eff}\left(  \theta\right)
\right]  =var_{\theta}\left[  \mathbb{IF}_{1,\tau\left(  \cdot\right)  }%
^{eff}\left(  \theta\right)  \right]  \right\}  ,
\]
$\left(  iv\right)  $ efficient testing score $\mathbb{ES}_{m}^{test}\left(
\theta\right)  \equiv\mathbb{ES}_{m,\tau\left(  \cdot\right)  }^{test}\left(
\theta\right)  \in\Gamma_{m}^{test}\left(  \theta,\tau^{\dagger}\right)  $ to
be
\[
\mathbb{ES}_{m,\tau\left(  \cdot\right)  }^{test}\left(  \theta\right)
=\mathbb{ES}_{1}^{test}-\Pi_{\theta}\left[  \mathbb{ES}_{1}^{test}|\Gamma
_{m}^{nuis,test}\left(  \theta,\tau^{\dagger}\right)  \right]  \equiv
\Pi_{\theta}\left[  \mathbb{ES}_{1}^{test}\left(  \theta\right)  |\Gamma
_{m}^{nuis,test,\perp}\left(  \theta,\tau^{\dagger}\right)  \right]
\]
where $\mathbb{ES}_{1}^{test}\left(  \theta\right)  \equiv\ \mathbb{ES}%
_{1,\tau\left(  \cdot\right)  }^{test}\left(  \theta\right)  \equiv
var_{\theta}\left\{  \mathbb{IF}_{1,\tau\left(  \cdot\right)  }^{eff}\left(
\theta\right)  \right\}  ^{-1}\mathbb{IF}_{1,\tau\left(  \cdot\right)  }%
^{eff}\left(  \theta\right)  ,$ $\left(  v\right)  $ efficient testing
information to be $var_{\theta}\left\{  \mathbb{ES}_{m}^{test}\left(
\theta\right)  \right\}  ,$ and $\left(  vi\right)  $ the efficient testing
variance to be $\left[  var_{\theta}\left\{  \mathbb{ES}_{m}^{test}\left(
\theta\right)  \right\}  \right]  ^{-1}.$
\end{definition}

Further define, for $\theta\in\Theta,$ the $m$th order (i) estimation nuisance
tangent space $\Gamma_{m}^{nuis}\left(  \theta\right)  $ to be $\Gamma
_{m}^{nuis}\left(  \theta\right)  \equiv$ $\left\{  \mathbb{A}_{m}\in
\Gamma_{m}\left(  \theta\right)  ;E\left[  \mathbb{A}_{m}\mathbb{IF}%
_{m,\tau\left(  \cdot\right)  }^{eff}\left(  \theta\right)  \right]
=0\right\}  ,$ and (ii) efficient estimation variance to be $var_{\theta
}\left[  \mathbb{IF}_{m,\tau\left(  \cdot\right)  }^{eff}\left(
\theta\right)  \right]  $.

\textbf{Remark:} For $m=1,$ the testing and estimation nuisance tangent spaces
$\Gamma_{m}^{nuis,test}\left(  \theta,\tau^{\dagger}\right)  $ and $\Gamma
_{m}^{nuis}\left(  \theta\right)  $ are identical. However for $m>1,$
$\Gamma_{m}^{nuis,test}\left(  \theta,\tau^{\dagger}\right)  $ is generally a
strict subset of $\Gamma_{m}^{nuis}\left(  \theta\right)  .$ For example, if
the model can be parametrized as $\theta=\left(  \tau,\rho\right)  $ and
$\Theta$ is the product of the parameter spaces for $\tau$ and $\rho,$ the
$\Gamma_{m}^{nuis,test}\left(  \theta,\tau^{\dagger}\right)  $ is the space of
mth order scores for $\rho;$ however, $\Gamma_{m}^{nuis}\left(  \theta\right)
$ also includes the mixed scores that have $s$ derivatives in the direction
$\tau$ and $m-s\geq1$ derivatives in $\rho$ directions. It is this strict
inclusion that gives rise to higher order phenomena that do not occur in the
first order theory.

\begin{theorem}
\label{gg}\textbf{: }Suppose\textbf{\ }$\mathbb{ES}_{m}^{test}\left(
\theta\right)  $ exists in $\mathcal{U}_{m}\left(  \theta\right)  .$ Then for
$\theta\in\Theta\left(  \tau^{\dagger}\right)  ,$ (i) the set of estimation
nuisance scores $\Gamma_{m}^{nuis}\left(  \theta\right)  $ includes the set of
testing nuisance scores $\Gamma_{m}^{nuis,test}\left(  \theta,\tau^{\dagger
}\right)  $ with equality of the sets when $m=1$, (ii) $\mathbb{IF}%
_{m,\tau\left(  \cdot\right)  }^{test}\left(  \theta\right)  ,\theta\in
\Theta\left(  \tau^{\dagger}\right)  $ is standardized if and only if
$E\left[  \mathbb{IF}_{m,\tau\left(  \cdot\right)  }^{test}\left(
\theta\right)  \mathbb{ES}_{m}^{test}\left(  \theta\right)  \right]  =1\ $if
and only if $E\left[  \mathbb{IF}_{m,\tau\left(  \cdot\right)  }^{test}\left(
\theta\right)  \mathbb{ES}_{1}^{test}\left(  \theta\right)  \right]  =1,$
(iii)%
\[
\left\{  \mathbb{IF}_{m,\tau\left(  \cdot\right)  }^{std,test}\right\}
=\left\{  E_{\theta}\left[  \mathbb{IF}_{m,\tau\left(  \cdot\right)  }%
^{test}\mathbb{ES}_{1}^{test}\left(  \theta\right)  \right]  ^{-1}%
\mathbb{IF}_{m,\tau\left(  \cdot\right)  }^{test};\text{ }\mathbb{IF}%
_{m,\tau\left(  \cdot\right)  }^{test}\in\left\{  \mathbb{IF}_{m,\tau\left(
\cdot\right)  }^{test}\right\}  \right\}  ,
\]
(iv) the set $\left\{  \mathbb{IF}_{m,\tau\left(  \cdot\right)  }\left(
\theta\right)  \right\}  $ of all mth order estimation influence functions is
contained in $\left\{  \mathbb{IF}_{m,\tau\left(  \cdot\right)  }%
^{std,test}\right\}  $ with equality of the sets when $m=1,$ (v)%
\[
\Pi_{\theta}\left[  \mathbb{IF}_{m,\tau\left(  \cdot\right)  }^{std,test}%
\left(  \theta\right)  |\Gamma_{m}^{test}\left(  \theta,\tau^{\dagger}\right)
\right]  =\left\{  var\left[  \mathbb{ES}_{m}^{test}\left(  \theta\right)
\right]  \right\}  ^{-1}\mathbb{ES}_{m}^{test}\left(  \theta\right)  ,
\]
(vi) $\left\{  var_{\theta}\left[  \mathbb{ES}_{m}^{test}\left(
\theta\right)  \right]  \right\}  ^{-1}\mathbb{ES}_{m}^{test}\left(
\theta\right)  \in\left\{  \mathbb{IF}_{m,\tau\left(  \cdot\right)
}^{std,test}\right\}  $ and has the minimum variance $\left\{  var_{\theta
}\left[  \mathbb{ES}_{m}^{test}\left(  \theta\right)  \right]  \right\}
^{-1}$ among members of $\left\{  \mathbb{IF}_{m,\tau\left(  \cdot\right)
}^{std,test}\right\}  $. In particular $\left\{  var_{\theta}\left[
\mathbb{ES}_{m}^{test}\left(  \theta\right)  \right]  \right\}  ^{-1}\leq
var_{\theta}\left[  \mathbb{IF}_{m,\tau\left(  \cdot\right)  }^{eff}\left(
\theta\right)  \right]  $ with equality when $m=1,$ (vii) Given $\mathbb{IF}%
_{m,\tau\left(  \cdot\right)  }^{test}\left(  \cdot\right)  \in\left\{
\mathbb{IF}_{m,\tau\left(  \cdot\right)  }^{test}\left(  \cdot\right)
\right\}  $,any smooth submodel $\widetilde{\theta}\left(  \zeta\right)  $
with range containing $\theta$ and contained in $\Theta\left(  \tau^{\dagger
}\right)  ,\ $and $s\leq m,$ we have
\[
\partial^{s}E_{\theta}\left[  \mathbb{IF}_{m,\tau\left(  \cdot\right)
}^{test}\left(  \widetilde{\theta}\left(  \zeta\right)  \right)  \right]
/\partial\zeta_{l_{1}}...\partial\zeta_{l_{_{s}}}|_{\zeta=\widetilde{\theta
}^{-1}\left\{  \theta\right\}  }=0.
\]
Thus, if $E_{\theta}\left[  \mathbb{IF}_{m,\tau\left(  \cdot\right)  }%
^{test}\left(  \theta^{\ast}\right)  \right]  $ is Fr\'{e}chet differentiable
w.r.t. $\theta^{\ast}$ to order $m+1$ for a norm $\left\vert \left\vert
\cdot\right\vert \right\vert ,$ $E_{\theta}\left[  \mathbb{IF}_{m,\tau\left(
\cdot\right)  }^{test}\left(  \theta+\delta\theta\right)  \right]  =O\left(
\left\vert \left\vert \delta\theta\ \right\vert \right\vert ^{m+1}\right)  $
for $\theta$ and $\theta+\delta\theta$ in an open neighborhood contained in
$\Theta\left(  \tau^{\dagger}\right)  $, since the Taylor expansion of
$E_{\theta}\left[  \mathbb{IF}_{m,\tau\left(  \cdot\right)  }^{test}\left(
\theta^{\ast}\right)  \right]  \ $around $\theta$ through order $m$ is
identically zero$.$
\end{theorem}

The proof of the Theorem will use the following two lemmas:

\begin{lemma}
\label{a} :For any $\mathbb{IF}_{m,\tau\left(  \cdot\right)  }^{test}\left(
\theta\right)  ,\theta\in\Theta\left(  \tau^{\dagger}\right)  $%
\[
E_{\theta}\left[  \mathbb{IF}_{m,\tau\left(  \cdot\right)  }^{test}\left(
\theta\right)  \mathbb{ES}_{1}^{test}\left(  \theta\right)  \right]
=E_{\theta}\left[  \mathbb{IF}_{m,\tau\left(  \cdot\right)  }^{test}\left(
\theta\right)  \mathbb{ES}_{m}^{test}\left(  \theta\right)  \right]
\]

\end{lemma}

\begin{proof}%
\begin{align*}
&  E_{\theta}\left[  \mathbb{IF}_{m,\tau\left(  \cdot\right)  }^{test}%
\mathbb{ES}_{m}^{test}\left(  \theta\right)  \right] \\
&  =E_{\theta}\left[  \mathbb{IF}_{m,\tau\left(  \cdot\right)  }^{test}%
\Pi_{\theta}\left[  \mathbb{ES}_{1}^{test}\left(  \theta\right)  |\Gamma
_{m}^{nuis,test,\perp}\left(  \theta,\tau^{\dagger}\right)  \right]  \right]
=E_{\theta}\left[  \mathbb{IF}_{m,\tau\left(  \cdot\right)  }^{test}%
\mathbb{ES}_{1}^{test}\left(  \theta\right)  \right]  ,
\end{align*}
where the last equality holds by $\mathbb{IF}_{m,\tau\left(  \cdot\right)
}^{test}\in\Gamma_{m}^{nuis,test,\perp}\left(  \theta,\tau^{\dagger}\right)  $
\end{proof}

\begin{lemma}
\label{b}\ For any $\mathbb{IF}_{m,\tau\left(  \cdot\right)  }^{test}\left(
\theta\right)  ,\theta\in\Theta\left(  \tau^{\dagger}\right)  ,$
\begin{align*}
&  \Pi_{\theta}\left[  \mathbb{IF}_{m,\tau\left(  \cdot\right)  }%
^{test}\left(  \theta\right)  |\Gamma_{m}^{test}\left(  \theta,\tau^{\dagger
}\right)  \right] \\
&  =E\left[  \mathbb{IF}_{m,\tau\left(  \cdot\right)  }^{test}\left(
\theta\right)  \mathbb{ES}_{m}^{test}\left(  \theta\right)  \right]  \left\{
var\left[  \mathbb{ES}_{m}^{test}\left(  \theta\right)  \right]  \right\}
^{-1}\mathbb{ES}_{m}^{test}\left(  \theta\right) \\
&  =E\left[  \mathbb{IF}_{m,\tau\left(  \cdot\right)  }^{test}\left(
\theta\right)  \mathbb{ES}_{1}^{test}\left(  \theta\right)  \right]  \left\{
var\left[  \mathbb{ES}_{m}^{test}\left(  \theta\right)  \right]  \right\}
^{-1}\mathbb{ES}_{m}^{test}\left(  \theta\right)
\end{align*}

\end{lemma}

\begin{proof}
$\Gamma_{m}^{test}\left(  \theta,\tau^{\dagger}\right)  =\left\{
c\mathbb{ES}_{m}^{test}\left(  \theta\right)  ;\text{ }c\in R^{1}\right\}
\oplus\Gamma_{m}^{nuis,test}\left(  \theta,\tau^{\dagger}\right)  .$ Thus, by
$\mathbb{IF}_{m,\tau\left(  \cdot\right)  }^{test}\left(  \theta\right)
\in\Gamma_{m}^{nuis,test,\perp}\left(  \theta,\tau^{\dagger}\right)  ,$
\begin{align*}
\Pi_{\theta}\left[  \mathbb{IF}_{m,\tau\left(  \cdot\right)  }^{test}\left(
\theta\right)  |\Gamma_{m}^{test}\left(  \theta,\tau^{\dagger}\right)
\right]   &  =\Pi_{\theta}\left[  \mathbb{IF}_{m,\tau\left(  \cdot\right)
}^{test}\left(  \theta\right)  |\left\{  c\mathbb{ES}_{m}^{test}\left(
\theta\right)  ;c\in R^{1}\right\}  \right] \\
&  =E\left[  \mathbb{IF}_{m,\tau\left(  \cdot\right)  }^{test}\left(
\theta\right)  \mathbb{ES}_{m}^{test}\left(  \theta\right)  \right]  \left\{
var\left[  \mathbb{ES}_{m}^{test}\left(  \theta\right)  \right]  \right\}
^{-1}\mathbb{ES}_{m}^{test}\left(  \theta\right)  .
\end{align*}
Now apply Lemma \ref{a}.
\end{proof}

\begin{proof}
(Theorem \ref{gg}) (i) is immediate from the definitions. (ii) and (iiii)
follow from
\begin{align*}
E\left[  \mathbb{IF}_{m,\tau\left(  \cdot\right)  }^{test}\left(
\theta\right)  \mathbb{ES}_{m}^{test}\left(  \theta\right)  \right]   &
=1\Leftrightarrow E\left[  \mathbb{IF}_{m,\tau\left(  \cdot\right)  }%
^{test}\left(  \theta\right)  \mathbb{ES}_{1}^{test}\left(  \theta\right)
\right]  =1\\
&  \Leftrightarrow E_{\theta}\left[  \mathbb{IF}_{m,\tau\left(  \cdot\right)
}^{test}\mathbb{IF}_{1,\tau\left(  \cdot\right)  }^{eff}\left(  \theta\right)
\right]  =var_{\theta}\left[  \mathbb{IF}_{1,\tau\left(  \cdot\right)  }%
^{eff}\left(  \theta\right)  \right]  ,
\end{align*}
where we have used Lemma \ref{a}. For (iv), note $\left\{  \mathbb{IF}%
_{m,\tau\left(  \cdot\right)  }\left(  \theta\right)  \right\}  \subset
\left\{  \mathbb{IF}_{m,\tau\left(  \cdot\right)  }^{test}\right\}  $ follows
from the fact that every smooth submodel through $\theta$ in model
$\mathcal{M}\left(  \Theta\left(  \tau^{\dagger}\right)  \right)  $ is a
smooth submodel through $\theta$ in model $\mathcal{M}\left(  \Theta\right)  $
. Thus it remains to prove that $\mathbb{IF}_{m,\tau\left(  \cdot\right)
}\left(  \theta\right)  $ is standardized. But, by Part 4 of Theorem
\ref{eift}, $E_{\theta}\left[  \mathbb{IF}_{m,\tau\left(  \cdot\right)
}\left(  \theta\right)  \mathbb{IF}_{1,\tau\left(  \cdot\right)  }%
^{eff}\left(  \theta\right)  \right]  =var_{\theta}\left[  \mathbb{IF}%
_{1,\tau\left(  \cdot\right)  }^{eff}\left(  \theta\right)  \right]  $. {}(v)
follows at once from Lemma \ref{a} and Part (ii). {}For (vi), note that
$\left\{  var_{\theta}\left[  \mathbb{ES}_{m}^{test}\left(  \theta\right)
\right]  \right\}  ^{-1}\mathbb{ES}_{m}^{test}\left(  \theta\right)
\in\left\{  \mathbb{IF}_{m,\tau\left(  \cdot\right)  }^{std,test}\right\}  $
by definition. Thus
\[
var_{\theta}\left\{  E_{\theta}\left[  \mathbb{IF}_{m,\tau\left(
\cdot\right)  }^{test}\mathbb{ES}_{m}^{test}\left(  \theta\right)  \right]
^{-1}\mathbb{IF}_{m,\tau\left(  \cdot\right)  }^{test}\right\}  \geq\left\{
var_{\theta}\left[  \mathbb{ES}_{m}^{test}\left(  \theta\right)  \right]
\right\}  ^{-1}%
\]
follows from (v). The result then follows from part (iii). Part (vii) is
proved analogously to Theorem \ref{eiet} except now all scores lie in
$\Gamma_{m}^{nuis}\left(  \theta\right)  $ by range $\widetilde{\theta}\left(
\zeta\right)  $ in $\Theta\left(  \tau^{\dagger}\right)  .$
\end{proof}

In the case of (locally) nonparametric models, we can explicitly characterize
$\Gamma_{m}^{test,\perp}\left(  \theta,\tau^{\dagger}\right)  $. Let $\left\{
\mathbb{U}_{j,j}^{test,\perp}\left(  \theta,\tau^{\dagger}\right)  \right\}  $
be the set of all $\mathbb{U}_{j,j}^{test,\perp}\left(  \theta,\tau^{\dagger
}\right)  =\mathbb{V}\left[  U_{j,j}^{test,\perp}\left(  \theta,\tau^{\dagger
}\right)  \right]  $ with the $U_{j,j}^{test,\perp}\left(  \theta
,\tau^{\dagger}\right)  =\sum_{l=1}^{\infty}c_{l}IF_{1,\tau\left(
\cdot\right)  ,i_{1}}^{eff}\left(  \theta\right)
%TCIMACRO{\dprod \limits_{s=2}^{j}}%
%BeginExpansion
{\displaystyle\prod\limits_{s=2}^{j}}
%EndExpansion
h_{l,s}\left(  O_{i_{s}};\theta\right)  \in\mathcal{U}_{j}\left(
\theta\right)  ,$ indexed by constants $c_{l}\in R^{1},$ and functions
$h_{l,s}\left(  O_{i_{s}};\theta\right)  $ satisfying $E_{\theta}\left[
h_{l,s}\left(  O_{i_{s}};\theta\right)  \right]  =0.$ We remark that the
subset of $\mathcal{U}_{j}\left(  \theta\right)  $ comprised of all $j$th
order degenerate U-statistics can be written $\left\{  \mathbb{V}\left[
\sum_{l=1}^{\infty}\
%TCIMACRO{\dprod \limits_{s=1}^{j}}%
%BeginExpansion
{\displaystyle\prod\limits_{s=1}^{j}}
%EndExpansion
h_{l,s}\left(  O_{i_{s}};\theta\right)  \right]  \right\}  $. Thus $\left\{
\mathbb{U}_{j,j}^{test,\perp}\left(  \theta,\tau^{\dagger}\right)  \right\}  $
simply restricts one of the functions $h_{l,s\ }\left(  O\ ;\theta\right)  $
to be $c_{l}IF_{1,\tau\left(  \cdot\right)  \ }^{eff}.$

\begin{theorem}
\label{ff} If the model $\mathcal{M}\left(  \Theta\right)  $ is (locally)
nonparametric, then $\Gamma_{m}^{test,\perp}\left(  \theta,\tau^{\dagger
}\right)  =\left\{  \sum_{j=2}^{m}\mathbb{U}_{j,j}^{test,\perp}\left(
\theta,\tau^{\dagger}\right)  ;\text{ }\mathbb{U}_{j,j}^{test,\perp}\left(
\theta,\tau^{\dagger}\right)  \in\left\{  \mathbb{U}_{j,j}^{test,\perp}\left(
\theta,\tau^{\dagger}\right)  \right\}  \right\}  $ .
\end{theorem}

\begin{proof}
\textbf{\ }Since the model is locally nonparametric $\Gamma_{m}^{test}\left(
\theta,\tau^{\dagger}\right)  $ includes the set of all mean zero first order
$U-$statistics $\mathcal{U}_{1}\left(  \theta\right)  \ $and thus any element
of $\Gamma_{m}^{test,\perp}\left(  \theta,\tau^{\dagger}\right)  $ must be a
sum of degenerate $U-statistics$ of orders $2$ through $m.$ We continue by
induction. First we prove the theorem for $m=2.$ $\ $Now, $\Gamma_{2}%
^{test}\left(  \theta,\tau^{\dagger}\right)  =\mathcal{U}_{1}\left(
\theta\right)  +\mathcal{U}_{2,2}^{nuis,test}\left(  \theta\right)  $ where
$\mathcal{U}_{2,2}^{nuis,test}$ is the closed linear span of the 2nd order
degenerate part $\sum_{s\neq j}S_{l_{1},j}S_{l_{2},s}$ of 2nd order scores
$\widetilde{\mathbb{S}}_{2,\overline{l}_{2}}=\sum_{j}S_{l_{1}l_{2},j}%
+\sum_{s\neq j}S_{l_{1},j}S_{l_{2},s}$ in model $\mathcal{M}\left(
\Theta\left(  \tau^{\dagger}\right)  \right)  ,$ where $\sum_{s\neq j}%
S_{l_{1},j}S_{l_{2},s}$ is a sum of products $S_{l_{1},j}S_{l_{2},s}$ of first
order scores in model $\mathcal{M}\left(  \Theta\left(  \tau^{\dagger}\right)
\right)  $ for two different subjects. By model $\mathcal{M}\left(
\Theta\right)  $ being (locally) nonparametric, the set of first order scores
in model $\mathcal{M}\left(  \Theta\left(  \tau^{\dagger}\right)  \right)  $
is precisely the set of random variables $\Gamma_{1}^{nuis,test}\left(
\theta,\tau^{\dagger}\right)  $ orthogonal to $IF_{1,\tau\left(  \cdot\right)
}^{eff}\left(  \theta\right)  .$ But the set of degenerate $U-statistics$ of
order $2$ orthogonal to the product of two scores in $\Gamma_{1}%
^{nuis,test}\left(  \theta,\tau^{\dagger}\right)  $ is clearly $\left\{
\mathbb{U}_{2,2}^{test,\perp}\left(  \theta,\tau^{\dagger}\right)  \right\}
.\ $ Suppose now the theorem is true for $m,m\geq2,$ we show it is true for
$m+1. $ By $\mathcal{M}\left(  \Theta\right)  $ (locally) nonparametric and
the induction assumption, $\Gamma_{m+1}^{test}\left(  \theta,\tau^{\dagger
}\right)  =\Gamma_{m}^{test}\left(  \theta,\tau^{\dagger}\right)
+\mathcal{U}_{m+1,m+1}^{test}\left(  \theta\right)  $ where $\mathcal{U}%
_{m+1,m+1}^{nuis,test}\left(  \theta\right)  $ is the closed linear span of
the sum of products of first order scores in model $\mathcal{M}\left(
\Theta\left(  \tau^{\dagger}\right)  \right)  $ for $m+1$ different subjects.
But $\left\{  \mathbb{U}_{m+1,m+1}^{test,\perp}\left(  \theta,\tau^{\dagger
}\right)  \right\}  $ is the set of set of degenerate $U-statistics$ of order
$m+1$ orthogonal to $\mathcal{U}_{m+1,m+1}^{nuis,test}\left(  \theta\right)
.$
\end{proof}

\subsection{Implicitly defined Functionals and Testing Influence Functions:}

In the following theorem we show that estimation influence functions
$\mathbb{IF}_{m,\psi\left(  \tau,\cdot\right)  }\left(  \theta\right)  $ for
the parameter $\psi\left(  \tau,\cdot\right)  $ evaluated at the solution
$\tau\left(  \theta\right)  $ to $0=\psi\left(  \tau\ ,\theta\right)  $ is
contained in the set $\left\{  \mathbb{IF}_{m,\tau\left(  \cdot\right)
}^{test}\left(  \theta\right)  \right\}  $ of testing influence functions for
$\tau\left(  \theta\right)  .$ We also derive the estimation influence
functions $\mathbb{IF}_{m,\tau\left(  \cdot\right)  }\left(  \theta\right)
=\sum_{s=1}^{m}\mathbb{IF}_{s,s,\tau\left(  \cdot\right)  }\left(
\theta\right)  $ for $\tau\left(  \theta\right)  $ in terms of the estimation
influence functions $\mathbb{IF}_{m,\psi\left(  \tau,\cdot\right)  }\left(
\theta\right)  $ for $\psi\left(  \tau,\cdot\right)  $ and their derivatives
with respect to $\tau.$

\begin{theorem}
\label{tt} Let $\tau\left(  \theta\right)  $ be the assumed unique functional
defined by $0=\psi\left(  \tau\left(  \theta\right)  ,\theta\right)
,\theta\in\Theta.$ Then, for $\theta\in\Theta\left(  \tau^{\dagger}\right)  $,
{}whenever $\mathbb{IF}_{m,\psi\left(  \tau^{\dagger},\cdot\right)  }\left(
\theta\right)  $ and $\mathbb{IF}_{\ m,\tau\left(  \cdot\right)  }\left(
\theta\right)  $ exist , (i) $\mathbb{IF}_{m,\psi\left(  \tau^{\dagger}%
,\cdot\right)  }\left(  \theta\right)  \in\left\{  \mathbb{IF}_{m,\tau\left(
\cdot\right)  }^{test}\left(  \theta\right)  \right\}  ,$ (ii) $\mathbb{IF}%
_{1,\tau\left(  \cdot\right)  }\left(  \theta\right)  =-\psi_{\tau}%
^{-1}\mathbb{IF}_{1,\psi\left(  \tau^{\dagger},\cdot\right)  }\left(
\theta\right)  \in\left\{  \mathbb{IF}_{1,\tau\left(  \cdot\right)
}^{std,test}\left(  \theta\right)  \right\}  $ {}where $\psi_{\tau}%
\equiv\partial\psi\left(  \tau,\theta\right)  /\partial\tau_{|\tau
=\tau^{\dagger}},$ (iii) $\mathbb{IF}_{m,m,\tau\left(  \cdot\right)  }\left(
\theta\right)  =-\psi_{\tau}^{-1}\left\{  \mathbb{IF}_{m,m,\psi\left(
\tau^{\dagger},\cdot\right)  }\left(  \theta\right)  +\mathbb{Q}_{m,m}\left(
\theta\right)  \right\}  ,\ $where $\mathbb{Q}_{m,m}\left(  \theta\right)
\equiv\mathbb{Q}_{m,m,\tau\left(  \cdot\right)  }\left(  \theta\right)
=\mathbb{V}\left\{  Q_{m,m}\left(  \theta\right)  \right\}  \in\left\{
\mathbb{U}_{m,m}^{test,\perp}\left(  \theta,\tau^{\dagger}\right)  \right\}
.$ For $m=2,$
\begin{align}
\mathbb{\ }Q_{2,2}\left(  \theta\right)   &  =\frac{1}{2}\psi_{\backslash
\tau^{2}}IF_{1,\tau\left(  \cdot\right)  ,i_{1}}\left(  \theta\right)
\ IF_{1,\tau\left(  \cdot\right)  ,i_{2}}\ \left(  \theta\right) \nonumber\\
&  +\frac{1}{2}\left[
\begin{array}
[c]{c}%
\left(
\begin{array}
[c]{c}%
\frac{\partial IF_{1,\psi\left(  \tau^{\dagger},\cdot\right)  ,i_{1}}\left(
\theta\right)  }{\partial\tau}\\
-E_{\theta}\left[  \frac{\partial IF_{1,\psi\left(  \tau^{\dagger}%
,\cdot\right)  ,i_{1}}\left(  \theta\right)  }{\partial\tau}\right]
\end{array}
\right)  IF_{1,\tau\left(  \cdot\right)  ,i_{2}}\ \left(  \theta\right) \\
+\left(
\begin{array}
[c]{c}%
\frac{\partial IF_{1,\psi\left(  \tau^{\dagger},\cdot\right)  ,i_{2}}\left(
\theta\right)  }{\partial\tau}\\
-E_{\theta}\left[  \frac{\partial IF_{1,\psi\left(  \tau^{\dagger}%
,\cdot\right)  ,i_{2}}\left(  \theta\right)  }{\partial\tau}\right]
\end{array}
\right)  IF_{1,\tau,i_{1}}\left(  \theta\right)
\end{array}
\right] \label{q22}%
\end{align}
where $\frac{\partial IF_{1,\psi\left(  \tau^{\dagger},\cdot\right)  ,i_{1}%
}\left(  \theta\right)  }{\partial\tau}=\partial IF_{1,\psi\left(  \tau
,\cdot\right)  ,i_{1}}\left(  \theta\right)  /\partial\tau_{|\tau
=\tau^{\dagger}}$. $Q_{m,m}\left(  \theta\right)  $ is given in the appendix
as well as the general formula.
\end{theorem}

\begin{proof}
(i) For $r\leq m,$ consider any suitably smooth $r\ $dimensional parametric
submodel $\widetilde{\theta}\left(  \zeta\right)  $ with range containing
$\theta$ and contained in $\Theta\left(  \tau^{\dagger}\right)  .$ Let
$\widetilde{\mathbb{S}}_{s\backslash\overline{l}_{s}}\left(  \theta\right)  $
be any associated $sth-order$ score $s\leq m$. By definition of $\tau\left(
\theta\right)  $, $\psi\left(  \tau\left(  \theta\left(  \zeta\right)
\right)  ,\theta\left(  \zeta\right)  \right)  =0.$ Hence, $0=\partial^{s}%
\psi\left(  \tau\left(  \theta\left(  \zeta\right)  \right)  ,\theta\left(
\zeta\right)  \right)  /\partial\zeta_{l_{1}}...\partial\zeta_{l_{_{s}}%
|\zeta=\widetilde{\theta}^{-1}\left(  \theta\right)  }.$ Now we expand the RHS
using the chain rule and note that the only non-zero term is the term
$\psi_{\backslash\overline{l}_{s}}\left(  \tau^{\dagger},\theta\right)  $ in
which all $s-$derivatives are taken with respect to the second $\theta\left(
\zeta\right)  $ in $\psi\left(  \tau\left(  \theta\left(  \zeta\right)
\right)  ,\theta\left(  \zeta\right)  \right)  ;$ all other terms include
derivatives of $\tau\left(  \theta\left(  \zeta\right)  \right)  ,\ $which are
zero by range $\widetilde{\theta}\left(  \zeta\right)  \subset\Theta\left(
\tau^{\dagger}\right)  .$ Further $\psi_{\backslash\overline{l}_{s}}\left(
\tau^{\dagger},\theta\right)  =E_{\theta}\left[  \mathbb{IF}_{m,\psi\left(
\tau^{\dagger},\cdot\right)  }\left(  \theta\right)  \widetilde{\mathbb{S}%
}_{s\backslash\overline{l}_{s}}\left(  \theta\right)  \right]  $ by the
definition of the estimation influence function $\mathbb{IF}_{m,\psi\left(
\tau^{\dagger},\cdot\right)  }\left(  \theta\right)  $. We conclude that
$\mathbb{IF}_{m,\psi\left(  \tau^{\dagger},\cdot\right)  }\left(
\theta\right)  $ is in $\Gamma_{m}^{nuis,test}\left(  \theta,\tau^{\dagger
}\right)  ^{\perp}.$ (ii) $\mathbb{IF}_{1,\tau\left(  \cdot\right)  }$
$=-\psi_{\tau}^{-1}\mathbb{IF}_{1,\psi\left(  \tau^{\dagger},\cdot\right)  }$
is straightforward. That $\mathbb{IF}_{1,\tau\left(  \cdot\right)  }$ is
contained in $\left\{  \mathbb{IF}_{1,\tau\left(  \cdot\right)  }%
^{std,test}\right\}  $ follows by Part (iv) of Theorem \ref{gg}. (iii) See
appendix for proof.
\end{proof}

\subsection{"Inefficiency" of the Efficient Score}

We now provide an example to show that, contrary to what one might expect
based on Part (vi) of Theorem \ref{gg}, inference concerning $\tau\left(
\theta\right)  $ may be more efficient when based on an 'inefficient' member
of the set $\left\{  \mathbb{IF}_{m,\tau\left(  \cdot\right)  }^{test}\left(
\theta\right)  \right\}  $ such as $\mathbb{IF}_{m,\psi\left(  \tau^{\dagger
},\cdot\right)  }\left(  \theta\right)  $ than when based on the efficient
score $\mathbb{ES}_{m,\tau\left(  \cdot\right)  }^{test}\left(  \theta\right)
.$ Without loss of generality, it is sufficient to consider the case $m=2$. In
the following example it is $\widetilde{\tau}_{k}\left(  \theta\right)  \ $and
$\widetilde{\psi}_{k}\left(  \tau^{\dagger},\theta\right)  $ that play the
role of $\tau\left(  \theta\right)  $ and $\psi\left(  \tau^{\dagger}%
,\theta\right)  $ in the preceding theorem, because $\widetilde{\tau}%
_{k}\left(  \theta\right)  $ and $\widetilde{\psi}_{k}\left(  \tau^{\dagger
},\theta\right)  $ have, but $\tau\left(  \theta\right)  $ and $\psi\left(
\tau^{\dagger},\theta\right)  $ do not have, higher order estimation and
testing influence functions.

\textbf{Example 1c (continued):} In this example, with $Y^{\ast}\left(
\tau\right)  \equiv Y^{\ast}-\tau A,$ $A\ $and\emph{\ }$Y^{\ast}$%
\emph{\ }binary,
\[
\psi\left(  \tau,\theta\right)  =E_{\theta}\left[  \left\{  Y^{\ast}\left(
\tau\right)  -E_{\theta}\left(  Y^{\ast}\left(  \tau\right)  |X\right)
\right\}  \left\{  A-E_{\theta}\left(  A|X\right)  \right\}  \right]
\]
and $\tau\left(  \theta\right)  $ satisfies $\psi\left(  \tau\left(
\theta\right)  ,\theta\right)  =0$. Let $\widetilde{\tau}_{k}\left(
\theta\right)  $ satisfy $\widetilde{\psi}_{k}\left(  \widetilde{\tau}%
_{k}\left(  \theta\right)  ,\theta\right)  =0$ where $\widetilde{\psi}%
_{k}\left(  \tau,\theta\right)  =E_{\theta}\left[  Y^{\ast}\left(
\tau\right)  A\right]  -E_{\theta}\left\{  \left[  \Pi_{\theta}\left[
B\left(  \tau\right)  |\overline{Z}_{k}\right]  \Pi_{\theta}\left[
P|\overline{Z}_{k}\right]  \right]  \right\}  $ is defined in Section 3.1 with
$\tau$ a real-valued index and $B\left(  \tau\right)  =b\left(  X,\tau\right)
=E_{\theta}\left(  Y^{\ast}\left(  \tau\right)  |X\right)  $. Note
$\widetilde{\psi}_{k,\tau}\left(  \tau,\theta\right)  \equiv\partial
\widetilde{\psi}_{k}\left(  \tau,\theta\right)  /\partial\tau=-\left\{
E_{\theta}\left[  A^{2}\right]  -E_{\theta}\left[  \left\{  \Pi_{\theta
}\left[  P|\overline{Z}_{k}\right]  \right\}  ^{2}\right]  \right\}  ,$
$\psi_{\tau}\left(  \tau,\theta\right)  =-E_{\theta}\left[  var_{\theta
}\left(  A|X\right)  \right]  ,$ $\widetilde{\psi}_{k,\tau^{2}}\left(
\tau,\theta\right)  =\psi_{\tau^{2}}\left(  \tau,\theta\right)  =0.$ Below we
freely use results of Theorems \ref{TBrate}, \ref{DRHOIF}, and \ref{EBrate}.
We suppose that $0<\sigma<var_{\theta}\left(  A|X\right)  $ and $E_{\theta
}\left[  A^{2}\right]  <c$ for some $\left(  \sigma,c\right)  ,$ $\beta
=\frac{\beta_{p}+\beta_{b}}{2}<1/4.$ Choose $k=k_{opt}\left(  2\right)
n^{2\sigma\ }=n^{\frac{2}{1+4\beta}+2\sigma},\sigma>0$ so the truncation bias
of $\widehat{\psi}_{2,k\ }\left(  \tau\right)  \equiv\psi_{2,k\ }\left(
\tau,\widehat{\theta}\right)  $ is $O_{p}\left(  n^{-\frac{4\beta}{4\beta+d}%
}\right)  \ $and $n^{-\frac{4\beta}{4\beta+d}}\ll var_{\theta}\left[
\widehat{\psi}_{2,k\ }\left(  \tau\right)  \right]  \asymp k/n^{2}%
=n^{-2\left(  \frac{4\beta}{4\beta+d}+\sigma\right)  }.$ We assume the given
$\left(  \beta_{g},\beta_{b},\beta_{p}\right)  $ are such that the order
$O_{p}\left[  n^{-\left(  \frac{\beta_{g}}{2\beta_{g}+d}+\frac{\beta_{b}%
}{d+2\beta_{b}}+\frac{\beta_{p}}{d+2\beta_{p}}\right)  }\right]  $ of the
estimation bias of $\widehat{\psi}_{2,k\ }\left(  \tau\right)  $ is
$O_{p}\left(  n^{-\frac{4\beta}{4\beta+d}}\right)  .$ Then $\left\vert
\widehat{\psi}_{2,k\ }\left(  \tau\right)  -\widetilde{\psi}_{k}\left(
\tau,\theta\right)  \right\vert $ and $\left\vert \widehat{\psi}%
_{2,k\ }\left(  \tau\right)  -\psi\left(  \tau,\theta\right)  \ \right\vert $
are $O_{p}\left(  n^{-\frac{4\beta}{4\beta+d}+\sigma}\right)  $ which just
exceeds the minimax rate $O_{p}\left(  n^{-\frac{4\beta}{4\beta+d}\ }\right)
$ for $\sigma$ very small.

Our goal is to compare the coverage and length of confidence intervals for
$\widetilde{\tau}_{k}\left(  \theta\right)  $ and $\tau\left(  \theta\right)
$ based on$\ $\emph{\ }%
\begin{align*}
C_{1-\alpha,\widetilde{\psi}_{k}\left(  \tau\right)  }  &  \equiv\left\{
\tau;-z_{1-\alpha/2}<\frac{\psi_{2,k}\left(  \tau,\widehat{\theta}\right)
}{\mathbb{W}_{2,\widetilde{\psi}_{k}\left(  \tau\right)  }\left(
\widehat{\theta}\right)  }<z_{1-\alpha/2}\right\}  ,\\
C_{1-\alpha,2,\widetilde{\tau}_{k}}  &  \equiv\left\{  \tau;-z_{1-\alpha
/2}<\frac{\tau_{2,k}\left(  \widehat{\theta}\right)  -\tau}{\mathbb{W}%
_{2,\widetilde{\tau}_{k}}\left(  \widehat{\theta}\right)  }<z_{1-\alpha
/2}\right\}  ,\\
C_{1-\alpha,2,ES}  &  \equiv\left\{  \tau;-z_{1-\alpha/2}<\frac{\mathbb{ES}%
_{2,\widetilde{\tau}_{k}}^{test}\left(  \widehat{\theta}\left(  \tau\right)
\right)  }{\mathbb{W}_{2,\widetilde{\tau}_{k}}^{ES}\left(  \widehat{\theta
}\left(  \tau\right)  \right)  }<z_{1-\alpha/2}\right\}  ,
\end{align*}
where $\mathbb{W}_{2,\widetilde{\psi}_{k}\left(  \tau\right)  }\left(
\widehat{\theta}\right)  ,\mathbb{W}_{2,\widetilde{\tau}_{k}}\left(
\widehat{\theta}\right)  ,\mathbb{W}_{2,\widetilde{\tau}_{k}}^{ES}\left(
\widehat{\theta}\left(  \tau\right)  \right)  $ are appropriate variance
estimators, $\widehat{\theta}$ is our usual split sample initial estimator,
and $\widehat{\theta}\left(  \tau^{\dagger}\right)  $ is an initial split
sample estimator depending on $\tau^{\dagger}$ that satisfies $\widetilde{\psi
}_{k}\left(  \tau,\widehat{\theta}\left(  \tau^{\dagger}\right)  \right)  =0$
if $\tau=\tau^{\dagger}$ , i.e., $\tau\left[  \widehat{\theta}\left(
\tau^{\dagger}\right)  \right]  =\tau^{\dagger},$ We assume that if
$\tau\left(  \theta\right)  =\tau^{\dagger}$ then the convergence rate under
$\theta$ of our estimator of $b\left(  X,\tau^{\ast}\right)  $ for any
$\tau^{\ast}$ remains $n^{-\frac{\beta_{b}}{d+2\beta_{b}}}$. Now the
assumption $0<\sigma<E_{\theta}\left[  var_{\theta}\left(  A|X\right)
\right]  , $ $E_{\theta}\left[  A^{2}\right]  <c$ implies $\left\vert
\widetilde{\tau}_{k}\left(  \theta\right)  -\tau\left(  \theta\right)
\right\vert /\left\vert \widetilde{\psi}_{k}\left(  \tau,\theta\right)
-\psi\left(  \tau,\theta\right)  \right\vert \ $is uniformly bounded away from
zero and infinity. It then follows from earlier results on $\widehat{\psi
}_{2,k\ }\left(  \tau\right)  ,$ the assumption $0<\sigma<E_{\theta}\left[
var_{\theta}\left(  A|X\right)  \right]  ,E_{\theta}\left[  A^{2}\right]  <c,$
and Theorem \ref{tt} that $C_{1-\alpha,\widetilde{\psi}_{k}\left(
\tau\right)  }$ is a uniform asymptotic $1-\alpha$ confidence interval for
both $\tau\left(  \theta\right)  $ and $\widetilde{\tau}_{k}\left(
\theta\right)  $ of length $O_{p}\left(  n^{-\frac{4\beta}{4\beta+d}+\sigma
}\right)  .$

The next theorem gives explicit formulae for $\psi_{2,k}\left(  \tau
,\widehat{\theta}\right)  ,$ $\mathbb{ES}_{2,\widetilde{\tau}_{k}}%
^{test}\left(  \widehat{\theta}\left(  \tau\right)  \right)  ,$ and
$\tau_{2,k}\left(  \widehat{\theta}\right)  .$ Using these formulae we
calculate the biases and variances necessary to compare the coverage of the
three intervals.

This comparison requires each of our three candidate procedures to be on the
same scale. Therefore we used standardized versions of the relevant statistics.

\begin{theorem}
\label{aa} Suppose the assumptions described in the preceding example hold. Then

(i)%
\begin{align*}
\psi_{2,k}\left(  \tau,\widehat{\theta}\right)   &  =\widetilde{\psi}_{k,\tau
}\left(  \tau,\widehat{\theta}\right)  +\mathbb{IF}_{2,\widetilde{\psi}%
_{k}\left(  \tau,\cdot\right)  }\left(  \widehat{\theta}\right) \\
&  =\widetilde{\psi}_{k,\tau}\left(  \tau,\widehat{\theta}\right)
+\mathbb{V}\left[  \left(  Y^{\ast}\left(  \tau\right)  -\widehat{b}\left(
X,\tau\right)  \right)  \left\{  A-\widehat{p}\left(  X\right)  \right\}
\right] \\
&  +\mathbb{V}\left[  \left\{  \left[  Y^{\ast}\left(  \tau\right)
-\widehat{b}\left(  X,\tau\right)  \right]  \overline{Z}_{k}^{T}\right\}
_{i_{1}}\left\{  \overline{Z}_{k}\left[  A-\widehat{p}\left(  X\right)
\right]  \right\}  _{i_{2}}\right]
\end{align*}
where $\widehat{b}\left(  X,\tau\right)  =\widehat{B}\left(  \tau\right)
=E_{\widehat{\theta}}\left(  Y^{\ast}\left(  \tau\right)  \ |X\right)  ,$
$\widehat{p}\left(  X\right)  =\widehat{P}=E_{\widehat{\theta}}\left(
A|X\right)  ;$

(ii): Let $\widehat{\epsilon}$ denote $Y-\widehat{b}\left(  X\right)  ,$ and
$\widehat{\Delta}$ denote $A-\widehat{p}\left(  X\right)  .$ Thus,%
\begin{align*}
&  \mathbb{ES}_{2,\widetilde{\tau}_{k}}^{test}\left(  \widehat{\theta}\left(
\tau^{\dagger}\right)  \right) \\
&  =v\left(  \widehat{\theta}\left(  \tau^{\dagger}\right)  \right)  \left\{
var_{\widehat{\theta}\left(  \tau^{\dagger}\right)  }\left[  \mathbb{IF}%
_{2,\widetilde{\psi}_{k}\left(  \tau^{\dagger},\cdot\right)  }\left(
\widehat{\theta}\left(  \tau^{\dagger}\right)  \right)  \right]  -var\left[
\mathbb{U}_{2,2,\widetilde{\tau}_{k}\left(  \cdot\right)  }^{\ast,test,\perp
}\left(  \widehat{\theta}\left(  \tau^{\dagger}\right)  ,\tau^{\dagger
}\right)  \right]  \right\}  ^{-1}\\
&  \times\left\{  \mathbb{IF}_{2,\widetilde{\psi}_{k}\left(  \tau^{\dagger
},\cdot\right)  }\left(  \widehat{\theta}\left(  \tau^{\dagger}\right)
\right)  -\mathbb{U}_{2,2,\widetilde{\tau}_{k}\left(  \cdot\right)  }%
^{\ast,test,\perp}\left(  \widehat{\theta}\left(  \tau^{\dagger}\right)
,\tau^{\dagger}\right)  \right\}
\end{align*}
\newline where%
\begin{align*}
&  U_{2,2,\widetilde{\tau}_{k}\left(  \cdot\right)  ,ij}^{\ast,test,\perp
}\left(  \widehat{\theta}\left(  \tau^{\dagger}\right)  ,\tau^{\dagger}\right)
\\
&  =\left(  E_{\widehat{\theta}}\left[  \widehat{\epsilon}_{i}^{2}%
\widehat{\Delta}_{i}^{2}\right]  \right)  ^{-1}\times\widehat{\epsilon}%
_{i}\widehat{\Delta}_{i}\\
&  \left\{
\begin{array}
[c]{c}%
-\left\{
\begin{array}
[c]{c}%
\left(  E_{\widehat{\theta}}\left[  \widehat{\epsilon}_{i}^{2}\widehat{\Delta
}_{i}^{2}\right]  \right)  ^{-1}E_{\widehat{\theta}}\left[  \widehat{\epsilon
}\widehat{\Delta}^{2}\overline{Z}_{k}^{T}\right] \\
\times E_{\widehat{\theta}}\left[  \widehat{\epsilon}^{2}\widehat{\Delta
}\overline{Z}_{k}^{T}\right]  \widehat{\epsilon}_{j}\widehat{\Delta}_{j}%
\end{array}
\right\} \\
+E_{\widehat{\theta}}\left[  \widehat{\epsilon}_{i}^{2}\widehat{\Delta}%
_{i}\overline{Z}_{k,i}^{T}\right]  \overline{Z}_{k,j}\widehat{\Delta}_{j}\\
+E_{\widehat{\theta}}\left[  \widehat{\epsilon}_{i}\widehat{\Delta}_{i}%
^{2}\overline{Z}_{k,i}^{T}\right]  \overline{Z}_{k,j}\widehat{\epsilon}_{j}%
\end{array}
\right\}
\end{align*}
and
\[
v\left(  \theta\right)  =E_{\theta}\left[  var_{\theta}\left(  A|X\right)
\right]
\]
Also,
\begin{align}
&  var_{\widehat{\theta}\left(  \tau^{\dagger}\right)  }\left\{
\mathbb{ES}_{2,\widetilde{\tau}_{k}\left(  \cdot\right)  }^{test}\left(
\widehat{\theta}\left(  \tau^{\dagger}\right)  \right)  \right\}
^{-1}\mathbb{ES}_{2,\widetilde{\tau}_{k}\left(  \cdot\right)  }^{test}\left(
\widehat{\theta}\left(  \tau^{\dagger}\right)  \right) \nonumber\\
&  =v\left(  \widehat{\theta}\left(  \tau^{\dagger}\right)  \right)
^{-1}\left\{  \mathbb{IF}_{2,\widetilde{\psi}_{k}\left(  \tau^{\dagger}%
,\cdot\right)  }\left(  \widehat{\theta}\left(  \tau^{\dagger}\right)
\right)  -\mathbb{U}_{2,2,\widetilde{\tau}_{k}\left(  \cdot\right)  }%
^{\ast,test,\perp}\left(  \widehat{\theta}\left(  \tau^{\dagger}\right)
,\tau^{\dagger}\right)  \right\} \label{joel}%
\end{align}
(iii)
\[
\tau_{2,k\ }\left(  \widehat{\theta}\right)  \equiv\widetilde{\tau}_{k}\left(
\widehat{\theta}\right)  +\mathbb{IF}_{1,\widetilde{\tau}_{k}\left(
\cdot\right)  }\left(  \widehat{\theta}\right)  +\mathbb{IF}%
_{2,2,\widetilde{\tau}_{k}\left(  \cdot\right)  }\left(  \widehat{\theta
}\right)  ,
\]
where
\[
\mathbb{IF}_{1,\widetilde{\tau}_{k}\left(  \cdot\right)  }\left(
\widehat{\theta}\right)  =\mathbb{IF}_{1,\widetilde{\tau}_{k}\left(
\cdot\right)  }\left(  \widehat{\theta}\right)  =\mathbb{V}\left\{  v\left(
\widehat{\theta}\right)  ^{-1}\left[  \left\{  Y-\widehat{b}\left(  X\right)
\right\}  \left\{  A-\widehat{p}\left(  X\right)  \right\}  \right]  \right\}
\]
with $Y=Y^{\ast}\left(  \tau\left(  \widehat{\theta}\right)  \right)
,\widehat{b}\left(  X\right)  =\widehat{b}\left(  X,\tau\left(
\widehat{\theta}\right)  \ \right)  ,$

${}${}${}${}$\mathbb{IF}_{2,2,\widetilde{\tau}_{k}\left(  \cdot\right)
}\left(  \widehat{\theta}\right)  =v\left(  \widehat{\theta}\right)
^{-1}\left[  \mathbb{IF}_{2,2,\widetilde{\psi}_{k}\left(  \tau\left(
\widehat{\theta}\right)  ,\cdot\right)  }\left(  \widehat{\theta}\right)
+\mathbb{Q}_{2,2,\widetilde{\tau}_{k}\left(  \cdot\right)  }\left(
\widehat{\theta}\right)  \right]  $ ${}$where
\begin{align*}
&  Q_{2,2,\widetilde{\tau}_{k}\left(  \cdot\right)  ,\overline{i}_{2}%
}\ \left(  \widehat{\theta}\right) \\
&  =-\frac{1}{2}v\left(  \widehat{\theta}\right)  ^{-1}\left[
\begin{array}
[c]{c}%
\left[  \left\{  A-\widehat{p}\left(  X\right)  \right\}  _{i_{1}}%
^{2}-v\left(  \widehat{\theta}\right)  \right]  \left[  \left\{
Y-\widehat{b}\left(  X\right)  \right\}  \left\{  A-\widehat{p}\left(
X\right)  \right\}  \right]  _{i_{2}}+\\
\left[  \left\{  A-\widehat{p}\left(  X\right)  \right\}  _{i_{2}}%
^{2}-v\left(  \widehat{\theta}\right)  \right]  \left[  \left\{
Y-\widehat{b}\left(  X\right)  \right\}  \left\{  A-\widehat{p}\left(
X\right)  \right\}  \right]  _{i_{1}}%
\end{array}
\right]
\end{align*}

\end{theorem}

\begin{proof}
The proof of $\left(  i\right)  $ was given earlier. The proofs of $\left(
ii\right)  $ and (iii) are in the appendix.
\end{proof}

\begin{theorem}
\label{bb}$.$\textbf{\ }Suppose $\widetilde{\tau}_{k}\left(  \theta\right)
=\tau^{\dagger}$ and the assumptions of the preceding theorem hold. {}{}${}
${}{}${}{}{}{}{}{}$Then

(i) $var_{\theta}\left[  \mathbb{U}_{2,2,\widetilde{\tau}_{k}\left(
\cdot\right)  }^{\ast,test,\perp}\left(  \widehat{\theta}\left(  \tau
^{\dagger}\right)  ,\tau^{\dagger}\right)  \right]  =o\left(  \frac{1}%
{n}\right)  ,$%
\begin{align*}
&  var_{\theta}\left[  v\left(  \widehat{\theta}\right)  ^{-1}\psi
_{2,k\ }\left(  \tau,\widehat{\theta}\right)  \right]  \times\left[
var_{\theta}\left\{  var_{\widehat{\theta}\left(  \tau^{\dagger}\right)
}\left\{  \mathbb{ES}_{2,\widetilde{\tau}_{k}\left(  \cdot\right)  }%
^{test}\left(  \widehat{\theta}\left(  \tau^{\dagger}\right)  \right)
\right\}  ^{-1}\mathbb{ES}_{2,\widetilde{\tau}_{k}\left(  \cdot\right)
}^{test}\left(  \widehat{\theta}\left(  \tau^{\dagger}\right)  \right)
\right\}  \right]  ^{-1}\\
&  =1+o_{p}\left(  1\right)
\end{align*}

(ii)
\begin{align*}
var_{\theta}\left[  \mathbb{Q}_{2,2,\widetilde{\tau}_{k}\left(  \cdot\right)
}\ \left(  \widehat{\theta}\right)  \right]   &  =o\left(  \frac{1}{n}\right)
,\\
var_{\theta}\left[  v\left(  \widehat{\theta}\right)  ^{-1}\ \psi
_{2,k\ }\left(  \tau,\widehat{\theta}\right)  \right]  /var_{\theta}\left\{
\tau_{2,k\ }\left(  \widehat{\theta}\right)  -\tau^{\dagger}\right\}   &
=1+o_{p}\left(  1\right)
\end{align*}

$\left(  iii\right)  $%
\begin{align*}
v\left(  \widehat{\theta}\right)  ^{-1}E_{\theta}\left[  \psi_{2,k\ }\left(
\tau^{\dagger},\widehat{\theta}\right)  \right]   &  =O_{p}\left\{  \left(
P-\widehat{P}\right)  \left(  B\left(  \tau^{\dagger}\right)  -\widehat{B}%
\left(  \tau^{\dagger}\right)  \right)  \left(  \frac{g\left(  X\right)
}{\widehat{g}\left(  X\right)  }-1\right)  \right\} \\
&  =O_{p}\left(  n^{-\left(  \frac{\beta_{g}}{2\beta_{g}+d}+\frac{\beta_{b}%
}{d+2\beta_{b}}+\frac{\beta_{p}}{d+2\beta_{p}}\right)  }\right)
\end{align*}

(iv)
\begin{align*}
&  E_{\theta}\left[  var_{\widehat{\theta}\left(  \tau^{\dagger}\right)
}\left\{  \mathbb{ES}_{2,\widetilde{\tau}_{k}\left(  \cdot\right)  }%
^{test}\left(  \widehat{\theta}\left(  \tau^{\dagger}\right)  \right)
\right\}  ^{-1}\ \mathbb{ES}_{2,\widetilde{\tau}_{k}\ }^{test}\left(
\widehat{\theta}\left(  \tau^{\dagger}\right)  \right)  \right] \\
&  =O_{p}\left\{  \left(  P-\widehat{P}\right)  \left(  B\left(  \tau
^{\dagger}\right)  -\widehat{B}\left(  \tau^{\dagger}\right)  \right)  \left[
\left(  \frac{g\left(  X\right)  }{\widehat{g}\left(  X\right)  }-1\right)
+\left(  P-\widehat{P}\right)  +\left(  B\left(  \tau^{\dagger}\right)
-\widehat{B}\left(  \tau^{\dagger}\right)  \right)  \right]  \right\} \\
&  =O_{p}\left[  \max\left\{  n^{-\left(  \frac{\beta_{g}}{2\beta_{g}+d}%
+\frac{\beta_{b}}{d+2\beta_{b}}+\frac{\beta_{p}}{d+2\beta_{p}}\right)
},n^{-\left(  \frac{\beta_{b}}{d+2\beta_{b}}+\frac{2\beta_{p}}{d+2\beta_{p}%
}\right)  },n^{-\left(  \frac{2\beta_{b}}{d+2\beta_{b}}+\frac{\beta_{p}%
}{d+2\beta_{p}}\right)  }\right\}  \right]
\end{align*}

( v)
\begin{align*}
E_{\theta}\left[  \tau_{2,k\ }\left(  \widehat{\theta}\right)  -\tau^{\dagger
}\right]   &  =O_{p}\left\{
\begin{array}
[c]{c}%
\left(  P-\widehat{P}\right)  \left(  \frac{g\left(  X\right)  }%
{\widehat{g}\left(  X\right)  }-1\right)  \left(  B-\widehat{B}\right)  +\\
\left(  P-\widehat{P}\right)  ^{2}\left(  P-\widehat{P}\right)  \left(
B-\widehat{B}\right) \\
+\left(  \frac{g\left(  X\right)  }{\widehat{g}\left(  X\right)  }-1\right)
^{2}\left[  \left(  P-\widehat{P}\right)  +\left(  \frac{g\left(  X\right)
}{\widehat{g}\left(  X\right)  }-1\right)  +\left(  B-\widehat{B}\right)
\right]
\end{array}
\right\} \\
&  =O_{p}\left\{  \max\left\{
\begin{array}
[c]{c}%
n^{-\left(  \frac{\beta_{g}}{2\beta_{g}+d}+\frac{\beta_{p}}{d+2\beta_{p}%
}+\frac{\beta_{b}}{d+2\beta_{b}}\right)  },\\
n^{-\left(  \frac{2\beta_{p}}{d+2\beta_{p}}\right)  }n^{-\frac{\beta_{p}%
}{d+2\beta_{p}}}n^{-\frac{\beta_{b}}{d+2\beta_{b}}},\\
n^{-\frac{2\beta_{g}}{2\beta_{g}+d}}\left\{  n^{-\frac{\beta_{b}}{d+2\beta
_{b}}}+n^{-\frac{\beta_{p}}{d+2\beta_{p}}}++n^{-\frac{\beta_{g}}{2\beta_{g}%
+d}}\right\}
\end{array}
\right\}  \right\}
\end{align*}

\end{theorem}

\begin{proof}
The proof of part (iii) was given earlier. The remaining parts are proved in
the Appendix.
\end{proof}

We conclude from this theorem that the savings in variance that comes with
using $\mathbb{ES}_{2,\widetilde{\tau}_{k}\left(  \cdot\right)  }%
^{test}\left(  \widehat{\theta}\left(  \tau^{\dagger}\right)  \right)
\mathbb{\ }$ rather than $\ \psi_{2,k\ }\left(  \tau,\widehat{\theta}\right)
$ is asymptotically negligible even in regard to constants. {}{}{}{}%
{}Similarly, we conclude that the difference in variance that comes with using
$\psi_{2,k\ }\left(  \tau,\widehat{\theta}\right)  $ rather than
$\ \mathbb{IF}_{2,2,\widetilde{\tau}_{k}\left(  \cdot\right)  }\left(
\widehat{\theta}\right)  $ is asymptotically negligible, again even in regard
to constants. Further, because $var_{\theta}\left[  \mathbb{U}%
_{2,2,\widetilde{\tau}_{k}\left(  \cdot\right)  }^{\ast,test,\perp}\left(
\widehat{\theta}\left(  \tau^{\dagger}\right)  ,\tau^{\dagger}\right)
\right]  $ and $var_{\theta}\left[  \mathbb{Q}_{2,2,\widetilde{\tau}%
_{k}\left(  \cdot\right)  }\ \left(  \widehat{\theta}\right)  \right]  $ are
of the order of $o\left(  \frac{1}{n}\right)  $ as their first order
degenerate kernels are both of order $o_{p}\left(  1\right)  ,$ and
$n^{\frac{4\beta}{4\beta+d}-\sigma}\left\{  \psi_{2,k\ }\left(  \tau
,\widehat{\theta}\right)  -E_{\theta}\left[  \psi_{2,k\ }\left(
\tau,\widehat{\theta}\right)  \right]  \right\}  $ is asymptotically normal,
we conclude that
\begin{align*}
&  n^{\frac{4\beta}{4\beta+d}-\sigma}\left\{  \tau_{2,k\ }\left(
\widehat{\theta}\right)  -E_{\theta}\left[  \tau_{2,k\ }\left(
\widehat{\theta}\right)  \right]  \right\}  ,\\
&  n^{\frac{4\beta}{4\beta+d}-\sigma}\left\{  \mathbb{ES}_{2,\widetilde{\tau
}_{k}\left(  \cdot\right)  }^{test}\left(  \widehat{\theta}\left(
\tau^{\dagger}\right)  \right)  \right\}  ^{-1}\left[  \mathbb{ES}%
_{2,\widetilde{\tau}_{k}\left(  \cdot\right)  }^{test}\left(  \widehat{\theta
}\left(  \tau^{\dagger}\right)  \right)  -E_{\theta}\left[  \mathbb{ES}%
_{2,\widetilde{\tau}_{k}\left(  \cdot\right)  }^{test}\left(  \widehat{\theta
}\left(  \tau^{\dagger}\right)  \right)  \right]  \right]
\end{align*}
and $n^{\frac{4\beta}{4\beta+d}-\sigma}v\left(  \widehat{\theta}\right)
^{-1}\left\{  \psi_{2,k\ }\left(  \tau,\widehat{\theta}\right)  -E_{\theta
}\left[  \psi_{2,k\ }\left(  \tau,\widehat{\theta}\right)  \right]  \right\}
$ are all asymptotically normal with the same asymptotic variance.

It then follows that a necessary condition for the intervals based on
$\psi_{2,k}\left(  \tau^{\dagger},\widehat{\theta}\right)  ,$ $\mathbb{ES}%
_{2,\widetilde{\tau}_{k}\left(  \cdot\right)  }^{test}\left(  \widehat{\theta
}\left(  \tau\right)  \right)  ,$ and $\tau_{2,k}\left(  \widehat{\theta
}\right)  -\tau$ to cover $\widetilde{\tau}_{k}\left(  \theta\right)
=\tau^{\dagger}$ at the nominal $1-\alpha$ level as $n\rightarrow\infty$ is
that
\begin{align*}
&  v\left(  \widehat{\theta}\right)  ^{-1}E_{\theta}\left[  \psi
_{2,k\ }\left(  \tau^{\dagger},\widehat{\theta}\right)  \right]  ,\\
&  var_{\widehat{\theta}\left(  \tau^{\dagger}\right)  }\left\{
\mathbb{ES}_{2,\widetilde{\tau}_{k}\left(  \cdot\right)  }^{test}\left(
\widehat{\theta}\left(  \tau^{\dagger}\right)  \right)  \right\}
^{-1}E_{\theta}\left[  \ \mathbb{ES}_{2,\widetilde{\tau}_{k}\ }^{test}\left(
\widehat{\theta}\left(  \tau^{\dagger}\right)  \right)  \right]
\end{align*}
and $E_{\theta}\left[  \tau_{2,k}\left(  \widehat{\theta}\right)
-\tau^{\dagger}\right]  $ are $O_{p}\left(  n^{-\frac{4\beta}{4\beta+d}%
+\sigma}\right)  .$

Now we know under the assumptions of theorem \ref{bb} that this necessary
condition holds for $v\left(  \widehat{\theta}\right)  ^{-1}E_{\theta}\left[
\psi_{2,k}\left(  \tau^{\dagger},\widehat{\theta}\right)  \right]  $ since
$v\left(  \widehat{\theta}\right)  $ is bounded away from zero and one and, by
assumption, $n^{-\left(  \frac{\beta_{g}}{2\beta_{g}+d}+\frac{\beta_{b}%
}{d+2\beta_{b}}+\frac{\beta_{p}}{d+2\beta_{p}}\right)  }=O_{p}\left(
n^{-\frac{4\beta}{4\beta+d}}\right)  .$ However this necessary condition need
not hold for either
\[
var_{\widehat{\theta}\left(  \tau^{\dagger}\right)  }\left\{  \mathbb{ES}%
_{2,\widetilde{\tau}_{k}\left(  \cdot\right)  }^{test}\left(  \widehat{\theta
}\left(  \tau^{\dagger}\right)  \right)  \right\}  ^{-1}E_{\theta}\left[
\ \mathbb{ES}_{2,\widetilde{\tau}_{k}\ }^{test}\left(  \widehat{\theta}\left(
\tau^{\dagger}\right)  \right)  \right]
\]
or $E_{\theta}\left[  \tau_{2,k}\left(  \widehat{\theta}\right)
-\tau^{\dagger}\right]  .$ For example, consider the following specification
consistent with our assumptions: $\beta_{p}/d=0$ , $\beta_{b}/d=\beta
_{g}/d=1/4.$ Then $\beta/d=1/8,$ so $n^{-\left(  \frac{\beta_{g}}{2\beta
_{g}+d}+\frac{\beta_{b}}{d+2\beta_{b}}+\frac{\beta_{p}}{d+2\beta_{p}}\right)
}=n^{-\frac{4\beta}{4\beta+d}}=n^{-1/3}.$ However, $E_{\theta}\left[
\tau_{2,k}\left(  \widehat{\theta}\right)  -\tau^{\dagger}\right]
\ $converges to zero at rate $n^{-\ \frac{\beta_{b}}{d+2\beta_{b}}}%
=n^{-\frac{1}{6}}$. Next
\begin{align*}
&  var_{\widehat{\theta}\left(  \tau^{\dagger}\right)  }\left\{
\mathbb{ES}_{2,\widetilde{\tau}_{k}\left(  \cdot\right)  }^{test}\left(
\widehat{\theta}\left(  \tau^{\dagger}\right)  \right)  \right\}
^{-1}E_{\theta}\left[  \ \mathbb{ES}_{2,\widetilde{\tau}_{k}\ }^{test}\left(
\widehat{\theta}\left(  \tau^{\dagger}\right)  \right)  \right] \\
&  =O_{p}\left(  n^{-\left(  \frac{\beta_{b}}{d+2\beta_{b}}+\frac{2\beta_{p}%
}{d+2\beta_{p}}\right)  }\right)  =n^{-1/6}>>O_{p}\left(  n^{-\frac{4\beta
}{4\beta+d}+\sigma}\right)  =n^{-1/3+\sigma}%
\end{align*}
for small $\sigma.$ We conclude that the intervals based on $\mathbb{ES}%
_{2,\widetilde{\tau}_{k}\left(  \cdot\right)  }^{test}\left(  \widehat{\theta
}\left(  \tau\right)  \right)  \ $ and $\tau_{2,k\ }\left(  \widehat{\theta
}\right)  -\tau$ fail to cover $\widetilde{\tau}_{k}\left(  \theta\right)
=\tau^{\dagger}$ at the nominal $1-\alpha$ level uniformly over $\Theta$ as
$n\rightarrow\infty$ . We reach the identical conclusion with regard to the
parameter $\tau\left(  \theta\right)  $ because under our assumptions
$\left\vert \tau\left(  \theta\right)  -\widetilde{\tau}_{k}\left(
\theta\right)  \right\vert =O_{p}\left(  n^{-\frac{4\beta}{4\beta+d}+\sigma
}\right)  $

Furthermore, by the argument used in the proof of theorem \ref{tt}, it is easy
to see that the length of each interval is $O_{p}\left(  k/n^{2}\right)
=O_{p}\left(  n^{-\frac{4\beta}{4\beta+d}+\sigma}\right)  .$ It follows that
if we try to improve the coverage of the intervals based on $\mathbb{ES}%
_{2,\widetilde{\tau}_{k}\left(  \cdot\right)  }^{test}\left(  \widehat{\theta
}\left(  \tau\right)  \right)  \ $ and $\tau_{2,k}\left(  \widehat{\theta
}\right)  -\tau$ by further increasing $k,$ the length of the intervals will
increase beyond $O_{p}\left(  n^{-\frac{4\beta}{4\beta+d}+\sigma}\right)  .$
We conclude that the interval based on $\psi_{2,k}\left(  \tau,\widehat{\theta
}\right)  $ is strictly preferred to the other two intervals when $\beta
_{p}/d=0$ , $\beta_{b}/d=\beta_{g}/d=1/4$ and is never worse in terms of
shrinkage rate and coverage than the other two intervals whatever be
$\beta_{p}$, $\beta_{b},$ and $\beta_{g}\ $. We reach the identical conclusion
with regard to the coverage of the parameter $\tau\left(  \theta\right)  $
because, under our assumptions including our choice of $k$, $\left\vert
\tau\left(  \theta\right)  -\widetilde{\tau}_{k}\left(  \theta\right)
\right\vert =O_{p}\left(  n^{-\frac{4\beta}{4\beta+d}\ }\right)  $ and
$n^{-\frac{4\beta}{4\beta+d}\ }<<n^{-\frac{4\beta}{4\beta+d}+\sigma},$ the
order of the interval lengths.

These results translate directly into analogous results concerning the
associated estimators. Under our assumptions the estimator solving
$\psi_{2,k\ }\left(  \tau,\widehat{\theta}\right)  =0$ converges to both
$\tau\left(  \theta\right)  $ and $\widetilde{\tau}_{k}\left(  \theta\right)
$ at rate $O_{p}\left(  n^{-\frac{4\beta}{4\beta+d}+\sigma}\right)  .$ In
contrast the rate of convergence of $\tau_{2,k}\left(  \widehat{\theta
}\right)  $ and the estimator solving $\mathbb{ES}_{2,\widetilde{\tau}%
_{k}\left(  \cdot\right)  }^{test}\left(  \widehat{\theta}\left(  \tau\right)
\right)  =0$ converge to $\tau\left(  \theta\right)  $ and $\widetilde{\tau
}_{k}\left(  \theta\right)  $ at the rates given in (iv) and (v) of theorem
\ref{bb}.

What is the intuition behind the above findings? First note that, as promised
by Theorem \ref{eiet} and part (vii) of the theorem in the last subsection,
the bias away from zero of $var_{\widehat{\theta}\left(  \tau^{\dagger
}\right)  }\left\{  \mathbb{ES}_{2,\widetilde{\tau}_{k}\left(  \cdot\right)
}^{test}\left(  \widehat{\theta}\left(  \tau^{\dagger}\right)  \right)
\right\}  ^{-1}E_{\theta}\left[  \ \mathbb{ES}_{2,\widetilde{\tau}_{k}%
\ }^{test}\left(  \widehat{\theta}\left(  \tau^{\dagger}\right)  \right)
\right]  ,E_{\theta}\left[  \tau_{2,k}\left(  \widehat{\theta}\right)
-\widetilde{\tau}_{k}\left(  \theta\right)  \right]  ,$ and $v\left(
\widehat{\theta}\right)  ^{-1}E_{\theta}\left[  \psi_{2,k}\left(
\tau^{\dagger},\widehat{\theta}\right)  \right]  $ are all $O_{p}\left(
\left\vert \left\vert \widehat{\theta}-\theta\ \right\vert \right\vert
^{3}\right)  .$ However the nature and convergence rate of the $O_{p}\left(
\left\vert \left\vert \widehat{\theta}-\theta\ \right\vert \right\vert
^{3}\right)  $ term can vary markedly between estimators, attaining a minimum
for $E_{\theta}\left[  \psi_{2,k}\left(  \tau^{\dagger},\widehat{\theta
}\right)  \right]  .$ Now it is not surprising that, for the same order of
variance, the order of $E_{\theta}\left[  \tau_{2,k}\left(  \widehat{\theta
}\right)  -\widetilde{\tau}_{k}\left(  \theta\right)  \right]  $ often exceeds
that of $E_{\theta}\left[  \psi_{2,k}\left(  \tau^{\dagger},\widehat{\theta
}\right)  \right]  .$ Confidence intervals for $\widetilde{\tau}_{k}\left(
\theta\right)  $ based on $\tau_{2,k}\left(  \widehat{\theta}\right)  $ are
centered at (i.e are symmetric around) $\tau_{2,k}\left(  \widehat{\theta
}\right)  ,$ which is a quite stringent constraint on the form of the
interval. In that sense, intervals based on $\tau_{2,k}\left(  \widehat{\theta
}\right)  $ are a higher order generalization of the first order asymptotic
Wald intervals for $\widetilde{\tau}_{k}\left(  \theta\right)  .$ It is well
known that when $\widetilde{\tau}_{k}\left(  \theta\right)  $ is an implicit
parameter that sets a functional such as $\widetilde{\psi}_{k}\left(
\tau,\theta\right)  $ to zero, first-order Wald confidence intervals are often
outperformed in finite samples by confidence sets obtained by inverting a
'score-like' test based on first order 'estimating functions' for the
functional that depend on the parameter $\widetilde{\tau}_{k}$ and,
frequently, on estimated nuisance parameters as well, although this fact is
not reflected in the first order asymptotics. Our example is higher order
version of this phenomenon, where the benefit of the interval $C_{1-\alpha
,\widetilde{\psi}_{\ k}\left(  \tau\right)  }$ obtained by inverting tests
based on the estimating function $\psi_{2,k}\left(  \tau,\widehat{\theta
}\right)  $ for the functional $\widetilde{\psi}_{k}\left(  \tau
,\theta\right)  $ is clearly and quantitatively revealed by the asymptotics.
Note that, like first order Wald intervals, the interval based on $\tau
_{2,k}\left(  \widehat{\theta}\right)  $ will differ from the interval for
$\widetilde{\tau}_{k}\left(  \theta\right)  $ based on applying an inverse
nonlinear monotone transform $h^{-1}\left(  \cdot\right)  $ to the end points
of a Wald interval for the transformed parameter $h\left\{  \widetilde{\tau
}_{k}\left(  \theta\right)  \right\}  $ that is centered on $h\left(
\tau\right)  _{2,k}\left(  \widehat{\theta}\right)  \equiv h\left(
\widetilde{\tau}_{k}\left(  \widehat{\theta}\right)  \right)  +\mathbb{IF}%
_{2,h\left(  \widetilde{\tau}_{k}\left(  \cdot\right)  \right)  }\left(
\widehat{\theta}\right)  .$ In contrast, like first order score-based
intervals, the intervals based on $\psi_{2,k}\left(  \tau,\widehat{\theta
}\right)  $ and $\mathbb{ES}_{2,\widetilde{\tau}_{k}\left(  \cdot\right)
}^{test}\left(  \widehat{\theta}\left(  \tau^{\dagger}\right)  \right)  $ are
invariant to monotone transformations of the parameter $\widetilde{\tau}%
_{k}\left(  \theta\right)  $.

More interesting and perhaps more surprising is that, for the same order of
variance, the order of $E_{\theta}\left[  var_{\widehat{\theta}\left(
\tau^{\dagger}\right)  }\left\{  \mathbb{ES}_{2,\widetilde{\tau}_{k}\left(
\cdot\right)  }^{test}\left(  \widehat{\theta}\left(  \tau^{\dagger}\right)
\right)  \right\}  ^{-1}\ \mathbb{ES}_{2,\widetilde{\tau}_{k}\ }^{test}\left(
\widehat{\theta}\left(  \tau^{\dagger}\right)  \right)  \right]  $ exceeds
that of $E_{\theta}\left[  \psi_{2,k}\left(  \tau^{\dagger},\widehat{\theta
}\right)  \right]  .$ The surprise derives from a failure to recognize that
the theorem \ref{gg} is simply too general to help select among competing
procedures . For example, this theorem implies that under law $\widehat{\theta
}\left(  \tau^{\dagger}\right)  ,$ (a) the variance $var_{\widehat{\theta
}\left(  \tau^{\dagger}\right)  }\left\{  \mathbb{ES}_{2,\widetilde{\tau}%
_{k}\left(  \cdot\right)  }^{test}\left(  \widehat{\theta}\left(
\tau^{\dagger}\right)  \right)  \right\}  ^{-1}$ of $\left[
var_{\widehat{\theta}\left(  \tau^{\dagger}\right)  }\left\{  \mathbb{ES}%
_{2,\widetilde{\tau}_{k}\left(  \cdot\right)  }^{test}\left(  \widehat{\theta
}\left(  \tau^{\dagger}\right)  \right)  \right\}  ^{-1}E_{\theta}\left[
\ \mathbb{ES}_{2,\widetilde{\tau}_{k}\ }^{test}\left(  \widehat{\theta}\left(
\tau^{\dagger}\right)  \right)  \right]  \right]  $ is less (and generally
strictly less ) than the variance of $v\left(  \widehat{\theta}\left(
\tau^{\dagger}\right)  \right)  ^{-1}E_{\theta}\left[  \psi_{2,k}\left(
\tau^{\dagger},\widehat{\theta}\left(  \tau^{\dagger}\right)  \right)
\right]  , $ while (b) both have bias of $O_{p}\left(  \left\vert \left\vert
\widehat{\theta}\left(  \tau^{\dagger}\right)  -\theta\ \right\vert
\right\vert ^{3}\right)  .$ At first blush, this might suggest that the
estimator solving $\mathbb{ES}_{2,\widetilde{\tau}_{k}\ }^{test}\left(
\widehat{\theta}\left(  \tau\right)  \right)  =0$ would likely have the same
bias but smaller variance than the estimator solving $\psi_{2,k}\left(
\tau,\widehat{\theta}\right)  =0. $ But we have seen that just the opposite is
true. The reason is that the difference between the variances in (a) is
negligible in the sense that their ratio is $1+o_{p}\left(  1\right)  $, while
the $O_{p}\left(  \left\vert \left\vert \widehat{\theta}\left(  \tau^{\dagger
}\right)  -\theta\ \right\vert \right\vert ^{3}\right)  $ biases are often of
quite different orders with that of $v\left(  \widehat{\theta}\left(
\tau^{\dagger}\right)  \right)  ^{-1}E_{\theta}\left[  \psi_{2,k}\left(
\tau^{\dagger},\widehat{\theta}\left(  \tau^{\dagger}\right)  \right)
\right]  $ always a minimum.

More generally, whenever the functional $\psi\left(  \tau,\theta\right)  $\ is
in our doubly robust class, Eq \ref{NE11} holds so $\widehat{\psi
}_{\mathcal{K}_{J}}^{eff}$\ is rate minimax (or near minimax if $\sigma$ is
chosen positive), and the suppositions of Theorem \ref{J1} hold for
$\widetilde{\psi}\left(  \tau\right)  =$\ $\widehat{\psi}_{\mathcal{K}_{J}%
}^{eff}\left(  \tau\right)  ,$\ Theorem \ref{J1} then implies the width of the
interval estimator for $\tau\left(  \theta\right)  $\ based on $\widehat{\psi
}_{\mathcal{K}_{J}}^{eff}\left(  \tau\right)  $\ converges to zero at the
convergence rate of $\widehat{\psi}_{\mathcal{K}_{J}}^{eff}\left(
\tau\right)  $\ to $\psi\left(  \tau,\theta\right)  .$

\section{Monotone Missing Data and Other Complex Functionals}

\subsection{Derivation of Higher Order Influence Functions for product of
functionals}

In this section, we discuss the construction of higher order influence
functions for a more general class of functionals than the one thus far
considered. {}An important application of this construction is in the
derivation of higher order influence functions in monotone missing data
problems. \ To begin, we must learn to construct higher order influence
functions for a functional with the product form:%
\[
\psi\left(  \theta;\zeta\right)  =\prod\limits_{s=1}^{\zeta}\psi_{s}\left(
\theta\right)
\]
here $\psi_{s}\left(  \theta\right)  ,$ $s=1,\ldots,\zeta,$ are known to be
higher order pathwise differentiable functionals and $\zeta$ is a known constant.

The following lemma gives the general form of higher order influence functions
of $\psi\left(  \theta;\zeta\right)  $ as a function of influence functions of
$\left\{  \psi_{s}\left(  \theta\right)  :s=1,\ldots,\zeta\right\}  .$
\ Before stating our lemma, we need additional notation. \ Given an ordered
set of $m$ positive integers $\overline{\mathbf{i}}_{m}=\left\{  i_{1}%
,i_{2},...,i_{m}\right\}  ,$ for any $r$ non-negative integers $\left\{
t_{1},t_{2},...,t_{r}\right\}  $, satisfying $\sum_{s\leq r}t_{s}=m,$ we
define $\Upsilon=\left(  \overline{i}_{1},_{t_{1}},\overline{i}_{2},_{t_{2}%
},...,\overline{i}_{r},_{t_{r}}\right)  $ to be an ordered partition of degree
$m$ and order $r$ if and only if for $1\leq s^{\ast}\leq r$ and $t_{s^{\ast}%
}>0$ we have$:$
\[
\overline{i}_{s^{\ast}},_{t_{s^{\ast}}}=\left\{
\begin{array}
[c]{c}%
i_{j(s^{\ast})+1},i_{j(s^{\ast})+2},...,i_{j(s^{\ast})+t_{s^{\ast}}}:\\
j(s^{\ast})=\\
\left\{
\begin{array}
[c]{c}%
\sum\limits_{q=1}^{s^{\ast}-1}t_{q}\text{ \ \ \ \ \ for }1<s^{\ast}\leq r\\
0\text{\ \ \ \ \ \ \ \ \ for \ \ }s^{\ast}=1
\end{array}
\right.
\end{array}
\right\}
\]
and for all $t_{s^{\ast}}=0$ we have $\overline{i}_{s^{\ast}},_{t_{s^{\ast}}%
}=\varnothing.$ \ Any such ordered partition satisfies%
\[
\overline{\mathbf{i}}_{m}=\bigcup\limits_{s=1}^{r}\overline{i}_{s},_{t_{s}}%
\]

\begin{lemma}
\label{HOIPROD}Let $if_{\psi_{s}\left(  \theta\right)  ;j,j}\left(  o_{i_{1}%
},...,o_{i_{j}};\theta\right)  =IF_{\psi_{s}\left(  \theta\right)
;j,j;\overline{\mathbf{i}}_{j}}\left(  \theta\right)  $ for $j\geq1$ be the
$jth$ order influence function of $\psi_{s}\left(  \theta\right)  ,$ and
define $IF_{\psi_{s}\left(  \theta\right)  ;0,0;\overline{\mathbf{i}}_{0}%
}\left(  \theta\right)  \equiv\psi_{s}\left(  \theta\right)  .$ Then the $mth$
order influence function of $\psi\left(  \theta;\zeta\right)  $ is given by:%
\begin{align*}
\mathbb{IF}_{\psi\left(  \theta;\zeta\right)  ,m}\left(  \theta\right)   &  =%
%TCIMACRO{\dsum \limits_{j=1}^{m}}%
%BeginExpansion
{\displaystyle\sum\limits_{j=1}^{m}}
%EndExpansion
\mathbb{IF}_{\psi\left(  \theta;\zeta\right)  ,j,j}\left(  \theta\right) \\
&  =%
%TCIMACRO{\dsum \limits_{j=1}^{m}}%
%BeginExpansion
{\displaystyle\sum\limits_{j=1}^{m}}
%EndExpansion
\mathbb{V}\left[  IF_{\psi\left(  \theta;\zeta\right)  ,j,j,\overline
{\mathbf{i}}_{j}}\left(  \theta\right)  \right]
\end{align*}
where
\[
\mathbb{V}\left[  IF_{\psi\left(  \theta;\zeta\right)  ,j,j,\overline
{\mathbf{i}}_{j}}\left(  \theta\right)  \right]  =\mathbb{V}\left[
%TCIMACRO{\dsum \limits_{\left\{  t_{1},...t_{\zeta}\right\}  \in
%\Upsilon_{\zeta;j}}}%
%BeginExpansion
{\displaystyle\sum\limits_{\left\{  t_{1},...t_{\zeta}\right\}  \in
\Upsilon_{\zeta;j}}}
%EndExpansion%
%TCIMACRO{\dprod \limits_{s=1}^{\zeta}}%
%BeginExpansion
{\displaystyle\prod\limits_{s=1}^{\zeta}}
%EndExpansion
IF_{\psi_{s}\left(  \theta\right)  ;t_{s},t_{s};\overline{i}_{s,t_{s}}}\left(
\theta\right)  \right]
\]%
\[
\Upsilon_{\zeta;j}=\left\{  t_{1},...t_{\zeta}:\sum\limits_{s=1}^{\zeta}%
t_{s}=j,t_{s}\geq0\right\}
\]

\end{lemma}

It is easy to generalize this lemma to functionals of the more general form
\[
\psi\left(  \theta\right)  =\psi\left(  \theta;\zeta_{1},\zeta_{2}\right)
=\sum_{1\leq v_{2}\leq\zeta_{1}}\prod\limits_{s=1}^{\zeta_{2}}\psi_{v_{2}%
,s}\left(  \theta\right)
\]
where both $\zeta_{1},\zeta_{2}$ are known constants, \ since the higher order
influence function of a linear combination of functionals is the linear
combination of the influence functions of the functionals.

\subsection{Application to Monotone Missing Data}

\subsubsection{Mapping Higher Order Full Data IFs to Observed Data IFs}

We next turn to the analysis of missing data models. {}Suppose that we have
derived the $mth$ order influence function $\mathbb{IF}_{\widetilde{\psi}%
_{k},m}^{full}\left(  \theta\right)  $ $m\geq2$ for a truncated parameter
$\widetilde{\psi}_{k}$ of a parameter $\psi\left(  \theta\right)  $ in our
doubly robust class of functionals based on i.i.d full data $L_{full}$; then,
according to theorem \ref{HOIPROD} the estimated $\mathbb{IF}_{\widetilde{\psi
}_{k};j,j}^{full}\left(  \widehat{\theta}\right)  ,$ $j\leq m,$ is of the
form
\begin{align*}
\sum_{1\leq l\leq k^{\left(  j-1\right)  }}\mathbb{V}\left[  \prod
\limits_{s=1}^{j}\widehat{U}_{l,s,i_{s}}^{(j)}\right]   &  =\sum_{1\leq l\leq
k^{\left(  j-1\right)  }}\left\{  \mathbb{V}\left(  \prod\limits_{s=1}%
^{j}u_{l,s}^{(j)}\left(  L_{i_{s}}^{full};\widehat{\theta}\right)  \right)
\right\} \\
&  =\sum_{1\leq l\leq k^{\left(  j-1\right)  }}\mathbb{V}\left[
\prod\limits_{s=1}^{j}\widehat{U}_{l,s,i_{s}}^{(j)}\right]
\end{align*}
where $u_{l,s}^{(j)}\left(  L_{i_{s}}^{full};\widehat{\theta}\right)  \equiv$
$\widehat{U}_{l,s,i_{s}}^{(j)}$ depends of one subject's data and has a
functional form which depends on the form of $H\left(  \cdot,\cdot\right)  $
corresponding to the functional of interest. \ For example, $\widehat{U}%
_{1,1,i_{1}}^{(2)}=\left[  \left(  A-\widehat{P}\right)  Z_{1}\right]
_{i_{1}}$corresponds to the first component of the vector contributed by
person $i_{1} $ to $\mathbb{IF}_{\widetilde{\psi}_{k};2,2}^{full}\left(
\widehat{\theta}\right)  $. Note that $E_{\widehat{\theta}}\left[
\widehat{U}_{l,s}^{(j)}\right]  =0.$\ \ \ Next, suppose that for a subject
some of the data may be missing, so that we observe
\[
\left\{  \left(  L_{obs,i}=R_{i}L_{full;i}+(1-R_{i})w(L_{full;i}%
),R_{i}\right)  :i=1,...,n\right\}
\]
where $R=I\left[  L_{obs,i}=L_{full,i}\right]  $ instead of the full data
$L_{full,i}.$ For now, $W_{i}=w(L_{full;i})$ is a known function of the full
data and $W_{i}$ is always observed. {}{}{}Below we shall extend our results
to monotone missing data.

Define
\begin{align*}
\pi\left(  W;\theta\right)   &  =P\left(  R=1|L^{full};\theta\right) \\
B_{l,s}^{(j)}\left(  W\right)   &  =E\left(  \widehat{U}_{l,s}^{(j)}%
|W;\theta\right)  .
\end{align*}
and we suppose $\pi\left(  W;\theta\right)  >\sigma>0.$If\ we know $\pi\left(
W;\theta\right)  ,$ we can use
\[
\mathbb{V}\left[  \prod\limits_{s=1}^{j}\left(  \frac{R_{i_{s}}\widehat{U}%
_{l,s,i_{s}}^{(j)}}{\pi\left(  L_{obs,,i_{s}};\theta\right)  }+\phi\left(
L_{obs,,i_{s}};\theta\right)  \right)  \right]
\]
instead of $\mathbb{V}\left[  \prod\limits_{s=1}^{j}\widehat{U}_{l,s,i_{s}%
}^{(j)}\right]  $, where $\phi\left(  \cdot\right)  $ is an arbitrary function
with finite variance which satisfies $E_{\theta}\left[  \phi\left(
L_{obs}\right)  |L_{full}\right]  =0.$ \ It is easy to verify that this
statistic is a function of the observed data and an unbiased estimator of%
\[
E_{\theta}\left[  \mathbb{V}\left[  \prod\limits_{s=1}^{j}\widehat{U}%
_{l,s,i_{s}}^{(j)}\right]  \right]
\]
so that in fact, in terms of rates, our observed data statistic has the same
bias and variance properties as the original full data higher order influence
function. \ Moreover, the most efficient (in terms of constants) choice for
$\phi\left(  L_{obs};\theta\right)  $ in our class of estimators is $-\left(
\frac{R}{\pi\left(  L_{obs}\right)  }-1\right)  B_{l,s}^{\left(  j\right)
}\left(  L_{obs};\theta\right)  ,$ while $B_{l,s}^{\left(  j\right)  }\left(
L_{obs};\theta\right)  $ is an unknown function to be estimated from the
observed data. \ This last two observations motivate our proposed strategy for
mapping higher order influence functions in the full data model to higher
order influence functions in the observed data model when the missingness
mechanism is unknown. \ First, define
\begin{align*}
\prod\limits_{s=1}^{j}\tau_{l,s}^{(j)}\left(  \theta\right)  \equiv &
E_{\theta}\left[  \mathbb{V}\left[  \prod\limits_{s=1}^{j}\widehat{U}%
_{l,s,i_{s}}^{(j)}\right]  \right] \\
&  =E_{\theta}\left[  \mathbb{V}\left[  \prod\limits_{s=1}^{j}\left(
\frac{R_{i_{s}}\left[  \widehat{U}_{l,s,i_{s}}^{(j)}-B_{l,s}^{(j)}\left(
L_{obs,i_{s}};\theta\right)  \right]  }{\pi\left(  L_{obs,,i_{s}}%
;\theta\right)  }+B_{l,s}^{(j)}\left(  L_{obs,i_{s}};\theta\right)  \right)
\right]  \right]
\end{align*}
which we notice is of the product form that was discussed earlier in this
section. \ In order to construct an $mth$ order influence function for
$\widetilde{\psi}_{k},$ it is now apparent that we must first succeed in
constructing $mth$ order influence functions for the set of parameters
$\left\{  \prod\limits_{s=1}^{j}\tau_{l,s}^{(j)}\left(  \theta\right)  :j\leq
m,l\leq k^{\left(  j-1\right)  }\right\}  $ so as to guarantee an estimation
bias of order $m+1.$ \ \ Now, for each $l$ and $s$, $\tau_{l,s}^{(j)}\left(
\theta\right)  $ is itself a member of our general doubly robust class so that
$\tau_{l,s}^{(j)}\left(  \theta\right)  =E_{\theta}\left[  H^{(j)}%
(P_{l,s},B_{l,s}^{(j)})\right]  $ where $H_{l,s,1}=-R,H_{l,s,2}=1,$
$H_{l,s,3}=R\widehat{U}_{l,s},H_{l,s,4}=0,$and $P_{l,s}=\pi\left(
L_{obs,,};\theta\right)  ^{-1},B_{l,s}^{(j)}=B_{l,s}\left(  L_{obs,}%
;\theta\right)  $. \ Thus, it immediately follows that neither $\left\{
\prod\limits_{s=1}^{j}\tau_{l,s}^{(j)}\left(  \theta\right)  :j\leq m,l\leq
k^{\left(  j-1\right)  }\right\}  $ nor $\widetilde{\psi}_{k}$ have higher
order influence functions. {}We proceed by constructing influence functions
for the set of truncated parameters $\left\{  \prod\limits_{s=1}%
^{j}\widetilde{\tau}_{l,s}^{(j)}\left(  \theta\right)  :j\leq m,l\leq
k^{\left(  j-1\right)  }\right\}  $ where $\widetilde{\tau}_{l,s}^{(j)}\left(
\theta\right)  $ are appropriately truncated versions of $\tau_{l,s}%
^{(j)}\left(  \theta\right)  $ as in section 3.2.2. \ We then define the
truncated parameter $\widetilde{\widetilde{\psi}}_{k,m}$
\[
\widetilde{\widetilde{\psi}}_{k,m}\left(  \theta\right)  =\psi\left(
\widehat{\theta}\right)  +\sum_{1\leq v\leq m}\sum_{1\leq l\leq k^{\left(
v-1\right)  }}\prod\limits_{s=1}^{v}\widetilde{\tau}_{l,s}^{(v)}\left(
\theta\right)
\]
Note that $\widetilde{\widetilde{\psi}}_{k,m}\left(  \theta\right)  $ differs
from the truncated parameters $\widetilde{\psi}_{k}\left(  \theta\right)  $ of
section 3.2.2, in that it depends on the order of the influence function we
plan to base our inference on. The next theorem gives the $mth$ order
influence function of $\widetilde{\widetilde{\psi}}_{k,m}\left(
\theta\right)  .$

\begin{theorem}
\label{map}Let $IF_{\widetilde{\tau}_{l,s}^{(j)}\left(  \theta\right)
;0,0;\overline{i}_{s,0}}\left(  \theta\right)  \equiv\widetilde{\tau}%
_{l,s}^{(j)}\left(  \theta\right)  ,$The $mth$ order influence function of
$\widetilde{\widetilde{\psi}}_{k,m}\left(  \theta\right)  $ is given by%
\[
\mathbb{IF}_{\widetilde{\widetilde{\psi}}_{k,m}\left(  \theta\right)
,m}\left(  \theta\right)  =%
%TCIMACRO{\dsum \limits_{j=1}^{m}}%
%BeginExpansion
{\displaystyle\sum\limits_{j=1}^{m}}
%EndExpansion
\mathbb{V}\left[  IF_{\widetilde{\widetilde{\psi}}_{k,m}\left(  \theta\right)
,j,j,\overline{\mathbf{i}}_{j}}\left(  \theta\right)  \right]
\]
where
\begin{equation}
IF_{\widetilde{\widetilde{\psi}}_{k,m}\left(  \theta\right)  ,j,j,\overline
{\mathbf{i}}_{j}}\left(  \theta\right)  =\sum_{1\leq v\leq j}\sum_{1\leq l\leq
k^{\left(  v-1\right)  }}%
%TCIMACRO{\dsum \limits_{\left\{  t_{1},...t_{v}\right\}  \in\Upsilon_{v;j}}}%
%BeginExpansion
{\displaystyle\sum\limits_{\left\{  t_{1},...t_{v}\right\}  \in\Upsilon_{v;j}%
}}
%EndExpansion%
%TCIMACRO{\dprod \limits_{s=1}^{v}}%
%BeginExpansion
{\displaystyle\prod\limits_{s=1}^{v}}
%EndExpansion
IF_{\widetilde{\tau}_{l,s}\left(  \theta\right)  ;t_{s},t_{s};\overline
{i}_{s,t_{s}}}\left(  \theta\right) \label{sum}%
\end{equation}%
\[
\Upsilon_{v;j}=\left\{  t_{1},...t_{v}:\sum\limits_{s=1}^{v}t_{s}=j,t_{s}%
\geq0\right\}
\]

\end{theorem}

\begin{proof}
The proof follows directly from the lemma \ref{HOIPROD} applied to each
element
\[
\left\{  \prod\limits_{s=1}^{v}\widetilde{\tau}_{l,s}^{(v)}\left(
\theta\right)  :1\leq v\leq j,1\leq l\leq k^{\left(  v-1\right)  },1\leq s\leq
v\right\}
\]

\end{proof}

\begin{corollary}
The $mth$ order estimated influence function of $\widetilde{\widetilde{\psi}%
}_{k,m}\left(  \theta\right)  $ is given by%
\[
\mathbb{IF}_{\widetilde{\widetilde{\psi}}_{k,m}\left(  \theta\right)
,m}\left(  \widehat{\theta}\right)  =%
%TCIMACRO{\dsum \limits_{j=1}^{m}}%
%BeginExpansion
{\displaystyle\sum\limits_{j=1}^{m}}
%EndExpansion
\mathbb{V}\left[  IF_{\widetilde{\widetilde{\psi}}_{k,m}\left(  \theta\right)
,j,j,\overline{\mathbf{i}}_{j}}\left(  \widehat{\theta}\right)  \right]
\]
where
\[
IF_{\widetilde{\widetilde{\psi}}_{k,m}\left(  \theta\right)  ,j,j,\overline
{\mathbf{i}}_{j}}\left(  \widehat{\theta}\right)  =\left[  \sum_{1\leq v\leq
m}\sum_{1\leq l\leq k^{\left(  v-1\right)  }}%
%TCIMACRO{\dsum \limits_{\left\{  t_{1},...t_{v}\right\}  \in\Upsilon_{v;j}%
%^{+}}}%
%BeginExpansion
{\displaystyle\sum\limits_{\left\{  t_{1},...t_{v}\right\}  \in\Upsilon
_{v;j}^{+}}}
%EndExpansion%
%TCIMACRO{\dprod \limits_{s=1}^{v}}%
%BeginExpansion
{\displaystyle\prod\limits_{s=1}^{v}}
%EndExpansion
IF_{\widetilde{\tau}_{l,s}\left(  \theta\right)  ;t_{s},t_{s};\overline
{i}_{s,t_{s}}}\left(  \widehat{\theta}\right)  \right]
\]%
\[
\Upsilon_{v;j}^{+}=\left\{  t_{1},...t_{v}:\sum\limits_{s=1}^{v}t_{s}%
=j,t_{s}>0\right\}
\]

\end{corollary}

\begin{proof}
This result follows immediately from theorem \ref{map} and the fact that
\[
IF_{\widetilde{\tau}_{l,s}\left(  \theta\right)  ;0,0;\overline{i}_{s,0}%
}\left(  \widehat{\theta}\right)  \equiv\widetilde{\tau}_{l,s}\left(
\widehat{\theta}\right)  =0
\]
by definition.
\end{proof}

\subsubsection{Two-occasion Monotone Missing Data}

We are now ready to use our theorem to derive higher order influence functions
for a functional $\psi\left(  \theta\right)  $ from monotone missing data.
\ \ We begin with the simple two-occasion case. Let $D_{i}=$ $\left(
L_{0,i},L_{1,i},Y_{i}\right)  \sim F\left(  D_{i},\vartheta\right)  ,$
$i=1,...n,$ be $n$ $i.i.d$ copies constitute the full data. \ The outcome of
interest, $Y_{i},$ is univariate. $\ L_{0,i}$ and $L_{1,i}$ are vectors of
continuous covariates with dimensions $d_{0}$ and $\left(  d_{1}-d_{0}\right)
$ respectively. \ The observed data $O_{i}=\left(  R_{0,i},L_{0,i}%
,R_{0,i}L_{1,i},R_{1,i},R_{0,i}R_{1,i}Y_{i}\right)  \sim F\left(  O_{i}%
,\theta=\left(  \vartheta,\gamma\right)  \right)  ,$ where both $R_{0,i}$ and
$R_{1,i}$ are binary missing indicators. The outcome $Y$ \ is observed if and
only if $R_{0}$ $=R_{1}$ = $1,$ while $L_{1}$ is observed if and only if
$R_{0}$ $=1$. \ The data is assumed to be missing at random, that is:
\[
\Pr\left(  R_{0}=1|D\right)  =\Pr\left(  R_{0}=1|L_{0}\right)  =\pi_{0}\left(
L_{0};\theta\right)
\]
and
\begin{align*}
\Pr\left(  R_{1}=1|R_{0}=1,D\right)   &  =\Pr\left(  R_{1}=1|R_{0}%
=1,L_{0},L_{1}\right) \\
&  =\pi_{1}\left(  L_{0},L_{1};\theta\right)
\end{align*}
\ Under this monotone missing-data pattern, the parameter of interest is given
by
\[
\psi\left(  \theta\right)  =E_{\theta}\left(  Y\right)  =E_{\theta}\left(
\frac{R_{0}R_{1}}{\pi_{0}\left(  L_{0};\theta\right)  \pi_{1}\left(
L_{0},L_{1};\theta\right)  }Y\right)
\]

We impose the following assumptions: $\left(  a1\right)  .$ $B_{0}%
=b_{0}\left(  L_{0};\theta\right)  =E_{\theta}\left[  Y|L_{0}\right]  \in
H\left(  \beta_{b_{0}},C_{B_{0}}\right)  ;$ $\left(  a2\right)  .$
$B_{1}=b_{1}\left(  L_{0},L_{1};\theta\right)  =E_{\theta}\left[
Y|L_{0},L_{1}\right]  \in H\left(  \beta_{b_{1}},C_{B_{1}}\right)  ;$ $\left(
b1\right)  .$ $\pi_{0}\left(  L_{0};\theta\right)  =\Pr\left(  R_{0}%
=1|L_{0}\right)  \in H\left(  \beta_{\pi_{0}},C_{\pi_{0}}\right)  $ and
$0<\sigma_{\pi_{0}}^{l}<\pi_{0}\left(  L_{0};\theta\right)  $ w.p. 1 $\left(
b2\right)  .$ $\ \pi_{1}\left(  L_{0},L_{1};\theta\right)  =\Pr\left(
R_{1}=1|L_{0},L_{1}\right)  \in H\left(  \beta_{\pi_{1}},C_{\pi_{1}}\right)  $
and $0<\sigma_{\pi_{1}}^{l}<\pi_{1}\left(  L_{0},L_{1};\theta\right)  $ w.p.
1$.$ $\left(  c1\right)  .$ The marginal density of $L_{0},$ $f_{0}=f\left(
L_{0};\theta\right)  ,$ falls in a H\"{o}lder ball $H\left(  \beta_{f_{0}%
},C_{f_{0}}\right)  ,$ and $0<\sigma_{f_{0}}<f_{0}\left(  L_{0};\theta\right)
$ w.p. 1, $\left\vert \left\vert f_{0}\right\vert \right\vert _{\infty}\leq
C_{f_{0}}^{\ast}<\infty.$ $\left(  c2\right)  .$ The marginal density of
$\left(  L_{0},L_{1}\right)  ,$ $f_{1}=f\left(  L_{0},L_{1};\theta\right)  ,$
falls in a H\"{o}lder ball $H\left(  \beta_{f_{1}},C_{f_{1}}\right)  ,$ and
$0<\sigma_{f_{1}}<f_{1}\left(  L_{0},L_{1};\theta\right)  $ w.p. 1,
$\left\vert \left\vert f_{1}\right\vert \right\vert _{\infty}\leq C_{f_{1}%
}^{\ast}<\infty.$

We define $g_{0}\left(  L_{0};\theta\right)  =\pi_{0}f_{0}$, $g_{1}\left(
L_{0},L_{1};\theta\right)  =\pi_{0}\pi_{1}f_{1}$ with corresponding H\"{o}lder
exponents $\beta_{g_{0}}=\min\left(  \beta_{f_{0}},\beta_{\pi_{0}}\right)  $
and $\beta_{g_{1}}=\min\left(  \beta_{f_{1}},\beta_{\pi_{0}},\beta_{\pi_{1}%
}\right)  $ respectively.

We next show how to apply theorem \ref{map} in a nested fashion in order to
derive higher order U-statistic estimators $\left\{  \widehat{\psi}_{m}%
:m\geq1\right\}  $ in the observed data model. \ In the first step of our
procedure, we derive higher order influence functions for a truncated
parameter in the artificial missing data problem in which the observed data is
given by $\ O_{i}^{\dag}=\left(  R_{0,i},L_{0,i},R_{0,i}L_{1,i},R_{0,i}%
Y_{i}\right)  $ rather than $O_{i}.$ \ In the second step, we construct
influence functions from a second artificial missing data problem with
$O_{i}^{\dag}$ now the full data and $O_{i}$ the observed data. \ This final
influence function is, in fact, the influence function for a truncated
parameter in the original monotone missing data model.

To follow this nested procedure, we derive a first stage class of estimators
\[
\left\{  \widehat{\psi}_{m}^{\dag}:m\geq1\right\}  ,
\]
which are functions of $\ O_{i}^{\dag}.$ \ In this model, $D_{i}$ is the full
data, $O_{i}^{\dag}$ is the observed data, $R_{0,i}$ is the missing indicator,
$L_{0,i}$ is a vector of always observed covariates, and $\left(  Y_{i}%
,L_{1i}\right)  $ is the outcome which might be missing. \ Since the parameter
of interest $\psi\left(  \theta\right)  $ is the marginal mean of $Y,$ this is
Example 2a\textbf{\ }of section 3.1\textbf{. }Therefore results from this
example may be applied, hence
\[
if_{1,\psi\left(  \theta\right)  }^{F}\left(  D_{i}\right)  =Y_{i}-\psi\left(
\theta\right)
\]
and
\[
if_{jj,\psi\left(  \theta\right)  }^{F}\left(  D_{i}\right)  =0\text{ \ \ for
\ }\forall\text{ {}}j\geq2.
\]

Moreover,
\[
if_{1,\psi\left(  \theta\right)  }\left(  O_{i}^{\dag}\right)  =\frac{R_{0,i}%
}{\pi_{0;i}}\left(  Y_{i}-B_{0;i}\right)  +B_{0;i}-\psi\left(  \theta\right)
\]
so that we can define
\begin{align*}
\widetilde{B}_{0} &  =\widehat{B}_{0}+\overline{Z}_{k_{0}}^{T}E_{\theta
}\left[  \frac{R_{0}}{\widehat{\pi}_{0}}\overline{Z}_{k_{0}}\overline
{Z}_{k_{0}}^{T}\right]  ^{-1}E_{\theta}\left[  \frac{R_{0}}{\widehat{\pi}_{0}%
}\left(  Y-\widehat{B}_{0}\right)  \overline{Z}_{k_{0}}\right] \\
\widetilde{\pi}_{0}^{-1} &  =\widehat{\pi}_{0}^{-1}\left(  1-\overline
{Z}_{k_{0}}^{T}E_{\theta}\left[  \frac{R_{0}}{\widehat{\pi}_{0}}\overline
{Z}_{k_{0}}\overline{Z}_{k_{0}}^{T}\right]  ^{-1}E_{\theta}\left[  \left(
\frac{R_{0}}{\widehat{\pi}_{0}}-1\right)  \overline{Z}_{k_{0}}\right]  \right)
\\
\widetilde{\psi}^{\dag}\left(  \theta\right)   &  =E_{\theta}\left[
\frac{R_{0}}{\widetilde{\pi}_{0}}\left(  Y-\widetilde{B}_{0}\right)
+\widetilde{B}_{0}\right]
\end{align*}
with $\dot{B}=1$ and $\dot{P}=\widehat{\pi}_{0}^{-1}.$ {}

Moreover%
\[
\widehat{\psi}_{m}^{\dag}=\psi\left(  \widehat{\theta}\right)  +\mathbb{V}%
_{n}\left(  \widehat{IF}_{m,\widetilde{\psi}^{\dag}}\right)
\]
\ $\left(  \widehat{B}_{0},\widehat{\pi}_{0},\widehat{f}\left(  L_{0}\right)
\right)  $ are rate optimal nonparametric estimators of $\left(  B_{0},\pi
_{0},f_{0}\right)  $ respectively estimated from the training sample and
$\overline{Z}_{k_{0}}=\overline{z}_{k_{0}}\left(  L_{0}\right)  =\widehat{E}%
\left(  \overline{\varphi}_{k_{0}}\left(  L_{0}\right)  \overline{\varphi
}_{k_{0}}^{T}\left(  L_{0}\right)  \right)  ^{-1/2}\overline{\varphi}_{k_{0}%
}\left(  L_{0}\right)  .$

From theorem \ref{TBrate} and eq.(\ref{EB6}) we also know that:%
\begin{gather*}
E\left(  \widehat{\psi}_{m}^{\dag}\right)  -\psi\left(  \theta\right)
=TB_{k_{0}}+EB_{m}\\
=O_{P}\left(  \max\left[  k_{0}^{-\frac{\beta_{b_{0}}+\beta_{\pi_{0}}}{d_{0}}%
},\left(  \frac{\log n}{n}\right)  ^{\frac{\left(  m-1\right)  \beta_{g_{0}}%
}{d_{0}+2\beta_{g_{0}}}}n^{-\left(  \frac{\beta_{b_{0}}}{d_{0}+2\beta_{b_{0}}%
}+\frac{\beta_{\pi_{0}}}{d_{0}+2\beta_{\pi_{0}}}\right)  }\right]  \right)
\end{gather*}

Next, we proceed with the second step of our procedure, where $O_{i}^{\dag}$
is now the full data and $O_{i}$ becomes the observed data. \ Then,
$O_{i}^{\dag}=O_{i}$ if $R_{0,i}=0$ or $R_{0,i}=R_{1,i}=1.$ Therefore we may
define a new missing indicator
\[
R=\left(  1-R_{0}\right)  +R_{0}R_{1}%
\]
with
\[
\Pr\left(  R_{i}=1|O_{i}\right)  =\left(  1-R_{0,i}\right)  +R_{0,i}\pi_{1,i}%
\]
In contrast with the first phase of our procedure, the full data influence
function now has non-zero higher order contributions; thus we can proceed with
the strategy layed out in the first part of this section. \ We can put
$\widehat{IF}_{jj,\widetilde{\psi}^{\dag}\left(  \theta\right)  ,\overline
{i}_{j}}$ in the format of eq.(\ref{sum}) as
\begin{align*}
&  \widehat{IF}_{jj,\widetilde{\psi}^{\dag}\left(  \theta\right)
,\overline{i}_{j}}\\
&  =\left(  -1\right)  ^{j-1}%
%TCIMACRO{\tsum \limits_{s_{1}=1}^{k_{0}}}%
%BeginExpansion
{\textstyle\sum\limits_{s_{1}=1}^{k_{0}}}
%EndExpansion
..%
%TCIMACRO{\tsum _{s_{j-1}=1}^{k_{0}}}%
%BeginExpansion
{\textstyle\sum_{s_{j-1}=1}^{k_{0}}}
%EndExpansion
\left\{
\begin{array}
[c]{c}%
\left(  \frac{R_{0}}{\widehat{\pi}_{0}}\left(  Y-\widehat{B}_{0}\right)
Z_{s_{1}}\right)  _{i_{1}}\left(  \left(  \frac{R_{0}}{\widehat{\pi}_{0}%
}-1\right)  Z_{s_{j-1}}\right)  _{i_{j}}\times\\%
%TCIMACRO{\tprod \limits_{t=2}^{j-1}}%
%BeginExpansion
{\textstyle\prod\limits_{t=2}^{j-1}}
%EndExpansion
\left(  \frac{R_{0}}{\widehat{\pi}_{0}}Z_{s_{t-1}}Z_{s_{t}}-I\left(
s_{t-1}=s_{t}\right)  \right)  _{i_{t}}%
\end{array}
\right\} \\
&  =\left(  -1\right)  ^{j-1}\sum_{l=1}^{k_{0}^{j-1}}\left\{
\begin{array}
[c]{c}%
\left(  \frac{R_{0}}{\widehat{\pi}_{0}}\left(  Y-\widehat{B}_{0}\right)
Z_{n\left(  l,1\right)  }\right)  _{i_{1}}\left(  \left(  \frac{R_{0}%
}{\widehat{\pi}_{0}}-1\right)  Z_{_{n\left(  l,j-1\right)  }}\right)  _{i_{j}%
}\times\\%
%TCIMACRO{\tprod \limits_{t=2}^{j-1}}%
%BeginExpansion
{\textstyle\prod\limits_{t=2}^{j-1}}
%EndExpansion
\left(  \frac{R_{0}}{\widehat{\pi}_{0}}Z_{n\left(  l,t-1\right)  }Z_{n\left(
l,t\right)  }-I\left(  n\left(  l,t-1\right)  =n\left(  l,t\right)  \right)
\right)  _{i_{t}}%
\end{array}
\right\}
\end{align*}

where $n\left(  l\right)  :\{1,2,...,k_{0}^{j-1}\}\rightarrow\{1,2,...,k_{0}%
\}^{j-1}$ is a one-to-one mapping that indexes all permutations of $\left\{
\left(  s_{1},s_{2},...s_{j-1}\right)  ,\text{ }1\leq s_{t}\leq k_{0}\right\}
,$ and $n\left(  l,t\right)  $ is the $t$th entry of $n\left(  l\right)  .$ We
define \ $\forall$ $j>1.$%
\begin{align*}
\widehat{U}_{l,s}^{\left(  j\right)  }  &  \equiv\left\{
\begin{tabular}
[c]{c}%
$\frac{R_{0}}{\widehat{\pi}_{0}}\left(  Y-\widehat{B}_{0}\right)  Z_{n\left(
l,1\right)  }$ \ \ \ for{}{}{}$s=1$\\
$\left\{
\begin{array}
[c]{c}%
\frac{R_{0}}{\widehat{\pi}_{0}}Z_{n\left(  l,s-1\right)  }Z_{n\left(
l,s\right)  }\\
-I\left(  n\left(  l,s-1\right)  =n\left(  l,s\right)  \right)
\end{array}
\right\}  $ \ for {}{}$1<s<j$\\
$\left(  -1\right)  ^{j-1}\left(  \frac{R_{0}}{\widehat{\pi}_{0}}-1\right)
Z_{_{n\left(  l,j-1\right)  }}$ \ {}for {}$s=j$%
\end{tabular}
\ \right. \\
\tau_{l,s}^{\left(  j\right)  }\left(  \theta\right)   &  \equiv E_{\theta
}\left(  \widehat{U}_{l,s}^{\left(  j\right)  }\right)
\end{align*}

Let $V_{i}=\left(  R_{0,i},L_{0,i},R_{0,i}L_{1,i}\right)  $ which is always
observed. Then
\begin{gather*}
\tau_{l,1}^{\left(  j\right)  }\left(  \theta\right)  =E_{\theta}\left(
\frac{R_{0}}{\widehat{\pi}_{0}}\left(  Y-\widehat{B}_{0}\right)  Z_{n\left(
l,1\right)  }\right) \\
=E_{\theta}\left(  \frac{R}{\pi}\left[
\begin{array}
[c]{c}%
\frac{R_{0}}{\widehat{\pi}_{0}}\left(  Y-\widehat{B}_{0}\right)  Z_{n\left(
l,1\right)  }-\\
E\left(  \frac{R_{0}}{\widehat{\pi}_{0}}\left(  Y-\widehat{B}_{0}\right)
Z_{n\left(  l,1\right)  }|V\right)
\end{array}
\right]  +E_{\theta}\left(  \frac{R_{0}}{\widehat{\pi}_{0}}\left(
Y-\widehat{B}_{0}\right)  Z_{n\left(  l,1\right)  }|V\right)  \right) \\
=E_{\theta}\left(  \frac{R_{1}}{\pi_{1}\left(  \theta\right)  }\frac{R_{0}%
}{\widehat{\pi}_{0}}\left(  YZ_{n\left(  l,1\right)  }-B_{1}\left(
\theta\right)  Z_{n\left(  l,1\right)  }\right)  +\frac{R_{0}}{\widehat{\pi
}_{0}}\left(  B_{1}\left(  \theta\right)  Z_{n\left(  l,1\right)
}-\widehat{B}_{0}Z_{n\left(  l,1\right)  }\right)  \right)
\end{gather*}

Thus, $\tau_{l,1}^{\left(  j\right)  }\left(  \theta\right)  $ falls into the
doubly-robust $H\left(  b,p\right)  $ class of functionals with $B^{\tau
_{l,1}^{\left(  j\right)  }}=B_{1}\left(  \theta\right)  Z_{n\left(
l,1\right)  },$ $P^{\tau_{l,1}^{\left(  j\right)  }}=\pi_{1}^{-1}\left(
\theta\right)  ,$ and corresponding $H_{1}^{\tau_{l,1}^{\left(  j\right)  }%
}=-R_{1}\frac{R_{0}}{\widehat{\pi}_{0}},$ $H_{2}^{\tau_{l,1}^{\left(
j\right)  }}=\frac{R_{0}}{\widehat{\pi}_{0}},$ $H_{3}^{\tau_{l,1}^{\left(
j\right)  }}=R_{1}\frac{R_{0}}{\widehat{\pi}_{0}}YZ_{n\left(  l,1\right)  },$
$H_{4}^{\tau_{l,1}^{\left(  j\right)  }}=-\frac{R_{0}}{\widehat{\pi}_{0}%
}\widehat{B}_{0}Z_{n\left(  l,1\right)  }.$ For any $s>1,$ $\widehat{U}%
_{l,s}^{\left(  j\right)  }=\widehat{u}_{l,s}^{\left(  j\right)  }\left(
O\right)  $ is a known function of the observed data, thus $if_{1,\tau
_{l,s}^{\left(  j\right)  }\left(  \theta\right)  }=\widehat{u}_{l,s}^{\left(
j\right)  }\left(  O\right)  -\tau_{l,s}^{\left(  j\right)  }\left(
\theta\right)  $ and $if_{m,m,\tau_{l,s}^{\left(  j\right)  }\left(
\theta\right)  }=0$ for any $m>1.$ Moreover, choosing $\dot{B}^{\tau
_{l,1}^{\left(  j\right)  }}=1,$ $\dot{P}^{\tau_{l,1}^{\left(  j\right)  }%
}=\widehat{\pi}_{1}^{-1},$ so that:
\begin{align*}
\widetilde{B_{1}Z_{n\left(  l,1\right)  }} &  =\left\{
\begin{array}
[c]{c}%
\widehat{B}_{1}Z_{n\left(  l,1\right)  }+\overline{W}_{k_{1}}^{T}E_{\theta
}\left(  \frac{R_{1}}{\widehat{\pi}_{1}}\frac{R_{0}}{\widehat{\pi}_{0}%
}\overline{W}_{k_{1}}\overline{W}_{k_{1}}^{T}\right)  ^{-1}\\
\times E_{\theta}\left(  \frac{R_{1}}{\widehat{\pi}_{1}}\frac{R_{0}%
}{\widehat{\pi}_{0}}\left(  Y-\widehat{B}_{1}\right)  Z_{n\left(  l,1\right)
}\overline{W}_{k_{1}}\right)
\end{array}
\right\} \\
\widetilde{\pi}_{1}^{-1} &  =\left\{  \widehat{\pi}_{1}^{-1}\left(
\begin{array}
[c]{c}%
1-\overline{W}_{k_{1}}^{T}E_{\theta}\left[  \frac{R_{1}}{\widehat{\pi}_{1}%
}\frac{R_{0}}{\widehat{\pi}_{0}}\overline{W}_{k_{1}}\overline{W}_{k_{1}}%
^{T}\right]  ^{-1}\\
\times E_{\theta}\left[  \frac{R_{0}}{\widehat{\pi}_{0}}\left(  \frac{R_{1}%
}{\widehat{\pi}_{1}}-1\right)  \overline{W}_{k_{1}}\right]
\end{array}
\right)  \right\}
\end{align*}
where $\overline{W}_{k_{1}}=\overline{w}_{k_{1}}\left(  L_{0},L_{1}\right)
=\widehat{E}\left(  \overline{\varphi}_{k_{1}}\left(  L_{0},L_{1}\right)
\overline{\varphi}_{k_{1}}^{T}\left(  L_{0},L_{1}\right)  \right)
^{-1/2}\overline{\varphi}_{k_{1}}\left(  L_{0},L_{1}\right)  ,$ and
$\overline{\varphi}_{k_{1}}\left(  L_{0},L_{1}\right)  $ is a $k_{1}%
-$dimensional vector of tensor product basis for functions of $\left(
L_{0},L_{1}\right)  .$ From theorem \ref{map}, for $1\leq l\leq k_{0}^{j-1}:$%

\begin{align*}
&  \widetilde{\tau}_{l,1}^{\left(  j\right)  }\left(  \theta\right) \\
&  \equiv E_{\theta}\left(
\begin{array}
[c]{c}%
R_{1}\widetilde{\pi}_{1}^{-1}\left(  \theta\right)  \frac{R_{0}}{\widehat{\pi
}_{0}}\left(  YZ_{n\left(  l,1\right)  }-\widetilde{B_{1}\left(
\theta\right)  Z_{n\left(  l,1\right)  }}\right) \\
+\frac{R_{0}}{\widehat{\pi}_{0}}\left(  \widetilde{B_{1}\left(  \theta\right)
Z_{n\left(  l,1\right)  }}-\widehat{B}_{0}Z_{n\left(  l,1\right)  }\right)
\end{array}
\right)
\end{align*}
have higher order influence functions $\left\{  \mathbb{IF}_{m,\widetilde{\tau
}_{l,1}^{\left(  j\right)  }}\left(  \theta\right)  ,\text{ }m\geq1\right\}
$. Let $\widetilde{\tau}_{l,t}^{\left(  j\right)  }\left(  \theta\right)
\equiv\tau_{l,t}^{\left(  j\right)  }\left(  \theta\right)  $ for any $t>1,$
and $\left\{  \mathbb{IF}_{m,\widetilde{\psi}_{j,j}}\left(  \theta\right)
,\text{ }m\geq1\right\}  $ be the higher order influence functions of
$\widetilde{\psi}_{j,j}\left(  \theta\right)  =\sum_{l=1}^{k_{0}^{j-1}}%
%TCIMACRO{\tprod \limits_{t=1}^{j}}%
%BeginExpansion
{\textstyle\prod\limits_{t=1}^{j}}
%EndExpansion
\widetilde{\tau}_{l,t}^{\left(  j\right)  }\left(  \theta\right)  $.

For $j=1,$ define%

\begin{align*}
\tau_{1,1}^{\left(  1\right)  }  &  =E_{\theta}\left(  \widehat{IF}%
_{1,\widetilde{\psi}^{\dag}}\left(  O_{i}^{\dag}\right)  \right) \\
&  =E_{\theta}\left(  \frac{R}{\pi}\left(  \widehat{IF}_{1,\widetilde{\psi
}^{\dag}}\left(  O_{i}^{\dag}\right)  -E_{\theta}\left(  \widehat{IF}%
_{1,\widetilde{\psi}^{\dag}}\left(  O_{i}^{\dag}\right)  |V_{i}\right)
\right)  +E_{\theta}\left(  \widehat{IF}_{1,\widetilde{\psi}^{\dag}}\left(
O_{i}^{\dag}\right)  |V_{i}\right)  \right) \\
&  =E_{\theta}\left(  \frac{R_{1}}{\pi_{1}\left(  \theta\right)  }\frac{R_{0}%
}{\widehat{\pi}_{0}}\left(  Y-B_{1}\left(  \theta\right)  \right)
+\frac{R_{0}}{\widehat{\pi}_{0}}\left(  B_{1}\left(  \theta\right)
-\widehat{B}_{0}\right)  +\widehat{B}_{0}-\psi\left(  \widehat{\theta}\right)
\right)
\end{align*}

$\tau_{1,1}^{\left(  1\right)  }$ also belongs to the $H\left(  \cdot
,\cdot\right)  $ class of models with $B^{\tau_{1,1}^{\left(  1\right)  }%
}=B_{1},$ $P^{\tau_{1,1}^{\left(  1\right)  }}=\pi_{1}^{-1},$ $H_{1}%
=-R_{1}\frac{R_{0}}{\widehat{\pi}_{0}},$ $H_{2}=\frac{R_{0}}{\widehat{\pi}%
_{0}},$ $H_{3}=R_{1}\frac{R_{0}}{\widehat{\pi}_{0}}Y,$ $H_{4}=\left(
1-\frac{R_{0}}{\widehat{\pi}_{0}}\right)  \widehat{B}_{0}-\psi\left(
\widehat{\theta}\right)  .$ We may choose $\dot{B}^{\tau_{1,1}^{\left(
1\right)  }}=B_{1},$ $\dot{P}^{\tau_{1,1}^{\left(  1\right)  }}=\widehat{\pi
}_{1}^{-1}\ {}$so that
\begin{align*}
\widetilde{B}_{1}  &  =\left\{
\begin{array}
[c]{c}%
\widehat{B}_{1}+\overline{W}_{k_{1}}^{T}E_{\theta}\left(  \frac{R_{1}%
}{\widehat{\pi}_{1}}\frac{R_{0}}{\widehat{\pi}_{0}}\overline{W}_{k_{1}%
}\overline{W}_{k_{1}}^{T}\right)  ^{-1}\\
\times E_{\theta}\left(  \frac{R_{1}}{\widehat{\pi}_{1}}\frac{R_{0}%
}{\widehat{\pi}_{0}}\left(  Y-\widehat{B}_{1}\right)  \overline{W}_{k_{1}%
}\right)
\end{array}
\right\} \\
\widetilde{\pi}_{1}^{-1}  &  =\left\{  \widehat{\pi}_{1}^{-1}\left(
\begin{array}
[c]{c}%
1-\overline{W}_{k_{1}}^{T}E_{\theta}\left[  \frac{R_{1}}{\widehat{\pi}_{1}%
}\frac{R_{0}}{\widehat{\pi}_{0}}\overline{W}_{k_{1}}\overline{W}_{k_{1}}%
^{T}\right]  ^{-1}\\
\times E_{\theta}\left[  \frac{R_{0}}{\widehat{\pi}_{0}}\left(  \frac{R_{1}%
}{\widehat{\pi}_{1}}-1\right)  \overline{W}_{k_{1}}\right]
\end{array}
\right)  \right\}
\end{align*}
and%
\[
\widetilde{\tau}_{1,1}^{\left(  1\right)  }\left(  \theta\right)  \equiv
E_{\theta}\left(  R_{1}\widetilde{\pi}_{1}^{-1}\frac{R_{0}}{\widehat{\pi}_{0}%
}\left(  Y-\widetilde{B}_{1}\right)  +\frac{R_{0}}{\widehat{\pi}_{0}}\left(
\widetilde{B}_{1}-\widehat{B}_{0}\right)  +\widehat{B}_{0}-\psi\left(
\widehat{\theta}\right)  \right)
\]
has higher order IFs $\left\{  \mathbb{IF}_{m,\widetilde{\tau}_{1,1}^{\left(
1\right)  }\left(  \theta\right)  },\text{ }m\geq1\right\}  .$

\begin{theorem}
Define $\widetilde{\widetilde{\psi}}_{k_{0},k_{1},m}\left(  \theta\right)
=\psi\left(  \widehat{\theta}\right)  +\widetilde{\tau}_{1,1}^{\left(
1\right)  }\left(  \theta\right)  +%
%TCIMACRO{\tsum \limits_{j=2}^{m}}%
%BeginExpansion
{\textstyle\sum\limits_{j=2}^{m}}
%EndExpansion
\widetilde{\psi}_{j,j}\left(  \theta\right)  ,$ then {}$\forall$ $m^{\ast}\leq
m$%
\begin{align*}
{}\mathbb{IF}_{m^{\ast},\widetilde{\widetilde{\psi}}_{k_{0},k_{1},m}\left(
\theta\right)  }\left(  \widehat{\theta}\right)   &  ={}\mathbb{IF}_{m^{\ast
},\widetilde{\widetilde{\psi}}_{k_{0},k_{1},m^{\ast}}\left(  \theta\right)
}\left(  \widehat{\theta}\right) \\
&  ={}\mathbb{IF}_{1,\widetilde{\widetilde{\psi}}_{k_{0},k_{1},m}\left(
\theta\right)  }\left(  \widehat{\theta}\right)  +{}%
%TCIMACRO{\tsum \limits_{j=2}^{m^{\ast}}}%
%BeginExpansion
{\textstyle\sum\limits_{j=2}^{m^{\ast}}}
%EndExpansion
\mathbb{IF}_{j,j,\widetilde{\widetilde{\psi}}_{k_{0},k_{1},m}\left(
\theta\right)  }\left(  \widehat{\theta}\right)
\end{align*}

with
\[
\widehat{IF}_{1,\widetilde{\widetilde{\psi}}_{k_{0},k_{1},m}\left(
\theta\right)  }\left(  O_{i}\right)  =\left\{
\begin{array}
[c]{c}%
\frac{R_{1,i}}{\widehat{\pi}_{1,i}}\frac{R_{0,i}}{\widehat{\pi}_{0,i}}\left(
Y_{i}-\widehat{B}_{1,i}\right)  +\\
\frac{R_{0,i}}{\widehat{\pi}_{0,i}}\left(  \widehat{B}_{1,i}-\widehat{B}%
_{0,i}\right)  +\widehat{B}_{0,i}-\psi\left(  \widehat{\theta}\right)
\end{array}
\right\}
\]

and $\forall$ \ $j\geq2$%
\begin{align*}
&  \widehat{IF}_{j,j,\widetilde{\widetilde{\psi}}_{k_{0},k_{1},m}\left(
\theta\right)  }\left(  O_{i_{1}},...O_{i_{j}}\right) \\
&  =\left(  -1\right)  ^{j-1}\left[
\begin{array}
[c]{c}%
\left\{
\begin{array}
[c]{c}%
\left[  \frac{R_{0}}{\widehat{\pi}_{0}}\frac{R_{1}}{\widehat{\pi}_{1}}\left(
Y-\widehat{B}_{1}\right)  \right]  _{i_{1}}\overline{W}_{k_{1},i_{1}}%
^{T}\times\\%
%TCIMACRO{\dprod \limits_{s=2}^{j-1}}%
%BeginExpansion
{\displaystyle\prod\limits_{s=2}^{j-1}}
%EndExpansion
\left(  \frac{R_{0}}{\widehat{\pi}_{0}}\frac{R_{1}}{\widehat{\pi}_{1}%
}\overline{W}_{k_{1}}\overline{W}_{k_{1}}^{T}-I\right)  _{i_{s}}\overline
{W}_{k_{1},i_{j}}\left[  \frac{R_{0}}{\widehat{\pi}_{0}}\left(  \frac{R_{1}%
}{\widehat{\pi}_{1}}-1\right)  \right]  _{i_{j}}%
\end{array}
\right\} \\%
%TCIMACRO{\tsum \limits_{t=2}^{j-1}}%
%BeginExpansion
{\textstyle\sum\limits_{t=2}^{j-1}}
%EndExpansion
\left\{
\begin{array}
[c]{c}%
\left(  \frac{R_{0}}{\widehat{\pi}_{0}}\left(  \frac{R_{1}}{\widehat{\pi}_{1}%
}-1\right)  \right)  _{i_{j}}\overline{W}_{k_{1},i_{j}}^{T}%
%TCIMACRO{\tprod \limits_{l=t+1}^{j-1}}%
%BeginExpansion
{\textstyle\prod\limits_{l=t+1}^{j-1}}
%EndExpansion
\left(  \frac{R_{1}}{\widehat{\pi}_{1}}\frac{R_{0}}{\widehat{\pi}_{0}%
}\overline{W}_{k_{1}}\overline{W}_{k_{1}}^{T}-I\right)  _{i_{l}}\\
\times\left[  \overline{W}_{k_{1}}\frac{R_{1}}{\widehat{\pi}_{1}}\frac{R_{0}%
}{\widehat{\pi}_{0}}\left(  Y-\widehat{B}_{1}\right)  \overline{Z}_{k_{0}}%
^{T}\right]  _{i_{1}}\times\\%
%TCIMACRO{\dprod \limits_{s=2}^{t-1}}%
%BeginExpansion
{\displaystyle\prod\limits_{s=2}^{t-1}}
%EndExpansion
\left(  \frac{R_{0}}{\widehat{\pi}_{0}}\overline{Z}_{k_{0}}\overline{Z}%
_{k_{0}}^{T}-I\right)  _{i_{s}}\left(  \frac{R_{0}}{\widehat{\pi}_{0}%
}-1\right)  _{i_{t}}\overline{Z}_{k_{0,}i_{t}}%
\end{array}
\right\} \\
+\left\{
\begin{array}
[c]{c}%
\left[  \frac{R_{1}}{\widehat{\pi}_{1}}\frac{R_{0}}{\widehat{\pi}_{0}}\left(
Y-\widehat{B}_{1}\right)  +\frac{R_{0}}{\widehat{\pi}_{0}}\left(
\widehat{B}_{1}-\widehat{B}_{0}\right)  \right]  _{i_{1}}\overline{Z}%
_{k_{0},i_{1}}^{T}\times\\%
%TCIMACRO{\dprod \limits_{s=2}^{j-1}}%
%BeginExpansion
{\displaystyle\prod\limits_{s=2}^{j-1}}
%EndExpansion
\left(  \frac{R_{0}}{\widehat{\pi}_{0}}\overline{Z}_{k_{0}}\overline{Z}%
_{k_{0}}^{T}-I\right)  _{i_{s}}\left(  \frac{R_{0}}{\widehat{\pi}_{0}%
}-1\right)  _{i_{j}}\overline{Z}_{k_{0,}i_{j}}%
\end{array}
\right\}
\end{array}
\right]
\end{align*}

\end{theorem}

\begin{proof}
This follows directly from theorem \ref{map} and the fact that
$\widetilde{\tau}_{l,t}^{\left(  j\right)  }\left(  \widehat{\theta}\right)
=$ $\tau_{l,t}^{\left(  j\right)  }\left(  \widehat{\theta}\right)  =0$ for
$t>1.$
\end{proof}

Suppose the class of estimators $\left\{  \widehat{\psi}_{m}\left(
\mathbf{O}\right)  ,\text{ }m\geq1\right\}  $ are defined as
\begin{align*}
\widehat{\psi}_{m}\left(  \mathbf{O}\right)   &  \equiv\widehat{\psi}%
_{1}\left(  \mathbf{O}\right)  +\sum_{j=2}^{m}\widehat{\psi}_{j,j}\left(
\mathbf{O}\right) \\
\widehat{\psi}_{1}\left(  \mathbf{O}\right)   &  \equiv\psi\left(
\widehat{\theta}\right)  +\mathbb{IF}_{1,\widetilde{\widetilde{\psi}}%
_{k_{0},k_{1},m}\left(  \theta\right)  }\left(  \widehat{\theta}\right) \\
\widehat{\psi}_{j,j}\left(  \mathbf{O}\right)   &  \equiv\mathbb{IF}%
_{j,j,\widetilde{\widetilde{\psi}}_{k_{0},k_{1},m}\left(  \theta\right)
}\left(  \widehat{\theta}\right)
\end{align*}

Then%
\begin{align}
&  E_{\theta}\left(  \widehat{\psi}_{m}\right)  -\psi\left(  \theta\right)
\nonumber\\
&  =\left\{
\begin{array}
[c]{c}%
E_{\theta}\left(  \psi\left(  \widehat{\theta}\right)  +\mathbb{IF}%
_{m,\widetilde{\widetilde{\psi}}_{k_{0},k_{1},m}\left(  \theta\right)
}\left(  \widehat{\theta}\right)  -\widetilde{\widetilde{\psi}}_{k_{0}%
,k_{1},m}\left(  \theta\right)  \right) \\
+\left(  \widetilde{\widetilde{\psi}}_{k_{0},k_{1},m}\left(  \theta\right)
-E_{\theta}\left[  \mathbb{IF}_{m,\widetilde{\psi}^{\dag}\left(
\theta\right)  }\left(  \widehat{\theta}\right)  +\psi\left(  \widehat{\theta
}\right)  \right]  \right) \\
+\left[  E_{\theta}\left(  \mathbb{IF}_{m,\widetilde{\psi}^{\dag}}\left(
\widehat{\theta}\right)  +\psi\left(  \widehat{\theta}\right)  \right)
-\widetilde{\psi}^{\dag}\left(  \theta\right)  \right]  +\left(
\widetilde{\psi}^{\dag}\left(  \theta\right)  -\psi\left(  \theta\right)
\right)
\end{array}
\right\} \label{bsp}%
\end{align}

The following theorem examines the bias and variance of $\widehat{\psi}_{m}.
$Define $\delta P_{t}=\frac{\pi_{t}}{\widehat{\pi}_{t}}-1,$ $\delta
B_{t}=B_{t}-\widehat{B}_{t},$ and $\delta g_{t}=\frac{g_{t}}{\widehat{g}_{t}%
}-1$ for $t=0,1.$ Let $q_{0}=\frac{\pi_{0}}{\widehat{\pi}_{0}}$ and
$q_{01}=\frac{\pi_{0}\pi_{1}}{\widehat{\pi}_{0}\widehat{\pi}_{1}}.$

\begin{theorem}
\label{BTOMM}Suppose conditions $\left(  a1\right)  -\left(  c2\right)  $ hold
then%
\[
E_{\theta}\left(  \widehat{\psi}_{1}|\widehat{\theta}\right)  -\psi\left(
\theta\right)  =E_{\theta}\left[  \frac{R_{0}}{\widehat{\pi}_{0}}\delta
P_{1}\delta B_{1}+\delta P_{0}\delta B_{0}\left\vert \widehat{\theta
}\right\vert \right]
\]
and $\forall$ \ $m>1$%
\[
E_{\theta}\left(  \widehat{\psi}_{m}|\widehat{\theta}\right)  -\psi\left(
\theta\right)  =BI_{m,1}+BI_{m,2}%
\]
where
\begin{align}
&  BI_{m,1}\left(  -1\right)  ^{m-1}\nonumber\\
&  =\left\{
\begin{array}
[c]{c}%
\left\{
\begin{array}
[c]{c}%
E_{\theta}\left\{  \left[  q_{0}\delta B_{0}\right]  \overline{Z}_{k_{0}}%
^{T}\right\}  E_{\theta}\left[  q_{0}\overline{Z}_{k_{0}}\overline{Z}_{k_{0}%
}^{T}\right]  ^{-1}\times\\
\left[  E_{\theta}\left(  q_{0}\overline{Z}_{k_{0}}\overline{Z}_{k_{0}}%
^{T}\right)  -I\right]  ^{m-1}E_{\theta}\left[  \overline{Z}_{k_{0}}\delta
P_{0}\right]
\end{array}
\right\}  _{-\left(  EB_{1}^{\left(  1\right)  }\right)  }\\
+\left\{
\begin{array}
[c]{c}%
E_{\theta}\left\{  \left[  q_{01}\delta B_{1}\right]  \overline{W}_{k_{1}}%
^{T}\right\}  E_{\theta}\left[  q_{01}\overline{W}_{k_{1}}\overline{W}_{k_{1}%
}^{T}\right]  ^{-1}\times\\
\left[  E_{\theta}\left(  q_{01}\overline{W}_{k_{1}}\overline{W}_{k_{1}}%
^{T}\right)  -I\right]  ^{m-1}E_{\theta}\left[  \overline{W}_{k_{1}}%
q_{0}\delta P_{1}\right]
\end{array}
\right\}  _{-\left(  EB_{1}^{\left(  2\right)  }\right)  }\\
+%
%TCIMACRO{\dsum \limits_{j=2}^{m-1}}%
%BeginExpansion
{\displaystyle\sum\limits_{j=2}^{m-1}}
%EndExpansion
\left\{
\begin{array}
[c]{c}%
E_{\theta}\left[  \overline{Z}_{k_{0}}\delta P_{0}^{T}\right]  \left[
E_{\theta}\left(  q_{0}\overline{Z}_{k_{0}}\overline{Z}_{k_{0}}^{T}-I\right)
\right]  ^{j-2}\times\\
E_{\theta}\left[  q_{01}\delta B_{1}\overline{Z}_{k_{0}}\overline{W}_{k_{1}%
}^{T}\right]  E_{\theta}\left[  q_{01}\overline{W}_{k_{1}}\overline{W}_{k_{1}%
}^{T}\right]  ^{-1}\times\\
E_{\theta}\left(  q_{01}\overline{W}_{k_{1}}\overline{W}_{k_{1}}^{T}-I\right)
^{m-j}E_{\theta}\left[  \overline{W}_{k_{1}}q_{0}\delta P_{1}\right]
\end{array}
\right\}  _{-\left(  EB_{jj}^{\left(  2\right)  }\right)  }\\
+E_{\theta}\left\{  \left[  q_{0}\delta P_{1}\delta B_{1}\right]  \overline
{Z}_{k_{0}}^{T}\right\}  \left[  E_{\theta}\left(  q_{0}\overline{Z}_{k_{0}%
}\overline{Z}_{k_{0}}^{T}\right)  -I\right]  ^{m-2}E_{\theta}\left[
\overline{Z}_{k_{0}}\delta P_{0}\right]  _{-\left(  EB_{mm}^{\left(  2\right)
}\right)  }%
\end{array}
\right\} \label{ebm}%
\end{align}%
\begin{align}
&  \left\vert BI_{m,1}\right\vert \nonumber\\
&  \leq\left[
\begin{array}
[c]{c}%
\left\{
\begin{array}
[c]{c}%
\left\vert \left\vert \widehat{f}_{0}\right\vert \right\vert _{\infty
}\left\vert \left\vert q_{0}^{-1/2}\right\vert \right\vert _{\infty}\left\vert
\left\vert q_{0}\right\vert \right\vert _{\infty}\left\vert \left\vert
\frac{\widehat{f}_{0}}{f_{0}}\right\vert \right\vert _{\infty}\times\\
\left\vert \left\vert \delta g_{0}\right\vert \right\vert _{\infty}%
^{m-1}\left\{  \int\left(  b_{0}-\widehat{b}_{0}\right)  ^{2}dL_{0}\right\}
^{1/2}\times\\
\left\vert \left\vert \frac{f_{0}}{\pi_{0}\widehat{\pi}_{0}}\right\vert
\right\vert _{\infty}\left\{  \int\left(  \pi_{0}-\widehat{\pi}_{0}\right)
^{2}dL_{0}\right\}  ^{1/2}%
\end{array}
\right\} \\
+\left\{
\begin{array}
[c]{c}%
\left\vert \left\vert \widehat{f}_{1}\right\vert \right\vert _{\infty
}\left\vert \left\vert q_{01}^{-1/2}\right\vert \right\vert _{\infty
}\left\vert \left\vert q_{01}\right\vert \right\vert _{\infty}\left\vert
\left\vert \frac{\widehat{f}_{1}}{f_{1}}\right\vert \right\vert _{\infty
}\times\\
\left\vert \left\vert \delta g_{1}\right\vert \right\vert _{\infty}%
^{m-1}\left\{  \int\left(  b_{1}-\widehat{b}_{1}\right)  ^{2}dL_{0}%
dL_{1}\right\}  ^{1/2}\\
\times\left\vert \left\vert \frac{\pi_{0}f_{1}}{\widehat{\pi}_{0}\pi
_{1}\widehat{\pi}_{1}}\right\vert \right\vert _{\infty}\left\{  \int\left(
\pi_{1}-\widehat{\pi}_{1}\right)  ^{2}dL_{0}dL_{1}\right\}  ^{1/2}%
\end{array}
\right\} \\
+%
%TCIMACRO{\dsum \limits_{j=2}^{m-1}}%
%BeginExpansion
{\displaystyle\sum\limits_{j=2}^{m-1}}
%EndExpansion
\left\{
\begin{array}
[c]{c}%
\left\vert \left\vert \frac{f_{1}}{\widehat{f}_{1}}\right\vert \right\vert
_{\infty}\left\vert \left\vert \frac{f_{0}^{2}}{\widehat{f}_{0}\widehat{\pi
}_{0}^{2}}\right\vert \right\vert _{\infty}\left\vert \left\vert
\frac{\widehat{f}_{1}}{f_{1}}\right\vert \right\vert _{\infty}\left\vert
\left\vert \frac{\pi_{0}f_{1}}{\widehat{\pi}_{0}\pi_{1}\widehat{\pi}_{1}%
}\right\vert \right\vert _{\infty}\\
\left\vert \left\vert q_{01}\right\vert \right\vert _{\infty}\left\vert
\left\vert \delta B_{1}\right\vert \right\vert _{\infty}\left\vert \left\vert
\delta g_{0}\right\vert \right\vert _{\infty}^{j-2}\left\{  \int\left(
\pi_{0}-\widehat{\pi}_{0}\right)  ^{2}dL_{0}\right\}  ^{1/2}\times\\
\left\vert \left\vert q_{01}^{-1/2}\right\vert \right\vert _{\infty}\left\vert
\left\vert \delta g_{1}\right\vert \right\vert _{\infty}^{m-j}\left\{
\int\left(  \pi_{1}-\widehat{\pi}_{1}\right)  ^{2}dL_{0}dL_{1}\right\}  ^{1/2}%
\end{array}
\right\} \\
+\left\{
\begin{array}
[c]{c}%
\left\vert \left\vert \frac{f_{1}}{\widehat{f}_{1}\widehat{\pi}_{1}^{2}}%
q_{0}\right\vert \right\vert _{\infty}\left\vert \left\vert \delta
B_{1}\right\vert \right\vert _{\infty}\left\{  \int\left(  \pi_{1}%
-\widehat{\pi}_{1}\right)  ^{2}dL_{0}dL_{1}\right\}  ^{1/2}\\
\times\left\vert \left\vert \frac{f_{0}^{2}}{\widehat{f}_{0}\widehat{\pi}%
_{0}^{2}}\right\vert \right\vert \left\vert \left\vert \delta g_{0}\right\vert
\right\vert _{\infty}^{m-2}\left\{  \int\left(  \pi_{0}-\widehat{\pi}%
_{0}\right)  ^{2}dL_{0}\right\}  ^{1/2}%
\end{array}
\right\}
\end{array}
\right] \label{ebm2}\\
&  =O_{p}\left(  \max\left[
\begin{array}
[c]{c}%
\left(  \frac{\log n}{n}\right)  ^{\frac{\left(  m-1\right)  \beta_{g_{0}}%
}{d_{0}+2\beta_{g_{0}}}}n^{-\left(  \frac{\beta_{b_{0}}}{d_{0}+2\beta_{b_{0}}%
}+\frac{\beta_{\pi_{0}}}{d_{0}+2\beta_{\pi_{0}}}\right)  },\\
\left(  \frac{\log n}{n}\right)  ^{\frac{\left(  m-1\right)  \beta_{g_{1}}%
}{d_{1}+2\beta_{g_{1}}}}n^{-\left(  \frac{\beta_{b_{1}}}{d_{1}+2\beta_{b_{1}}%
}+\frac{\beta_{\pi_{1}}}{d_{1}+2\beta_{\pi_{1}}}\right)  },\\%
%TCIMACRO{\dsum \limits_{j=2}^{m}}%
%BeginExpansion
{\displaystyle\sum\limits_{j=2}^{m}}
%EndExpansion
\left(  \frac{\log n}{n}\right)  ^{\frac{\beta_{b_{1}}}{d_{1}+2\beta_{b_{1}}%
}+\frac{\left(  j-2\right)  \beta_{g_{0}}}{d_{0}+2\beta_{g_{0}}}+\frac{\left(
m-j\right)  \beta_{g_{1}}}{d_{1}+2\beta_{g_{1}}}}n^{-\frac{\beta_{\pi_{0}}%
}{d_{0}+2\beta_{\pi_{0}}}-\frac{\beta_{\pi_{1}}}{d_{1}+2\beta_{\pi_{1}}}}%
\end{array}
\right]  \right) \label{ebm3}%
\end{align}
and%
\begin{align}
&  BI_{m,2}\nonumber\\
&  =\left\{
\begin{array}
[c]{c}%
E_{\theta}\left(  \Pi_{\theta}^{\bot}\left(  q_{0}^{1/2}\delta B_{0}|\left(
q_{0}^{1/2}\overline{Z}_{k_{0}}\right)  \right)  \Pi_{\theta}^{\bot}\left(
q_{0}^{-1/2}\delta P_{0}|\left(  q_{0}^{1/2}\overline{Z}_{k_{0}}\right)
\right)  \right) \\
+E_{\theta}\left(  \Pi_{\theta}^{\bot}\left(  q_{01}^{1/2}\delta B_{1}|\left(
q_{01}^{1/2}\overline{W}_{k_{1}}\right)  \right)  \Pi_{\theta}^{\bot}\left(
q_{0}q_{01}^{-1/2}\delta P_{1}|\left(  q_{01}^{1/2}\overline{W}_{k_{1}%
}\right)  \right)  \right)
\end{array}
\right\}  +I\left(  m>2\right)  \times\nonumber\\
&  \left\{
\begin{array}
[c]{c}%
-E_{\theta}\left[
\begin{array}
[c]{c}%
\Pi_{\theta}^{\bot}\left[  \left(
\begin{array}
[c]{c}%
\frac{\pi_{1}^{1/2}}{\widehat{\pi}_{1}^{1/2}}\delta B_{1}\Pi_{\theta}\left(
q_{0}^{1/2}\delta P_{0}|\left(  q_{0}^{1/2}\overline{Z}_{k_{0}}\right)
\right)
\end{array}
\right)  |\left(  q_{01}^{1/2}\overline{W}_{k_{1}}\right)  \right] \\
\times\Pi_{\theta}^{\bot}\left[  q_{0}q_{01}^{-1/2}\delta P_{1}|\left(
q_{01}^{1/2}\overline{W}_{k_{1}}\right)  \right]
\end{array}
\right]  _{-\left(  TB\left(  m,2\right)  \right)  }\\
+\left(  -1\right)  ^{m}E_{\theta}\left[
\begin{array}
[c]{c}%
\Pi_{\theta}^{\bot}\left[  \left(
\begin{array}
[c]{c}%
q_{01}^{1/2}\delta B_{1}E_{\theta}\left[  \delta P_{0}\overline{Z}_{k_{0}}%
^{T}\right]  E_{\theta}\left[  q_{0}\overline{Z}_{k_{0}}\overline{Z}_{k_{0}%
}^{T}\right]  ^{-1}\\
\times\left(  E_{\theta}\left[  q_{0}\overline{Z}_{k_{0}}\overline{Z}_{k_{0}%
}^{T}-I\right]  \right)  ^{m-2}\overline{Z}_{k_{0}}%
\end{array}
\right)  |\left(  q_{01}^{1/2}\overline{W}_{k_{1}}\right)  \right] \\
\times\Pi_{\theta}^{\bot}\left[  q_{0}q_{01}^{-1/2}\delta P_{1}|\left(
q_{01}^{1/2}\overline{W}_{k_{1}}\right)  \right]
\end{array}
\right]  _{-\left(  TB\left(  m,3\right)  \right)  }%
\end{array}
\right\} \label{tbm}%
\end{align}%
\begin{align}
&  \left\vert BI_{m,2}\right\vert \nonumber\\
&  =O_{p}\left(  \max\left[
\begin{array}
[c]{c}%
k_{0}^{-\left(  \beta_{b_{0}}+\beta_{\pi_{0}}\right)  /d_{0}},k_{1}^{-\left(
\beta_{b_{1}}+\beta_{\pi_{1}}\right)  /d_{1}},k_{1}^{-\beta_{\pi_{1}}/d_{1}%
}k_{0}^{-\beta_{\pi_{0}}/d_{0}}\left(  \frac{\log n}{n}\right)  ^{\frac
{\beta_{b_{1}}}{d_{1}+\beta_{b_{1}}}}\\
k_{1}^{-\left(  \min\left(  \beta_{\pi_{0}},\beta_{b_{1}}\right)  +\beta
_{\pi_{1}}\right)  /d_{1}},\left(  \frac{\log n}{n}\right)  ^{-\frac
{\beta_{b_{1}}}{d_{1}+\beta_{b_{1}}}-\frac{\left(  m-2\right)  \beta_{g_{0}}%
}{d_{0}+\beta_{g_{0}}}}n^{-\frac{\beta_{\pi_{0}}}{d_{0}+\beta_{\pi_{0}}}}%
k_{1}^{-\beta_{\pi_{1}}/d_{1}}%
\end{array}
\right]  \right) \label{tbm2}%
\end{align}
Moreover,
\[
var\left(  \widehat{\psi}_{m}|\widehat{\theta}\right)  =O_{P}\left(  \frac
{1}{n}\max\left\{  1,\frac{1}{n^{m}}\max\left(  k_{0}^{m-1},k_{0}^{m-2}%
k_{1},\cdots,k_{0}k_{1}^{m-2},k_{1}^{m-1}\right)  \right\}  \right)
\]

\end{theorem}

The optimal choice of $k_{0}$ and $k_{1}$ in this class of higher order
estimators, will depend on the size of the effective smoothness exponents
$\frac{\beta_{b_{1}}}{d_{1}}$,$\frac{\beta_{\pi_{1}}}{d_{1}}$ and $\frac
{\beta_{g_{1}}}{d_{1}}$ relative to the size of the exponents $\frac
{\beta_{b_{0}}}{d_{0}}$,$\frac{\beta_{\pi_{0}}}{d_{0}}$ and $\frac
{\beta_{g_{0}}}{d_{0}}.$ \ A general prescription for finding an optimal
estimator in our class is to choose a pair $\left(  k_{0,opt}(m_{opt}%
),k_{1,opt}(m_{opt})\right)  $ that minimizes the maximum asymptotic MSE over
the model among the candidate $\widehat{\psi}_{m}=\widehat{\psi}_{m,\left(
k_{0},k_{1}\right)  }.$ Here, \ the estimator $\widehat{\psi}_{m,\left(
k_{0,opt}(m),k_{1,opt}(m)\right)  }$ uses the pair $\left(  k_{0,opt}%
(m),k_{1,opt}(m)\right)  $ of $m$ that equates the order of the variance to
the order of the maximum between the squared truncation and estimation biases
(which are given in the theorem)

\subsubsection{Three-occasion Monotone Missing Data}

Theorem \ref{map} can be applied in a nested fashion to the estimation of
general functionals in nonparametric models with monotone missingness (that is
with an arbitrary number of occasions). {}Building on the two-occasion case,
\[
O_{i}=\left(  L_{0},R_{0},R_{0}L_{1},R_{1},R_{0}R_{1}L_{2},R_{2},R_{0}%
R_{1}R_{2}Y\right)  ,
\]
where $R_{2}$ is the missing indicator for the third occasion. \ We also write
$\pi_{2}=\Pr\left(  R_{2}=1|R_{0}=R_{1}=1,L_{0},L_{1},L_{2}\right)  ,$ and
$B_{2}=E\left(  Y|L_{0},L_{1},L_{2}\right)  .$ Let
\[
\overline{V}_{k_{2}}\equiv\widehat{E}\left(  \Phi_{k_{2}}\left(  L_{0}%
,L_{1},L_{2}\right)  \Phi_{k_{2}}^{T}\left(  L_{0},L_{1},L_{2}\right)
\right)  ^{-1/2}\Phi_{k_{2}}\left(  L_{0},L_{1},L_{2}\right)
\]
and $\Phi_{k_{2}}\left(  L_{0},L_{1},L_{2}\right)  $ is a $k_{2}-$dimensional
vector of tensor product basis for functions of $\left(  L_{0},L_{1}%
,L_{2}\right)  $ with finite variance. The truncated parameter
$\widetilde{\psi}_{k_{0},k_{1},k_{2}}^{\left(  3\right)  }\left(
\theta\right)  $ is given by :%

\[
\widetilde{\psi}_{k_{0},k_{1},k_{2}}^{\left(  3\right)  }\left(
\theta\right)  =\psi\left(  \widehat{\theta}\right)  +\widetilde{\tau}%
_{1}^{\left(  3\right)  }\left(  \theta\right)  +%
%TCIMACRO{\tsum \limits_{j=2}^{m}}%
%BeginExpansion
{\textstyle\sum\limits_{j=2}^{m}}
%EndExpansion
\widetilde{\psi}_{j,j}^{\left(  3\right)  }\left(  \theta\right)
\]
where for all $1\leq s\leq k_{1},$ $1\leq t\leq k_{0},$%

\begin{align*}
\widetilde{\pi}_{2}^{\left(  3\right)  -1}  &  =\left\{  \widehat{\pi}%
_{2}^{-1}\left(
\begin{array}
[c]{c}%
1-\overline{V}_{k_{2}}^{T}E_{\theta}\left(  \frac{R_{2}}{\widehat{\pi}_{2}%
}\frac{R_{1}}{\widehat{\pi}_{1}}\frac{R_{0}}{\widehat{\pi}_{0}}\overline
{V}_{k_{2}}\overline{V}_{k_{2}}^{T}\right)  ^{-1}\\
\times E_{\theta}\left[  \frac{R_{0}}{\widehat{\pi}_{0}}\frac{R_{1}%
}{\widehat{\pi}_{1}}\left(  \frac{R_{2}}{\widehat{\pi}_{2}}-1\right)
\overline{V}_{k_{2}}\right]
\end{array}
\right)  \right\} \\
\widetilde{B}_{2}^{\left(  3\right)  }  &  =\left\{
\begin{array}
[c]{c}%
\widehat{B}_{2}+\overline{V}_{k_{2}}^{T}E_{\theta}\left(  \frac{R_{2}%
}{\widehat{\pi}_{2}}\frac{R_{1}}{\widehat{\pi}_{1}}\frac{R_{0}}{\widehat{\pi
}_{0}}\overline{V}_{k_{2}}\overline{V}_{k_{2}}^{T}\right)  ^{-1}\\
\times E_{\theta}\left(  \frac{R_{2}}{\widehat{\pi}_{2}}\frac{R_{1}%
}{\widehat{\pi}_{1}}\frac{R_{0}}{\widehat{\pi}_{0}}\left(  Y-\widehat{B}%
_{2}\right)  \overline{V}_{k_{2}}\right)
\end{array}
\right\} \\
\widetilde{W_{s}B}_{2}^{\left(  3\right)  }  &  =\left\{
\begin{array}
[c]{c}%
\widehat{B}_{2}W_{s}+\overline{V}_{k_{2}}^{T}E_{\theta}\left(  \frac{R_{2}%
}{\widehat{\pi}_{2}}\frac{R_{1}}{\widehat{\pi}_{1}}\frac{R_{0}}{\widehat{\pi
}_{0}}\overline{V}_{k_{2}}\overline{V}_{k_{2}}^{T}\right)  ^{-1}\\
\times E_{\theta}\left(  \frac{R_{2}}{\widehat{\pi}_{2}}\frac{R_{1}%
}{\widehat{\pi}_{1}}\frac{R_{0}}{\widehat{\pi}_{0}}\left(  Y-\widehat{B}%
_{2}\right)  W_{s}\overline{V}_{k_{2}}\right)
\end{array}
\right\} \\
\widetilde{Z_{t}W_{s}B}_{2}^{\left(  3\right)  }  &  =\left\{
\begin{array}
[c]{c}%
\widehat{B}_{2}Z_{t}W_{s}+\overline{V}_{k_{2}}^{T}E_{\theta}\left(
\frac{R_{2}}{\widehat{\pi}_{2}}\frac{R_{1}}{\widehat{\pi}_{1}}\frac{R_{0}%
}{\widehat{\pi}_{0}}\overline{V}_{k_{2}}\overline{V}_{k_{2}}^{T}\right)
^{-1}\\
\times E_{\theta}\left(  \frac{R_{2}}{\widehat{\pi}_{2}}\frac{R_{1}%
}{\widehat{\pi}_{1}}\frac{R_{0}}{\widehat{\pi}_{0}}\left(  Y-\widehat{B}%
_{2}\right)  Z_{t}W_{s}\overline{V}_{k_{2}}\right)
\end{array}
\right\} \\
\widetilde{Z_{t}B}_{2}^{\left(  3\right)  }  &  =\left\{
\begin{array}
[c]{c}%
\widehat{B}_{2}Z_{t}+\overline{V}_{k_{2}}^{T}E_{\theta}\left(  \frac{R_{2}%
}{\widehat{\pi}_{2}}\frac{R_{1}}{\widehat{\pi}_{1}}\frac{R_{0}}{\widehat{\pi
}_{0}}\overline{V}_{k_{2}}\overline{V}_{k_{2}}^{T}\right)  ^{-1}\\
\times E_{\theta}\left(  \frac{R_{2}}{\widehat{\pi}_{2}}\frac{R_{1}%
}{\widehat{\pi}_{1}}\frac{R_{0}}{\widehat{\pi}_{0}}\left(  Y-\widehat{B}%
_{2}\right)  Z_{t}\overline{V}_{k_{2}}\right)
\end{array}
\right\}
\end{align*}
and%
\[
\widetilde{\tau}_{1}^{\left(  3\right)  }\left(  \theta\right)  =\left\{
\begin{array}
[c]{c}%
E\left(  \frac{R_{2}}{\widetilde{\pi}_{2}^{\left(  3\right)  }}\frac{R_{1}%
}{\widehat{\pi}_{1}}\frac{R_{0}}{\widehat{\pi}_{0}}\left(  Y-\widetilde{B}%
_{2}^{\left(  3\right)  }\right)  +\frac{R_{1}}{\widehat{\pi}_{1}}\frac{R_{0}%
}{\widehat{\pi}_{0}}\left(  \widetilde{B}_{2}^{\left(  3\right)  }%
-\widehat{B}_{1}\right)  \right) \\
E\left(  \frac{R_{0}}{\widehat{\pi}_{0}}\left(  \widehat{B}_{1}-\widehat{B}%
_{0}\right)  +\widehat{B}_{0}-\psi\left(  \widehat{\theta}\right)  \right)
\end{array}
\right\}
\]

\begin{align*}
&  \widetilde{\psi}_{j,j}^{\left(  3\right)  }\left(  \theta\right) \\
&  =\left(  -1\right)  ^{j-1}\left[
\begin{array}
[c]{c}%
\left\{  \left(  A_{s}\right)  _{1\times k_{1}}\left[  E\left(  \frac{R_{0}%
}{\widehat{\pi}_{0}}\frac{R_{1}}{\widehat{\pi}_{1}}\overline{W}_{k_{1}%
}\overline{W}_{k_{1}}^{T}-I\right)  \right]  ^{j-2}E\left[  \overline
{W}_{k_{1}}\frac{R_{0}}{\widehat{\pi}_{0}}\left(  \frac{R_{1}}{\widehat{\pi
}_{1}}-1\right)  \right]  \right\} \\
+%
%TCIMACRO{\tsum \limits_{q=2}^{j-1}}%
%BeginExpansion
{\textstyle\sum\limits_{q=2}^{j-1}}
%EndExpansion
\left\{
\begin{array}
[c]{c}%
E\left(  \left(  \frac{R_{0}}{\widehat{\pi}_{0}}\left(  \frac{R_{1}%
}{\widehat{\pi}_{1}}-1\right)  \right)  \overline{W}_{k_{1}}^{T}\right)
\left[  E\left(  \frac{R_{1}}{\widehat{\pi}_{1}}\frac{R_{0}}{\widehat{\pi}%
_{0}}\overline{W}_{k_{1}}\overline{W}_{k_{1}}^{T}-I\right)  \right]
^{j-q-1}\\
\times\left(  B_{s,t}\right)  _{k_{1}\times k_{0}}\times\\
\left[  E\left(  \frac{R_{0}}{\widehat{\pi}_{0}}\overline{Z}_{k_{0}}%
\overline{Z}_{k_{0}}^{T}-I\right)  \right]  ^{q-2}E\left[  \left(  \frac
{R_{0}}{\widehat{\pi}_{0}}-1\right)  \overline{Z}_{k_{0}}\right]
\end{array}
\right\} \\
+\left\{  \left(  C_{t}\right)  _{1\times k_{0}}\times\left[  E\left(
\frac{R_{0}}{\widehat{\pi}_{0}}\overline{Z}_{k_{0}}\overline{Z}_{k_{0}}%
^{T}-I\right)  \right]  ^{j-2}E\left[  \left(  \frac{R_{0}}{\widehat{\pi}_{0}%
}-1\right)  \overline{Z}_{k_{0}}\right]  \right\}
\end{array}
\right]
\end{align*}

where
\begin{align*}
A_{s}  &  =E\left(  \frac{R_{2}}{\widetilde{\pi}_{2}^{\left(  3\right)  }%
}\frac{R_{0}}{\widehat{\pi}_{0}}\frac{R_{1}}{\widehat{\pi}_{1}}\left(
YW_{s}-\widetilde{W_{s}B}_{2}^{\left(  3\right)  }\right)  +\frac{R_{0}%
}{\widehat{\pi}_{0}}\frac{R_{1}}{\widehat{\pi}_{1}}\left(  \widetilde{W_{s}%
B}_{2}^{\left(  3\right)  }-\widehat{B}_{1}W_{s}\right)  \right) \\
B_{s,t}  &  =E\left(  \frac{R_{2}}{\widetilde{\pi}_{2}^{\left(  3\right)  }%
}\frac{R_{0}}{\widehat{\pi}_{0}}\frac{R_{1}}{\widehat{\pi}_{1}}\left(
YW_{s}Z_{t}-\widetilde{Z_{t}W_{s}B}_{2}^{\left(  3\right)  }\right)
+\frac{R_{0}}{\widehat{\pi}_{0}}\frac{R_{1}}{\widehat{\pi}_{1}}\left(
\widetilde{Z_{t}W_{s}B}_{2}^{\left(  3\right)  }-\widehat{B}_{1}W_{s}%
Z_{t}\right)  \right) \\
C_{t}  &  =E\left(  \frac{R_{2}}{\widetilde{\pi}_{2}^{\left(  3\right)  }%
}\frac{R_{0}}{\widehat{\pi}_{0}}\frac{R_{1}}{\widehat{\pi}_{1}}\left(
YZ_{t}-\widetilde{Z_{t}B}_{2}^{\left(  3\right)  }\right)  +\frac{R_{0}%
}{\widehat{\pi}_{0}}\frac{R_{1}}{\widehat{\pi}_{1}}\left(  \widetilde{Z_{t}%
B}_{2}^{\left(  3\right)  }-\widehat{B}_{1}Z_{t}\right)  \right) \\
&  +E\left(  \frac{R_{0}}{\widehat{\pi}_{0}}\left(  \widehat{B}_{1}%
-\widehat{B}_{0}\right)  Z_{t}\right)
\end{align*}
\newline We only give out the final expression for $IF_{r,r,\widetilde{\psi
}_{k_{0},k_{1},k_{2}}^{\left(  3\right)  }}\left(  \widehat{\theta}\right)  $
without any technical details; bias and variance properties of this estimator
will be published elsewhere.%

\begin{align*}
&  \left(  -1\right)  ^{r-1}IF_{r,r,\widetilde{\psi}_{k_{0},k_{1},k_{2}%
}^{\left(  3\right)  }\left(  \theta\right)  }\left(  \widehat{\theta}\right)
\\
&  =\left[  \frac{R_{0}R_{1}R_{2}}{\widehat{\pi}_{0}\widehat{\pi}%
_{1}\widehat{\pi}_{2}}\left(  Y-\widehat{B}_{2}\right)  \overline{V}_{k_{2}%
}^{T}\right]
%TCIMACRO{\dprod \limits_{q=2}^{r-1}}%
%BeginExpansion
{\displaystyle\prod\limits_{q=2}^{r-1}}
%EndExpansion
\left(  \frac{R_{0}R_{1}R_{2}}{\widehat{\pi}_{0}\widehat{\pi}_{1}\widehat{\pi
}_{2}}\overline{V}_{k_{2}}\overline{V}_{k_{2}}^{T}-I\right)  _{i_{q}}\left[
\overline{V}_{k_{2}}\frac{R_{0}R_{1}}{\widehat{\pi}_{0}\widehat{\pi}_{1}%
}\left(  \frac{R_{2}}{\widehat{\pi}_{2}}-1\right)  \right]  _{i_{r}}\\
&  +%
%TCIMACRO{\dsum \limits_{m=2}^{r-1}}%
%BeginExpansion
{\displaystyle\sum\limits_{m=2}^{r-1}}
%EndExpansion
\left\{
\begin{array}
[c]{c}%
\begin{array}
[c]{c}%
\left[  \frac{R_{0}}{\widehat{\pi}_{0}}\left(  \frac{R_{1}}{\widehat{\pi}_{1}%
}-1\right)  \right]  _{i_{1}}\overline{W}_{k_{1}}^{T}%
%TCIMACRO{\dprod \limits_{q=2}^{m-1}}%
%BeginExpansion
{\displaystyle\prod\limits_{q=2}^{m-1}}
%EndExpansion
\left(  \frac{R_{0}R_{1}}{\widehat{\pi}_{0}\widehat{\pi}_{1}}\overline
{W}_{k_{1}}\overline{W}_{k_{1}}^{T}-I\right)  _{i_{q}}\times\\
\left[  \frac{R_{0}R_{1}R_{2}}{\widehat{\pi}_{0}\widehat{\pi}_{1}\widehat{\pi
}_{2}}\left(  Y-\widehat{B}_{2}\right)  \overline{W}_{k_{1}}\overline
{V}_{k_{2}}^{T}\right]  _{i_{m}}%
%TCIMACRO{\dprod \limits_{s=m+1}^{r-1}}%
%BeginExpansion
{\displaystyle\prod\limits_{s=m+1}^{r-1}}
%EndExpansion
\left(  \frac{R_{0}R_{1}R_{2}}{\widehat{\pi}_{0}\widehat{\pi}_{1}\widehat{\pi
}_{2}}\overline{V}_{k_{2}}\overline{V}_{k_{2}}^{T}-I\right)  _{i_{s}}\times\\
\left[  \overline{V}_{k_{2}}\frac{R_{0}R_{1}}{\widehat{\pi}_{0}\widehat{\pi
}_{1}}\left(  \frac{R_{2}}{\widehat{\pi}_{2}}-1\right)  \right]  _{i_{r}}%
\end{array}
\\
+%
%TCIMACRO{\dsum \limits_{j=2}^{m-1}}%
%BeginExpansion
{\displaystyle\sum\limits_{j=2}^{m-1}}
%EndExpansion
\left[
\begin{array}
[c]{c}%
\left(  \frac{R_{0}}{\widehat{\pi}_{0}}-1\right)  _{i_{1}}\overline{Z}%
_{k_{0},i_{1}}^{T}%
%TCIMACRO{\dprod \limits_{s=2}^{j-1}}%
%BeginExpansion
{\displaystyle\prod\limits_{s=2}^{j-1}}
%EndExpansion
\left(  \frac{R_{0}}{\widehat{\pi}_{0}}\overline{Z}_{k_{0}}\overline{Z}%
_{k_{0}}^{T}-I\right)  _{i_{s}}\left[  \frac{R_{0}R_{1}R_{2}}{\widehat{\pi
}_{0}\widehat{\pi}_{1}\widehat{\pi}_{2}}\left(  Y-\widehat{B}_{2}\right)
\overline{Z}_{k_{0}}\overline{V}_{k_{2}}^{T}\right]  _{i_{j}}\\
\times%
%TCIMACRO{\dprod \limits_{q=j+1}^{r+j-m-1}}%
%BeginExpansion
{\displaystyle\prod\limits_{q=j+1}^{r+j-m-1}}
%EndExpansion
\left(  \frac{R_{0}R_{1}R_{2}}{\widehat{\pi}_{0}\widehat{\pi}_{1}\widehat{\pi
}_{2}}\overline{V}_{k_{2}}\overline{V}_{k_{2}}^{T}-I\right)  _{i_{q}}\left[
\overline{V}_{k_{2}}\frac{R_{0}R_{1}}{\widehat{\pi}_{0}\widehat{\pi}_{1}%
}\left(  \frac{R_{2}}{\widehat{\pi}_{2}}-1\right)  \overline{W}_{k_{1}}%
^{T}\right]  _{i_{r+j-m}}\times\\%
%TCIMACRO{\dprod \limits_{q=r+j-m+1}^{r-1}}%
%BeginExpansion
{\displaystyle\prod\limits_{q=r+j-m+1}^{r-1}}
%EndExpansion
\left(  \frac{R_{0}R_{1}}{\widehat{\pi}_{0}\widehat{\pi}_{1}}\overline
{W}_{k_{1}}\overline{W}_{k_{1}}^{T}-I\right)  _{i_{q}}\left[  \overline
{W}_{k_{1}}\frac{R_{0}}{\widehat{\pi}_{0}}\left(  \frac{R_{1}}{\widehat{\pi
}_{1}}-1\right)  \right]  _{i_{r}}%
\end{array}
\right] \\
+\left(  \frac{R_{0}}{\widehat{\pi}_{0}}-1\right)  _{i_{1}}\overline{Z}%
_{k_{0},i_{1}}^{T}%
%TCIMACRO{\dprod \limits_{s=2}^{m-1}}%
%BeginExpansion
{\displaystyle\prod\limits_{s=2}^{m-1}}
%EndExpansion
\left(  \frac{R_{0}}{\widehat{\pi}_{0}}\overline{Z}_{k_{0}}\overline{Z}%
_{k_{0}}^{T}-I\right)  _{i_{s}}\left[  \overline{Z}_{k_{0}}\frac{R_{0}%
R_{1}R_{2}}{\widehat{\pi}_{0}\widehat{\pi}_{1}\widehat{\pi}_{2}}\left(
Y-\widehat{B}_{2}\right)  \overline{V}_{k_{2}}^{T}\right]  _{i_{m}}\times\\%
%TCIMACRO{\dprod \limits_{q=m+1}^{r-1}}%
%BeginExpansion
{\displaystyle\prod\limits_{q=m+1}^{r-1}}
%EndExpansion
\left(  \frac{R_{0}R_{1}R_{2}}{\widehat{\pi}_{0}\widehat{\pi}_{1}\widehat{\pi
}_{2}}\overline{V}_{k_{2}}\overline{V}_{k_{2}}^{T}-I\right)  \left[
\overline{V}_{k_{2}}\frac{R_{0}R_{1}}{\widehat{\pi}_{0}\widehat{\pi}_{1}%
}\left(  \frac{R_{2}}{\widehat{\pi}_{2}}-1\right)  \right]  _{i_{r}}%
\end{array}
\right\}
\end{align*}%
\[
+\left\{
\begin{array}
[c]{c}%
\left[
\begin{array}
[c]{c}%
\frac{R_{0}R_{1}R_{2}}{\widehat{\pi}_{0}\widehat{\pi}_{1}\widehat{\pi}_{2}%
}\left(  Y-\widehat{B}_{2}\right)  +\\
\frac{R_{0}R_{1}}{\widehat{\pi}_{0}\widehat{\pi}_{1}}\left(  \widehat{B}%
_{2}-\widehat{B}_{1}\right)
\end{array}
\right]  _{i_{1}}\overline{W}_{k_{1},i_{1}}^{T}%
%TCIMACRO{\dprod \limits_{q=2}^{r-1}}%
%BeginExpansion
{\displaystyle\prod\limits_{q=2}^{r-1}}
%EndExpansion
\left(  \frac{R_{0}R_{1}}{\widehat{\pi}_{0}\widehat{\pi}_{1}}\overline
{W}_{k_{1}}\overline{W}_{k_{1}}^{T}-I\right)  _{i_{q}}\left[  \overline
{W}_{k_{1}}\frac{R_{0}}{\widehat{\pi}_{0}}\left(  \frac{R_{1}}{\widehat{\pi
}_{1}}-1\right)  \right]  _{i_{r}}\\
+%
%TCIMACRO{\dsum \limits_{j=2}^{r-1}}%
%BeginExpansion
{\displaystyle\sum\limits_{j=2}^{r-1}}
%EndExpansion
\left(  \frac{R_{0}}{\widehat{\pi}_{0}}-1\right)  _{i_{1}}\overline{Z}%
_{k_{0},i_{1}}^{T}%
%TCIMACRO{\dprod \limits_{s=2}^{j-1}}%
%BeginExpansion
{\displaystyle\prod\limits_{s=2}^{j-1}}
%EndExpansion
\left(  \frac{R_{0}}{\widehat{\pi}_{0}}\overline{Z}_{k_{0}}\overline{Z}%
_{k_{0}}^{T}-I\right)  _{i_{s}}\left[  \left(
\begin{array}
[c]{c}%
\frac{R_{0}R_{1}R_{2}}{\widehat{\pi}_{0}\widehat{\pi}_{1}\widehat{\pi}_{2}%
}\left(  Y-\widehat{B}_{2}\right) \\
+\frac{R_{0}R_{1}}{\widehat{\pi}_{0}\widehat{\pi}_{1}}\left(  \widehat{B}%
_{2}-\widehat{B}_{1}\right)
\end{array}
\right)  \overline{Z}_{k_{0}}\overline{W}_{k_{1}}^{T}\right]  _{i_{j}}\times\\%
%TCIMACRO{\dprod \limits_{q=j+1}^{r-1}}%
%BeginExpansion
{\displaystyle\prod\limits_{q=j+1}^{r-1}}
%EndExpansion
\left(  \frac{R_{0}R_{1}}{\widehat{\pi}_{0}\widehat{\pi}_{1}}\overline
{W}_{k_{1}}\overline{W}_{k_{1}}^{T}-I\right)  _{i_{q}}\left[  \overline
{W}_{k_{1}}\frac{R_{0}}{\widehat{\pi}_{0}}\left(  \frac{R_{1}}{\widehat{\pi
}_{1}}-1\right)  \right]  _{i_{r}}\\
+\left[  \left(
\begin{array}
[c]{c}%
\frac{R_{0}R_{1}R_{2}}{\widehat{\pi}_{0}\widehat{\pi}_{1}\widehat{\pi}_{2}%
}\left(  Y-\widehat{B}_{2}\right) \\
+\frac{R_{0}R_{1}}{\widehat{\pi}_{0}\widehat{\pi}_{1}}\left(  \widehat{B}%
_{2}-\widehat{B}_{1}\right) \\
+\frac{R_{0}}{\widehat{\pi}_{0}}\left(  \widehat{B}_{1}-\widehat{B}%
_{0}\right)
\end{array}
\right)  \right]  _{i_{1}}\overline{Z}_{k_{0},i_{1}}^{T}%
%TCIMACRO{\dprod \limits_{s=2}^{r-1}}%
%BeginExpansion
{\displaystyle\prod\limits_{s=2}^{r-1}}
%EndExpansion
\left(  \frac{R_{0}}{\widehat{\pi}_{0}}\overline{Z}_{k_{0}}\overline{Z}%
_{k_{0}}^{T}-I\right)  _{i_{s}}\overline{Z}_{k_{0},i_{r}}\left(  \frac{R_{0}%
}{\widehat{\pi}_{0}}-1\right)  _{i_{r}}%
\end{array}
\right\}
\]

\bigskip

\section{Appendix}

In the following, we assume all parametric submodels are sufficiently smooth
and regular that expectation and differentiation operators commute as needed.
We also define $\mathbb{IF}_{1,1}$ to be $\mathbb{IF}_{1}.$

\begin{proof}
(Theorem \ref{eiet}) Define the bias function $B_{m}\left[  \theta^{\dagger
},\theta\right]  $ of $\mathbb{IF}_{m}\left(  \theta\right)  $ to be
$E_{\theta^{\dagger}}\left[  \mathbb{IF}_{m}\left(  \theta\right)  \right]  .$
$\ $\ Define
\[
B_{m,l_{1}^{\ast}...l_{j}^{\ast}l_{j+1},...l_{s}}\left[  \theta,\theta\right]
=\partial^{s}B_{m}\left[  \widetilde{\theta}\left(  \varsigma^{\ast}\right)
,\widetilde{\theta}\left(  \varsigma\right)  \right]  /\partial\varsigma
_{l_{1}...}^{\ast}\partial\varsigma_{l_{j}}^{\ast}\partial\varsigma
_{l_{j+1}...}\partial\varsigma_{l_{s}}|_{\varsigma^{\ast}=\widetilde{\theta
}^{-1}\left\{  \theta\right\}  ,\varsigma=\widetilde{\theta}^{-1}\left\{
\theta\right\}  }%
\]
where we reserve $\ast$ for differentiation with respect to the first argument
of $B_{m}\left[  \cdot,\cdot\right]  .$ Thus for $s\leq m,$
\[
\psi_{\backslash l_{1}...l_{s}}\left(  \theta\right)  =B_{m,l_{1}^{\ast
}...l_{s}^{\ast}}\left[  \theta,\theta\right]
\]
To prove the theorem we will first need to show that:
\begin{equation}
B_{m,l_{1}^{\ast}...l_{j}^{\ast}l_{j+1},...l_{s}}\left[  \theta,\theta\right]
=0\text{ for }m\geq s>j>0\label{interinfo0}%
\end{equation}
To this end note that for $j<m,$
\begin{align*}
\psi_{\backslash l_{1}...l_{j+1}}\left(  \theta\right)   &  =\partial
\psi_{\backslash l_{1}...l_{j}}\left(  \theta\right)  /\partial\varsigma
_{l_{j+1}}=\partial B_{m,l_{1}^{\ast}...l_{j}^{\ast}}\left[  \theta
,\theta\right]  /\partial\varsigma_{l_{j+1}}\\
&  =B_{m,l_{1}^{\ast}...l_{j}^{\ast}l_{j+1}^{\ast}}\left[  \theta
,\theta\right]  +B_{m,l_{1}^{\ast}...l_{j}^{\ast}l_{j+1}}\left[  \theta
,\theta\right] \\
&  =\psi_{\backslash l_{1}...l_{j+1}}\left(  \theta\right)  +B_{m,l_{1}^{\ast
}...l_{j}^{\ast}l_{j+1}}\left[  \theta,\theta\right]
\end{align*}
where the 2nd equality is by the definition of $\mathbb{IF}_{m}\left(
\theta\right)  $, the third is by the chain rule, and the fourth is again by
the definition of $\mathbb{IF}_{m}\left(  \theta\right)  $. \ Hence
$B_{m,l_{1}^{\ast}...l_{j}^{\ast}l_{j+1}}\left[  \theta,\theta\right]  =0.$
Hence for $j\leq m-2,$
\begin{align*}
0  &  =\partial B_{m,l_{1}^{\ast}...l_{j}^{\ast}l_{j+1}}\left[  \theta
,\theta\right]  /\partial\varsigma_{l_{j+2}..}=B_{m,l_{1}^{\ast}...l_{j}%
^{\ast}l_{j+2}^{\ast}l_{j+1}}\left[  \theta,\theta\right]  +B_{m,l_{1}^{\ast
}...l_{j}^{\ast}l_{j+1}l_{j+2}}\left[  \theta,\theta\right] \\
&  =0+B_{m,l_{1}^{\ast}...l_{j}^{\ast}l_{j+1}l_{j+2}}\left[  \theta
,\theta\right]
\end{align*}
$\ $where the last equality holds because we just proved $B_{m,l_{1}^{\ast
}...l_{j}^{\ast}l_{j+1}}\left[  \theta,\theta\right]  =0$ for arbitrary
indices. \ Iterating this argument proves $\left(  \ref{interinfo0}\right)  $.
We complete the proof by induction on $s$ for some $s<m.$ Given a $s=1$
dimensional regular parametric submodel $\widetilde{\theta}\left(
\varsigma\right)  $, $E_{\theta\left(  \varsigma\right)  }\left[
\mathbb{IF}_{m}\left(  \theta\left(  \varsigma\right)  \right)  \right]  =0$
by assumption. Hence, by regularity of the model, $0=B_{m,l_{1}^{\ast}%
.}\left[  \theta,\theta\right]  +B_{m,l_{1}.}\left[  \theta,\theta\right]  .$
Therefore $B_{m,l_{1}.}\left[  \theta,\theta\right]  =-\psi_{\backslash l_{1}%
}\left(  \theta\right)  .$ Now suppose the theorem if true for $s.$ Then
\begin{align*}
-\psi_{\backslash l_{1}...l_{s+1}}\left(  \theta\right)   &  =-\partial
\psi_{\backslash l_{1}...l_{s}}\left(  \theta\right)  /\partial\varsigma
_{l_{s+1}}=\partial B_{m,l_{1}...l_{s}}\left[  \theta,\theta\right]
/\partial\varsigma_{l_{s+1}.}\\
&  =B_{m,l_{s+1}^{\ast}l_{1}...l_{s}}\left[  \theta,\theta\right]
+B_{m,l_{1}...l_{s+1}}\left[  \theta,\theta\right]  =0+B_{m,l_{1}...l_{s+1}%
}\left[  \theta,\theta\right]  \
\end{align*}
where the second equality is by the induction assumption, the third by the
chain rule, and the last by equation $\left(  \ref{interinfo0}\right)  $
\end{proof}

\begin{proof}
(Theorem \ref{eift}) (1): Consider two influence functions $\mathbb{IF}%
_{m}^{\left(  1\right)  }\left(  \theta\right)  $ and $\mathbb{IF}%
_{m}^{\left(  2\right)  }\left(  \theta\right)  $ for $\psi\left(
\theta\right)  .$ Then $E_{\theta}\left[  \left\{  \mathbb{IF}_{m}^{\left(
1\right)  }\left(  \theta\right)  -\mathbb{IF}_{m}^{\left(  2\right)  }\left(
\theta\right)  \right\}  \widetilde{\mathbb{S}}_{s,\overline{l}_{s}}\left(
\theta\right)  \right]  =\psi_{\backslash\overline{l}_{s}}\left(
\theta\right)  -\psi_{\backslash\overline{l}_{s}}\left(  \theta\right)
=0\ $for any score $\widetilde{\mathbb{S}}_{s,\overline{l}_{s}}\left(
\theta\right)  ,s\leq m\ $\ and hence for any linear combination of scores.
But, by definition, linear combinations of scores are dense in $\Gamma
_{m}\left(  \theta\right)  .$ Thus $\mathbb{IF}_{m}^{\left(  1\right)
}\left(  \theta\right)  $ and $\mathbb{IF}_{m}^{\left(  2\right)  }\left(
\theta\right)  $ have the same projection on $\Gamma_{m}\left(  \theta\right)
.$ (2-3): Essentially immediate from the definitions. (4): For $t\leq
s,\psi_{\backslash\overline{l}_{t}}\left(  \theta\right)  =E_{\theta}\left[
\mathbb{IF}_{m}\left(  \theta\right)  \widetilde{\mathbb{S}}_{t,\overline
{l}_{t}}\left(  \theta\right)  \right]  =E_{\theta}\left[  \Pi_{m,\theta
}\left[  \mathbb{IF}_{m}\left(  \theta\right)  |\mathcal{U}_{t}\left(
\theta\right)  \right]  \widetilde{\mathbb{S}}_{t,\overline{l}_{t}}\left(
\theta\right)  \right]  $ for any $\widetilde{\mathbb{S}}_{t,\overline{l}_{t}%
}\left(  \theta\right)  $. (5.a): follows from (1). (5.b): follows from (4).
Degeneracy of $\mathbb{IF}_{mm}\left(  \theta\right)  $ follows at once from
the fact that $\mathbb{IF}_{mm}\left(  \theta\right)  \in\mathcal{U}%
_{m-1}\left(  \theta\right)  ^{\perp}$ in $\mathcal{U}_{m}\left(
\theta\right)  $. \newline Proof of part (5.c) requires the following.
\end{proof}

\begin{lemma}
\label{joel2}Suppose, for $m\geq1,$ $\mathbb{IF}_{m,m}\left(  \theta\right)  $
and $if_{1,if_{m\ ,m\ }\left(  O_{i_{1},}..O_{i_{m}};\cdot\right)  }\left(
O_{i_{m+1}};\theta\right)  $ exist w.p.1 for a kernel $IF_{m,m}\left(
\theta\right)  .$ Then, (i):$if_{1,if_{m\ ,m\ }\left(  O_{i_{1},}..O_{i_{m}%
};\cdot\right)  }\left(  O_{i_{m+1}};\theta\right)  s_{l_{t}}\left(
O_{i_{m+1}}\right)  $\newline$,-if_{m\ ,m,\backslash l_{t}\ }\left(  O_{i_{1}%
},...,O_{i_{m\ }};\theta\right)  ,$ and $if_{m\ ,m\ }\left(  O_{i_{1}%
},...,O_{i_{m}};\cdot\right)  s_{l_{t}}\left(  O_{i_{_{m}}}\right)  $ each
have the same mean given $O_{i_{1}},...,O_{i_{m-1}}$, (ii) $E\left[
if_{m\ ,m,\backslash l_{t}\ }\left(  O_{i_{1},}..O_{i_{m}};\theta\right)
|O_{i_{1}},...,O_{i_{m-2}}\right]  =0,$ \newline(iii) $E_{\theta}\left[
if_{1,if_{m\ ,m\ }\left(  O_{i_{1},}..O_{i_{m}};\cdot\right)  }\left(
O_{i_{m+1}};\theta\right)  |O_{i_{1}},...,O_{i_{m-2}},O_{i_{m+1}}\right]  =0,$
so
\begin{align*}
&  \Pi\left[  \mathbb{V}\left[  if_{1,if_{m\ ,m\ }\left(  O_{i_{1},}%
..O_{i_{m}};\cdot\right)  }\left(  O_{i_{m+1}};\theta\right)  \right]
|\mathcal{U}_{m}\left(  \theta\right)  \right] \\
&  =\Pi\left[  \mathbb{V}\left[  if_{1,if_{m\ ,m\ }\left(  O_{i_{1}%
,}..O_{i_{m}};\cdot\right)  }\left(  O_{i_{m+1}};\theta\right)  \right]
|\mathcal{U}_{m}\left(  \theta\right)  \cap\mathcal{U}_{m-2}^{\bot}\left(
\theta\right)  \right]
\end{align*}
\qquad\ and (iv) $\mathbb{IF}_{m,m,\backslash l_{t}}\left(  \theta\right)  $
satisfies $\Pi_{\theta}\left[  \mathbb{IF}_{m,m,\backslash l_{t}}\left(
\theta\right)  |\mathcal{U}_{m-2}\left(  \theta\right)  \right]  =0\ $and
\[
\Pi_{\theta}\left[  \mathbb{IF}_{m,m,\backslash l_{t}}\left(  \theta\right)
|\mathcal{U}_{m-1}\left(  \theta\right)  \right]  =-\mathbb{V}\left[
mE_{\theta}\left[  IF_{m,m,\backslash l_{t},\overline{i}_{m}}^{sym}\left(
\theta\right)  |O_{i_{1}},...,O_{i_{m-1}}\right]  \right]  .
\]

\end{lemma}

\begin{proof}
(i):By $IF_{m,m}\left(  \theta\right)  $ degenerate$,$\newline$E_{\theta
}\left[  IF_{m,m,\backslash l_{t},\overline{i}_{m}}\left(  \theta\right)
|O_{i_{1}},...,O_{i_{m-1}}\right]  =-E_{\theta}\left[  IF_{m,m,\overline
{i}_{m}}\left(  \theta\right)  s_{l_{t}}\left(  O_{i_{m}}\right)  |O_{i_{1}%
},...,O_{i_{m-1}}\right]  .$ Further, by definition,
\begin{align*}
&  E_{\theta}\left[  if_{1,if_{m\ ,m\ }^{sym}\left(  O_{i_{1},}..O_{i_{m}%
};\cdot\right)  }\left(  O_{i_{m+1}};\theta\right)  s_{l_{t}}\left(
O_{i_{m+1}}\right)  |O_{i_{1}},...,O_{i_{m}}\right] \\
&  =E_{\theta}\left[  IF_{m,m,\backslash l_{t},\overline{i}_{m}}\left(
\theta\right)  |O_{i_{1}},...,O_{i_{m}}\right]  .
\end{align*}
(ii): By $IF_{m,m}\left(  \theta\right)  $ degenerate $0=E_{\theta}\left[
IF_{m,m,\overline{i}_{m}}\left(  \theta\right)  s_{l_{t}}\left(  O_{i_{m}%
}\right)  |O_{i_{1}},...,O_{i_{m-2}}\right]  $ wp1 and so (ii) follows from
(i).\newline(iii): (i) and (ii) imply
\begin{align*}
0  &  =E_{\theta}\left[  if_{1,if_{m\ ,m\ }\left(  O_{i_{1},}..O_{i_{m}}%
;\cdot\right)  }\left(  O_{i_{m+1}};\theta\right)  s_{l_{t}}\left(
O_{i_{m+1}}\right)  |O_{i_{1}},...,O_{i_{m-2}}\right] \\
&  =E_{\theta}\left[
\begin{array}
[c]{c}%
E_{\theta}\left\{  if_{1,if_{m\ ,m\ }\left(  O_{i_{1},}..O_{i_{m}}%
;\cdot\right)  }\left(  O_{i_{m+1}};\theta\right)  |\left(  O_{i_{m+1}%
},O_{i_{1}},...,O_{i_{m-2}}\right)  \right\} \\
\times s_{l_{t}}\left(  O_{i_{m+1}}\right)  |O_{i_{1}},...,O_{i_{m-2}}%
\end{array}
\right]
\end{align*}
But, by $s_{l_{t}}\left(  O_{i_{m+1}}\right)  $ an arbitrary mean zero
function,
\begin{align*}
&  E_{\theta}\left\{  if_{1,if_{m\ ,m\ }\left(  O_{i_{1},}..O_{i_{m}}%
;\cdot\right)  }\left(  O_{i_{m+1}};\theta\right)  |\left(  O_{i_{m+1}%
},O_{i_{1}},...,O_{i_{m-2}}\right)  \right\} \\
&  =E_{\theta}\left\{  if_{1,if_{m\ ,m\ }\left(  O_{i_{1},}..O_{i_{m}}%
;\cdot\right)  }\left(  O_{i_{m+1}};\theta\right)  |O_{i_{1}},...,O_{i_{m-2}%
}\right\}  =0
\end{align*}
(iv): By definition, $\Pi_{\theta}\left[  \mathbb{IF}_{m,m,\backslash l_{t}%
}\left(  \theta\right)  |\mathcal{U}_{m-1}\left(  \theta\right)  \right]
=\mathbb{V}\left[  \left\{  I-d_{m,\theta}\right\}  \left\{
IF_{m,m,\backslash l_{t},\overline{i}_{m}}\left(  \theta\right)  \right\}
\right]  .$ The result follows by Eq. $\left(  \ref{deg}\right)  $ and part (ii).
\end{proof}

\begin{proof}
\textbf{Theorem 5c(ii): }Consider a $m$-dimensional parametric submodel
$f\left(  O;\widetilde{\theta}\left(  \zeta\right)  \right)  =f\left(
O;\theta\right)  \left\{  1+\sum_{l=1}^{m}\zeta_{j}a_{j}\left(  O\right)
\right\}  ,$ $\zeta^{T}=\left(  \zeta_{1},...,\zeta_{m}\right)  ,$ with
$E_{\theta}\left[  a_{l}\left(  O\right)  \right]  =0.$ Since this model is
linear in the $\zeta_{j},$ $f_{/l_{1}...l_{m}}\left(  O_{j};\theta\right)  =0$
for $m>1$. Hence $\widetilde{\mathbb{S}}_{m,\overline{l}_{m}}\left(
\theta\right)  \ $is degenerate of order $m$, i.e., $\widetilde{\mathbb{S}%
}_{m,\overline{l}_{m}}\left(  \theta\right)  \in\mathcal{U}_{m-1}^{\bot
}\left(  \theta\right)  .$ Since $\mathbb{IF}_{m-1}\left(  \theta\right)  $
exists, on setting $l_{s}=s\ $for $s=1,...,m,$
\[
\partial^{m-1}\psi\left(  \widetilde{\theta}\left(  \zeta\right)  \right)
/\prod\limits_{j=1}^{m-1}\partial\zeta_{j|\zeta=0}\equiv\psi_{\backslash
\overline{l}_{m-1}}\left(  \theta\right)  =E_{\theta}\left[  \mathbb{IF}%
_{m-1}\left(  \theta\right)  \widetilde{\mathbb{S}}_{m-1,\overline{l}_{m-1}%
}\left(  \theta\right)  \right]  .
\]
Differentiating the last display with respect to $\zeta_{m}$ and evaluating at
$\zeta=0$, we obtain
\begin{align*}
\psi_{\backslash\overline{l}_{m}}\left(  \theta\right)   &  =E_{\theta}\left[
\mathbb{IF}_{m-1}\left(  \theta\right)  \widetilde{\mathbb{S}}_{m,\overline
{l}_{m}}\left(  \theta\right)  \right]  +E_{\theta}\left[  \mathbb{IF}%
_{m-1_{,}\backslash l_{m}}\left(  \theta\right)  \widetilde{\mathbb{S}%
}_{m-1,\overline{l}_{m-1}}\left(  \theta\right)  \right] \\
&  =E_{\theta}\left[  \mathbb{IF}_{m-1_{,}\backslash l_{m}}\left(
\theta\right)  \widetilde{\mathbb{S}}_{m-1,\overline{l}_{m-1}}\left(
\theta\right)  \right]
\end{align*}
Now $E_{\theta}\left[  \mathbb{IF}_{m-1_{,}\backslash l_{m}}\left(
\theta\right)  \widetilde{\mathbb{S}}_{m-1,\overline{l}_{m-1}}\left(
\theta\right)  \right]  $\newline$=E_{\theta}\left[  \mathbb{IF}%
_{m-2_{,}\backslash l_{m}}\left(  \theta\right)  \widetilde{\mathbb{S}%
}_{m-1,\overline{l}_{m-1}}\left(  \theta\right)  \right]  +E_{\theta}\left[
\mathbb{IF}_{m-1_{,}m-1,\backslash l_{m}}\left(  \theta\right)
\widetilde{\mathbb{S}}_{m-1,\overline{l}_{m-1}}\left(  \theta\right)  \right]
.$ Setting $s_{l_{r}}\left(  O_{i_{r}},\theta\right)  =a_{r}\left(  O_{i_{r}%
}\right)  ,$ $\widetilde{\mathbb{S}}_{m-1,\overline{l}_{m-1}}\left(
\theta\right)  =\sum_{i_{1}\neq...\neq i_{m-1}}\prod\limits_{r=1}^{m-1}%
a_{r}\left(  O_{i_{r}},\theta\right)  $ is degenerate of order $m-1$ so
\begin{align*}
&  E_{\theta}\left[  \mathbb{IF}_{m-1_{,}m-1,\backslash l_{m}}\left(
\theta\right)  \widetilde{\mathbb{S}}_{m-1,\overline{l}_{m-1}}\left(
\theta\right)  \right] \\
&  =\left(  m-1\right)  !E_{\theta}\left(  \left[  if_{m-1,m-1}^{sym}%
,_{\backslash l_{m}}\left(  O_{i_{1},}..O_{i_{m-1}};\theta\right)  \right]
\prod\limits_{r=1}^{m-1}a_{r}\left(  O_{i_{r}},\theta\right)  \right)
\end{align*}
and $E_{\theta}\left[  \mathbb{IF}_{m-2_{,}\backslash l_{m}}\left(
\theta\right)  \widetilde{\mathbb{S}}_{m-1,\overline{l}_{m-1}}\left(
\theta\right)  \right]  =0.$ Hence
\[
\psi_{\backslash\overline{l}_{m}}\left(  \theta\right)  =\left(  m-1\right)
!E_{\theta}\left(  if_{m-1,m-1}^{sym},_{\backslash l_{m}}\left(  O_{i_{1}%
,}..O_{i_{m-1}};\theta\right)  \prod\limits_{r=1}^{m-1}a_{r}\left(  O_{i_{r}%
},\theta\right)  \right)
\]
Now, by the assumed existence of $\mathbb{IF}_{m}\left(  \theta\right)  $, we
also have $\psi_{\backslash\overline{l}_{m}}\left(  \theta\right)  =E_{\theta
}\left[  \mathbb{IF}_{m}\left(  \theta\right)  \widetilde{\mathbb{S}}%
_{m}\left(  \theta\right)  \right]  =m!E_{\theta}\left(  if_{m,m}^{sym}\left(
O_{i_{1},}..O_{i_{m}};\theta\right)  \prod\limits_{r=1}^{m}a_{r}\left(
O_{i_{r}},\theta\right)  \right)  $. It follows that, for any choice of $m-1$
mean zero functions $a_{r}\left(  O\right)  $ under $\theta,$
\begin{align*}
0  &  =E_{\theta}\left(  \left\{
\begin{array}
[c]{c}%
if_{m-1,m-1}^{sym},_{\backslash l_{m}}\left(  O_{i_{1},}..O_{i_{m-1}}%
;\theta\right) \\
-mE_{\theta}\left[  if_{m,m}^{sym}\left(  O_{i_{1},}..O_{i_{m}};\theta\right)
a_{m}\left(  O_{i_{m}},\theta\right)  |O_{i_{1},}..O_{i_{m-1}}\right]
\end{array}
\right\}  \times\prod\limits_{r=1}^{m-1}a_{r}\left(  O_{i_{r}},\theta\right)
\right) \\
&  =E_{\theta}\left(  r\left(  O_{i_{1},}..O_{i_{m-1}};\theta\right)
\prod\limits_{r=1}^{m-1}a_{r}\left(  O_{i_{r}},\theta\right)  \right)
\end{align*}
where
\begin{align*}
&  r\left(  O_{i_{1},}..O_{i_{m-1}};\theta\right) \\
&  \equiv d_{m-1,\theta}\left[  if_{m-1,m-1}^{sym},_{\backslash l_{m}}\left(
O_{i_{1},}..O_{i_{m-1}};\theta\right)  \right]  -mE_{\theta}\left[
if_{m,m}^{sym}\left(  O_{i_{1},}..O_{i_{m}};\theta\right)  a_{m}\left(
O_{i_{m}},\theta\right)  |O_{i_{1},}..O_{i_{m-1}}\right]
\end{align*}
The last equality follows from $if_{m-1,m-1}^{sym},_{\backslash l_{m}}\left(
O_{i_{1},}..O_{i_{m-1}};\theta\right)  -d_{m-1,\theta}\left[  if_{m-1,m-1}%
^{sym},_{\backslash l_{m}}\left(  O_{i_{1},}..O_{i_{m-1}};\theta\right)
\right]  $ orthogonal to $\prod\limits_{r=1}^{m-1}a_{r}\left(  O_{i_{r}%
},\theta\right)  .$ We conclude $r\left(  O_{i_{1},}..O_{i_{m-1}}%
;\theta\right)  =0$ with probability $1$ because $r\left(  O_{i_{1}%
,}..O_{i_{m-1}};\theta\right)  $ is a degenerate U-statistic kernel of order
$m-1$ and all degenerate U-statistics of order $m-1$ have kernels that are the
(possibly infinite) sum of products of $m-1$ mean zero functions. It follows
that, on a set $\mathcal{O}_{m-1}$ which has probability $1$ under $F^{\left(
m-1\right)  }\left(  \cdot,\theta\right)  ,$
\begin{align*}
&  if_{m-1,m-1}^{sym},_{\backslash l_{m}}\left(  o_{i_{1},}...,o_{i_{m-1}%
};\theta\right) \\
&  =E_{\theta}\left[  \left\{  m\times if_{m,m}^{sym}\left(  o_{i_{1}%
,}...,o_{i_{m-1}},O_{i_{m},};\theta\right)  a_{m}\left(  O_{i_{m}}%
,\theta\right)  \right\}  \right] \\
&  +\left\{  I-d_{m-1,\theta}\right\}  \left[  if_{m-1,m-1}^{sym},_{\backslash
l_{m}}\left(  o_{i_{1},}...,o_{i_{m-1}};\theta\right)  \right] \\
&  =E_{\theta}\left[  \left\{
\begin{array}
[c]{c}%
m\times if_{m,m}^{sym}\left(  o_{i_{1},}..o_{i_{m-1}},O;\theta\right) \\
-\sum_{j=1}^{m-1}if_{m-1,m-1}^{sym}\left(  o_{i_{1},}...,o_{i_{j-1}%
},O,o_{i_{j+1}},...,o_{i_{m-1}};\theta\right)
\end{array}
\right\}  a_{m}\left(  O,\theta\right)  \right]
\end{align*}
since, by parts (i) and (ii) of the lemma \ref{joel2} and Eq. $\left(
\ref{deg}\right)  ,$
\begin{align*}
&  \left\{  I-d_{m-1,\theta}\right\}  \left[  if_{m-1,m-1}^{sym},_{\backslash
l_{m}}\left(  o_{i_{1},}...,o_{i_{m-1}};\theta\right)  \right] \\
&  =-E_{\theta}\left[  \sum_{j=1}^{m-1}if_{m-1,m-1}^{sym}\left(  o_{i_{1}%
,}...,o_{i_{j-1}},O,o_{i_{j+1}},...,o_{i_{m-1}};\theta\right)  a_{m}\left(
O,\theta\right)  \right]
\end{align*}
Here $I$ is the identity operator. Now since the model $f\left(
O;\widetilde{\theta}\left(  \zeta\right)  \right)  =f\left(  O;\theta\right)
\left\{  1+\zeta_{m}a_{m}\left(  O\right)  \right\}  $ with $\zeta_{s}=0$ for
$s<m$ has score $a_{m}\left(  O\right)  $ and such scores are dense in the
subspace of L$_{2}\left(  F\left(  \cdot,\theta\right)  \right)  $ with mean
zero, it follows that $if_{m-1,m-1}^{sym}\left(  o_{i_{1},}..o_{i_{m-1}%
};\theta\right)  $ has influence function $m\times if_{m,m}^{sym}\left(
o_{i_{1},}..o_{i_{m-1}},O;\theta\right)  -\sum_{j=1}^{m-1}if_{m-1,m-1}%
^{sym}\left(  o_{i_{1},}...,o_{i_{j-1}},O,o_{i_{j+1}},...,o_{i_{m-1}}%
;\theta\right)  \ $on the set $\mathcal{O}_{m-1}$. Thus $m\times
if_{m,m}^{sym}\left(  o_{i_{1},}..o_{i_{m-1}},O_{i_{m}};\theta\right)
=d_{m,\theta}\left[  if_{1,if_{m-1,m-1}^{sym}\left(  o_{i_{1},}..o_{i_{m-1}%
};\cdot\right)  }\left(  O_{i_{m}};\theta\right)  \right]  .$

Below $f\left(  O;\widetilde{\theta}\left(  \zeta\right)  \right)  \ ,$
$\zeta^{T}=\left(  \zeta_{1},...,\zeta_{s}\right)  $ denotes an arbitrary
smooth $s $ -dimensional parametric submodel and $\zeta_{t}$ denotes an
arbitrary component of $\zeta.$
\end{proof}

\begin{corollary}
For $m\geq2$,
\begin{align}
\Pi_{\theta}\left[  \mathbb{IF}_{m-1,m-1,\backslash l_{t}}\left(
\theta\right)  |\mathcal{U}_{m-2}^{\bot}\left(  \theta\right)  \right]   &
=-\Pi_{\theta}\left[  \mathbb{IF}_{m,m,\backslash l_{t}}\left(  \theta\right)
|\mathcal{U}_{m-1}\left(  \theta\right)  \right] \label{aaa}\\
\mathbb{IF}_{m,\backslash l_{t}}\left(  \theta\right)   &  =\Pi_{\theta
}\left[  \mathbb{IF}_{m,m,\backslash l_{t}}\left(  \theta\right)
|\mathcal{U}_{m-1}^{\bot}\left(  \theta\right)  \right] \label{bbb}%
\end{align}%
\begin{align}
&  E_{\theta}\left[  \mathbb{IF}_{m,\backslash l_{m+1}}\left(  \theta\right)
\widetilde{\mathbb{S}}_{m,\overline{l}_{m}}\left(  \theta\right)  \right]
\nonumber\\
&  =m!E_{\theta}\left(  if_{m,m,\backslash l_{m+1}}^{sym}\left(  O_{i_{1}%
,}..O_{i_{l_{m}}};\theta\right)  \prod\limits_{r=1}^{m}S_{l_{r}}\left(
O_{i_{r}},\theta\right)  \right) \label{ccc}%
\end{align}

\end{corollary}

\begin{proof}
$\left(  \ref{aaa}\right)  $:By lemma \ref{joel2} and Theorem 5c(ii)$,$
\begin{align*}
&  \Pi_{\theta}\left[  \mathbb{IF}_{m,m,\backslash l_{t}}\left(
\theta\right)  |\mathcal{U}_{m-1}\left(  \theta\right)  \right] \\
&  =\mathbb{V}\left[  mE_{\theta}\left(  IF_{m,m,\overline{i}_{m}}%
^{sym}\left(  \theta\right)  s_{l_{t}}\left(  O_{i_{m}}\right)  |O_{i_{1}%
},..,O_{i_{m-1}}\right)  \right] \\
&  =\mathbb{V}\left[  mE_{\theta}\left(  m^{-1}d_{m,\theta}\left\{
if_{1,if_{m-1,m-1}^{sym}\left(  O_{i_{1},}..O_{i_{m-1}};\cdot\right)  }\left(
O_{i_{m}};\theta\right)  \right\}  s_{l_{t}}\left(  O_{i_{m}}\right)
|O_{i_{1}},..,O_{i_{m-1}}\right)  \right]
\end{align*}
Now, by part (iii) of lemma \ref{joel2} and Eq. $\left(  \ref{deg}\right)  ,$
the RHS is
\begin{align*}
&  \mathbb{V}\left[  E_{\theta}\left(  if_{1,if_{m-1,m-1}^{sym}\left(
O_{i_{1},}..O_{i_{m-1}};\cdot\right)  }\left(  O_{i_{m}};\theta\right)
s_{l_{t}}\left(  O_{i_{m}}\right)  |O_{i_{1}},..,O_{i_{m-1}}\right)  \right]
\\
&  -\mathbb{V}\left\{  E\left[  E\left[  \left(  m-1\right)  E\left[
if_{1,if_{m-1,m-1}^{sym}\left(  O_{i_{1},}..O_{i_{m-1}};\cdot\right)  }\left(
O_{i_{m}};\theta\right)  |O_{i_{m}},O_{i_{1}},..,O_{i_{m-2}}\right]  \right]
s_{l_{t}}\left(  O_{i_{m}}\right)  |O_{i_{1}},..,O_{i_{m-1}}\right]  \right\}
\\
&  =\mathbb{V}\left[  IF_{m-1,m-1,\backslash l_{t}}^{sym}\left(
\theta\right)  \right]  -\mathbb{V}\left\{  \left(  m-1\right)  E_{\theta
}\left[  IF_{m-1,m-1,\backslash l_{t}}^{sym}\left(  \theta\right)  |O_{i_{1}%
},...,O_{i_{m-2}}\right]  \right\}
\end{align*}
On the other hand, by part (iv) of the lemma \ref{joel2}, $\Pi_{\theta}\left[
\mathbb{IF}_{m-1,m-1,\backslash l_{t}}\left(  \theta\right)  |\mathcal{U}%
_{m-2}^{\bot}\left(  \theta\right)  \right]  =\mathbb{V}\left[
IF_{m-1,m-1,\backslash l_{t}}\left(  \theta\right)  \right]  -\mathbb{V}%
\left[  \left(  m-1\right)  E_{\theta}\left[  IF_{m-1,m-1,\backslash
l_{t},\overline{i}_{m-1}}^{sym}\left(  \theta\right)  |O_{i_{1}}%
,...,O_{i_{m-2}}\right]  \right]  . $

$\left(  \ref{bbb}\right)  :$ Write
\begin{align*}
\mathbb{IF}_{m,\backslash l_{t}}\left(  \theta\right)   &  =\Pi_{\theta
}\left[  \mathbb{IF}_{m,m,\backslash l_{t}}\left(  \theta\right)
|\mathcal{U}_{m-1}^{\bot}\left(  \theta\right)  \right]  +\left\{  \Pi\left[
\mathbb{IF}_{2,2,\backslash l_{t}}\left(  \theta\right)  |\mathcal{U}%
_{1}\left(  \theta\right)  \right]  +\mathbb{IF}_{1,\backslash l_{t}}\left(
\theta\right)  \right\} \\
&  +\sum_{j=2}^{m-1}\left\{  \Pi\left[  \mathbb{IF}_{j+1,j+1,\backslash l_{t}%
}\left(  \theta\right)  |\mathcal{U}_{j}\left(  \theta\right)  \right]
+\Pi\left[  \mathbb{IF}_{jj,\backslash l_{t}}\left(  \theta\right)
|\mathcal{U}_{j-1}^{\bot}\left(  \theta\right)  \right]  \right\}
\end{align*}
The RHS is $\Pi_{\theta}\left[  \mathbb{IF}_{m,m,\backslash l_{t}}\left(
\theta\right)  |\mathcal{U}_{m-1}^{\bot}\left(  \theta\right)  \right]  $ by
eq. $\left(  \ref{aaa}\right)  .$

$\left(  \ref{ccc}\right)  :$ $E_{\theta}\left[  \mathbb{IF}_{m,\backslash
l_{m+1}}\left(  \theta\right)  \widetilde{\mathbb{S}}_{m,\overline{l}_{m}%
}\left(  \theta\right)  \right]  =E_{\theta}\left[  \Pi\left[  \mathbb{IF}%
_{m,m,\backslash l_{m+1}}\left(  \theta\right)  |\mathcal{U}_{m-1}^{\bot
}\left(  \theta\right)  \right]  \widetilde{\mathbb{S}}_{m,\overline{l}_{m}%
}\left(  \theta\right)  \right]  $ by eq. $\left(  \ref{bbb}\right)  .$ But
the RHS of this equation is the RHS of eq. $\left(  \ref{ccc}\right)  .$
\end{proof}

\begin{proof}
(\textbf{Theorem} \textbf{5c(i)): }By assumption $\psi_{\backslash\overline
{l}_{m-1}}\left(  \theta\right)  =E_{\theta}\left(  \mathbb{IF}_{m-1}\left(
\theta\right)  \widetilde{\mathbb{S}}_{m-1,\overline{l}_{m-1}}\left(
\theta\right)  \right)  .$ Hence
\[
\psi_{\backslash\overline{l}_{m\ }}\left(  \theta\right)  =E_{\theta}\left(
\mathbb{IF}_{m-1}\left(  \theta\right)  \widetilde{\mathbb{S}}_{m,\overline
{l}_{m}}\left(  \theta\right)  \right)  +E_{\theta}\left[  \mathbb{IF}%
_{m-1,\backslash l_{m}}\left(  \theta\right)  \widetilde{\mathbb{S}%
}_{m-1,\overline{l}_{m-1}}\left(  \theta\right)  \right]
\]
By eq. $\left(  \ref{ccc}\right)  ,$ and the assumption $if_{m-1\ ,m-1\ }%
^{sym}\left(  O_{i_{1},}..O_{i_{m}};\theta\right)  $ has an influence
function, we obtain
\begin{align*}
&  E_{\theta}\left[  \mathbb{IF}_{m-1,\backslash l_{m}}\left(  \theta\right)
\widetilde{\mathbb{S}}_{m-1,\overline{l}_{m-1}}\left(  \theta\right)  \right]
\\
&  =\left(  m-1\right)  !E_{\theta}\left(  if_{1if_{m-1,m-1}^{sym}\left(
O_{i_{1},}..O_{i_{m-1}};\cdot\right)  }\left(  O_{i_{m}},\theta\right)
S_{l_{m}}\left(  O_{i_{m}},\theta\right)  \prod\limits_{r=1}^{m-1}S_{l_{r}%
}\left(  O_{i_{r}},\theta\right)  \right)  .
\end{align*}
We conclude that $\mathbb{IF}_{m,m}$ exists and equals $\mathbb{V}\left[
m^{-1}d_{m,\theta}\left\{  if_{1if_{m-1,m-1}^{sym}\left(  O_{i_{1}%
,}..O_{i_{m-1}};\cdot\right)  }\left(  O_{i_{m}},\theta\right)  \right\}
\right]  .$
\end{proof}

Below is an alternative proof of Theorem 5c (i) and (ii)

\begin{proof}
First we show that for any $j$-dimensional parametric submodel
$\widetilde{\theta}\left(  \zeta\right)  ,$
\[
\frac{\partial\left(  \mathbb{IF}_{j,\psi}\left(  \theta\right)  +\psi\left(
\theta\right)  \right)  }{\partial l_{j}}\in\mathcal{U}_{j-1}^{\bot_{j,\theta
}}\left(  \theta\right)
\]
where $\mathcal{U}_{0}\left(  \theta\right)  =\phi.$ {}From eq
(\ref{interinfo0}) we know that
\[
E_{\theta}\left(  \frac{\partial\left(  \mathbb{IF}_{j,\psi}\left(
\theta\right)  +\psi\left(  \theta\right)  \right)  }{\partial l_{j}%
}\widetilde{\mathbb{S}}_{r,\overline{l}_{r}}\left(  \theta\right)  \right)
=E_{\theta}\left(  \frac{\partial\left(  \mathbb{IF}_{j,\psi}\left(
\theta\right)  \right)  }{\partial l_{r}}\widetilde{\mathbb{S}}_{r,\overline
{l}_{r}}\left(  \theta\right)  \right)  =0
\]
for all $1\leq r<j.$ Since $\mathcal{M}\left(  \Theta\right)  $ is locally
nonparametric, i.e., $\mathcal{U}_{j-1}\left(  \theta\right)  =\Gamma
_{j-1}\left(  \theta\right)  ,$ and
\[
E_{\theta}\left(  \frac{\partial\left(  \mathbb{IF}_{j,\psi}\left(
\theta\right)  +\psi\left(  \theta\right)  \right)  }{\partial l_{j}}\right)
=-E_{\theta}\left(  \mathbb{IF}_{j,\psi}\left(  \theta\right)
\widetilde{\mathbb{S}}_{1,l_{j}}\left(  \theta\right)  \right)  +\psi
_{\backslash l_{j}}=0
\]
therefore we have
\begin{equation}
\frac{\partial\left(  \mathbb{IF}_{j,\psi}\left(  \theta\right)  +\psi\left(
\theta\right)  \right)  }{\partial l_{j}}\in\mathcal{U}_{j-1}^{\bot_{j,\theta
}}\left(  \theta\right) \label{eqn1proof}%
\end{equation}

ii) If $\mathbb{IF}_{m}$ exists, we have $\mathbb{IF}_{m}=\mathbb{V}\left(
IF_{m}\right)  $ with $IF_{m}$ symmetric, such that%
\begin{align*}
\psi_{\backslash\overline{l}_{m}}\left(  \theta\right)   &  =E\left(
\mathbb{IF}_{m}\mathbb{S}_{m,\overline{l}_{m}}\right) \\
&  =E_{\theta}\left(  \mathbb{V}\left[  IF_{m}^{c}\left(  \theta\right)
\right]  \mathbb{D}_{m}^{(\mathbb{S}_{m,\overline{l}_{m}})}\right) \\
&  +E_{\theta}\left(  \mathbb{V}\left(  mE\left[  if_{m}\left(  O_{i_{1}%
},...,O_{i_{m}};\theta\right)  |O_{i_{1}},...O_{i_{m-1}}\right]  ^{c}\right)
\mathbb{D}_{m-1}^{(\mathbb{S}_{m,\overline{l}_{m}})}\right) \\
&  +E\left(  \mathbb{IF}_{m-2}\mathbb{S}_{m,\overline{l}_{m}}\right)
\end{align*}
Therefore, $if_{m-1,m-1}\left(  O_{i_{1}},...,O_{i_{m-1}};\theta\right)
=mE\left[  if_{m}\left(  O_{i_{1}},...,O_{i_{m}};\theta\right)  |O_{i_{1}%
},...O_{i_{m-1}}\right]  ^{c}$ wp1$.$ \ So that the lhs has an influence
function since:
\begin{align*}
&  \left\{  E_{\theta}\left[  if_{m}\left(  O_{i_{1}},...,O_{i_{m}}%
;\theta\right)  |o_{i_{1}},...o_{i_{m-1}}\right]  ^{c}\right\}  _{\backslash
l_{m}}\\
&  =\left\{  E_{\theta}\left[  if_{m}\left(  O_{i_{1}},...,O_{i_{m}}%
;\theta\right)  |o_{i_{1}},...o_{i_{m-1}}\right]  \right\}  _{\backslash
l_{m}}\\
&  +\sum_{t=0}^{m-2}\left(  -1\right)  ^{m-1-t}\sum_{\substack{i_{r_{1}}\neq
i_{r_{2}}..\neq i_{r_{t}} \\\overline{i}_{r_{t}}\subset\overline{i}_{m-1}
}}E_{\theta}\left(  if_{m}\left(  O_{i_{1}},O_{i_{2}},..O_{i_{m}}%
;\theta\right)  |o_{i_{r_{1}}},o_{i_{r_{2}}},..o_{i_{r_{t}}}\right)
_{\backslash l_{m}}\\
&  =\left\{  E_{\theta}\left[  if_{m}\left(  O_{i_{1}},...,O_{i_{m}}%
;\theta\right)  S_{l_{m}}\left(  O_{i_{m}}\right)  |o_{i_{1}},...o_{i_{m-1}%
}\right]  \right\} \\
&  +\sum_{t=0}^{m-2}\left[
\begin{array}
[c]{c}%
\left(  -1\right)  ^{m-1-t}\left(  m-t\right) \\
\sum_{\substack{i_{r_{1}}\neq i_{r_{2}}..\neq i_{r_{t}} \\\overline{i}_{r_{t}%
}\subset\overline{i}_{m-1}}}E\left(
\begin{array}
[c]{c}%
E\left[  if_{m}\left(  O_{i_{1}},...,O_{i_{m}};\theta\right)  |o_{i_{r_{1}}%
},o_{i_{t}},..O_{i_{t+1}}\right]  S_{l_{t+1}}\left(  O_{i_{t+1}}\right) \\
|o_{i_{r_{1}}},o_{i_{t}},..o_{i_{t}}%
\end{array}
\right)
\end{array}
\right]
\end{align*}
by equation \ref{eqn1proof}.

i) if the first order influence function of $if_{m-1,m-1,\theta}\left(
O_{i_{1},}..O_{i_{m-1}}\right)  ,$
\[
if_{1,if_{m-1,m-1,\theta}\left(  O_{i_{1},}..O_{i_{m-1}};\cdot\right)
}\left(  O_{i_{m}};\theta\right)
\]
exists, then
\begin{gather*}
\psi_{\backslash\overline{l}_{m}}\left(  \theta\right)  =E\left(
\mathbb{IF}_{m-1}\mathbb{S}_{m,\overline{l}_{m}}\right)  +\\
E_{\theta}\left(  \left(  m-1\right)  !if_{1,if_{m-1,m-1,\theta}\left(
O_{i_{1},}..O_{i_{m-1}};\cdot\right)  }\left(  O_{i_{m}},\theta\right)
%TCIMACRO{\tprod \limits_{r=1}^{m}}%
%BeginExpansion
{\textstyle\prod\limits_{r=1}^{m}}
%EndExpansion
S_{l_{r}}\left(  O_{i_{r}},\theta\right)  \right) \\
=E_{\theta}\left(  \left(  m-1\right)  !if_{1,if_{m-1,m-1,\theta}\left(
O_{i_{1},}..O_{i_{m-1}};\cdot\right)  }^{c}\left(  O_{i_{m}},\theta\right)
%TCIMACRO{\tprod \limits_{r=1}^{m}}%
%BeginExpansion
{\textstyle\prod\limits_{r=1}^{m}}
%EndExpansion
S_{l_{r}}\left(  O_{i_{r}},\theta\right)  \right)
\end{gather*}

If we switch the order of differentiating ${}l_{m}$ with $l_{j}$ $\left(
1\leq j\leq m-1\right)  ,$ since
\[
if_{m-1,m-1}\left(  O_{i_{1}},..O_{i_{m-1}},\theta\right)
\]
is symmetric, we will have
\begin{gather*}
\psi_{\backslash\overline{l}_{m}}\left(  \theta\right)  =E\left(
\mathbb{IF}_{m-1}\mathbb{S}_{m,\overline{l}_{m}}\right)  +\\
E_{\theta}\left(  \left(  m-1\right)  !if_{1,if_{m-1,m-1,\theta}\left(
O_{i_{1},}.O_{i_{m}}.O_{i_{m-1}};\cdot\right)  }^{c}\left(  O_{i_{j}}%
,\theta\right)
%TCIMACRO{\tprod \limits_{r=1}^{m}}%
%BeginExpansion
{\textstyle\prod\limits_{r=1}^{m}}
%EndExpansion
S_{l_{r}}\left(  O_{i_{r}},\theta\right)  \right)
\end{gather*}

which means
\begin{gather*}
E_{\theta}\left(  \left[
\begin{array}
[c]{c}%
if_{1,if_{m-1,m-1,\theta}\left(  O_{i_{1},}..O_{i_{m-1}};\cdot\right)  }%
^{C}\left(  O_{i_{m}},\theta\right) \\
-if_{1,if_{m-1,m-1,\theta}\left(  O_{i_{1},}.O_{i_{m}}.O_{i_{m-1}}%
;\cdot\right)  }^{C}\left(  O_{i_{j}},\theta\right)
\end{array}
\right]
%TCIMACRO{\tprod \limits_{r=1}^{m}}%
%BeginExpansion
{\textstyle\prod\limits_{r=1}^{m}}
%EndExpansion
S_{l_{r}}\left(  O_{i_{r}},\theta\right)  \right)  =0\\
\Longleftrightarrow if_{1,if_{m-1,m-1,\theta}\left(  O_{i_{1},}..O_{i_{m-1}%
};\cdot\right)  }^{C}\left(  O_{i_{m}},\theta\right)
=if_{1,if_{m-1,m-1,\theta}\left(  O_{i_{1},}.O_{i_{m}}.O_{i_{m-1}}%
;\cdot\right)  }^{C}\left(  O_{i_{j}},\theta\right)  \text{ w.p.1 }%
\end{gather*}

therefore
\[
\psi_{\backslash\overline{l}_{m}}\left(  \theta\right)  =E\left(
\mathbb{IF}_{m-1}\mathbb{S}_{m,\overline{l}_{m}}\right)  +E\left(
\mathbb{IF}_{m,m}\left(  \theta\right)  \mathbb{S}_{m,\overline{l}_{m}%
}\right)
\]

with
\[
\mathbb{IF}_{m,m}\equiv\frac{1}{m}\mathbb{V}\left(  if_{1,if_{m-1,m-1}\left(
O_{i_{1}},...,O_{i_{m-1}};\cdot\right)  \ {}}^{c}\left(  O_{i_{m}}%
;\theta\right)  \right)
\]

\end{proof}

\bigskip

\begin{proof}
\bigskip(Proof of eq. (\ref{imm})) We have proved in the text the following
results that will be used repeatedly throughout the proof:%
\[
v\left(  x;\theta\right)  K_{f_{X},\infty}\left(  x,X\right)  =v\left(
X;\theta\right)  K_{f_{X},\infty}\left(  x,X\right)  ,
\]%
\begin{align*}
&  if_{1,b\left(  x_{i_{1}};\cdot\right)  }\left(  O_{i_{2}};\theta\right) \\
&  =-E_{\theta}\left[  H_{1}|x_{i_{1}}\right]  ^{-\frac{1}{2}}K_{f_{X},\infty
}\left(  x_{i_{1}},X_{i_{2}}\right)  E_{\theta}\left[  H_{1}|X_{i_{2}}\right]
^{-\frac{1}{2}}\varepsilon_{b,i_{2}}\left(  \theta\right) \\
&  =-\frac{K_{Leb,\infty}\left(  x_{i_{1}},X_{i_{2}}\right)  }{g\left(
x_{i_{1}}\right)  ^{\frac{1}{2}}g\left(  X_{i_{2}}\right)  ^{\frac{1}{2}}%
}\varepsilon_{b,i_{2}}\left(  \theta\right)  ,
\end{align*}%
\begin{align*}
&  if_{1,p\left(  x_{i_{1}};\cdot\right)  }\left(  O_{i_{2}};\theta\right) \\
&  =-E_{\theta}\left[  H_{1}|x_{i_{1}}\right]  ^{-\frac{1}{2}}K_{f_{X},\infty
}\left(  x_{i_{1}},X_{i_{2}}\right)  E_{\theta}\left[  H_{1}|X_{i_{2}}\right]
^{-\frac{1}{2}}\varepsilon_{p,i_{2}}\left(  \theta\right) \\
&  =-\frac{K_{Leb,\infty}\left(  x_{i_{1}},X_{i_{2}}\right)  }{g\left(
x_{i_{1}}\right)  ^{\frac{1}{2}}g\left(  X_{i_{2}}\right)  ^{\frac{1}{2}}%
}\varepsilon_{p,i_{2}}\left(  \theta\right)  ,
\end{align*}
and
\[
IF_{1,K_{f_{X},\infty}\left(  X_{i_{1}},X_{i_{2}}\right)  }=-\left\{
K_{f_{X},\infty}\left(  X_{i_{1}},X_{i_{3}}\right)  K_{f_{X},\infty}\left(
X_{i_{3}},X_{i_{2}}\right)  -K_{f_{X},\infty}\left(  X_{i_{1}\ },X_{i_{2}%
}\right)  \right\}  .
\]

In addition, by an analogous argument, we can also show that
\begin{align*}
&  if_{1,g\left(  x;\cdot\right)  }\left(  O\right) \\
&  =IF_{1,E_{\theta}\left[  H_{1}|x\right]  }f_{X}\left(  x\right)
+E_{\theta}\left[  H_{1}|x\right]  IF_{1,f_{X}\left(  x;\cdot\right)  }\\
&  =\left(  H_{1}-E\left[  H_{1}|X\right]  \right)  K_{Leb,\infty}\left(
X,x\right)  +E_{\theta}\left[  H_{1}|x\right]  \left(  K_{Leb,\infty}\left(
X,x\right)  -f_{X}\left(  x\right)  \right) \\
&  =H_{1}K_{Leb,\infty}\left(  X,x\right)  -g\left(  x\right)
\end{align*}
\newline

Now, we are ready to prove that eq. (\ref{imm}) holds for any $m\geq2$ by induction.

i) The case where $m=2$ was proved in the text.

ii) We now assume eq. (\ref{imm}) holds for $m$ for some $m\geq2,$ and shall
prove it is also true for $m+1.$ $\ $By assumption,%
\begin{align*}
&  IF_{m,m,\psi,\overline{i}_{m}}\left(  \theta\right) \\
&  =\varepsilon_{b,i_{1}}\left(  \theta\right)  g\left(  X_{i_{1}}\right)
^{-\frac{1}{2}}\left[
\begin{array}
[c]{c}%
\sum_{j=0}^{m-2}c(m,j)\times\\
\prod\limits_{s=1}^{j}\frac{H_{1,i_{s+1}}}{g\left(  X_{i_{s+1}}\right)
}K_{Leb,\infty}\left(  X_{i_{s}},X_{i_{s+1}}\right) \\
\times K_{Leb,\infty}\left(  X_{i_{j+1}},X_{i_{m}}\right)
\end{array}
\right]  g\left(  X_{i_{m}}\right)  ^{-\frac{1}{2}}\varepsilon_{p,i_{m}%
}\left(  \theta\right)
\end{align*}

Following the results from part 5c) of Theorem \ref{eift}, $\ $%
\ $IF_{m+1,m+1,\psi,\overline{i}_{m}}\left(  \theta\right)  $ exists if and
only if $if_{1,IF_{m,m,\psi,\overline{i}_{m}}\left(  \cdot\right)  }$ exists.
\ By the chain rule,%
\begin{align*}
&  if_{1,if_{m,m,\psi,\overline{i}_{m}}\left(  \mathbf{O}_{i_{m}}\cdot\right)
}\left(  O_{i_{m+1}};\theta\right) \\
&  =\left\{
\begin{array}
[c]{c}%
\left(
\begin{array}
[c]{c}%
H_{1,i_{1}}if_{1,\varepsilon_{b,i_{1}}\left(  \cdot\right)  }\left(
O_{i_{m+1}}\right)  \varepsilon_{p,i_{m}}\left(  \theta\right)  +\\
H_{1,i_{m}}\varepsilon_{b,i_{1}}\left(  \theta\right)  if_{1,}\varepsilon
_{p,i_{m}}\left(  \cdot\right)  \left(  O_{i_{m+1}}\right)
\end{array}
\right)  g\left(  X_{i_{1}}\right)  ^{-\frac{1}{2}}g\left(  X_{i_{m}}\right)
^{-\frac{1}{2}}\\
-\frac{1}{2}\varepsilon_{b,i_{1}}\left(  \theta\right)  \varepsilon_{p,i_{m}%
}\left(  \theta\right)  g\left(  X_{i_{1}}\right)  ^{-\frac{1}{2}}g\left(
X_{i_{m}}\right)  ^{-\frac{1}{2}}\times\\
\left[  \frac{if_{1,g\left(  X_{i_{1}}\right)  }\left(  O_{i_{m+1}}\right)
}{g\left(  X_{i_{1}}\right)  }+\frac{if_{1,g\left(  X_{i_{m}}\right)  }\left(
O_{i_{m+1}}\right)  }{g\left(  X_{i_{m}}\right)  }\right]
\end{array}
\right\}  \times\\
&  \left[
\begin{array}
[c]{c}%
\sum_{j=0}^{m-2}c(m,j)\times\\
\prod\limits_{s=1}^{j}\frac{H_{1,i_{s+1}}}{g\left(  X_{i_{s+1}}\right)
}K_{Leb,\infty}\left(  X_{i_{s}},X_{i_{s+1}}\right) \\
\times K_{Leb,\infty}\left(  X_{i_{j+1}},X_{i_{m}}\right)
\end{array}
\right] \\
&  -\varepsilon_{b,i_{1}}\left(  \theta\right)  \varepsilon_{p,i_{m}}\left(
\theta\right)  g\left(  X_{i_{1}}\right)  ^{-\frac{1}{2}}g\left(  X_{i_{m}%
}\right)  ^{-\frac{1}{2}}\sum_{j=0}^{m-2}c(m,j)\times\\
&  \left[
\begin{array}
[c]{c}%
\sum_{t=1}^{j}\frac{H_{1,i_{t+1}}K_{Leb,\infty}\left(  X_{i_{t}},X_{i_{t+1}%
}\right)  }{g^{2}\left(  X_{i_{t+1}}\right)  }if_{1,g\left(  X_{i_{t+1}%
}\right)  }\left(  O_{i_{m+1}}\right)  \times\\
\prod\limits_{s\neq t}\frac{H_{1,i_{s+1}}}{g\left(  X_{i_{s+1}}\right)
}K_{Leb,\infty}\left(  X_{i_{s}},X_{i_{s+1}}\right)  K_{Leb,\infty}\left(
X_{i_{j+1}},X_{i_{m}}\right)
\end{array}
\right]
\end{align*}%
\[
=
\]%
\begin{align*}
&  \left\{
\begin{array}
[c]{c}%
-\frac{H_{1,i_{1}}K_{Leb,\infty}\left(  X_{i_{1}},X_{i_{m+1}}\right)
}{g\left(  X_{i_{1}}\right)  g\left(  X_{i_{m+1}}\right)  ^{\frac{1}{2}%
}g\left(  X_{i_{m}}\right)  ^{\frac{1}{2}}}\varepsilon_{b,i_{m+1}}\left(
\theta\right)  \varepsilon_{p,i_{m}}\left(  \theta\right) \\
-\varepsilon_{b,i_{1}}\left(  \theta\right)  \frac{H_{1,i_{m}}K_{Leb,\infty
}\left(  X_{i_{m}},X_{i_{m+1}}\right)  }{g\left(  X_{i_{m}}\right)  g\left(
X_{i_{1}}\right)  ^{\frac{1}{2}}g\left(  X_{i_{m+1}}\right)  ^{\frac{1}{2}}%
}\varepsilon_{p,i_{m+1}}\\
-\frac{1}{2}\varepsilon_{b,i_{1}}\left(  \theta\right)  \varepsilon_{p,i_{m}%
}\left(  \theta\right)  g\left(  X_{i_{1}}\right)  ^{-\frac{1}{2}}g\left(
X_{i_{m}}\right)  ^{-\frac{1}{2}}\times\\
\left[  \frac{H_{1,i_{m+1}}K_{Leb,\infty}\left(  X_{i_{m+1}},X_{i_{1}}\right)
}{g\left(  X_{i_{1}}\right)  }+\frac{H_{1,i_{m+1}}K_{Leb,\infty}\left(
X_{i_{m+1}},X_{i_{m}}\right)  }{g\left(  X_{i_{m}}\right)  }-2\right]
\end{array}
\right\} \\
\times &  \left[
\begin{array}
[c]{c}%
\sum_{j=0}^{m-2}c(m,j)\times\\
\prod\limits_{s=1}^{j}\frac{H_{1,i_{s+1}}}{g\left(  X_{i_{s+1}}\right)
}K_{Leb,\infty}\left(  X_{i_{s}},X_{i_{s+1}}\right) \\
\times K_{Leb,\infty}\left(  X_{i_{j+1}},X_{i_{m}}\right)
\end{array}
\right] \\
&  -\varepsilon_{b,i_{1}}\left(  \theta\right)  \varepsilon_{p,i_{m}}\left(
\theta\right)  g\left(  X_{i_{1}}\right)  ^{-\frac{1}{2}}g\left(  X_{i_{m}%
}\right)  ^{-\frac{1}{2}}\sum_{j=0}^{m-2}c(m,j)\times\\
&  \left[
\begin{array}
[c]{c}%
\sum_{t=1}^{j}\frac{H_{1,i_{t+1}}K_{Leb,\infty}\left(  X_{i_{t}},X_{i_{t+1}%
}\right)  }{g^{2}\left(  X_{i_{t+1}}\right)  }\left(  H_{1,i_{m+1}%
}K_{Leb,\infty}\left(  X_{i_{m+1}},X_{i_{t+1}}\right)  -g\left(  X_{i_{t+1}%
}\right)  \right)  \times\\
\prod\limits_{s\neq t}\frac{H_{1,i_{s+1}}}{g\left(  X_{i_{s+1}}\right)
}K_{Leb,\infty}\left(  X_{i_{s}},X_{i_{s+1}}\right)  K_{Leb,\infty}\left(
X_{i_{j+1}},X_{i_{m}}\right)
\end{array}
\right]
\end{align*}%
\[
=
\]
\newline%
\begin{align*}
&  +\sum_{j=0}^{m-2}\left(  -1\right)  ^{j}\binom{m-2}{j}\left(
\begin{array}
[c]{c}%
\varepsilon_{b,i_{m+1}}\left(  \theta\right)  g\left(  X_{i_{m+1}}\right)
^{-\frac{1}{2}}\frac{H_{1,i_{1}}K_{Leb,\infty}\left(  X_{i_{m+1}},X_{i_{1}%
}\right)  }{g\left(  X_{i_{1}}\right)  }\\
\prod\limits_{s=1}^{j}\frac{H_{1,i_{s+1}}}{g\left(  X_{i_{s+1}}\right)
}K_{Leb,\infty}\left(  X_{i_{s}},X_{i_{s+1}}\right)  K_{Leb,\infty}\left(
X_{i_{j+1}},X_{i_{m}}\right) \\
\times g\left(  X_{i_{m}}\right)  ^{-\frac{1}{2}}\varepsilon_{p,i_{m}}\left(
\theta\right)
\end{array}
\right) \\
&  +\sum_{j=0}^{m-2}\left(  -1\right)  ^{j}\binom{m-2}{j}\left(
\begin{array}
[c]{c}%
\varepsilon_{b,i_{1}}\left(  \theta\right)  g\left(  X_{i_{1}}\right)
^{-\frac{1}{2}}\prod\limits_{s=1}^{j}\frac{H_{1,i_{s+1}}}{g\left(  X_{i_{s+1}%
}\right)  }K_{Leb,\infty}\left(  X_{i_{s}},X_{i_{s+1}}\right) \\
\frac{H_{1,i_{m}}K_{Leb,\infty}\left(  X_{i_{j+1}},X_{i_{m}}\right)
}{g\left(  X_{i_{m}}\right)  }K_{Leb,\infty}\left(  X_{i_{m}},X_{i_{m+1}%
}\right)  ^{\frac{1}{2}}\\
\times g\left(  X_{i_{m+1}}\right)  ^{-\frac{1}{2}}\varepsilon_{p,i_{m+1}}%
\end{array}
\right) \\
&  +\frac{1}{2}\sum_{j=0}^{m-2}\left(  -1\right)  ^{j}\binom{m-2}{j}\left\{
\begin{array}
[c]{c}%
\varepsilon_{b,i_{1}}\left(  \theta\right)  \varepsilon_{p,i_{m}}\left(
\theta\right)  \left(
\begin{array}
[c]{c}%
\frac{H_{1,i_{m+1}}K_{Leb,\infty}\left(  X_{i_{m+1}},X_{i_{1}}\right)
}{g\left(  X_{i_{1}}\right)  }\\
+\frac{H_{1,i_{m+1}}K_{Leb,\infty}\left(  X_{i_{m+1}},X_{i_{j+1}}\right)
}{g\left(  X_{i_{m+1}}\right)  }%
\end{array}
\right) \\
\times\prod\limits_{s=1}^{j}\frac{H_{1,i_{s+1}}}{g\left(  X_{i_{s+1}}\right)
}K_{Leb,\infty}\left(  X_{i_{s}},X_{i_{s+1}}\right)  \times\\
K_{Leb,\infty}\left(  X_{i_{j+1}},X_{i_{m}}\right)  g\left(  X_{i_{1}}\right)
^{-\frac{1}{2}}g\left(  X_{i_{m}}\right)  ^{-\frac{1}{2}}%
\end{array}
\right\} \\
&  +IF_{m,m,\psi,\overline{i}_{m}}\left(  \theta\right) \\
&  +\sum_{j=0}^{m-2}\left(  -1\right)  ^{j}\binom{m-2}{j}\left[  \sum
_{t=1}^{j}\left(
\begin{array}
[c]{c}%
\varepsilon_{b,i_{1}}\left(  \theta\right)  \varepsilon_{p,i_{m}}\left(
\theta\right)  g\left(  X_{i_{1}}\right)  ^{-\frac{1}{2}}g\left(  X_{i_{m}%
}\right)  ^{-\frac{1}{2}}\times\\
\frac{H_{1,i_{t+1}}K_{Leb,\infty}\left(  X_{i_{t}},X_{i_{t+1}}\right)
}{g\left(  X_{i_{t+1}}\right)  }H_{1,i_{m+1}}\frac{K_{Leb,\infty}\left(
X_{i_{t+1}},X_{i_{m+1}}\right)  }{g\left(  X_{i_{m+1}}\right)  }\\
\times\prod\limits_{s\neq t}\frac{H_{1,i_{s+1}}}{g\left(  X_{i_{s+1}}\right)
}K_{Leb,\infty}\left(  X_{i_{s}},X_{i_{s+1}}\right)  K_{Leb,\infty}\left(
X_{i_{j+1}},X_{i_{m}}\right)
\end{array}
\right)  \right] \\
&  -\sum_{j=0}^{m-2}\left(  -1\right)  ^{j}\binom{m-2}{j}j\left(
\begin{array}
[c]{c}%
\varepsilon_{b,i_{1}}\left(  \theta\right)  \varepsilon_{p,i_{m}}\left(
\theta\right)  g\left(  X_{i_{1}}\right)  ^{-\frac{1}{2}}g\left(  X_{i_{m}%
}\right)  ^{-\frac{1}{2}}\times\\
\prod\limits_{s=1}^{j}\frac{H_{1,i_{s+1}}}{g\left(  X_{i_{s+1}}\right)
}K_{Leb,\infty}\left(  X_{i_{s}},X_{i_{s+1}}\right)  K_{Leb,\infty}\left(
X_{i_{j+1}},X_{i_{m}}\right)
\end{array}
\right)
\end{align*}

Applying the operator $d_{m+1,\theta},$ which is defined in Eq$.\left(
\ref{deg}\right)  ,$ on the statistic above, it is straightforward to show
that
\begin{gather*}
\mathbb{IF}_{m+1,m+1,\psi,\overline{i}_{m+1}}\left(  \theta\right)  =\frac
{1}{m+1}\left(  \mathbb{V}\left\{  d_{m+1,\theta}\left[  if_{1,IF_{m,m,\psi
,\overline{i}_{m}}\left(  \cdot\right)  }\left(  \theta\right)  \right]
\right\}  \right) \\
=\mathbb{V}\left\{  d_{m+1,\theta}\left[
\begin{array}
[c]{c}%
\varepsilon_{b,i_{1}}\left(  \theta\right)  g\left(  X_{i_{1}}\right)
^{-\frac{1}{2}}\prod\limits_{s=1}^{m-1}\frac{H_{1,i_{s+1}}}{g\left(
X_{i_{s+1}}\right)  }K_{Leb,\infty}\left(  X_{i_{s}},X_{i_{s+1}}\right) \\
\times K_{Leb,\infty}\left(  X_{i_{m}},X_{i_{m+1}}\right)  g\left(
X_{i_{m+1}}\right)  ^{-\frac{1}{2}}\varepsilon_{p,i_{m+1}}\left(
\theta\right)
\end{array}
\right]  \right\} \\
=\mathbb{V}\left\{
\begin{array}
[c]{c}%
\varepsilon_{b,i_{1}}\left(  \theta\right)  g\left(  X_{i_{1}}\right)
^{-\frac{1}{2}}\left[
\begin{array}
[c]{c}%
\sum_{j=0}^{m-1}c(m+1,j)\times\\
\prod\limits_{s=1}^{j}\frac{H_{1,i_{s+1}}}{g\left(  X_{i_{s+1}}\right)
}K_{Leb,\infty}\left(  X_{i_{s}},X_{i_{s+1}}\right) \\
\times K_{Leb,\infty}\left(  X_{i_{j+1}},X_{i_{m+1}}\right)
\end{array}
\right] \\
\times g\left(  X_{i_{m+1}}\right)  ^{-\frac{1}{2}}\varepsilon_{p,i_{m+1}%
}\left(  \theta\right)
\end{array}
\right\}
\end{gather*}

\end{proof}

\begin{proof}
(Theorem \ref{TBformula}) By eq (\ref{eta1}) and eq (\ref{alpha1}) and part a
of the preceding lemma
\begin{align*}
\widetilde{\overline{\eta}}_{k}  &  =-E\left[  \dot{B}\dot{P}H_{1}\overline
{Z}_{k}\overline{Z}_{k}^{T}\right]  ^{-1}E\left[  \overline{Z}_{k}\dot
{P}\left(  \widehat{B}-B\right)  H_{1}\right] \\
\widetilde{\overline{\alpha}}_{k}  &  =-E\ \left[  \dot{B}\dot{P}%
H_{1}\overline{Z}_{k}\overline{Z}_{k}^{T}\right]  ^{-1}E\ \left[  \overline
{Z}_{k}\dot{B}\left(  \widehat{P}-P\right)  H_{1}\right]
\end{align*}
and hence
\begin{align*}
\widetilde{B}-\widehat{B}  &  =\overset{\cdot}{-B}\overline{Z}_{k}%
^{T}E\ \left[  \dot{B}\dot{P}H_{1}\overline{Z}_{k}\overline{Z}_{k}^{T}\right]
^{-1}E\left[  \dot{P}\dot{B}\overline{Z}_{k}\left(  B-\widehat{B}\ \right)
\left\{  \dot{B}\right\}  ^{-1}H_{1}\right] \\
\widetilde{P}-\widehat{P}  &  =-\dot{P}\overline{Z}_{k}^{T}E_{\ }\left[
\dot{B}\dot{P}H_{1}\overline{Z}_{k}\overline{Z}_{k}^{T}\right]  ^{-1}E\left[
\dot{B}\dot{P}\overline{Z}_{k}\left(  P-\widehat{P}\right)  \left\{  \dot
{P}\right\}  ^{-1}H_{1}\right]
\end{align*}
Thus
\begin{align*}
Q\frac{\widetilde{P}-\widehat{P}}{\dot{P}}  &  =-Q\overline{Z}_{k}%
^{T}E_{\theta}\left[  Q^{2}\overline{Z}_{k}\overline{Z}_{k}^{T}\right]
^{-1}E\left[  Q^{2}\overline{Z}_{k}\frac{\left(  P-\widehat{P}\right)  }%
{\dot{P}}\right] \\
&  =\Pi\left[  \frac{P-\widehat{P}}{\dot{P}}Q|Q\overline{Z}_{k}\right] \\
Q\frac{\widetilde{B}-\widehat{B}}{\dot{B}}  &  =-Q\overline{Z}_{k}%
^{T}E_{\theta}\left[  Q^{2}\overline{Z}_{k}\overline{Z}_{k}^{T}\right]
^{-1}E\left[  Q^{2}\overline{Z}_{k}\frac{\left(  B-\widehat{B}\right)  }%
{\dot{B}}\right] \\
&  =\Pi\left[  \left(  \frac{B-\widehat{B}}{\dot{B}}\right)  Q|Q\overline
{Z}_{k}\right]
\end{align*}
and hence
\begin{align*}
Q\left(  \frac{\left(  P-\widetilde{P}\right)  }{\dot{P}}\right)   &
=\Pi^{\perp}\left[  \left(  \frac{P-\widehat{P}}{\dot{P}}\right)
Q|Q\overline{Z}_{k}\right] \\
Q\left(  \frac{B-\widetilde{B}}{\dot{B}}\right)   &  =\Pi^{\perp}\left[
\left(  \frac{B-\widehat{B}}{\dot{B}}\right)  Q|Q\overline{Z}_{k}\right]
\end{align*}
But by the previous theorem,
\[
TB_{k}=E\left[  \left\{  \widetilde{B}\ -B\right\}  \left\{  \widetilde{P}%
\ -P\right\}  H_{1}\right]  =E\left[  Q\left(  \frac{B-\widetilde{B}}{\dot{B}%
}\right)  Q\left(  \frac{\left(  P-\widetilde{P}\right)  }{\dot{P}}\right)
\right]
\]
proving the theorem.
\end{proof}

\begin{proof}
(Theorem \ref{TBrate}) Under our assumptions, the following holds uniformly
for $\theta\in\Theta.$%
\begin{align*}
TB_{k}^{2} &  =\left\{  E\left[  \Pi^{\perp}\left[  \left(  \frac
{P-\widehat{P}}{\dot{P}}\right)  Q|Q\overline{Z}_{k}\right]  \Pi^{\perp
}\left[  \left(  \frac{B-\widehat{B}}{\dot{B}}\right)  Q|Q\overline{Z}%
_{k}\right]  \right]  \right\}  ^{2}\\
&  \leq E\left\{  \Pi^{\perp}\left[  \left(  \frac{P-\widehat{P}}{\dot{P}%
}\right)  Q|Q\overline{Z}_{k}\right]  ^{2}\right\}  \times E\left\{
\Pi^{\perp}\left[  \left(  \frac{B-\widehat{B}}{\dot{B}}\right)
Q|Q\overline{Z}_{k}\right]  ^{2}\right\}
\end{align*}
by Cauchy Shwartz. Now
\begin{align*}
&  E\left\{  \Pi^{\perp}\left[  \left(  \frac{P-\widehat{P}}{\dot{P}}\right)
Q|Q\overline{Z}_{k}\right]  ^{2}\right\} \\
&  =inf_{\varsigma_{_{l}}}\int_{R^{d}}Q^{2}\left(  \frac{p\left(  X\right)
-\widehat{p}\left(  X\right)  }{\dot{p}\left(  X\right)  }-\sum_{l=1}%
^{k}\varsigma_{_{l}}z_{l}\left(  X\right)  \right)  ^{2}f\left(  X\right)
dX\\
&  =inf_{\varsigma_{_{l}}}\int_{R^{d}}Q^{2}\left(  \frac{\left(  p\left(
X\right)  -\widehat{p}\left(  X\right)  \right)  }{\dot{p}\left(  X\right)
}-\sum_{l=1}^{k}\varsigma_{_{l}}^{\ast}\varphi_{l}\left(  X\right)  \right)
^{2}Q^{2}f\left(  X\right)  dX\\
&  =inf_{\varsigma_{_{l}}}\int_{R^{d}}\left(  \left(  p\left(  X\right)
-\widehat{p}\left(  X\right)  \right)  -\sum_{l=1}^{k}\varsigma_{_{l}}^{\ast
}\varphi_{l}\left(  X\right)  \right)  ^{22}Q^{2}\ f\left(  X\right)  dX\\
&  \leq\left\vert \left\vert Q^{2}f\left(  X\right)  \right\vert \right\vert
_{\infty}inf_{\varsigma_{_{l}}}\int_{R^{d}}\left(  \left(  p\left(  X\right)
-\widehat{p}\left(  X\right)  \right)  -\sum_{l=1}^{k}\varsigma_{_{l}}^{\ast
}\varphi_{l}\left(  X\right)  \right)  ^{2}dX\\
&  \leq\left\vert \left\vert Q^{2}f\left(  X\right)  \right\vert \right\vert
_{\infty}O_{p}\left(  k^{-2\beta_{p}/d}\right)  =O_{p}\left(  k^{-2\beta
_{p}/d}\right)
\end{align*}
The last equality follows from the fact that under the stated assumptions
$\left\vert \left\vert Q^{2}f\left(  X\right)  \right\vert \right\vert
_{\infty}=O(1)$ . Similarly $E\left\{  \Pi^{\perp}\left[  \left(
\frac{B-\widehat{B}}{\dot{B}}\right)  Q|Q\overline{Z}_{k}\right]
^{2}\right\}  =O_{p}\left(  k^{-2\beta_{b}/d}\right)  .$\newline
\end{proof}

\begin{proof}
(Theorem \ref{DRHOIF}) By theorem \ref{FOIF}
\begin{align*}
IF_{1,\widetilde{\psi}_{k},i_{i}}\,\left(  \theta\right)   &
=if_{1,\widetilde{\psi}_{k}}\left(  O_{i_{1}},\theta\right) \\
&  =H_{i_{i}}\left(  \widetilde{b}\left(  \theta\right)  ,\widetilde{p}\left(
\theta\right)  \right)  -\widetilde{\psi}_{k}\left(  \theta\right) \\
&  =h\left(  O_{i_{1}},\widetilde{b}\left(  X_{i_{1}},\theta\right)
,\widetilde{p}\left(  X_{i_{1}},\theta\right)  \right)  -\widetilde{\psi}%
_{k}\left(  \theta\right)
\end{align*}
$\ \ $and by part 5.c of theorem \ref{eift}
\[
\mathbb{V}\left[  IF_{22,\widetilde{\psi}_{k},\overline{i}_{2}}\right]
=\frac{1}{2}\left\{  \Pi_{\theta}\left[  \mathbb{V}\left[
IF_{1,if_{1,\widetilde{\psi}_{k}}\left(  O_{i_{1}},\cdot\right)  ,i_{2}%
}\,\left(  \theta\right)  \right]  |\mathcal{U}_{1}^{\perp_{\theta,2}}\left(
\theta\right)  \right]  \right\}  .
\]
Now
\begin{align*}
IF_{1,if_{1,\widetilde{\psi}_{k},}\,\left(  O_{i_{1}},\theta\right)  ,i_{2}%
}\,\left(  \theta\right)   &  =h_{\widetilde{b}}\left(  O_{i_{1}%
},\widetilde{b}\left(  X_{i_{1}},\theta\right)  ,\widetilde{p}\left(
X_{i_{1}},\theta\right)  \right)  IF_{1,\widetilde{b}\left(  X_{i_{1}}%
,\cdot\right)  ,i_{2}}\,\left(  \theta\right) \\
&  +h_{\widetilde{p}}\left(  O_{i_{1}},\widetilde{b}\left(  X_{i_{1}}%
,\theta\right)  ,\widetilde{p}\left(  X_{i_{1}},\theta\right)  \right)
IF_{1,\widetilde{p}\left(  X_{i_{1}},\cdot\right)  ,i_{2}}\,\left(
\theta\right)
\end{align*}
where
\begin{align*}
h_{\widetilde{b}}\left(  O_{i_{1}},\widetilde{b}\left(  X_{i_{1}}%
,\theta\right)  ,\widetilde{p}\left(  X_{i_{1}},\theta\right)  \right)   &
=H_{1_{i_{1}}}\widetilde{p}\left(  X_{i_{1}},\theta\right)  +H_{2_{i_{1}}}\\
h_{\widetilde{p}}\left(  O_{i_{1}},\widetilde{b}\left(  X_{i_{1}}%
,\theta\right)  ,\widetilde{p}\left(  X_{i_{1}},\theta\right)  \right)   &
=H_{1i_{1}}\widetilde{b}\left(  X_{i_{1}},\theta\right)  +H_{3i_{1}}%
\end{align*}%
\begin{align*}
IF_{1,\widetilde{b}\left(  X_{i_{1}},\cdot\right)  ,i_{2}}\,\left(
\theta\right)   &  =IF_{1,b^{\ast}\left(  X_{i_{1}},\widetilde{\overline{\eta
}}_{k}\left(  \cdot\right)  \right)  ,i_{2}}\,\left(  \theta\right) \\
&  =\dot{B}_{i_{1}}\overline{Z}_{ki_{1}}^{T}IF_{1,\widetilde{\overline{\eta}%
}_{k}\left(  \cdot\right)  ,i_{2}}\,\left(  \theta\right) \\
&  =\overset{\cdot}{-B}\overline{Z}_{ki_{1}}^{T}\left\{  E_{\theta}\left[
\dot{P}\dot{B}H_{1}\overline{Z}_{k}\overline{Z}_{k}^{T}\right]  \right\}
^{-1}\left[  \left\{  H_{1}\widetilde{b}\left(  X,\theta\right)
+H_{3}\right\}  \dot{P}\overline{Z}_{k}\right]  _{i_{2}},
\end{align*}
and
\[
IF_{1,\widetilde{p}\left(  X_{i_{1}},\cdot\right)  ,i_{2}}\,\left(
\theta\right)  =-\dot{P}_{i_{1}}\overline{Z}_{ki_{1}}^{T}\left\{  E_{\theta
}\left[  \dot{P}\dot{B}H_{1}\overline{Z}_{k}\overline{Z}_{k}^{T}\right]
\right\}  ^{-1}\left[  \left\{  H_{1}\widetilde{p}\left(  X,\theta\right)
+H_{2}\right\}  \dot{B}\overline{Z}_{k}\right]  _{i_{2}}.
\]%
\begin{align*}
&  IF_{1,if_{1,\widetilde{\psi}_{k},}\,\left(  O_{i_{1}},\cdot\right)  ,i_{2}%
}\,\left(  \theta\right) \\
&  =-\left\{  H_{1}\widetilde{p}\left(  X,\theta\right)  +H_{2}\right\}
_{i_{1}}\dot{B}_{i_{1}}\overline{Z}_{ki_{1}}^{T}\left\{  E_{\theta}\left[
\dot{P}\dot{B}H_{1}\overline{Z}_{k}\overline{Z}_{k}^{T}\right]  \right\}
^{-1}\\
&  \times\left[  \left\{  H_{1}\widetilde{b}\left(  X,\theta\right)
+H_{3}\right\}  \dot{P}\overline{Z}_{k}\right]  _{i_{2}}\\
&  -\left\{  H_{1}\widetilde{b}\left(  X,\theta\right)  +H_{3}\right\}
_{i_{1}}\dot{P}_{i_{1}}\overline{Z}_{ki_{1}}^{T}\left\{  E_{\theta}\left[
\dot{P}\dot{B}H_{1}\overline{Z}_{k}\overline{Z}_{k}^{T}\right]  \right\}
^{-1}\\
&  \times\left[  \left\{  H_{1}\widetilde{p}\left(  X,\theta\right)
+H_{2}\right\}  \dot{B}\overline{Z}_{k}\right]  _{i_{2}}%
\end{align*}
and further
\[
\Pi_{\theta}\left[  \mathbb{V}\left[  IF_{1,if_{1,\widetilde{\psi}_{k}%
,}\,\left(  O_{i_{1}},\cdot\right)  ,i_{2}}\,\left(  \theta\right)  \right]
|\mathcal{U}_{1}\left(  \theta\right)  \right]  =0
\]
since
\[
E_{\theta}\left[  \left\{  H_{1}\widetilde{p}\left(  X,\theta\right)
+H_{2}\right\}  \dot{B}\overline{Z}_{k}\right]  =E_{\theta}\left[  \left\{
H_{1}\widetilde{b}\left(  X,\theta\right)  +H_{3}\right\}  \dot{P}\overline
{Z}_{k}\right]  =0
\]
and thus $IF_{1,if_{1,\widetilde{\psi}_{k},}\,\left(  O_{i_{1}},\cdot\right)
,i_{2}}\,\left(  \theta\right)  $ is degenerate. \ Because
$IF_{1,if_{1,\widetilde{\psi}_{k},}\,\left(  O_{i_{1}},\cdot\right)  ,i_{2}%
}\,\left(  \theta\right)  $ has two terms, it appears that
$IF_{22,\widetilde{\psi}_{k},\overline{i}_{2}}$ will consist of two terms.
\ However by the symmetry upon interchange of $i_{2}$ and $i_{1},$\ and the
permutation invariance of the operator $\mathbb{V}$
\begin{align*}
&  \mathbb{V}\left[  IF_{1,if_{1,\widetilde{\psi}_{k},}\,\left(  O_{i_{1}%
},\cdot\right)  ,i_{2}}\,\left(  \theta\right)  \right] \\
&  =\mathbb{V}\left[
\begin{array}
[c]{c}%
-2\left\{  H_{1}\widetilde{p}\left(  X,\theta\right)  +H_{2}\right\}  _{i_{1}%
}\dot{B}_{i_{1}}\overline{Z}_{ki_{1}}^{T}\left\{  E_{\theta}\left[  \dot
{P}\dot{B}H_{1}\overline{Z}_{k}\overline{Z}_{k}^{T}\right]  \right\}  ^{-1}\\
\times\left[  \overline{Z}_{k}\left\{  H_{1}\widetilde{b}\left(
X,\theta\right)  +H_{3}\right\}  \dot{P}\right]  _{i_{2}}%
\end{array}
\right]
\end{align*}
Thus we can take
\begin{align*}
IF_{22,\widetilde{\psi}_{k},\overline{i}_{2}}  &  =-\left\{  H_{1}%
\widetilde{p}\left(  X,\theta\right)  +H_{2}\right\}  _{i_{1}}\dot{B}_{i_{1}%
}\overline{Z}_{ki_{1}}^{T}\left\{  E_{\theta}\left[  \dot{P}\dot{B}%
H_{1}\overline{Z}_{k}\overline{Z}_{k}^{T}\right]  \right\}  ^{-1}\\
&  \times\left[  \overline{Z}_{k}\left\{  H_{1}\widetilde{b}\left(
X,\theta\right)  +H_{3}\right\}  \dot{P}\right]  _{i_{2}}%
\end{align*}
as was to be proved. \ We now complete the proof of the Theorem by induction.
\ We assume it is true for $IF_{mm,\widetilde{\psi}_{k},\overline{i}_{m\ }}$
and prove it is true for $IF_{\left(  m+1\right)  \left(  m+1\right)
,\widetilde{\psi}_{k},\overline{i}_{m+1}}.$ Now%
\[
\mathbb{V}\left[  IF_{\left(  m+1\right)  \left(  m+1\right)  ,\widetilde{\psi
}_{k},\overline{i}_{m+1}}\left(  \theta\right)  \right]  =\frac{1}%
{m}\mathbb{V}\left[  \Pi_{\theta}\left[  IF_{1,if_{mm,\widetilde{\psi}_{k}%
,}\,\left(  O_{\overline{i}_{m}},\cdot\right)  ,i_{m+1}}\,\left(
\theta\right)  |\mathcal{U}_{m}^{\perp_{\theta,m+1}}\left(  \theta\right)
\right]  \right]
\]
Now by the induction hypothesis,
\begin{align*}
&  if_{mm,\widetilde{\psi}_{k},}\,\left(  O_{\overline{i}_{m}},\theta\right)
\\
&  =\left(  -1\right)  ^{m-1}\left[  \left(  H_{1}\widetilde{P}\left(
\theta\right)  +H_{2}\right)  \dot{B}\overline{Z}_{k}^{T}\right]  _{i_{1}}\\
&  \times\left[
%TCIMACRO{\dprod \limits_{s=3}^{m}}%
%BeginExpansion
{\displaystyle\prod\limits_{s=3}^{m}}
%EndExpansion
\left\{  E_{\theta}\left[  \dot{P}\dot{B}H_{1}\overline{Z}_{k}\overline{Z}%
_{k}^{T}\right]  \right\}  ^{-1}\left\{
\begin{array}
[c]{c}%
\left(  \dot{P}\dot{B}H_{1}\overline{Z}_{k}\overline{Z}_{k}^{T}\right)
_{i_{s}}\\
-E_{\theta}\left[  \dot{P}\dot{B}H_{1}\overline{Z}_{k}\overline{Z}_{k}%
^{T}\right]
\end{array}
\right\}  \right] \\
&  \times\left\{  E_{\theta}\left[  \dot{P}\dot{B}H_{1}\overline{Z}%
_{k}\overline{Z}_{k}^{T}\right]  \right\}  ^{-1}\left[  \overline{Z}%
_{k}\left(  H_{1}\widetilde{B}\left(  \theta\right)  +H_{3}\right)  \dot
{P}\right]  _{i_{2}}%
\end{align*}
The derivatives with respect to the $\theta^{\prime}s$ in $\widetilde{P}%
\left(  \theta\right)  ,\widetilde{B}\left(  \theta\right)  \ $and in the
$m-1$ terms $\left\{  E_{\theta}\left[  \dot{P}\dot{B}H_{1}\overline{Z}%
_{k}\overline{Z}_{k}^{T}\right]  \right\}  ^{-1}$ will each contribute a term
to $\mathbb{V}\left[  IF_{\left(  m+1\right)  \left(  m+1\right)
,\widetilde{\psi}_{k},\overline{i}_{m+1}}\left(  \theta\right)  \right]  .$
\ However differentiating wrt to the $\theta$ in the $m-2$ terms $E_{\theta
}\left[  \dot{P}\dot{B}H_{1}\overline{Z}_{k}\overline{Z}_{k}^{T}\right]  $
will not contribute to $\mathbb{V}\left[  IF_{\left(  m+1\right)  \left(
m+1\right)  ,\widetilde{\psi}_{k},\overline{i}_{m+1}}\left(  \theta\right)
\right]  $ as the contribution from each of these $m-2$ terms to
$IF_{1,if_{mm,\widetilde{\psi}_{k},}\,\left(  O_{\overline{i}_{m}}%
,\cdot\right)  ,i_{m+1}}\,\left(  \theta\right)  $ is only a function of $m$
units' data and is thus an element of $\mathcal{U}_{m}\left(  \theta\right)  $
which is orthogonal to the space $\mathcal{U}_{m}^{\perp_{\theta,m+1}}\left(
\theta\right)  $ that is projected on.. Now
\begin{align*}
&  IF_{1,\left\{  E_{\theta}\left[  \dot{P}\dot{B}H_{1}\overline{Z}%
_{k}\overline{Z}_{k}^{T}\right]  \right\}  ^{-1},i_{m+1}}\left(  \theta\right)
\\
&  =-\left\{  E_{\theta}\left[  \dot{P}\dot{B}H_{1}\overline{Z}_{k}%
\overline{Z}_{k}^{T}\right]  \right\}  ^{-1}\left\{
\begin{array}
[c]{c}%
\left(  \dot{P}\dot{B}H_{1}\overline{Z}_{k}\overline{Z}_{k}^{T}\right)
_{i_{m+1}}\\
-E_{\theta}\left[  \dot{P}\dot{B}H_{1}\overline{Z}_{k}\overline{Z}_{k}%
^{T}\right]
\end{array}
\right\}  \left\{  E_{\theta}\left[  \dot{P}\dot{B}H_{1}\overline{Z}%
_{k}\overline{Z}_{k}^{T}\right]  \right\}  ^{-1}%
\end{align*}
so upon permuting the unit indices, the contribution of each of these $m-1$
terms to $IF_{1,if_{mm,\widetilde{\psi}_{k},}\,\left(  O_{\overline{i}_{m}%
},\cdot\right)  ,i_{m+1}}\,\left(  \theta\right)  $ is%
\begin{align}
&  -\left(  -1\right)  ^{m-1}\left[  \left(  H_{1}\widetilde{P}\left(
\theta\right)  +H_{2}\right)  \dot{B}\overline{Z}_{k}^{T}\right]  _{i_{1}%
}\label{IFm+1}\\
&  \times\left[
%TCIMACRO{\dprod \limits_{s=3}^{m+1}}%
%BeginExpansion
{\displaystyle\prod\limits_{s=3}^{m+1}}
%EndExpansion
\left\{  E_{\theta}\left[  \dot{P}\dot{B}H_{1}\overline{Z}_{k}\overline{Z}%
_{k}^{T}\right]  \right\}  ^{-1}\left\{
\begin{array}
[c]{c}%
\left(  \dot{P}\dot{B}H_{1}\overline{Z}_{k}\overline{Z}_{k}^{T}\right)
_{i_{s}}\\
-E_{\theta}\left[  \dot{P}\dot{B}H_{1}\overline{Z}_{k}\overline{Z}_{k}%
^{T}\right]
\end{array}
\right\}  \right] \nonumber\\
&  \times\left\{  E_{\theta}\left[  \dot{P}\dot{B}H_{1}\overline{Z}%
_{k}\overline{Z}_{k}^{T}\right]  \right\}  ^{-1}\left[  \overline{Z}%
_{k}\left(  H_{1}\widetilde{B}\left(  \theta\right)  +H_{3}\right)  \dot
{P}\right]  _{i_{2}}\nonumber
\end{align}
\ \ which is already degenerate ( i.e., orthogonal to $\mathcal{U}_{m}\left(
\theta\right)  ).$ \ Differentiating with respect to the $\theta^{\prime}s$ of
$\widetilde{P}\left(  \theta\right)  ,\widetilde{B}\left(  \theta\right)
\ $in $IF_{1,if_{mm,\widetilde{\psi}_{k},}\,\left(  O_{\overline{i}_{m}}%
,\cdot\right)  ,i_{m+1}}\,\left(  \theta\right)  $ we obtain
\begin{align*}
&  =\left(  -1\right)  ^{m-1}IF_{1,\widetilde{b}\left(  X_{i_{1}}%
,\cdot\right)  ,i_{m+1}}\,\left(  \theta\right)  \left[  H_{1}\dot{B}%
\overline{Z}_{k}^{T}\right]  _{i_{1}}\\
&  \left[
%TCIMACRO{\dprod \limits_{s=3}^{m}}%
%BeginExpansion
{\displaystyle\prod\limits_{s=3}^{m}}
%EndExpansion
\left\{  E_{\theta}\left[  \dot{P}\dot{B}H_{1}\overline{Z}_{k}\overline{Z}%
_{k}^{T}\right]  \right\}  ^{-1}\left\{  \left(  \dot{P}\dot{B}H_{1}%
\overline{Z}_{k}\overline{Z}_{k}^{T}\right)  _{i_{s}}-E_{\theta}\left[
\dot{P}\dot{B}H_{1}\overline{Z}_{k}\overline{Z}_{k}^{T}\right]  \right\}
\right]  \times\\
&  \left\{  E_{\theta}\left[  \dot{P}\dot{B}H_{1}\overline{Z}_{k}\overline
{Z}_{k}^{T}\right]  \right\}  ^{-1}\left[  \overline{Z}_{k}\left(
H_{1}\widetilde{B}\left(  \theta\right)  +H_{3}\right)  \dot{P}\right]
_{i_{2}}\\
&  +\left(  -1\right)  ^{m-1}\left[  \left(  H_{1}\widetilde{P}\left(
\theta\right)  +H_{2}\right)  \dot{B}\overline{Z}_{k}^{T}\right]  _{i_{1}%
}\times\\
&  \left[
%TCIMACRO{\dprod \limits_{s=3}^{m}}%
%BeginExpansion
{\displaystyle\prod\limits_{s=3}^{m}}
%EndExpansion
\left\{  E_{\theta}\left[  \dot{P}\dot{B}H_{1}\overline{Z}_{k}\overline{Z}%
_{k}^{T}\right]  \right\}  ^{-1}\left\{  \left(  \dot{P}\dot{B}H_{1}%
\overline{Z}_{k}\overline{Z}_{k}^{T}\right)  _{i_{s}}-E_{\theta}\left[
\dot{P}\dot{B}H_{1}\overline{Z}_{k}\overline{Z}_{k}^{T}\right]  \right\}
\right]  \times\\
&  \left[  \overline{Z}_{k}H_{1}\dot{P}\right]  IF_{1,\widetilde{p}\left(
X_{i_{2}},\cdot\right)  ,i_{m+1}}\,\left(  \theta\right)
\end{align*}
Substituting in the above expressions for $IF_{1,\widetilde{b}\left(
X_{i_{1}},\cdot\right)  ,i_{m+1}}\,\left(  \theta\right)  $ and
$IF_{1,\widetilde{p}\left(  X_{i_{2}},\cdot\right)  ,i_{m+1}}\,\left(
\theta\right)  ,$ then projecting on $\mathcal{U}_{m}^{\perp_{\theta,m+1}%
}\left(  \theta\right)  ,$ and again permuting unit indices, we obtain two
identical terms both equal to eq ($\ref{IFm+1})$ . Thus we obtain $m+1$
identical terms in all. Upon dividing by $m+1$, we conclude that
$\mathbb{V}\left[  IF_{\left(  m+1\right)  \left(  m+1\right)
,\widetilde{\psi}_{k},\overline{i}_{m+1}}\left(  \left(  \theta\right)
\right)  \right]  \ $equals $\mathbb{V}$ operating on $\left(  \ref{IFm+1}%
\right)  ,$ proving the theorem.\newline
\end{proof}

\begin{proof}
(Theorem \ref{EBrate}) Equation (\ref{EB6}) follows from eq. (\ref{EB4}) by
our assumption of rate optimality of the initial estimators. We next prove eq.
(\ref{EB3}) by induction. \ For $m=1.$%
\begin{align*}
EB_{1}  &  =E\left[  \widehat{\psi}_{1,k}\right]  -\widetilde{\psi}_{k}\\
&  =E\left[  \widehat{B}\widehat{P}H_{1}+\widehat{B}H_{2}+\widehat{P}%
H_{3}\right]  -E\left[  \widetilde{B}\widetilde{P}H_{1}+\widetilde{B}%
H_{2}+\widetilde{P}H_{3}\right] \\
&  =E\left[  \dot{B}\dot{P}\overline{Z}_{k}^{T}\left(  P-\widehat{P}\right)
\left\{  \dot{P}\right\}  ^{-1}H_{1}\right]  E_{\theta}\left[  \dot{B}\dot
{P}H_{1}\overline{Z}_{k}\overline{Z}_{k}^{T}\right]  ^{-1}\\
&  \times E\left[  \dot{P}\dot{B}\overline{Z}_{k}\left(  B-\widehat{B}%
\ \right)  \left\{  \dot{B}\right\}  ^{-1}H_{1}\right]
\end{align*}

where the last equality follows from
\begin{align*}
\widetilde{B}  &  =\overset{\cdot}{\widehat{B}-B}\overline{Z}_{k}^{T}%
E_{\theta}\left[  \dot{B}\dot{P}H_{1}\overline{Z}_{k}\overline{Z}_{k}%
^{T}\right]  ^{-1}E\left[  \dot{P}\dot{B}\overline{Z}_{k}\left(
B-\widehat{B}\ \right)  \left\{  \dot{B}\right\}  ^{-1}H_{1}\right] \\
\widetilde{P}  &  =\widehat{P}-\dot{P}\overline{Z}_{k}^{T}E_{\theta}\left[
\dot{B}\dot{P}H_{1}\overline{Z}_{k}\overline{Z}_{k}^{T}\right]  ^{-1}E\left[
\dot{B}\dot{P}\overline{Z}_{k}\left(  P-\widehat{P}\right)  \left\{  \dot
{P}\right\}  ^{-1}H_{1}\right]
\end{align*}

Next, we proceed by induction. Assume \ref{EB3} holds for $m-1\geq1,$ we next
show that it holds for $m.$%
\begin{align*}
EB_{m}  &  =EB_{m-1}+E\left[  IF_{mm,\widetilde{\psi}_{k},\overline{i}_{j}%
}\right] \\
&  =\left\{
\begin{array}
[c]{c}%
\left(  -1\right)  ^{m-2}E\left[  Q^{2}\left(  \frac{B-\widehat{B}}{\dot{B}%
}\right)  \overline{Z}_{k}^{T}\right] \\
\left\{  E\left[  Q^{2}\overline{Z}_{k}\overline{Z}_{k}^{T}\right]
-I_{k\times k}\right\}  ^{m-2}\left\{  E\left[  Q^{2}\overline{Z}_{k}%
\overline{Z}_{k}^{T}\right]  \right\}  ^{-1}E\left[  Q^{2}\left(
\frac{P-\widehat{P}}{\dot{P}}\right)  \right]
\end{array}
\right\} \\
&  +\left\{
\begin{array}
[c]{c}%
\left(  -1\right)  ^{m-1}E\left[  Q^{2}\left(  \frac{B-\widehat{B}}{\dot{B}%
}\right)  \overline{Z}_{k}^{T}\right] \\
\times\left\{  E\left[  Q^{2}\overline{Z}_{k}\overline{Z}_{k}^{T}\right]
-I_{k\times k}\right\}  ^{m-2}E\left[  Q^{2}\left(  \frac{P-\widehat{P}}%
{\dot{P}}\right)  \right]
\end{array}
\right\} \\
&  =\left(  -1\right)  ^{m-1}E\left[  Q^{2}\left(  \frac{B-\widehat{B}}%
{\dot{B}}\right)  \overline{Z}_{k}^{T}\right]  \left\{  E\left[
Q^{2}\overline{Z}_{k}\overline{Z}_{k}^{T}\right]  -I_{k\times k}\right\}
^{m-1}\\
&  \times\left\{  E\left[  Q^{2}\overline{Z}_{k}\overline{Z}_{k}^{T}\right]
\right\}  ^{-1}E\left[  Q^{2}\left(  \frac{P-\widehat{P}}{\dot{P}}\right)
\right]
\end{align*}

Finally we prove that \ref{EB3} implies \ref{EB4}. For any random variable $H
$ define
\begin{align*}
\widehat{R}\left(  H\right)   &  =\widehat{\Pi}\left[  \delta gH|\widehat{Q}%
\overline{Z}_{k}\right]  =\widehat{Q}\overline{Z}_{k}^{T}\widehat{E}\left[
\delta g\widehat{Q}\overline{Z}_{k}H\right] \\
\widehat{R}^{t}\left(  H\right)   &  =\widehat{R}\circ\widehat{R}^{t-1}\left(
H\right)  \text{ \ for }t\geq2
\end{align*}

where $g\left(  X\right)  =E\left(  H_{1}|X\right)  f\left(  X\right)  $ and
$\delta g=\frac{g\left(  x\right)  -\widehat{g}\left(  X\right)  }%
{\widehat{g}\left(  X\right)  }=\frac{Q^{2}f-\widehat{Q}^{2}\widehat{f}%
}{\widehat{Q}^{2}\widehat{f}}.$

Then
\begin{align*}
&  \left(  -1\right)  ^{m-1}EB_{m}\\
&  =\left\{
\begin{array}
[c]{c}%
E\left[  Q^{2}\left(  \frac{B-\widehat{B}}{\dot{B}}\right)  \overline{Z}%
_{k}^{T}\right]  \left\{  \widehat{E}\left[  \delta g\text{ }\widehat{Q}%
^{2}\overline{Z}_{k}\overline{Z}_{k}^{T}\right]  \right\}  ^{m-2}\times\\
\widehat{E}\left[  \delta g\text{ }\widehat{Q}\overline{Z}_{k}\frac
{\widehat{Q}}{Q}\Pi\left(  Q\overline{Z}_{k}\left(  \frac{P-\widehat{P}}%
{\dot{P}}\right)  |\left(  Q\overline{Z}_{k}\right)  \right)  \right]
\end{array}
\right\} \\
&  =E\left[  Q\left(  \frac{B-\widehat{B}}{\dot{B}}\right)  \frac
{Q}{\widehat{Q}}\widehat{R}^{m-1}\left(  \frac{\widehat{Q}}{Q}\Pi\left(
Q\left(  \frac{P-\widehat{P}}{\dot{P}}\right)  |\left(  Q\overline{Z}%
_{k}\right)  \right)  \right)  \right]
\end{align*}%
\begin{align*}
&  =\widehat{E}\left[  Q\left(  \frac{B-\widehat{B}}{\dot{B}}\right)
\frac{Qf}{\widehat{Q}\widehat{f}}\widehat{R}^{m-1}\left(  \frac{\widehat{Q}%
}{Q}\Pi\left(  Q\left(  \frac{P-\widehat{P}}{\dot{P}}\right)  |\left(
Q\overline{Z}_{k}\right)  \right)  \right)  \right] \\
&  \leq\left\{  \widehat{E}\left[  Q\left(  \frac{B-\widehat{B}}{\dot{B}%
}\right)  \frac{Qf}{\widehat{Q}\widehat{f}}\right]  ^{2}\right\}  ^{\frac
{1}{2}}\left\{  \widehat{E}\left[  \widehat{R}^{m-1}\left(  \frac{\widehat{Q}%
}{Q}\Pi\left(  Q\left(  \frac{P-\widehat{P}}{\dot{P}}\right)  |\left(
Q\overline{Z}_{k}\right)  \right)  \right)  \right]  ^{2}\right\}  ^{\frac
{1}{2}}%
\end{align*}
by Cauchy Shwartz.

Now
\begin{align*}
&  \widehat{E}\left[  \widehat{R}^{m-1}\left(  \frac{\widehat{Q}}{Q}\Pi\left(
Q\left(  \frac{P-\widehat{P}}{\dot{P}}\right)  |\left(  Q\overline{Z}%
_{k}\right)  \right)  \right)  \right]  ^{2}\\
&  \leq\widehat{E}\left[  \left(  \delta g\right)  ^{2}\left[  \widehat{R}%
^{m-2}\left(  \frac{\widehat{Q}}{Q}\Pi\left(  Q\left(  \frac{P-\widehat{P}%
}{\dot{P}}\right)  |\left(  Q\overline{Z}_{k}\right)  \right)  \right)
\right]  ^{2}\right]
\end{align*}

by the projection operator having operator norm equal to 1
\[
\leq\left\vert \left\vert \delta g\right\vert \right\vert _{\infty}%
^{2}\widehat{E}\left[  \widehat{R}^{m-2}\left(  \frac{\widehat{Q}}{Q}%
\Pi\left(  Q\left(  \frac{P-\widehat{P}}{\dot{P}}\right)  |\left(
Q\overline{Z}_{k}\right)  \right)  \right)  \right]  ^{2}%
\]

Iterating this calculation $m-1$ times we find
\begin{align*}
&  \widehat{E}\left[  \widehat{R}^{m-1}\left(  \frac{\widehat{Q}}{Q}\Pi\left(
Q\left(  \frac{P-\widehat{P}}{\dot{P}}\right)  |\left(  Q\overline{Z}%
_{k}\right)  \right)  \right)  \right]  ^{2}\\
&  \leq\left\vert \left\vert \delta g\right\vert \right\vert _{\infty
}^{2\left(  m-1\right)  }\widehat{E}\left[  \frac{\widehat{Q}}{Q}\Pi\left(
Q\left(  \frac{P-\widehat{P}}{\dot{P}}\right)  |\left(  Q\overline{Z}%
_{k}\right)  \right)  \right]  ^{2}\\
&  \leq\left\vert \left\vert \delta g\right\vert \right\vert _{\infty
}^{2\left(  m-1\right)  }\left\vert \left\vert \frac{\widehat{G}}%
{G}\right\vert \right\vert _{\infty}E\left(  \Pi\left(  Q\left(
\frac{P-\widehat{P}}{\dot{P}}\right)  |\left(  Q\overline{Z}_{k}\right)
\right)  \right)  ^{2}\\
&  \leq\left\vert \left\vert \delta g\right\vert \right\vert _{\infty
}^{2\left(  m-1\right)  }\left\vert \left\vert \frac{\widehat{G}}%
{G}\right\vert \right\vert _{\infty}\int Q^{2}\left(  \frac{P-\widehat{P}%
}{\dot{P}}\right)  ^{2}f\left(  X\right)  dX\\
&  \leq\left\vert \left\vert \delta g\right\vert \right\vert _{\infty
}^{2\left(  m-1\right)  }\left\vert \left\vert \frac{\widehat{G}}%
{G}\right\vert \right\vert _{\infty}\left\vert \left\vert \frac{Q^{2}f}%
{\dot{P}^{2}}\right\vert \right\vert _{\infty}\int\left(  p\left(  X\right)
-\widehat{p}\left(  X\right)  \right)  ^{2}dX\\
&  =\left\vert \left\vert \delta g\right\vert \right\vert _{\infty}^{2\left(
m-1\right)  }\left\vert \left\vert \frac{Q^{2}f}{\dot{P}^{2}}\right\vert
\right\vert _{\infty}\left[  \int\left(  p\left(  X\right)  -\widehat{p}%
\left(  X\right)  \right)  ^{2}dX\right]  \left(  1+o_{p}\left(  1\right)
\right)
\end{align*}

by $\frac{\widehat{G}}{G}=\left(  1+o_{p}\left(  1\right)  \right)  .$ \ Next
\begin{align*}
&  \widehat{E}\left[  Q\left(  \frac{B-\widehat{B}}{\dot{B}}\right)  \frac
{Qf}{\widehat{Q}\widehat{f}}\right]  ^{2}\\
&  =\int\frac{Q^{2}f}{\dot{B}^{2}}\frac{G}{\widehat{G}}\left(  b\left(
X\right)  -\widehat{b}\left(  X\right)  \right)  ^{2}dX\\
&  \leq\left\vert \left\vert \frac{Q^{2}f}{\dot{B}^{2}}\right\vert \right\vert
_{\infty}\left\vert \left\vert \frac{G}{\widehat{G}}\right\vert \right\vert
_{\infty}\int\left(  b\left(  X\right)  -\widehat{b}\left(  X\right)  \right)
^{2}dX\\
&  =\left\vert \left\vert \frac{Q^{2}f}{\dot{B}^{2}}\right\vert \right\vert
_{\infty}\left[  \int\left(  b\left(  X\right)  -\widehat{b}\left(  X\right)
\right)  ^{2}dX\right]  \left(  1+o_{p}\left(  1\right)  \right)
\end{align*}

Then we know%

\[
|EB_{m}|\leq\left\{
\begin{array}
[c]{c}%
\left\vert \left\vert \delta g\right\vert \right\vert _{\infty}^{m-1}%
\left\vert \left\vert \left(  \frac{\dot{B}}{\dot{P}}G\right)  \right\vert
\right\vert _{\infty}^{\frac{1}{2}}\left\vert \left\vert \frac{\dot{P}}%
{\dot{B}}G\right\vert \right\vert _{\infty}^{\frac{1}{2}}\left(
1+o_{p}\left(  1\right)  \right)  \times\\
\left\{  \int\left(  p\left(  X\right)  -\widehat{p}\left(  X\right)  \right)
^{2}dX\right\}  ^{\frac{1}{2}}\left\{  \int\left(  b\left(  X\right)
-\widehat{b}\left(  X\right)  \right)  ^{2}dX\right\}  ^{\frac{1}{2}}%
\end{array}
\right\}
\]

\end{proof}

To prove theorem \ref{var_multi}, we first give the following univariate
result. Suppose that we have a set
\[
\left\{  \overline{\phi}_{k_{j}}^{k_{j+1}}\left(  X\right)  ,\text{
}j=0,1,...,2m\right\}
\]
such that for each $\left(  k_{j},k_{j+1}\right)  $ pair, $\overline{\phi
}_{k_{j}}^{kj+1}\left(  X\right)  $ either spans $\mathcal{V}_{\log_{2}\left(
k^{\ast}\right)  }$ or spans $\bigoplus\limits_{v=\log_{2}\left(
k_{j}-1\right)  }^{\log_{2}\left(  k_{j+1}\right)  -1}\mathcal{W}_{v}$ where
$\log_{2}\left(  k_{j+1}\right)  -1$ and $\log_{2}\left(  k_{j}-1\right)  $
are both nonnegative integers and $\log_{2}\left(  k_{j}-1\right)  \leq
\log_{2}\left(  k_{j+1}\right)  -1.$ Let $K_{\left(  k_{j},k_{j+1}\right)
}\left(  X,Y\right)  \equiv\overline{\phi}_{k_{j}}^{kj+1}\left(  X\right)
^{T}\overline{\phi}_{k_{j}}^{kj+1}\left(  Y\right)  ,$ we first introduce an
important preliminary lemma.

\begin{lemma}
\label{var_kernel}For $m\geq0$ and $1\leq j\leq m+1,$ if\ $\ k^{\ast}\leq
k_{2j-1}\asymp k^{\ast},$ then $k_{2j-2}=1;$ otherwise if $k_{2j-1}\gg
k^{\ast},$ then $\log_{2}\left(  k_{2j-2}-1\right)  $ and $\log_{2}\left(
k_{2j-1}\right)  -1$ are both nonnegative integers and $k_{2j-2}=o\left(
k_{2j-1}\right)  $. \ Thus, \
\[
\left(  E\left[
%TCIMACRO{\dprod \limits_{j=1}^{m+1}}%
%BeginExpansion
{\displaystyle\prod\limits_{j=1}^{m+1}}
%EndExpansion
K_{\left(  k_{2j-2},k_{2j-1}\right)  }\left(  X,X\right)  \right]  \right)
^{-1}\left(  \Pi_{j=1}^{m+1}k_{2j-1}\right)  =O\left(  1\right)
\]
\newline
\end{lemma}

\begin{proof}
(Lemma \ref{var_kernel}) case 1: $m=0,$%
\begin{align*}
&  E\left[  K_{\left(  k_{0},k_{1}\right)  }\left(  X,X\right)  \right] \\
&  =E\left[  \overline{\phi}_{k_{0}}^{k_{1}}\left(  X\right)  ^{T}%
\overline{\phi}_{k_{0}}^{k_{1}}\left(  X\right)  \right] \\
&  =k_{1}-k_{0}%
\end{align*}

which follows from the orthonormality of wavelets.

case 2: $m=1.$ \ If $k_{1}\asymp k_{3}\asymp k^{\ast},$ then,
\begin{align*}
&  E\left[  K_{\left(  k_{0},k_{1}\right)  }\left(  X,X\right)  K_{\left(
k_{2},k_{3}\right)  }\left(  X,X\right)  \right] \\
&  \geq E\left[  \left(  K_{\left(  1,k^{\ast}\right)  }\left(  X,X\right)
\right)  ^{2}\right]  \geq\left(  E\left[  K_{\left(  1,k^{\ast}\right)
}\left(  X,X\right)  \right]  \right)  ^{2}=\left(  k^{\ast}\right)
^{2}\asymp k_{1}k_{3}.
\end{align*}
Similarly, if $k_{1}\asymp k_{3}\gg k^{\ast},$ and we further assume
$k_{1}<k_{3}$ WLOG, \ then%
\begin{align*}
&  E\left[  K_{\left(  k_{0},k_{1}\right)  }\left(  X,X\right)  K_{\left(
k_{2},k_{3}\right)  }\left(  X,X\right)  \right] \\
&  \geq E\left[  \left(  K_{\left(  \frac{k_{1}}{2}+1,k_{1}\right)  }\left(
X,X\right)  \right)  ^{2}\right]  \geq\left(  E\left[  K_{\left(  \frac{k_{1}%
}{2}+1,k_{1}\right)  }\left(  X,X\right)  \right]  \right)  ^{2}\\
&  =\left(  \frac{k_{1}}{2}\right)  ^{2}\asymp k_{1}k_{3};
\end{align*}
otherwise, WLOG, we assume that $k_{1}=o\left(  k_{3}\right)  ,$ then%
\begin{align*}
&  E\left[  K_{\left(  k_{0},k_{1}\right)  }\left(  X,X\right)  K_{\left(
k_{2},k_{3}\right)  }\left(  X,X\right)  \right] \\
&  =E\left[  \sum_{r=k_{0}}^{k_{1}}\phi_{r}^{2}\left(  X\right)  \sum
_{s=k_{2}}^{k_{3}}\phi_{s}^{2}\left(  X\right)  \right] \\
&  \geq E\left[  \sum_{r=k_{0}}^{k_{1}}\phi_{r}^{2}\left(  X\right)
\sum_{s=\frac{k_{3}}{2}+1}^{k_{3}}\phi_{s}^{2}\left(  X\right)  \right]
\end{align*}

Here we further assume that $k_{1}=k^{\ast}.$ The proof for the case when
$k_{1}>k^{\ast}$ follows similarly.

Let $\phi^{w}\left(  x\right)  $ indicate the correponding compactly supported
father/mother wavelet on $\left[  L_{w},U_{w}\right]  ,$ whose absolute value
is bounded above by $M_{w}.$ By the continuity of $\phi^{w}\left(  x\right)
,$ there exists a set $A,$ which is a union of finite number of disjoint open
intervals, such that $\left\vert \phi^{w}\left(  x\right)  \right\vert $ is
greater than $\sqrt{\frac{1}{2\left(  U_{w}-L_{w}\right)  }}$ on $A$.
moreover, the Lebesgue measure (i.e., the length) of $\ A$ is greater than
$\frac{1}{2M_{w}^{2}}$ because $\phi^{w}\left(  x\right)  $ is bounded and has
unit length. \ Specifically,%
\begin{align*}
1  &  =\int\left(  \phi^{w}\left(  x\right)  \right)  ^{2}dx\\
&  =\int_{\left\vert \phi^{w}\left(  x\right)  \right\vert ^{2}\leq\frac
{1}{2\left(  U_{w}-L_{w}\right)  }}\left(  \phi^{w}\left(  x\right)  \right)
^{2}dx+\int_{\left\vert \phi^{w}\left(  x\right)  \right\vert ^{2}>\frac
{1}{2\left(  U_{w}-L_{w}\right)  }}\left(  \phi^{w}\left(  x\right)  \right)
^{2}dx\\
&  \leq\frac{1}{2\left(  U_{w}-L_{w}\right)  }\left(  U_{w}-L_{w}\right)
+M_{w}^{2}\mu\left(  A\right)  .
\end{align*}

By definition $\left\{  \phi_{r}\left(  X\right)  :1\leq r\leq k^{\ast
}\right\}  $ is a sequence of level $\log_{2}k^{\ast}$ scaled and translated
father wavelets on the unit interval $\left[  0,1\right]  $, \ therefore it is
obvious that the lebesgue measure of the set
\begin{align*}
\widetilde{A}  &  \equiv\left\{  x:\sum_{r=1}^{k^{\ast}}\phi_{r}^{2}\left(
x\right)  >\frac{k^{\ast}}{2\left(  U_{w}-L_{w}\right)  }\right\} \\
&  \supset%
%TCIMACRO{\dbigcup \limits_{r=1}^{k^{\ast}}}%
%BeginExpansion
{\displaystyle\bigcup\limits_{r=1}^{k^{\ast}}}
%EndExpansion
\left\{  x:\frac{1}{k^{\ast}}\phi_{r}^{2}\left(  x\right)  >\frac{1}{2\left(
U_{w}-L_{w}\right)  }\right\}
\end{align*}
is greater than $\frac{1}{2M_{w}^{2}\left(  U_{w}-L_{w}\right)  }.$
Furthermore, for $1\leq r\leq k^{\ast}$, the set $\left\{  x:\frac{1}{k^{\ast
}}\phi_{r}^{2}\left(  x\right)  >\frac{1}{2\left(  U_{w}-L_{w}\right)
}\right\}  $ consists of multiple disjoint open intervals whose lengths are
all of order $\frac{1}{k^{\ast}}.$\ In contrast, the support for any level
$\log_{2}k_{3}-1$ scaled and translated mother wavelet $\phi_{r},$
\ $\frac{k_{3}}{2}+1\leq r\leq k_{3}$, is of order $k_{3}^{-1}\ll\left(
k^{\ast}\right)  ^{-1}$ as $k^{\ast}=o\left(  k_{3}\right)  .$ Therefore, \ at
least $\frac{1}{2}\mu\left(  \widetilde{A}\right)  $ proportion of the level
$\log_{2}k_{3}-1$ scaled and translated mother wavelets $\left(  \phi
_{\frac{k_{3}}{2}+1},...\phi_{k_{3}}\right)  $ have their support inside
$\widetilde{A}.$ Hence, \ \
\begin{align*}
&  E\left[  K_{\left(  k_{0},k_{1}\right)  }\left(  X,X\right)  K_{\left(
k_{2},k_{3}\right)  }\left(  X,X\right)  \right] \\
&  >E\left[  1_{\widetilde{A}}\sum_{r=1}^{k^{\ast}}\phi_{r}^{2}\left(
X\right)  \sum_{s=\frac{k_{3}}{2}+1}^{k_{3}}\phi_{s}^{2}\left(  X\right)
\right] \\
&  \geq\frac{k^{\ast}}{2\left(  U_{w}-L_{w}\right)  }E\left[  1_{\widetilde{A}%
}\sum_{s=\frac{k_{3}}{2}+1}^{k_{3}}\phi_{s}^{2}\left(  X\right)  \right] \\
&  >\frac{k^{\ast}}{2\left(  U_{w}-L_{w}\right)  }\frac{1}{4M_{w}^{2}\left(
U_{w}-L_{w}\right)  }\frac{k_{3}}{2}\\
&  \asymp k_{1}k_{3}.
\end{align*}

case 3: $m>1.$ \ WLOG, we assume that $k_{1}\leq k_{3}...\leq k_{2m+1}$ and
$k_{1}\asymp k_{3}..\asymp k_{2l_{1}-1}\ll k_{2l_{1}+1}\asymp..\asymp
k_{2l_{2}-1}\ll...\ll k_{2l_{t-1}+1}\asymp...\asymp k_{2l_{t}-1}=k_{2m+1}$ for
$1\leq t\leq m+1,$ $0=l_{0}<l_{1}...<l_{t}=m+1,$ then
\begin{align*}
&  E\left[
%TCIMACRO{\dprod \limits_{j=1}^{m+1}}%
%BeginExpansion
{\displaystyle\prod\limits_{j=1}^{m+1}}
%EndExpansion
K_{\left(  k_{2j-2},k_{2j-1}\right)  }\left(  X,X\right)  \right] \\
&  \geq E\left[  \left(  K_{\left(  k_{0},k_{1}\right)  }\right)  ^{l_{1}}%
%TCIMACRO{\dprod \limits_{r=2}^{t}}%
%BeginExpansion
{\displaystyle\prod\limits_{r=2}^{t}}
%EndExpansion
\left(  K_{\left(  \frac{k_{2l_{r-1}+1}}{2}+1,k_{2l_{r-1}+1}\right)  }\right)
^{l_{r}-l_{r-1}}\right] \\
&  >E\left[  \left(  \sum_{r=k_{0}}^{k_{1}}\phi_{r}^{2}\left(  X\right)
\right)  ^{l_{1}}%
%TCIMACRO{\dprod \limits_{r=2}^{t}}%
%BeginExpansion
{\displaystyle\prod\limits_{r=2}^{t}}
%EndExpansion
\left(  \sum_{r=\frac{k_{2l_{r-1}+1}}{2}+1}^{k_{2l_{r-1}+1}}\phi_{r}%
^{2}\left(  X\right)  \right)  ^{l_{r}-l_{r-1}}\right]
\end{align*}

By a similar argument as that in case 2, we can show that there exists a set
$\widetilde{A}_{1},$ which consists of multiple disjoint open intervals, such
that $\sum_{r=k_{0}}^{k_{1}}\phi_{r}^{2}\left(  x\right)  >c_{1}^{\phi}k_{1},$
$\forall$ $x\in\widetilde{A}_{1}$ and $\mu\left(  \widetilde{A}_{1}\right)
>\delta_{1}^{\phi}$ for some positive constants $c_{1}^{\phi},\delta_{1}%
^{\phi}.$ moreover, by the multiresolution analysis (MRA) property of
compactly supported wavelets and the fact that $k_{1}=o\left(  k_{2l_{1}%
+1}\right)  $, it is obvious that at least $\frac{1}{2}\delta_{1}^{\phi}$
proportion of the level $log_{2}\left(  k_{2l_{1}+1}\right)  -1$ scaled and
translated mother wavelets $\left\{  \phi_{\frac{k_{2l_{1}+1}}{2}+1}%
,...\phi_{k_{2l_{1}+1}}\right\}  $ have support inside $\widetilde{A}_{1}. $
Hence we can find a set $\widetilde{A}_{2}\subset\widetilde{A}_{1}$, which is
also a union of disjoint open intervals, such that $\sum_{r=\frac{k_{2l_{1}%
+1}}{2}+1}^{k_{2l_{1}+1}}\phi_{r}^{2}\left(  x\right)  >c_{2}^{\phi}%
k_{2l_{1}+1},$ $\forall$ $x\in\widetilde{A}_{2}$ and $\mu\left(
\widetilde{A}_{2}\right)  >\delta_{2}^{\phi}\frac{\delta_{1}^{\phi}}{2}$ for
some positive constants $c_{2}^{\phi},\delta_{2}^{\phi}.$ \ Furthermore, by
applying this algorithm in a nested fashion $t-1$ times, we can find a
decreasing sequence of sets $\widetilde{A}_{1}\supset\widetilde{A}_{2}%
\supset...\supset\widetilde{A}_{t-2}\supset\widetilde{A}_{t-1},$ such that for
any $1\leq s\leq t-2,$ $\
%TCIMACRO{\dsum \limits_{r=\frac{k_{2l_{s}+1}}{2}+1}^{k_{2l_{s}+1}}}%
%BeginExpansion
{\displaystyle\sum\limits_{r=\frac{k_{2l_{s}+1}}{2}+1}^{k_{2l_{s}+1}}}
%EndExpansion
\phi_{r}^{2}\left(  x\right)  >c_{s+1}^{\phi}k_{2l_{s}+1}, $ \ $\forall$
$x\in\widetilde{A}_{s+1}, $ and \ $\mu\left(  \widetilde{A}_{s+1}\right)
>\delta_{s+1}^{\phi}%
%TCIMACRO{\dprod \limits_{r=1}^{s}}%
%BeginExpansion
{\displaystyle\prod\limits_{r=1}^{s}}
%EndExpansion
\frac{\delta_{r}^{\phi}}{2}$ for some positive constants $\left\{
c_{s+1}^{\phi},\text{ }1\leq s\leq t-2\right\}  $ and $\left\{  \delta
_{r}^{\phi},\text{ }1\leq r\leq t-1\right\}  .$ \ In addition, \ $%
%TCIMACRO{\dprod \limits_{r=1}^{t-1}}%
%BeginExpansion
{\displaystyle\prod\limits_{r=1}^{t-1}}
%EndExpansion
\frac{\delta_{r}^{\phi}}{2}$ proportion of the level $\log_{2}k_{2l_{t-1}%
+1}-1$ scaled and translated mother wavelets $\left\{  \phi_{\frac
{k_{2l_{t-1}+1}}{2}+1},...\phi_{k_{2l_{t-1}+1}}\right\}  $ have support inside
$\widetilde{A}_{t-1}.$

Hence, \
\begin{align*}
&  E\left[
%TCIMACRO{\dprod \limits_{j=1}^{m+1}}%
%BeginExpansion
{\displaystyle\prod\limits_{j=1}^{m+1}}
%EndExpansion
K_{\left(  k_{2j-2},k_{2j-1}\right)  }\left(  X,X\right)  \right] \\
&  >\left(
%TCIMACRO{\dprod \limits_{r=1}^{t-1}}%
%BeginExpansion
{\displaystyle\prod\limits_{r=1}^{t-1}}
%EndExpansion
k_{2l_{r-1}+1}^{l_{r}-l_{r-1}}\right)  E \left[  \left(  1_{A_{t-1}}%
\sum_{r=\frac{k_{2l_{t-1}+1}}{2}+1}^{k_{2l_{t-1}+1}}\phi_{r}^{2}\left(
X\right)  \right)  ^{l_{t}-l_{t-1}}\right] \\
&  \geq\left(
%TCIMACRO{\dprod \limits_{r=1}^{t-1}}%
%BeginExpansion
{\displaystyle\prod\limits_{r=1}^{t-1}}
%EndExpansion
k_{2l_{r-1}+1}^{l_{r}-l_{r-1}}\right)  E\left[  \left(  1_{A_{t-1}}%
\sum_{r=\frac{k_{2l_{t-1}+1}}{2}+1}^{k_{2l_{t-1}+1}}\phi_{r}^{2}\left(
X\right)  \right)  \right]  ^{l_{t}-l_{t-1}}\\
&  >\left(
%TCIMACRO{\dprod \limits_{r=1}^{t-1}}%
%BeginExpansion
{\displaystyle\prod\limits_{r=1}^{t-1}}
%EndExpansion
k_{2l_{r-1}+1}^{l_{r}-l_{r-1}}\right)  \left(  \frac{k_{2l_{t-1}+1}}{2}%
%TCIMACRO{\dprod \limits_{r=1}^{t-1}}%
%BeginExpansion
{\displaystyle\prod\limits_{r=1}^{t-1}}
%EndExpansion
\frac{\delta_{r}^{\phi}}{2}\right)  ^{l_{t}-l_{t-1}}\times%
%TCIMACRO{\dprod \limits_{j=1}^{m+1}}%
%BeginExpansion
{\displaystyle\prod\limits_{j=1}^{m+1}}
%EndExpansion
k_{2j-1}%
\end{align*}

\end{proof}

Now, we are ready to prove the variance results for the univariate case.

\begin{lemma}
\label{var_uni} Suppose the assumptions in Lemma \ref{var_kernel} hold, then
\begin{align*}
&  \left\Vert \overline{\phi}_{k_{0}}^{k_{1}}\left(  X_{i_{1}}\right)  ^{T}%
%TCIMACRO{\dprod \limits_{j=1}^{m}}%
%BeginExpansion
{\displaystyle\prod\limits_{j=1}^{m}}
%EndExpansion
\left\{  \overline{\phi}_{k_{2j-2}}^{k_{2j-1}}\left(  X_{i_{j+1}}\right)
\overline{\phi}_{k_{2j}}^{k_{2j+1}}\left(  X_{i_{j+1}}\right)  ^{T}\right\}
\overline{\phi}_{k_{2m}}^{k_{2m+1}}\left(  X_{i_{m+2}}\right)  \right\Vert
^{2}\\
&  =\left\Vert
%TCIMACRO{\dprod \limits_{j=1}^{m+1}}%
%BeginExpansion
{\displaystyle\prod\limits_{j=1}^{m+1}}
%EndExpansion
K_{\left(  k_{2j-2},k_{2j-1}\right)  }\left(  X_{i_{j}},X_{i_{j+1}}\right)
\right\Vert ^{2}\\
&  \asymp\Pi_{j=1}^{m+1}k_{2j-1}%
\end{align*}

\end{lemma}

\begin{proof}
(Lemma: \ref{var_uni}) case 1:$\ m=0.$%
\begin{align*}
&  E\left[  \overline{\phi}_{k_{0}}^{k_{1}}\left(  X_{i_{1}}\right)
^{T}\overline{\phi}_{k_{0}}^{k_{1}}\left(  X_{i_{2}}\right)  \right]  ^{2}\\
&  =E\left[  \left(  K_{\left(  k_{0},k_{1}\right)  }\left(  X_{i_{1}%
},X_{i_{2}}\right)  \right)  ^{2}\right] \\
&  =E\left[  \left(  K_{\left(  k_{0},k_{1}\right)  }\left(  X_{i_{1}%
},X_{i_{1}}\right)  \right)  \right] \\
&  =k_{1}-k_{0}%
\end{align*}
which follows from the orthonormality of wavelets. \ 

Next, we prove the lemma for the case where $m\geq1$ in two steps.

a). We first prove that
\[
\left\Vert
%TCIMACRO{\dprod \limits_{j=1}^{m+1}}%
%BeginExpansion
{\displaystyle\prod\limits_{j=1}^{m+1}}
%EndExpansion
K_{\left(  k_{2j-2},k_{2j-1}\right)  }\left(  X_{i_{j}},X_{i_{j+1}}\right)
\right\Vert ^{2}=O\left(  \Pi_{j=1}^{m+1}k_{2j-1}\right)
\]

Consider the case $m=1$
\begin{align*}
&  E\left[  \left(  K_{\left(  k_{0},k_{1}\right)  }\left(  X_{i_{1}}%
,X_{i_{2}}\right)  K_{\left(  k_{2},k_{3}\right)  }\left(  X_{i_{2}},X_{i_{3}%
}\right)  \right)  ^{2}\right] \\
&  =E\left[  \left(  K_{\left(  k_{0},k_{1}\right)  }\left(  X_{i_{2}%
},X_{i_{2}}\right)  K_{\left(  k_{2},k_{3}\right)  }\left(  X_{i_{2}}%
,X_{i_{2}}\right)  \right)  \right] \\
&  \leq||K_{\left(  k_{0},k_{1}\right)  }\left(  X,X\right)  ||_{\infty
}||K_{\left(  k_{2},k_{3}\right)  }\left(  X,X\right)  ||_{\infty}%
\end{align*}
if $\left(  k_{0},k_{1}\right)  =(1,k^{\ast}),$ then $||K_{\left(  k_{0}%
,k_{1}\right)  }\left(  X,X\right)  ||_{\infty}=O\left(  ||\phi_{1}^{2}\left(
X\right)  ||_{\infty}\right)  =O\left(  k_{1}-k_{0}\right)  $. Similarly, if
$k_{0}>k^{\ast}$ and $\log_{2}\left(  k_{0}-1\right)  =\log_{2}\left(
k_{1}\right)  -1,$ then $||K_{\left(  k_{0},k_{1}\right)  }\left(  X,X\right)
||_{\infty}=O\left(  ||\phi_{k_{0}}^{2}\left(  X\right)  ||_{\infty}\right)
=O\left(  k_{1}-k_{0}\right)  ;$ otherwise, if $k_{0}>k^{\ast}$ and $\log
_{2}\left(  k_{0}-1\right)  >\log_{2}\left(  k_{1}\right)  -1,$ then
\begin{align*}
||K_{\left(  k_{0},k_{1}\right)  }\left(  X,X\right)  ||_{\infty}  &
=O\left(  \sum_{t=0}^{int\left(  \log_{2}\left(  \frac{k_{1}}{k_{0}}\right)
\right)  }k_{0}2^{t}\right) \\
&  =O\left(  k_{1}\right)
\end{align*}
Finally, \ for $m>1$%
\begin{align*}
&  E\left[  \left(
%TCIMACRO{\dprod \limits_{j=1}^{m+1}}%
%BeginExpansion
{\displaystyle\prod\limits_{j=1}^{m+1}}
%EndExpansion
K_{\left(  k_{2j-2},k_{2j-1}\right)  }\left(  X_{i_{j}},X_{i_{j+1}}\right)
\right)  ^{2}\right] \\
&  =E\left[  K_{\left(  k_{0},k_{1}\right)  }\left(  X_{i_{1}},X_{i_{2}%
}\right)  ^{2}\left(
%TCIMACRO{\dprod \limits_{j=2}^{m}}%
%BeginExpansion
{\displaystyle\prod\limits_{j=2}^{m}}
%EndExpansion
K_{\left(  k_{2j-2},k_{2j-1}\right)  }\left(  X_{i_{j}},X_{i_{j+1}}\right)
\right)  ^{2}K_{\left(  k_{2m},k_{2m+1}\right)  }\left(  X_{i_{m+1}%
},X_{i_{m+2}}\right)  ^{2}\right] \\
&  =E\left[  K_{\left(  k_{0},k_{1}\right)  }\left(  X_{i_{2}},X_{i_{2}%
}\right)  \left(
%TCIMACRO{\dprod \limits_{j=2}^{m}}%
%BeginExpansion
{\displaystyle\prod\limits_{j=2}^{m}}
%EndExpansion
K_{\left(  k_{2j-2},k_{2j-1}\right)  }\left(  X_{i_{j}},X_{i_{j+1}}\right)
\right)  ^{2}K_{\left(  J_{2m},k_{2m+1}\right)  }\left(  X_{i_{m+2}%
},X_{i_{m+2}}\right)  \right] \\
&  \leq||K_{\left(  k_{0},k_{1}\right)  }\left(  X,X\right)  ||_{\infty
}E\left[  \left(
%TCIMACRO{\dprod \limits_{j=2}^{m}}%
%BeginExpansion
{\displaystyle\prod\limits_{j=2}^{m}}
%EndExpansion
K_{\left(  k_{2j-2},k_{2j-1}\right)  }\left(  X_{i_{j}},X_{i_{j+1}}\right)
\right)  ^{2}\right]  ||K_{\left(  k_{2m},k_{2m+1}\right)  }\left(
X,X\right)  ||_{\infty}\\
&  =||K_{\left(  k_{0},k_{1}\right)  }\left(  X,X\right)  ||_{\infty
}||K_{\left(  k_{2m},k_{2m+1}\right)  }\left(  X,X\right)  ||_{\infty}E\left[
\left(
%TCIMACRO{\dprod \limits_{j=2}^{m}}%
%BeginExpansion
{\displaystyle\prod\limits_{j=2}^{m}}
%EndExpansion
K_{\left(  k_{2j-2},k_{2j-1}\right)  }\left(  X_{i_{j}},X_{i_{j+1}}\right)
\right)  ^{2}\right] \\
&  \leq%
%TCIMACRO{\dprod \limits_{j=1}^{m+1}}%
%BeginExpansion
{\displaystyle\prod\limits_{j=1}^{m+1}}
%EndExpansion
||K_{\left(  k_{2j-2},k_{2j-1}\right)  }\left(  X,X\right)  ||_{\infty}\\
&  =O\left(  \Pi_{j=1}^{m+1}k_{2j-1}\right)
\end{align*}

b). In the second step, we prove that%
\begin{equation}
\frac{\Pi_{j=1}^{m+1}k_{2j-1}}{\left\Vert
%TCIMACRO{\dprod \limits_{j=1}^{m+1}}%
%BeginExpansion
{\displaystyle\prod\limits_{j=1}^{m+1}}
%EndExpansion
K_{\left(  k_{2j-2},k_{2j-1}\right)  }\left(  X_{i_{j}},X_{i_{j+1}}\right)
\right\Vert ^{2}}=O\left(  1\right) \label{var_lower}%
\end{equation}

We shall first show that $\left\Vert
%TCIMACRO{\dprod \limits_{j=1}^{m+1}}%
%BeginExpansion
{\displaystyle\prod\limits_{j=1}^{m+1}}
%EndExpansion
K_{\left(  1,k_{2j-1}\right)  }\left(  X_{i_{j}},X_{i_{j+1}}\right)
\right\Vert ^{2}\asymp%
%TCIMACRO{\dprod \limits_{j=1}^{m+1}}%
%BeginExpansion
{\displaystyle\prod\limits_{j=1}^{m+1}}
%EndExpansion
k_{2j-1},$ and then complete the proof by showing that
\begin{equation}
\left\Vert
%TCIMACRO{\dprod \limits_{j=1}^{m+1}}%
%BeginExpansion
{\displaystyle\prod\limits_{j=1}^{m+1}}
%EndExpansion
K_{\left(  k_{2j-2},k_{2j-1}\right)  }\left(  X_{i_{j}},X_{i_{j+1}}\right)
\right\Vert ^{2}\asymp\left\Vert
%TCIMACRO{\dprod \limits_{j=1}^{m+1}}%
%BeginExpansion
{\displaystyle\prod\limits_{j=1}^{m+1}}
%EndExpansion
K_{\left(  1,k_{2j-1}\right)  }\left(  X_{i_{j}},X_{i_{j+1}}\right)
\right\Vert ^{2}\label{eq1}%
\end{equation}
Specifically,
\begin{align*}
&  \left\Vert
%TCIMACRO{\dprod \limits_{j=1}^{m+1}}%
%BeginExpansion
{\displaystyle\prod\limits_{j=1}^{m+1}}
%EndExpansion
K_{\left(  1,k_{2j-1}\right)  }\left(  X_{i_{j}},X_{i_{j+1}}\right)
\right\Vert ^{2}\\
&  =E\left[  K_{\left(  1,k_{1}\right)  }\left(  X_{i_{2}},X_{i_{2}}\right)
\left(  K_{\left(  1,k_{3}\right)  }\left(  X_{i_{2}},X_{i_{3}}\right)
\right)  ^{2}%
%TCIMACRO{\dprod \limits_{j=3}^{m+1}}%
%BeginExpansion
{\displaystyle\prod\limits_{j=3}^{m+1}}
%EndExpansion
\left(  K_{\left(  1,k_{2j-1}\right)  }\left(  X_{i_{j}},X_{i_{j+1}}\right)
\right)  ^{2}\right] \\
&  =E\left[  K_{\left(  1,k_{1}\right)  }\left(  X_{i_{3}},X_{i_{3}}\right)
K_{\left(  1,k_{3}\right)  }\left(  X_{i_{3}},X_{i_{3}}\right)
%TCIMACRO{\dprod \limits_{j=3}^{m+1}}%
%BeginExpansion
{\displaystyle\prod\limits_{j=3}^{m+1}}
%EndExpansion
\left(  K_{\left(  1,k_{2j-1}\right)  }\left(  X_{i_{j}},X_{i_{j+1}}\right)
\right)  ^{2}\right] \\
&  -E\left[  k_{1}k_{3}\left\{  h_{X_{i_{3}}}\left(  X_{i_{3}}\right)
-\left(  \mathcal{K}_{\left(  1,k_{3}\right)  }\mathcal{\circ}h_{X_{i_{3}}%
}\right)  \left(  X_{i_{3}}\right)  \right\}
%TCIMACRO{\dprod \limits_{j=3}^{m+1}}%
%BeginExpansion
{\displaystyle\prod\limits_{j=3}^{m+1}}
%EndExpansion
\left(  K_{\left(  1,k_{2j-1}\right)  }\left(  X_{i_{j}},X_{i_{j+1}}\right)
\right)  ^{2}\right]
\end{align*}
where $h_{X_{i_{3}}}\left(  x\right)  =\frac{1}{k_{1}k_{3}}K_{\left(
1,k_{1}\right)  }\left(  x,x\right)  K_{\left(  1,k_{3}\right)  }\left(
x,X_{i_{3}}\right)  $ and $\left(  \mathcal{K}_{\left(  1,k_{3}\right)
}\mathcal{\circ}h\right)  \left(  \cdot\right)  =E\left[  h\left(  X\right)
K_{\left(  1,k_{3}\right)  }\left(  X,\cdot\right)  \right]  .$

As previously shown, \ there exists a positive constant $C_{\phi}$ such that
$\underset{x,X_{i_{3}}}{sup}\left\vert h_{X_{i_{3}}}\left(  x\right)
\right\vert \leq C_{\phi}.$ \ By the regularity property of compactly
supported wavelets and the approximation property of the kernel $K_{\left(
1,k_{3}\right)  }\left(  \cdot,\cdot\right)  ,$ it is obvious that $\left\vert
\left\vert h_{X_{i_{3}}}\left(  X_{i_{2}}\right)  -\left(  \mathcal{K}%
_{\left(  1,k_{3}\right)  }\mathcal{\circ}h_{X_{i_{3}}}\right)  \left(
X_{i_{3}}\right)  \right\vert \right\vert _{\infty}=o\left(  1\right)  .$
Thus,
\begin{align*}
&  E\left[
\begin{array}
[c]{c}%
k_{1}k_{3}\left\{  h_{X_{i_{3}}}\left(  X_{i_{2}}\right)  -\left(
\mathcal{K}_{\left(  1,k_{3}\right)  }\mathcal{\circ}h_{X_{i_{3}}}\right)
\left(  X_{i_{3}}\right)  \right\} \\
\times%
%TCIMACRO{\dprod \limits_{j=3}^{m+1}}%
%BeginExpansion
{\displaystyle\prod\limits_{j=3}^{m+1}}
%EndExpansion
\left(  K_{\left(  1,k_{2j-1}\right)  }\left(  X_{i_{j}},X_{i_{j+1}}\right)
\right)  ^{2}%
\end{array}
\right] \\
&  =o\left(  \Pi_{j=1}^{m+1}k_{2j-1}\right)
\end{align*}
\newline

In fact, by arguing similarly as above,
\begin{align*}
&  \left\Vert
%TCIMACRO{\dprod \limits_{j=1}^{m+1}}%
%BeginExpansion
{\displaystyle\prod\limits_{j=1}^{m+1}}
%EndExpansion
K_{\left(  1,k_{2j-1}\right)  }\left(  X_{i_{j}},X_{i_{j+1}}\right)
\right\Vert ^{2}\\
&  =E\left[  \left\{
\begin{array}
[c]{c}%
K_{\left(  1,k_{1}\right)  }\left(  X_{i_{3}},X_{i_{3}}\right)  K_{\left(
1,k_{3}\right)  }\left(  X_{i_{3}},X_{i_{3}}\right) \\
\times%
%TCIMACRO{\dprod \limits_{j=3}^{m+1}}%
%BeginExpansion
{\displaystyle\prod\limits_{j=3}^{m+1}}
%EndExpansion
\left(  K_{\left(  1,k_{2j-1}\right)  }\left(  X_{i_{j}},X_{i_{j+1}}\right)
\right)  ^{2}%
\end{array}
\right\}  \right]  +o\left(  \Pi_{j=1}^{m+1}k_{2j-1}\right) \\
&  =E\left[  \left\{
\begin{array}
[c]{c}%
K_{\left(  1,k_{1}\right)  }\left(  X_{i_{4}},X_{i_{4}}\right)  K_{\left(
1,k_{3}\right)  }\left(  X_{i_{4}},X_{i_{4}}\right)  K_{\left(  1,k_{5}%
\right)  }\left(  X_{i_{4}},X_{i_{4}}\right) \\
\times%
%TCIMACRO{\dprod \limits_{j=4}^{m+1}}%
%BeginExpansion
{\displaystyle\prod\limits_{j=4}^{m+1}}
%EndExpansion
\left(  K_{\left(  1,k_{2j-1}\right)  }\left(  X_{i_{j}},X_{i_{j+1}}\right)
\right)  ^{2}%
\end{array}
\right\}  \right] \\
&  +o\left(  \Pi_{j=1}^{m+1}k_{2j-1}\right) \\
&  =E\left[
%TCIMACRO{\dprod \limits_{j=1}^{m+1}}%
%BeginExpansion
{\displaystyle\prod\limits_{j=1}^{m+1}}
%EndExpansion
K_{\left(  1,k_{2j-1}\right)  }\left(  X,X\right)  \right]  +o\left(
\Pi_{j=1}^{m+1}k_{2j-1}\right) \\
&  \asymp%
%TCIMACRO{\dprod \limits_{j=1}^{m+1}}%
%BeginExpansion
{\displaystyle\prod\limits_{j=1}^{m+1}}
%EndExpansion
k_{2j-1}%
\end{align*}

Now, we prove eq. (\ref{eq1}).
\begin{align*}
&
%TCIMACRO{\dprod \limits_{j=1}^{m+1}}%
%BeginExpansion
{\displaystyle\prod\limits_{j=1}^{m+1}}
%EndExpansion
K_{\left(  k_{2j-2},k_{2j-1}\right)  }\left(  X_{i_{j}},X_{i_{j+1}}\right) \\
&  =K_{\left(  1,k_{1}\right)  }\left(  X_{i_{1}},X_{i_{2}}\right)
%TCIMACRO{\dprod \limits_{j=2}^{m+1}}%
%BeginExpansion
{\displaystyle\prod\limits_{j=2}^{m+1}}
%EndExpansion
K_{\left(  k_{2j-2},k_{2j-1}\right)  }\left(  X_{i_{j}},X_{i_{j+1}}\right) \\
&  -K_{\left(  1,k_{0}\right)  }\left(  X_{i_{1}},X_{i_{2}}\right)
%TCIMACRO{\dprod \limits_{j=2}^{m+1}}%
%BeginExpansion
{\displaystyle\prod\limits_{j=2}^{m+1}}
%EndExpansion
K_{\left(  k_{2j-2},k_{2j-1}\right)  }\left(  X_{i_{j}},X_{i_{j+1}}\right)
\end{align*}%
\[
=
\]%
\begin{align*}
&  K_{\left(  1,k_{1}\right)  }\left(  X_{i_{1}},X_{i_{2}}\right)  \left\{
K_{\left(  1,k_{3}\right)  }\left(  X_{i_{2}},X_{i_{3}}\right)  -K_{\left(
1,k_{2}\right)  }\left(  X_{i_{2}},X_{i_{3}}\right)  \right\} \\
&  \times%
%TCIMACRO{\dprod \limits_{j=2}^{m+1}}%
%BeginExpansion
{\displaystyle\prod\limits_{j=2}^{m+1}}
%EndExpansion
K_{\left(  k_{2j-2},k_{2j-1}\right)  }\left(  X_{i_{j}},X_{i_{j+1}}\right) \\
&  -K_{\left(  1,k_{0}\right)  }\left(  X_{i_{1}},X_{i_{2}}\right)
%TCIMACRO{\dprod \limits_{j=2}^{m+1}}%
%BeginExpansion
{\displaystyle\prod\limits_{j=2}^{m+1}}
%EndExpansion
K_{\left(  k_{2j-2},k_{2j-1}\right)  }\left(  X_{i_{j}},X_{i_{j+1}}\right) \\
&  =%
%TCIMACRO{\dprod \limits_{j=1}^{m+1}}%
%BeginExpansion
{\displaystyle\prod\limits_{j=1}^{m+1}}
%EndExpansion
K_{\left(  1,k_{2j-1}\right)  }\left(  X_{i_{j}},X_{i_{j+1}}\right)
-\sum_{\left\{  \left(  k_{2j-2}^{\ast},k_{2j-1}^{\ast}\right)  \right\}  }%
%TCIMACRO{\dprod \limits_{j=1}^{m+1}}%
%BeginExpansion
{\displaystyle\prod\limits_{j=1}^{m+1}}
%EndExpansion
K_{\left(  k_{2j-2}^{\ast},k_{2j-1}^{\ast}\right)  }\left(  X_{i_{j}%
},X_{i_{j+1}}\right)
\end{align*}
where $\left(  k_{2j-2}^{\ast},k_{2j-1}^{\ast}\right)  $ may be $\left(
1,k_{2j-2}\right)  $ or $\left(  k_{2j-2},k_{2j-1}\right)  ,$ but $%
%TCIMACRO{\dprod \limits_{j=1}^{m+1}}%
%BeginExpansion
{\displaystyle\prod\limits_{j=1}^{m+1}}
%EndExpansion
k_{2j-1}^{\ast}=o\left(
%TCIMACRO{\dprod \limits_{j=1}^{m+1}}%
%BeginExpansion
{\displaystyle\prod\limits_{j=1}^{m+1}}
%EndExpansion
k_{2j-1}\right)  .$ Therefore $\left\vert \left\vert
%TCIMACRO{\dprod \limits_{j=1}^{m+1}}%
%BeginExpansion
{\displaystyle\prod\limits_{j=1}^{m+1}}
%EndExpansion
K_{\left(  k_{2j-2},k_{2j-1}\right)  }\left(  X_{i_{j}},X_{i_{j+1}}\right)
\right\vert \right\vert ^{2}\asymp\left\vert \left\vert
%TCIMACRO{\dprod \limits_{j=1}^{m+1}}%
%BeginExpansion
{\displaystyle\prod\limits_{j=1}^{m+1}}
%EndExpansion
K_{\left(  1,k_{2j-1}\right)  }\left(  X_{i_{j}},X_{i_{j+1}}\right)
\right\vert \right\vert ^{2}$ \ as a direct consequence of the previously
proved result that $\left\vert \left\vert K_{\left(  k_{2j-2}^{\ast}%
,k_{2j-1}^{\ast}\right)  }\left(  X_{i_{j}},X_{i_{j+1}}\right)  \right\vert
\right\vert ^{2}=O\left(
%TCIMACRO{\dprod \limits_{j=1}^{m+1}}%
%BeginExpansion
{\displaystyle\prod\limits_{j=1}^{m+1}}
%EndExpansion
k_{2j-1}^{\ast}\right)  =o\left(
%TCIMACRO{\dprod \limits_{j=1}^{m+1}}%
%BeginExpansion
{\displaystyle\prod\limits_{j=1}^{m+1}}
%EndExpansion
k_{2j-1}\right)  .$ The proof is complete.
\end{proof}

\begin{proof}
(Theorem \ref{var_multi}) We note that since the linear span of $\overline
{\varphi}_{1}^{k_{j}}\left(  X\right)  $ equals $\bigotimes\limits_{1\leq
r\leq l}\mathcal{V}_{r,\log_{2}\left(  k_{j,r}\right)  },$we have
\begin{align*}
&  K_{\left(  1,k_{j}\right)  }\left(  X_{i_{j}},X_{i_{j+1}}\right) \\
&  =\sum_{\left\{
\begin{array}
[c]{c}%
t_{1},...,t_{l}:\\
1\leq t_{u}\leq k_{j,u}\\
u=1,...,l
\end{array}
\right\}  }\prod\limits_{u=1}^{l}\overline{\varphi}_{t_{u}}\left(  X_{i_{j}%
}^{u}\right)  \prod\limits_{u^{\prime}=1}^{l}\overline{\varphi}_{t_{u^{\prime
}}}\left(  X_{i_{j+1}}^{u^{\prime}}\right)
\end{align*}
So that
\begin{align*}
&  K_{\left(  1,k_{j}\right)  }\left(  X_{i_{j}},X_{i_{j+1}}\right) \\
&  =\prod\limits_{u=1}^{l}\sum_{1\leq t_{u}\leq k_{j,u}}\overline{\varphi
}_{t_{u}}\left(  X_{i_{j}}^{u}\right)  \overline{\varphi}_{t_{u}}\left(
X_{i_{j+1}}^{u}\right) \\
&  =\prod\limits_{u=1}^{l}K_{\left(  1,k_{j,u}\right)  }\left(  X_{i_{j}}%
^{u},X_{i_{j+1}}^{u}\right)
\end{align*}
The remainder of the proof then follows from Lemma \ref{var_uni} .
\end{proof}

\begin{proof}
(Theorem \ref{3.19}) \ In the proof of Theorem \ref{3.19}, the following lemma
plays a central role. \ Note that expectations and probabilities remain
conditional on $\widehat{\theta}$ even when it is suppressed in the notation.
\end{proof}

\begin{lemma}
\label{3.19_lemma}Let $\widehat{IF}_{m,m,\widetilde{\psi}_{k}\left(
\cdot\right)  }^{\left(  s\right)  }$ be the symmetric kernel of
${}\ \widehat{\mathbb{IF}}_{m,m,\widetilde{\psi}_{k}\left(  \cdot\right)  },$
then for any $m\geq2$ and $1\leq m_{1}<m_{2},$
\begin{align}
&  Var_{\theta}\left[  \left(  \widehat{\mathbb{IF}}_{m,m,\widetilde{\psi}%
_{k}\left(  \cdot\right)  }\right)  ^{2}\right]  -Var_{\widehat{\theta}%
}\left[  \left(  \widehat{\mathbb{IF}}_{m,m,\widetilde{\psi}_{k}\left(
\cdot\right)  }\right)  ^{2}\right] \nonumber\\
&  =O\left(  n^{-m}\left\{  E_{\theta}\left[  \left(  \widehat{IF}%
_{m,m,\overline{i}_{m}}^{\left(  s\right)  }\right)  ^{2}\right]
-E_{\widehat{\theta}}\left[  \left(  \widehat{IF}_{m,m,\overline{i}_{m}%
}^{\left(  s\right)  }\right)  ^{2}\right]  \right\}  \right) \nonumber\\
&  +O\left(  \frac{1}{n}\left\{  E_{\theta}\left[  \widehat{IF}_{m,m,\overline
{i}_{m}}^{\left(  s\right)  }\right]  \right\}  ^{2}\right) \nonumber\\
&  +o\left(  \max\left(  \frac{1}{n},\frac{k^{m-2}}{n^{m-1}}\right)  \right)
\label{lemma_1}%
\end{align}

and
\begin{align}
&  Cov_{\theta}\left(  \widehat{\mathbb{IF}}_{m_{1},m_{1},\widetilde{\psi}%
_{k}\left(  \cdot\right)  },\text{ }\widehat{\mathbb{IF}}_{m_{2}%
,m_{2},\widetilde{\psi}_{k}\left(  \cdot\right)  }\right) \nonumber\\
&  =O\left(  \frac{1}{n}E_{\theta}\left[  \widehat{IF}_{m_{1},m_{1}%
,\overline{i}_{m}}^{\left(  s\right)  }\right]  E_{\theta}\left[
\widehat{IF}_{m_{2},m_{2},\overline{i}_{m}}^{\left(  s\right)  }\right]
\right)  +o\left(  \max\left(  \frac{1}{n},\frac{k^{m_{1}-2}}{n^{m_{1}-1}%
}\right)  \right) \label{lemma_2}%
\end{align}

\end{lemma}

\bigskip The proof of this lemma is delayed as some technical details are
involved. We first use this lemma to prove Theorem \ref{3.19}.

\begin{proof}
(Theorem \ref{3.19}). By the degeneracy of $\widehat{\mathbb{IF}%
}_{t,t,\widetilde{\psi}_{k}\left(  \cdot\right)  }$ for any $t\leq m$ under
$F\left(  \cdot;\widehat{\theta}\right)  ,$
\begin{align*}
&  \frac{Var_{\theta}\left[  \widehat{\mathbb{IF}}_{m,\widetilde{\psi}%
_{k}\left(  \cdot\right)  }|\widehat{\theta}\right]  -Var_{\widehat{\theta}%
}\left[  \widehat{\mathbb{IF}}_{m,\widetilde{\psi}_{k}\left(  \cdot\right)
}|\widehat{\theta}\right]  }{Var_{\widehat{\theta}}\left[
\widehat{\mathbb{IF}}_{m,\widetilde{\psi}_{k}\left(  \cdot\right)
}|\widehat{\theta}\right]  }\\
&  =\frac{\left\{
\begin{array}
[c]{c}%
%TCIMACRO{\dsum \limits_{t=1}^{m}}%
%BeginExpansion
{\displaystyle\sum\limits_{t=1}^{m}}
%EndExpansion
\left(  Var_{\theta}\left[  \widehat{\mathbb{IF}}_{t,t,\widetilde{\psi}%
_{k}\left(  \cdot\right)  }\right]  -Var_{\widehat{\theta}}\left[
\widehat{\mathbb{IF}}_{t,t,\widetilde{\psi}_{k}\left(  \cdot\right)  }\right]
\right) \\
+2%
%TCIMACRO{\dsum \limits_{1\leq t_{1}<t_{2}\leq m}}%
%BeginExpansion
{\displaystyle\sum\limits_{1\leq t_{1}<t_{2}\leq m}}
%EndExpansion
Cov_{\theta}\left[  \widehat{\mathbb{IF}}_{t_{1},t_{1},\widetilde{\psi}%
_{k}\left(  \cdot\right)  },\widehat{\mathbb{IF}}_{t_{2},t_{2},\widetilde{\psi
}_{k}\left(  \cdot\right)  }\right]
\end{array}
\right\}  }{Var_{\widehat{\theta}}\left[  \widehat{\mathbb{IF}}%
_{m,\widetilde{\psi}_{k}\left(  \cdot\right)  }|\widehat{\theta}\right]  },
\end{align*}

which equals
\begin{align*}
&  \frac{%
%TCIMACRO{\dsum \limits_{t=1}^{m}}%
%BeginExpansion
{\displaystyle\sum\limits_{t=1}^{m}}
%EndExpansion
n^{-t}\left(  E_{\theta}\left[  \left(  \widehat{IF}_{t,t,\overline{i}_{m}%
}^{\left(  s\right)  }\right)  ^{2}\right]  -E_{\widehat{\theta}}\left[
\left(  \widehat{IF}_{t,t,\overline{i}_{m}}^{\left(  s\right)  }\right)
^{2}\right]  \right)  }{%
%TCIMACRO{\dsum \limits_{t=1}^{m}}%
%BeginExpansion
{\displaystyle\sum\limits_{t=1}^{m}}
%EndExpansion
n^{-t}E_{\widehat{\theta}}\left[  \left(  \widehat{IF}_{t,t,\overline{i}_{m}%
}^{\left(  s\right)  }\right)  ^{2}\right]  }\left(  1+o\left(  1\right)
\right) \\
&  +\frac{o\left(  \max\left(  \frac{1}{n},\frac{k^{m-2}}{n^{m-1}}\right)
\right)  }{Var_{\widehat{\theta}}\left[  \widehat{\mathbb{IF}}%
_{m,\widetilde{\psi}_{k}\left(  \cdot\right)  }|\widehat{\theta}\right]  }%
\end{align*}

by Lemma \ref{3.19_lemma}.

By assumption, $\underset{o\in\mathcal{O}}{sup}\left\vert f\left(
o;\widehat{\theta}\right)  -f\left(  o;\theta\right)  \right\vert
\rightarrow0$ as $||\widehat{\theta}-\theta||\rightarrow0$, hence
\begin{align*}
&  \frac{E_{\theta}\left[  \left(  \widehat{IF}_{t,t,\overline{i}_{m}%
}^{\left(  s\right)  }\right)  ^{2}\right]  -E_{\widehat{\theta}}\left[
\left(  \widehat{IF}_{t,t,\overline{i}_{m}}^{\left(  s\right)  }\right)
^{2}\right]  }{E_{\widehat{\theta}}\left[  \left(  \widehat{IF}_{t,t,\overline
{i}_{m}}^{\left(  s\right)  }\right)  ^{2}\right]  }\\
&  =\frac{E_{\widehat{\theta}}\left[  \left(  \widehat{IF}_{t,t,\overline
{i}_{m}}^{\left(  s\right)  }\right)  ^{2}\left(
%TCIMACRO{\dprod \limits_{s=1}^{t}}%
%BeginExpansion
{\displaystyle\prod\limits_{s=1}^{t}}
%EndExpansion
\frac{f\left(  O_{i_{s}};\theta\right)  }{f\left(  O_{i_{s}};\widehat{\theta
}\right)  }-1\right)  \right]  }{E_{\widehat{\theta}}\left[  \left(
\widehat{IF}_{t,t,\overline{i}_{m}}^{\left(  s\right)  }\right)  ^{2}\right]
}\\
&  \leq\sup_{\mathbf{O}_{\overline{i}_{t}}}\left\vert
%TCIMACRO{\dprod \limits_{s=1}^{t}}%
%BeginExpansion
{\displaystyle\prod\limits_{s=1}^{t}}
%EndExpansion
\frac{f\left(  O_{i_{s}};\theta\right)  }{f\left(  O_{i_{s}};\widehat{\theta
}\right)  }-1\right\vert \\
&  =O\left(  \underset{o\in\mathcal{O}}{sup}\left\vert f\left(
o;\widehat{\theta}\right)  -f\left(  o;\theta\right)  \right\vert \right)
=o\left(  1\right)  .
\end{align*}

Furthermore, $\frac{o\left(  \max\left(  \frac{1}{n},\frac{k^{m-2}}{n^{m-1}%
}\right)  \right)  }{Var_{\widehat{\theta}}\left[  \widehat{\mathbb{IF}%
}_{m,\widetilde{\psi}_{k}\left(  \cdot\right)  }|\widehat{\theta}\right]
}=o\left(  1\right)  $ as well since $Var_{\widehat{\theta}}\left[
\widehat{\mathbb{IF}}_{m,\widetilde{\psi}_{k}\left(  \cdot\right)
}|\widehat{\theta}\right]  \times\max\left(  \frac{1}{n},\frac{k^{m-1}}{n^{m}%
}\right)  $ from Theorem \ref{var_multi}. Thus the proof of Thoerem \ref{3.19}
is complete.
\end{proof}

Before giving out the proof of Lemma \ref{3.19_lemma}, we first introduce two
useful propositions.

\begin{proposition}
For any $m\geq2$ and $1\leq t\leq m,$ \
\begin{equation}
E_{\theta}\left[  \left\{  E_{\theta}\left(  \widehat{IF}_{m,m,\overline
{i}_{m}}|\mathbf{O}_{-i_{t}}\right)  \right\}  ^{2}\right]  =o\left(
k^{m-2}\right) \label{3.19_1}%
\end{equation}

\end{proposition}

\begin{proof}
As proved in Theorem \ref{ON}, for any $m\geq2,$
\begin{align*}
&  \widehat{IF}_{m,m,\overline{i}_{m}}\\
&  =\left(  -1\right)  ^{j-1}\left\{
\begin{array}
[c]{c}%
\left[  \left(  H_{1}\widehat{P}+H_{2}\right)  \dot{B}\overline{Z}_{k}%
^{T}\right]  _{i_{1}}\left[
%TCIMACRO{\dprod \limits_{s=3}^{m}}%
%BeginExpansion
{\displaystyle\prod\limits_{s=3}^{m}}
%EndExpansion
\left\{  \left(  \dot{P}\dot{B}H_{1}\overline{Z}_{k}\overline{Z}_{k}%
^{T}\right)  _{i_{s}}-I_{k\times k}\right\}  \right] \\
\times\left[  \overline{Z}_{k}\left(  H_{1}\widehat{B}+H_{3}\right)  \dot
{P}\right]  _{i_{2}}%
\end{array}
\right\}
\end{align*}

We first consider the case when $m=2.$
\begin{align*}
&  -E_{\theta}\left[  \widehat{IF}_{2,2,i_{1}i_{2}}|O_{i_{2}}\right] \\
&  =E_{\theta}\left[  Q^{2}\left(  \frac{P-\widehat{P}}{\dot{P}}\right)
\overline{Z}_{k}^{T}\right]  \left[  \overline{Z}_{k}\left(  H_{1}%
\widehat{B}+H_{3}\right)  \dot{P}\right]  _{i_{2}}\\
&  =\widehat{E}\left[  \frac{f\left(  X\right)  }{\widehat{f}\left(  X\right)
}\frac{Q^{2}}{\widehat{Q}}\left(  \frac{P-\widehat{P}}{\dot{P}}\right)
\widehat{Q}\overline{Z}_{k}^{T}\right]  \left[  \frac{\widehat{Q}}%
{\widehat{Q}}\overline{Z}_{k}\left(  H_{1}\widehat{B}+H_{3}\right)  \dot
{P}\right]  _{i_{2}}\\
&  =\widehat{\Pi}\left[  \left(  \frac{f\left(  X\right)  }{\widehat{f}\left(
X\right)  }\frac{Q^{2}}{\widehat{Q}}\left(  \frac{P-\widehat{P}}{\dot{P}%
}\right)  \right)  |\left(  \widehat{Q}\overline{Z}_{k}\right)  _{i_{2}%
}\right]  \left(  \frac{\left(  H_{1}\widehat{B}+H_{3}\right)  \dot{P}%
}{\widehat{Q}}\right)  _{i_{2}}\\
&  =\left\vert \left\vert P-\widehat{P}\right\vert \right\vert _{2}\frac
{T_{c}\left(  O_{i_{2}}\right)  }{\left\vert \left\vert P-\widehat{P}%
\right\vert \right\vert _{2}}\left(  \frac{\left(  H_{1}\widehat{B}%
+H_{3}\right)  \dot{P}}{\widehat{Q}}\right)  _{i_{2}}%
\end{align*}

where
\[
T_{c}\left(  O\right)  \equiv\widehat{\Pi}\left[  \left(  \frac{f\left(
X\right)  }{\widehat{f}\left(  X\right)  }\frac{Q^{2}}{\widehat{Q}}%
\frac{P-\widehat{P}}{\dot{P}}\right)  |\left(  \widehat{Q}\overline{Z}%
_{k}\right)  \right]
\]

Since assumptions $\left(  \text{\ref{dot1}}\right)  -\left(  \text{\ref{dot3}%
}\right)  $ and $Ai)-Aiv)$ are satisfied, it is easy to show that $E\left[
\left(  \frac{T_{c}\left(  O_{i_{2}}\right)  }{\left\vert \left\vert
P-\widehat{P}\right\vert \right\vert _{2}}\left(  \frac{\left(  H_{1}%
\widehat{B}+H_{3}\right)  \dot{P}}{\widehat{Q}}\right)  _{i_{2}}\right)
^{2}\right]  =O\left(  1\right)  $, and thus
\[
E\left[  \left(  E_{\theta}\left[  \widehat{IF}_{2,2,i_{1}i_{2}}|O_{i_{2}%
}\right]  \right)  ^{2}\right]  =o\left(  1\right)  .
\]

Similarly, we can prove that $E\left[  \left(  E_{\theta}\left[
\widehat{IF}_{2,2,i_{1}i_{2}}|O_{i_{1}}\right]  \right)  ^{2}\right]
=o\left(  1\right)  .$

Next, we proceed by induction. We assume eq. (\ref{3.19_1}) holds for $m-1$
and prove it is also true for $m$ by considering different values of $t.$

i) If $t=1,$ \ then
\begin{align*}
&  \left(  -1\right)  ^{m-1}E_{\theta}\left(  \widehat{IF}_{m,m,\overline
{i}_{m}}|\mathbf{O}_{-i_{1}}\right) \\
&  =E_{\theta}\left[  Q^{2}\left(  \frac{P-\widehat{P}}{\dot{P}}\right)
\overline{Z}_{k}^{T}\right]  \left[
%TCIMACRO{\dprod \limits_{s=3}^{m}}%
%BeginExpansion
{\displaystyle\prod\limits_{s=3}^{m}}
%EndExpansion
\left\{
\begin{array}
[c]{c}%
\left(  \dot{P}\dot{B}H_{1}\overline{Z}_{k}\overline{Z}_{k}^{T}\right)
_{i_{s}}\\
-I_{k\times k}%
\end{array}
\right\}  \right] \\
&  \times\left[  \overline{Z}_{k}\left(  H_{1}\widehat{B}+H_{3}\right)
\dot{P}\right]  _{i_{2}}\\
&  =\left\{
\begin{array}
[c]{c}%
\widehat{E}\left[  \frac{f\left(  X\right)  }{\widehat{f}\left(  X\right)
}Q^{2}\left(  \frac{P-\widehat{P}}{\dot{P}}\right)  \overline{Z}_{k}%
^{T}\right]  \left(  \dot{P}\dot{B}H_{1}\overline{Z}_{k}\overline{Z}_{k}%
^{T}\right)  _{i_{3}}\\
-E_{\theta}\left[  Q^{2}\left(  \frac{P-\widehat{P}}{\dot{P}}\right)
\overline{Z}_{k}^{T}\right]
\end{array}
\right\}  \times\\
&
%TCIMACRO{\dprod \limits_{s=4}^{m}}%
%BeginExpansion
{\displaystyle\prod\limits_{s=4}^{m}}
%EndExpansion
\left\{
\begin{array}
[c]{c}%
\left(  \dot{P}\dot{B}H_{1}\overline{Z}_{k}\overline{Z}_{k}^{T}\right)
_{i_{s}}\\
-I_{k\times k}%
\end{array}
\right\}  \left[  \overline{Z}_{k}\left(  H_{1}\widehat{B}+H_{3}\right)
\dot{P}\right]  _{i_{2}}%
\end{align*}%
\[
=
\]%
\begin{align*}
&  \left\{  \frac{\dot{P}\dot{B}H_{1}}{\widehat{Q}}\widehat{\Pi}\left[
\left(  \frac{f\left(  X\right)  }{\widehat{f}\left(  X\right)  }\frac{Q^{2}%
}{\widehat{Q}}\left(  \frac{P-\widehat{P}}{\dot{P}}\right)  \right)
|\widehat{Q}\overline{Z}_{k}\right]  \overline{Z}_{k}^{T}\right\}  _{i_{3}%
}\times\\
&
%TCIMACRO{\dprod \limits_{s=4}^{m}}%
%BeginExpansion
{\displaystyle\prod\limits_{s=4}^{m}}
%EndExpansion
\left\{
\begin{array}
[c]{c}%
\left(  \dot{P}\dot{B}H_{1}\overline{Z}_{k}\overline{Z}_{k}^{T}\right)
_{i_{s}}\\
-I_{k\times k}%
\end{array}
\right\}  \left[  \overline{Z}_{k}\left(  H_{1}\widehat{B}+H_{3}\right)
\dot{P}\right]  _{i_{2}}\\
&  -E_{\theta}\left(  \widehat{IF}_{m-1,m-1,i_{1}i_{2}i_{4}...i_{m}%
}|\mathbf{O}_{-i_{1}}\right) \\
&  =\left(  \frac{\dot{P}\dot{B}H_{1}}{\widehat{Q}}\right)  _{i_{3}}%
T_{c}\left(  O_{i_{3}}\right)  \overline{Z}_{k,i_{3}}^{T}%
%TCIMACRO{\dprod \limits_{s=4}^{m}}%
%BeginExpansion
{\displaystyle\prod\limits_{s=4}^{m}}
%EndExpansion
\left\{
\begin{array}
[c]{c}%
\left(  \dot{P}\dot{B}H_{1}\overline{Z}_{k}\overline{Z}_{k}^{T}\right)
_{i_{s}}\\
-I_{k\times k}%
\end{array}
\right\} \\
&  \times\left[  \overline{Z}_{k}\left(  H_{1}\widehat{B}+H_{3}\right)
\dot{P}\right]  _{i_{2}}-E_{\theta}\left(  \widehat{IF}_{m-1,m-1,i_{1}%
i_{2}i_{4}...i_{m}}|\mathbf{O}_{-i_{1}}\right)
\end{align*}

From the fact that $\left(  a-b\right)  ^{2}\leq2\left(  a^{2}+b^{2}\right)
,$ we have
\begin{align*}
&  E_{\theta}\left[  \left(  E_{\theta}\left[  \widehat{IF}_{m,m,\overline
{i}_{m}}\right]  |\mathbf{O}_{-i_{1}}\right)  ^{2}\right] \\
&  \leq2E\left[  \left(
\begin{array}
[c]{c}%
\left(  \frac{\dot{P}\dot{B}H_{1}}{\widehat{Q}}\right)  _{i_{3}}T_{c}\left(
O_{i_{3}}\right)  \overline{Z}_{k,i_{3}}^{T}\times\\%
%TCIMACRO{\dprod \limits_{s=4}^{m}}%
%BeginExpansion
{\displaystyle\prod\limits_{s=4}^{m}}
%EndExpansion
\left\{
\begin{array}
[c]{c}%
\left(  \dot{P}\dot{B}H_{1}\overline{Z}_{k}\overline{Z}_{k}^{T}\right)
_{i_{s}}\\
-I_{k\times k}%
\end{array}
\right\}  \left[  \overline{Z}_{k}\left(  H_{1}\widehat{B}+H_{3}\right)
\dot{P}\right]  _{i_{2}}%
\end{array}
\right)  ^{2}\right] \\
&  +2E\left[  \left(  E_{\theta}\left[  \widehat{IF}_{m-1,m-1,i_{1}i_{2}%
i_{4}...i_{m}}|\mathbf{O}_{-i_{1}}\right]  \right)  ^{2}\right]
\end{align*}
\newline

From Theorem \ref{var_multi}, it can be shown that
\begin{align*}
&  E\left[  \left(
\begin{array}
[c]{c}%
\left(  \frac{\dot{P}\dot{B}H_{1}}{\widehat{Q}}\right)  _{i_{3}}T_{c}\left(
O_{i_{3}}\right)  \overline{Z}_{k,i_{3}}^{T}\times\\%
%TCIMACRO{\dprod \limits_{s=4}^{m}}%
%BeginExpansion
{\displaystyle\prod\limits_{s=4}^{m}}
%EndExpansion
\left\{
\begin{array}
[c]{c}%
\left(  \dot{P}\dot{B}H_{1}\overline{Z}_{k}\overline{Z}_{k}^{T}\right)
_{i_{s}}\\
-I_{k\times k}%
\end{array}
\right\}  \left[  \overline{Z}_{k}\left(  H_{1}\widehat{B}+H_{3}\right)
\dot{P}\right]  _{i_{2}}%
\end{array}
\right)  ^{2}\right] \\
&  =\left\vert \left\vert P-\widehat{P}\right\vert \right\vert _{\infty
}O\left(  k^{m-2}\right) \\
&  =o\left(  k^{m-2}\right)
\end{align*}

By the induction assumption, $E_{\theta}\left[  \left\{  E_{\theta}\left(
\widehat{IF}_{m-1,m-1,i_{1}i_{2}i_{4}...i_{m}}|O_{i_{1}}\right)  \right\}
^{2}\right]  =o\left(  k^{m-3}\right)  .$ Therefore eq. (\ref{3.19_1}) holds
when $t=1.$

ii) Following the same argument as above, we can prove that eq. (\ref{3.19_1})
also holds for $t=2.$

iii) If $3\leq t\leq m,$ \ WLOG, assume $t=3,$ then%
\begin{align*}
&  \left(  -1\right)  ^{m-1}E\left[  \widehat{IF}_{m,m,\overline{i}_{m}%
}|\mathbf{O}_{-i_{3}}\right] \\
&  =\left[  \left(  H_{1}\widehat{P}+H_{2}\right)  \dot{B}\overline{Z}_{k}%
^{T}\right]  _{i_{1}}\left(  E_{\theta}\left[  Q^{2}\overline{Z}_{k}%
\overline{Z}_{k}^{T}\right]  -I\right)  \times\\
&
%TCIMACRO{\dprod \limits_{s=4}^{m}}%
%BeginExpansion
{\displaystyle\prod\limits_{s=4}^{m}}
%EndExpansion
\left\{
\begin{array}
[c]{c}%
\left(  \dot{P}\dot{B}H_{1}\overline{Z}_{k}\overline{Z}_{k}^{T}\right)
_{i_{s}}\\
-I_{k\times k}%
\end{array}
\right\}  \left[  \overline{Z}_{k}\left(  H_{1}\widehat{B}+H_{3}\right)
\dot{P}\right]  _{i_{2}}\\
&  =\left[  \left(  H_{1}\widehat{P}+H_{2}\right)  \dot{B}\overline{Z}_{k}%
^{T}\right]  _{i_{1}}\widehat{E}\left[  \delta g\text{ }\widehat{Q}%
^{2}\overline{Z}_{k}\overline{Z}_{k}^{T}\right]  \left(  \dot{P}\dot{B}%
H_{1}\overline{Z}_{k}\overline{Z}_{k}^{T}\right)  _{i_{4}}\\
&  \times%
%TCIMACRO{\dprod \limits_{s=5}^{m}}%
%BeginExpansion
{\displaystyle\prod\limits_{s=5}^{m}}
%EndExpansion
\left\{  \left(  \dot{P}\dot{B}H_{1}\overline{Z}_{k}\overline{Z}_{k}%
^{T}\right)  _{i_{s}}-I_{k\times k}\right\}  \left[  \overline{Z}_{k}\left(
H_{1}\widehat{B}+H_{3}\right)  \dot{P}\right]  _{i_{2}}\left(  \equiv
\widehat{T}\right) \\
&  -E_{\theta}\left[  \widehat{IF}_{m-1,m-1,i_{1}i_{2}i_{3}i_{5}..i_{m}%
}|\mathbf{O}_{-i_{3}}\right]
\end{align*}

moreover, it can be shown that $E_{\theta}\left[  \widehat{T}^{2}\right]
=O\left(  \left\vert \left\vert \delta g\right\vert \right\vert _{\infty}%
^{2}k^{m-2}\right)  =o\left(  k^{m-2}\right)  $ following the proof of Theorem
\ref{var_multi} but replacing $K_{k}\left(  X_{i_{1}},X_{i_{2}}\right)  $ with
$K_{k}^{\ddag}\left(  X_{i_{1}},X_{i_{2}}\right)  \equiv\overline{\phi}%
_{0}^{k}\left(  X_{i_{1}}\right)  ^{T}\widehat{E}\left[  \delta g\text{
}\widehat{Q}^{2}\overline{Z}_{k}\overline{Z}_{k}^{T}\right]  \overline{\phi
}_{0}^{k}\left(  X_{i_{2}}\right)  .$ Specifically,
\begin{align*}
&  \int\left(  K_{k}^{\ddag}\left(  X_{i_{1}},X_{i_{2}}\right)  \right)
^{2}d\mu\left(  O_{i_{1}};\theta\right) \\
&  =tr\left(  \widehat{E}\left[  \delta g\text{ }\widehat{Q}^{2}\overline
{Z}_{k}\overline{Z}_{k}^{T}\right]  \overline{\phi}_{0}^{k}\left(  X_{i_{2}%
}\right)  \overline{\phi}_{0}^{k,T}\left(  X_{i_{2}}\right)  \widehat{E}%
\left[  \delta g\text{ }\widehat{Q}^{2}\overline{Z}_{k}\overline{Z}_{k}%
^{T}\right]  \right) \\
&  =\overline{\phi}_{0}^{k,T}\left(  X_{i_{2}}\right)  \left(  \widehat{E}%
\left[  \delta g\text{ }\widehat{Q}^{2}\overline{Z}_{k}\overline{Z}_{k}%
^{T}\right]  \right)  ^{2}\overline{\phi}_{0}^{k}\left(  X_{i_{2}}\right) \\
&  \leq\left\vert \left\vert \delta g\right\vert \right\vert _{\infty}%
^{2}\overline{\phi}_{0}^{k,T}\left(  X_{i_{2}}\right)  \overline{\phi}_{0}%
^{k}\left(  X_{i_{2}}\right)
\end{align*}

The last inequality holds because $\left\vert \left\vert \delta g\right\vert
\right\vert _{\infty}I_{k\times k}-\widehat{E}\left[  \delta g\text{
}\widehat{Q}^{2}\overline{Z}_{k}\overline{Z}_{k}^{T}\right]  $ is a
semi-positive definite symmetric matrix. $\ \ E_{\theta}\left[  \left(
E_{\theta}\left[  \widehat{IF}_{m-1,m-1,i_{1}i_{2}i_{3}i_{5}..i_{m}%
}|\mathbf{O}_{-i_{3}}\right]  \right)  ^{2}\right]  $ is of order $o\left(
k^{m-3}\right)  $ by induction assumption. \ Now, the proof of this
proposition is complete. moreover, by arguing similarly, the result above can
be generalized in a straightforward manner to the following proposition.
\end{proof}

\begin{proposition}
For any $m\geq2$ and $1\leq t<m,$%
\[
E_{\theta}\left[  \left\{  E_{\theta}\left(  \widehat{IF}_{m,m,\overline
{i}_{m}}|O_{i_{s_{1}}},...O_{i_{s_{t}}}\right)  \right\}  ^{2}\right]
=o\left(  k^{t-1}\right)
\]

\end{proposition}

Finally, we are now ready to prove Lemma \ref{3.19_lemma}.

\begin{proof}
(Lemma \ref{3.19_lemma}) Throughout the proof, we repeatedly use the result
that $\widehat{\mathbb{IF}}_{m,m,\widetilde{\psi}_{k}\left(  \cdot\right)  }$
for any $m\geq2$ is degenerate under $F\left(  \cdot;\widehat{\theta}\right)
. $ We first prove eq. (\ref{lemma_1}). $\ \ $%
\begin{align*}
&  Var_{\theta}\left[  \left(  \widehat{\mathbb{IF}}_{m,m,\widetilde{\psi}%
_{k}\left(  \cdot\right)  }\right)  ^{2}\right]  -Var_{\widehat{\theta}%
}\left[  \left(  \widehat{\mathbb{IF}}_{m,m,\widetilde{\psi}_{k}\left(
\cdot\right)  }\right)  ^{2}\right] \\
&  =E_{\theta}\left[  \left(  \widehat{\mathbb{IF}}_{m,m,\widetilde{\psi}%
_{k}\left(  \cdot\right)  }\right)  ^{2}\right]  -E_{\widehat{\theta}}\left[
\left(  \widehat{\mathbb{IF}}_{m,m,\widetilde{\psi}_{k}\left(  \cdot\right)
}\right)  ^{2}\right]  -\left(  E_{\theta}\left[  \widehat{\mathbb{IF}%
}_{m,m,\widetilde{\psi}_{k}\left(  \cdot\right)  }\right]  \right)  ^{2}\\
&  =E_{\widehat{\theta}}\left[  \left(  \widehat{\mathbb{IF}}%
_{m,m,\widetilde{\psi}_{k}\left(  \cdot\right)  }\right)  ^{2}\left(
%TCIMACRO{\tprod \limits_{i=1}^{n}}%
%BeginExpansion
{\textstyle\prod\limits_{i=1}^{n}}
%EndExpansion
\frac{f\left(  O_{i}\right)  }{\widehat{f}\left(  O_{i}\right)  }-1\right)
\right]  -\left(  E_{\theta}\left[  \widehat{\mathbb{IF}}_{m,m,\widetilde{\psi
}_{k}\left(  \cdot\right)  }\right]  \right)  ^{2}%
\end{align*}
can be written as a sum of four terms as below:%
\begin{align*}
&  E_{\widehat{\theta}}\left[  \left(  \widehat{\mathbb{IF}}%
_{m,m,\widetilde{\psi}_{k}\left(  \cdot\right)  }\right)  ^{2}\left(
%TCIMACRO{\tprod \limits_{i=1}^{n}}%
%BeginExpansion
{\textstyle\prod\limits_{i=1}^{n}}
%EndExpansion
\frac{f\left(  O_{i}\right)  }{\widehat{f}\left(  O_{i}\right)  }-1\right)
\right]  -\left(  E_{\theta}\left[  \widehat{\mathbb{IF}}_{m,m,\widetilde{\psi
}_{k}\left(  \cdot\right)  }\right]  \right)  ^{2}\\
&  =E_{\widehat{\theta}}\left\{
\begin{array}
[c]{c}%
\left[  \binom{n}{m}\right]  ^{-2}\left(
%TCIMACRO{\dsum \limits_{i_{1}<i_{2}..<i_{m}}}%
%BeginExpansion
{\displaystyle\sum\limits_{i_{1}<i_{2}..<i_{m}}}
%EndExpansion
\widehat{IF}_{m,m,\overline{i}_{m}}^{\left(  s\right)  }\right) \\
\times\left(
%TCIMACRO{\dsum \limits_{r_{1}<r_{2}..<r_{m}}}%
%BeginExpansion
{\displaystyle\sum\limits_{r_{1}<r_{2}..<r_{m}}}
%EndExpansion
\widehat{IF}_{m,m,\overline{r}_{m}}^{\left(  s\right)  }\right)  \left(
%TCIMACRO{\tprod \limits_{i=1}^{n}}%
%BeginExpansion
{\textstyle\prod\limits_{i=1}^{n}}
%EndExpansion
\frac{f\left(  O_{i}\right)  }{\widehat{f}\left(  O_{i}\right)  }-1\right)
\end{array}
\right\} \\
&  -\left(  E_{\theta}\left[  \widehat{\mathbb{IF}}_{m,m,\widetilde{\psi}%
_{k}\left(  \cdot\right)  }\right]  \right)  ^{2}\\
&  =E_{\widehat{\theta}}\left\{
\begin{array}
[c]{c}%
\left[  \binom{n}{m}\right]  ^{-2}\left[
\begin{array}
[c]{c}%
%TCIMACRO{\dsum \limits_{i_{1}<i_{2}..<i_{m}}}%
%BeginExpansion
{\displaystyle\sum\limits_{i_{1}<i_{2}..<i_{m}}}
%EndExpansion
\left(  \widehat{IF}_{m,m,\overline{i}_{m}}^{\left(  s\right)  }\right)
^{2}\\
+%
%TCIMACRO{\dsum \limits_{\overline{i}_{m}\cap\overline{r}_{m}=\emptyset}}%
%BeginExpansion
{\displaystyle\sum\limits_{\overline{i}_{m}\cap\overline{r}_{m}=\emptyset}}
%EndExpansion
\widehat{IF}_{m,m,\overline{i}_{m}}^{\left(  s\right)  }\widehat{IF}%
_{m,m,\overline{r}_{m}}^{\left(  s\right)  }\\
+%
%TCIMACRO{\dsum \limits_{1\leq\#\left(  \overline{i}_{m}\cap\overline{r}%
%_{m}\right)  <m}}%
%BeginExpansion
{\displaystyle\sum\limits_{1\leq\#\left(  \overline{i}_{m}\cap\overline{r}%
_{m}\right)  <m}}
%EndExpansion
\widehat{IF}_{m,m,\overline{i}_{m}}^{\left(  s\right)  }\widehat{IF}%
_{m,m,\overline{r}_{m}}^{\left(  s\right)  }%
\end{array}
\right] \\
\times\left(
%TCIMACRO{\tprod \limits_{i=1}^{n}}%
%BeginExpansion
{\textstyle\prod\limits_{i=1}^{n}}
%EndExpansion
\frac{f\left(  O_{i}\right)  }{\widehat{f}\left(  O_{i}\right)  }-1\right)
\end{array}
\right\} \\
&  -\left(  E_{\theta}\left[  \widehat{\mathbb{IF}}_{m,m,\widetilde{\psi}%
_{k}\left(  \cdot\right)  }\right]  \right)  ^{2}%
\end{align*}
where $\#\left(  \overline{i}_{m}\cap\overline{r}_{m}\right)  $ is the number
of elements in the intersection set $\left\{  i_{1},...,i_{m}\right\}
\cap\left\{  r_{1},...,r_{m}\right\}  .$
\ \ \ \ \ \ \ \ \ \ \ \ \ \ \ \ \ \ \ \ \ \ \ \ \ \ \ \ \ \ \ \ \ \ \ \ \ \ \ \ The
first term:
\begin{align*}
&  E_{\widehat{\theta}}\left[  \left[  \binom{n}{m}\right]  ^{-2}%
%TCIMACRO{\dsum \limits_{i_{1}<i_{2}..<i_{m}}}%
%BeginExpansion
{\displaystyle\sum\limits_{i_{1}<i_{2}..<i_{m}}}
%EndExpansion
\widehat{IF}_{m,m,\overline{i}_{m}}^{\left(  s\right)  2}\left(
%TCIMACRO{\tprod \limits_{i=1}^{n}}%
%BeginExpansion
{\textstyle\prod\limits_{i=1}^{n}}
%EndExpansion
\frac{f\left(  O_{i}\right)  }{\widehat{f}\left(  O_{i}\right)  }-1\right)
\right] \\
&  =\left[  \binom{n}{m}\right]  ^{-1}E_{\widehat{\theta}}\left[
\widehat{IF}_{m,m,\overline{i}_{m}}^{\left(  s\right)  2}\left(
%TCIMACRO{\tprod \limits_{i=1}^{n}}%
%BeginExpansion
{\textstyle\prod\limits_{i=1}^{n}}
%EndExpansion
\frac{f\left(  O_{i}\right)  }{\widehat{f}\left(  O_{i}\right)  }-1\right)
\right] \\
&  =\left[  \binom{n}{m}\right]  ^{-1}\left(  E_{\theta}\left[  \widehat{IF}%
_{m,m,\overline{i}_{m}}^{\left(  s\right)  2}\right]  -E_{\widehat{\theta}%
}\left[  \widehat{IF}_{m,m,\overline{i}_{m}}^{\left(  s\right)  2}\right]
\right)
\end{align*}
The second term:
\begin{align*}
&  E_{\widehat{\theta}}\left[  \left[  \binom{n}{m}\right]  ^{-2}\left(
%TCIMACRO{\dsum \limits_{\overline{i}_{m}\cap\overline{r}_{m}=\emptyset}}%
%BeginExpansion
{\displaystyle\sum\limits_{\overline{i}_{m}\cap\overline{r}_{m}=\emptyset}}
%EndExpansion
\widehat{IF}_{m,m,\overline{i}_{m}}^{\left(  s\right)  }\widehat{IF}%
_{m,m,\overline{r}_{m}}^{\left(  s\right)  }\right)  \left(
%TCIMACRO{\tprod \limits_{i=1}^{n}}%
%BeginExpansion
{\textstyle\prod\limits_{i=1}^{n}}
%EndExpansion
\frac{f\left(  O_{i}\right)  }{\widehat{f}\left(  O_{i}\right)  }-1\right)
\right] \\
&  =\left[  \binom{n}{m}\right]  ^{-2}E_{\theta}\left[
%TCIMACRO{\dsum \limits_{\overline{i}_{m}\cap\overline{r}_{m}=\emptyset}}%
%BeginExpansion
{\displaystyle\sum\limits_{\overline{i}_{m}\cap\overline{r}_{m}=\emptyset}}
%EndExpansion
\widehat{IF}_{m,m,\overline{i}_{m}}^{\left(  s\right)  }\widehat{IF}%
_{m,m,\overline{r}_{m}}^{\left(  s\right)  }\right] \\
&  =\left[  \binom{n}{m}\right]  ^{-2}\binom{n}{m}\binom{n-m}{m}E_{\theta
}\left[  \widehat{IF}_{m,m,\overline{i}_{m}}^{\left(  s\right)  }%
\widehat{IF}_{m,m,\overline{r}_{m}}^{\left(  s\right)  }\right] \\
&  =\frac{\left(  n-m\right)  ...\left(  n-2m+1\right)  }{n\left(  n-1\right)
...\left(  n-m+1\right)  }\left(  E_{\theta}\left[  \widehat{IF}%
_{m,m,\overline{i}_{m}}^{\left(  s\right)  }\right]  \right)  ^{2}%
\end{align*}
Substracting the fourth term from the second term, we have
\begin{align*}
&  \left[  \frac{\left(  n-m\right)  ...\left(  n-2m+1\right)  }{n\left(
n-1\right)  ...\left(  n-m+1\right)  }-1\right]  \left(  E_{\theta}\left[
\widehat{IF}_{m,m,\overline{i}_{m}}^{\left(  s\right)  }\right]  \right)
^{2}\\
&  =O\left(  \frac{1}{n}\left(  E_{\theta}\left[  \widehat{IF}_{m,m,\overline
{i}_{m}}^{\left(  s\right)  }\right]  \right)  ^{2}\right)
\end{align*}
The third term:%
\begin{align*}
&  E_{\widehat{\theta}}\left[  \left[  \binom{n}{m}\right]  ^{-2}\left(
%TCIMACRO{\dsum \limits_{1\leq\#\left(  \overline{i}_{m}\cap\overline{r}%
%_{m}\right)  <m}}%
%BeginExpansion
{\displaystyle\sum\limits_{1\leq\#\left(  \overline{i}_{m}\cap\overline{r}%
_{m}\right)  <m}}
%EndExpansion
\widehat{IF}_{m,m,\overline{i}_{m}}^{\left(  s\right)  }\widehat{IF}%
_{m,m,\overline{r}_{m}}^{\left(  s\right)  }\right)  \left(
%TCIMACRO{\tprod \limits_{i=1}^{n}}%
%BeginExpansion
{\textstyle\prod\limits_{i=1}^{n}}
%EndExpansion
\frac{f\left(  O_{i}\right)  }{\widehat{f}\left(  O_{i}\right)  }-1\right)
\right] \\
&  =\sum_{t=1}^{m-1}E_{\theta}\left[  \left[  \binom{n}{m}\right]
^{-2}\left(
%TCIMACRO{\dsum \limits_{\#\left(  \overline{i}_{m}\cap\overline{r}_{m}\right)
%=t}}%
%BeginExpansion
{\displaystyle\sum\limits_{\#\left(  \overline{i}_{m}\cap\overline{r}%
_{m}\right)  =t}}
%EndExpansion
\widehat{IF}_{m,m,\overline{i}_{m}}^{\left(  s\right)  }\widehat{IF}%
_{m,m,\overline{r}_{m}}^{\left(  s\right)  }\right)  \right] \\
&  =\sum_{t=1}^{m-1}\left[  \binom{n}{m}\right]  ^{-2}\binom{n}{2m-t}\left[
\binom{m}{t}\right]  ^{2}E_{\theta}\left[  \widehat{IF}_{m,m,i_{1}i_{2}%
..i_{t}i_{t+1}..i_{m}}^{\left(  s\right)  }\widehat{IF}_{m,m,i_{1}i_{2}%
..i_{t}i_{m+1}..i_{2m-t}}^{\left(  s\right)  }\right] \\
&  =O\left(  \sum_{t=1}^{m-1}n^{-t}E_{\theta}\left[  \left(  E_{\theta}\left[
\widehat{IF}_{m,m,\overline{i}_{m}}^{\left(  s\right)  }|\mathbf{O}%
_{\overline{i}_{t}}\right]  \right)  ^{2}\right]  \right) \\
&  =O\left(  \sum_{t=1}^{m-1}n^{-t}\left\{  E_{\theta}\left[  \left(
\widehat{IF}_{m,m,\overline{i}_{m}}|O_{s_{1}},...O_{s_{t}}\right)
^{2}\right]  E_{\theta}\left[  \left(  \widehat{IF}_{m,m,\overline{i}_{m}%
}|O_{v_{1}},...O_{v_{t}}\right)  ^{2}\right]  \right\}  ^{1/2}\right) \\
&  =o\left(  \max\left(  \frac{1}{n},\frac{k^{m-2}}{n^{m-1}}\right)  \right)
\end{align*}
The last two equalities follow from Cauchy-Shwartz inequality and Lemma
\ref{3.19_lemma}. Specifically,
\begin{align*}
&  E_{\theta}\left(  E_{\theta}\left[  \widehat{IF}_{m,m,\overline{i}_{m}%
}^{\left(  s\right)  }|\mathbf{O}_{\overline{i}_{t}}\right]  \right)  ^{2}\\
&  =E_{\theta}\left[  E_{\theta}\left(  \widehat{IF}_{m,m,\overline{i}%
_{m}^{\ast}}|O_{i_{1}},...O_{i_{t}}\right)  E_{\theta}\left(  \widehat{IF}%
_{m,m,\overline{r}_{m}^{\ast}}|O_{i_{1}},...O_{i_{t}}\right)  \right] \\
&  \leq\left\{  E_{\theta}\left[  \left(  E_{\theta}\left(  \widehat{IF}%
_{m,m,\overline{i}_{m}^{\ast}}|O_{i_{1}},...O_{i_{t}}\right)  \right)
^{2}\right]  \right\}  ^{1/2}\\
&  \times\left\{  E_{\theta}\left[  \left(  E_{\theta}\left(  \widehat{IF}%
_{m,m,\overline{r}_{m}^{\ast}}|O_{i_{1}},...O_{i_{t}}\right)  \right)
^{2}\right]  \right\}  ^{1/2}.
\end{align*}

where $\overline{i}_{m}^{\ast}$ and $\overline{r}_{m}^{\ast}$ are two
permutations of $\left(  i_{1},i_{2},...i_{m}\right)  .$

Next, we prove eq. \ref{lemma_2} for any $1\leq m_{1}<m_{2}.$ \ Here we also
rewrite $Cov_{\theta}\left(  \widehat{\mathbb{IF}}_{m_{1},m_{1}%
,\widetilde{\psi}_{k}\left(  \cdot\right)  },\text{ }\widehat{\mathbb{IF}%
}_{m_{2},m_{2},\widetilde{\psi}_{k}\left(  \cdot\right)  }\right)  $ as a sum
of four terms.
\begin{align*}
&  Cov_{\theta}\left(  \widehat{\mathbb{IF}}_{m_{1},m_{1},\widetilde{\psi}%
_{k}\left(  \cdot\right)  },\text{ }\widehat{\mathbb{IF}}_{m_{2}%
,m_{2},\widetilde{\psi}_{k}\left(  \cdot\right)  }\right) \\
&  =E_{\theta}\left[  \widehat{\mathbb{IF}}_{m_{1},m_{1},\widetilde{\psi}%
_{k}\left(  \cdot\right)  }\widehat{\mathbb{IF}}_{m_{2},m_{2},\widetilde{\psi
}_{k}\left(  \cdot\right)  }\right]  -E_{\theta}\left[  \widehat{IF}%
_{m_{1},m_{1},\overline{i}_{m}}^{\left(  s\right)  }\right]  E_{\theta}\left[
\widehat{IF}_{m_{2},m_{2},\overline{i}_{m}}^{\left(  s\right)  }\right] \\
&  =E_{\theta}\left[  \frac{1}{\binom{n}{m_{1}}\binom{n}{m_{2}}}\sum
_{i_{1}<i_{2}..<i_{m_{1}}}\widehat{IF}_{m_{1},m_{1},\overline{i}_{m_{1}}%
}^{\left(  s\right)  }\sum_{r_{1}<r_{2}..<r_{m_{2}}}\widehat{IF}_{m_{2}%
,m_{2},\overline{r}_{m_{2}}}^{\left(  s\right)  }\right] \\
&  -E_{\theta}\left[  \widehat{IF}_{m_{1},m_{1},\overline{i}_{m}}^{\left(
s\right)  }\right]  E_{\theta}\left[  \widehat{IF}_{m_{2},m_{2},\overline
{i}_{m}}^{\left(  s\right)  }\right] \\
&  =E_{\theta}\left[  \frac{1}{\binom{n}{m_{1}}\binom{n}{m_{2}}}\left\{
\begin{array}
[c]{c}%
\left(
%TCIMACRO{\dsum \limits_{\overline{i}_{m_{1}}\subset\overline{r}_{m_{2}}}}%
%BeginExpansion
{\displaystyle\sum\limits_{\overline{i}_{m_{1}}\subset\overline{r}_{m_{2}}}}
%EndExpansion
+%
%TCIMACRO{\dsum \limits_{\overline{i}_{m_{1}}\cap\overline{r}_{m_{2}}%
%=\emptyset}}%
%BeginExpansion
{\displaystyle\sum\limits_{\overline{i}_{m_{1}}\cap\overline{r}_{m_{2}%
}=\emptyset}}
%EndExpansion
+%
%TCIMACRO{\dsum \limits_{1\leq\#\left(  \overline{i}_{m_{1}}\cap\overline
%{r}_{m_{2}}\right)  <m_{1}}}%
%BeginExpansion
{\displaystyle\sum\limits_{1\leq\#\left(  \overline{i}_{m_{1}}\cap\overline
{r}_{m_{2}}\right)  <m_{1}}}
%EndExpansion
\right) \\
\times\widehat{IF}_{m_{1},m_{1},\overline{i}_{m_{1}}}^{\left(  s\right)
}\widehat{IF}_{m_{2},m_{2},\overline{r}_{m_{2}}}^{\left(  s\right)  }%
\end{array}
\right\}  \right] \\
&  -E_{\theta}\left[  \widehat{IF}_{m_{1},m_{1},\overline{i}_{m}}^{\left(
s\right)  }\right]  E_{\theta}\left[  \widehat{IF}_{m_{2},m_{2},\overline
{i}_{m}}^{\left(  s\right)  }\right]
\end{align*}
The first term:
\begin{align*}
&  E_{\theta}\left[  \frac{1}{\binom{n}{m_{1}}\binom{n}{m_{2}}}%
%TCIMACRO{\dsum \limits_{\overline{i}_{m_{1}}\subset\overline{r}_{m_{2}}}}%
%BeginExpansion
{\displaystyle\sum\limits_{\overline{i}_{m_{1}}\subset\overline{r}_{m_{2}}}}
%EndExpansion
\widehat{IF}_{m_{1},m_{1},\overline{i}_{m_{1}}}^{\left(  s\right)
}\widehat{IF}_{m_{2},m_{2},\overline{r}_{m_{2}}}^{\left(  s\right)  }\right]
\\
&  =\frac{1}{\binom{n}{m_{1}}\binom{n}{m_{2}}}\binom{n}{m_{2}}\binom{m_{2}%
}{m_{1}}E_{\theta}\left[  \widehat{IF}_{m_{1},m_{1},\overline{i}_{m_{1}}%
}^{\left(  s\right)  }\widehat{IF}_{m_{2},m_{2},\overline{r}_{m_{2}}}^{\left(
s\right)  }\right] \\
&  =\frac{\binom{m_{2}}{m_{1}}}{\binom{n}{m_{1}}}E_{\theta}\left[
\widehat{IF}_{m_{1},m_{1},\overline{i}_{m_{1}}}^{\left(  s\right)  }E_{\theta
}\left[  \widehat{IF}_{m_{2},m_{2},\overline{r}_{m_{2}}}^{\left(  s\right)
}|\mathbf{O}_{\overline{i}_{m_{1}}}\right]  \right] \\
&  \leq\frac{\binom{m_{2}}{m_{1}}}{\binom{n}{m_{1}}}\left\{  E_{\theta}\left[
\left(  \widehat{IF}_{m_{1},m_{1},\overline{i}_{m_{1}}}^{\left(  s\right)
}\right)  ^{2}\right]  E_{\theta}\left[  \left(  E_{\theta}\left[
\widehat{IF}_{m_{2},m_{2},\overline{r}_{m_{2}}}^{\left(  s\right)
}|\mathbf{O}_{\overline{i}_{m_{1}}}\right]  \right)  ^{2}\right]  \right\}
^{1/2}\\
&  =o\left(  \frac{k^{m_{1}-1}}{n^{m_{1}}}\right)  ,
\end{align*}
which follows from Cauchy-Shwartz inequality, Theorem (\ref{var_multi}), and
Lemma \ref{3.19_lemma}.

The difference between the second and the fourth terms equals%
\begin{align*}
&  E_{\theta}\left[  \frac{1}{\binom{n}{m_{1}}\binom{n}{m_{2}}}%
%TCIMACRO{\dsum \limits_{\overline{i}_{m_{1}}\cap\overline{r}_{m_{2}}%
%=\emptyset}}%
%BeginExpansion
{\displaystyle\sum\limits_{\overline{i}_{m_{1}}\cap\overline{r}_{m_{2}%
}=\emptyset}}
%EndExpansion
\widehat{IF}_{m_{1},m_{1},\overline{i}_{m_{1}}}^{\left(  s\right)
}\widehat{IF}_{m_{2},m_{2},\overline{r}_{m_{2}}}^{\left(  s\right)  }\right]
\\
&  -E_{\theta}\left[  \widehat{IF}_{m_{1},m_{1},\overline{i}_{m}}^{\left(
s\right)  }\right]  E_{\theta}\left[  \widehat{IF}_{m_{2},m_{2},\overline
{i}_{m}}^{\left(  s\right)  }\right] \\
&  =\left(  \frac{1}{\binom{n}{m_{1}}\binom{n}{m_{2}}}\binom{n}{m_{1}}%
\binom{n-m_{1}}{m_{2}}-1\right)  E_{\theta}\left[  \widehat{IF}_{m_{1}%
,m_{1},\overline{i}_{m}}^{\left(  s\right)  }\right]  E_{\theta}\left[
\widehat{IF}_{m_{2},m_{2},\overline{i}_{m}}^{\left(  s\right)  }\right] \\
&  =O\left(  \frac{1}{n}E_{\theta}\left[  \widehat{IF}_{m_{1},m_{1}%
,\overline{i}_{m}}^{\left(  s\right)  }\right]  E_{\theta}\left[
\widehat{IF}_{m_{2},m_{2},\overline{i}_{m}}^{\left(  s\right)  }\right]
\right)
\end{align*}
The third term:%
\begin{align*}
&  E_{\theta}\left[  \frac{1}{\binom{n}{m_{1}}\binom{n}{m_{2}}}%
%TCIMACRO{\dsum \limits_{1\leq\#\left(  \overline{i}_{m_{1}}\cap\overline
%{r}_{m_{2}}\right)  <m_{1}}}%
%BeginExpansion
{\displaystyle\sum\limits_{1\leq\#\left(  \overline{i}_{m_{1}}\cap\overline
{r}_{m_{2}}\right)  <m_{1}}}
%EndExpansion
\widehat{IF}_{m_{1},m_{1},\overline{i}_{m_{1}}}^{\left(  s\right)
}\widehat{IF}_{m_{2},m_{2},\overline{r}_{m_{2}}}^{\left(  s\right)  }\right]
\\
&  =%
%TCIMACRO{\dsum \limits_{t=1}^{m_{1}-1}}%
%BeginExpansion
{\displaystyle\sum\limits_{t=1}^{m_{1}-1}}
%EndExpansion
\frac{1}{\binom{n}{m_{1}}\binom{n}{m_{2}}}E_{\theta}\left[
%TCIMACRO{\dsum \limits_{\#\left(  \overline{i}_{m_{1}}\cap\overline{r}_{m_{2}%
%}\right)  =t}}%
%BeginExpansion
{\displaystyle\sum\limits_{\#\left(  \overline{i}_{m_{1}}\cap\overline
{r}_{m_{2}}\right)  =t}}
%EndExpansion
\widehat{IF}_{m_{1},m_{1},\overline{i}_{m_{1}}}^{\left(  s\right)
}\widehat{IF}_{m_{2},m_{2},\overline{r}_{m_{2}}}^{\left(  s\right)  }\right]
\\
&  =%
%TCIMACRO{\dsum \limits_{t=1}^{m_{1}-1}}%
%BeginExpansion
{\displaystyle\sum\limits_{t=1}^{m_{1}-1}}
%EndExpansion
\left\{
\begin{array}
[c]{c}%
\frac{\binom{n}{m_{2}+m_{1}-t}\binom{m_{2}+m_{1}-t}{m_{2}}\binom{m_{2}}{t}%
}{\binom{n}{m_{1}}\binom{n}{m_{2}}}\times\\
E_{\theta}\left[  \widehat{IF}_{m_{1},m_{1},i_{1}i_{2}..i_{t}i_{t+1}%
..i_{m_{1}}}^{\left(  s\right)  }\widehat{IF}_{m_{2},m_{2},i_{1}%
..i_{t}i_{m_{1}+1}..i_{m_{1}+m_{2}-t}}^{\left(  s\right)  }\right]
\end{array}
\right\} \\
&  =o\left(  \max\left(  \frac{1}{n},\frac{k^{m_{1}-2}}{n^{m_{1}-1}}\right)
\right)  ,
\end{align*}
\newline which also follows from Cauchy-Shwartz inequality and Lemma
\ref{3.19_lemma}.
\end{proof}

\bigskip

\begin{proof}
(Eq. (\ref{multirob_bias})) \ We prove the bias property of $\widehat{\psi
}_{m,k}^{\operatorname{mod}}$ by induction.

For $m=2,$The estimation bias is given by%
\begin{align*}
&  -\left\{  E\left[  Q^{2}\left(  \frac{P-\widehat{P}}{\dot{P}}\right)
\overline{Z}_{k}^{T}\right]  E\left[  Q^{2}\left(  \frac{B-\widehat{B}}%
{\dot{B}}\right)  \overline{Z}_{k}\right]  \right\} \\
&  +\left\{
\begin{array}
[c]{c}%
E\left[  Q^{2}\left(  \frac{P-\widehat{P}}{\dot{P}}\right)  \overline{Z}%
_{k}^{T}\right]  \left\{  E\left[  \dot{P}\dot{B}H_{1}\overline{Z}%
_{k}\overline{Z}_{k}^{T}\right]  \right\}  ^{-1}\\
\times E\left[  \overline{Z}_{k}Q^{2}\left(  \frac{B-\widehat{B}}{\dot{B}%
}\right)  \right]
\end{array}
\right\}
\end{align*}%
\[
=
\]%
\begin{align*}
&  \left\{
\begin{array}
[c]{c}%
E\left[  \left(  H_{1}\widehat{P}+H_{2}\right)  \dot{B}\overline{Z}_{k}%
^{T}\right]  \left[  \left\{  E\left[  \dot{P}\dot{B}H_{1}\overline{Z}%
_{k}\overline{Z}_{k}^{T}\right]  \right\}  ^{-1}-I\right] \\
\times E\left[  \overline{Z}_{k}\left(  H_{1}\widehat{B}+H_{3}\right)  \dot
{P}\right]
\end{array}
\right\} \\
&  =-\left\{
\begin{array}
[c]{c}%
E\left[  \left(  H_{1}\widehat{P}+H_{2}\right)  \dot{B}\overline{Z}_{k}%
^{T}\right]  \left\{  E\left[  \dot{P}\dot{B}H_{1}\overline{Z}_{k}\overline
{Z}_{k}^{T}\right]  -I\right\} \\
\times E\left[  \dot{P}\dot{B}H_{1}\overline{Z}_{k}\overline{Z}_{k}%
^{T}\right]  ^{-1}E\left[  \overline{Z}_{k}\left(  H_{1}\widehat{B}%
+H_{3}\right)  \dot{P}\right]
\end{array}
\right\}
\end{align*}

Suppose the bias formula holds for $m$, then the bias at $m+1$ is%
\begin{align*}
&  \left(  -1\right)  ^{m-1}\left\{
\begin{array}
[c]{c}%
E\left[  Q^{2}\left(  \frac{P-\widehat{P}}{\dot{P}}\right)  \overline{Z}%
_{k}^{T}\right]  \left\{  E\left[  Q^{2}\overline{Z}_{k}\overline{Z}_{k}%
^{T}\right]  -I\right\} \\
\times%
%TCIMACRO{\dprod \limits_{s=3}^{m}}%
%BeginExpansion
{\displaystyle\prod\limits_{s=3}^{m}}
%EndExpansion
\left\{  \widehat{E}_{s}\left[  \dot{P}\dot{B}H_{1}\overline{Z}_{k}%
\overline{Z}_{k}^{T}\right]  \right\}  ^{-1}\left\{  E\left[  Q^{2}%
\overline{Z}_{k}\overline{Z}_{k}^{T}\right]  -\widehat{E}_{s}\left[
Q^{2}\overline{Z}_{k}\overline{Z}_{k}^{T}\right]  \right\} \\
\times\left\{  E\left[  Q^{2}\overline{Z}_{k}\overline{Z}_{k}^{T}\right]
\right\}  ^{-1}E\left[  Q^{2}\overline{Z}_{k}\left(  \frac{B-\widehat{B}}%
{\dot{B}}\right)  \right]
\end{array}
\right\} \\
&  +\left(  -1\right)  ^{m}E\left[  Q^{2}\left(  \frac{P-\widehat{P}}{\dot{P}%
}\right)  \overline{Z}_{k}^{T}\right]  \left(  E\left[  \left(  \dot{P}\dot
{B}H_{1}\overline{Z}_{k}\overline{Z}_{k}^{T}\right)  _{i_{2}}\right]
-I\right)  \times\\
&  \left[
%TCIMACRO{\dprod \limits_{s=3}^{m}}%
%BeginExpansion
{\displaystyle\prod\limits_{s=3}^{m}}
%EndExpansion
\left\{  \widehat{E}_{s}\left[  \dot{P}\dot{B}H_{1}\overline{Z}_{k}%
\overline{Z}_{k}^{T}\right]  \right\}  ^{-1}\left\{  E\left[  \left(  \dot
{P}\dot{B}H_{1}\overline{Z}_{k}\overline{Z}_{k}^{T}\right)  _{i_{s}}\right]
-\widehat{E}_{s}\left[  \dot{P}\dot{B}H_{1}\overline{Z}_{k}\overline{Z}%
_{k}^{T}\right]  \right\}  \right] \\
&  \times\left\{  \widehat{E}_{m+1}\left[  \dot{P}\dot{B}H_{1}\overline{Z}%
_{k}\overline{Z}_{k}^{T}\right]  \right\}  ^{-1}E\left[  \overline{Z}_{k}%
Q^{2}\left(  \frac{B-\widehat{B}}{\dot{B}}\right)  \right]
\end{align*}%
\[
=
\]%
\begin{align*}
&  \left(  -1\right)  ^{m}\left\{
\begin{array}
[c]{c}%
E\left[  Q^{2}\left(  \frac{P-\widehat{P}}{\dot{P}}\right)  \overline{Z}%
_{k}^{T}\right]  \left\{  E\left[  Q^{2}\overline{Z}_{k}\overline{Z}_{k}%
^{T}\right]  -I\right\} \\
\times%
%TCIMACRO{\dprod \limits_{s=3}^{m}}%
%BeginExpansion
{\displaystyle\prod\limits_{s=3}^{m}}
%EndExpansion
\left\{  \widehat{E}_{s}\left[  \dot{P}\dot{B}H_{1}\overline{Z}_{k}%
\overline{Z}_{k}^{T}\right]  \right\}  ^{-1}\left\{
\begin{array}
[c]{c}%
E\left[  Q^{2}\overline{Z}_{k}\overline{Z}_{k}^{T}\right] \\
-\widehat{E}_{s}\left[  Q^{2}\overline{Z}_{k}\overline{Z}_{k}^{T}\right]
\end{array}
\right\} \\
\times\left[  \left\{  \widehat{E}_{m+1}\left[  \dot{P}\dot{B}H_{1}%
\overline{Z}_{k}\overline{Z}_{k}^{T}\right]  \right\}  ^{-1}-\left\{  E\left[
Q^{2}\overline{Z}_{k}\overline{Z}_{k}^{T}\right]  \right\}  ^{-1}\right] \\
\times E\left[  \overline{Z}_{k}Q^{2}\left(  \frac{B-\widehat{B}}{\dot{B}%
}\right)  \right]
\end{array}
\right\} \\
&  =\left(  -1\right)  ^{m}\left\{
\begin{array}
[c]{c}%
E\left[  Q^{2}\left(  \frac{P-\widehat{P}}{\dot{P}}\right)  \overline{Z}%
_{k}^{T}\right]  \left\{  E\left[  Q^{2}\overline{Z}_{k}\overline{Z}_{k}%
^{T}\right]  -I\right\} \\
\times%
%TCIMACRO{\dprod \limits_{s=3}^{m+1}}%
%BeginExpansion
{\displaystyle\prod\limits_{s=3}^{m+1}}
%EndExpansion
\left\{  \widehat{E}_{s}\left[  \dot{P}\dot{B}H_{1}\overline{Z}_{k}%
\overline{Z}_{k}^{T}\right]  \right\}  ^{-1}\left\{
\begin{array}
[c]{c}%
E\left[  Q^{2}\overline{Z}_{k}\overline{Z}_{k}^{T}\right] \\
-\widehat{E}_{s}\left[  Q^{2}\overline{Z}_{k}\overline{Z}_{k}^{T}\right]
\end{array}
\right\} \\
\times\left\{  E\left[  Q^{2}\overline{Z}_{k}\overline{Z}_{k}^{T}\right]
\right\}  ^{-1}E\left[  \overline{Z}_{k}Q^{2}\left(  \frac{B-\widehat{B}}%
{\dot{B}}\right)  \right]
\end{array}
\right\}
\end{align*}
which completes the proof.
\end{proof}

\textbf{Motivation and proofs for Section \ref{case2}.}

Before proving theorem \ref{beyond4th}, we provide some calculations as
motivation and several preliminary lemmas.

Since $E_{\theta}\left(  H_{1}B+H_{3}|X\right)  =E_{\theta}\left(
H_{1}P+H_{2}|X\right)  =0,$ we can show that
\begin{align*}
&  E_{\theta}\left(  \widehat{\mathbb{U}}_{m}\left(  \left(  l\right)
_{k\left(  l,0\right)  }^{k\left(  l,1\right)  },1\leq l\leq m-1\right)
\right) \\
&  =\left(  E_{\theta}\left(  \widehat{\epsilon}\overline{Z}_{k\left(
1,0\right)  }^{k\left(  1,1\right)  T}\right)
%TCIMACRO{\tprod \limits_{u=2}^{m-1}}%
%BeginExpansion
{\textstyle\prod\limits_{u=2}^{m-1}}
%EndExpansion
\left[  E_{\theta}\left(  \dot{B}\dot{P}H_{1}\overline{Z}_{k\left(
u-1,0\right)  }^{k\left(  u-1,1\right)  }\overline{Z}_{k\left(  u,0\right)
}^{k\left(  u,1\right)  T}-I_{k_{u-1}\times k_{u}}\right)  \right]  E_{\theta
}\left(  \overline{Z}_{k\left(  m-1,0\right)  }^{k\left(  m-1,1\right)
}\widehat{\Delta}\right)  \right) \\
&  =\left\{  \widehat{E}\left(  \left(  \delta g+1\right)  \delta
b\overline{Z}_{k\left(  1,0\right)  }^{k\left(  1,1\right)  T}\right)
%TCIMACRO{\tprod \limits_{u=2}^{m-1}}%
%BeginExpansion
{\textstyle\prod\limits_{u=2}^{m-1}}
%EndExpansion
\widehat{E}\left(  \delta g\widehat{Q}^{2}\overline{Z}_{k\left(  u-1,0\right)
}^{k\left(  u-1,1\right)  }\overline{Z}_{k\left(  u,0\right)  }^{k\left(
u,1\right)  T}\right)  \widehat{E}\left(  \left(  \delta g+1\right)  \delta
p\overline{Z}_{k\left(  m-1,0\right)  }^{k\left(  m-1,1\right)  }\right)
\right\}
\end{align*}
with $\delta b$ $=\dot{P}\widehat{E}\left(  H_{1}|X\right)  \left(
\widehat{B}-B\right)  ,$ $\delta p=\dot{B}\widehat{E}\left(  H_{1}|X\right)
\left(  \widehat{P}-P\right)  ,$ $\delta g=\frac{g-\widehat{g}}{\widehat{g}}$
and $\widehat{Q}^{2}=\dot{B}\dot{P}\widehat{E}\left(  H_{1}|X\right)  .$

Below we give a useful representation of $E_{\theta}\left[
\widehat{\mathbb{U}}_{m}\left(  \left(  l\right)  _{k\left(  l,0\right)
}^{k\left(  l,1\right)  },1\leq l\leq m-1\right)  \right]  .\mathit{\ \ }$Let
\begin{align*}
&  B_{m}\left(  \widehat{\mathbb{U}}_{m}\left(  \left(  l\right)  _{k\left(
l,0\right)  }^{k\left(  l,1\right)  },1\leq l\leq m-1\right)  \right) \\
&  \equiv\left\{  \widehat{E}\left(  \delta b\overline{Z}_{k\left(
1,0\right)  }^{k\left(  1,1\right)  T}\right)  \left[
%TCIMACRO{\tprod \limits_{u=2}^{m-1}}%
%BeginExpansion
{\textstyle\prod\limits_{u=2}^{m-1}}
%EndExpansion
\widehat{E}\left(  \delta g\widehat{Q}^{2}\overline{Z}_{k\left(  u-1,0\right)
}^{k\left(  u-1,1\right)  }\overline{Z}_{k\left(  u,0\right)  }^{k\left(
u,1\right)  T}\right)  \right]  \widehat{E}\left(  \delta p\overline
{Z}_{k\left(  m-1,0\right)  }^{k\left(  m-1,1\right)  }\right)  \right\} \\
&  B_{m+1}^{bg}\left(  \widehat{\mathbb{U}}_{m}\left(  \left(  l\right)
_{k\left(  l,0\right)  }^{k\left(  l,1\right)  },1\leq l\leq m-1\right)
\right) \\
&  \equiv\left\{  \widehat{E}\left(  \delta g\delta b\overline{Z}_{k\left(
1,0\right)  }^{k\left(  1,1\right)  T}\right)
%TCIMACRO{\tprod \limits_{u=2}^{m-1}}%
%BeginExpansion
{\textstyle\prod\limits_{u=2}^{m-1}}
%EndExpansion
\widehat{E}\left(  \delta g\widehat{Q}^{2}\overline{Z}_{k\left(  u-1,0\right)
}^{k\left(  u-1,1\right)  }\overline{Z}_{k\left(  u,0\right)  }^{k\left(
u,1\right)  T}\right)  \widehat{E}\left(  \delta p\overline{Z}_{k\left(
m-1,0\right)  }^{k\left(  m-1,1\right)  }\right)  \right\}
\end{align*}%
\begin{align*}
&  B_{m+1}^{pg}\left(  \widehat{\mathbb{U}}_{m}\left(  \left(  l\right)
_{k\left(  l,0\right)  }^{k\left(  l,1\right)  },1\leq l\leq m-1\right)
\right) \\
&  \equiv\left\{  \widehat{E}\left(  \delta b\overline{Z}_{k\left(
1,0\right)  }^{k\left(  1,1\right)  T}\right)
%TCIMACRO{\tprod \limits_{u=2}^{m-1}}%
%BeginExpansion
{\textstyle\prod\limits_{u=2}^{m-1}}
%EndExpansion
\widehat{E}\left(  \delta g\widehat{Q}^{2}\overline{Z}_{k\left(  u-1,0\right)
}^{k\left(  u-1,1\right)  }\overline{Z}_{k\left(  u,0\right)  }^{k\left(
u,1\right)  T}\right)  \widehat{E}\left(  \delta g\delta p\overline
{Z}_{k\left(  m-1,0\right)  }^{k\left(  m-1,1\right)  }\right)  \right\} \\
&  B_{m+2}\left(  \widehat{\mathbb{U}}_{m}\left(  \left(  l\right)  _{k\left(
l,0\right)  }^{k\left(  l,1\right)  },1\leq l\leq m-1\right)  \right) \\
&  \equiv\left\{  \widehat{E}\left(  \delta g\delta b\overline{Z}_{k\left(
1,0\right)  }^{k\left(  1,1\right)  T}\right)
%TCIMACRO{\tprod \limits_{u=2}^{m-1}}%
%BeginExpansion
{\textstyle\prod\limits_{u=2}^{m-1}}
%EndExpansion
\widehat{E}\left(  \delta g\widehat{Q}^{2}\overline{Z}_{k\left(  u-1,0\right)
}^{k\left(  u-1,1\right)  }\overline{Z}_{k\left(  u,0\right)  }^{k\left(
u,1\right)  T}\right)  \widehat{E}\left(  \delta g\delta p\overline
{Z}_{k\left(  m-1,0\right)  }^{k\left(  m-1,1\right)  }\right)  \right\}
\end{align*}

Then we may write
\begin{align*}
&  E_{\theta}\left(  \widehat{\mathbb{U}}_{m}\left(  \left(  l\right)
_{k\left(  l,0\right)  }^{k\left(  l,1\right)  },1\leq l\leq m-1\right)
\right) \\
&  =\left(  B_{m}+B_{m+1}^{bg}+B_{m+1}^{pg}+B_{m+2}\right)  \left(
\widehat{\mathbb{U}}_{m}\left(  \left(  l\right)  _{k\left(  l,0\right)
}^{k\left(  l,1\right)  },1\leq l\leq m-1\right)  \right)
\end{align*}

We shall require the following:

\begin{lemma}
\label{cp}Assume $k_{1}=\left(  k_{1}\left(  l,s\right)  \right)  _{\left(
t-1\right)  \times2}$ is a $\left(  t-1\right)  \times2$ dimensional matrix
with $l\in\left\{  0,1,..,t-2\right\}  ,$ $s\in\left\{  0,1\right\}  ,$ and
$k_{1}\left(  0,0\right)  =k_{1}\left(  t-2,0\right)  \equiv0.$ ${}$For
$\forall$ $t>3,$ if
\begin{gather*}
\chi\left(  t,k_{1};\theta\right)  \equiv\\
\left\{
\begin{array}
[c]{c}%
\left(
\begin{array}
[c]{c}%
B_{t}\left(  \widehat{\mathbb{U}}_{t-2}\left(  \left(  l\right)
_{k_{1}\left(  l,0\right)  }^{k_{1}\left(  l,1\right)  },1\leq l\leq
t-3\right)  \right) \\
-B_{t}^{bg}\left(  \widehat{\mathbb{U}}_{t-1}\left(  \left(  l\right)
_{k_{1}\left(  l,0\right)  }^{k_{1}\left(  l,1\right)  },\left(  t-2\right)
_{k_{1}\left(  t-2,0\right)  }^{k_{1}\left(  t-2,0\right)  },1\leq l\leq
t-3\right)  \right)
\end{array}
\right) \\
-\left(
\begin{array}
[c]{c}%
B_{t}^{pg}\left(  \widehat{\mathbb{U}}_{t-1}\left(  \left(  1\right)
_{k_{1}\left(  0,0\right)  }^{k_{1}\left(  0,1\right)  },\left(  l+1\right)
_{k_{1}\left(  l,0\right)  }^{k_{1}\left(  l,1\right)  },1\leq l\leq
t-3\right)  \right) \\
-B_{t}\left(  \widehat{\mathbb{U}}_{t}\left(  \left(  1\right)  _{k_{1}\left(
0,0\right)  }^{k_{1}\left(  0,1\right)  },\left(  l+1\right)  _{k_{1}\left(
l,0\right)  }^{k_{1}\left(  l,1\right)  },\left(  t-1\right)  _{k_{1}\left(
t-2,0\right)  }^{k_{1}\left(  t-2,0\right)  },1\leq l\leq t-3\right)  \right)
\end{array}
\right)
\end{array}
\right\}
\end{gather*}

then%
\[
\left\vert \chi\left(  t,k_{1};\theta\right)  \right\vert =O_{p}\left(
\left(  \frac{\log n}{n}\right)  ^{-\frac{\beta_{g}}{d+2\beta_{g}}}\left(
k_{1}\left(  0,1\right)  \right)  ^{-\frac{\beta_{b}}{d}}\left(  k_{1}\left(
t-2,1\right)  \right)  ^{-\frac{\beta p}{d}}\right)
\]

\end{lemma}

This lemma explains how to use higher order U-statistics to estimate the $t$th
order contribution of $E_{\theta}\left(  \widehat{\mathbb{U}}_{t-2}\left(
\left(  l\right)  _{k_{1}\left(  l,0\right)  }^{k_{1}\left(  l,1\right)
},1\leq l\leq t-3\right)  \right)  $ with a residual bias not exceeding
$\left(  \frac{\log n}{n}\right)  ^{-\frac{\beta_{g}}{d+2\beta_{g}}}\left(
k_{1}\left(  0,1\right)  \right)  ^{-\frac{\beta_{b}}{d}}\left(  k_{1}\left(
t-2,1\right)  \right)  ^{-\frac{\beta p}{d}}.$ Our estimator uses this idea to
reduce fourth and higher order estimation bias to the optimal rate.

\begin{proof}
(Lemma \ref{cp})%
\begin{align*}
&  \chi\left(  t,k_{1};\theta\right) \\
&  =\left\{
\begin{array}
[c]{c}%
\widehat{E}\left(  \delta b\delta g\overline{Z}_{k_{1}\left(  1,0\right)
}^{k_{1}\left(  1,1\right)  T}\right)
%TCIMACRO{\tprod \limits_{u=2}^{t-3}}%
%BeginExpansion
{\textstyle\prod\limits_{u=2}^{t-3}}
%EndExpansion
\widehat{E}\left(  \delta g\widehat{Q}^{2}\overline{Z}_{k_{1}\left(
u-1,0\right)  }^{k_{1}\left(  u-1,1\right)  }\overline{Z}_{k_{1}\left(
u,0\right)  }^{k_{1}\left(  u,1\right)  T}\right)  \times\\
\widehat{E}\left(  \delta g\widehat{Q}\overline{Z}_{k_{1}\left(  t-3,0\right)
}^{k_{1}\left(  t-3,1\right)  }\left(  \widehat{Q}^{-1}\delta p-\widehat{\Pi
}\left(  \widehat{Q}^{-1}\delta p|\widehat{Q}\overline{Z}_{k_{1}\left(
t-2,0\right)  }^{k_{1}\left(  t-2,1\right)  }\right)  \right)  \right)
\end{array}
\right\} \\
&  -\left\{
\begin{array}
[c]{c}%
\widehat{E}\left(  \widehat{Q}\delta g\overline{Z}_{k_{1}\left(  1,0\right)
}^{k_{1}\left(  1,1\right)  T}\widehat{\Pi}\left(  \widehat{Q}^{-1}\delta
b|\left(  \widehat{Q}\overline{Z}_{k_{1}\left(  0,0\right)  }^{k_{1}\left(
0,1\right)  }\right)  \right)  \right)
%TCIMACRO{\tprod \limits_{u=2}^{t-3}}%
%BeginExpansion
{\textstyle\prod\limits_{u=2}^{t-3}}
%EndExpansion
\widehat{E}\left(  \delta g\widehat{Q}^{2}\overline{Z}_{k_{1}\left(
u-1,0\right)  }^{k_{1}\left(  u-1,1\right)  }\overline{Z}_{k_{1}\left(
u,0\right)  }^{k_{1}\left(  u,1\right)  T}\right) \\
\times\widehat{E}\left(  \delta g\widehat{Q}\overline{Z}_{k_{1}\left(
t-3,0\right)  }^{k_{1}\left(  t-3,1\right)  }\left(  \widehat{Q}^{-1}\delta
p-\widehat{\Pi}\left(  \widehat{Q}^{-1}\delta p|\left(  \widehat{Q}%
\overline{Z}_{k_{1}\left(  t-2,0\right)  }^{k_{1}\left(  t-2,1\right)
}\right)  \right)  \right)  \right)
\end{array}
\right\} \\
&  =\left\{
\begin{array}
[c]{c}%
\widehat{E}\left(  \widehat{Q}\delta g\overline{Z}_{k_{1}\left(  1,0\right)
}^{k_{1}\left(  1,1\right)  T}\widehat{\Pi}^{\bot}\left(  \widehat{Q}%
^{-1}\delta b|\left(  \widehat{Q}\overline{Z}_{k_{1}\left(  0,0\right)
}^{k_{1}\left(  0,1\right)  }\right)  \right)  \right)
%TCIMACRO{\tprod \limits_{u=2}^{t-3}}%
%BeginExpansion
{\textstyle\prod\limits_{u=2}^{t-3}}
%EndExpansion
\widehat{E}\left(  \delta g\widehat{Q}^{2}\overline{Z}_{k_{1}\left(
u-1,0\right)  }^{k_{1}\left(  u-1,1\right)  }\overline{Z}_{k_{1}\left(
u,0\right)  }^{k_{1}\left(  u,1\right)  T}\right) \\
\times\widehat{E}\left(  \delta g\widehat{Q}\overline{Z}_{k_{1}\left(
t-3,0\right)  }^{k_{1}\left(  t-3,1\right)  }\widehat{\Pi}^{\bot}\left(
\widehat{Q}^{-1}\delta p|\left(  \widehat{Q}\overline{Z}_{k_{1}\left(
t-2,0\right)  }^{k_{1}\left(  t-2,1\right)  }\right)  \right)  \right)
\end{array}
\right\} \\
&  =\left\{  \widehat{E}\left(  \delta g\left[  \left(
%TCIMACRO{\tprod \limits_{u=1}^{t-3}}%
%BeginExpansion
{\textstyle\prod\limits_{u=1}^{t-3}}
%EndExpansion
R_{u}\right)  \left(  \widehat{\Pi}^{\bot}\left(  \widehat{Q}^{-1}\delta
b|\left(  \widehat{Q}\overline{Z}_{k_{1}\left(  0,0\right)  }^{k_{1}\left(
0,1\right)  }\right)  \right)  \right)  \right]  \widehat{\Pi}^{\bot}\left(
\widehat{Q}^{-1}\delta p|\left(  \widehat{Q}\overline{Z}_{k_{1}\left(
t-2,0\right)  }^{k_{1}\left(  t-2,1\right)  }\right)  \right)  \right)
\right\}
\end{align*}
where$R_{u}\left(  H\right)  =\widehat{E}\left(  \delta g\widehat{Q}%
\overline{Z}_{k_{1}\left(  u,0\right)  }^{k_{1}\left(  u,1\right)  T}H\right)
\widehat{Q}\overline{Z}_{k_{1}\left(  u,0\right)  }^{k_{1}\left(  u,1\right)
}=\widehat{\Pi}\left(  \delta gH|\left(  \widehat{Q}\overline{Z}_{k_{1}\left(
u,0\right)  }^{k_{1}\left(  u,1\right)  }\right)  \right)  ,$ $\left(
%TCIMACRO{\tprod \limits_{u=1}^{s}}%
%BeginExpansion
{\textstyle\prod\limits_{u=1}^{s}}
%EndExpansion
R_{u}\right)  \left(  H\right)  =R_{s}\left[  \left(
%TCIMACRO{\tprod \limits_{u=1}^{s-1}}%
%BeginExpansion
{\textstyle\prod\limits_{u=1}^{s-1}}
%EndExpansion
R_{u}\right)  \left(  H\right)  \right]  ,$ $\widehat{\Pi}\left(  \cdot
|\cdot\right)  =\Pi_{\widehat{\theta}}\left(  \cdot|\cdot\right)  , $ and
$\widehat{\Pi}^{\bot}\left(  H|\Gamma\right)  =H-$ $\widehat{\Pi}\left(
H|\Gamma\right)  .$

Since projection operator has operator norm of 1, we have
\begin{align*}
&  \left\vert \chi\left(  t,k_{1};\theta\right)  \right\vert \\
&  \leq\left[
\begin{array}
[c]{c}%
\left\vert \left\vert \delta g\right\vert \right\vert _{\infty}^{t-2}\left\{
\widehat{E}\left(  \widehat{\Pi}^{\bot}\left(  \widehat{Q}^{-1}\delta
b|\left(  \widehat{Q}\overline{Z}_{0}^{k_{1}\left(  0,1\right)  }\right)
\right)  \right)  ^{2}\right\}  ^{1/2}\\
\times\left\{  \widehat{E}\left(  \widehat{\Pi}^{\bot}\left(  \widehat{Q}%
^{-1}\delta p|\left(  \widehat{Q}\overline{Z}_{0}^{k_{1}\left(  t-2,1\right)
}\right)  \right)  \right)  ^{2}\right\}  ^{1/2}%
\end{array}
\right] \\
&  =O_{p}\left(  \left(  \frac{\log n}{n}\right)  ^{-\frac{\left(  t-2\right)
\beta_{g}}{d+2\beta_{g}}}\left(  k_{1}\left(  0,1\right)  \right)
^{-\frac{\beta_{b}}{d}}\left(  k_{1}\left(  t-2,1\right)  \right)
^{-\frac{\beta p}{d}}\right)
\end{align*}

The last equality holds from theorem (\ref{TBrate}).
\end{proof}

Before proving theorem \ref{beyond4th}, let's first introduce a lemma which
will be used in the proof multiple times.

\begin{lemma}
\label{ON2}$\forall$ {}$2\times2$ matrix $k\equiv\left(  k\left(  l,s\right)
\right)  _{2\times2}$\smallskip, $l\in\left\{  1,2\right\}  ,$ $s\in\left\{
0,1\right\}  ,$ and $0<k\left(  l,0\right)  <k\left(  l,1\right)  $,%
\[
\widehat{E}\left[  \left(  \widehat{Q}\overline{Z}_{k\left(  l,0\right)
}^{k\left(  l,1\right)  }\right)  \left(  \widehat{Q}\overline{Z}_{k\left(
l,0\right)  }^{k\left(  l,1\right)  T}\right)  \right]  =I,\text{ }l=1,2
\]

Moreover%
\begin{align*}
&  \widehat{E}\left(  \delta b\overline{Z}_{k\left(  1,0\right)  }^{k\left(
1,1\right)  T}\right)  \widehat{E}\left(  \delta g\widehat{Q}^{2}\overline
{Z}_{k\left(  1,0\right)  }^{k\left(  1,1\right)  }\overline{Z}_{k\left(
2,0\right)  }^{k\left(  2,1\right)  T}\right)  \widehat{E}\left[  \delta
p\overline{Z}_{k\left(  2,0\right)  }^{k\left(  2,1\right)  }\right] \\
&  =\widehat{E}\left(  \delta g\widehat{\Pi}\left(  \widehat{Q}^{-1}\delta
b|\left(  \widehat{Q}\overline{Z}_{k\left(  1,0\right)  }^{k\left(
1,1\right)  }\right)  \right)  \widehat{\Pi}\left(  \widehat{Q}^{-1}\delta
p|\left(  \widehat{Q}\overline{Z}_{k\left(  2,0\right)  }^{k\left(
2,1\right)  }\right)  \right)  \right)
\end{align*}

and $\widehat{\Pi}\left(  H|\left(  \widehat{Q}\overline{Z}_{k\left(
1,0\right)  }^{k\left(  1,1\right)  }\right)  \right)  =\widehat{\Pi}\left(
H|\left(  \widehat{Q}\overline{Z}_{k\left(  1,0\right)  }^{c_{k}}\right)
\right)  +\widehat{\Pi}\left(  H|\left(  \widehat{Q}\overline{Z}_{c_{k}%
}^{k\left(  1,1\right)  }\right)  \right)  $

if $k\left(  1,0\right)  <c_{k}<k\left(  1,1\right)  .$
\end{lemma}

\begin{proof}
The orthonormality of $\left\{  \widehat{Q}Z_{l}:l=1,2,...\right\}  $ under
$E_{\widehat{\theta}}$ comes directly from definition. Hence $\widehat{\Pi
}\left(  \widehat{Q}^{-1}\delta b|\left(  \widehat{Q}\overline{Z}_{k\left(
1,0\right)  }^{k\left(  1,1\right)  }\right)  \right)  =\widehat{E}\left(
\widehat{Q}^{-1}\delta b\widehat{Q}\overline{Z}_{k\left(  1,0\right)
}^{k\left(  1,1\right)  T}\right)  \widehat{Q}\overline{Z}_{k\left(
1,0\right)  }^{k\left(  1,1\right)  }=\widehat{E}\left(  \delta b\overline
{Z}_{k\left(  1,0\right)  }^{k\left(  1,1\right)  T}\right)  \widehat{Q}%
\overline{Z}_{k\left(  1,0\right)  }^{k\left(  1,1\right)  }$. {}And
$lin\left(  \widehat{Q}\overline{Z}_{k\left(  1,0\right)  }^{k\left(
1,1\right)  }\right)  ,$ the linear space generated by $\widehat{Q}%
\overline{Z}_{k\left(  1,0\right)  }^{k\left(  1,1\right)  },$ ${}$equals
$\left(  \widehat{Q}\overline{Z}_{k\left(  1,0\right)  }^{c_{k}}\right)
\oplus\left(  \widehat{Q}\overline{Z}_{c_{k}}^{k\left(  1,1\right)  }\right)
.$
\end{proof}

\begin{proof}
(Theorem \ref{beyond4th}\smallskip) {}Let $\mathbb{H}_{v}^{\ast}%
=\mathbb{G}\left(  s,v\right)  =\mathbb{Q}_{v}\equiv0$ for $v>m\left(
\beta_{g},\beta_{b},\beta_{p}\right)  ,$ and $\mathbb{G}\left(  s,2\right)
=\mathbb{Q}_{2}\equiv0$. For $\forall$ $3<t,$ {}define
\begin{align*}
B_{t}^{\left(  H^{\ast}\right)  }  &  =\left(  -1\right)  ^{t-1}\left(
B_{t}\left(  \mathbb{H}_{t-2}^{\ast}\right)  -B_{t}^{bg}\left(  \mathbb{H}%
_{t-1}^{\ast}\right)  -B_{t}^{pg}\left(  \mathbb{H}_{t-1}^{\ast}\right)
+B_{t}\left(  \mathbb{H}_{t}^{\ast}\right)  \right) \\
B_{t}^{\left(  G\right)  }  &  =%
%TCIMACRO{\dsum \limits_{s=1}^{J}}%
%BeginExpansion
{\displaystyle\sum\limits_{s=1}^{J}}
%EndExpansion
B_{t}^{\left(  G,s\right)  }\\
B_{t}^{\left(  G,s\right)  }  &  =\left(  -1\right)  ^{t-1}\left(
\begin{array}
[c]{c}%
B_{t}\left(  \mathbb{G}\left(  s,t-2\right)  \right)  -B_{t}^{bg}\left(
\mathbb{G}\left(  s,t-1\right)  \right) \\
-B_{t}^{pg}\left(  \mathbb{G}\left(  s,t-1\right)  \right)  +B_{t}\left(
\mathbb{G}\left(  s,t\right)  \right)
\end{array}
\right) \\
B_{t}^{\left(  Q\right)  }  &  =\left(  -1\right)  ^{t-1}\left(  B_{t}\left(
\mathbb{Q}_{t-2}\right)  -B_{t}^{bg}\left(  \mathbb{Q}_{t-1}\right)
-B_{t}^{pg}\left(  \mathbb{Q}_{t-1}\right)  +B_{t}\left(  \mathbb{Q}%
_{t}\right)  \right)
\end{align*}
then
\begin{align*}
&  E\left(  \widehat{\psi}_{\mathcal{K}_{J}}^{eff}\left(  \beta_{g},\beta
_{b},\beta_{p}\right)  \right)  -\psi\left(  \theta\right) \\
&  =\left[  E_{\theta}\left(  H_{1}\widehat{B}\widehat{P}+H_{2}\widehat{B}%
+H_{3}\widehat{P}+H_{4}\right)  -\psi\left(  \theta\right)  \right]
_{-\left(  L4.1\right)  }\\
&  -B_{2}\left(  \mathbb{H}_{2}^{\ast}\right)  +B_{3}\left(  \mathbb{H}%
_{3}^{\ast}-\mathbb{H}_{2}^{\ast\ }\right)  +%
%TCIMACRO{\dsum \limits_{s=1}^{J}}%
%BeginExpansion
{\displaystyle\sum\limits_{s=1}^{J}}
%EndExpansion
B_{3}\left(  \mathbb{G}\left(  s,3\right)  \right)  +B_{3}\left(
\mathbb{Q}_{3}\right) \\
&  +\sum_{t=4}^{m\left(  \beta_{g},\beta_{b},\beta_{p}\right)  +2}\left(
B_{t}^{\left(  H^{\ast}\right)  }+B_{t}^{\left(  G\right)  }+B_{t}^{\left(
Q\right)  }\right)
\end{align*}%
\begin{gather*}
\left(  L4.1\right)  =E_{\theta}\left(  H_{1}\left(  B-\widehat{B}\right)
\left(  P-\widehat{P}\right)  \right) \\
=\widehat{E}\left[  \delta g\widehat{E}\left(  H_{1}|X\right)  \left(
B-\widehat{B}\right)  \left(  P-\widehat{P}\right)  \right] \\
+\widehat{E}\left(  \widehat{E}\left(  H_{1}|X\right)  \left(  B-\widehat{B}%
\right)  \left(  P-\widehat{P}\right)  \right) \\
=\widehat{E}\left(  \delta g\widehat{Q}^{-2}\delta b\delta p\right)
_{-\left(  L4.2\right)  }+\widehat{E}\left(  \widehat{Q}^{-2}\delta b\delta
p\right)  _{-\left(  L4.3\right)  }%
\end{gather*}

\begin{align*}
&  \left(  L4.3\right)  -B_{2}\left(  H_{2}^{\ast}\right) \\
&  =\widehat{E}\left(  \widehat{Q}^{-2}\delta b\delta p\right)  -\widehat{E}%
\left(  \delta b\overline{Z}_{0}^{k_{-1}T}\right)  \widehat{E}\left(  \delta
p\overline{Z}_{0}^{k_{-1}}\right) \\
&  =\widehat{E}\left(  \widehat{Q}^{-2}\delta b\delta p\right)  -\widehat{E}%
\left[  \widehat{\Pi}\left(  \widehat{Q}^{-1}\delta b|\widehat{Q}\overline
{Z}_{0}^{k_{-1}}\right)  \widehat{\Pi}\left(  \widehat{Q}^{-1}\delta
p|\widehat{Q}\overline{Z}_{0}^{k_{-1}}\right)  \right] \\
&  =\widehat{E}\left[  \widehat{\Pi}^{\bot}\left(  \widehat{Q}^{-1}\delta
b|\left(  \widehat{Q}\overline{Z}_{0}^{k_{-1}}\right)  \right)  \widehat{\Pi
}^{\bot}\left(  \widehat{Q}^{-1}\delta p|\left(  \widehat{Q}\overline{Z}%
_{0}^{k_{-1}}\right)  \right)  \right]
\end{align*}
Then
\[
\left\vert \left(  L4.3\right)  -B_{2}\left(  H_{2}^{\ast}\right)  \right\vert
=O_{p}\left(  k_{-1}^{2\beta/d}\right)
\]%
\begin{align*}
&  \left(  L4.2\right)  -B_{3}\left(  \mathbb{H}_{2}^{\ast}\right)
+B_{3}\left(  \mathbb{H}_{2}^{\ast}\right)  +%
%TCIMACRO{\dsum \limits_{s=1}^{J}}%
%BeginExpansion
{\displaystyle\sum\limits_{s=1}^{J}}
%EndExpansion
B_{3}\left(  \mathbb{G}\left(  s,3\right)  \right)  +B_{3}\left(
\mathbb{Q}_{3}\right) \\
&  =\left\{
\begin{array}
[c]{c}%
\widehat{E}\left(  \widehat{Q}^{-2}\delta b\delta p\delta g\right)
-\widehat{E}\left(  \delta b\delta g\overline{Z}_{0}^{k_{-1}T}\right)
\widehat{E}\left(  \delta p\overline{Z}_{0}^{k_{-1}}\right) \\
-\widehat{E}\left(  \delta b\overline{Z}_{0}^{k_{-1}T}\right)  \widehat{E}%
\left(  \delta p\delta g\overline{Z}_{0}^{k_{-1}}\right)
\end{array}
\right\}  _{-\left(  L5.1\right)  }\\
&  +\widehat{E}\left(  \delta b\overline{Z}_{0}^{k_{0}T}\right)
\widehat{E}\left(  \delta g\widehat{Q}^{2}\overline{Z}_{0}^{k_{0}}\overline
{Z}_{0}^{k_{0}T}\right)  \widehat{E}\left[  \delta p\overline{Z}_{0}^{k_{0}%
}\right] \\
&  +\widehat{E}\left(  \delta b\overline{Z}_{k_{0}}^{k_{-1}T}\right)
\widehat{E}\left(  \delta g\widehat{Q}^{2}\overline{Z}_{k_{0}}^{k_{-1}%
}\overline{Z}_{0}^{k_{0}T}\right)  \widehat{E}\left[  \delta p\overline{Z}%
_{0}^{k_{0}}\right] \\
&  +\widehat{E}\left(  \delta b\overline{Z}_{0}^{k_{0}T}\right)
\widehat{E}\left(  \delta g\widehat{Q}^{2}\overline{Z}_{0}^{k_{0}}\overline
{Z}_{k_{0}}^{k_{-1}T}\right)  \widehat{E}\left[  \delta p\overline{Z}_{k_{0}%
}^{k_{-1}}\right] \\
&  +\sum_{s=1}^{J}\left\{
\begin{array}
[c]{c}%
\widehat{E}\left(  \delta b\overline{Z}_{k_{2s-2}}^{k_{2s-1}T}\right)
\widehat{E}\left(  \delta g\widehat{Q}^{2}\overline{Z}_{k_{2s-2}}^{k_{2s-1}%
}\overline{Z}_{k_{2s-2}}^{k_{2s}T}\right)  \widehat{E}\left[  \delta
p\overline{Z}_{k_{2s-2}}^{k_{2s}}\right] \\
+\widehat{E}\left(  \delta b\overline{Z}_{k_{2s-2}}^{k_{2s}T}\right)
\widehat{E}\left(  \delta g\widehat{Q}^{2}\overline{Z}_{k_{2s-2}}^{k_{2s}%
}\overline{Z}_{k_{2s}}^{k_{2s-1}T}\right)  \widehat{E}\left[  \delta
p\overline{Z}_{k_{2s}}^{k_{2s-1}}\right]
\end{array}
\right\} \\
&  +\widehat{E}\left(  \delta b\overline{Z}_{k_{2J}}^{k_{2J+1}T}\right)
\widehat{E}\left(  \delta g\widehat{Q}^{2}\overline{Z}_{k_{2J}}^{k_{2J+1}%
}\overline{Z}_{k_{2J}}^{k_{2J+1}T}\right)  \widehat{E}\left[  \delta
p\overline{Z}_{k_{2J}}^{k_{2J+1}}\right]
\end{align*}
because
\begin{align*}
&  \left(  L5.1\right) \\
&  =\widehat{E}\left(
\begin{array}
[c]{c}%
\delta g\left\{  \widehat{\Pi}\left(  \widehat{Q}^{-1}\delta b|\left(
\widehat{Q}\overline{Z}_{0}^{k_{-1}}\right)  \right)  +\widehat{\Pi}^{\bot
}\left(  \widehat{Q}^{-1}\delta b|\left(  \widehat{Q}\overline{Z}_{0}^{k_{-1}%
}\right)  \right)  \right\} \\
\times\left\{  \widehat{\Pi}\left(  \widehat{Q}^{-1}\delta p|\left(
\widehat{Q}\overline{Z}_{0}^{k_{-1}}\right)  \right)  +\widehat{\Pi}^{\bot
}\left(  \widehat{Q}^{-1}\delta p|\left(  \widehat{Q}\overline{Z}_{0}^{k_{-1}%
}\right)  \right)  \right\}
\end{array}
\right) \\
&  -\widehat{E}\left(  \delta g\delta p\widehat{Q}^{-1}\widehat{\Pi}\left(
\widehat{Q}^{-1}\delta b|\left(  \widehat{Q}\overline{Z}_{0}^{k_{-1}}\right)
\right)  \right) \\
&  -\widehat{E}\left(  \delta g\delta b\widehat{Q}^{-1}\widehat{\Pi}\left(
\widehat{Q}^{-1}\delta p|\left(  \widehat{Q}\overline{Z}_{0}^{k_{-1}}\right)
\right)  \right) \\
&  =\widehat{E}\left(  \delta g\widehat{\Pi}^{\bot}\left(  \widehat{Q}%
^{-1}\delta b|\left(  \widehat{Q}\overline{Z}_{0}^{k_{-1}}\right)  \right)
\widehat{\Pi}^{\bot}\left(  \widehat{Q}^{-1}\delta p|\left(  \widehat{Q}%
\overline{Z}_{0}^{k_{-1}}\right)  \right)  \right) \\
&  -\widehat{E}\left(  \delta g\widehat{\Pi}\left(  \widehat{Q}^{-1}\delta
b|\left(  \widehat{Q}\overline{Z}_{0}^{k_{-1}}\right)  \right)  \widehat{\Pi
}\left(  \widehat{Q}^{-1}\delta p|\left(  \widehat{Q}\overline{Z}_{0}^{k_{-1}%
}\right)  \right)  \right)
\end{align*}
and by lemma \ref{ON2}, we now have
\begin{align*}
&  \left(  L4.2\right)  -B_{3}\left(  \mathbb{H}_{2}^{\ast}\right)
+B_{3}\left(  \mathbb{H}_{2}^{\ast}\right)  +%
%TCIMACRO{\dsum \limits_{s=1}^{J}}%
%BeginExpansion
{\displaystyle\sum\limits_{s=1}^{J}}
%EndExpansion
B_{3}\left(  \mathbb{G}\left(  s,3\right)  \right)  +B_{3}\left(
\mathbb{Q}_{3}\right) \\
&  =\widehat{E}\left(  \delta g\widehat{\Pi}^{\bot}\left(  \widehat{Q}%
^{-1}\delta b|\left(  \widehat{Q}\overline{Z}_{0}^{k_{-1}}\right)  \right)
\widehat{\Pi}^{\bot}\left(  \widehat{Q}^{-1}\delta p|\left(  \widehat{Q}%
\overline{Z}_{0}^{k_{-1}}\right)  \right)  \right)  _{-\left(  L5.2\right)
}\\
&  -%
%TCIMACRO{\dsum \limits_{s=0}^{J}}%
%BeginExpansion
{\displaystyle\sum\limits_{s=0}^{J}}
%EndExpansion%
\begin{array}
[c]{c}%
\widehat{E}\left[  \delta g\widehat{\Pi}\left(  \widehat{Q}^{-1}\delta
b|\left(  \widehat{Q}\overline{Z}_{k_{2s}}^{k_{2s-1}}\right)  \right)
\widehat{\Pi}\left(  \widehat{Q}^{-1}\delta p|\left(  \widehat{Q}\overline
{Z}_{k_{2s+1}}^{k_{2s-1}}\right)  \right)  \right]  _{-\left(  L5.3\right)
}\\
+\widehat{E}\left[  \delta g\widehat{\Pi}\left(  \widehat{Q}^{-1}\delta
b|\left(  \widehat{Q}\overline{Z}_{k_{2s+1}}^{k_{2s-1}}\right)  \right)
\widehat{\Pi}\left(  \widehat{Q}^{-1}\delta p|\left(  \widehat{Q}\overline
{Z}_{k_{2s}}^{k_{2s+1}}\right)  \right)  \right]  _{-\left(  L5.4\right)  }%
\end{array}
\end{align*}
Therefore it is easy to show from theorem \ref{TBrate} that
\begin{align*}
\left\vert \left(  L5.2\right)  \right\vert  &  =O_{p}\left(  \left(
\frac{\log n}{n}\right)  ^{-\frac{\beta_{g}}{d+2\beta_{g}}}k_{-1}^{2\beta
/d}\right) \\
\left\vert \left(  L5.3\right)  \right\vert  &  =O_{p}\left(  \left(
\frac{\log n}{n}\right)  ^{-\frac{\beta_{g}}{d+2\beta_{g}}}k_{2s}^{-\beta
_{b}/d}k_{2s+1}^{-\beta_{p}/d}\right) \\
\left\vert \left(  L5.4\right)  \right\vert  &  =O_{p}\left(  \left(
\frac{\log n}{n}\right)  ^{-\frac{\beta_{g}}{d+2\beta_{g}}}k_{2s+1}%
^{-\beta_{b}/d}k_{2s}^{-\beta_{p}/d}\right)
\end{align*}

$\forall$ {}$3<t\leq m\left(  \beta_{g},\beta_{b},\beta_{p}\right)  ,$%
\begin{align*}
&  \left(  -1\right)  ^{t-1}B_{t}^{\left(  H^{\ast}\right)  }\\
&  =\left(  B_{t}\left(  \mathbb{H}_{t-2}^{\ast}\right)  -B_{t}^{bg}\left(
\mathbb{H}_{t-1}^{\ast}\right)  -B_{t}^{pg}\left(  \mathbb{H}_{t-1}^{\ast
}\right)  +B_{t}\left(  \mathbb{H}_{t}^{\ast}\right)  \right) \\
&  =\left\{
\begin{array}
[c]{c}%
B_{t}\left(  \widehat{\mathbb{U}}_{t-2}\left(  _{0}^{k_{0}}\right)  \right)
-B_{t}^{bg}\left(  \widehat{\mathbb{U}}_{t-1}\left(  _{0}^{k_{0}}\right)
\right) \\
-\left[  B_{t}^{pg}\left(  \widehat{\mathbb{U}}_{t-1}\left(  _{0}^{k_{0}%
}\right)  \right)  -B_{t}\left(  \widehat{\mathbb{U}}_{t}\left(  _{0}^{k_{0}%
}\right)  \right)  \right]
\end{array}
\right\}  _{-\left(  L6.1\right)  }\\
&  +%
%TCIMACRO{\dsum \limits_{u=1}^{t-3}}%
%BeginExpansion
{\displaystyle\sum\limits_{u=1}^{t-3}}
%EndExpansion
\left\{
\begin{array}
[c]{c}%
\left[  B_{t}\left(  \widehat{\mathbb{U}}_{t-2}^{\left(  u\right)  }\left(
_{k_{0}}^{k_{-1}},_{0}^{k_{0}}\right)  \right)  -B_{t}^{bg}\left(
\widehat{\mathbb{U}}_{t-1}^{\left(  u\right)  }\left(  _{k_{0}}^{k_{-1}}%
,_{0}^{k_{0}}\right)  \right)  \right] \\
-\left[  B_{t}^{pg}\left(  \widehat{\mathbb{U}}_{t-1}^{\left(  u+1\right)
}\left(  _{k_{0}}^{k_{-1}},_{0}^{k_{0}}\right)  \right)  -B_{t}\left(
\widehat{\mathbb{U}}_{t}^{\left(  u+1\right)  }\left(  _{k_{0}}^{k_{-1}}%
,_{0}^{k_{0}}\right)  \right)  \right]
\end{array}
\right\}  _{-\left(  L6.2.u\right)  }\\
&  -B_{t}^{pg}\left(  \widehat{\mathbb{U}}_{t-1}^{\left(  1\right)  }\left(
_{k_{0}}^{k_{-1}},_{0}^{k_{0}}\right)  \right)  -B_{t}\left(
\widehat{\mathbb{U}}_{t}^{\left(  1\right)  }\left(  _{k_{0}}^{k_{-1}}%
,_{0}^{k_{0}}\right)  \right)  _{-\left(  L6.3\right)  }\\
&  -B_{t}^{bg}\left(  \widehat{\mathbb{U}}_{t-1}^{\left(  t-2\right)  }\left(
_{k_{0}}^{k_{-1}},_{0}^{k_{0}}\right)  \right)  -B_{t}\left(
\widehat{\mathbb{U}}_{t}^{\left(  t-1\right)  }\left(  _{k_{0}}^{k_{-1}}%
,_{0}^{k_{0}}\right)  \right)  _{-\left(  L6.4\right)  }%
\end{align*}
It can be shown that $\left(  L6.1\right)  =\chi\left(  t,k_{1};\theta\right)
$ with $k_{1}\left(  l,0\right)  \equiv0,$ $k_{1}\left(  l,1\right)  \equiv
k_{0},$ $0\leq l\leq t-2;$ $\left(  L6.2.u\right)  =\chi\left(  t,k_{1}%
;\theta\right)  $ with $k_{1}\left(  l,0\right)  =0,$ $k_{1}\left(
l,1\right)  =k_{0}$ for $l\neq u$ and $k_{1}\left(  u,0\right)  =k_{0},$
$k_{1}\left(  u,1\right)  =k_{-1}.$ And
\begin{align*}
&  \left\vert \left(  L6.3\right)  \right\vert \\
&  =\left\vert \left\{
\begin{array}
[c]{c}%
\widehat{E}\left(  \delta g\widehat{Q}\widehat{\Pi}\left(  \widehat{Q}%
^{-1}\delta b|\left(  \widehat{Q}\overline{Z}_{k_{0}}^{k_{-1}}\right)
\right)  \overline{Z}_{0}^{k_{0}T}\right)  \left(  \widehat{E}\left(  \delta
g\widehat{Q}^{2}\overline{Z}_{0}^{k_{0}}\overline{Z}_{0}^{k_{0}T}\right)
\right)  ^{t-4}\\
\times E\left(  \delta g\widehat{Q}\overline{Z}_{0}^{k_{0}}\widehat{\Pi}%
^{\bot}\left(  \widehat{Q}^{-1}\delta b|\left(  \widehat{Q}\overline{Z}%
_{0}^{k_{0}}\right)  \right)  \right)
\end{array}
\right\}  \right\vert \\
&  =O_{p}\left(  \left(  \frac{\log n}{n}\right)  ^{-\frac{\left(  t-2\right)
\beta_{g}}{d+2\beta_{g}}}k_{0}^{-2\beta/d}\right)
\end{align*}
Similarly
\[
\left\vert \left(  L6.4\right)  \right\vert =O_{p}\left(  \left(  \frac{\log
n}{n}\right)  ^{-\frac{\left(  t-2\right)  \beta_{g}}{d+2\beta_{g}}}%
k_{0}^{-2\beta/d}\right)
\]
From lemma \ref{cp}, we now have$\left\vert \left(  L6.1\right)  \right\vert
=O_{p}\left(  \left(  \frac{\log n}{n}\right)  ^{-\frac{\left(  t-2\right)
\beta_{g}}{d+2\beta_{g}}}k_{0}^{-2\beta/d}\right)  $ therefore
\[
\left\vert B_{t}^{\left(  H^{\ast}\right)  }\right\vert =O_{p}\left(  \left(
\frac{\log n}{n}\right)  ^{-\frac{\left(  t-2\right)  \beta_{g}}{d+2\beta_{g}%
}}k_{0}^{-2\beta/d}\right)
\]
$\forall$ $t\geq5,$ following the same argument as above, we know as well
\begin{align*}
\left\vert B_{t}^{\left(  G\right)  }\right\vert  &  =O_{p}\left(  \left(
\frac{\log n}{n}\right)  ^{-\frac{\left(  t-2\right)  \beta_{g}}{d+2\beta_{g}%
}}k_{0}^{-2\beta/d}\right) \\
\left\vert B_{t}^{\left(  Q\right)  }\right\vert  &  =O_{p}\left(  \left(
\frac{\log n}{n}\right)  ^{-\frac{\left(  t-2\right)  \beta_{g}}{d+2\beta_{g}%
}}k_{0}^{-2\beta/d}\right)
\end{align*}
$\forall$ {}$1\leq s\leq J$%
\begin{align*}
B_{4}^{\left(  G,s\right)  }  &  =B_{4}\left(  \widehat{\mathbb{U}}%
_{3}^{\left(  1,2\right)  }\left(  _{k_{2s-2}}^{k_{2s-1}},_{k_{2s-2}}^{k_{2s}%
},_{0}^{k_{0}}\right)  \right)  +B_{4}\left(  \widehat{\mathbb{U}}%
_{3}^{\left(  1,2\right)  }\left(  _{k_{2s-2}}^{k_{2s}},_{k_{2s}}^{k_{2s-1}%
},_{0}^{k_{0}}\right)  \right) \\
&  -B_{4}\left(  \widehat{\mathbb{U}}_{4}^{\left(  1,2\right)  }\left(
_{k_{2s-2}}^{k_{2s-1}},_{k_{2s-2}}^{k_{2s}},_{0}^{k_{0}}\right)
+\widehat{\mathbb{U}}_{4}^{\left(  2,3\right)  }\left(  _{k_{2s-2}}^{k_{2s-1}%
},_{k_{2s-2}}^{k_{2s}},_{0}^{k_{0}}\right)  \right) \\
&  -B_{4}\left(  \widehat{\mathbb{U}}_{4}^{\left(  1,2\right)  }\left(
_{k_{2s-2}}^{k_{2s}},_{k_{2s}}^{k_{2s-1}},_{0}^{k_{0}}\right)
+\widehat{\mathbb{U}}_{4}^{\left(  2,3\right)  }\left(  _{k_{2s-2}}^{k_{2s}%
},_{k_{2s}}^{k_{2s-1}},_{0}^{k_{0}}\right)  \right)
\end{align*}
From lemma \ref{ON2}
\begin{align*}
&  B_{4}\left(  \widehat{\mathbb{U}}_{3}^{\left(  1,2\right)  }\left(
_{k_{2s-2}}^{k_{2s-1}},_{k_{2s-2}}^{k_{2s}},_{0}^{k_{0}}\right)  \right)
-B_{4}\left(  \widehat{\mathbb{U}}_{4}^{\left(  1,2\right)  }\left(
_{k_{2s-2}}^{k_{2s-1}},_{k_{2s-2}}^{k_{2s}},_{0}^{k_{0}}\right)
+\widehat{\mathbb{U}}_{4}^{\left(  2,3\right)  }\left(  _{k_{2s-2}}^{k_{2s-1}%
},_{k_{2s-2}}^{k_{2s}},_{0}^{k_{0}}\right)  \right) \\
&  =\widehat{E}\left(  \delta b\delta g\overline{Z}_{k_{2s-2}}^{k_{2s-1}%
T}\right)  \widehat{E}\left(  \delta g\widehat{Q}^{2}\overline{Z}_{k_{2s-2}%
}^{k_{2s-1}}\overline{Z}_{k_{2s-2}}^{k_{2s}T}\right)  \widehat{E}\left(
\delta p\overline{Z}_{k_{2s-2}}^{k_{2s}}\right) \\
&  -\widehat{E}\left(  \delta b\overline{Z}_{0}^{k_{0}T}\right)
\widehat{E}\left(  \delta g\widehat{Q}^{2}\overline{Z}_{0}^{k_{0}}\overline
{Z}_{k_{2s-2}}^{k_{2s-1}T}\right)  \widehat{E}\left(  \delta g\widehat{Q}%
^{2}\overline{Z}_{k_{2s-2}}^{k_{2s-1}}\overline{Z}_{k_{2s-2}}^{k_{2s}%
T}\right)  \widehat{E}\left[  \delta p\overline{Z}_{k_{2s-2}}^{k_{2s}}\right]
\end{align*}%
\begin{align*}
&  +\widehat{E}\left(  \delta b\overline{Z}_{k_{2s-2}}^{k_{2s-1}T}\right)
\widehat{E}\left(  \delta g\widehat{Q}^{2}\overline{Z}_{k_{2s-2}}^{k_{2s-1}%
}\overline{Z}_{k_{2s-2}}^{k_{2s}T}\right)  \widehat{E}\left(  \delta p\delta
g\overline{Z}_{k_{2s-2}}^{k_{2s}}\right) \\
&  -\widehat{E}\left(  \delta b\overline{Z}_{k_{2s-2}}^{k_{2s-1}T}\right)
\widehat{E}\left(  \delta g\widehat{Q}^{2}\overline{Z}_{k_{2s-2}}^{k_{2s-1}%
}\overline{Z}_{k_{2s-2}}^{k_{2s}T}\right)  \widehat{E}\left(  \delta
g\widehat{Q}^{2}\overline{Z}_{k_{2s-2}}^{k_{2s}}\overline{Z}_{0}^{k_{0}%
T}\right)  \widehat{E}\left[  \delta p\overline{Z}_{0}^{k_{0}}\right] \\
&  =\widehat{E}\left(  \widehat{\Pi}^{\bot}\left(  \widehat{Q}^{-1}\delta
b|\left(  \widehat{Q}\overline{Z}_{0}^{k_{0}}\right)  \right)  \delta
g\overline{Z}_{k_{2s-2}}^{k_{2s-1}T}\right)  \widehat{E}\left(  \delta
g\widehat{Q}\overline{Z}_{k_{2s-2}}^{k_{2s-1}}\widehat{\Pi}\left(
\widehat{Q}^{-1}\delta p|\left(  \widehat{Q}\overline{Z}_{k_{2s-2}}^{k_{2s}%
}\right)  \right)  \right) \\
&  +\widehat{E}\left(  \widehat{\Pi}\left(  \widehat{Q}^{-1}\delta b|\left(
\widehat{Q}\overline{Z}_{k_{2s-2}}^{k_{2s-1}}\right)  \right)  \delta
g\overline{Z}_{k_{2s-2}}^{k_{2s}T}\right)  \widehat{E}\left(  \delta
g\widehat{Q}\overline{Z}_{k_{2s-2}}^{k_{2s}}\widehat{\Pi}^{\bot}\left(
\widehat{Q}^{-1}\delta p|\left(  \widehat{Q}\overline{Z}_{0}^{k_{0}}\right)
\right)  \right)
\end{align*}
Similarly,%
\begin{align*}
&  B_{4}\left(  \widehat{\mathbb{U}}_{3}^{\left(  1,2\right)  }\left(
_{k_{2s-2}}^{k_{2s}},_{k_{2s}}^{k_{2s-1}},_{0}^{k_{0}}\right)  \right) \\
&  -B_{4}\left(  \widehat{\mathbb{U}}_{4}^{\left(  1,2\right)  }\left(
_{k_{2s-2}}^{k_{2s}},_{k_{2s}}^{k_{2s-1}},_{0}^{k_{0}}\right)
+\widehat{\mathbb{U}}_{4}^{\left(  2,3\right)  }\left(  _{k_{2s-2}}^{k_{2s}%
},_{k_{2s}}^{k_{2s-1}},_{0}^{k_{0}}\right)  \right) \\
&  =\widehat{E}\left(  \widehat{\Pi}^{\bot}\left(  \widehat{Q}^{-1}\delta
b|\left(  \widehat{Q}\overline{Z}_{0}^{k_{0}}\right)  \right)  \delta
g\overline{Z}_{k_{2s-2}}^{k_{2s}T}\right)  \widehat{E}\left(  \delta
g\widehat{Q}\overline{Z}_{k_{2s-2}}^{k_{2s}}\widehat{\Pi}\left(
\widehat{Q}^{-1}\delta p|\left(  \widehat{Q}\overline{Z}_{k_{2s}}^{k_{2s-1}%
}\right)  \right)  \right) \\
&  +\widehat{E}\left(  \widehat{\Pi}\left(  \widehat{Q}^{-1}\delta b|\left(
\widehat{Q}\overline{Z}_{k_{2s-2}}^{k_{2s}}\right)  \right)  \delta
g\overline{Z}_{k_{2s}}^{k_{2s-1}T}\right)  \widehat{E}\left(  \delta
g\widehat{Q}\overline{Z}_{k_{2s}}^{k_{2s-1}}\widehat{\Pi}^{\bot}\left(
\widehat{Q}^{-1}\delta p|\left(  \widehat{Q}\overline{Z}_{0}^{k_{0}}\right)
\right)  \right)
\end{align*}
Therefore%
\[
\left\vert B_{4}^{\left(  G,s\right)  }\right\vert =O_{p}\left(  \left(
\frac{\log n}{n}\right)  ^{\frac{2\beta_{g}}{d+2\beta_{g}}}\max\left(
k_{2s-2}^{-\beta_{b}/d}k_{0}^{-\beta_{p}/d},k_{0}^{-\beta_{b}/d}%
k_{2s-2}^{-\beta_{p}/d}\right)  \right)
\]
Applying the above argument to $B_{4}^{\left(  Q\right)  },$ we can show
that:
\[
\left\vert B_{4}^{\left(  Q\right)  }\right\vert =O_{p}\left(  \left(
\frac{\log n}{n}\right)  ^{\frac{2\beta_{g}}{d+2\beta_{g}}}\max\left(
k_{2J}^{-\beta_{b}/d}k_{0}^{-\beta_{p}/d},k_{0}^{-\beta_{b}/d}k_{2J}%
^{-\beta_{p}/d}\right)  \right)
\]
In addition, {}$\forall$ {}$L\in\left\{  H^{\ast},G,Q\right\}  ,$ $t>m\left(
\beta_{g},\beta_{b},\beta_{p}\right)  ,$
\[
\left\vert B_{t}^{\left(  L\right)  }\right\vert =O_{p}\left(  \left\vert
\left\vert \delta g\right\vert \right\vert _{\infty}^{\left(  t-2\right)
}\left\vert \left\vert \delta b\right\vert \right\vert _{2}\left\vert
\left\vert \delta p\right\vert \right\vert _{2}\right)
\]
which completes the proof of bias. {}The order of variance follows directly
from theorem \ref{var_multi}.
\end{proof}

\begin{proof}
(Theorem \ref{tt} (iii)) As proved in (ii), $\ $%
\begin{equation}
\psi_{\tau}if_{1,\tau\left(  \cdot\right)  }\left(  o;\theta\right)
=-if_{1,\psi\left(  \tau^{\dagger},\cdot\right)  }\left(  o;\tau
,\theta\right)  .\label{if1_tau}%
\end{equation}
By part 5c) of Theorem \ref{eift}, consider any suitably smooth one
dimensional parametric submodel $\widetilde{\theta}\left(  \zeta\right)  $
with range containing $\theta$ and contained in $\Theta\left(  \tau^{\dagger
}\right)  $, and differentiate both sides of eq. (\ref{if1_tau}) wrt. $\zeta.$
Then,
\begin{align*}
&  \psi_{\tau}\frac{\partial if_{1,\tau\left(  \cdot\right)  }\left(
o;\theta\right)  }{\partial\zeta}+\left(  \psi_{\tau\tau}\frac{\partial\tau
}{\partial\zeta}+\frac{\partial\psi_{\tau}\left(  \tau^{\dagger}%
,\widetilde{\theta}\left(  \zeta\right)  \right)  }{\partial\zeta
}|_{\widetilde{\theta}\left(  \zeta\right)  =\theta}\right)  if_{1,\tau\left(
\cdot\right)  }\left(  o;\theta\right) \\
&  =-\frac{\partial if_{1,\psi\left(  \tau^{\dagger},\cdot\right)  }\left(
o;\tau,\theta\right)  }{\partial\tau}|_{\tau=\tau^{\dagger}}\frac{\partial
\tau}{\partial\zeta}-\frac{\partial if_{1,\psi\left(  \tau^{\dagger}%
,\cdot\right)  }\left(  o;\tau^{\dagger},\theta\right)  }{\partial\zeta}\\
&  =-\frac{\partial if_{1,\psi\left(  \tau^{\dagger},\cdot\right)  }\left(
o;\theta\right)  }{\partial\tau}E_{\theta}\left[  if_{1,\tau\left(
\cdot\right)  }\left(  O_{2};\theta\right)  S_{\zeta}\left(  \theta\right)
\right] \\
&  -E_{\theta}\left[  if_{1,if_{1,\psi}\left(  o;\cdot\right)  }\left(
O_{2};\theta\right)  S_{\zeta}\left(  \theta\right)  \right]  .
\end{align*}

Further, $\frac{\partial\tau}{\partial\zeta}=E_{\theta}\left[  if_{1,\tau
\left(  \cdot\right)  }\left(  O_{2},\theta\right)  S_{\zeta}\left(
\theta\right)  \right]  $ and
\begin{align*}
\frac{\partial\psi_{\tau}\left(  \tau^{\dagger},\widetilde{\theta}\left(
\zeta\right)  \right)  }{\partial\zeta}|_{\widetilde{\theta}\left(
\zeta\right)  =\theta}  &  =\frac{\partial E_{\theta}\left[  if_{1,\psi\left(
\tau,\cdot\right)  }\left(  O_{2};\theta\right)  S_{\zeta}\left(
\theta\right)  \right]  }{\partial\tau}|_{\tau=\tau^{\ast}}\\
&  =E_{\theta}\left[  \frac{\partial if_{1,\psi\left(  \tau^{\dagger}%
,\cdot\right)  }\left(  O_{2};\theta\right)  }{\partial\tau}S_{\zeta}\left(
\theta\right)  \right]
\end{align*}

Thus,
\begin{align*}
&  \psi_{\tau}\frac{\partial if_{1,\tau\left(  \cdot\right)  }\left(
o;\theta\right)  }{\partial\zeta}\\
&  =-\frac{\partial if_{1,\psi\left(  \tau^{\dagger},\cdot\right)  }\left(
o;\theta\right)  }{\partial\tau}E_{\theta}\left[  if_{1,\tau\left(
\cdot\right)  }\left(  O_{2};\theta\right)  S_{\zeta}\left(  \theta\right)
\right] \\
&  -E_{\theta}\left[  if_{1,if_{1,\psi}\left(  o;\cdot\right)  }\left(
O_{2},\theta\right)  S_{\zeta}\left(  \theta\right)  \right] \\
&  -if_{1,\tau\left(  \cdot\right)  }\left(  o;\theta\right)  \left(
\begin{array}
[c]{c}%
\psi_{\tau\tau}E_{\theta}\left[  if_{1,\tau\left(  \cdot\right)  }\left(
O_{2},\theta\right)  S_{\zeta}\left(  \theta\right)  \right] \\
+E_{\theta}\left[  \frac{\partial if_{1,\psi\left(  \tau^{\dagger}%
,\cdot\right)  }\left(  O_{2};\theta\right)  }{\partial\tau}S_{\zeta}\left(
\theta\right)  \right]
\end{array}
\right)
\end{align*}
for any $o\in\mathcal{O}$ \ wp. 1. \ That is, there exists a first order
influence function for $if_{1,\tau\left(  \cdot\right)  }\left(
o;\theta\right)  $, and
\begin{align*}
&  \mathbb{IF}_{2,2,\tau\left(  \cdot\right)  }\left(  \theta\right) \\
&  =\frac{1}{2}\mathbb{V}\left\{  d_{2,\theta}\left[  if_{1,if_{1,\tau\left(
\cdot\right)  }\left(  O_{1};\cdot\right)  }\left(  O_{2};\theta\right)
\right]  \right\} \\
&  =-\psi_{\tau}^{-1}\left\{
\begin{array}
[c]{c}%
\mathbb{IF}_{2,2,\tau\left(  \cdot\right)  }\left(  \theta\right)  +\frac
{1}{2}\psi_{\tau\tau}\mathbb{V}\left(  if_{1,\tau\left(  \cdot\right)
}\left(  O_{1};\theta\right)  if_{1,\tau\left(  \cdot\right)  }\left(
O_{2},\theta\right)  \right) \\
+\frac{1}{2}\mathbb{V}\left(
\begin{array}
[c]{c}%
if_{1,\tau\left(  \cdot\right)  }\left(  O_{1};\theta\right)  d_{1,\theta
}\left(  \frac{\partial if_{1,\psi\left(  \tau^{\dagger},\cdot\right)
}\left(  O_{2};\theta\right)  }{\partial\tau}\right) \\
+d_{1,\theta}\left(  \frac{\partial if_{1,\psi\left(  \tau^{\dagger}%
,\cdot\right)  }\left(  O_{1};\theta\right)  }{\partial\tau}\right)
if_{1,\tau\left(  \cdot\right)  }\left(  O_{2};\theta\right)
\end{array}
\right)
\end{array}
\right\}
\end{align*}
which completes the proof. \ \ Note that $d_{m,\theta}\left(  \cdot\right)  $
is defined in eq. (\ref{deg}).
\end{proof}

\begin{proof}
(Part ii and iii of Theorem \ref{aa})

(iii) The formulae for $\mathbb{IF}_{1,\tau\left(  \cdot\right)  }\left(
\widehat{\theta}\right)  =\mathbb{IF}_{1,\widetilde{\tau}_{k}\left(
\cdot\right)  }\left(  \widehat{\theta}\right)  $ and $\mathbb{IF}%
_{2,2,\widetilde{\tau}_{k}\left(  \cdot\right)  }\left(  \widehat{\theta
}\right)  $ follow from Theorem \ref{tt}. To verify the formula for
$Q_{2,2,\widetilde{\tau}_{k}\left(  \cdot\right)  }\left(  \widehat{\theta
}\right)  \ $, substitute $\psi_{\backslash\tau^{2}}\left(  \tau^{\dagger
},\widehat{\theta}\right)  =0\ $and $\partial IF_{1,\widetilde{\psi}%
_{k}\left(  \tau,\cdot\right)  ,i_{1}}\left(  \widehat{\theta}\right)
/\partial\tau=-\left\{  A-\widehat{p}\left(  X\right)  \right\}  _{i_{1}}%
^{2}\ $in the formula for $Q_{2,2,\widetilde{\tau}_{k}\left(  \cdot\right)
}\left(  \widehat{\theta}\right)  $ in Theorem \ref{tt}.

(ii): To obtain Eq (\ref{joel}), note by Theorem \ref{gg}, we have
\begin{align*}
&  var_{\widehat{\theta}}\left\{  \mathbb{ES}_{2,\widetilde{\tau}_{k}\left(
\cdot\right)  }^{test}\left(  \widehat{\theta}\left(  \tau^{\dagger}\right)
\right)  \right\}  ^{-1}\mathbb{ES}_{2,\widetilde{\tau}_{k}\left(
\cdot\right)  }^{test}\left(  \widehat{\theta}\left(  \tau^{\dagger}\right)
\right) \\
&  =E_{\widehat{\theta}\left(  \tau^{\dagger}\right)  }\left[  \mathbb{IF}%
_{2,\widetilde{\psi}_{k}\left(  \tau^{\dagger},\cdot\right)  }\left(
\widehat{\theta}\left(  \tau^{\dagger}\right)  \right)  \mathbb{ES}%
_{1,\widetilde{\tau}_{k}\left(  \cdot\right)  }^{test}\left(  \widehat{\theta
}\left(  \tau^{\dagger}\right)  \right)  \right]  ^{-1}\times\\
&  \Pi_{\widehat{\theta}\left(  \tau^{\dagger}\right)  }\left[  \mathbb{IF}%
_{2,\widetilde{\psi}_{k}\left(  \tau^{\dagger},\cdot\right)  }\left(
\widehat{\theta}\left(  \tau^{\dagger}\right)  \right)  |\Gamma_{2}%
^{test}\left(  \widehat{\theta}\left(  \tau^{\dagger}\right)  ,\tau^{\dagger
}\right)  \right]
\end{align*}
But by Theorem \ref{tt} and the definition of $\mathbb{ES}_{1}^{test},$ we
have
\begin{gather*}
\mathbb{ES}_{1,\widetilde{\tau}_{k}\left(  \cdot\right)  }^{test}\left(
\widehat{\theta}\left(  \tau^{\dagger}\right)  \right)  =\mathbb{ES}%
_{1,\tau\left(  \cdot\right)  }^{test}\left(  \widehat{\theta}\left(
\tau^{\dagger}\right)  \right) \\
=v\left(  \widehat{\theta}\left(  \tau^{\dagger}\right)  \right)
E_{\widehat{\theta}\left(  \tau^{\dagger}\right)  }\left[  \left\{  Y^{\ast
}\left(  \tau^{\dagger}\right)  -\widehat{b}\left(  X,\tau^{\dagger}\right)
\right\}  ^{2}\left\{  A-\widehat{p}\left(  X\right)  \right\}  ^{2}\right]
^{-1}\\
\times\left\{  Y^{\ast}\left(  \tau^{\dagger}\right)  -\widehat{b}\left(
X,\tau^{\dagger}\right)  \right\}  \left\{  A-\widehat{p}\left(  X\right)
\right\}
\end{gather*}
thus, we obtain $E_{\widehat{\theta}\left(  \tau^{\dagger}\right)  }\left[
\mathbb{IF}_{2,\widetilde{\psi}_{k}\left(  \tau^{\dagger},\cdot\right)
}\left(  \widehat{\theta}\left(  \tau^{\dagger}\right)  \right)
\mathbb{ES}_{1,\widetilde{\tau}_{k}\left(  \cdot\right)  }^{test}\left(
\widehat{\theta}\left(  \tau^{\dagger}\right)  \right)  \right]
^{-1}=v\left(  \widehat{\theta}\left(  \tau^{\dagger}\right)  \right)  ^{-1}$

Now
\begin{align*}
&  \Pi_{\widehat{\theta}}\left[  \mathbb{IF}_{2,\widetilde{\psi}_{k}\left(
\tau^{\dagger},\cdot\right)  }\left(  \widehat{\theta}\left(  \tau^{\dagger
}\right)  \right)  |\Gamma_{2}^{test}\left(  \widehat{\theta}\left(
\tau^{\dagger}\right)  ,\tau^{\dagger}\right)  \right] \\
&  =\mathbb{IF}_{2,\widetilde{\psi}_{k}\left(  \tau^{\dagger},\cdot\right)
}\left(  \widehat{\theta}\left(  \tau^{\dagger}\right)  \right)
-\Pi_{\widehat{\theta}}\left[  \mathbb{IF}_{2,\widetilde{\psi}_{k}\left(
\tau^{\dagger},\cdot\right)  }\left(  \widehat{\theta}\left(  \tau^{\dagger
}\right)  \right)  |\left\{  \mathbb{U}_{2,2,\widetilde{\tau}_{k}\left(
\cdot\right)  }^{test,\perp}\left(  \widehat{\theta}\left(  \tau^{\dagger
}\right)  ,\tau^{\dagger}\right)  \right\}  \right]
\end{align*}
Let $\widehat{\epsilon}$ denote $Y-\widehat{b}\left(  X\right)  ,$ and
$\widehat{\Delta}$ denote $A-\widehat{p}\left(  X\right)  .$ \ Next, we show
that
\[
\Pi_{\widehat{\theta}}\left[  \mathbb{IF}_{2,\widetilde{\psi}_{k}\left(
\tau^{\dagger},\cdot\right)  }\left(  \widehat{\theta}\left(  \tau^{\dagger
}\right)  \right)  |\left\{  \mathbb{U}_{2,2,\widetilde{\tau}_{k}\left(
\cdot\right)  }^{test,\perp}\left(  \widehat{\theta}\left(  \tau^{\dagger
}\right)  ,\tau^{\dagger}\right)  \right\}  \right]  =\mathbb{U}%
_{2,2,\widetilde{\tau}_{k}\left(  \cdot\right)  }^{\ast,test,\perp}\left(
\widehat{\theta}\left(  \tau^{\dagger}\right)  ,\tau^{\dagger}\right)
\]
where%
\begin{align*}
&  U_{2,2,\widetilde{\tau}_{k}\left(  \cdot\right)  ,ij}^{\ast,test,\perp
}\left(  \widehat{\theta}\left(  \tau^{\dagger}\right)  ,\tau^{\dagger}\right)
\\
&  =\left(  E_{\widehat{\theta}}\left[  \widehat{\epsilon}_{i}^{2}%
\widehat{\Delta}_{i}^{2}\right]  \right)  ^{-1}\times\widehat{\epsilon}%
_{i}\widehat{\Delta}_{i}\\
&  \left\{
\begin{array}
[c]{c}%
-\left\{
\begin{array}
[c]{c}%
\left(  E_{\widehat{\theta}}\left[  \widehat{\epsilon}_{i}^{2}\widehat{\Delta
}_{i}^{2}\right]  \right)  ^{-1}E_{\widehat{\theta}}\left[  \widehat{\epsilon
}\widehat{\Delta}^{2}\overline{Z}_{k}^{T}\right] \\
\times E_{\widehat{\theta}}\left[  \widehat{\epsilon}^{2}\widehat{\Delta
}\overline{Z}_{k}^{T}\right]  \widehat{\epsilon}_{j}\widehat{\Delta}_{j}%
\end{array}
\right\} \\
+E_{\widehat{\theta}}\left[  \widehat{\epsilon}_{i}^{2}\widehat{\Delta}%
_{i}\overline{Z}_{k,i}^{T}\right]  \overline{Z}_{k,j}\widehat{\Delta}_{j}\\
+E_{\widehat{\theta}}\left[  \widehat{\epsilon}_{i}\widehat{\Delta}_{i}%
^{2}\overline{Z}_{k,i}^{T}\right]  \overline{Z}_{k,j}\widehat{\epsilon}_{j}%
\end{array}
\right\}
\end{align*}
As proved in Theorem \ref{ff},%
\[
\left\{  \mathbb{U}_{2,2,\widetilde{\tau}_{k}\left(  \cdot\right)
}^{test,\perp}\left(  \theta\left(  \tau^{\dagger}\right)  ,\tau^{\dagger
}\right)  \right\}  =\left\{  \mathbb{V}\left\{  IF_{1,\tau\left(
\cdot\right)  ,i}^{eff}\left(  \theta\right)  h\left(  O_{j};\theta\right)
\right\}  :\text{ }\forall\text{ \ }E_{\theta}\left[  h\left(  O_{j}%
;\theta\right)  \right]  =0\right\}
\]

We assume that
\begin{align*}
&  \Pi_{\widehat{\theta}}\left[  \mathbb{IF}_{2,\widetilde{\psi}_{k}\left(
\tau^{\dagger},\cdot\right)  }\left(  \widehat{\theta}\left(  \tau^{\dagger
}\right)  \right)  |\left\{  \mathbb{U}_{2,2,\widetilde{\tau}_{k}\left(
\cdot\right)  }^{test,\perp}\left(  \widehat{\theta}\left(  \tau^{\dagger
}\right)  ,\tau^{\dagger}\right)  \right\}  \right] \\
&  =\mathbb{V}\left\{  IF_{1,\tau\left(  \cdot\right)  ,i}^{eff}\left(
\theta\right)  h^{\ast}\left(  O_{j};\theta\right)  \right\}  ,
\end{align*}
then by the definition of the projection, for any $h\left(  O_{j}%
;\theta\right)  $ such that $E_{\theta}\left[  h\left(  O_{j};\theta\right)
\right]  =0,$ we have
\begin{align*}
&  E_{\widehat{\theta}}\left[  \mathbb{IF}_{2,\widetilde{\psi}_{k}\left(
\tau^{\dagger},\cdot\right)  }\mathbb{V}\left\{  IF_{1,\tau\left(
\cdot\right)  ,i}^{eff}\left(  \theta\right)  h\left(  O_{j};\theta\right)
\right\}  \right] \\
&  =E_{\widehat{\theta}}\left[  \mathbb{V}\left\{  IF_{1,\tau\left(
\cdot\right)  ,i}^{eff}\left(  \theta\right)  h^{\ast}\left(  O_{j}%
;\theta\right)  \right\}  \mathbb{V}\left\{  IF_{1,\tau\left(  \cdot\right)
,i}^{eff}\left(  \theta\right)  h\left(  O_{j};\theta\right)  \right\}
\right]  ,
\end{align*}
which is equivalent to
\begin{align*}
&  v\left(  \widehat{\theta}\right)  ^{-1}\left\{  E_{\widehat{\theta}}\left[
\widehat{\epsilon}_{i}^{2}\widehat{\Delta}_{i}\overline{Z}_{k,i}^{T}%
\overline{Z}_{k,j}\widehat{\Delta}_{j}h\left(  O_{j}\right)  \right]
+E_{\widehat{\theta}}\left[  \widehat{\epsilon}_{i}\widehat{\Delta}_{i}%
^{2}\overline{Z}_{k,i}^{T}\overline{Z}_{k,j}\widehat{\epsilon}_{j}h\left(
O_{j}\right)  \right]  \right\} \\
&  =v\left(  \widehat{\theta}\right)  ^{-2}\left\{  E_{\widehat{\theta}%
}\left[  \widehat{\epsilon}_{i}^{2}\widehat{\Delta}_{i}^{2}h^{\ast}\left(
O_{j}\right)  h\left(  O_{j}\right)  \right]  +E_{\widehat{\theta}}\left[
\widehat{\epsilon}_{i}\widehat{\Delta}_{i}h\left(  O_{i}\right)
\widehat{\epsilon}_{j}\widehat{\Delta}_{j}h^{\ast}\left(  O_{j}\right)
\right]  \right\}  .
\end{align*}

As the equation above holds for any mean zero function $h\left(
O;\theta\right)  ,$ therefore%
\begin{align*}
&  \left\{  E_{\widehat{\theta}}\left[  \widehat{\epsilon}_{i}^{2}%
\widehat{\Delta}_{i}^{2}\right]  h^{\ast}\left(  O\right)  +\widehat{\epsilon
}\widehat{\Delta}E_{\widehat{\theta}}\left[  \widehat{\epsilon}_{j}%
\widehat{\Delta}_{j}h^{\ast}\left(  O_{j}\right)  \right]  \right\} \\
&  =v\left(  \widehat{\theta}\right)  \left\{  E_{\widehat{\theta}}\left[
\widehat{\epsilon}_{i}^{2}\widehat{\Delta}_{i}\overline{Z}_{k,i}^{T}\right]
\overline{Z}_{k}\widehat{\Delta}+E_{\widehat{\theta}}\left[  \widehat{\epsilon
}_{i}\widehat{\Delta}_{i}^{2}\overline{Z}_{k,i}^{T}\right]  \overline{Z}%
_{k}\widehat{\epsilon}\right\}
\end{align*}%
\[
\Leftrightarrow
\]%
\begin{align*}
&  h^{\ast}\left(  O\right) \\
&  =\left(  E_{\widehat{\theta}}\left[  \widehat{\epsilon}_{i}^{2}%
\widehat{\Delta}_{i}^{2}\right]  \right)  ^{-1}\left\{
\begin{array}
[c]{c}%
c_{h}\left(  \widehat{\theta}\right)  \widehat{\epsilon}\widehat{\Delta}\\
+v\left(  \widehat{\theta}\right)  E_{\widehat{\theta}}\left[
\widehat{\epsilon}_{i}^{2}\widehat{\Delta}_{i}\overline{Z}_{k,i}^{T}\right]
\overline{Z}_{k}\widehat{\Delta}\\
+v\left(  \widehat{\theta}\right)  E_{\widehat{\theta}}\left[
\widehat{\epsilon}_{i}\widehat{\Delta}_{i}^{2}\overline{Z}_{k,i}^{T}\right]
\overline{Z}_{k}\widehat{\epsilon}%
\end{array}
\right\}
\end{align*}
and $c_{h}\left(  \widehat{\theta}\right)  $ is determined by the following
equation
\begin{align*}
&  c_{h}\left(  \widehat{\theta}\right)  \widehat{\epsilon}\widehat{\Delta
}+v\left(  \widehat{\theta}\right)  E_{\widehat{\theta}}\left[
\widehat{\epsilon}_{i}^{2}\widehat{\Delta}_{i}\overline{Z}_{k,i}^{T}\right]
\overline{Z}_{k}\widehat{\Delta}\\
&  +v\left(  \widehat{\theta}\right)  E_{\widehat{\theta}}\left[
\widehat{\epsilon}_{i}\widehat{\Delta}_{i}^{2}\overline{Z}_{k,i}^{T}\right]
\overline{Z}_{k}\widehat{\epsilon}\\
&  +\widehat{\epsilon}\widehat{\Delta}E_{\widehat{\theta}}\left[
\widehat{\epsilon}\widehat{\Delta}\left(  E_{\widehat{\theta}}\left[
\widehat{\epsilon}_{i}^{2}\widehat{\Delta}_{i}^{2}\right]  \right)
^{-1}\left\{
\begin{array}
[c]{c}%
c_{h}\left(  \widehat{\theta}\right)  \widehat{\epsilon}\widehat{\Delta}\\
+v\left(  \widehat{\theta}\right)  E_{\widehat{\theta}}\left[
\widehat{\epsilon}_{i}^{2}\widehat{\Delta}_{i}\overline{Z}_{k,i}^{T}\right]
\overline{Z}_{k}\widehat{\Delta}\\
+v\left(  \widehat{\theta}\right)  E_{\widehat{\theta}}\left[
\widehat{\epsilon}_{i}\widehat{\Delta}_{i}^{2}\overline{Z}_{k,i}^{T}\right]
\overline{Z}_{k}\widehat{\epsilon}%
\end{array}
\right\}  \right] \\
&  =v\left(  \widehat{\theta}\right)  \left\{  E_{\widehat{\theta}}\left[
\widehat{\epsilon}_{i}^{2}\widehat{\Delta}_{i}\overline{Z}_{k,i}^{T}\right]
\overline{Z}_{k}\widehat{\Delta}+E_{\widehat{\theta}}\left[  \widehat{\epsilon
}_{i}\widehat{\Delta}_{i}^{2}\overline{Z}_{k,i}^{T}\right]  \overline{Z}%
_{k}\widehat{\epsilon}\right\}
\end{align*}%
\[
\Leftrightarrow
\]%
\[
\widehat{\epsilon}\widehat{\Delta}\left[
\begin{array}
[c]{c}%
c_{h}\left(  \widehat{\theta}\right)  +\left(  E_{\widehat{\theta}}\left[
\widehat{\epsilon}_{i}^{2}\widehat{\Delta}_{i}^{2}\right]  \right)
^{-1}\times\\
\left\{
\begin{array}
[c]{c}%
E_{\widehat{\theta}}\left[  \widehat{\epsilon}^{2}\widehat{\Delta}^{2}\right]
c_{h}\left(  \widehat{\theta}\right) \\
+2v\left(  \widehat{\theta}\right)  E_{\widehat{\theta}}\left[
\widehat{\epsilon}\widehat{\Delta}^{2}\overline{Z}_{k}^{T}\right]
E_{\widehat{\theta}}\left[  \widehat{\epsilon}^{2}\widehat{\Delta}\overline
{Z}_{k}^{T}\right]
\end{array}
\right\}
\end{array}
\right]  =0
\]%
\[
\Leftrightarrow
\]%
\begin{align*}
&  c_{h}\left(  \widehat{\theta}\right) \\
&  =-v\left(  \widehat{\theta}\right)  \left(  E_{\widehat{\theta}}\left[
\widehat{\epsilon}_{i}^{2}\widehat{\Delta}_{i}^{2}\right]  \right)
^{-1}E_{\widehat{\theta}}\left[  \widehat{\epsilon}\widehat{\Delta}%
^{2}\overline{Z}_{k}^{T}\right]  E_{\widehat{\theta}}\left[  \widehat{\epsilon
}^{2}\widehat{\Delta}\overline{Z}_{k}^{T}\right]
\end{align*}

In summary,%
\begin{align*}
&  \mathbb{U}_{2,2,\widetilde{\tau}_{k}\left(  \cdot\right)  }^{\ast
,test,\perp}\left(  \widehat{\theta}\left(  \tau^{\dagger}\right)
,\tau^{\dagger}\right) \\
&  =\left(  E_{\widehat{\theta}}\left[  \widehat{\epsilon}_{i}^{2}%
\widehat{\Delta}_{i}^{2}\right]  \right)  ^{-1}\times\\
&  \mathbb{V}\left\{
\begin{array}
[c]{c}%
\widehat{\epsilon}_{i}\widehat{\Delta}_{i}\times\\
\left\{
\begin{array}
[c]{c}%
-\left\{
\begin{array}
[c]{c}%
\left(  E_{\widehat{\theta}}\left[  \widehat{\epsilon}_{i}^{2}\widehat{\Delta
}_{i}^{2}\right]  \right)  ^{-1}E_{\widehat{\theta}}\left[  \widehat{\epsilon
}\widehat{\Delta}^{2}\overline{Z}_{k}^{T}\right] \\
\times E_{\widehat{\theta}}\left[  \widehat{\epsilon}^{2}\widehat{\Delta
}\overline{Z}_{k}^{T}\right]  \widehat{\epsilon}_{j}\widehat{\Delta}_{j}%
\end{array}
\right\} \\
+E_{\widehat{\theta}}\left[  \widehat{\epsilon}_{i}^{2}\widehat{\Delta}%
_{i}\overline{Z}_{k,i}^{T}\right]  \overline{Z}_{k,j}\widehat{\Delta}_{j}\\
+E_{\widehat{\theta}}\left[  \widehat{\epsilon}_{i}\widehat{\Delta}_{i}%
^{2}\overline{Z}_{k,i}^{T}\right]  \overline{Z}_{k,j}\widehat{\epsilon}_{j}%
\end{array}
\right\}
\end{array}
\right\}
\end{align*}

To obtain $\mathbb{ES}_{2,\widetilde{\tau}_{k}\left(  \cdot\right)  }%
^{test}\left(  \widehat{\theta}\left(  \tau^{\dagger}\right)  \right)  ,$ we
divide Eq. (\ref{joel}) by $var_{\widehat{\theta}}\left\{  \mathbb{ES}%
_{2,\widetilde{\tau}_{k}\left(  \cdot\right)  }^{test}\left(  \widehat{\theta
}\left(  \tau^{\dagger}\right)  \right)  \right\}  ^{-1}$. To obtain
$var_{\widehat{\theta}}\left\{  \mathbb{ES}_{2,\widetilde{\tau}_{k}\left(
\cdot\right)  }^{test}\left(  \widehat{\theta}\left(  \tau^{\dagger}\right)
\right)  \right\}  ^{-1}$, we take the variance of both sides of Eq.
(\ref{joel}) under law $\widehat{\theta}\left(  \tau^{\dagger}\right)  $
giving
\begin{align*}
&  var_{\widehat{\theta}\left(  \tau^{\dagger}\right)  }\left\{
\mathbb{ES}_{2,\widetilde{\tau}_{k}\left(  \cdot\right)  }^{test}\left(
\widehat{\theta}\left(  \tau^{\dagger}\right)  \right)  \right\}  ^{-1}\\
&  =v\left(  \widehat{\theta}\left(  \tau^{\dagger}\right)  \right)
^{-2}var_{\widehat{\theta}\left(  \tau^{\dagger}\right)  }\left[
\Pi_{\widehat{\theta}\left(  \tau^{\dagger}\right)  }\left[  \mathbb{IF}%
_{2,\widetilde{\psi}_{k}\left(  \tau^{\dagger},\cdot\right)  }\left(
\widehat{\theta}\left(  \tau^{\dagger}\right)  \right)  |\Gamma_{2}%
^{test}\left(  \widehat{\theta}\left(  \tau^{\dagger}\right)  ,\tau^{\dagger
}\right)  \right]  \right] \\
&  =v\left(  \widehat{\theta}\left(  \tau^{\dagger}\right)  \right)
^{-2}\left\{  var_{\widehat{\theta}\left(  \tau^{\dagger}\right)  }\left[
\mathbb{IF}_{2,\widetilde{\psi}_{k}\left(  \tau^{\dagger},\cdot\right)
}\left(  \widehat{\theta}\left(  \tau^{\dagger}\right)  \right)  \right]
-var\left[  \mathbb{U}_{2,2,\widetilde{\tau}_{k}\left(  \cdot\right)  }%
^{\ast,test,\perp}\left(  \widehat{\theta}\left(  \tau^{\dagger}\right)
,\tau^{\dagger}\right)  \right]  \right\}
\end{align*}

\end{proof}

\begin{proof}
(Theorem \ref{bb})\textbf{\ }except part (iii) which was proved in Theorem
\ref{EBrate}.

Parts (i) and (ii):\ That $var_{\theta}\left[  \mathbb{U}_{2,2,\widetilde{\tau
}_{k}\left(  \cdot\right)  }^{\ast test,\perp}\left(  \widehat{\theta}\left(
\tau^{\dagger}\right)  ,\tau^{\dagger}\right)  \right]  =o_{P}\left(
1/n\right)  \ \ $and $var_{\theta}\left[  \mathbb{Q}_{2,2,\widetilde{\tau}%
_{k}\left(  \cdot\right)  }\ \left(  \widehat{\theta}\right)  \right]
=o_{P}\left(  1/n\right)  \ $is a straightforward calculation. The remainder
of $\left(  i\right)  $ and $\left(  ii\right)  $ follows from the fact that
$var_{\theta}\left[  \psi_{2,k\ }\left(  \tau,\widehat{\theta}\right)
\right]  \asymp\max\left(  \frac{1}{n},\frac{k}{n^{2}}\right)  .$

Part (iv): By part (ii) of Theorem \ref{aa}, it is sufficient to show that
\begin{align*}
&  E_{\theta}\left[  U_{2,2,\widetilde{\tau}_{k}\left(  \cdot\right)  }^{\ast
test,\perp}\left(  \widehat{\theta}\left(  \tau^{\dagger}\right)
,\tau^{\dagger}\right)  \right] \\
&  =O_{p}\left\{  \left(  P-\widehat{P}\right)  \left(  B\left(  \tau
^{\dagger}\right)  -\widehat{B}\left(  \tau^{\dagger}\right)  \right)  \left[
\left(  P-\widehat{P}\right)  +\left(  B\left(  \tau^{\dagger}\right)
-\widehat{B}\left(  \tau^{\dagger}\right)  \right)  \right]  \right\}
\end{align*}
Below we show
\begin{align*}
&  E_{\theta}\left[  \mathbb{U}_{2,2,\widetilde{\tau}_{k}\left(  \cdot\right)
}^{\ast,test,\perp}\left(  \widehat{\theta}\left(  \tau^{\dagger}\right)
,\tau^{\dagger}\right)  \right] \\
&  =\left(  E_{\widehat{\theta}}\left[  \widehat{\epsilon}_{i}^{2}%
\widehat{\Delta}_{i}^{2}\right]  \right)  ^{-1}\times\\
&  \left\{
\begin{array}
[c]{c}%
E_{\theta}\left[  \left(  B\left(  \tau^{\dagger}\right)  -\widehat{B}\left(
\tau^{\dagger}\right)  \right)  \left(  P-\widehat{P}\right)  \right]
\times\\
\left\{
\begin{array}
[c]{c}%
E_{\theta}\left[  \left(  P-\widehat{P}\right)  \overline{Z}_{k}^{T}\right]
E_{\widehat{\theta}}\left[  \widehat{\epsilon}_{i}^{2}\widehat{\Delta}%
_{i}\overline{Z}_{k,i}^{T}\right] \\
+E_{\theta}\left[  \left(  B\left(  \tau^{\dagger}\right)  -\widehat{B}\left(
\tau^{\dagger}\right)  \right)  \overline{Z}_{k}^{T}\right]
E_{\widehat{\theta}}\left[  \widehat{\epsilon}_{i}\widehat{\Delta}_{i}%
^{2}\overline{Z}_{k,i}^{T}\right] \\
-\left\{
\begin{array}
[c]{c}%
\left(  E_{\widehat{\theta}}\left[  \widehat{\epsilon}_{i}^{2}\widehat{\Delta
}_{i}^{2}\right]  \right)  ^{-1}E_{\widehat{\theta}}\left[  \widehat{\epsilon
}\widehat{\Delta}^{2}\overline{Z}_{k}^{T}\right] \\
\times E_{\widehat{\theta}}\left[  \widehat{\epsilon}^{2}\widehat{\Delta
}\overline{Z}_{k}^{T}\right]  \times\\
E_{\theta}\left[  \left(  B\left(  \tau^{\dagger}\right)  -\widehat{B}\left(
\tau^{\dagger}\right)  \right)  \left(  P-\widehat{P}\right)  \right]
\end{array}
\right\}
\end{array}
\right\}
\end{array}
\right\}
\end{align*}
which is
\[
O_{p}\left\{  \left(  P-\widehat{P}\right)  \left(  B\left(  \tau^{\dagger
}\right)  -\widehat{B}\left(  \tau^{\dagger}\right)  \right)  \left[  \left(
P-\widehat{P}\right)  +\left(  B\left(  \tau^{\dagger}\right)  -\widehat{B}%
\left(  \tau^{\dagger}\right)  \right)  \right]  \right\}
\]
when, as is the case under our assumptions $E_{\widehat{\theta}\left(
\tau^{\dagger}\right)  }\left[  \left\{  Y\left(  \tau^{\dagger}\right)
-\widehat{B}\left(  \tau^{\dagger}\right)  \right\}  \left\{  A-\widehat{P}%
\right\}  ^{2}\overline{Z}_{k}^{T}\right]  $ and $E_{\widehat{\theta}\left(
\tau^{\dagger}\right)  }\left[  \left\{  Y\left(  \tau^{\dagger}\right)
-\widehat{B}\left(  \tau^{\dagger}\right)  \right\}  \left\{  A-\widehat{P}%
\right\}  ^{2}\overline{Z}_{k}^{T}\right]  $ are both order $O_{p}\left(
1\right)  ,$ but would be
\[
O_{p}\left\{  \left(  P-\widehat{P}\right)  \left(  B\left(  \tau^{\dagger
}\right)  -\widehat{B}\left(  \tau^{\dagger}\right)  \right)  \left[  \left(
P-\widehat{P}\right)  ^{2}+\left(  B\left(  \tau^{\dagger}\right)
-\widehat{B}\left(  \tau^{\dagger}\right)  \right)  ^{2}\right]  \right\}
\]
in the (unlikely) special case in which the semiparametric regression model
was precisely true, since then
\[
E_{\widehat{\theta}}\left[  Y\left(  \tau^{\dagger}\right)  A|X\right]
/\left\{  E_{\widehat{\theta}}\left[  Y\left(  \tau^{\dagger}\right)
|X\right]  E_{\widehat{\theta}}\left[  A|X\right]  \right\}  =1+o_{p}\left(
1\right)
\]
so
\[
E_{\widehat{\theta}\left(  \tau^{\dagger}\right)  }\left[  \left\{  Y\left(
\tau^{\dagger}\right)  -\widehat{B}\left(  \tau^{\dagger}\right)  \right\}
^{2}\left\{  A-\widehat{P}\right\}  \overline{Z}_{k}^{T}\right]  =O_{p}\left(
P-\widehat{P}\right)
\]
and
\[
E_{\widehat{\theta}\left(  \tau^{\dagger}\right)  }\left[  \left\{  Y\left(
\tau^{\dagger}\right)  -\widehat{B}\left(  \tau^{\dagger}\right)  \right\}
\left\{  A-\widehat{P}\right\}  ^{2}\overline{Z}_{k}^{T}\right]  =O_{p}\left(
B\left(  \tau^{\dagger}\right)  -\widehat{B}\right)  .
\]
The expression for $E_{\theta}\left[  \mathbb{U}_{2,2,\widetilde{\tau}%
_{k}\left(  \cdot\right)  }^{\ast,test,\perp}\left(  \widehat{\theta}\left(
\tau^{\dagger}\right)  ,\tau^{\dagger}\right)  \right]  $ is obtained from the
formula for $U_{2,2,\widetilde{\tau}_{k}\left(  \cdot\right)  }^{\ast
test,\perp}\left(  \widehat{\theta}\left(  \tau^{\dagger}\right)
,\tau^{\dagger}\right)  $ in Theorem \ref{aa} by noting that, because
$\mathbb{V}\left[  \left(  Y\left(  \tau^{\dagger}\right)  -\widehat{B}\left(
\tau^{\dagger}\right)  \right)  \left(  A-\widehat{P}\right)  \right]
=\psi_{1,k}\left(  \tau^{\dagger},\widehat{\theta}\left(  \tau^{\dagger
}\right)  \right)  $ and
\begin{gather*}
\widetilde{\psi}_{k}\left(  \tau^{\dagger},\theta\right)  =0,\\
E_{\theta}\left[  \left(  Y\left(  \tau^{\dagger}\right)  -\widehat{B}\left(
\tau^{\dagger}\right)  \right)  \left(  A-\widehat{P}\right)  \right] \\
=E_{\theta}\left[  \psi_{1,k}\left(  \tau^{\dagger},\widehat{\theta}\left(
\tau^{\dagger}\right)  \right)  -\widetilde{\psi}_{k}\left(  \tau^{\dagger
},\theta\right)  \right] \\
=E_{\theta}\left[  \left(  P-\widehat{P}\right)  \overline{Z}_{k}^{T}\right]
E_{\theta}\left[  \overline{Z}_{k}\overline{Z}_{k}^{T}\right]  E_{\theta
}\left[  \left(  B\left(  \tau^{\dagger}\right)  -\widehat{B}\left(
\tau^{\dagger}\right)  \right)  \overline{Z}_{k}\right]
\end{gather*}
by Theorem \ref{EBrate}.

Part (v): We first note that by Theorem \ref{eiet}, $\frac{E_{\theta}\left[
\tau_{2,k\ }\left(  \widehat{\theta}\right)  -\widetilde{\tau}_{k}\left(
\theta\right)  \right]  }{E_{\theta}\left[  \mathbb{IF}_{3,3,\widetilde{\tau
}_{k}\left(  \cdot\right)  }\left(  \widehat{\theta}\right)  \right]
}=-\left(  1+o_{p}\left(  1\right)  \right)  .$ {}It can be shown that
\begin{gather*}
\mathbb{IF}_{3,3,\widetilde{\tau}_{k}\left(  \cdot\right)  }\left(
\theta\right)  =\left(  -\psi_{\tau}\right)  ^{-1}\times\\
\mathbb{V}\left\{
\begin{array}
[c]{c}%
IF_{3,3,\psi\left(  \tau^{\dag},\cdot\right)  ,i_{1}i_{2}i_{3}}\left(
\theta\right)  +\frac{1}{6}\psi_{\backslash\tau^{3}}IF_{1,\tau\left(
\cdot\right)  ,i_{1}}\left(  \theta\right)  IF_{1,\tau\left(  \cdot\right)
,i_{2}}\left(  \theta\right)  IF_{1,\tau\left(  \cdot\right)  ,i_{3}}\left(
\theta\right) \\
+\frac{1}{3}\psi_{\backslash\tau^{2}}\left(
\begin{array}
[c]{c}%
IF_{1,\tau\left(  \cdot\right)  ,i_{1}}\left(  \theta\right)  IF_{2,2,\tau
\left(  \cdot\right)  ,i_{2}i_{3}}\left(  \theta\right) \\
+IF_{1,\tau\left(  \cdot\right)  ,i_{2}}\left(  \theta\right)  IF_{2,2,\tau
\left(  \cdot\right)  ,i_{1}i_{3}}\left(  \theta\right) \\
+IF_{1,\tau\left(  \cdot\right)  ,i_{3}}\left(  \theta\right)  IF_{2,2,\tau
\left(  \cdot\right)  ,i_{1}i_{2}}\left(  \theta\right)
\end{array}
\right) \\
+\frac{1}{3}\left(
\begin{array}
[c]{c}%
d_{1,\theta}\left(  \frac{\partial IF_{1,\psi\left(  \tau^{\dag},\cdot\right)
,i_{1}}\left(  \theta\right)  }{\partial\tau}\right)  IF_{2,2,\tau\left(
\cdot\right)  ,i_{2}i_{3}}\left(  \theta\right) \\
+d_{1,\theta}\left(  \frac{\partial IF_{1,\psi\left(  \tau^{\dag}%
,\cdot\right)  ,i_{2}}\left(  \theta\right)  }{\partial\tau}\right)
IF_{2,2,\tau\left(  \cdot\right)  ,i_{1}i_{3}}\left(  \theta\right) \\
+d_{1,\theta}\left(  \frac{\partial IF_{1,\psi\left(  \tau^{\dag}%
,\cdot\right)  ,i_{3}}\left(  \theta\right)  }{\partial\tau}\right)
IF_{2,2,\tau\left(  \cdot\right)  ,i_{1}i_{2}}\left(  \theta\right)
\end{array}
\right) \\
+\frac{1}{3}\left(
\begin{array}
[c]{c}%
IF_{1,\tau\left(  \cdot\right)  ,i_{1}}\left(  \theta\right)  d_{2,\theta
}\left(  \frac{\partial IF_{2,2,\psi\left(  \tau^{\dag},\cdot\right)
i_{2}i_{3}}\left(  \theta\right)  }{\partial\tau}\right) \\
+IF_{1,\tau\left(  \cdot\right)  ,i_{2}}\left(  \theta\right)  d_{2,\theta
}\left(  \frac{\partial IF_{2,2,\psi\left(  \tau^{\dag},\cdot\right)
i_{1}i_{3}}\left(  \theta\right)  }{\partial\tau}\right) \\
+IF_{1,\tau\left(  \cdot\right)  ,i_{3}}\left(  \theta\right)  d_{2,\theta
}\left(  \frac{\partial IF_{2,2,\psi\left(  \tau^{\dag},\cdot\right)
i_{1}i_{2}}\left(  \theta\right)  }{\partial\tau}\right)
\end{array}
\right) \\
+\frac{1}{6}\left(
\begin{array}
[c]{c}%
d_{1,\theta}\left(  \frac{\partial^{2}IF_{1,\psi\left(  \tau^{\dag}%
,\cdot\right)  ,i_{1}}\left(  \theta\right)  }{\partial\tau^{2}}\right)
IF_{1,\tau\left(  \cdot\right)  ,i_{2}}\left(  \theta\right)  IF_{1,\tau
\left(  \cdot\right)  ,i_{3}}\left(  \theta\right) \\
+d_{1,\theta}\left(  \frac{\partial^{2}IF_{1,\psi\left(  \tau^{\dag}%
,\cdot\right)  ,i_{2}}\left(  \theta\right)  }{\partial\tau^{2}}\right)
IF_{1,\tau\left(  \cdot\right)  ,i_{1}}\left(  \theta\right)  IF_{1,\tau
\left(  \cdot\right)  ,i_{3}}\left(  \theta\right) \\
+d_{1,\theta}\left(  \frac{\partial^{2}IF_{1,\psi\left(  \tau^{\dag}%
,\cdot\right)  ,i_{3}}\left(  \theta\right)  }{\partial\tau^{2}}\right)
IF_{1,\tau\left(  \cdot\right)  ,i_{1}}\left(  \theta\right)  IF_{1,\tau
\left(  \cdot\right)  ,i_{2}}\left(  \theta\right)
\end{array}
\right)
\end{array}
\right\}
\end{gather*}

From the fact that $\partial^{3}\widetilde{\psi}_{k}\left(  \tau
,\theta\right)  /\partial\tau^{3}=\partial^{2}\widetilde{\psi}_{k}\left(
\tau,\theta\right)  /\partial\tau^{2}=\partial^{2}IF_{1,\widetilde{\psi}%
_{k}\left(  \tau,\cdot\right)  }\left(  \widehat{\theta}\right)  /\partial
\tau^{2}\mathbb{=}0,$ we conclude that the order of $E_{\theta}\left[
\tau_{2,k\ }\left(  \widehat{\theta}\right)  -\widetilde{\tau}_{k}\left(
\theta\right)  \right]  $ is equal to the order of%
\begin{align*}
&  E_{\theta}\left[  IF_{3,3,\widetilde{\psi}_{k}\left(  \tau,\cdot\right)
}\left(  \widehat{\theta}\right)  \ \right]  +E_{\theta}\left[  d_{1,\theta
}\left(  \partial IF_{1,\widetilde{\psi}_{k}\left(  \tau,\cdot\right)
}\left(  \widehat{\theta}\right)  /\partial\tau\right)  \right]  E_{\theta
}\left[  IF_{2,2,\widetilde{\tau}_{k}\left(  \cdot\right)  }\left(
\widehat{\theta}\right)  \right] \\
&  +E_{\theta}\left[  d_{2,\theta}\left(  \partial IF_{2,2,\widetilde{\psi
}_{k}\left(  \tau,\cdot\right)  }\left(  \widehat{\theta}\right)
/\partial\tau\right)  \right]  E_{\theta}\left[  IF_{1,\widetilde{\tau}%
_{k}\left(  \cdot\right)  }\left(  \widehat{\theta}\right)  \right]
\end{align*}
Now
\begin{gather*}
E_{\theta}\left[  IF_{3,3,\widetilde{\psi}_{k}\left(  \tau^{\dagger}%
,\cdot\right)  }\left(  \widehat{\theta}\right)  \ \right]  =O_{p}\left[
\left(  P-\widehat{P}\right)  \left(  B\ -\widehat{B}\ \right)  \left(
\frac{g\left(  X\right)  }{\widehat{g}\left(  X\right)  }-1\right)  \right]
,\\
E_{\theta}\left[  d_{1,\theta}\left(  \partial IF_{1,\widetilde{\psi}%
_{k}\left(  \tau,\cdot\right)  }\left(  \widehat{\theta}\right)  /\partial
\tau\right)  \right]  =E_{\theta}\left[  \left(  A-\widehat{P}\right)
^{2}\right]  -E_{\widehat{\theta}}\left[  \left(  A-\widehat{P}\right)
^{2}\right] \\
=E_{\widehat{\theta}}\left[  \left(  \frac{f\left(  A|X\right)  }%
{\widehat{f}\left(  A|X\right)  }\frac{g\left(  X\right)  }{\widehat{g}\left(
X\right)  }-1\right)  \left(  A-\widehat{P}\right)  ^{2}\right] \\
=E_{\widehat{\theta}}\left[  \left(  \left[  \left(  \frac{f\left(
A|X\right)  }{\widehat{f}\left(  A|X\right)  }-1\right)  \frac{g\left(
X\right)  }{\widehat{g}\left(  X\right)  }\right]  \right)  \left(
A-\widehat{P}\right)  ^{2}\right] \\
+E_{\widehat{\theta}}\left[  \left(  \frac{g\left(  X\right)  }{\widehat{g}%
\left(  X\right)  }-1\right)  \left(  A-\widehat{P}\right)  ^{2}\right]
=O_{p}\left[  \left(  P-\widehat{P}\right)  +\left(  \frac{g\left(  X\right)
}{\widehat{g}\left(  X\right)  }-1\right)  \right]
\end{gather*}
by $A$ binary,
\begin{gather*}
E_{\theta}\left[  \partial IF_{2,2,\widetilde{\psi}_{k}\left(  \tau
,\cdot\right)  }\left(  \widehat{\theta}\right)  /\partial\tau\right]
=-\left\{  E_{\theta}\left[  \left(  A-\widehat{P}\right)  \right]  \right\}
^{2}=O_{p}\left[  \left(  P-\widehat{P}\right)  ^{2}\right] \\
E_{\theta}\left[  IF_{1,\widetilde{\tau}_{k}\left(  \cdot\right)  }\left(
\widehat{\theta}\right)  \right]  =E_{\theta}\left[  \left(  A-\widehat{P}%
\right)  \left(  Y\ -\widehat{B}\ \right)  \right] \\
=E_{\widehat{\theta}}\left[  \left(  P-\widehat{P}\right)  \left(
B\ -\widehat{B}\ \right)  \right]
\end{gather*}
and using $\partial^{2}\widetilde{\psi}_{k}\left(  \tau,\theta\right)
/\partial\tau^{2}=0$ and the explicit expression for $\mathbb{IF}%
_{2,2,\widetilde{\tau}_{k}\left(  \cdot\right)  }\left(  \theta\right)  $ in
Eq (\ref{q22})
\begin{align*}
&  E_{\theta}\left[  IF_{2,2,\widetilde{\tau}_{k}\left(  \cdot\right)
}\left(  \widehat{\theta}\right)  \right] \\
&  =O_{p}\left[  E_{\theta}\left[  IF_{2,2,\widetilde{\psi}_{k}\left(
\tau^{\dagger},\cdot\right)  }\left(  \widehat{\theta}\right)  \right]
\right]  +E_{\theta}\left[  \partial IF_{1,\widetilde{\psi}_{k}\left(
\tau,\cdot\right)  }\left(  \widehat{\theta}\right)  /\partial\tau\right] \\
&  \times E_{\theta}\left[  IF_{1,\widetilde{\tau}_{k}\left(  \cdot\right)
}\left(  \widehat{\theta}\right)  \right] \\
&  =O_{p}\left[  \left(  P-\widehat{P}\right)  \left(  B\ -\widehat{B}%
\ \right)  \right] \\
&  +O_{p}\left[  \left(  P-\widehat{P}\right)  +\left(  \frac{g\left(
X\right)  }{\widehat{g}\left(  X\right)  }-1\right)  \right]  \times\\
&  O_{p}\left[  \left(  P-\widehat{P}\right)  +\left(  \frac{g\left(
X\right)  }{\widehat{g}\left(  X\right)  }-1\right)  +\left(  B\ -\widehat{B}%
\ \right)  \right]
\end{align*}
Combining terms completes the proof.
\end{proof}

Next, we motivate and derive the formula of $\mathbb{IF}_{m,m,\tau\left(
\cdot\right)  }\left(  \theta\right)  $ for an assumed unique functional
$\tau\left(  \theta\right)  $ implicitly defined by $0=\psi\left(  \tau\left(
\theta\right)  ,\theta\right)  ,$ $\theta\in\Theta.$ \ To motivate the general
formula of $\mathbb{IF}_{m,m,\tau\left(  \cdot\right)  }\left(  \theta\right)
$ for arbitrary $m,$ we first consider the following formula for
$\mathbb{IF}_{44,\tau\left(  \cdot\right)  }$ which was derived from
$\mathbb{IF}_{1,\tau\left(  \cdot\right)  }$ following part 5c) of Theorem
\ref{eift}.%

\begin{align*}
&  -\psi_{\backslash\tau}\mathbb{IF}_{44,\tau\left(  \cdot\right)  }\\
&  =\mathbb{IF}_{4,4,\psi\left(  \tau,\cdot\right)  }\\
&  +\frac{1}{24}\mathbb{V}\left\{
\begin{array}
[c]{c}%
\psi_{\backslash\tau^{4}}IF_{1,\tau}\left(  i\right)  IF_{1,\tau}\left(
j\right)  IF_{1,\tau}\left(  s\right)  IF_{1,\tau}\left(  t\right) \\
+d_{1,\theta}\left(  \frac{\partial IF_{1,\tau}\left(  i\right)  }%
{\partial\tau^{3}}\right)  IF_{1,\tau}\left(  j\right)  IF_{1,\tau}\left(
s\right)  IF_{1,\psi\left(  \tau,\cdot\right)  }\left(  t\right) \\
+IF_{1,\tau}\left(  i\right)  d_{1,\theta}\left(  \frac{\partial IF_{1,\tau
}\left(  j\right)  }{\partial\tau^{3}}\right)  IF_{1,\tau}\left(  s\right)
IF_{1,\psi\left(  \tau,\cdot\right)  }\left(  t\right) \\
+IF_{1,\tau}\left(  i\right)  IF_{1,\tau}\left(  j\right)  d_{1,\theta}\left(
\frac{\partial IF_{1,\tau}\left(  s\right)  }{\partial\tau^{3}}\right)
IF_{1,\psi\left(  \tau,\cdot\right)  }\left(  t\right) \\
+IF_{1,\tau}\left(  i\right)  IF_{1,\tau}\left(  j\right)  IF_{1,\tau}\left(
s\right)  d_{1,\theta}\left(  \frac{\partial IF_{1,\psi\left(  \tau
,\cdot\right)  }\left(  t\right)  }{\partial\tau^{3}}\right)
\end{array}
\right\} \\
&  +\frac{1}{2}\mathbb{V}\left\{
\begin{array}
[c]{c}%
\psi_{\backslash\tau^{3}}IF_{22,\tau}\left(  ij\right)  IF_{1,\tau}\left(
s\right)  IF_{1,\tau}\left(  t\right) \\
+d_{1,\theta}\left(  \frac{\partial IF_{1,\psi}\left(  t\right)  }%
{\partial\tau^{2}}\right)  IF_{22,\tau}\left(  ij\right)  IF_{1,\tau}\left(
s\right) \\
+d_{1,\theta}\left(  \frac{\partial IF_{1,\psi}\left(  s\right)  }%
{\partial\tau^{2}}\right)  IF_{22,\tau}\left(  ij\right)  IF_{1,\tau}\left(
t\right) \\
+d_{1,\theta}\left(  \frac{\partial IF_{22,\psi}\left(  ij\right)  }%
{\partial\tau^{2}}\right)  IF_{1,\tau}\left(  t\right)  IF_{1,\tau}\left(
s\right)
\end{array}
\right\} \\
&  +\frac{1}{2}\mathbb{V}\left\{
\begin{array}
[c]{c}%
\psi_{\backslash\tau^{2}}IF_{22,\tau}\left(  ij\right)  IF_{22,\tau}\left(
st\right)  +\\
d_{2,\theta}\left(  \frac{\partial IF_{22,\psi}\left(  ij\right)  }%
{\partial\tau}\right)  IF_{22,\tau}\left(  st\right) \\
+IF_{22,\tau}\left(  ij\right)  d_{2,\theta}\left(  \frac{\partial
IF_{22,\psi}\left(  st\right)  }{\partial\tau}\right)
\end{array}
\right\} \\
&  +\mathbb{V}\left\{
\begin{array}
[c]{c}%
\psi_{\backslash\tau^{2}}IF_{1,\tau}\left(  i\right)  IF_{33,\tau}\left(
jst\right)  +\\
d_{1,\theta}\left(  \frac{\partial IF_{1,\psi}\left(  i\right)  }{\partial
\tau}\right)  IF_{33,\tau}\left(  jst\right)  +d_{3,\theta}\left(
\frac{\partial IF_{33,\psi}\left(  jst\right)  }{\partial\tau}\right)
IF_{1,\tau}\left(  i\right)
\end{array}
\right\}
\end{align*}

This formula for $-\psi_{\backslash\tau}\mathbb{IF}_{44,\tau\left(
\cdot\right)  }$ reveals a very nice pattern. Note that in addition to
$\mathbb{IF}_{4,4,\psi\left(  \tau,\cdot\right)  },$ the RHS consists of four
pieces with leading terms $\psi_{\backslash\tau^{4}}IF_{1,\tau}\left(
i\right)  IF_{1,\tau}\left(  j\right)  IF_{1,\tau}\left(  s\right)
IF_{1,\tau}\left(  t\right)  ,$ $\psi_{\backslash\tau^{3}}IF_{22,\tau}\left(
ij\right)  IF_{1,\tau}\left(  s\right)  IF_{1,\tau}\left(  t\right)  ,$
$~\psi_{\backslash\tau^{2}}IF_{22,\tau}\left(  ij\right)  IF_{22,\tau}\left(
st\right)  ,$ and $\psi_{\backslash\tau^{2}}IF_{1,\tau}\left(  i\right)
IF_{33,\tau}\left(  jst\right)  $ respectively. \ Within each piece, the
remaining terms can be constructed simply by applying the algorithm
\textbf{TE} described below.

\begin{algorithm}
\textbf{TE }i) Remove the partial derivative of $\psi$ wrt. $\tau;$ ii) for
each factor of the leading term, replace the $\tau\left(  \cdot\right)  $ in
the subscript with $\psi\left(  \tau,\cdot\right)  ;$ and iii) partially
differentiate the newly replaced factor wrt. $\tau$, and make the partial
derivative degenerate. Here the order of the partial derivative equals the
total number of the factors in the leading term minus 1.
\end{algorithm}

Moreover, each piece corresponds to one of the following ways of representing
the number 4 as a sum: $1+1+1+1,$ $1+1+2$, $2+2$, and $1+3.$ Furthermore,
assume $m$ can be written as $m=%
%TCIMACRO{\dsum \limits_{r=1}^{m-1}}%
%BeginExpansion
{\displaystyle\sum\limits_{r=1}^{m-1}}
%EndExpansion
\varkappa_{m,r}\times r$ $\ \ $with $\ \varkappa_{m,r}\geq0,$ e.g.,
$4=4\times1+0\times2+0\times3,$ then the number in front of the piece
corresponding to the sum representation $\left(  \varkappa_{m,1}%
,\varkappa_{m,2},...,\varkappa_{m,m-1}\right)  $ equals $\left(
%TCIMACRO{\dprod \limits_{r=1}^{m-1}}%
%BeginExpansion
{\displaystyle\prod\limits_{r=1}^{m-1}}
%EndExpansion
\varkappa_{m,r}!\right)  ^{-1}.$ Note that $0!=1.$

Now we are ready to generalize this expression to arbitrary $m,$ and prove it
by induction.

\begin{lemma}
\label{testing} Let $\tau\left(  \theta\right)  $ be the assumed unique
functional defined by $0=\psi\left(  \tau\left(  \theta\right)  ,\theta
\right)  ,\theta\in\Theta.$ Then, for $\theta\in\Theta\left(  \tau^{\dagger
}\right)  $, {}whenever $\mathbb{IF}_{m,\psi\left(  \tau^{\dagger}%
,\cdot\right)  }\left(  \theta\right)  $ and $\mathbb{IF}_{\ m,\tau\left(
\cdot\right)  }\left(  \theta\right)  $ exist ,
\begin{gather}
\mathbb{IF}_{1,\tau\left(  \cdot\right)  }\left(  \theta\right)
=-\psi_{\backslash\tau}^{-1}\mathbb{IF}_{1,\psi\left(  \tau,\cdot\right)
}\left(  \theta\right) \label{ts1}\\
-\psi_{\backslash\tau}\mathbb{IF}_{m,m,\tau\left(  \cdot\right)  }\left(
\theta\right) \nonumber\\
=\mathbb{IF}_{m,m,\psi\left(  \tau,\cdot\right)  }\left(  \theta\right)
\nonumber\\
+\sum_{\left(  \varkappa_{m,1},\varkappa_{m,2},...,\varkappa_{m,m-1}\right)
}\left(
%TCIMACRO{\dprod \limits_{r=1}^{m-1}}%
%BeginExpansion
{\displaystyle\prod\limits_{r=1}^{m-1}}
%EndExpansion
\varkappa_{m,r}!\right)  ^{-1}\mathbb{V}\left\{
\begin{array}
[c]{c}%
\psi_{\backslash\tau^{sum\left(  \overline{\varkappa}_{m}\right)  }}%
%TCIMACRO{\dprod \limits_{r=1}^{m-1}}%
%BeginExpansion
{\displaystyle\prod\limits_{r=1}^{m-1}}
%EndExpansion
\left(
%TCIMACRO{\dprod \limits_{s=1}^{\varkappa_{m,r}}}%
%BeginExpansion
{\displaystyle\prod\limits_{s=1}^{\varkappa_{m,r}}}
%EndExpansion
IF_{r,r,\tau\left(  \cdot\right)  ,\overline{i}_{r,s}}^{\left(  s\right)
}\left(  \theta\right)  \right) \\
+%
%TCIMACRO{\dsum \limits_{r,s}^{m-1,\varkappa_{m,r}}}%
%BeginExpansion
{\displaystyle\sum\limits_{r,s}^{m-1,\varkappa_{m,r}}}
%EndExpansion
\left(
\begin{array}
[c]{c}%
d_{r,\theta}\left(  \frac{\partial^{\left[  sum\left(  \overline{\varkappa
}_{m}\right)  -1\right]  }\left(  IF_{r,r,\psi\left(  \tau,\cdot\right)
,\overline{i}_{r,s}}^{\left(  s\right)  }\right)  }{\partial\tau^{\left[
sum\left(  \overline{\varkappa}_{m}\right)  -1\right]  }}\right)  \times\\%
%TCIMACRO{\dprod \limits_{\left(  r_{1},s_{1}\right)  \neq\left(  r,s\right)
%}}%
%BeginExpansion
{\displaystyle\prod\limits_{\left(  r_{1},s_{1}\right)  \neq\left(
r,s\right)  }}
%EndExpansion
IF_{r_{1},r_{1},\tau\left(  \cdot\right)  ,\overline{i}_{r_{1},s_{1}}%
}^{\left(  s\right)  }\left(  \theta\right)
\end{array}
\right)
\end{array}
\right\} \label{tsm}%
\end{gather}

where
\begin{align*}
\overline{\varkappa}_{m}  &  \equiv\left(  \varkappa_{m,1},\varkappa
_{m,2},...,\varkappa_{m,m-1}\right)  ,\text{ and }m=%
%TCIMACRO{\dsum \limits_{r=1}^{m-1}}%
%BeginExpansion
{\displaystyle\sum\limits_{r=1}^{m-1}}
%EndExpansion
\varkappa_{m,r}\times r,\text{ }\varkappa_{m,r}\geq0\text{ \ \ }\forall\text{
}1\leq r\leq m-1\\
sum\left(  \overline{\varkappa}_{m}\right)   &  \equiv\sum_{r=1}%
^{m-1}\varkappa_{m,r}\\
\overline{i}_{r,s}  &  \equiv\left(  i_{l_{r,s}+1},...,i_{l_{r,s}+r}\right)
\text{ \ where \ }l_{r,s}\equiv\sum_{q=1}^{r-1}\varkappa_{m,q}\times q+\left(
s-1\right)  r
\end{align*}

Note that
\[
\mathbb{V}\left\{
%TCIMACRO{\dsum \limits_{r,s}^{m-1,\varkappa_{m,r}}}%
%BeginExpansion
{\displaystyle\sum\limits_{r,s}^{m-1,\varkappa_{m,r}}}
%EndExpansion
\left(
\begin{array}
[c]{c}%
d_{r,\theta}\left(  \frac{\partial^{\left[  sum\left(  \overline{\varkappa
}_{m}\right)  -1\right]  }\left(  IF_{r,r,\psi\left(  \tau,\cdot\right)
,\overline{i}_{r,s}}^{\left(  s\right)  }\right)  }{\partial\tau^{\left[
sum\left(  \overline{\varkappa}_{m}\right)  -1\right]  }}\right)  \times\\%
%TCIMACRO{\dprod \limits_{\left(  r_{1},s_{1}\right)  \neq\left(  r,s\right)
%}}%
%BeginExpansion
{\displaystyle\prod\limits_{\left(  r_{1},s_{1}\right)  \neq\left(
r,s\right)  }}
%EndExpansion
IF_{r_{1},r_{1},\tau\left(  \cdot\right)  ,\overline{i}_{r_{1},s_{1}}%
}^{\left(  s\right)  }\left(  \theta\right)
\end{array}
\right)  \right\}
\]
can be constructed by applying Algorithm \textbf{TE }to the leading term
\[
\mathbb{V}\left\{  \psi_{\backslash\tau^{sum\left(  \overline{\varkappa}%
_{m}\right)  }}%
%TCIMACRO{\dprod \limits_{r=1}^{m-1}}%
%BeginExpansion
{\displaystyle\prod\limits_{r=1}^{m-1}}
%EndExpansion
\left(
%TCIMACRO{\dprod \limits_{s=1}^{\varkappa_{m,r}}}%
%BeginExpansion
{\displaystyle\prod\limits_{s=1}^{\varkappa_{m,r}}}
%EndExpansion
IF_{r,r,\tau\left(  \cdot\right)  ,\overline{i}_{r,s}}^{\left(  s\right)
}\left(  \theta\right)  \right)  \right\}  .
\]

\end{lemma}

\begin{proof}
Eq. (\ref{ts1}) has been proved in Theorem \ref{tt}. \ Next, we shall prove
eq. (\ref{tsm}) by induction. \ The case where $m=2$ was proved in Theorem
\ref{tt} as well. \ Now, we assume eq. (\ref{tsm}) holds for $m\,\ $and prove
it is also true for $m+1$ by part 5c) of Theorem \ref{eift}. Specifically, by
induction assumption,
\begin{align*}
&  -\psi_{\backslash\tau}if_{m,m,\tau\left(  \cdot\right)  }^{\left(
s\right)  }\left(  \mathbf{o}_{i_{m}};\theta\right) \\
&  =\ if_{m,m,\psi\left(  \tau,\cdot\right)  }^{\left(  s\right)  }\left(
\mathbf{o}_{i_{m}};\theta\right)  +\\
&  \sum_{\left(  \varkappa_{m,1},\varkappa_{m,2},...,\varkappa_{m,m-1}\right)
}\left(
%TCIMACRO{\dprod \limits_{r=1}^{m-1}}%
%BeginExpansion
{\displaystyle\prod\limits_{r=1}^{m-1}}
%EndExpansion
\varkappa_{m,r}!\right)  ^{-1}\times\\
&  \left\{
\begin{array}
[c]{c}%
\psi_{\backslash\tau^{sum\left(  \overline{\varkappa}_{m}\right)  }}%
%TCIMACRO{\dprod \limits_{r=1}^{m-1}}%
%BeginExpansion
{\displaystyle\prod\limits_{r=1}^{m-1}}
%EndExpansion
\left(
%TCIMACRO{\dprod \limits_{s=1}^{\varkappa_{m,r}}}%
%BeginExpansion
{\displaystyle\prod\limits_{s=1}^{\varkappa_{m,r}}}
%EndExpansion
IF_{r,r,\tau\left(  \cdot\right)  ,\overline{i}_{r,s}}^{\left(  s\right)
}\left(  \theta\right)  \right) \\
+%
%TCIMACRO{\dsum \limits_{r,s}^{m-1,\varkappa_{m,r}}}%
%BeginExpansion
{\displaystyle\sum\limits_{r,s}^{m-1,\varkappa_{m,r}}}
%EndExpansion
\left(
\begin{array}
[c]{c}%
d_{r,\theta}\left(  \frac{\partial^{\left[  sum\left(  \overline{\varkappa
}_{m}\right)  -1\right]  }\left(  IF_{r,r,\psi\left(  \tau,\cdot\right)
,\overline{i}_{r,s}}^{\left(  s\right)  }\right)  }{\partial\tau^{\left[
sum\left(  \overline{\varkappa}_{m}\right)  -1\right]  }}\right)  \times\\%
%TCIMACRO{\dprod \limits_{\left(  r_{1},s_{1}\right)  \neq\left(  r,s\right)
%}}%
%BeginExpansion
{\displaystyle\prod\limits_{\left(  r_{1},s_{1}\right)  \neq\left(
r,s\right)  }}
%EndExpansion
IF_{r_{1},r_{1},\tau\left(  \cdot\right)  ,\overline{i}_{r_{1},s_{1}}%
}^{\left(  s\right)  }\left(  \theta\right)
\end{array}
\right)
\end{array}
\right\}
\end{align*}

Consider any sufficiently smooth $1-$dimensional parametric submodel
$\theta_{t}$ mapping $R$ to $\Theta$. \ For any $\theta$ in the range of
$\theta_{t},$ differentiate both sides of the above equation w.r.t $t$ and
evaluate at $t^{\ast}\equiv$ $\theta_{t}^{-1}\left(  \theta\right)  ,$ then
\begin{align}
&  -\psi_{\backslash\tau}if_{m,m,\tau\left(  \cdot\right)  ,\backslash
t}^{\left(  s\right)  }\left(  \mathbf{o}_{i_{m}};\theta\right)  -\left(
\psi_{\backslash\tau^{2}}\tau_{\backslash t}\left(  \theta\right)
+\frac{\partial}{\partial t}\psi_{\backslash\tau}\left(  \tau,\theta
_{t}\right)  |_{t=t^{\ast}}\right)  if_{m,m,\tau\left(  \cdot\right)
}^{\left(  s\right)  }\left(  \mathbf{o}_{i_{m}};\theta\right) \nonumber\\
&  =\ \frac{\partial IF_{m,m,\psi\left(  \tau,\cdot\right)  }^{\left(
s\right)  }\left(  \tau,\theta\right)  }{\partial\tau}\tau_{\backslash
t}\left(  \theta\right)  +\ \frac{\partial IF_{m,m,\psi\left(  \tau
,\cdot\right)  }^{\left(  s\right)  }\left(  \tau,\theta_{t}\right)
}{\partial t}|_{t=t^{\ast}}\nonumber\\
&  +\sum_{\left(  \varkappa_{m,1},\varkappa_{m,2},...,\varkappa_{m,m-1}%
\right)  }\left(
%TCIMACRO{\dprod \limits_{r=1}^{m-1}}%
%BeginExpansion
{\displaystyle\prod\limits_{r=1}^{m-1}}
%EndExpansion
\varkappa_{m,r}!\right)  ^{-1}\times\nonumber\\
&  \left\{
\begin{array}
[c]{c}%
\left[
\begin{array}
[c]{c}%
\psi_{\backslash\tau^{sum\left(  \overline{\varkappa}_{m}\right)  +1}}%
\tau_{\backslash t}\left(  \theta\right) \\
+\frac{\partial\psi_{\backslash\tau^{sum\left(  \overline{\varkappa}%
_{m}\right)  }}\left(  \tau,\theta_{t}\right)  }{\partial t}|_{t=t^{\ast}}%
\end{array}
\right]
%TCIMACRO{\dprod \limits_{r=1}^{m-1}}%
%BeginExpansion
{\displaystyle\prod\limits_{r=1}^{m-1}}
%EndExpansion
\left(
%TCIMACRO{\dprod \limits_{s=1}^{\varkappa_{m,r}}}%
%BeginExpansion
{\displaystyle\prod\limits_{s=1}^{\varkappa_{m,r}}}
%EndExpansion
IF_{r,r,\tau\left(  \cdot\right)  ,\overline{i}_{r,s}}^{\left(  s\right)
}\left(  \theta\right)  \right) \\
+\psi_{\backslash\tau^{sum\left(  \overline{\varkappa}_{m}\right)  }}%
%TCIMACRO{\dsum \limits_{\left(  r,s\right)  }^{sum\left(  \overline{\varkappa
%}_{m}\right)  }}%
%BeginExpansion
{\displaystyle\sum\limits_{\left(  r,s\right)  }^{sum\left(  \overline
{\varkappa}_{m}\right)  }}
%EndExpansion
IF_{r,r,\tau\left(  \cdot\right)  ,\overline{i}_{r,s},\backslash t}^{\left(
s\right)  }\left(  \theta\right)
%TCIMACRO{\dprod \limits_{\left(  r_{1},s_{1}\right)  \neq\left(  r,s\right)
%}}%
%BeginExpansion
{\displaystyle\prod\limits_{\left(  r_{1},s_{1}\right)  \neq\left(
r,s\right)  }}
%EndExpansion
IF_{r_{1},r_{1},\tau\left(  \cdot\right)  ,\overline{i}_{r_{1},s_{1}}%
}^{\left(  s\right)  }\left(  \theta\right) \\
+%
%TCIMACRO{\dsum \limits_{\left(  r,s\right)  }^{sum\left(  \overline{\varkappa
%}_{m}\right)  }}%
%BeginExpansion
{\displaystyle\sum\limits_{\left(  r,s\right)  }^{sum\left(  \overline
{\varkappa}_{m}\right)  }}
%EndExpansion
\left\{
\begin{array}
[c]{c}%
\left[
\begin{array}
[c]{c}%
\frac{\partial^{sum\left(  \overline{\varkappa}_{m}\right)  }\left(
IF_{r,r,\psi\left(  \tau,\cdot\right)  ,\overline{i}_{r,s}}^{\left(  s\right)
}\right)  }{\partial\tau^{sum\left(  \overline{\varkappa}_{m}\right)  }}%
\tau_{\backslash t}\left(  \theta\right) \\
+\frac{\partial^{sum\left(  \overline{\varkappa}_{m}\right)  }\left(
IF_{r,r,\psi\left(  \tau,\cdot\right)  ,\overline{i}_{r,s}}^{\left(  s\right)
}\left(  \tau,\theta_{t}\right)  \right)  }{\partial\tau^{\left[  sum\left(
\overline{\varkappa}_{m}\right)  -1\right]  }\partial t}|_{t=t^{\ast}}%
\end{array}
\right] \\
\times%
%TCIMACRO{\dprod \limits_{\left(  r_{1},s_{1}\right)  \neq\left(  r,s\right)
%}}%
%BeginExpansion
{\displaystyle\prod\limits_{\left(  r_{1},s_{1}\right)  \neq\left(
r,s\right)  }}
%EndExpansion
IF_{r_{1},r_{1},\tau\left(  \cdot\right)  ,\overline{i}_{r_{1},s_{1}}%
}^{\left(  s\right)  }\left(  \theta\right) \\
+\frac{\partial^{\left[  sum\left(  \overline{\varkappa}_{m}\right)
-1\right]  }\left(  IF_{r,r,\psi\left(  \tau,\cdot\right)  ,\overline{i}%
_{r,s}}^{\left(  s\right)  }\right)  }{\partial\tau^{\left[  sum\left(
\overline{\varkappa}_{m}\right)  -1\right]  }}\\
\times%
%TCIMACRO{\dsum \limits_{\left(  r_{1},s_{1}\right)  \neq\left(  r,s\right)
%}}%
%BeginExpansion
{\displaystyle\sum\limits_{\left(  r_{1},s_{1}\right)  \neq\left(  r,s\right)
}}
%EndExpansion
IF_{r_{1},r_{1},\tau\left(  \cdot\right)  ,\overline{i}_{r_{1},s_{1}%
},\backslash t}^{\left(  s\right)  }\left(  \theta\right) \\
\times%
%TCIMACRO{\dprod \limits_{\left(  r_{2},s_{2}\right)  \neq\left(  r_{1}%
%,s_{1}\right)  \neq\left(  r,s\right)  }}%
%BeginExpansion
{\displaystyle\prod\limits_{\left(  r_{2},s_{2}\right)  \neq\left(
r_{1},s_{1}\right)  \neq\left(  r,s\right)  }}
%EndExpansion
IF_{r_{2},r_{2},\tau\left(  \cdot\right)  ,\overline{i}_{r_{2},s_{2}}%
}^{\left(  s\right)  }\left(  \theta\right)
\end{array}
\right\}
\end{array}
\right\} \label{m2m+1}%
\end{align}

Note that $\frac{\partial^{p}\left(  IF_{r,r,\psi\left(  \tau,\cdot\right)
,\overline{i}_{r,s}}^{\left(  s\right)  }\left(  \tau,\theta_{t}\right)
\right)  }{\partial\tau^{p-1}\partial t}|_{t=t^{\ast}}$ is the derivative of
$\frac{\partial^{p-1}\left(  IF_{r,r,\psi\left(  \tau,\cdot\right)
,\overline{i}_{r,s}}^{\left(  s\right)  }\right)  }{\partial\tau^{p-1}}$
w.r.t. $t$ while fixing $\tau$ at $\tau\left(  \theta\right)  .$ Therefore,
\begin{align*}
&  \frac{\partial^{p}\left(  IF_{r,r,\psi\left(  \tau,\cdot\right)
,\overline{i}_{r,s}}^{\left(  s\right)  }\left(  \tau,\theta_{t}\right)
\right)  }{\partial\tau^{p-1}\partial t}|_{t=t^{\ast}}\\
&  =\frac{\partial^{p-1}}{\partial\tau^{p-1}}\left(  \frac{\partial}{\partial
t}IF_{r,r,\psi\left(  \tau,\cdot\right)  ,\overline{i}_{r,s}}^{\left(
s\right)  }\left(  \tau,\theta_{t}\right)  |_{t=t^{\ast}}\right) \\
&  =\frac{\partial^{p-1}}{\partial\tau^{p-1}}E_{\theta}\left[
if_{1,IF_{r,r,\psi\left(  \tau,\cdot\right)  ,\overline{i}_{r,s}}^{\left(
s\right)  }\left(  \cdot\right)  }\left(  O_{m+1};\theta\right)
s_{1,t}\left(  O_{m+1}\right)  \right] \\
&  =E_{\theta}\left[  \frac{\partial^{p-1}if_{1,IF_{r,r,\psi\left(  \tau
,\cdot\right)  ,\overline{i}_{r,s}}^{\left(  s\right)  }\left(  \cdot\right)
}\left(  O_{m+1};\theta\right)  }{\partial\tau^{p-1}}s_{1,t}\left(
O_{m+1}\right)  \right]  ,
\end{align*}

and%
\begin{align*}
&  -\psi_{\backslash\tau}if_{1,if_{m,m,\tau\left(  \cdot\right)  ,\backslash
t}^{\left(  s\right)  }\left(  \mathbf{o}_{i_{m}};\cdot\right)  }\left(
O_{i_{m+1}};\theta\right) \\
&  =if_{1,if_{m,m,\psi\left(  \tau,\cdot\right)  }^{\left(  s\right)  }\left(
\mathbf{o}_{i_{m}};\cdot\right)  }\left(  O_{i_{m+1}};\theta\right) \\
&  +\left\{
\begin{array}
[c]{c}%
\ \psi_{\backslash\tau^{2}}IF_{1,\tau\left(  \cdot\right)  ,i_{m+1}}\left(
\theta\right)  if_{m,m,\tau\left(  \cdot\right)  }^{\left(  s\right)  }\left(
\mathbf{o}_{i_{m}};\theta\right) \\
+\frac{\partial}{\partial\tau}IF_{1,\psi\left(  \tau,\cdot\right)  ,i_{m+1}%
}if_{m,m,\tau\left(  \cdot\right)  }^{\left(  s\right)  }\left(
\mathbf{o}_{i_{m}};\theta\right) \\
+IF_{m,m,\psi\left(  \tau,\cdot\right)  ,\backslash\tau}^{\left(  s\right)
}IF_{1,\tau\left(  \cdot\right)  ,i_{m+1}}%
\end{array}
\right\} \\
&  +\sum_{\left(  \varkappa_{m,1},\varkappa_{m,2},...,\varkappa_{m,m-1}%
\right)  }\left(
%TCIMACRO{\dprod \limits_{r=1}^{m-1}}%
%BeginExpansion
{\displaystyle\prod\limits_{r=1}^{m-1}}
%EndExpansion
\varkappa_{m,r}!\right)  ^{-1}\times\\
&  \left\{
\begin{array}
[c]{c}%
\left\{
\begin{array}
[c]{c}%
\left\{
\begin{array}
[c]{c}%
\left[
\begin{array}
[c]{c}%
\psi_{\backslash\tau^{sum\left(  \overline{\varkappa}_{m}\right)  +1}%
}IF_{1,\tau\left(  \cdot\right)  ,i_{m+1}}\left(  \theta\right) \\
+\frac{\partial^{sum\left(  \overline{\varkappa}_{m}\right)  }}{\partial
\tau^{sum\left(  \overline{\varkappa}_{m}\right)  }}\left(  IF_{1,\psi\left(
\tau,\cdot\right)  ,i_{m+1}}\right)
\end{array}
\right] \\
\times%
%TCIMACRO{\dprod \limits_{r=1}^{m-1}}%
%BeginExpansion
{\displaystyle\prod\limits_{r=1}^{m-1}}
%EndExpansion
\left(
%TCIMACRO{\dprod \limits_{s=1}^{\varkappa_{m,r}}}%
%BeginExpansion
{\displaystyle\prod\limits_{s=1}^{\varkappa_{m,r}}}
%EndExpansion
IF_{r,r,\tau\left(  \cdot\right)  ,\overline{i}_{r,s}}^{\left(  s\right)
}\left(  \theta\right)  \right)
\end{array}
\right\} \\
+%
%TCIMACRO{\dsum \limits_{\left(  r,s\right)  }^{sum\left(  \overline{\varkappa
%}_{m}\right)  }}%
%BeginExpansion
{\displaystyle\sum\limits_{\left(  r,s\right)  }^{sum\left(  \overline
{\varkappa}_{m}\right)  }}
%EndExpansion
\left[
\begin{array}
[c]{c}%
\frac{\partial^{sum\left(  \overline{\varkappa}_{m}\right)  }\left(
IF_{r,r,\psi\left(  \tau,\cdot\right)  ,\overline{i}_{r,s}}^{\left(  s\right)
}\right)  }{\partial\tau^{sum\left(  \overline{\varkappa}_{m}\right)  }%
}IF_{1,\tau\left(  \cdot\right)  ,i_{m+1}}\left(  \theta\right) \\
\times%
%TCIMACRO{\dprod \limits_{\left(  r_{1},s_{1}\right)  \neq\left(  r,s\right)
%}}%
%BeginExpansion
{\displaystyle\prod\limits_{\left(  r_{1},s_{1}\right)  \neq\left(
r,s\right)  }}
%EndExpansion
IF_{r_{1},r_{1},\tau\left(  \cdot\right)  ,\overline{i}_{r_{1},s_{1}}%
}^{\left(  s\right)  }\left(  \theta\right)
\end{array}
\right]
\end{array}
\right\} \\
+%
%TCIMACRO{\dsum \limits_{\left(  r,s\right)  }^{sum\left(  \overline{\varkappa
%}_{m}\right)  }}%
%BeginExpansion
{\displaystyle\sum\limits_{\left(  r,s\right)  }^{sum\left(  \overline
{\varkappa}_{m}\right)  }}
%EndExpansion
\left(
\begin{array}
[c]{c}%
\left[
\begin{array}
[c]{c}%
\psi_{\backslash\tau^{sum\left(  \overline{\varkappa}_{m}\right)  }%
}if_{1,IF_{r,r,\tau\left(  \cdot\right)  ,\overline{i}_{r,s}}^{\left(
s\right)  }\left(  \cdot\right)  }\left(  O_{i_{m+1}};\theta\right) \\
\times%
%TCIMACRO{\dprod \limits_{\left(  r_{1},s_{1}\right)  \neq\left(  r,s\right)
%}}%
%BeginExpansion
{\displaystyle\prod\limits_{\left(  r_{1},s_{1}\right)  \neq\left(
r,s\right)  }}
%EndExpansion
IF_{r_{1},r_{1},\tau\left(  \cdot\right)  ,\overline{i}_{r_{1},s_{1}}%
}^{\left(  s\right)  }\left(  \theta\right)
\end{array}
\right] \\
+\left[
\begin{array}
[c]{c}%
\frac{\partial^{sum\left(  \overline{\varkappa}_{m}\right)  }\left(
if_{1,IF_{r,r,\psi\left(  \tau,\cdot\right)  ,\overline{i}_{r,s}}^{\left(
s\right)  }\left(  \cdot\right)  }\left(  O_{i_{m+1}};\theta\right)  \right)
}{\partial\tau^{\left[  sum\left(  \overline{\varkappa}_{m}\right)  -1\right]
}}\\
\times%
%TCIMACRO{\dprod \limits_{\left(  r_{1},s_{1}\right)  \neq\left(  r,s\right)
%}}%
%BeginExpansion
{\displaystyle\prod\limits_{\left(  r_{1},s_{1}\right)  \neq\left(
r,s\right)  }}
%EndExpansion
IF_{r_{1},r_{1},\tau\left(  \cdot\right)  ,\overline{i}_{r_{1},s_{1}}%
}^{\left(  s\right)  }\left(  \theta\right)
\end{array}
\right]
\end{array}
\right) \\
+%
%TCIMACRO{\dsum \limits_{\left(  r,s\right)  }^{sum\left(  \overline{\varkappa
%}_{m}\right)  }}%
%BeginExpansion
{\displaystyle\sum\limits_{\left(  r,s\right)  }^{sum\left(  \overline
{\varkappa}_{m}\right)  }}
%EndExpansion
\left\{
\begin{array}
[c]{c}%
\frac{\partial^{\left[  sum\left(  \overline{\varkappa}_{m}\right)  -1\right]
}\left(  IF_{r,r,\psi\left(  \tau,\cdot\right)  ,\overline{i}_{r,s}}^{\left(
s\right)  }\right)  }{\partial\tau^{\left[  sum\left(  \overline{\varkappa
}_{m}\right)  -1\right]  }}\times\\%
%TCIMACRO{\dsum \limits_{\left(  r_{1},s_{1}\right)  \neq\left(  r,s\right)
%}}%
%BeginExpansion
{\displaystyle\sum\limits_{\left(  r_{1},s_{1}\right)  \neq\left(  r,s\right)
}}
%EndExpansion
\left(
\begin{array}
[c]{c}%
if_{1,IF_{r_{1},r_{1},\tau\left(  \cdot\right)  ,\overline{i}_{r_{1},s_{1}}%
}^{\left(  s\right)  }\left(  \cdot\right)  }\left(  O_{i_{m+1}}%
;\theta\right)  \times\\%
%TCIMACRO{\dprod \limits_{\left(  r_{2},s_{2}\right)  \neq\left(  r_{1}%
%,s_{1}\right)  \neq\left(  r,s\right)  }}%
%BeginExpansion
{\displaystyle\prod\limits_{\left(  r_{2},s_{2}\right)  \neq\left(
r_{1},s_{1}\right)  \neq\left(  r,s\right)  }}
%EndExpansion
IF_{r_{2},r_{2},\tau\left(  \cdot\right)  ,\overline{i}_{r_{2},s_{2}}%
}^{\left(  s\right)  }\left(  \theta\right)
\end{array}
\right)
\end{array}
\right\}
\end{array}
\right\}
\end{align*}

Consider the last term
\[%
%TCIMACRO{\dsum \limits_{\left(  r,s\right)  }^{sum\left(  \overline{\varkappa
%}_{m}\right)  }}%
%BeginExpansion
{\displaystyle\sum\limits_{\left(  r,s\right)  }^{sum\left(  \overline
{\varkappa}_{m}\right)  }}
%EndExpansion
\left\{
\begin{array}
[c]{c}%
\frac{\partial^{\left[  sum\left(  \overline{\varkappa}_{m}\right)  -1\right]
}\left(  IF_{r,r,\psi\left(  \tau,\cdot\right)  ,\overline{i}_{r,s}}^{\left(
s\right)  }\right)  }{\partial\tau^{\left[  sum\left(  \overline{\varkappa
}_{m}\right)  -1\right]  }}%
%TCIMACRO{\dsum \limits_{\left(  r_{1},s_{1}\right)  \neq\left(  r,s\right)
%}}%
%BeginExpansion
{\displaystyle\sum\limits_{\left(  r_{1},s_{1}\right)  \neq\left(  r,s\right)
}}
%EndExpansion
if_{1,IF_{r_{1},r_{1},\tau\left(  \cdot\right)  ,\overline{i}_{r_{1},s_{1}}%
}^{\left(  s\right)  }\left(  \cdot\right)  }\left(  O_{i_{m+1}};\theta\right)
\\
\times%
%TCIMACRO{\dprod \limits_{\left(  r_{2},s_{2}\right)  \neq\left(  r_{1}%
%,s_{1}\right)  \neq\left(  r,s\right)  }}%
%BeginExpansion
{\displaystyle\prod\limits_{\left(  r_{2},s_{2}\right)  \neq\left(
r_{1},s_{1}\right)  \neq\left(  r,s\right)  }}
%EndExpansion
IF_{r_{2},r_{2},\tau\left(  \cdot\right)  ,\overline{i}_{r_{2},s_{2}}%
}^{\left(  s\right)  }\left(  \theta\right)
\end{array}
\right\}
\]

WLOG, we exchange $\left(  r,s\right)  $ with $\left(  r_{1},s_{1}\right)  ,$
then we have
\[%
%TCIMACRO{\dsum \limits_{\left(  r,s\right)  \neq\left(  r_{1},s_{1}\right)
%}^{sum\left(  \overline{\varkappa}_{m}\right)  ,sum\left(  \overline
%{\varkappa}_{m}\right)  }}%
%BeginExpansion
{\displaystyle\sum\limits_{\left(  r,s\right)  \neq\left(  r_{1},s_{1}\right)
}^{sum\left(  \overline{\varkappa}_{m}\right)  ,sum\left(  \overline
{\varkappa}_{m}\right)  }}
%EndExpansion
\left[
\begin{array}
[c]{c}%
\frac{\partial^{\left[  sum\left(  \overline{\varkappa}_{m}\right)  -1\right]
}\left(  IF_{r_{1},r_{1},\psi\left(  \tau,\cdot\right)  ,\overline{i}%
_{r_{1},s_{1}}}^{\left(  s\right)  }\right)  }{\partial\tau^{\left[
sum\left(  \overline{\varkappa}_{m}\right)  -1\right]  }}\times\\
if_{1,IF_{r,r,\tau\left(  \cdot\right)  ,\overline{i}_{r,s}}^{\left(
s\right)  }\left(  \cdot\right)  }\left(  O_{i_{m+1}};\theta\right)
%TCIMACRO{\dprod \limits_{\left(  r_{2},s_{2}\right)  \neq\left(  r_{1}%
%,s_{1}\right)  \neq\left(  r,s\right)  }}%
%BeginExpansion
{\displaystyle\prod\limits_{\left(  r_{2},s_{2}\right)  \neq\left(
r_{1},s_{1}\right)  \neq\left(  r,s\right)  }}
%EndExpansion
IF_{r_{2},r_{2},\tau\left(  \cdot\right)  ,\overline{i}_{r_{2},s_{2}}%
}^{\left(  s\right)  }\left(  \theta\right)
\end{array}
\right]
\]
and the sum of the last two terms equals
\[%
%TCIMACRO{\dsum \limits_{\left(  r,s\right)  }^{sum\left(  \overline{\varkappa
%}_{m}\right)  }}%
%BeginExpansion
{\displaystyle\sum\limits_{\left(  r,s\right)  }^{sum\left(  \overline
{\varkappa}_{m}\right)  }}
%EndExpansion
\left(
\begin{array}
[c]{c}%
\left[
\begin{array}
[c]{c}%
\psi_{\backslash\tau^{sum\left(  \overline{\varkappa}_{m}\right)  }%
}if_{1,IF_{r,r,\tau\left(  \cdot\right)  ,\overline{i}_{r,s}}^{\left(
s\right)  }\left(  \cdot\right)  }\left(  O_{i_{m+1}};\theta\right) \\
\times%
%TCIMACRO{\dprod \limits_{\left(  r_{1},s_{1}\right)  \neq\left(  r,s\right)
%}}%
%BeginExpansion
{\displaystyle\prod\limits_{\left(  r_{1},s_{1}\right)  \neq\left(
r,s\right)  }}
%EndExpansion
IF_{r_{1},r_{1},\tau\left(  \cdot\right)  ,\overline{i}_{r_{1},s_{1}}%
}^{\left(  s\right)  }\left(  \theta\right)
\end{array}
\right] \\
+\left[
\begin{array}
[c]{c}%
\frac{\partial^{sum\left(  \overline{\varkappa}_{m}\right)  }\left(
if_{1,IF_{r,r,\psi\left(  \tau,\cdot\right)  ,\overline{i}_{r,s}}^{\left(
s\right)  }\left(  \cdot\right)  }\left(  O_{i_{m+1}};\theta\right)  \right)
}{\partial\tau^{\left[  sum\left(  \overline{\varkappa}_{m}\right)  -1\right]
}}\\
\times%
%TCIMACRO{\dprod \limits_{\left(  r_{1},s_{1}\right)  \neq\left(  r,s\right)
%}}%
%BeginExpansion
{\displaystyle\prod\limits_{\left(  r_{1},s_{1}\right)  \neq\left(
r,s\right)  }}
%EndExpansion
IF_{r_{1},r_{1},\tau\left(  \cdot\right)  ,\overline{i}_{r_{1},s_{1}}%
}^{\left(  s\right)  }\left(  \theta\right)
\end{array}
\right] \\
+\left(
\begin{array}
[c]{c}%
if_{1,IF_{r,r,\tau\left(  \cdot\right)  ,\overline{i}_{r,s}}^{\left(
s\right)  }\left(  \cdot\right)  }\left(  O_{i_{m+1}};\theta\right)  \times\\%
%TCIMACRO{\dsum \limits_{\left(  r_{1},s_{1}\right)  \neq\left(  r,s\right)
%}}%
%BeginExpansion
{\displaystyle\sum\limits_{\left(  r_{1},s_{1}\right)  \neq\left(  r,s\right)
}}
%EndExpansion
\left[
\begin{array}
[c]{c}%
\frac{\partial^{\left[  sum\left(  \overline{\varkappa}_{m}\right)  -1\right]
}\left(  IF_{r_{1},r_{1},\psi\left(  \tau,\cdot\right)  ,\overline{i}%
_{r_{1},s_{1}}}^{\left(  s\right)  }\right)  }{\partial\tau^{\left[
sum\left(  \overline{\varkappa}_{m}\right)  -1\right]  }}\times\\%
%TCIMACRO{\dprod \limits_{\left(  r_{2},s_{2}\right)  \neq\left(  r_{1}%
%,s_{1}\right)  \neq\left(  r,s\right)  }}%
%BeginExpansion
{\displaystyle\prod\limits_{\left(  r_{2},s_{2}\right)  \neq\left(
r_{1},s_{1}\right)  \neq\left(  r,s\right)  }}
%EndExpansion
IF_{r_{2},r_{2},\tau\left(  \cdot\right)  ,\overline{i}_{r_{2},s_{2}}%
}^{\left(  s\right)  }\left(  \theta\right)
\end{array}
\right]
\end{array}
\right)
\end{array}
\right)
\]

Now, we have shown that, in addition to $if_{1,if_{m,m,\psi\left(  \tau
,\cdot\right)  }^{\left(  s\right)  }\left(  \mathbf{o}_{i_{m}};\cdot\right)
}\left(  O_{i_{m+1}};\theta\right)  ,$ the RHS of eq.(\ref{m2m+1}) can be
written as the sum of three pieces with the following leading terms%
\begin{align}
&  \psi_{\backslash\tau^{2}}IF_{1,\tau\left(  \cdot\right)  ,i_{m+1}}\left(
\theta\right)  if_{m,m,\tau\left(  \cdot\right)  }^{\left(  s\right)  }\left(
\mathbf{o}_{i_{m}};\theta\right) \nonumber\\
&  +\sum_{\left(  \varkappa_{m,1},\varkappa_{m,2},...,\varkappa_{m,m-1}%
\right)  }\left(
%TCIMACRO{\dprod \limits_{r=1}^{m-1}}%
%BeginExpansion
{\displaystyle\prod\limits_{r=1}^{m-1}}
%EndExpansion
\varkappa_{m,r}!\right)  ^{-1}\times\nonumber\\
&  \left\{
\begin{array}
[c]{c}%
\left[  \psi_{\backslash\tau^{sum\left(  \overline{\varkappa}_{m}\right)  +1}%
}IF_{1,\tau\left(  \cdot\right)  ,i_{m+1}}\left(  \theta\right)  \right]
%TCIMACRO{\dprod \limits_{r=1}^{m-1}}%
%BeginExpansion
{\displaystyle\prod\limits_{r=1}^{m-1}}
%EndExpansion
\left(
%TCIMACRO{\dprod \limits_{s=1}^{\varkappa_{m,r}}}%
%BeginExpansion
{\displaystyle\prod\limits_{s=1}^{\varkappa_{m,r}}}
%EndExpansion
IF_{r,r,\tau\left(  \cdot\right)  ,\overline{i}_{r,s}}^{\left(  s\right)
}\left(  \theta\right)  \right) \\
+%
%TCIMACRO{\dsum \limits_{\left(  r,s\right)  }^{sum\left(  \overline{\varkappa
%}_{m}\right)  }}%
%BeginExpansion
{\displaystyle\sum\limits_{\left(  r,s\right)  }^{sum\left(  \overline
{\varkappa}_{m}\right)  }}
%EndExpansion
\left(
\begin{array}
[c]{c}%
\psi_{\backslash\tau^{sum\left(  \overline{\varkappa}_{m}\right)  }%
}if_{1,IF_{r,r,\tau\left(  \cdot\right)  ,\overline{i}_{r,s}}^{\left(
s\right)  }\left(  \cdot\right)  }\left(  O_{i_{m+1}};\theta\right) \\
\times%
%TCIMACRO{\dprod \limits_{\left(  r_{1},s_{1}\right)  \neq\left(  r,s\right)
%}}%
%BeginExpansion
{\displaystyle\prod\limits_{\left(  r_{1},s_{1}\right)  \neq\left(
r,s\right)  }}
%EndExpansion
IF_{r_{1},r_{1},\tau\left(  \cdot\right)  ,\overline{i}_{r_{1},s_{1}}%
}^{\left(  s\right)  }\left(  \theta\right)
\end{array}
\right)
\end{array}
\right\}  ,\label{lead_m+1}%
\end{align}
while the remaining terms can be constructed by applying the algorithm
\textbf{TE }to the above leading terms.

By part 5c) of Theorem \ref{eift} and the induction assumption, next, we only
need to prove that eq. (\ref{lead_m+1}) is actually a kernel of the following
$\left(  m+1\right)  $th order U-statistic
\begin{align}
&  \left(  m+1\right)  \sum_{\left(  \varkappa_{m+1,1},\varkappa
_{m+1,2},...,\varkappa_{m+1,m}\right)  }\left(
%TCIMACRO{\dprod \limits_{r=1}^{m}}%
%BeginExpansion
{\displaystyle\prod\limits_{r=1}^{m}}
%EndExpansion
\varkappa_{m+1,r}!\right)  ^{-1}\times\nonumber\\
&  \mathbb{V}\left\{  \psi_{\backslash\tau^{sum\left(  \overline{\varkappa
}_{m+1}\right)  }}%
%TCIMACRO{\dprod \limits_{r=1}^{m}}%
%BeginExpansion
{\displaystyle\prod\limits_{r=1}^{m}}
%EndExpansion
\left(
%TCIMACRO{\dprod \limits_{s=1}^{\varkappa_{m+1,r}}}%
%BeginExpansion
{\displaystyle\prod\limits_{s=1}^{\varkappa_{m+1,r}}}
%EndExpansion
IF_{r,r,\tau\left(  \cdot\right)  ,\overline{i}_{r,s}^{\ast}}^{\left(
s\right)  }\left(  \theta\right)  \right)  \right\} \label{lead2}%
\end{align}

\ where \newline%
\begin{align*}
\overline{\varkappa}_{m+1}  &  \equiv\left(  \varkappa_{m+1,1},\varkappa
_{m+1,2},...,\varkappa_{m+1,m}\right)  ,\text{ }\\
\text{and }m+1  &  =%
%TCIMACRO{\dsum \limits_{r=1}^{m}}%
%BeginExpansion
{\displaystyle\sum\limits_{r=1}^{m}}
%EndExpansion
\varkappa_{m+1,r}\times r,\text{ }\varkappa_{m+1,r}\geq0\text{ \ \ }%
\forall\text{ }1\leq r\leq m\\
sum\left(  \overline{\varkappa}_{m+1}\right)   &  \equiv\sum_{r=1}%
^{m}\varkappa_{m+1,r}\\
\overline{i}_{r,s}^{\ast}  &  \equiv\left(  i_{l_{r,s}+1}^{\ast}%
,...,i_{l_{r,s}+r}^{\ast}\right)  \text{ \ where \ }l_{r,s}^{\ast}\equiv
\sum_{q=1}^{r-1}\varkappa_{m+1,q}\times q+\left(  s-1\right)  r
\end{align*}

This can be proved following a simple but important fact that, for any sum
representation of$\ \ m+1=%
%TCIMACRO{\dsum \limits_{r=1}^{m}}%
%BeginExpansion
{\displaystyle\sum\limits_{r=1}^{m}}
%EndExpansion
\varkappa_{m+1,r}\times r,$ either i) $\varkappa_{m+1,1}=\varkappa_{m+1,m}=1,$
and $\varkappa_{m+1,r}=0$ for $\forall$ $1<r<m;$ or ii) $\varkappa_{m+1,m}=0$
and there exists a sum representation of $m=$ $%
%TCIMACRO{\dsum \limits_{r=1}^{m-1}}%
%BeginExpansion
{\displaystyle\sum\limits_{r=1}^{m-1}}
%EndExpansion
\varkappa_{m,r}\times r$ such that,

iia) $\varkappa_{m+1,1}=\varkappa_{m,1}+1$ and $\varkappa_{m+1,r}%
=\varkappa_{m,r}$ for $1<r<m,$ or

iib) $\exists$ $1\leq r^{\ast}\leq m-2,$ \ such that $\varkappa_{m+1,r^{\ast}%
}=\varkappa_{m,r^{\ast}}-1,$ $\varkappa_{m+1,r^{\ast}+1}=\varkappa_{m,r^{\ast
}+1}+1,$ and $\varkappa_{m+1,r}=\varkappa_{m,r}$ for $\forall$ $r\neq r^{\ast
},r^{\ast}+1.$

Define%
\begin{align*}
&  LT1\left(  \varkappa_{m,1},...,\varkappa_{m,m-1}\right) \\
&  =\left(
%TCIMACRO{\dprod \limits_{r=1}^{m-1}}%
%BeginExpansion
{\displaystyle\prod\limits_{r=1}^{m-1}}
%EndExpansion
\varkappa_{m,r}!\right)  ^{-1}\left[  \psi_{\backslash\tau^{sum\left(
\overline{\varkappa}_{m}\right)  +1}}IF_{1,\tau\left(  \cdot\right)  ,i_{m+1}%
}\left(  \theta\right)  \right] \\
&  \times%
%TCIMACRO{\dprod \limits_{r=1}^{m-1}}%
%BeginExpansion
{\displaystyle\prod\limits_{r=1}^{m-1}}
%EndExpansion
\left(
%TCIMACRO{\dprod \limits_{s=1}^{\varkappa_{m,r}}}%
%BeginExpansion
{\displaystyle\prod\limits_{s=1}^{\varkappa_{m,r}}}
%EndExpansion
IF_{r,r,\tau\left(  \cdot\right)  ,\overline{i}_{r,s}}^{\left(  s\right)
}\left(  \theta\right)  \right) \\
&  LT2\left(  \varkappa_{m,1},...,\varkappa_{m,m-1};r\right) \\
&  =\left(
%TCIMACRO{\dprod \limits_{r=1}^{m-1}}%
%BeginExpansion
{\displaystyle\prod\limits_{r=1}^{m-1}}
%EndExpansion
\varkappa_{m,r}!\right)  ^{-1}\sum_{s=1}^{\varkappa_{m,r}}\left(
\begin{array}
[c]{c}%
\psi_{\backslash\tau^{sum\left(  \overline{\varkappa}_{m}\right)  }%
}if_{1,IF_{r,r,\tau\left(  \cdot\right)  ,\overline{i}_{r,s}}^{\left(
s\right)  }\left(  \cdot\right)  }\left(  O_{i_{m+1}};\theta\right) \\
\times%
%TCIMACRO{\dprod \limits_{\left(  r_{1},s_{1}\right)  \neq\left(  r,s\right)
%}}%
%BeginExpansion
{\displaystyle\prod\limits_{\left(  r_{1},s_{1}\right)  \neq\left(
r,s\right)  }}
%EndExpansion
IF_{r_{1},r_{1},\tau\left(  \cdot\right)  ,\overline{i}_{r_{1},s_{1}}%
}^{\left(  s\right)  }\left(  \theta\right)
\end{array}
\right)  .
\end{align*}

Note that for any $\left(  \varkappa_{m+1,1},\varkappa_{m+1,2},...,\varkappa
_{m+1,m}=0\right)  ,$ if $\varkappa_{m+1,1}>0,$ then
\begin{align*}
&  \mathbb{V}\left\{  LT1\left(  \varkappa_{m+1,1}-1,\varkappa_{m+1,2}%
...,\varkappa_{m+1,m-1}\right)  \right\} \\
&  =\varkappa_{m+1,1}\left(
%TCIMACRO{\dprod \limits_{r=1}^{m}}%
%BeginExpansion
{\displaystyle\prod\limits_{r=1}^{m}}
%EndExpansion
\varkappa_{m+1,r}!\right)  ^{-1}\mathbb{V}\left\{  \psi_{\backslash
\tau^{sum\left(  \overline{\varkappa}_{m+1}\right)  }}%
%TCIMACRO{\dprod \limits_{r=1}^{m}}%
%BeginExpansion
{\displaystyle\prod\limits_{r=1}^{m}}
%EndExpansion
\left(
%TCIMACRO{\dprod \limits_{s=1}^{\varkappa_{m+1,r}}}%
%BeginExpansion
{\displaystyle\prod\limits_{s=1}^{\varkappa_{m+1,r}}}
%EndExpansion
IF_{r,r,\tau\left(  \cdot\right)  ,\overline{i}_{r,s}^{\ast}}^{\left(
s\right)  }\left(  \theta\right)  \right)  \right\}  ,
\end{align*}
and for any $\ 2\leq r\leq m-1$ such that $\varkappa_{m+1,r}>0,$
\begin{align*}
&  \mathbb{V}\left\{  d_{m+1}\left(  LT2\left(  \varkappa_{m+1,1}%
,..,\varkappa_{m+1,r-1}+1,\varkappa_{m+1,r}-1,..,\varkappa_{m+1,m-1},r\right)
\right)  \right\} \\
&  =\varkappa_{m+1,r}r\left(
%TCIMACRO{\dprod \limits_{r=1}^{m}}%
%BeginExpansion
{\displaystyle\prod\limits_{r=1}^{m}}
%EndExpansion
\varkappa_{m+1,r}!\right)  ^{-1}\mathbb{V}\left\{  \psi_{\backslash
\tau^{sum\left(  \overline{\varkappa}_{m+1}\right)  }}%
%TCIMACRO{\dprod \limits_{r=1}^{m}}%
%BeginExpansion
{\displaystyle\prod\limits_{r=1}^{m}}
%EndExpansion
\left(
%TCIMACRO{\dprod \limits_{s=1}^{\varkappa_{m+1,r}}}%
%BeginExpansion
{\displaystyle\prod\limits_{s=1}^{\varkappa_{m+1,r}}}
%EndExpansion
IF_{r,r,\tau\left(  \cdot\right)  ,\overline{i}_{r,s}^{\ast}}^{\left(
s\right)  }\left(  \theta\right)  \right)  \right\}  .
\end{align*}
As $%
%TCIMACRO{\dsum \limits_{r=1}^{m-1}}%
%BeginExpansion
{\displaystyle\sum\limits_{r=1}^{m-1}}
%EndExpansion
r\varkappa_{m+1,r}=m+1.$ \ Now, it is obvious that the term with $\left(
\varkappa_{m+1,1},\varkappa_{m+1,2},...,\varkappa_{m+1,m}=0\right)  $ in eq.
(\ref{lead2}) comes from the following terms in eq. (\ref{lead_m+1})
\begin{align*}
&  I\left\{  \varkappa_{m+1,1}>0\right\}  LT1\left(  \varkappa_{m+1,1}%
-1,\varkappa_{m+1,2}...,\varkappa_{m+1,m-1}\right) \\
&  +\sum_{r=2}^{m-1}I\left\{  \varkappa_{m+1,r}>0\right\}  LT2\left(
\varkappa_{m+1,1},..,\varkappa_{m+1,r-1}+1,\varkappa_{m+1,r}-1,..,\varkappa
_{m+1,m-1},r\right)  ,
\end{align*}
while the\ term with $\left(  \varkappa_{m+1,1}=1,0,...,0,\varkappa
_{m+1,m}=1\right)  $ in eq. (\ref{lead2}) comes from the following terms in
eq. (\ref{lead_m+1})
\begin{align*}
&  \psi_{\backslash\tau^{2}}IF_{1,\tau\left(  \cdot\right)  ,i_{m+1}}\left(
\theta\right)  if_{m,m,\tau\left(  \cdot\right)  }^{\left(  s\right)  }\left(
\mathbf{o}_{i_{m}};\theta\right) \\
&  +\psi_{\backslash\tau^{2}}if_{1,IF_{m-1,m-1,\tau\left(  \cdot\right)
,\overline{i}_{m-1,1}}^{\left(  s\right)  }\left(  \cdot\right)  }\left(
O_{i_{m+1}};\theta\right)  IF_{1,\tau\left(  \cdot\right)  ,i_{1}}^{\left(
s\right)  }\left(  \theta\right)  .
\end{align*}

\end{proof}

\begin{proof}
(Theorem \ref{HOIPROD}) We proceed by induction. \ For $j=1,$
\begin{align*}
\psi_{\backslash l_{1}}\left(  \theta\right)   &  =\sum_{r=1}^{\zeta}\left\{
\psi_{r,\backslash l_{1}}\left(  \theta\right)  \times\left[  \prod
\limits_{s\leq\zeta,\zeta\neq r}\psi_{s}\left(  \theta\right)  \right]
\right\} \\
&  =\sum_{r=1}^{\zeta}\left\{  E_{\theta}\left[  \mathbb{IF}_{\psi_{r}%
,1}\left(  \theta\right)  \mathbb{\zeta}_{1;\backslash l_{1}}\left(
\theta\right)  \right]  \times\left[  \prod\limits_{s\leq\zeta,s\neq r}%
\psi_{s}\left(  \theta\right)  \right]  \right\}
\end{align*}
therefore
\begin{align*}
\mathbb{IF}_{\psi\left(  \theta;\zeta\right)  ,1,1}\left(  \theta\right)   &
=\mathbb{V}\left\{  \sum_{r=1}^{\zeta}IF_{\psi_{r};1,1;i_{1}}\left(
\theta\right)  \times\left[  \prod\limits_{s\leq\zeta,s\neq r}\psi_{s}\left(
\theta\right)  \right]  \right\} \\
&  =\mathbb{V}\left[
%TCIMACRO{\dsum \limits_{\left\{  t_{1},...t_{\zeta}\right\}  \in
%\Upsilon_{\zeta;1}}}%
%BeginExpansion
{\displaystyle\sum\limits_{\left\{  t_{1},...t_{\zeta}\right\}  \in
\Upsilon_{\zeta;1}}}
%EndExpansion%
%TCIMACRO{\dprod \limits_{s=1}^{\zeta}}%
%BeginExpansion
{\displaystyle\prod\limits_{s=1}^{\zeta}}
%EndExpansion
IF_{\psi_{s}\left(  \theta\right)  ;t_{s},t_{s};\overline{i}_{s,t_{s}}}\left(
\theta\right)  \right]
\end{align*}
Assume that the lemma holds for $j,$i.e.:
\[
\mathbb{IF}_{\psi\left(  \theta;\zeta\right)  ,j,j}\left(  \theta\right)
=\mathbb{V}\left[
%TCIMACRO{\dsum \limits_{\left\{  t_{1},...t_{\zeta}\right\}  \in
%\Upsilon_{\zeta;j}}}%
%BeginExpansion
{\displaystyle\sum\limits_{\left\{  t_{1},...t_{\zeta}\right\}  \in
\Upsilon_{\zeta;j}}}
%EndExpansion%
%TCIMACRO{\dprod \limits_{s=1}^{\zeta}}%
%BeginExpansion
{\displaystyle\prod\limits_{s=1}^{\zeta}}
%EndExpansion
IF_{\psi_{s}\left(  \theta\right)  ;t_{s},t_{s};\overline{i}_{s,t_{s}}}\left(
\theta\right)  \right]
\]
we now show that it holds for $j+1.$ Now,
\begin{align*}
&  \left(  j+1\right)  \mathbb{IF}_{\psi\left(  \theta;\zeta\right)
;j+1,j+1}\left(  \theta\right) \\
&  =\mathbb{V}\left[  if_{if_{\psi\left(  \theta;\zeta\right)  ;j,j}\left(
O_{i_{1}},...,O_{i_{j}};\cdot\right)  ;1,1\ }\left(  O_{i_{j+1}}%
;\theta\right)  \right] \\
&  -\Pi_{\theta,m}\left[  \mathbb{V}\left[  if_{1,if_{\psi\left(  \theta
;\zeta\right)  ;j,j}\left(  O_{i_{1}},...,O_{i_{j}};\cdot\right)  \ }\left(
O_{i_{j+1}};\theta\right)  \right]  |\mathcal{U}_{j}\left(  \theta\right)
\right]
\end{align*}
so that $\left(  j+1\right)  \mathbb{IF}_{\psi\left(  \theta;\zeta\right)
;j+1,j+1}\left(  \theta\right)  =$
\begin{align*}
&  \mathbb{V}\left[
%TCIMACRO{\dsum \limits_{\left\{  t_{1},...t_{\zeta}\right\}  \in
%\Upsilon_{\zeta;j}}}%
%BeginExpansion
{\displaystyle\sum\limits_{\left\{  t_{1},...t_{\zeta}\right\}  \in
\Upsilon_{\zeta;j}}}
%EndExpansion%
%TCIMACRO{\dsum \limits_{r=1}^{\zeta}}%
%BeginExpansion
{\displaystyle\sum\limits_{r=1}^{\zeta}}
%EndExpansion
\left\{
\begin{array}
[c]{c}%
\left\{
\begin{array}
[c]{c}%
if_{IF_{\psi_{r}\left(  \theta\right)  ;t_{r},t_{r};\overline{i}_{r,t_{r}}%
}\left(  \theta\right)  ;1,1\ }\left(  O_{i_{r,t_{r}+1}};\theta\right) \\
-\Pi_{\theta,t_{r}}\left[  \mathbb{V}\left[  if_{IF_{\psi_{r}\left(
\theta\right)  ;t_{r},t_{r};\overline{i}_{r,t_{r}}}\left(  \theta\right)
;1,1\ }\left(  O_{i_{r,t_{r}+1}};\theta\right)  \right]  |\mathcal{U}_{t_{r}%
}\left(  \theta\right)  \right]
\end{array}
\right\} \\
\times\left(
%TCIMACRO{\dprod \limits_{s\leq\zeta,s\neq r}}%
%BeginExpansion
{\displaystyle\prod\limits_{s\leq\zeta,s\neq r}}
%EndExpansion
IF_{\psi_{s}\left(  \theta\right)  ;t_{s},t_{s};\overline{i}_{s,t_{s}}}\left(
\theta\right)  \right)
\end{array}
\right\}  \right] \\
&  =\mathbb{V}\left[
%TCIMACRO{\dsum \limits_{\left\{  t_{1},...t_{\zeta}\right\}  \in
%\Upsilon_{\zeta;j}}}%
%BeginExpansion
{\displaystyle\sum\limits_{\left\{  t_{1},...t_{\zeta}\right\}  \in
\Upsilon_{\zeta;j}}}
%EndExpansion%
%TCIMACRO{\dsum \limits_{r=1}^{\zeta}}%
%BeginExpansion
{\displaystyle\sum\limits_{r=1}^{\zeta}}
%EndExpansion
\left\{
\begin{array}
[c]{c}%
\left(  t_{r}+1\right)  IF_{\psi_{r}\left(  \theta\right)  ;t_{r}%
+1,t_{r}+1;\overline{i}_{r,t_{r}+1}}\left(  \theta\right) \\
\times\left(
%TCIMACRO{\dprod \limits_{s\leq\zeta,s\neq r}}%
%BeginExpansion
{\displaystyle\prod\limits_{s\leq\zeta,s\neq r}}
%EndExpansion
IF_{\psi_{s}\left(  \theta\right)  ;t_{s},t_{s};\overline{i}_{s,t_{s}}}\left(
\theta\right)  \right)
\end{array}
\right\}  \right]
\end{align*}
Now, consider an arbitrary term in the double sum corresponding to the index
vector
\[
\left(  t_{1}^{\prime},...,t_{r^{\ast}}^{\prime},...t_{\zeta}^{\prime}\right)
\]
where the star indicates the index of the second summation . Then this term
and terms corresponding to
\begin{align*}
&  \left(  t_{1^{\ast}}^{\prime}-1,t_{2}^{\prime},...,t_{r}^{\prime
}+1,...,t_{\zeta}^{\prime}\right)  ,\\
&  \left(  t_{1}^{\prime},t_{2^{\ast}}^{\prime}-1,...,t_{r}^{\prime
}+1,...,t_{\zeta}^{\prime}\right)  ,\\
&  \vdots\\
&  \left(  t_{1}^{\prime},...t_{r-1^{\ast}}^{\prime}-1,t_{r}^{\prime
}+1,...,t_{\zeta}^{\prime}\right)  ,\\
&  \left(  t_{1}^{\prime},...t_{r}^{\prime}+1,t_{r+1^{\ast}}^{\prime
}-1,...,t_{\zeta}^{\prime}\right)  ,\\
&  \vdots\\
&  \left(  t_{1}^{\prime},...,t_{r}^{\prime}+1,...,t_{\zeta}^{\prime
}-1\right)
\end{align*}
all share the common factor
\[
IF_{\psi_{r^{\ast}}\left(  \theta\right)  ;t_{r^{\ast}}+1,t_{r^{\ast}%
}+1;\overline{i}_{r^{\ast},t_{r^{\ast}}+1}}\left(  \theta\right)
\times\left(
%TCIMACRO{\dprod \limits_{s\leq\zeta,s\neq r^{\ast}}}%
%BeginExpansion
{\displaystyle\prod\limits_{s\leq\zeta,s\neq r^{\ast}}}
%EndExpansion
IF_{\psi_{s}\left(  \theta\right)  ;t_{s}^{\prime},t_{s}^{\prime};\overline
{i}_{s,t_{s}^{^{\prime}}}}\left(  \theta\right)  \right)
\]
but with respective multiplicative constants $\left(  t_{r^{\ast}}^{\prime
}+1\right)  ,t_{1^{\ast}}^{\prime},t_{2^{\ast}}^{\prime},...,t_{r-1^{\ast}%
}^{\prime},t_{r+1^{\ast}}^{\prime},...,t_{\zeta}^{\prime}$. \ So that the
total contribution of terms in the double summation that share this common
factor is given by:%
\begin{align*}
&  \left(  \left(  t_{r^{\ast}}^{\prime}+1\right)  +\sum_{r\neq r^{\ast}}%
t_{r}^{\prime}\right)  \left(
\begin{array}
[c]{c}%
IF_{\psi_{r^{\ast}}\left(  \theta\right)  ;t_{r^{\ast}}^{\prime}+1,t_{r^{\ast
}}^{\prime}+1;\overline{i}_{r^{\ast},t_{r^{\ast}}^{\prime}+1}}\left(
\theta\right) \\
\times\left(
%TCIMACRO{\dprod \limits_{s\leq\zeta,s\neq r^{\ast}}}%
%BeginExpansion
{\displaystyle\prod\limits_{s\leq\zeta,s\neq r^{\ast}}}
%EndExpansion
IF_{\psi_{s}\left(  \theta\right)  ;t_{s}^{\prime},t_{s}^{\prime};\overline
{i}_{s,t_{s}^{^{\prime}}}}\left(  \theta\right)  \right)
\end{array}
\right) \\
&  =(j+1)\left(
\begin{array}
[c]{c}%
IF_{\psi_{r^{\ast}}\left(  \theta\right)  ;t_{r^{\ast}}^{\prime}+1,t_{r^{\ast
}}^{\prime}+1;\overline{i}_{r^{\ast},t_{r^{\ast}}^{\prime}+1}}\left(
\theta\right) \\
\times\left(
%TCIMACRO{\dprod \limits_{s\leq\zeta,s\neq r^{\ast}}}%
%BeginExpansion
{\displaystyle\prod\limits_{s\leq\zeta,s\neq r^{\ast}}}
%EndExpansion
IF_{\psi_{s}\left(  \theta\right)  ;t_{s}^{\prime},t_{s}^{\prime};\overline
{i}_{s,t_{s}^{^{\prime}}}}\left(  \theta\right)  \right)
\end{array}
\right)
\end{align*}
Repeating this argument over all common factors in the set
\[
\mathcal{A=}\left\{  \left(
\begin{array}
[c]{c}%
IF_{\psi_{r^{\ast}}\left(  \theta\right)  ;t_{r^{\ast}}+1,t_{r^{\ast}%
}+1;\overline{i}_{r^{\ast},t_{r^{\ast}}+1}}\left(  \theta\right) \\
\times\left(
%TCIMACRO{\dprod \limits_{s\leq\zeta,s\neq r^{\ast}}}%
%BeginExpansion
{\displaystyle\prod\limits_{s\leq\zeta,s\neq r^{\ast}}}
%EndExpansion
IF_{\psi_{s}\left(  \theta\right)  ;t_{s}^{\prime},t_{s}^{\prime};\overline
{i}_{s,t_{s}^{^{\prime}}}}\left(  \theta\right)  \right)
\end{array}
\right)  :\left\{  t_{1},...t_{\zeta}\right\}  \in\Upsilon_{\zeta;j}\right\}
\]
that appear in the double sum, we recover the desired sum%

\[
\left(  j+1\right)  \mathbb{IF}_{\psi\left(  \theta;\zeta\right)
;j+1,j+1}\left(  \theta\right)  =\left(  j+1\right)  \mathbb{V}\left[
%TCIMACRO{\dsum \limits_{\left\{  t_{1},...t_{\zeta}\right\}  \in
%\Upsilon_{\zeta;j+1}}}%
%BeginExpansion
{\displaystyle\sum\limits_{\left\{  t_{1},...t_{\zeta}\right\}  \in
\Upsilon_{\zeta;j+1}}}
%EndExpansion
\left(
%TCIMACRO{\dprod \limits_{s\leq\zeta}}%
%BeginExpansion
{\displaystyle\prod\limits_{s\leq\zeta}}
%EndExpansion
IF_{\psi_{s}\left(  \theta\right)  ;t_{s},t_{s};\overline{i}_{s,t_{s}}}\left(
\theta\right)  \right)  \right]
\]
This is because the set of all common factors in $\mathcal{A}$ is precisely :%
\[
\left\{
%TCIMACRO{\dprod \limits_{s\leq\zeta}}%
%BeginExpansion
{\displaystyle\prod\limits_{s\leq\zeta}}
%EndExpansion
IF_{\psi_{s}\left(  \theta\right)  ;t_{s},t_{s};\overline{i}_{s,t_{s}}}\left(
\theta\right)  :\left\{  t_{1},...t_{\zeta}\right\}  \in\Upsilon_{\zeta
;j+1}\right\}
\]
This concludes the proof.
\end{proof}

\begin{proof}
(Theorem \ref{BTOMM}) \bigskip for $m=1,$%
\begin{align*}
&  E_{\theta}\left(  \widehat{\psi}_{1}\right)  -\psi\left(  \theta\right) \\
&  =E_{\theta}\left[  \frac{R_{1}}{\widehat{\pi}_{1}}\frac{R_{0}}%
{\widehat{\pi}_{0}}\left(  Y-\widehat{B}_{1}\right)  \right]  +E_{\theta
}\left[  \frac{R_{0}}{\widehat{\pi}_{0}}\left(  \widehat{B}_{1}-\widehat{B}%
_{0}\right)  \right]  +E_{\theta}\left(  \widehat{B}_{0}-B_{0}\right) \\
&  =E_{\theta}\left[  \frac{R_{0}}{\widehat{\pi}_{0}}\delta P_{1}\delta
B_{1}+\delta P_{0}\delta B_{0}\right]
\end{align*}
Next we proceed to prove eq.(\ref{ebm}) and eq.(\ref{tbm}) by induction,

First of all,
\begin{align*}
&  E_{\theta}\left[  \frac{R_{1}}{\widehat{\pi}_{1}}\frac{R_{0}}{\widehat{\pi
}_{0}}\left(  Y-\widehat{B}_{1}\right)  +\frac{R_{0}}{\widehat{\pi}_{0}%
}\left(  \widehat{B}_{1}-\widehat{B}_{0}\right)  |L_{0}\right] \\
&  =E_{\theta}\left(  \frac{R_{0}}{\widehat{\pi}_{0}}\left(  B_{1}%
-\widehat{B}_{1}\right)  \left(  \frac{R_{1}}{\widehat{\pi}_{1}}-1\right)
+\frac{R_{0}}{\widehat{\pi}_{0}}\left(  B_{1}-\widehat{B}_{1}+\widehat{B}%
_{1}-\widehat{B}_{0}\right)  |L_{0}\right) \\
&  =E_{\theta}\left(  \frac{R_{0}}{\widehat{\pi}_{0}}\delta B_{1}\delta
P_{1}+\frac{R_{0}}{\widehat{\pi}_{0}}\delta B_{0}|L_{0}\right)
\end{align*}
and
\begin{gather*}
E_{\theta}\left[  \frac{R_{0}}{\widehat{\pi}_{0}}\delta P_{1}\delta
B_{1}\right] \\
=E_{\theta}\left(  \Pi_{\theta}\left(  q_{01}^{1/2}\delta B_{1}|\left(
q_{01}^{1/2}\overline{W}_{k_{1}}\right)  \right)  \Pi_{\theta}\left(
q_{0}q_{01}^{-1/2}\delta P_{1}|\left(  q_{01}^{1/2}\overline{W}_{k_{1}%
}\right)  \right)  \right)  +\\
E_{\theta}\left(  \Pi_{\theta}^{\bot}\left(  q_{01}^{1/2}\delta B_{1}|\left(
q_{01}^{1/2}\overline{W}_{k_{1}}\right)  \right)  \Pi_{\theta}^{\bot}\left(
q_{0}q_{01}^{-1/2}\delta P_{1}|\left(  q_{01}^{1/2}\overline{W}_{k_{1}%
}\right)  \right)  \right)
\end{gather*}
with
\begin{align*}
&  E_{\theta}\left(  \Pi_{\theta}\left(  q_{01}^{1/2}\delta B_{1}|\left(
q_{01}^{1/2}\overline{W}_{k_{1}}\right)  \right)  \Pi_{\theta}\left(
q_{0}q_{01}^{-1/2}\delta P_{1}|\left(  q_{01}^{1/2}\overline{W}_{k_{1}%
}\right)  \right)  \right) \\
&  =E_{\theta}\left(  q_{01}\delta B_{1}\overline{W}_{k_{1}}^{T}\right)
E_{\theta}\left(  q_{01}\overline{W}_{k_{1}}\overline{W}_{k_{1}}^{T}\right)
^{-1}E_{\theta}\left(  q_{0}\delta P_{1}\overline{W}_{k_{1}}\right)
\end{align*}
For $m=2,$%
\begin{align*}
&  E_{\theta}\left(  \widehat{\psi}_{2}\right)  -\psi\left(  \theta\right) \\
&  =E_{\theta}\left[  \frac{R_{0}}{\widehat{\pi}_{0}}\delta P_{1}\delta
B_{1}+\delta P_{0}\delta B_{0}\right]  +E_{\theta}\left(  \widehat{\psi}%
_{2,2}\right) \\
&  =E_{\theta}\left[  \frac{R_{0}}{\widehat{\pi}_{0}}\delta P_{1}\delta
B_{1}+\delta P_{0}\delta B_{0}\right]  -E_{\theta}\left[  \frac{R_{0}\pi_{1}%
}{\widehat{\pi}_{0}\widehat{\pi}_{1}}\delta B_{1}\overline{W}_{k_{1}}%
^{T}\right]  E_{\theta}\left[  \overline{W}_{k_{1}}\frac{R_{0}}{\widehat{\pi
}_{0}}\delta P_{1}\right] \\
&  -E_{\theta}\left\{  \left(  \frac{R_{0}}{\widehat{\pi}_{0}}\delta
P_{1}\delta B_{1}+\frac{R_{0}}{\widehat{\pi}_{0}}\delta B_{0}\right)
\overline{Z}_{k_{0}}^{T}\right\}  E_{\theta}\left[  \overline{Z}_{k_{0}}\delta
P_{0}\right] \\
&  =-E_{\theta}\left(  q_{0}\delta B_{0}\overline{Z}_{k_{0}}^{T}\right)
E_{\theta}\left(  q_{0}\overline{Z}_{k_{0}}\overline{Z}_{k_{0}}^{T}\right)
^{-1}\left[  E_{\theta}\left(  q_{0}\overline{Z}_{k_{0}}\overline{Z}_{k_{0}%
}^{T}\right)  -I\right]  E_{\theta}\left(  \overline{Z}_{k_{0}}\delta
P_{0}\right) \\
&  -E_{\theta}\left(  q_{01}\delta B_{1}\overline{W}_{k_{1}}^{T}\right)
E_{\theta}\left(  q_{01}\overline{W}_{k_{1}}\overline{W}_{k_{1}}^{T}\right)
^{-1}\left(  E_{\theta}\left(  q_{01}\overline{W}_{k_{1}}\overline{W}_{k_{1}%
}^{T}\right)  -I\right)  E_{\theta}\left(  \overline{W}_{k_{1}}q_{0}\delta
P_{1}\right) \\
&  -E_{\theta}\left(  q_{0}\delta P_{1}\delta B_{1}\overline{Z}_{k_{0}}%
^{T}\right)  E_{\theta}\left[  \overline{Z}_{k_{0}}\delta P_{0}\right] \\
&  +E_{\theta}\left(  \Pi_{\theta}^{\bot}\left(  q_{0}^{1/2}\delta
B_{0}|\left(  q_{0}^{1/2}\overline{Z}_{k_{0}}\right)  \right)  \Pi_{\theta
}^{\bot}\left(  q_{0}^{-1/2}\delta P_{0}|\left(  q_{0}^{1/2}\overline
{Z}_{k_{0}}\right)  \right)  \right) \\
&  +E_{\theta}\left(  \Pi_{\theta}^{\bot}\left(  q_{01}^{1/2}\delta
B_{1}|\left(  q_{01}^{1/2}\overline{W}_{k_{1}}\right)  \right)  \Pi_{\theta
}^{\bot}\left(  q_{0}q_{01}^{-1/2}\delta P_{1}|\left(  q_{01}^{1/2}%
\overline{W}_{k_{1}}\right)  \right)  \right)
\end{align*}
If e.q(\ref{ebm}) holds for $m-1\geq2,$ we next show it also holds for $m,$%
\begin{align*}
&  E_{\theta}\left(  \widehat{\psi}_{m}\right)  -\psi\left(  \theta\right) \\
&  =E_{\theta}\left(  \widehat{\psi}_{m,m}\right)  +\left(  E_{\theta}\left(
\widehat{\psi}_{m-1}\right)  -\widehat{\psi}_{m}\right) \\
&  =BI_{m-1,2}+\left(  -1\right)  ^{m-1}\times\\
&  E_{\theta}\left\{
\begin{array}
[c]{c}%
E_{\theta}\left[  q_{0}\delta P_{1}\delta B_{1}\overline{Z}_{k_{0}}%
^{T}\right]  \left[  E_{\theta}\left(  q_{0}\overline{Z}_{k_{0}}\overline
{Z}_{k_{0}}^{T}-I\right)  \right]  ^{m-2}E_{\theta}\left(  \overline{Z}%
_{k_{0}}\delta P_{0}\right) \\
+E_{\theta}\left[  q_{0}\delta B_{0}\overline{Z}_{k_{0}}^{T}\right]  \left[
E_{\theta}\left(  q_{0}\overline{Z}_{k_{0}}\overline{Z}_{k_{0}}^{T}-I\right)
\right]  ^{m-2}E_{\theta}\left(  \overline{Z}_{k_{0}}\delta P_{0}\right)
_{\_\left(  L1.1\right)  }\\
+E_{\theta}\left(  q_{01}\delta B_{1}\overline{W}_{k_{1}}^{T}\right)  \left[
E_{\theta}\left(  q_{01}\overline{W}_{k_{1}}\overline{W}_{k_{1}}^{T}-I\right)
\right]  ^{m-2}E_{\theta}\left(  \overline{W}_{k_{1}}q_{0}\delta P_{1}\right)
\__{\left(  L1.2\right)  }\\
+%
%TCIMACRO{\tsum \limits_{j=2}^{m-1}}%
%BeginExpansion
{\textstyle\sum\limits_{j=2}^{m-1}}
%EndExpansion%
\begin{array}
[c]{c}%
E_{\theta}\left[  \overline{Z}_{k_{0}}\delta P_{0}^{T}\right]  \left[
E_{\theta}\left(  q_{0}\overline{Z}_{k_{0}}\overline{Z}_{k_{0}}^{T}-I\right)
\right]  ^{j-2}E_{\theta}\left[  q_{01}\delta B_{1}\overline{Z}_{k_{0}%
}\overline{W}_{k_{1}}^{T}\right] \\
\times E_{\theta}\left(  q_{01}\overline{W}_{k_{1}}\overline{W}_{k_{1}}%
^{T}-I\right)  ^{m-1-j}E_{\theta}\left[  \overline{W}_{k_{1}}q_{0}\delta
P_{1}\right]  \__{\left(  L1.3.j\right)  }%
\end{array}
\end{array}
\right\}
\end{align*}%
\[
-\left(  -1\right)  ^{m-1}\times
\]%
\[
\left\{
\begin{array}
[c]{c}%
\left\{
\begin{array}
[c]{c}%
E_{\theta}\left\{  \left[  q_{0}\delta B_{0}\right]  \overline{Z}_{k_{0}}%
^{T}\right\}  E_{\theta}\left[  q_{0}\overline{Z}_{k_{0}}\overline{Z}_{k_{0}%
}^{T}\right]  ^{-1}\times\\
\left[  E_{\theta}\left(  q_{0}\overline{Z}_{k_{0}}\overline{Z}_{k_{0}}%
^{T}\right)  -I\right]  ^{m-2}E_{\theta}\left[  \overline{Z}_{k_{0}}\delta
P_{0}\right]
\end{array}
\right\}  _{-\left(  L2.1\right)  }\\
+\left\{
\begin{array}
[c]{c}%
E_{\theta}\left\{  \left[  q_{01}\delta B_{1}\right]  \overline{W}_{k_{1}}%
^{T}\right\}  E_{\theta}\left[  q_{01}\overline{W}_{k_{1}}\overline{W}_{k_{1}%
}^{T}\right]  ^{-1}\times\\
\left[  E_{\theta}\left(  q_{01}\overline{W}_{k_{1}}\overline{W}_{k_{1}}%
^{T}\right)  -I\right]  ^{m-2}E_{\theta}\left[  \overline{W}_{k_{1}}%
q_{0}\delta P_{1}\right]
\end{array}
\right\}  _{-\left(  L2.2\right)  }\\
+%
%TCIMACRO{\dsum \limits_{j=2}^{m-2}}%
%BeginExpansion
{\displaystyle\sum\limits_{j=2}^{m-2}}
%EndExpansion
\left\{
\begin{array}
[c]{c}%
E_{\theta}\left[  \overline{Z}_{k_{0}}\delta P_{0}^{T}\right]  \left[
E_{\theta}\left(  q_{0}\overline{Z}_{k_{0}}\overline{Z}_{k_{0}}^{T}-I\right)
\right]  ^{j-2}\times\\
E_{\theta}\left[  q_{01}\delta B_{1}\overline{Z}_{k_{0}}\overline{W}_{k_{1}%
}^{T}\right]  E_{\theta}\left[  q_{01}\overline{W}_{k_{1}}\overline{W}_{k_{1}%
}^{T}\right]  ^{-1}\times\\
E_{\theta}\left(  q_{01}\overline{W}_{k_{1}}\overline{W}_{k_{1}}^{T}-I\right)
^{m-1-j}E_{\theta}\left[  \overline{W}_{k_{1}}q_{0}\delta P_{1}\right]
\end{array}
\right\}  _{\_\left(  L2.3.j\right)  }\\
+E_{\theta}\left\{  \left[  q_{0}\delta P_{1}\delta B_{1}\right]  \overline
{Z}_{k_{0}}^{T}\right\}  \left[  E_{\theta}\left(  q_{0}\overline{Z}_{k_{0}%
}\overline{Z}_{k_{0}}^{T}\right)  -I\right]  ^{m-3}E_{\theta}\left[
\overline{Z}_{k_{0}}\delta P_{0}\right]  _{\_\left(  L2.4\right)  }%
\end{array}
\right\}
\]
It can be shown
\begin{align*}
&  \left(  L1.1\right)  -\left(  L2.1\right) \\
&  =\left(  -1\right)  ^{m-1}\left\{
\begin{array}
[c]{c}%
E_{\theta}\left\{  \left[  q_{0}\delta B_{0}\right]  \overline{Z}_{k_{0}}%
^{T}\right\}  E_{\theta}\left[  q_{0}\overline{Z}_{k_{0}}\overline{Z}_{k_{0}%
}^{T}\right]  ^{-1}\times\\
\left[  E_{\theta}\left(  q_{0}\overline{Z}_{k_{0}}\overline{Z}_{k_{0}}%
^{T}\right)  -I\right]  ^{m-1}E_{\theta}\left[  \overline{Z}_{k_{0}}\delta
P_{0}\right]
\end{array}
\right\}
\end{align*}%
\begin{align*}
&  \left(  L1.2\right)  -\left(  L2.2\right) \\
&  =\left(  -1\right)  ^{m-1}\left\{
\begin{array}
[c]{c}%
E_{\theta}\left\{  \left[  q_{01}\delta B_{1}\right]  \overline{W}_{k_{1}}%
^{T}\right\}  E_{\theta}\left[  q_{01}\overline{W}_{k_{1}}\overline{W}_{k_{1}%
}^{T}\right]  ^{-1}\times\\
\left[  E_{\theta}\left(  q_{01}\overline{W}_{k_{1}}\overline{W}_{k_{1}}%
^{T}\right)  -I\right]  ^{m-1}E_{\theta}\left[  \overline{W}_{k_{1}}%
q_{0}\delta P_{1}\right]
\end{array}
\right\}
\end{align*}
$\forall$ $2\leq j<m-1,$%
\begin{align*}
&  \left(  L1.3.j\right)  -\left(  L2.3.j\right) \\
&  =\left(  -1\right)  ^{m-1}\left\{
\begin{array}
[c]{c}%
E_{\theta}\left[  \overline{Z}_{k_{0}}\delta P_{0}^{T}\right]  \left[
E_{\theta}\left(  q_{0}\overline{Z}_{k_{0}}\overline{Z}_{k_{0}}^{T}-I\right)
\right]  ^{j-2}\times\\
E_{\theta}\left[  q_{01}\delta B_{1}\overline{Z}_{k_{0}}\overline{W}_{k_{1}%
}^{T}\right]  E_{\theta}\left[  q_{01}\overline{W}_{k_{1}}\overline{W}_{k_{1}%
}^{T}\right]  ^{-1}\times\\
E_{\theta}\left(  q_{01}\overline{W}_{k_{1}}\overline{W}_{k_{1}}^{T}-I\right)
^{m-j}E_{\theta}\left[  \overline{W}_{k_{1}}q_{0}\delta P_{1}\right]
\end{array}
\right\}
\end{align*}
If $\zeta_{m}\left(  L_{0},\theta\right)  \equiv E_{\theta}\left[  \delta
P_{0}\overline{Z}_{k_{0}}^{T}\right]  \left[  E_{\theta}\left(  q_{0}%
\overline{Z}_{k_{0}}\overline{Z}_{k_{0}}^{T}-I\right)  \right]  ^{m-3}%
\overline{Z}_{k_{0}}, $ then
\begin{align*}
&  \left(  L1.3.m-1\right)  -\left(  L2.4\right) \\
&  =E_{\theta}\left[  q_{01}\delta B_{1}\zeta_{m}\left(  L_{0},\theta\right)
\overline{W}_{k_{1}}^{T}\right]  E_{\theta}\left[  \overline{W}_{k_{1}}%
q_{0}\delta P_{1}\right]  -E_{\theta}\left\{  \zeta_{m}\left(  L_{0}%
,\theta\right)  q_{0}\delta P_{1}\delta B_{1}\right\} \\
&  =\left(  -1\right)  ^{m-1}\left\{
\begin{array}
[c]{c}%
E_{\theta}\left(  q_{01}\delta B_{1}\zeta_{m}\left(  L_{0},\theta\right)
\overline{W}_{k_{1}}^{T}\right)  E_{\theta}\left(  q_{01}\overline{W}_{k_{1}%
}\overline{W}_{k_{1}}^{T}\right)  ^{-1}\times\\
\left(  E_{\theta}\left(  q_{01}\overline{W}_{k_{1}}\overline{W}_{k_{1}}%
^{T}-I\right)  \right)  E_{\theta}\left(  \overline{W}_{k_{1}}q_{0}\delta
P_{1}\right)
\end{array}
\right\} \\
&  +\left(  -1\right)  ^{m}E_{\theta}\left(
\begin{array}
[c]{c}%
\Pi_{\theta}^{\bot}\left[  \left(
\begin{array}
[c]{c}%
q_{01}^{1/2}\delta B_{1}\zeta_{m}\left(  L_{0},\theta\right)
\end{array}
\right)  |\left(  q_{01}^{1/2}\overline{W}_{k_{1}}\right)  \right] \\
\times\Pi_{\theta}^{\bot}\left[  q_{0}q_{01}^{-1/2}\delta P_{1}|\left(
q_{01}^{1/2}\overline{W}_{k_{1}}\right)  \right]
\end{array}
\right)  _{-\left(  TB_{m-1,m-1}^{\left(  2\right)  }\right)  }%
\end{align*}
In summary%
\begin{gather*}
E_{\theta}\left(  \widehat{\psi}_{m}\right)  -\psi\left(  \theta\right) \\
=\left(  -1\right)  ^{m-1}BI_{m,1}+BI_{m-1,2}+\left(  TB_{m-1,m-1}^{\left(
2\right)  }\right)
\end{gather*}
Next we show $\left(  TB_{m-1,m-1}^{\left(  2\right)  }\right)  =BI_{m,2}%
-BI_{m-1,2},$%
\begin{align*}
&  \zeta_{m}\left(  L_{0},\theta\right) \\
&  =E_{\theta}\left[  \delta P_{0}\overline{Z}_{k_{0}}^{T}\right]  \left[
E_{\theta}\left(  q_{0}\overline{Z}_{k_{0}}\overline{Z}_{k_{0}}^{T}-I\right)
\right]  ^{m-3}\overline{Z}_{k_{0}}\\
&  =\left\{
\begin{array}
[c]{c}%
E_{\theta}\left[  \delta P_{0}\overline{Z}_{k_{0}}^{T}\right]  \left(  \left(
\begin{array}
[c]{c}%
E_{\theta}\left(  q_{0}\overline{Z}_{k_{0}}\overline{Z}_{k_{0}}^{T}\right)
^{-1}+\\
I-E_{\theta}\left(  q_{0}\overline{Z}_{k_{0}}\overline{Z}_{k_{0}}^{T}\right)
^{-1}%
\end{array}
\right)  \right) \\
\times\left[  E_{\theta}\left(  q_{0}\overline{Z}_{k_{0}}\overline{Z}_{k_{0}%
}^{T}-I\right)  \right]  ^{m-3}\overline{Z}_{k_{0}}%
\end{array}
\right\} \\
&  =\left\{  E_{\theta}\left[  \delta P_{0}\overline{Z}_{k_{0}}^{T}\right]
E_{\theta}\left(  q_{0}\overline{Z}_{k_{0}}\overline{Z}_{k_{0}}^{T}\right)
^{-1}\left[  E_{\theta}\left(  q_{0}\overline{Z}_{k_{0}}\overline{Z}_{k_{0}%
}^{T}-I\right)  \right]  ^{m-3}\overline{Z}_{k_{0}}\right\} \\
&  +\left\{
\begin{array}
[c]{c}%
E_{\theta}\left[  \delta P_{0}\overline{Z}_{k_{0}}^{T}\right]  E_{\theta
}\left(  q_{0}\overline{Z}_{k_{0}}\overline{Z}_{k_{0}}^{T}\right)
^{-1}E_{\theta}\left(  q_{0}\overline{Z}_{k_{0}}\overline{Z}_{k_{0}}%
^{T}-I\right) \\
\times\left[  E_{\theta}\left(  q_{0}\overline{Z}_{k_{0}}\overline{Z}_{k_{0}%
}^{T}-I\right)  \right]  ^{m-3}\overline{Z}_{k_{0}}%
\end{array}
\right\}
\end{align*}
Applying the above expression of $\zeta_{m}\left(  L_{0},\theta\right)  $ to
e.q$\left(  L3.1\right)  ,$ we find that
\[
\left(  TB_{m-1,m-1}^{\left(  2\right)  }\right)  +BI_{m-1,2}=BI_{m,2}%
\]

which completes the induction.

We want to mention that $TB_{j,j}^{\left(  2\right)  }$ $\left(  2\leq
j<m\right)  $ equals $\sum_{l=1}^{k_{0}^{j-1}}\left(
%TCIMACRO{\tprod \limits_{t=1}^{j}}%
%BeginExpansion
{\textstyle\prod\limits_{t=1}^{j}}
%EndExpansion
\widetilde{\tau}_{l,t}^{\left(  j\right)  }\left(  \theta\right)  -%
%TCIMACRO{\tprod \limits_{t=1}^{j}}%
%BeginExpansion
{\textstyle\prod\limits_{t=1}^{j}}
%EndExpansion
\tau_{l,t}^{\left(  j\right)  }\left(  \theta\right)  \right)  $,
$BI_{2,2}=\left(  \widetilde{\psi}^{\dag}\left(  \theta\right)  -\psi\left(
\theta\right)  \right)  +\widetilde{\tau}_{1,1}^{\left(  1\right)  }\left(
\theta\right)  -$ $\tau_{1,1}^{\left(  1\right)  }\left(  \theta\right)  ,$
$EB_{1}^{\left(  1\right)  }=$ $E_{\theta}\left(  \mathbb{IF}%
_{m,\widetilde{\psi}^{\dag}}\left(  \widehat{\theta}\right)  +\psi\left(
\widehat{\theta}\right)  \right)  -\widetilde{\psi}^{\dag}\left(
\theta\right)  ,$ $EB_{1}^{\left(  2\right)  }=E\left(  \mathbb{IF}%
_{m,\widetilde{\tau}_{1,1}^{\left(  1\right)  }\left(  \theta\right)  }\left(
\widehat{\theta}\right)  \right)  -\widetilde{\tau}_{1,1}^{\left(  1\right)
}\left(  \theta\right)  ,$ and $EB_{jj}^{\left(  2\right)  }=E\left(
\mathbb{IF}_{m,\widetilde{\widetilde{\psi}}_{jj}\left(  \theta\right)
}\left(  \widehat{\theta}\right)  \right)  -\widetilde{\widetilde{\psi}}%
_{jj}\left(  \theta\right)  .$ Therefore the results from this theorem is
consistent with eq.(\ref{bsp}). Technical details are not presented here.

Eq.(\ref{ebm2}) follows from eq.(\ref{ebm}) similarly as in the proof of
theorem \ref{BTOMM} and eq.(\ref{ebm3}) can be derived from eq.(\ref{ebm2}) by our
assumption of rate optimality of the initial estimator. Next we prove
eq.(\ref{tbm2}),%
\begin{align*}
&  \left\vert BI_{2,2}\right\vert \\
&  \leq\left\vert E\left(  \Pi_{\theta}^{\bot}\left(  q_{0}^{1/2}\delta
B_{0}|\left(  q_{0}^{1/2}\overline{Z}_{k_{0}}\right)  \right)  \Pi_{\theta
}^{\bot}\left(  q_{0}^{-1/2}\delta P_{0}|\left(  q_{0}^{1/2}\overline
{Z}_{k_{0}}\right)  \right)  \right)  \right\vert \\
&  +\left\vert E\left(  \Pi_{\theta}^{\bot}\left(  q_{01}^{1/2}\delta
B_{1}|\left(  q_{01}^{1/2}\overline{W}_{k_{1}}\right)  \right)  \Pi_{\theta
}^{\bot}\left(  q_{0}q_{01}^{-1/2}\delta P_{1}|\left(  q_{01}^{1/2}%
\overline{W}_{k_{1}}\right)  \right)  \right)  \right\vert \\
&  \leq\left\{  E\left(  \Pi_{\theta}^{\bot}\left(  q_{0}^{1/2}\delta
B_{0}|\left(  q_{0}^{1/2}\overline{Z}_{k_{0}}\right)  \right)  \right)
^{2}\right\}  ^{1/2}\left\{  E\left(  \Pi_{\theta}^{\bot}\left(  q_{0}%
^{-1/2}\delta P_{0}|\left(  q_{0}^{1/2}\overline{Z}_{k_{0}}\right)  \right)
\right)  ^{2}\right\}  ^{1/2}\\
&  +\left\{  E\left(  \Pi_{\theta}^{\bot}\left(  q_{01}^{1/2}\delta
B_{1}|\left(  q_{01}^{1/2}\overline{W}_{k_{1}}\right)  \right)  \right)
^{2}\right\}  ^{1/2}\left\{  E\left(  \Pi_{\theta}^{\bot}\left(  q_{0}%
q_{01}^{-1/2}\delta P_{1}|\left(  q_{01}^{1/2}\overline{W}_{k_{1}}\right)
\right)  \right)  ^{2}\right\}  ^{1/2}\\
&  =O_{p}\left(  \max\left[  k_{0}^{-\left(  \beta_{b_{0}}+\beta_{\pi_{0}%
}\right)  /d_{0}},k_{0}^{-\left(  \beta_{b_{1}}+\beta_{\pi_{1}}\right)
/d_{1}}\right]  \right)
\end{align*}
(Cauchy-Shwartz inequality. The last equality can easily be shown as in the
proof of theorem \ref{ON})%
\begin{align*}
&  \left\vert \left(  TB\left(  m,2\right)  \right)  \right\vert \\
&  =\left\vert E\left[
\begin{array}
[c]{c}%
\Pi_{\theta}^{\bot}\left[  \left(
\begin{array}
[c]{c}%
\frac{\pi_{1}^{1/2}}{\widehat{\pi}_{1}^{1/2}}\delta B_{1}\Pi_{\theta}\left(
q_{0}^{1/2}\delta P_{0}|\left(  q_{0}^{1/2}\overline{Z}_{k_{0}}\right)
\right)
\end{array}
\right)  |\left(  q_{01}^{1/2}\overline{W}_{k_{1}}\right)  \right] \\
\times\Pi_{\theta}^{\bot}\left[  q_{0}q_{01}^{-1/2}\delta P_{1}|\left(
q_{01}^{1/2}\overline{W}_{k_{1}}\right)  \right]
\end{array}
\right]  \right\vert \\
&  \leq\left\vert E\left[
\begin{array}
[c]{c}%
\Pi_{\theta}^{\bot}\left[  \left(
\begin{array}
[c]{c}%
\delta B_{1}q_{01}^{1/2}\delta P_{0}%
\end{array}
\right)  |\left(  q_{01}^{1/2}\overline{W}_{k_{1}}\right)  \right] \\
\times\Pi_{\theta}^{\bot}\left[  q_{0}q_{01}^{-1/2}\delta P_{1}|\left(
q_{01}^{1/2}\overline{W}_{k_{1}}\right)  \right]
\end{array}
\right]  \right\vert \\
&  +\left\vert E\left[
\begin{array}
[c]{c}%
\Pi_{\theta}^{\bot}\left[  \left(
\begin{array}
[c]{c}%
\frac{\pi_{1}^{1/2}}{\widehat{\pi}_{1}^{1/2}}\delta B_{1}\Pi_{\theta}^{\bot
}\left(  q_{0}^{1/2}\delta P_{0}|\left(  q_{0}^{1/2}\overline{Z}_{k_{0}%
}\right)  \right)
\end{array}
\right)  |\left(  q_{01}^{1/2}\overline{W}_{k_{1}}\right)  \right] \\
\times\Pi_{\theta}^{\bot}\left[  q_{0}q_{01}^{-1/2}\delta P_{1}|\left(
q_{01}^{1/2}\overline{W}_{k_{1}}\right)  \right]
\end{array}
\right]  \right\vert \\
&  \leq O_{p}\left(  \max\left(  k_{1}^{-\left(  \min\left(  \beta_{\pi_{0}%
},\beta_{b_{1}}\right)  +\beta_{\pi_{1}}\right)  /d_{1}},k_{1}^{-\beta
_{\pi_{1}}/d_{1}}k_{0}^{-\beta_{\pi_{0}}/d_{0}}\left(  \frac{\log n}%
{n}\right)  ^{\frac{\beta_{b_{1}}}{d_{1}+\beta_{b_{1}}}}\right)  \right)
\end{align*}
The last inequality holds because
\begin{align*}
&  E\left(  \Pi_{\theta}^{\bot}\left[  \left(
\begin{array}
[c]{c}%
\frac{\pi_{1}^{1/2}}{\widehat{\pi}_{1}^{1/2}}\delta B_{1}\Pi_{\theta}^{\bot
}\left(  q_{0}^{1/2}\delta P_{0}|\left(  q_{0}^{1/2}\overline{Z}_{k_{0}%
}\right)  \right)
\end{array}
\right)  |\left(  q_{01}^{1/2}\overline{W}_{k_{1}}\right)  \right]  \right)
^{2}\\
&  \leq E\left(  \frac{\pi_{1}^{1/2}}{\widehat{\pi}_{1}^{1/2}}\delta B_{1}%
\Pi_{\theta}^{\bot}\left(  q_{0}^{1/2}\delta P_{0}|\left(  q_{0}%
^{1/2}\overline{Z}_{k_{0}}\right)  \right)  \right)  ^{2}\\
&  \leq\left\vert \left\vert \frac{\pi_{1}^{1/2}}{\widehat{\pi}_{1}^{1/2}%
}\delta B_{1}\right\vert \right\vert _{\infty}^{2}E\left(  \Pi_{\theta}^{\bot
}\left(  q_{0}^{1/2}\delta P_{0}|\left(  q_{0}^{1/2}\overline{Z}_{k_{0}%
}\right)  \right)  \right)  ^{2}%
\end{align*}
(The first inequality holds because the orthogonal projection operator has a
norm no greater than 1)%
\begin{align*}
&  \left\vert TB\left(  m,3\right)  \right\vert \\
&  =\left\vert E\left[
\begin{array}
[c]{c}%
\Pi_{\theta}^{\bot}\left[  \left(
\begin{array}
[c]{c}%
q_{01}^{1/2}\delta B_{1}E\left[  \delta P_{0}\overline{Z}_{k_{0}}^{T}\right]
E\left[  q_{0}\overline{Z}_{k_{0}}\overline{Z}_{k_{0}}^{T}\right]  ^{-1}\\
\times\left(  E\left[  q_{0}\overline{Z}_{k_{0}}\overline{Z}_{k_{0}}%
^{T}-I\right]  \right)  ^{m-2}\overline{Z}_{k_{0}}%
\end{array}
\right)  |\left(  q_{01}^{1/2}\overline{W}_{k_{1}}\right)  \right] \\
\times\Pi_{\theta}^{\bot}\left[  q_{0}q_{01}^{-1/2}\delta P_{1}|\left(
q_{01}^{1/2}\overline{W}_{k_{1}}\right)  \right]
\end{array}
\right]  \right\vert \\
&  \leq\left\{
\begin{array}
[c]{c}%
\left\{  E\left(
\begin{array}
[c]{c}%
q_{01}^{1/2}\delta B_{1}E\left[  \delta P_{0}\overline{Z}_{k_{0}}^{T}\right]
E\left[  q_{0}\overline{Z}_{k_{0}}\overline{Z}_{k_{0}}^{T}\right]  ^{-1}\\
\times\left(  E\left[  q_{0}\overline{Z}_{k_{0}}\overline{Z}_{k_{0}}%
^{T}-I\right]  \right)  ^{m-2}\overline{Z}_{k_{0}}%
\end{array}
\right)  ^{2}\right\}  ^{1/2}\\
\times\left\{  E\left(  \Pi_{\theta}^{\bot}\left[  q_{0}q_{01}^{-1/2}\delta
P_{1}|\left(  q_{01}^{1/2}\overline{W}_{k_{1}}\right)  \right]  \right)
^{2}\right\}  ^{1/2}%
\end{array}
\right\}
\end{align*}
(Cauchy-Shwartz inequality and projection operator has operator norm of 1)%
\begin{align*}
&  \leq\left\{
\begin{array}
[c]{c}%
\left\vert \left\vert q_{01}^{1/2}\frac{f_{1}}{\widehat{f}_{1}}\right\vert
\right\vert _{\infty}\left\vert \left\vert \delta B_{1}\right\vert \right\vert
_{\infty}\left\vert \left\vert q_{0}^{-1/2}\frac{\widehat{f}_{0}}{f_{0}%
}\right\vert \right\vert _{\infty}\left\vert \left\vert \frac{f_{0}}{\pi
_{0}\widehat{\pi}_{0}}\right\vert \right\vert _{\infty}\\
\times\left\vert \left\vert \delta g_{0}\right\vert \right\vert _{\infty
}^{m-2}\left\{  \int\left(  \pi_{0}-\widehat{\pi}_{0}\right)  ^{2}%
dL_{0}\right\}  ^{1/2}\times\\
\left\{  E\left(  \Pi_{\theta}^{\bot}\left[  q_{0}q_{01}^{-1/2}\delta
P_{1}|\left(  q_{01}^{1/2}\overline{W}_{k_{1}}\right)  \right]  \right)
^{2}\right\}  ^{1/2}%
\end{array}
\right\} \\
&  =O_{p}\left(  \left(  \frac{\log n}{n}\right)  ^{-\frac{\beta_{b_{1}}%
}{d_{1}+\beta_{b_{1}}}-\frac{\left(  m-2\right)  \beta_{g_{0}}}{d_{0}%
+\beta_{g_{0}}}}n^{-\frac{\beta_{\pi_{0}}}{d_{0}+\beta_{\pi_{0}}}}%
k_{1}^{-\beta_{\pi_{1}}/d_{1}}\right)
\end{align*}

\end{proof}


\begin{thebibliography}{99}                                                                                               %
\bibitem {arellano}Arellano M. (2003) "\textit{Panel Data Econometrics".
}Oxford University Press: Advanced Texts in Econometrics.

\bibitem {bhatt}Bhattacharyya A. (1947) "On some analogues of the amount of
information and their use in statistical estimation II-III"
\ \textit{Sankhy\={a}}, Vol.8, 3, 201--218.

\bibitem {bickel1}Bickel P. , Klassen C., Ritov Y., Wellner J. (1993) "
\textit{Efficient and Adaptive Estimation for Semiparametric Models}". Springer.

\bibitem {bickel3}Bickel P., and Ritov Y. (2003) "Nonparametric estimators
which can be "plugged-in". \textit{Annals of Statist.} 31(4), 1033--53.

\bibitem {birge}Birge L., Massart P.(1995) "Estimation of Integral Functionals
of a Density". \ \ \textit{Annals of Statistics}. 23(1), 11-29.

\bibitem {He}He, X. and Shao, Q.M. (2000). "On parameters of increasing
dimension". \textit{Journal of Multivariate Analysis}, 73(1), 120-135.

\bibitem {klassen}Klassen C. (1987) "Consistent Estimation of the Influence
Function of Locally Asymptotically Linear Estimators". {}\textit{Annals of
Statistics. }15(4), 1548-62.

\bibitem {lee}Lee A.J. (1990) "\textit{U-Statistics: Theory and Practice}".
Marcel Dekker, New York.

\bibitem {lindsay}Lindsay R, Waterman B. (1996) "Projected Score Methods for
Approximating Conditional Scores". \textit{Biometrika} 83(1):1-13.

\bibitem {Li}Li L., Tchetgen E., van der Vaart AW, and Robins JM. (2006)
"Robust Inference with Higher Order Inference Functions: Part II." 2005
\textit{JSM Proceedings. }American Statistical Associations. 2558-2565.

\bibitem {mallat}Mallat SG. (1998) "\textit{A Wavelet Tour of Signal
Processing}". Acedemic Press.

\bibitem {newey}Newey W., {}Hsieh F., and Robins JM. (2004) {}{}{}"Twicing
kernels and a small bias property of semiparametric estimators" {}%
\textit{Econometrica} 72, 947-962.

\bibitem {pfanzagl}Pfanzagl J. (1990) "\textit{Estimation in Semiparametric
Models: Some Recent Developments}". \ Lecture Notes in Statistics 31,
Springer-Verlag, Berlin 1985, 505 pages

\bibitem {portnoy}Portnoy S. (1988). "Asymptotic Behavior of Likelihood
Methods for Exponential Families When the Number of Parameters Tends to
Infinity". {}{}{}\textit{Annals of Statistics. }{}16, 356-366

\bibitem {pyke}Pyke R., "Spacings (with discussion)", \textit{J. Roy. Statis.
Soc. B} 27 (1965), 395--449.

\bibitem {bickel2}Ritov Y., Bickel P. (1990) "Achieving information bounds in
non- and semi-parametric models". {}\textit{Annals of Statistics. {}}18, 925-938.

\bibitem {robins1}Robins J., Ritov, Y (1997). "Toward a
Curse-of-dimensionality Appropriate (CODA) Asymptotic Theory for
Semiparametric Models". \textit{Statistics in Medicine.} 16 285-319.

\bibitem {robins3}Robins JM. (2004) "Optimal Structural Nested Models for
Optimal Sequential Decisions". in DY Lin and P. Heagerty (Eds.)
\textit{Proceedings of the Second Seattle Symposium in Biostatistics}. New
York Springer.

\bibitem {rotnitzky}Robins JM, Rotnitzky A. (2001).Comment on the Bickel and
Kwon article, "Inference for semiparametric models: Some questions and an
answer\textquotedblright\ Statistica Sinica, 11(4):920-936. ["On Double Robustness."]

\bibitem {Vaart2}Robins J.M, Li L, Tchetgen Eric, van der Vaart AW (2007),
"Asymptotic Normality of Degenerate U-statistics". Working paper.

\bibitem {robins4}Robins J.M, van der Vaart AW. (2006) " Adaptive
Nonparametric Confidence Sets". \ \textit{Annals of statistics}. 34(1):229-253.

\bibitem {Robins 5}Robins JM, Li L, Tchetgen E, van der Vaart A. (2008).
Higher order influence functions and minimax estimation of nonlinear
functionals. Probability and Statistics: Essays in Honor of David A. Freedman 2:335-421.

\bibitem {mcleish}Small D., McLeish C.(1994) "\textit{Hilbert Space Methods in
Probability and Statistical Inference}". Wiley.

\bibitem {tchetgen}Tchetgen E., Li L., van der Vaart AW, and Robins JM. (2006)
"Robust Inference with Higher Order Inference Functions: Part I." 2005
\textit{JSM Proceedings. }American Statistical Associations. 2644-2651.

\bibitem {tchetgen2}Tchetgen E., Li L., van der Vaart AW, and Robins JM.
(2007) "Higher Order U-statistics Estimators for Longitudinal Missing Data and
Causal Inference Models". working paper

\bibitem {van der laan}van der Laan M and Dudoit S.(2003). Asymptotics of
Cross-Validated Risk Estimation in Estimator Selection and Performance
Assessment. Technical report.

\bibitem {laan2}van der Laan M, and Robins JM. (2003) "\textit{Unified Methods
for Censored Longitudinal Data and Causality" }Springer Series in Statistics.

\bibitem {cai}Wang, L., Brown, L. D., Cai, T. Levine, M. (2006). "Effect of
mean on variance function estimation in nonparametric regression". Technical
Report .

\bibitem {cai2}Cai, T., Levine, M. Wang, L. (2006) "Variance function
estimation in multivariate nonparametric regression". Technical Report.

\bibitem {vaart}van der Vaart, AW (1991). "On Differentiable Functionals".
\textit{Ann. Statist}. 19 178-204.

\bibitem {vaart3}van der Vaart, AW (1998). \ \ "\textit{Asymptotic
Statistics}". \ Cambridge Series in Statistical and Probabilistic Mathematics.
\end{thebibliography}
\end{document}